\documentclass[12pt]{book}
\usepackage{mathpazo}
\usepackage{physics}
\usepackage{graphicx}
\usepackage[Rejne]{fncychap}
\usepackage{amsmath, amssymb, amsthm}
\usepackage[letterpaper, margin=1.2in]{geometry}
\usepackage{enumitem}
\usepackage{hyperref}
\usepackage{listings}
\usepackage{verbatim}
\usepackage{cite}
\usepackage{lettrine}
\usepackage{tabularx,multirow}
\usepackage{pagecolor}
\usepackage{afterpage}
\usepackage{xcolor}
\usepackage{imakeidx}

\usepackage{etoolbox}
\patchcmd{\part}{\thispagestyle{plain}}{\thispagestyle{empty}}{}{}
\patchcmd{\chapter}{\thispagestyle{plain}}{\thispagestyle{empty}}{}{}

\usepackage[]{graphicx}

\makeatletter
\def\@partimage{}
\newcommand{\partimage}[2][]{\gdef\@partimage{\includegraphics[#1]{#2}}}
\def\@part[#1]#2{%
    \ifnum \c@secnumdepth >-2\relax
      \refstepcounter{part}%
      \addcontentsline{toc}{part}{\thepart\hspace{1em}#1}%
    \else
      \addcontentsline{toc}{part}{#1}%
    \fi
    \markboth{}{}%
    {\centering
     \interlinepenalty \@M
     \normalfont
     \ifnum \c@secnumdepth >-2\relax

        \vspace{24pt}
     
       \huge\bfseries \textsc{\partname}\nobreakspace\thepart
       \par
       \vskip 20\p@
     \fi
     \Huge \bfseries #2\vfil\@partimage\vfil\par}%
    \@endpart}
\makeatother

\usepackage{lscape}
\usepackage{tikz}

\usepackage{wrapfig}
\definecolor{darkblue}{HTML}{004C93} 
\usepackage{tcolorbox}
\tcbset{
    center title,
		top=1cm,
		right=0cm,
		left=0cm,
    colback=darkblue,
    colframe=white,
    width=15.85cm,
		height=5.62cm,
		halign=right,
		valign=top,
		flush right,
		sharp corners=all
    }

\graphicspath{{./}{images/}}

\newcommand{\tud}[2]{^{#1}_{\phantom{#1}#2}}
\newcommand{\tudu}[3]{^{#1\phantom{#2}#3}_{\phantom{#1}#2}}
\newcommand{\tdu}[2]{_{#1}^{\phantom{#1}#2}}
\newcommand{\tdud}[3]{_{#1\phantom{#2}#3}^{\phantom{#1}#2}}
\newcommand{\tudud}[4]{^{#1\phantom{#2}#3}_{\phantom{#1}#2\phantom{#3}#4}}

\newcommand{\defining}[1]{\textbf{#1}}

\newcommand{\so}{{\text{so}}}
\newcommand{\sw}{{\text{sw}}}
\newcommand{\tp}{{\text{tp}}}

\newcommand{\sL}{{{}^*\!L}}
\newcommand{\sV}{{{}^*\!V}}
\newcommand{\sRs}{{{}^*\!R^*}}
\newcommand{\sWs}{{{}^*\!W^*}}
\newcommand{\sMs}{{{}^*\!M^*}}
\newcommand{\ldv}[2]{\frac{\text{D}#1}{\dd#2}}
\newcommand{\heart}{\ensuremath\heartsuit}

\newcommand{\os}[1]{\mathcal O\qty(\mathcal S^{#1})}

\newtheorem{proposition}{Proposition}[chapter]
\newtheorem{definition}{Definition}[chapter]
\newtheorem{conjecture}{Conjecture}[chapter]
\newtheorem{theorem}{Theorem}[chapter]
\newtheorem{lemma}{Lemma}[chapter]
\newtheorem{main_result}{Main result}[chapter]
\newtheorem{preliminary_result}{Preliminary result}[chapter]

\usepackage{fancyhdr}
\fancypagestyle{plain}{
    \lhead{}
    \fancyhead[R]{\thepage}
    \fancyhead[L]{\leftmark}
    
    \fancyfoot{}
}

\fancypagestyle{intro}{
    \lhead{}
    \fancyhead[R]{\thepage}
    \fancyhead[L]{\textsc{INTRODUCTION}}
    
    \fancyfoot{}
}

\fancypagestyle{contents}{
    \lhead{}
    \fancyhead[R]{\thepage}
    \fancyhead[L]{\textsc{CONTENTS}}
    
    \fancyfoot{}
}

\fancypagestyle{conclusion}{
    \lhead{}
    \fancyhead[R]{\thepage}
    \fancyhead[L]{\textsc{CONCLUSION}}
    
    \fancyfoot{}
}

\fancypagestyle{biblio}{
    \lhead{}
    \fancyhead[R]{\thepage}
    \fancyhead[L]{\nouppercase\textsc{BIBLIOGRAPHY}}
    
    \fancyfoot{}
}

\fancypagestyle{index}{
    \lhead{}
    \fancyhead[R]{\thepage}
    \fancyhead[L]{\textsc{INDEX OF THE MAIN EQUATIONS}}
    
    \fancyfoot{}
}

\newcommand{\sint}{\rotatebox[origin=c]{-90}{$\backslash$}\hspace{-14pt}\int}
\newcommand{\sintline}{\scriptsize\rotatebox[origin=c]{-90}{$\backslash$}\hspace{-10pt}\int}

\usepackage{empheq}
\newcommand*\widefbox[1]{\fbox{\hspace{2em}#1\hspace{2em}}}

\newcommand{\boxedeqn}[2]{
\begin{empheq}[box=\widefbox]{align}
  #1
\end{empheq}\index{#2}

\noindent
}

\newcommand{\cell}[4]{\fill[color=#4] (#1,#2) rectangle ({#1+1.5},{#2+.75});
\node[] () at ({#1+.75},{#2+.375}) {#3}
}

\usepackage{makeidx}
\makeindex

\makeindex[columns=2,title=Index of the main equations]

\begin{document}

\renewcommand{\labelitemi}{{\color{black!40}$\blacktriangleright$}}

\pagestyle{empty}
\pagenumbering{gobble}

\color{white}

\definecolor{dukeblue}{rgb}{0.15, 0.0, 0.61}

\newpagecolor{dukeblue}\afterpage{\restorepagecolor}

\begin{center}

\phantom{moumoulamouette}

\vspace{\stretch{0.2}}

\textbf{\fontsize{34}{50}\selectfont The Motion of Test Bodies\\\vspace{18pt}
around Kerr Black Holes}

\vspace{21pt}

\rule{0.8\textwidth}{2pt}

\vspace{21pt}

{\LARGE Adrien Druart}

\vspace{\stretch{0.6}}

\includegraphics[width=0.8\textwidth]{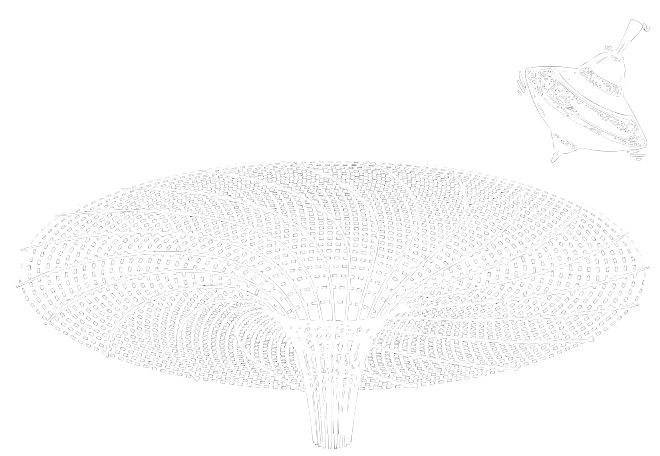}

\vspace{\stretch{0.6}}

\textsc{\large---\hspace{6pt} PhD Thesis \hspace{6pt}---}

\vspace{\stretch{0.1}}
    
\end{center}

\color{black}
\pagenumbering{gobble}
\pagestyle{empty}

\phantom{a}
\newpage

\begin{center}

\phantom{moumoulamouette}

\vspace{\stretch{0.2}}

\textbf{\fontsize{34}{50}\selectfont The Motion of Test Bodies\\\vspace{18pt}
around Kerr Black Holes}

\vspace{21pt}

\rule{0.8\textwidth}{2pt}

\vspace{21pt}

\large{\textbf{Thesis presented by Adrien DRUART}} \\
\color{black}
\large{in fulfilment of the requirements of the PhD Degree in Sciences ("Docteur en Sciences")\\\vspace{6pt}
Ann\'ee acad\'emique 2022-2023}

\vspace{\stretch{0.6}}

\textbf{Supervisor:}\\ Geoffrey COMP\`ERE (Universit\'e Libre de Bruxelles)

\vspace{\stretch{0.2}}

\textbf{Thesis jury : } \\
Riccardo ARGURIO (Universit\'e Libre de Bruxelles, Chair) \\
Stéphane DETOURNAY (Universit\'e Libre de Bruxelles, Secretary) \\
Tanja HINDERER (Utrecht University) \\
Justin VINES (Max-Planck-Institute for Gravitational Physics)

\vspace{\stretch{0.6}}

\includegraphics[height=2.5cm]{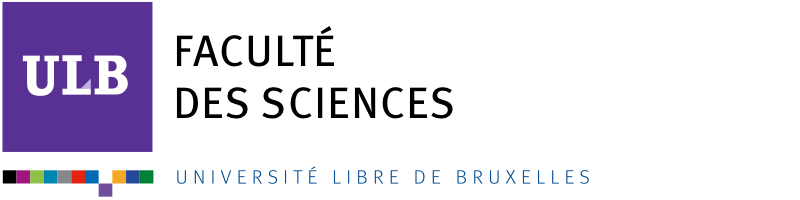} 
\includegraphics[height=2.5cm]{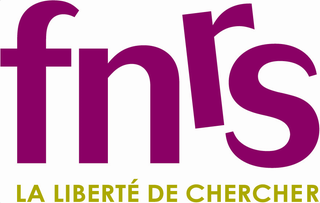}
    
\end{center}

\newpage

\phantom{a}

\newpage

\pagestyle{empty}

\phantom{a}
\vspace{\stretch{0.3}}
\begin{flushright}
\begin{minipage}{0.8\textwidth}\textsf{
In my entire scientific life, extending over forty-five years, the most shattering experience has been the realization that an exact solution of Einstein's equations of general relativity, discovered by the New Zealand mathematician Roy Kerr, provides the absolute exact representation of untold numbers of massive black holes that populate the universe. This ``shuddering before the beautiful,'' this incredible fact that a discovery motivated by a search after the beautiful in mathematics should find its exact replica in Nature, persuades me to say that beauty is that to which the human mind responds at its deepest and most profound level.
}
\end{minipage}

\vspace{24pt}

\textsf{--- Subrahmanijan Chandrasekhar }

\end{flushright}
\vspace{\stretch{0.6}}

\newpage
\phantom{a}
\newpage

\pagenumbering{arabic}
\pagestyle{plain}

\pagenumbering{gobble}
\phantom{a}
\vspace{30pt}
\pagestyle{empty}

\begin{center}\Large
    \textsc{Summary of the thesis}
\end{center}

\vspace{\stretch{0.75}}

This thesis aims to explore the properties of the motion of finite size, compact test bodies around a Kerr black hole in the small mass-ratio approximation. The small body is modelled as a perturbation of Kerr geometry, neglecting its gravitational back-reaction but including deviations from a purely geodesic motion by allowing it to possess a non-trivial internal structure. Such a body can be accurately described by a worldline endowed with a collection of multipole moments. Hereafter, we shall always consider the multipole expansion truncated at quadrupole order. Moreover, only spin-induced quadrupole moment will be taken into account, thus discarding the presence of any tidal-type deformation. For astrophysically realistic objects, this approximation is consistent with expanding the equations of motion up to second order in the body's spin magnitude.

The text is structured as follows. The first part is devoted to an extended review of geodesic motion in Kerr spacetime, including Hamiltonian formulation and classification of timelike geodesics, with a particular emphasis put on near-horizon geodesics of high spin black holes. The second part introduces the equations of motion for extended test bodies in generic curved spacetime, also known as Mathisson-Papapetrou-Dixon (MPD) equations. They are derived from a generic action principle, and their physical significance and mathematical consistency is examined in details. The third part discusses conserved quantities for the MPD equations in Kerr spacetime, restricting to the aforementioned quadrupole approximation. The conservation is required to hold perturbatively in the test body's spin magnitude, and the related conserved quantities are build through the explicit resolution of the conservation constraint equations. Finally, the covariant Hamiltonian formulation of test body motion in curved spacetime is presented, and an Hamiltonian reproducing the spin-induced quadrupole MPD equations is derived. Two applications of the Hamiltonian formalism are subsequently discussed: (i) the integrability properties of MPD equations in Kerr and Schwarzschild spacetimes and (ii) the Hamilton-Jacobi formulation of MPD equations in Kerr spacetime. It is shown that the constants of motion obtained in the previous part directly arise while solving the Hamilton-Jacobi equation at first order in the spin magnitude. Some expectations regarding the computation at quadratic order close the discussion.  
\vspace{\stretch{1}}

\newpage

\phantom{a}

\newpage

\pagestyle{plain}

\pagenumbering{Roman}
\pagestyle{contents}
\tableofcontents

\pagestyle{empty}

\addcontentsline{toc}{chapter}{Acknowledgements}
\chapter*{Acknowledgements}

My warmest thanks first go to my supervisor, Geoffrey Compère. His constant availability, the trust and the freedom he provided me all along these four years have been fundamental to the achievement of the present work.

\vspace{24pt}
\noindent
It would like to thank all the people I had the occasion to collaborate with during my PhD degree: Lorenzo Küchler, Justin Vines and Paul Ramond on the research side, Riccardo Argurio and Stéphane Detournay on the teaching side. A big thank you to the three wonderful secretaries of our group, who always managed to make things easy on the organizational level.

\vspace{24pt}
\noindent
Enfin, je tiens à remercier de tout mon coeur ma famille et mes amis, qu'ils soient de Liège, de Bruxelles, de Clerheid ou d'ailleurs. Vous tous qui m'avez entouré pendant ces années et sans qui je ne serais pas moi-même aujourd'hui, cette thèse est également la vôtre.

\vspace{24pt}
\noindent
Finally, thanks to Stéphane Detournay and Tanja Hinderer for pointing a couple of typos in the draft.

\addcontentsline{toc}{part}{\textsc{Introduction}}
\part*{\textsc{Introduction}}
\chaptermark{\textsc{Introduction}}
\pagestyle{intro}

\lettrine{T}{his} thesis is devoted to the understanding of some theoretical aspects of the motion of small objects in the neighbourhood of supermassive blacks holes. By small objects, we have in mind either neutron stars or stellar mass black holes, which both share the property of not being (too much) tidally deformed even in a strong gravitational field, thus remaining compact at any time. Any bounded binary system composed of a ``central'' supermassive black hole and of such a stellar mass companion will dissipate energy and angular momentum though the emission of Gravitational Waves (GWs), finally leading to its coalescence.  

Our framework for modelling this kind of systems will be General Relativity (GR), and our main concern will be the understanding on how the internal structure of the small compact object (``the secondary'') will affect its motion around the much more massive black hole (``the primary'', or the ``central'' black hole). These deviations originate from the fact that the secondary is not a point-wise particle, but can be spinning, exhibit quadrupole or higher order multipole moments induced by either its proper rotation or its tidal deformability\ldots All these effects will be collectively referred to as \textit{finite size effects}. The zeroth order approximation to this description (corresponding to switch off all the finite size effects) amounts to study timelike geodesics in Kerr spacetime, which is the most generic GR solution accounting for an astrophysically realistic stationary black hole. The finite size induced corrections can then be studied as perturbations added on the top of this geodesic motion. 

This introduction aims to both motivate the present work from the current context of GW observations and to briefly set the background context in which the forthcoming discussion will take place. This thesis is divided into four parts, which are intended to be rather independent one to another, and more specific introductions will be provided at the beginning of each of them. We encourage the reader willing to acquire an overview of this text to first read the four part's introductions before turning to the chapters themselves.

\subsubsection{Observational motivations and large mass ratio binaries}

\begin{figure}
    \centering
    \includegraphics[width=\textwidth]{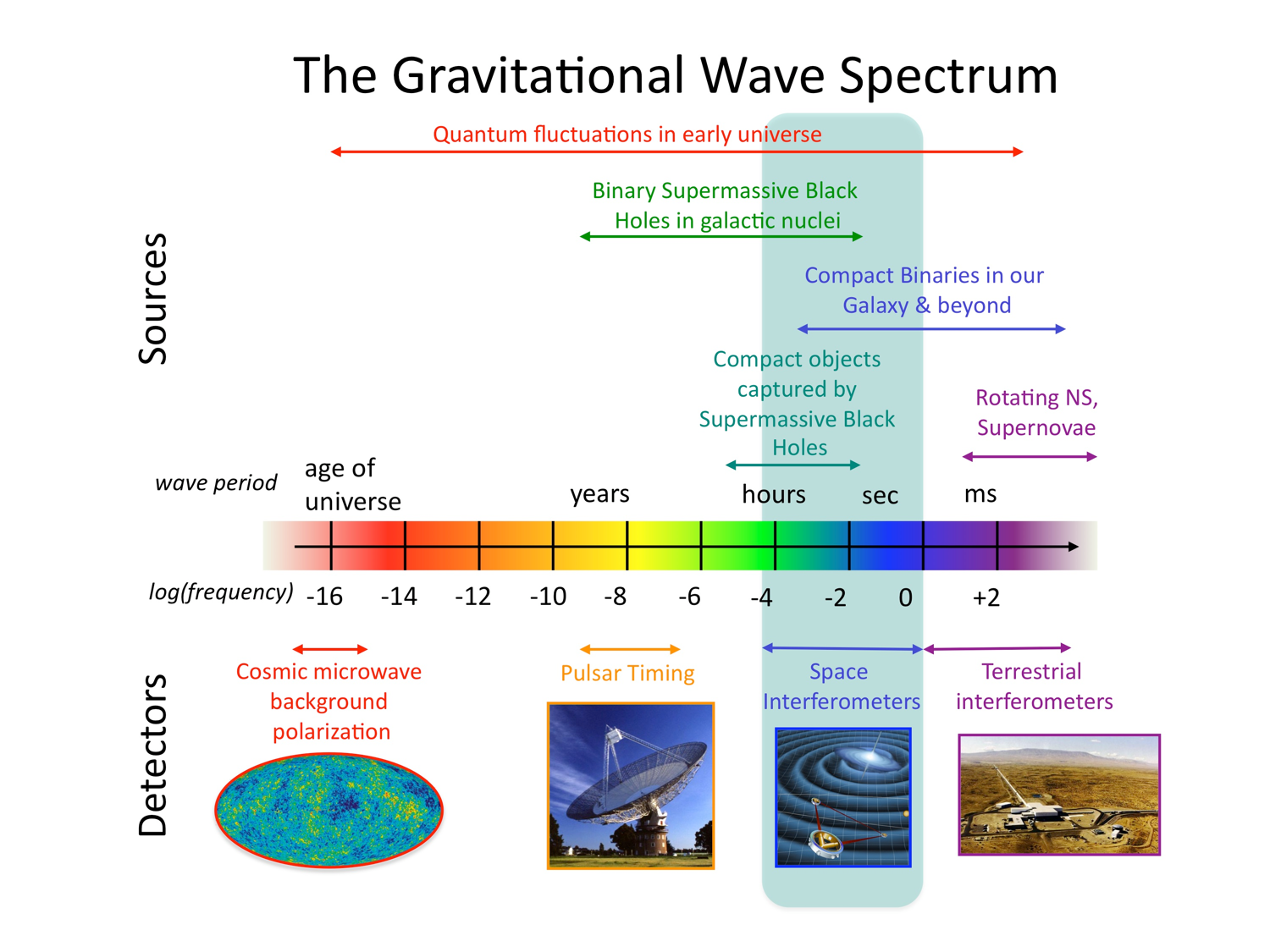}
    \caption{The gravitational wave spectrum: astrophysical sources, and corresponding detectors. Image credit: NASA Goddard Space Flight Center.}
    \label{fig:spectrum}
\end{figure}

Since 2015 and the first detection of a binary black hole coalescence by the LIGO collaboration \cite{LIGOScientific:2016aoc}, we have entered into a gravitational wave astronomy era, having potentially huge amounts of informations to bring to the scientific community, both on the theoretical and astrophysical sides.

The main GW sources for the present detectors are the coalescences of systems of binary stellar mass black holes and/or neutron stars. Such systems exhibit a mass ratio which is typically not much greater than 1:10 \cite{LIGOScientific:2021djp}. Altogether, this yields the frequency of the GW signal emitted to be centered around one hundred Hertz, thus exactly lying in the LIGO/VIRGO/KAGRA terrestrial interferometers band. However, the upcoming space-based detectors such as the Laser Interferometer Space Antenna (LISA) mission \cite{LISA:2017pwj} will be sensitive to GW signals centered around the millihertz, thus allowing the detection of GWs emitted by radically different types of astrophysical sources, see Figure \ref{fig:spectrum}.

One of these new in-band phenomena corresponds to the capture of a compact, stellar mass object by a supermassive black hole. This phenomenon is known as an \textit{Extreme Mass Ratio Inspiral} (or EMRI for short), provided that the mass ratio $\epsilon$ between the two bodies satisfies
\begin{align*}
10^{-4}\geq\epsilon\triangleq\frac{\mu}{M}.
\end{align*}
Here, $M$ denotes the mass of the supermassive black hole while $\mu$ stands for the mass of the small compact object. The prospective observation of this type of events motivates the modelling of black hole binaries in the small mass ratio regime up to a high precision, since accurate parameter extraction from the LISA data would require to keep track of the orbital phase of the binary with a precision of about one radian over the whole in-band inspiral, which can last for a few hundred thousand cycles \cite{AmaroSeoane:2007aw,Babak:2017tow}.

\subsubsection{Extreme mass ratio inspirals in General Relativity}

\begin{figure}
    \centering
    \includegraphics[width=\textwidth]{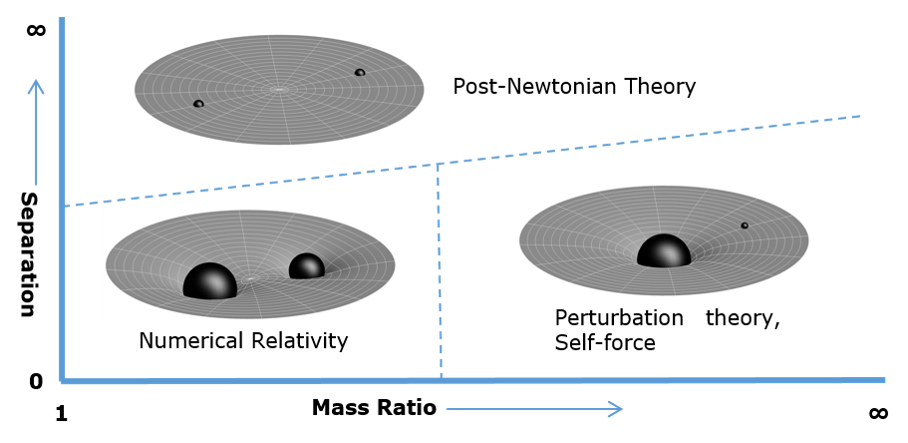}
    \caption{The various approaches currently used for modelling extreme mass ratio inspirals. Image credit: Timothy Rias.}
    \label{fig:corners}
\end{figure}

This huge accuracy requirement has motivated a community-scaled effort aimed at providing precise models of EMRIs within the framework of GR in the last decades. Achieving this task requires to solve the general relativistic two body problem in the strong field/small mass ratio regime. As depicted in Figure \ref{fig:corners}, various methods can be used for tackling this problem, depending on both the separation and the mass ratio between the objects. For EMRIs, the most adapted method is to treat the secondary compact object as a perturbation moving in the curved geometry generated by the supermassive black hole, modelled as a rotating black hole within GR. This amounts to apply black hole perturbation theory over a Kerr background, which is also known as \textit{self-force} theory in the literature, see \cite{Pound:2021qin} for a state of the art review.

The presence of a small mass ratio parameter $\epsilon\ll 1$ allows to perform the analysis perturbatively, order by order in $\epsilon$. The actual motion of the secondary withing the curved spacetime generated by the primary will take the form of a forced geodesic equation:
\begin{center}
    \begin{tikzpicture}
    \node[fill=green!10,rounded corners,draw=green,thick](A) at (-0.25,0) {\large$\frac{\text{D} z^\mu}{\dd\tau}=\epsilon\, f^\mu[g_{\alpha\beta},z^\alpha,T_{\alpha\beta}]$};
    \node[fill=blue!10,rounded corners, draw=blue,thick] (B) at (-5.5,-2) {\small$f^\mu=0$: geodesic motion};
    \node[fill=purple!10,rounded corners, draw=purple,thick] (C) at (-0.25,-2) {\small$f^\mu=f^\mu_\text{GSF}$: self-force corrections};
    \node[fill=red!10,rounded corners, draw=red,thick] (D) at (5.5,-2) {\small$f^\mu=f^\mu_\text{spin}$: finite size effects};
    \draw[ultra thick,->] (A.south)--(B.north);
    \draw[ultra thick,->] (A.south)--(C.north);
    \draw[ultra thick,->] (A.south)--(D.north);
    \end{tikzpicture}
\end{center}
which shall be solved consistently with the field equations for obtaining the evolution of the position of the secondary $z^\mu(\tau)$ over the whole inspiral, whose duration scales as $1/\epsilon$ (this timescale is known as the \textit{radiation reaction time}). 

Adapted techniques like the two timescale expansion \cite{Hinderer:2008dm} allow to solve perturbatively these equations. Actually, one can show that the orbital phase of the secondary schematically takes the form \cite{Hinderer:2008dm,Witzany:2018ahb}
\begin{align*}
\phi&=\phi^{(1)}_\text{avg}&&\mathcal O\qty(\epsilon^{-1})\\
&+\phi^{(1)}_\text{osc}+\phi^{(2)}_\text{avg}+\phi^{(1)}_\text{spin}&&\mathcal O\qty(\epsilon^{0})\\
&+\phi^{(2)}_\text{osc}+\phi^{(3)}_\text{avg}+\phi^{(2)}_\text{spin}&&\mathcal O\qty(\epsilon^{1})\\
&+\ldots
\end{align*}
The various contributions $\phi^{(i)}_\text{\ldots}$ to the right hand side of this equation (originating from the forcing terms $f^{\mu}_\text{\ldots}$ of the forced geodesic equation) have two main origins:
\begin{itemize}
    \item \textbf{Self-force corrections:} they are due to the self-interaction between the mass of the secondary and its own gravitational field. This is the main effect which has been addressed in the literature, see \cite{Pound:2021qin} and references therein, and the most difficult to treat. It is of prime importance, since it is its presence that will lead to the dissipative nature of the inspiral through the emission of gravitational waves. Meeting the LISA precision requirements ($\Delta\phi\sim 1\,\text{rad}$) requires both the knowledge of the first order conservative piece of the self-force (leading to the term $\phi^{(1)}_\text{osc}$), and to the first and second order averaged, dissipative parts of the self-force ($\phi^{(i)}_\text{avg}$).
    \item \textbf{Finite size effects corrections:} they originate from the non point-like nature of the secondary. The leading piece $\phi^{(1)}_\text{spin}$ of these corrections arises at $\mathcal O\qty(\epsilon^0)$ in the orbital phase, and is due to the spin (intrinsic angular momentum) of the compact body. As shown in \cite{Warburton:2017sxk}, discarding this term can lead to a dephasing of a few dozen of cycles over the whole inspiral.
\end{itemize}

Actually, at the level of the equations of motion of the secondary, both self-force and finite size corrections arise at the same order, if one consider astrophysically realistic secondaries. As will be detailed in Part \ref{part:spinnin_bodies} of this thesis, one can introduce a spin magnitude parameter $\mathcal S$ defined from the body's spin dipole tensor $S^{\mu\nu}$ as $\mathcal S^2\triangleq\frac{1}{2}S_{\mu\nu}S^{\mu\nu}$. One can show that the leading order spin force term then scales as
\begin{align*}
    \abs{f^\mu_\text{spin}}\sim \frac{\mathcal S}{\mu M}.
\end{align*}
However, a realistic astrophysical secondary will always spin at a rate lower than the one of a maximally spinning Kerr black hole\footnote{As we will see in Chapter \ref{chap:kerr}, this maximal bound on the spin is necessary to prevent the appearance of a naked singularity in spacetime.}, yielding $\mathcal S\leq \mu^2$. Gathering these results allows to write
\begin{align*}
    \abs{f^\mu_\text{spin}}\lesssim\epsilon.
\end{align*}
The leading forcing term describing finite size effects therefore scales at most as the leading self-force term. Moreover, from this short reasoning, one see that $\mathcal S$ itself can be formally used as a small expansion parameter.

In this thesis, we will discard self-force effects but account for the leading and the first subleading finite-size corrections, thus considering \textit{extended test bodies}. This approximation amounts to truncate the actual motion at zeroth order in the mass-ratio $\epsilon$ and at second order in the spin magnitude $\mathcal S$. Therefore, for timescales much shorter than the radiation-reaction time $1/\epsilon$, it provides a valid approximation of the motion. Moreover, even if the independent study of finite size effects does not directly provide the motion over the whole inspiral, its results can in principle be used to inform the self-forced motion \cite{Warburton:2017sxk,Pound:2021qin,Hinderer:2008dm,Witzany:2019nml}. 

The first subleading finite size correction (quadratic in the spin) is due to the quadrupole moment of the secondary, its study deserves some attention since it is the first term of the expansion where the nature of the compact object plays a role. Moreover, since they appear at 2PN order in the post-Newtonian expansion \cite{Kidder:1992fr,Kidder:1995zr,Will:1996zj,Gergely:1999pd,Mikoczi:2005dn,Racine:2008kj} (see \cite{Porto:2008jj,Bohe:2015ana} for the 3PN order and \cite{Kastha:2019brk} for a recent status) quadratic effects are generically relevant for gravitational waveform modeling of compact binaries.

Finally, let us notice that being able to consider finite-size effects as corrections added on the top of geodesic motion in Kerr spacetime (describing the central, supermassive black hole) will be of prime importance for most of the analytical computations of this thesis. This originates from the fact that the Kerr geometry exhibits a \emph{hidden symmetry}, not corresponding to any spacetime isometry, responsible for the integrability of geodesic motion and for the separability of various field equations. 

\subsubsection{How this thesis is organized}
The vast majority of the chapters of this thesis aims to provide a pedagogical exposition of the subject, and will be denoted with a (P) in the upcoming plan. Nevertheless, a few chapters, denoted with a (T), contain technical derivations and details which can be skipped during a first reading of this work. They are nevertheless included in order to give a feeling of the computational complexity lying behind some results provided in this work. Along the text, some computations were also performed or checked thanks to the software \textit{Mathematica}. The related notebooks are available on simple request.

Part \ref{part:geodesics} of the thesis is concerned with the zeroth order approximation for EMRIs motion, timelike geodesic motion in Kerr spacetime. Chapter \ref{chap:kerr} (P) reviews the main features of the Kerr metric and of its timelike geodesics, while Chapter \ref{chap:hamilton_kerr} (P) describes the associated Hamiltonian formulation. Finally, Chapter \ref{chap:classification} (T) describes the classification of the polar geodesic motion in generic Kerr spacetime, as well as the radial motion for near-horizon geodesics of extremal Kerr black holes.

In Part \ref{part:spinnin_bodies}, we depart from the geodesic approximation and turn on spin and finite-size effects, which are studied for a generic curved background. Chapter \ref{chap:EOM} (P) derives the equations of motion encompassing finite size effects, also known as Mathisson-Papapetrou-Dixon (MPD) equations. Chapter \ref{chap:skeleton} (P) discuss the physical significance of these equations from another point of view, namely the gravitational skeletonization, which amounts to replace the smooth compact body by a worldline endowed with a collection of multipoles. Chapter \ref{chap:SSC} (P) describes the necessity of supplementing the MPD equations with supplementary algebraic conditions for obtaining a closed set of equations. These conditions are known as \textit{spin supplementary conditions} and can be understood making some specific choice for the worldline upon which the multipole moments of the body are defined. Finally, Chapter \ref{chap:quadrupole} (P) aims at deriving the structure of the quadrupole moment in the case where it is only induced by the proper rotation of the test body, which will be the approximation used for the remaining of the thesis (thus discarding the presence of any tidal-type effects).

As any dynamical system, the understanding of the motion of finite size test bodies will be enormously facilitated if one is able to find quantities which are conserved along the motion. This is the core of the present thesis, to which Part \ref{part:conserved_quantities} is devoted. Chapter \ref{chap:building} (P) aims at introducing a generic procedure for building conserved quantities directly from conservation constraint equations. Conserved quantities for MPD equations in Kerr spacetime at first and second order in the spin magnitude are respectively investigated in Chapters \ref{chap:first_order_solutions} and \ref{chap:second_order_solution} (T). A readable summary of the results (P) can be found in the specific introduction to Part \ref{part:conserved_quantities}.

The thesis ends by a rough discussion of a covariant Hamiltonian formulation of MPD equations, valid at quadratic order in $\mathcal S$. Chapter \ref{chap:symplectic} (P) introduces the related phase space, Poisson brackets structure and symplectic coordinates relevant for coping with the problem. Chapter \ref{chap:covariant_H} (P) deals with the construction of the Hamiltonian. The two last chapters discuss two applications of the Hamiltonian formulation: Chapter \ref{chap:integrability} (P) is devoted to the non-integrability of MPD equations in Kerr spacetime, while Chapter \ref{chap:HJ} reviews the solution of the associated Hamilton-Jacobi equation, and its relation with the constants of motion obtained in Part \ref{part:conserved_quantities}.

\subsubsection{Personal contributions}

The original contributions presented in this thesis are based upon the three following publications \cite{Compere:2020eat,Compere:2021kjz,Compere:2023alp}:
\begin{itemize}
    \item G.~Comp\`ere and A.~Druart, ``Near-horizon geodesics of high-spin black holes,'' Phys. Rev. D \textbf{101} (2020) no.8, [erratum: Phys. Rev. D \textbf{102} (2020) no.2, 029901] [arXiv:2001.03478 [gr-qc]]: complete classification and discussion of the physical properties of both generic Kerr polar geodesic motion and near-horizon extremal Kerr radial geodesic motion.
    \item G.~Comp\`ere and A.~Druart, ``Complete set of quasi-conserved quantities for spinning particles around Kerr,'' SciPost Phys. \textbf{12} (2022) no.1 [arXiv:2105.12454 [gr-qc]]: investigation of the conserved quantities for MPD equations at linear order in the spin magnitude, including the proof of uniqueness of Rüdiger's quadratic invariant in Kerr spacetime and discussion of the non-integrability of linearized MPD equations in Kerr spacetime. 
    \item G.~Comp\`ere, A.~Druart and J.~Vines, ``Generalized Carter constant for quadrupolar test bodies in Kerr spacetime,'' [arXiv:2302.14549 [gr-qc]]: generalization of Rüdiger's deformed Carter constant to quadratic order in the spin magnitude.
\end{itemize}

Moreover, Part \ref{part:hamilton} also contains a few results that do not appear in the literature, at least to our knowledge: the covariant Hamiltonian for spin-induced quadrupole MPD equations \eqref{HQ_offshell} and the discussion about the second order swing region solution to the associated Hamilton-Jacobi equation of Section \ref{sec:HJ:second}.

\subsubsection{Conventions}
All along this text, we stay within the realm of General Relativity, choosing to follow the conventions of \cite{carroll2003spacetime}. We will always consider a 4d Lorentzian manifold equipped with a metric $g_{\mu\nu}$. The metric signature is chosen to be $(-+++)$. Unless otherwise stated, lowercase Greek indices run from $0$ to $3$ and denote spacetime indices. Lowercase Latin indices denote purely spatial indices and run from $1$ to $3$. Uppercase Latin indices represent tetrad indices. The Einstein summation convention is used everywhere. $\nabla_\alpha$ denotes the covariant derivative with respect to the Levi-Civita connection, the Riemann tensor is defined such that $\comm{\nabla_\alpha}{\nabla_\beta}A_\mu=-R\tud{\lambda}{\mu\alpha\beta}A_\lambda$ and the Ricci tensor is $R_{\mu\nu}=R\tud{\lambda}{\mu\lambda\nu}$. Covariant derivatives will sometimes be denoted by a semicolon, while a comma might be used for partial derivatives.

\newpage

\phantom{a}

\newpage

\pagestyle{plain}

\pagenumbering{arabic}

\partimage[width=\textwidth]{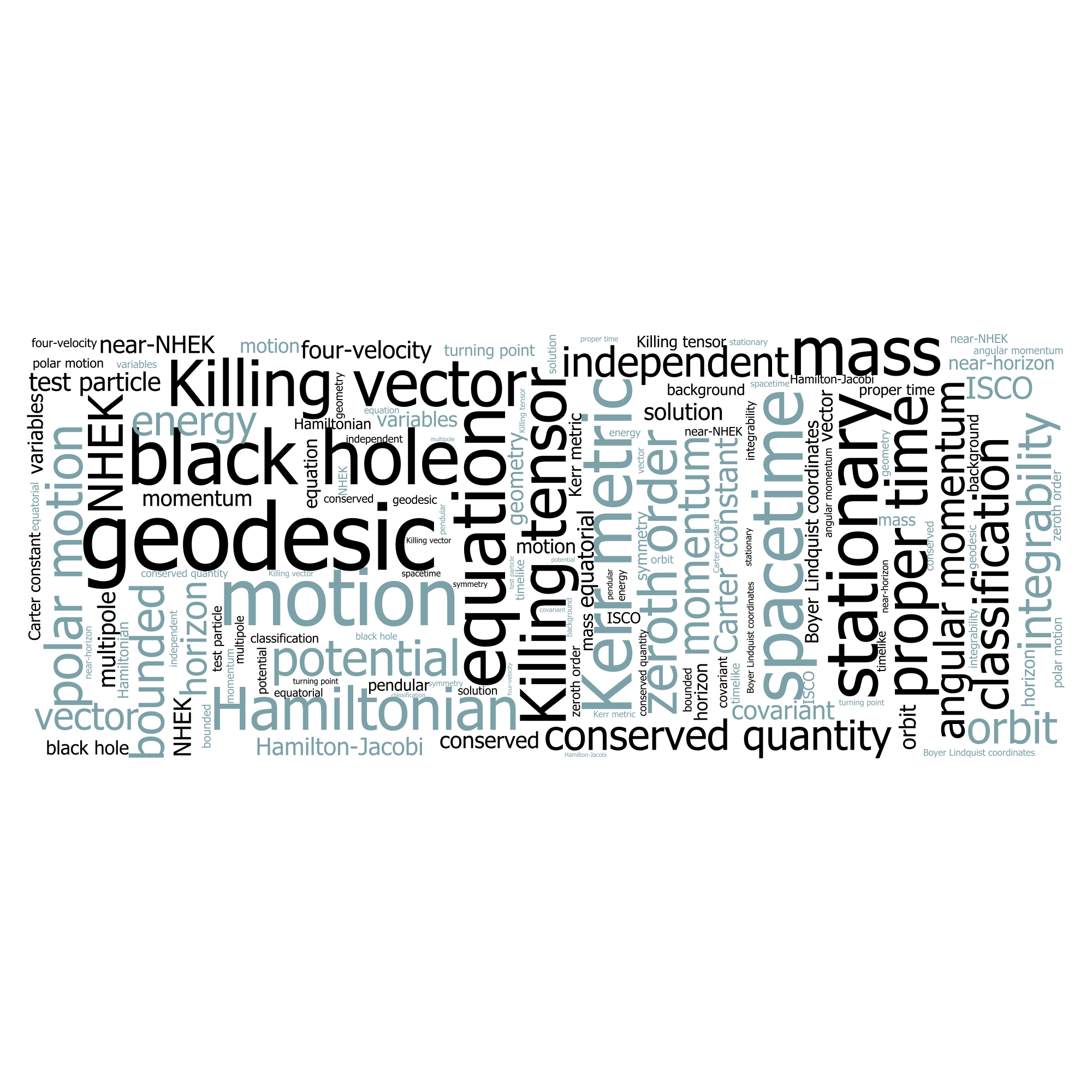}

\part[{\textsc{Geodesic Motion in Kerr Spacetime}}]{\textsc{Geodesic Motion\\in Kerr Spacetime}}
\label{part:geodesics}

{\renewcommand{\thefigure}{I.1}
\begin{figure}[h!]
    \centering
    \begin{tabular}{ccc}
       \includegraphics[height=6.5cm]{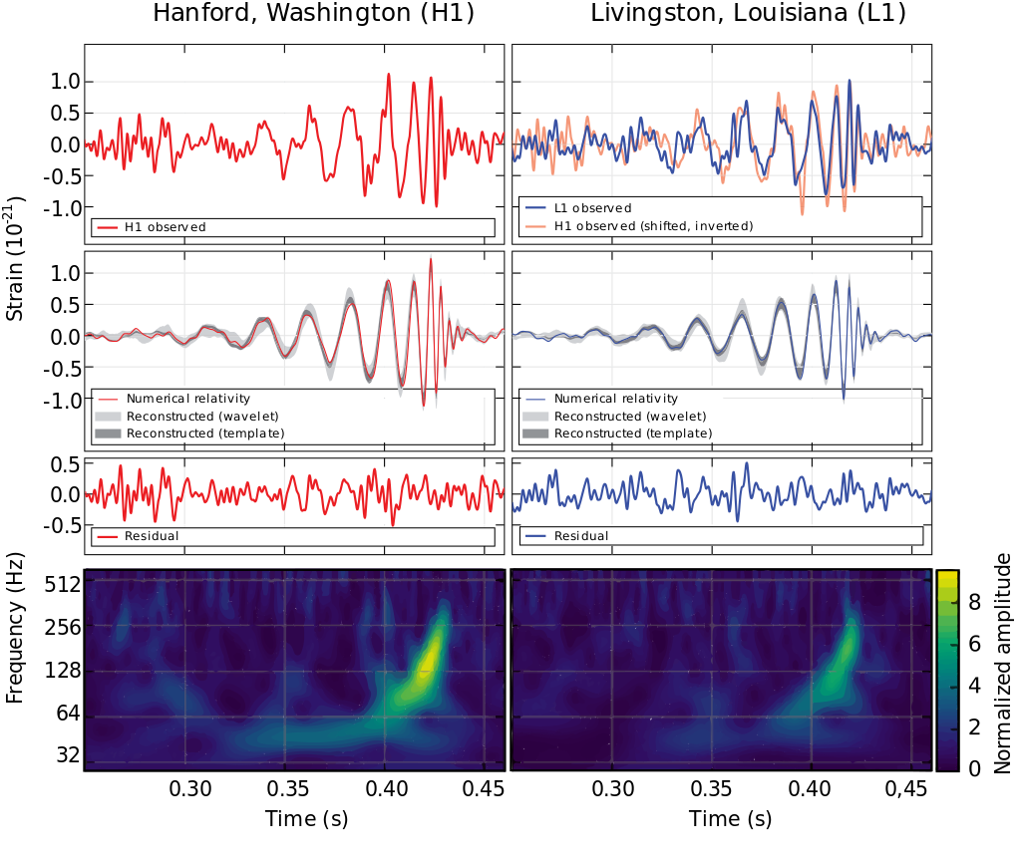}  & \hspace{1cm}&\includegraphics[height=6.5cm]{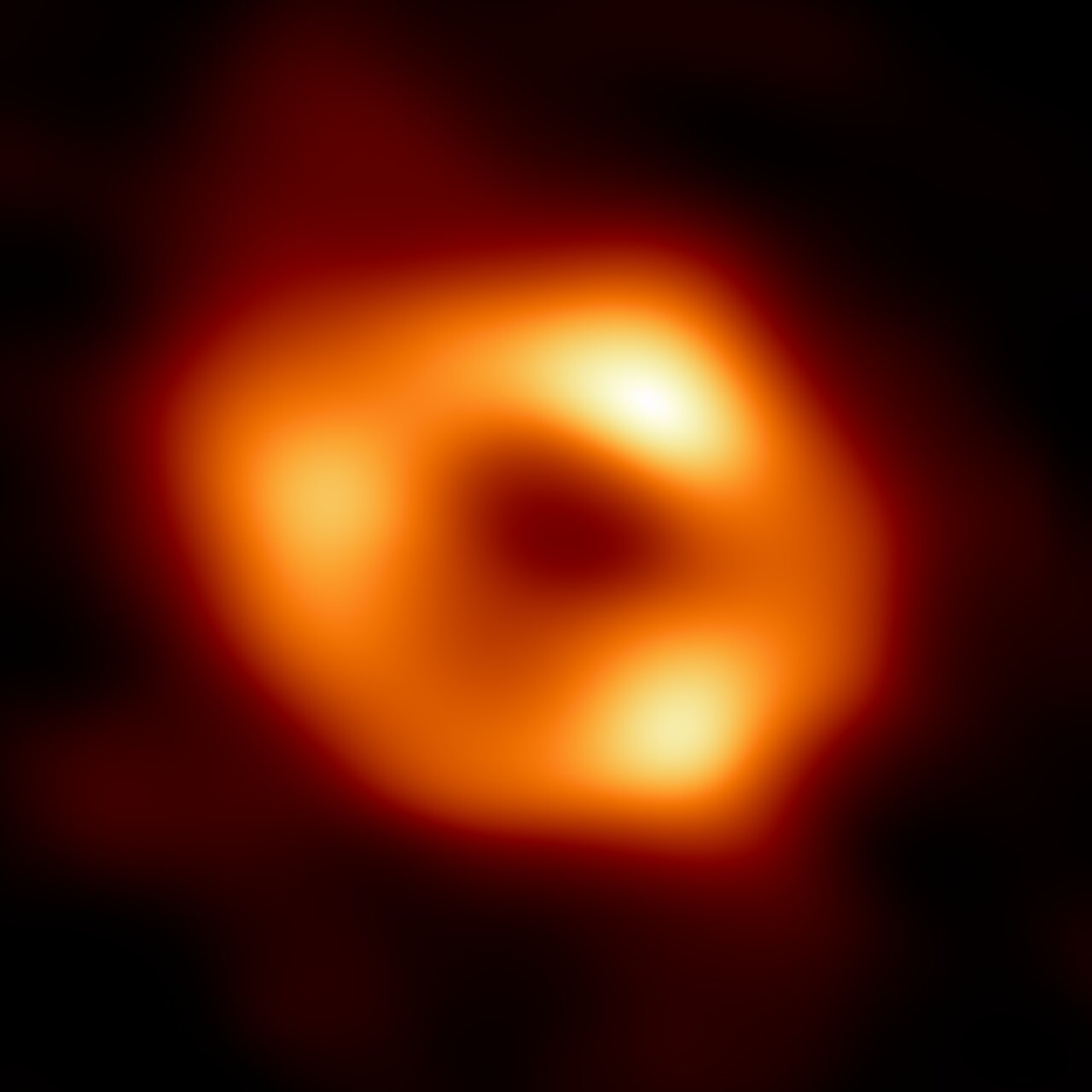} \\
        (a) & & (b)
    \end{tabular}
    \caption{Kerr black holes in the sky? (a) GW150914, the first detection of the coalescence of a binary black hole system by the LIGO collaboration \cite{collinson1976}; (b) the shadow of Sagittarius A$^*$, the supermassive black hole lying at the center of our galaxy by the Event Horizon Telescope collaboration \cite{EventHorizonTelescope:2022wkp}.}
    \label{fig:kerr_in_the_sky}
\end{figure}
}

\lettrine[lines=2]{T}{his} first part will be devoted to the study of geodesic motion in Kerr spacetime. Chapter \ref{chap:kerr} will describe the basics features of Kerr metric, discuss its geodesic equations and provide their formal solutions. We then turn to the discussion of geodesic motion in the near-horizon region of highly spinning black holes, which exhibits an enhanced symmetry group. Chapter \ref{chap:hamilton_kerr} then examines again the problem of geodesic motion in Kerr spacetime, tackling it from the perspective of Hamiltonian mechanics. This point of view will enable us to obtain in an easy way some fundamental results, including (i) the proof of complete integrability of Kerr geodesic motion, (ii) the derivation of the standard form of Kerr geodesic equations through solving the associated Hamilton-Jacobi equation and (iii) the formulation of equations of motion in terms of action-angle variables. Finally, in Chapter \ref{chap:classification}, we will establish an exhaustive and comprehensive classification of polar geodesic motion in generic Kerr spacetime, as well as radial motion in the near-horizon region of high spin Kerr black holes. Explicit solutions to the geodesic equations will also be provided. Before we start, this introduction briefly summarizes the history and the state of the art of each of these topics.

\subsubsection{Kerr geometry}

The exact solutions to Einstein Equations describing static black holes (namely, the Schwarzschild and Reissner-Nordström metrics) were discovered very quickly after the birth of General Relativity in 1915. However, it took nearly fifty years to discover a solution representing a rotating black hole, more accurate for depicting the astrophysical black holes that surrounds us. This solution is the Kerr metric, due to the New Zealand mathematician R. Kerr in 1963 \cite{Kerr:1963ud}. It would be too long to include here a detailed review of the discovery of this solution and of all the subsequent developments it has generated. The interested reader should consult fruitfully the personal reminiscences of R. Kerr himself \cite{Kerr:2004zz,Kerr:2007dk}, as well as the very nice historical and scientific review written by S. Teukolsky \cite{Teukolsky:2014vca}.

The exceptional features of Kerr solution come from the existence of uniqueness theorems \cite{Teukolsky:2014vca} that give a fundamental importance to Kerr metric within the realm of General Relativity. Actually, it is the most generic solution of Einstein Equations in vacuum which is stationary and asymptotically flat. Kerr metric is therefore the stationary state predicted by General Relativity towards which are believed to tend all the massive stars collapsing into black holes, as well as the mergers of extremely compact objects (black holes and/or neutrons stars) that can be observed in our Universe, see Fig. \ref{fig:kerr_in_the_sky}. It therefore provides an unique possibility for testing Einstein's theory within the strong field regime. 

\subsubsection{Geodesics in Kerr spacetime}

There are at least two motivations for studying geodesic motion in Kerr spacetime. First, null geodesics (together with the modeling of light sources) underpin the field of black hole imaging \cite{Bardeen:1972fi,Luminet:1979nyg,Falcke:1999pj,Vazquez:2003zm,James:2015yla,Luminet:2019hfx,Gralla:2019xty,Gralla:2019drh}, which has recently become an observational science \cite{Akiyama:2019cqa}. Second, as already discussed in the main introduction, timelike geodesics provide the zeroth-order motion of binary systems in the perturbative small-mass-ratio expansion, which leads in the adiabatic approximation to the leading-order gravitational waveforms of extreme-mass-ratio inspirals (EMRIs) \cite{1973ApJ...185..635T,Sasaki:1981sx,Ryan:1995zm,Finn:2000sy}. As we have seen, they are two directions in which this zeroth order geodesic motion can be refined in order to obtain a more accurate description of EMRIs dynamics: including gravitational self-force effects \cite{poisson_2004,Pound:2015tma,Pound:2021qin}, which are responsible for the dissipative character of the dynamics, finally leading to the coalescence of the system. Another refinement appearing at the same order in the small mass ratio expansion is the inclusion of corrections due to the finite size nature of the secondary. In realistic situations, the latter is not a point particle but a compact object that possesses some internal structure, which can be accurately described by the means of a tower of multipoles (mass, spin, quadrupole\ldots). In astrophysically realistic situations, the appearance of this structure can be modelled by adding perturbative corrections to the geodesic dynamics. Because this subject will be the main focus of the subsequent parts of this thesis, a deep understanding of geodesic motion in Kerr spacetime appears to be of prime importance.

The study of timelike and null geodesics of the Kerr metric has a long history which is still ongoing \cite{carter68,Wilkins:1972rs,Bardeen:1972fi,Chandrasekhar:1984siy,Rauch:1994aa,Neill:1995aa,Schmidt:2002qk,Mino:2003yg,Vazquez:2003zm,Kraniotis:2005zm,Dexter:2009fg,Fujita:2009bp,Kraniotis:2010gx,Hackmann:2015vla,Porfyriadis:2016gwb,Compere:2017hsi,Kapec:2019hro,Gralla:2019ceu,Rana:2019bsn,Stein:2019buj,Teo:2020sey,Compere:2020eat,Compere:2021bkk}, and is briefly summarized in the introduction of \cite{Compere:2021bkk}. Two remarkable benchmarks are the discovery of the fourth constant of motion by B. Carter in 1968 \cite{carter68}, which allowed the separation of the geodesic equations and the derivation of closed-form analytical solutions for bound geodesic motion by R. Fujita and W. Hikida \cite{Fujita:2009bp}.

In recent years, the community's efforts were mostly directed towards topics directly useful for the modelling of extreme mass-ratio inspirals (EMRIs) dynamics, such as understanding the location of the innermost stable spherical orbit (ISSO) in Kerr spacetime \cite{Stein:2019buj,Compere:2020eat} or obtaining analytic expressions for plunging geodesics \cite{Dyson:2023fws}. Another program that has been completed recently is the construction of a classification of all the possible forms that geodesic motion can exhibit in Kerr spacetime. This classification was initiated by the one of near-horizon geodesics \cite{Compere:2017hsi,Kapec:2019hro,Compere:2020eat}. It has succeeded into providing a complete classification of the polar geodesic motion \cite{Kapec:2019hro,Compere:2020eat} and, more recently, of generic radial motion \cite{Compere:2021bkk}.

\subsubsection{Action-angle formulation} 
Kerr bounded geodesic motion is well-known to be tri-periodic, in its radial, polar and azimutal directions. A formulation that directly reflects this behaviour is the action-angle formalism of Hamiltonian mechanics, which was applied to Kerr spacetime for the first time in the early 2000's \cite{Schmidt:2002qk}. The unfamiliar reader will find an introduction to action-angle formalism in the classical textbook of Goldstein \cite{Goldstein2001}, together with a self-contained introduction to Hamiltonian mechanics and Hamilton-Jacobi theory. Arnold's classical text \cite{arnold1989mathematical} reviews the fundamentals of the symplectic formulation of Hamilton's mechanics and provides the proof of the central result known as the \textit{Liouville-Arnold theorem}. 
        
However, as will be review in Chapter \ref{chap:hamilton_kerr}, the non-compactness of level sets for geodesic motion in Kerr disables us to use directly these results. The generalized theorem for setting up action-angle formalism for non-compact level sets was developed in \cite{Fiorani_2003}. A brief description of this result in the context of EMRIs evolution (on which the present review is heavily based) can be found in E. Flanagan and T. Hinderer's classical paper \cite{Hinderer:2008dm}. Action-angle variables description of Kerr geodesics was initiated by W. Schmidt who derived explicit expressions for the fundamental frequencies and the action variables \cite{Schmidt:2002qk}. Coherent mathematical foundations were provided later by E. Flanagan and T. Hinderer \cite{Hinderer:2008dm}. This formulation of Kerr bounded geodesic motion lies at the heart of the two timescale analysis for describing self-forced EMRIs dynamics \cite{Hinderer:2008dm}. See also \cite{Pound:2021qin} for a recent review of the field.

\subsubsection{Near-horizon geometries}

The spin $a$ of a Kerr black hole admits a maximal bound, given by its mass $M$ ($a^2\leq M^2$). In the extremely high-spin limit, a throat-like geometry possessing a conformal $\textsf{SL}(2,\mathbb R)$ symmetry shows up close to the event horizon of the hole \cite{Bardeen:1999px,Amsel:2009ev,Dias:2009ex,Bredberg:2009pv}. This leads to the appearance of very peculiar physics in this region. The first studies of geodesic motion in the high-spin near-horizon region were restricted either to equatorial orbits \cite{Porfyriadis:2014fja,Hadar:2014dpa,Hadar:2015xpa,Gralla:2015rpa,Hadar:2016vmk,Compere:2017hsi,Hod:2017uof}, to specific orbits \cite{AlZahrani:2010qb}, or to parametrically generic geodesics \cite{Kapec:2019hro} that discard relevant measure-zero sets in parameter space such as the separatrix between bound and unbound motion. Complete classifications of geodesics in the high-spin near-horizon Kerr region were then obtained in recent years, as well as the explicit solutions of the related equations of motion \cite{Kapec:2019hro,Compere:2020eat}.  

It is known since many years that any null orbit that enters or leaves the near-horizon region has a polar motion bounded by the minimal angle $\cos^2\theta_{\text{\text{min}}} = 2 \sqrt{3}-3$ ($47^\circ \lessapprox \theta \lessapprox 133^\circ$), which corresponds to the polar inclination of the velocity-of-light surface in the near-horizon and high-spin limit \cite{AlZahrani:2010qb,Porfyriadis:2016gwb}. This property was recently proven to hold also for any timelike geodesic \cite{Compere:2020eat}. The polar motion is more restricted for the innermost bound spherical orbits (IBSOs): $\cos^2\theta_{\text{\text{min}}} = 1/3$ ($55^\circ \lessapprox \theta \lessapprox 125^\circ$) \cite{Hod:2017uof,Stein:2019buj} and even more restricted for the innermost stable spherical orbits (ISSOs): $\cos^2\theta_{\text{\text{min}}} = 3-2 \sqrt{2}$ ($65^\circ \lessapprox \theta \lessapprox 115^\circ$), as independently shown in \cite{Stein:2019buj,Compere:2020eat}.

Conformal symmetry in the near-horizon high spin Kerr geometry leads to potentially observable signatures if such high spin black holes are realized in nature. The behavior of null geodesics on the image of an extremely spinning Kerr black hole leads to the NHEKline \cite{Bardeen:1972fi,Gralla:2017ufe} and to specific polarization whorls \cite{Gates:2018hub}. Gravitational waveforms on adiabatic inspirals lead to exponentially decaying tails at fixed oscillation frequencies with amplitudes suppressed as $(1-a^2/M^2)^{1/6}$ \cite{Porfyriadis:2014fja,Gralla:2016qfw}, while plunging trajectories lead to impact-dependent polynomial quasinormal ringing with a power ranging from inverse time to square root of inverse time \cite{Compere:2017hsi}. It was shown in \cite{Compere:2017hsi} that conformal symmetry together with a discrete symmetry leads to equivalence classes of equatorial timelike geodesics with circular orbits as distinguished representatives. This allows to simplify the computation of Teukolsky waveforms by applying conformal transformations to the seed circular waveform \cite{Hadar:2014dpa,Compere:2017hsi}. In \cite{Compere:2020eat}, we showed that -- in whole generality -- conformal symmetry and discrete symmetries lead to equivalence classes which each admit spherical orbits as distinguished representatives. For orbits with an angular momentum lower than the ISSO one, no spherical orbit exists but a ``complex spherical orbit'' exists that generates the equivalence class. Such a complex spherical orbit can be used as a seed and complexified conformal transformations allow to reach all real subcritical geodesics. 

\chapter[Kerr Geometry and its Geodesics]{Kerr Geometry\\and its Geodesics}\label{chap:kerr}


The main goal of this chapter is to set the stage for all the forthcoming developments of this thesis, by reviewing the main features of Kerr spacetime and of its geodesics and setting on the fly numerous notations and conventions that will be extensively used latter on. We will end by discussing a side topic, the near-horizon limits of (near-)extremal Kerr black holes, whose geodesic radial classification shall deserve an extended analysis in Chapter \ref{chap:classification}.

\section{Kerr metric and its main features}

One of the most widely used coordinate system for expressing the Kerr metric \cite{Kerr:1963ud} are the   \defining{Boyer-Lindquist coordinates} $(t,r,\theta,\varphi)$ in which Kerr's solution reads \cite{Wald:1984rg,Misner1973,carroll2003spacetime}
\boxedeqn{
      \begin{split}
      \dd s^2&=-\frac{  \Delta(r) }{\Sigma(r,\cos\theta)}\qty(  \dd t-a\sin^2\theta  \dd\varphi)^2+\Sigma(r,\cos\theta)\qty(\frac{  \dd r^2}{  \Delta(r) }+  \dd\theta^2)\\
      &\quad+\frac{\sin^2\theta}{\Sigma(r,\cos\theta)}\qty[\qty(r^2+a^2)  \dd\varphi-a  \dd t]^2,\label{kerr_metric}
      \end{split}
    }{Kerr metric}
with
\begin{align}
      \Delta(r) \triangleq r^2-2Mr+a^2,\qquad \Sigma(r,\cos\theta)\triangleq r^2+a^2\cos^2\theta.
\end{align}
Eq. \eqref{kerr_metric} actually provides us with a two-parameters family of solutions to the Einstein equations, that we will denote $\textsf{Kerr}\qty(M,a)$. Using \textit{e.g.} Komar integrals, one can show that $M$ can be interpreted as the mass of the black hole, whereas $a$ is its angular momentum per unit of mass, $a=J/M$ \cite{carroll2003spacetime}. The interested reader can find a readable derivation of the metric Eq. \eqref{kerr_metric} in Carter's contribution to \textit{Les Houches} proceedings \cite{Carter:1973rla}.

For numerical symbolic evaluation, it is often useful to get rid of the trigonometric functions appearing in Eq. \eqref{kerr_metric} by using the variable\footnote{Actually, depending on the context, we will sometimes define $z\triangleq\cos\theta$ or $z\triangleq\cos^2\theta$. The former is well-suited when dealing with the equations of motion, which are better expressed in terms of $\cos\theta$ because it allows to keep track of the direction $\pm_\theta$ of the motion, whereas the latter is mostly useful for establishing the classification of polar geodesic motion, since the corresponding potential is a function of $\cos^2\theta$.} $z\triangleq\cos\theta$ instead of $\theta$. In terms of the coordinates $(t,r,z,\varphi)$, the non-vanishing components of the Kerr metric read
\begin{align}
g_{tt}&=-\frac{  \Delta(r) }{\Sigma(r,z)},\qquad
g_{rr}=\frac{\Sigma(r,z)}{  \Delta(r) },\qquad
g_{zz}=\frac{\Sigma(r,z)}{1-z^2},\nonumber\\
g_{\varphi\varphi}&= \frac{1-z^2}{\Sigma(r,z)}\qty[\qty(r^2+a^2)^2-a^2  \Delta(r) \qty(1-z^2)],\quad
g_{t\varphi}=g_{\varphi t}=\frac{2aMr\qty(z^2-1)}{\Sigma(r,z)}
\end{align}
whereas the components of the inverse metric are
\begin{align}\label{inverse_metric}
    g^{tt}&=-1-\frac{2Mr\qty(r^2+a^2)}{\Sigma(r,z)  \Delta(r) },\qquad
    g^{rr}=\frac{  \Delta(r) }{\Sigma(r,z)},\qquad
    g^{zz}=\frac{1-z^2}{\Sigma(r,z)},\nonumber\\
    g^{\varphi\varphi}&=\frac{1}{\Sigma(r,z)}\qty(\frac{1}{1-z^2}-\frac{a^2}{\Delta(r)}),\qquad
    g^{t\varphi}=g^{\varphi t}=-\frac{2aMr}{\Sigma(r,z)\Delta(r)}.
\end{align}
We finally notice that the determinant of the metric is simply
\begin{align}
    g\triangleq  \det g_{\mu\nu}=-\Sigma(r,z)=-\qty(r^2+a^2z^2).
\end{align}

\subsection{Existence and uniqueness of Kerr solution}

The power of the Kerr solution is mathematically stated through the \defining{Carter-Robinson theorem} \cite{Mazur:1986ki}:
\begin{theorem}[Carter-Robinson]
Any asymptotically flat, static and axisymmetric solution to vacuum Einstein equations which is non-singular on and outside its event horizon is a member of $\textsf{Kerr}\qty(M,a)$.
\end{theorem}

The strength of this result can be further enhanced by convoking \textit{Hawking rigidity theorem}: under certain assumptions on the matter fields present in spacetime (and which are verified in vacuum), any stationary solution of Einstein equations is axisymmetric \cite{Hawking:1971vc}. The combination of these two results has stringent consequences if we stay within the realm of General Relativity: Kerr metric is the most generic stationary, asymptotically flat solution to vacuum Einstein Equations that is non-singular on and outside its event horizon. Macroscopic astrophysical objects being expected to be electromagnetically neutral, Kerr solution is therefore largely believed to be the most generic stationary state reached after the gravitational collapse of enough massive stars or the merging of extremely compact objects (black holes and/or neutron stars). Notice that this statement is not a direct consequence of the uniqueness theorems discussed above, a large amount of physics shall be convoked for reaching this conclusion \cite{Misner1973,carroll2003spacetime}.

This prominent status of the Kerr solution largely accounts for the huge, commu-nity-scaled effort put in understanding its properties since its discovery, and corroborates the central role that its study will deserve in the present thesis: in studying the motion of bodies around Kerr black holes, we are \textit{really} saying something about the astrophysical systems involving black holes that surround us.
Another striking feature of Kerr solution is that it has ``no hair'': the entire geometry of the spacetime (that is, the entire gravitational field of the black hole) is only characterized by two parameters, its mass and its spin. Any multipole expansion of the spacetime will lead to a collection of multipole moments whose value is entirely fixed by $M$ and $a$ \cite{Hansen:1974zz}. In our analysis, this will also be what happens when we will consider the test body as a Kerr black hole itself: the form and the coupling strength of its high order multipole moments (quadrupole and higher) will be entirely determined by the knowledge of its linear momentum (monopole) and spin (current-type dipole), see Chapter \ref{chap:quadrupole}.

\subsection{Generic properties of Kerr spacetime}

Before going further, let us have a look at some generic features of the Kerr metric. A first point of interest is to get a bit more intuition about the meaning of Boyer-Lindquist coordinates. In the $M\to 0$ limit, Kerr metric \eqref{kerr_metric} reduces to flat (Minkowski) spacetime expressed in ellipsoidal coordinates \cite{carroll2003spacetime}. Moreover, in the spinless limit $a\to 0$, the geometry reduces to Schwarzschild solution (that is, $\textsf{Kerr}(M,0)=\textsf{Schwarzschild}(M)$) and the Boyer-Lindquist coordinates reduce to standard Schwarzschild ones:
\begin{align}
      \dd s^2\stackrel{a\to 0}{\longrightarrow}-f(r)  \dd t^2+\frac{  \dd r^2}{f(r)}+r^2  \dd\Omega^2,\label{schwarzschild_metric}
\end{align}
with $f(r)\triangleq 1-\frac{2M}{r}$ and $  \dd\Omega^2\triangleq  \dd\theta^2+\sin^2\theta  \dd\varphi^2$. Therefore, $t$ takes the interpretation of the proper time of an asymptotically far away observer, $r$ is the radius to the origin of the spacetime and $(\theta,\varphi)$ are respectively polar and azimutal angles.

As it has been already stated numerous times, Kerr spacetime is a black hole spacetime. It is characterized by two event horizons located at radial distances
\begin{align}
    r_\pm\triangleq M\pm\sqrt{M^2-a^2}\label{event_horizons}
\end{align}
which are known as the \textit{outer} ($r=r_+$) and the \textit{inner} ($r=r_-$) horizons, since $r_+\geq r_-$. In the continuation of this thesis, we shall only be concerned with phenomena occurring in the \textit{exterior} Kerr spacetime, \textit{i.e.} phenomena taking place in the region of spacetime located outside of the outer event horizon, at $r>r_+$.

Finally, in order to avoid the appearance of a naked singularity (and thus to prevent violating the \textit{Cosmic Censorship Conjecture}), the outer horizon radius described by Eq. \eqref{event_horizons} should be a real number. This enforces the magnitude of the black hole spin to be bounded by its mass,
\begin{align}
    \abs{a}\leq M.
\end{align}
The special case of a maximally spinning black hole (also referred to as an \defining{extremal black hole}) $\abs{a}=M$ deserves special attention, since the structure of the spacetime will dramatically change close to the event horizon of the hole, as it will be extensively discussed in Section \ref{sec:near_horizon}.

\subsection{Symmetries: Killing vectors and tensors}

In this thesis, one of our main concerns with respect to the Kerr metric is the study of its symmetries. This is their existence that will allow to solve the equations of motion for test bodies. They are also closely related to the existence of conserved quantities along the motion and to its integrability properties, as will be extensively discussed in the continuation of this text.

As always in differential geometry, isometries of the spacetime will be encoded under the form of Killing vectors. These symmetries are \textit{explicit} symmetries of the metric. However, Kerr geometry also possesses other symmetries -- referred to as \defining{hidden symmetries} -- which are generated by higher rank Killing-type tensorial objects. Even if these symmetries are not ``true'' symmetries of the metric, they play a fundamental role for studying the motion in Kerr spacetime since they also provide us with conserved quantities and are strongly related to the separability of the Hamilton-Jacobi equation. Table \ref{tab:killing_quantities} summarizes the defining equations for Killing vectors, Killing and Killing-Yano rank 2 tensor as well as their conformal counterparts. These objects are not all independent one of another, but the existence of one type of Killing object often imply the existence of others. Actually,
\begin{proposition}\label{prop:killing_quantities}
The following statements hold:
\begin{enumerate}
    \item If $Y_{\mu\nu}$ is a Killing-Yano tensor, then its Hodge dual $Y^*_{\mu\nu}\triangleq\frac{1}{2}\epsilon_{\mu\nu\rho\sigma}Y^{\rho\sigma}$ is a conformal Killing-Yano tensor.
    \item If $Y_{\mu\nu}$ is a Killing-Yano tensor, then $K_{\mu\nu}\triangleq Y\tdu{\mu}{\lambda}Y_{\nu\lambda}$ is a Killing tensor.
    \item If $Y_{\mu\nu}$ is a conformal Killing-Yano tensor, then $K_{\mu\nu}\triangleq Y\tdu{\mu}{\lambda}Y_{\nu\lambda}$ is a conformal Killing tensor.
    \item If $Y_{\mu\nu}$ is a Killing-Yano tensor or a conformal Killing-Yano tensor, then $K_{\mu\nu}\triangleq Y\tdu{\mu}{\lambda}Y^*_{\nu\lambda}$ is a conformal Killing tensor.
    \item In Ricci-flat spacetimes ($R_{\mu\nu}=0$), if $Y_{\mu\nu}$ is a conformal Killing-Yano tensor, then $\xi^\mu=-\frac{1}{3}\nabla_\lambda Y^{\lambda\mu}$ is a Killing vector.
\end{enumerate}
\end{proposition}
The proofs of some of these assertions -- as well as many other properties of spacetimes admitting Killing-Yano tensors -- will be provided in Section \ref{sec:KY}. We now list the various Killing-type objects that appear in Kerr spacetime.

\subsubsection{Explicit symmetries} Since Kerr spacetime is stationary, it admits a timelike Killing vector
\begin{align}
    \xi\triangleq\partial_t.
\end{align}
Its axisymmetric character leads to the existence of another Killing vector, which is simply
\begin{align}
    \eta\triangleq\partial_\varphi.
\end{align}
Notice that there exists also a discrete $\mathbb Z_2$ symmetry -- that we shall refer to as the $\uparrow\!\downarrow$-flip -- which flips the sign of the time and the azimutal coordinates:
\begin{align}
    \uparrow\!\downarrow:\qquad t\to-t,\qquad\varphi\to-\varphi.
\end{align}
Physically, it corresponds to ``rewind'' the movie in time, thus considering a black hole of opposite spin flowing backwards in time. This symmetry will play an important role when we will consider conformal mapping between the near horizon geodesics of highly spinning Kerr black holes, see Chapter \ref{chap:classification}.

\subsubsection{Hidden symmetries}
In Petrov's algebraic classification \cite{Petrov:2000bs}, Kerr spacetime turns out to be of type D, thus possessing two distinct, doubly degenerated, principal null directions that we will denote $l^\mu$ and $n^\mu$. Building a null orthonormal tetrad by adding two (self-conjugated) complex null directions $m^\mu$ and $\bar m^\mu$, the metric can be written in terms of the null tetrad legs as
\begin{align}
    g_{\mu\nu}=-2 l_{(\mu}n_{\nu)}+2 m_{(\mu}\bar m_{\nu)}.\label{null_tetrad_decomposition}
\end{align}
A common choice for the null tetrad is \defining{Kinnersley tetrad} \cite{Kinnersley:1969zza}, which reads, in Boyer-Lindquist coordinates
\begin{subequations}\label{kinnersley}
\begin{align}
    l^\mu&{=}\frac{1}{\Delta}\qty(r^2+a^2,\Delta,0,a),\qquad
    n^\mu{=}\frac{1}{2\Sigma}\qty(r^2+a^2,-\Delta,0,a),\\
    m^\mu&{=}\frac{1}{\sqrt{2}\mathcal R}\qty(ia\sin\theta,0,\sqrt{1-z^2},\frac{i}{\sin\theta}).
\end{align}
\end{subequations}
Here, the scalar quantity $\mathcal R$ is defined as\footnote{Notice that $\mathcal R$ is directly related to the only non-vanishing Kerr's Weyl scalar $\Psi_2$ by $\Psi_2=-\frac{M}{\mathcal{\bar{R}}^3}$.}
\begin{align}
    \mathcal R\triangleq r+ia\cos\theta.
\end{align}

Kerr spacetime admits a Killing-Yano tensor which is advantageously written in terms of the null directions as
\begin{align}
    Y_{\mu\nu}=i\qty(\mathcal R-\bar{\mathcal{R}})l_{[\mu}n_{\nu]}-i\qty(\mathcal R+\bar{\mathcal{R}})m_{[\mu}\bar m_{\nu]}.
\end{align}
In Boyer-Lindquist coordinates, this becomes
\boxedeqn{
    \frac{1}{2}Y_{\mu\nu}   \dd x^\mu \wedge   \dd x^\nu &= a \cos\theta   \dd r \wedge (  \dd t - a \sin^2\theta   \dd \varphi)\nonumber\\
    &\quad+  r \sin\theta   \dd \theta \wedge \left[ ( r^2+a^2)   \dd \varphi - a   \dd t\right].\label{kerr_KY}
}{Kerr Killing-Yano tensor}
This allows us to define a Killing tensor as
\begin{subequations}
    \begin{align}
    K_{\mu\nu}&=Y\tdu{\mu}{\lambda}Y_{\nu\lambda}\\
    &=-\frac{1}{2}\qty(\mathcal R-\bar{\mathcal{R}})^2l_{(\mu}n_{\nu)}+\frac{1}{2}\qty(\mathcal R+\bar{\mathcal{R}})^2m_{(\mu}\bar m_{\nu)}\\
    &=2\Sigma l_{(\mu}n_{\nu)}+r^2 g_{\mu\nu}.\label{kerr_K_tensor}
\end{align}
\end{subequations}
The last equality allow to express the Killing tensor in terms of only manifestly real quantities and is obtain from the second line by using the decomposition Eq. \eqref{null_tetrad_decomposition} of the metric and noticing that $\abs{\mathcal R}^2=\Sigma$. 

Kerr's Killing tensor is often presented as the fundamental quantity related to Kerr spacetime hidden symmetry. Let us stress that, in the light of Proposition \ref{prop:killing_quantities}, we see that (conformal) Killing-Yano tensor appear to be more fundamental objects that Killing tensors, since their existence allow to systematically build a (conformal) Killing tensor, and even a Killing vector in Ricci-flat spacetime. However, the converse is not true: any Killing tensor  cannot be written as the contraction of two Killing-Yano tensors. Collinson \cite{collinson1976} worked out necessary and sufficient conditions under which such a decomposition is realizable, which are satisfied in the Kerr case. 

Regarding test motion in Kerr spacetime, the fundamental nature of the Killing-Yano tensor only shows up when studying the motion of spinning test bodies: at geodesic level, the knowledge of the Killing tensor is sufficient to construct a constant of motion -- the Carter constant -- that will allow the equations of motion to be solvable, as we will discuss in this chapter and in the following one. However, generalizing this Carter constant to include spin effects will explicitly require Killing tensor to be written as the contraction of two Killing-Yano tensors, as we will see in full details in Part \ref{part:conserved_quantities} of this text. Moreover, despite Kerr Killing tensor becomes reducible to combinations of Killing vectors in the Schwarzschild limit $a\to 0$, the Kerr Killing-Yano tensor does not and thus remains a hidden symmetry of Schwarzschild spacetime.

Finally, notice the remarkable identity
\begin{align}
    K^{\mu\nu}\xi_\nu=a\qty(a\xi^\mu+\eta^\mu)\triangleq\tilde\eta^\mu,
\end{align}
which relates the two Kerr Killing vectors thanks to the Killing tensor. The right hand side $\tilde \eta^\mu$ is itself a Killing vector.

\begin{table}[t!]
    \centering
    \begin{tabular}{c|c|c}
         & \textsc{standard} & \textsc{conformal}\\\hline
        \rule{0pt}{15pt}Killing vector & $\nabla_{(\mu} X_{\nu)}=0$ & $\nabla_{(\mu}X_{\nu)}=\frac{1}{4}g_{\mu\nu}\nabla_\rho X^\rho$\\
        \rule{0pt}{15pt}Killing-Yano tensor (antisymmetric) & $\nabla_{(\mu}Y_{\nu)\rho}=0$ & $\nabla_{(\mu}Y_{\nu)\rho}=-2g_{\mu[\nu}\xi_{\rho]}$\\
        \rule{0pt}{15pt}Killing tensor (symmetric) & $\nabla_{(\mu}K_{\nu\rho)}=0$ & $\nabla_{(\mu}K_{\nu\rho)}=K_{(\mu}g_{\nu\rho)}$
    \end{tabular}
    \caption{Ranks 1 and 2 Killing and conformal Killing objects for a $d=4$ manifold. We have here defined $\xi^\mu\triangleq-\frac{1}{3}\nabla_\lambda Y^{\lambda\mu}$ and $K_\mu\triangleq\frac{1}{6}\qty(2\nabla_\lambda K\tud{\lambda}{\mu}+\nabla_\mu K\tud{\lambda}{\lambda})$.}
    \label{tab:killing_quantities}
\end{table}

\section{Geodesic equations}

In this section, we will have a first look at Kerr geodesic equations, depicting how the symmetries discussed above allow to write them as a set of four first order ODEs. We will then briefly show how formal solutions to the geodesic equations can be written. These solutions will be the cornerstone used to discuss classification of Kerr geodesics and their explicit solutions, as we will undertake in Chapter \ref{chap:classification}.

Hereafter, we restrict to timelike geodesics, describing the physical motion of massive, point-wise test particles. This is the relevant problem to study, since it is the ``zeroth order'' approximation for describing the motion of massive compact objects around a Kerr black hole. However, notice that the study of null geodesics -- of prime importance for the fields of ray tracing and black hole imaging -- only deserves minor modifications with respect to the timelike case, see \cite{Kapec:2019hro,Compere:2020eat,Gralla:2019ceu} for more details.

\subsection{Conserved quantities along geodesics}\label{sec:conserved_killing_geodesic}

Let us first introduce a bit of formalism. We denote $z^\mu(\tau)$ the particle's worldline and $v^\mu \triangleq \dv{z^\mu}{\tau}$ its four-velocity. The parameter $\tau$ will always denote the proper time, other parametrizations being usually denoted by the symbol $\lambda$. As usually, one has $v_\mu v^\mu=-1$. The particle's linear momentum is $p_\mu=\mu v_\mu$, with $\mu>0$ the dynamical mass of the particle, which leads to $p_\mu p^\mu=-\mu^2$.

For any vector non-null $X^\mu$, we denote $\hat X^{\mu}\triangleq\frac{X^\mu}{\sqrt{\abs{X_\mu X^\mu}}}$ its normalized version. In particular, one obtains directly
\begin{align}
    \hat p^\mu\triangleq\frac{p^\mu}{\mu}=v^\mu.
\end{align}
This relation between the normalized linear momentum and the four-velocity is of course trivial in this case, where it follows simply from the definitions given above. We can therefore equivalently use the linear momentum or the four-momentum in our computations. However, as we will see in Part \ref{part:spinnin_bodies}, it will acquires a totally new meaning when we shall include spin effects, the latter breaking the parallelism between linear momentum and four-velocity.

A function of the dynamical variables $Q=Q(x^\mu,p_\mu)$ is \defining{conserved} along the worldline $z^\mu(\tau)$ if
\begin{align}
    \dv{Q}{\tau}=0\quad\Leftrightarrow\quad\frac{\text D Q}{\dd\tau}\triangleq v^\lambda\nabla_\lambda Q=0.
\end{align}

In what follows, we will always consider ``proper'' constants of motion, \textit{i.e.} constants normalized by the right power of $\mu$ to make them independent of the particle's mass. It amounts to consider constants of motions ``per unit of test mass'', and will enable us to get rid of the $\mu$ factors appearing in the equations. There appears to be four independent quantities that are conserved along geodesic motion in Kerr spacetime:
\begin{itemize}
    \item The dynamical mass $\mu$, since it is defined as (minus the square of) the norm of the geodesic tangent vector;
    \item As it is well-known, for any Killing vector $X^\mu$, the quantity $p_\mu X^\mu$ is conserved along geodesic motion. In Kerr spacetime, this leads to the two conserved quantities
    \begin{align}
        E_0\triangleq -\hat p_\mu\xi^\mu,\qquad L_0\triangleq \hat p_\mu\eta^\mu,\label{kerr_KV}
    \end{align}
    which respectively take the interpretation of particle's (proper) energy and projection of (proper) angular momentum onto the rotation axis of the black hole.
    \item The existence of the Killing tensor \eqref{kerr_K_tensor} directly implies that
    \begin{align}
        K_0\triangleq K_{\mu\nu} \hat p^\mu \hat p^\nu
    \end{align}
    is conserved along geodesics. For subsequent purposes, it is useful to define the shifted quantity
    \begin{align}
        Q_0 = K_0- (L_0 - a E_0)^2.
    \end{align}
    This is nothing but the celebrated (proper) \defining{Carter constant},\footnote{Notice that there is sometime a confusion in the literature about which of $Q_0$ and $K_0$ is called ``Carter constant''. With respect to the original paper of Carter \cite{carter68}, our $K_0$ is Carter's $\mathcal K/\mu^2$, whereas our $Q_0$ is his $Q/\mu^2$.} which was discovered by Carter \cite{carter68} as the separation constant of the geodesic Hamilton-Jacobi equation in Kerr spacetime, see Section \ref{sec:HJ_geodesic} for a discussion of this approach.
\end{itemize}

The existence of these four constants allows to reduce Kerr timelike geodesic motion to the following set of four first-order ODEs:
\begin{subequations}\label{kerr_geodesics}
\boxedeqn{
\Sigma\,  \dv{ t }{\tilde\tau}&=a\qty(L_0 -a E_0\sin^2\theta)+( r ^2+a^2)\frac{P_0( r )}{  \Delta ( r )},\label{eq:kerr_t}\\
\Sigma\,  \dv{ r }{\tilde\tau}&=\pm_{r}\sqrt{R( r )},\label{eq:kerr_r}\\
\Sigma\,  \dv{\cos\theta}{\tilde\tau}&=\pm_\theta \sqrt{\Theta(\cos^2\theta)},\label{eq:kerr_theta}\\
\Sigma\,  \dv{ \varphi }{\tilde\tau}&=-a E_0+ L_0 \csc^2\theta+a\frac{P_0( r )}{  \Delta ( r )},\label{eq:kerr_phi}
}{Kerr timelike geodesic equations}
\end{subequations}
with $\tilde\tau\triangleq\mu\tau$ and where we have defined
\begin{subequations}
    \begin{align}
P_0( r )&\triangleq E_0( r ^2+a^2)-a  L_0,\\
R( r )&\triangleq P_0^2( r )-  \Delta ( r )\qty( K_0+r ^2),\label{eq:kerr_vr}\\
\Theta(\cos^2\theta)&\triangleq Q_0 \qty(1-\cos^2\theta)+ \cos^2\theta\qty[a^2(E_0^2-1) \qty(1-\cos^2\theta)-  L_0^2].\label{eq:kerr_vtheta}
\end{align}
\end{subequations}
Moreover, one can show that Carter constant $Q_0$ takes the value
\boxedeqn{
    Q_0=\hat p_\theta^2+\cos^2\theta\qty[a^2\qty(1-E_0^2)+\qty(\frac{L_0}{\sin\theta})^2].
}{Carter constant}
Even if all these results can be directly obtained from tricky algebraic manipulations of the conserved quantities, their easiest derivation is obtained through Hamilton-Jacobi formulation, and will be discussed in Section \ref{sec:HJ_geodesic}.

\subsection{Formal solutions}
In a given Kerr geometry $\textsf{Kerr}(M,a)$, a geodesic is fully characterized by the quadruplet of parameters $\qty(\mu \geq 0,E_0 \geq 0, L_0 \in \mathbb R,Q_0 \geq -(L_0 - a E_0)^2)$, its initial spacetime position and two signs, $s_r^i \equiv \pm_{\hat r}|_{\tau = \tau_i}$, $s_\theta^i \equiv \pm_{\theta}|_{\tau = \tau_i}$, that correspond to the signs of the radial and polar velocity at the initial time $\tau_i$. 

The motion can be integrated using Mino time \cite{Mino:2003yg} defined as $ \dd\lambda \triangleq \dd\tilde\tau/\Sigma$ thanks to the property
\begin{align}
  \dd\lambda = \frac{  \dd r}{\pm_{ r}  \sqrt{R( r )} } = \frac{  \dd\cos \theta}{\pm_\theta\sqrt{\Theta(\cos^2\theta)}}. 
\end{align}
We consider a timelike geodesic path linking the initial event $\qty( t _i, r _i,\theta_i, \varphi _i)$ at  Mino time $\lambda_i$ and the final event  $\qty( t _f, r _f,\theta_f, \varphi _f)$ at Mino time $\lambda_f$. The geodesic can be formally integrated as 
\begin{align}
 t(\lambda_f)-{t}(\lambda_i) &= a( L_0  - a  E_0)(\lambda_f-\lambda_i)  +a^2 E_0\,  \qty( T_{\theta}(\lambda_f)-T_{\theta}(\lambda_i) )+T_{ r}(\lambda_f)-T_{ r}(\lambda_i) ,\nonumber \\
 \varphi (\lambda_f) - \varphi (\lambda_i) &= ( L_0 - a E_0) (\lambda_f-\lambda_i) +  L_0  \qty( \Phi_\theta (\lambda_f)-\Phi_\theta (\lambda_i)) + a \qty(\Phi_{r}(\lambda_f)-\Phi_{r}(\lambda_i)) \nonumber
\end{align}
where
\begin{subequations}
    \begin{align}
\lambda &= \sint \frac{  \dd r}{\pm_{ r}  \sqrt{ R( r )} } = \sint \frac{  \dd\cos\theta}{\pm_\theta \sqrt{\Theta(\cos^2\theta)}},\label{eqn:lambda}\\
T_\theta(\lambda) &\triangleq \sint \frac{\cos^2 \theta \,  \dd\cos \theta}{\pm_\theta \sqrt{\Theta (\cos^2 \theta)}}, \qquad \Phi_\theta(\lambda) \triangleq  \sint \frac{   \dd\cos\theta}{\pm_\theta \sqrt{\Theta (\cos^2\theta)}}\qty(\csc^2 \theta-1), \label{eqn:polarIntegrals}\\
T_{r}(\lambda) &\triangleq \sint \frac{(r^2 + a^2)P_0(r) \,  \dd r}{\pm_{r}  \sqrt{R( r )}   \Delta ( r )},\qquad \Phi_{r}(\lambda) \triangleq \sint \frac{P_0(r) \,  \dd r}{\pm_{r}  \sqrt{R( r )}   \Delta ( r )}. 
\end{align}
\end{subequations}
The notation $\sintline$ (already used in \cite{Kapec:2019hro,Gralla:2019ceu}) indicates that the signs $\pm_r$ and $\pm_\theta$ are flipped each time a zero of $R$ and $\Theta$, respectively, is encountered. Since the signs $\pm_r$, $\pm_\theta$ are identical to the signs of $  \dd r$ and $  \dd\cos\theta$, respectively, the integral \eqref{eqn:lambda} is monotonic around each turning point, as it should be in order to define an increasing Mino time $\lambda$ along the geodesic. Note that ${T}_\theta,~\Phi_\theta$ are normalized to be vanishing for equatorial motion. The initial signs $s_r^i \equiv \pm_{\hat r}|_{\lambda = \lambda_i}$, $s_\theta^i \equiv \pm_{\theta}|_{\lambda = \lambda_i}$, as well as the initial spacetime position, are fixed as a part of the specification of the orbit. If we denote by $w(\lambda),m(\lambda)$ the number of turning points in the radial and polar motion, respectively, at Mino time $\lambda$, then as the velocity changes sign at each turning point,
\begin{align}
\pm_r = s_r^i (-1)^w, \qquad \pm_\theta = s_\theta^i (-1)^m.\label{signs}
\end{align}

\section{Near-horizon geodesics of high spin Kerr}\label{sec:near_horizon}

When we start considering the behaviour of geodesics close to the horizon of Kerr black holes whose spin is close to its maximal value (\textit{i.e.} $a^2\to M^2$), a totally unexpected behaviour of Kerr geometry shows up, leading to the appearance of a throat-like geometry near the horizon. For exactly extremal black holes ($a^2=M^2$), the spacetime separates in three independent, geodesically complete spacetimes: the pre-existing Kerr spacetime, and two new spacetimes called \textit{Near Horizon Extremal Kerr} (NHEK) \textit{and near-NHEK}. This remarkable fact was first noticed by Bardeen, Press and Teukolsky in 1972 \cite{Bardeen:1972fi}.

\subsection{Kerr ISCO, and the need for near-horizon limits in the high spin regime}

To reach this conclusion, let us start by studying the behaviour of some specific circular (that is equatorial, $\theta=\pi/2$) geodesics in the high spin limit. The deviation from extremality of the black hole can be characterized by the parameter
\begin{equation}
\lambda\triangleq\sqrt{1-\frac{a^2}{M^2}},
\end{equation}
which tends to $0$ as $a^2\to M^2$.

The condition for having a circular orbit is given by
\begin{align}
    R(r)=0,\qquad\text{ et }\qquad R'(r)=0,
\end{align}
which can be solved for $E_0$ and $L_0$ as \cite{Bardeen:1972fi}
\begin{subequations}
    \begin{align}
    E_0&=\frac{r^{3/2}-2Mr^{1/2}\pm aM^{1/2}}{r^{3/4}\sqrt{r^{3/2}-3M r^{1/2}\pm 2aM^{1/2}}},\\
    L_0&=\frac{\pm M^{1/2}\qty(r^2\mp 2a M^{1/2}r^{1/2}+a^2)}{r^{3/4}\sqrt{r^{3/2}-3M r^{1/2}\pm 2aM^{1/2}}}.
\end{align}
\end{subequations}
The upper (resp. lower) signs correspond here to prograde ($L_0>0$) (resp. retrograde ($L_0<0$)) orbits. The existence of such orbits therefore require the condition
\begin{align}
    r^{3/2}-3Mr^{1/2}\pm 2aM^{1/2}\geq 0
\end{align}
to hold. It is possible to show that this inequality is only saturated for null geodesics (since $E_0$ scales as $1/\mu$, which implies $E_0\to\infty$ for $\mu\to 0$). We now review the near-extremal behaviour of some specific prograde circular geodesics. More details may be found in \cite{Bardeen:1972fi}.
\begin{itemize}
    \item The \defining{Innermost Stable Circular Orbit (ISCO)} is characterized by the stability condition
    \begin{align}
        \dv{E_0}{r}\eval_{r=r_\textsf{ISCO}}=0.
    \end{align}
    Close to extremality, the ISCO radius scales as
    \begin{align}
        r_\textsf{ISCO}=M\qty(1+2^{1/3}\lambda^{2/3})+\mathcal O\qty(\lambda).
    \end{align}
    Section \ref{sec:ISSO} will study the generalization of the ISCO to non-equatorial orbits, the \defining{Innermost Stable Spherical Orbit (ISSO)}.
    \item Wilkins showed \cite{Wilkins:1972rs} that a necessary condition for a Kerr timelike geodesic to be radially bounded was
    \begin{align}
        E_0<1.
    \end{align}
    A consequence of this result is that the \defining{Innermost Bounded Circular Orbit (IBCO)} can be determined by the condition
    \begin{align}
        E_0=1.
    \end{align}
    Near extremality, we obtain
    \begin{align}
        r_\textsf{IBCO}=M\qty(1+\sqrt{2}\lambda)+o\qty(\lambda).
    \end{align}
    \item A last case of interest is the \defining{light ring}, which is the outermost circular orbit for which circular geodesics are null. This amounts to find the largest root of
    \begin{align}
    r^{3/2}-3Mr^{1/2}\pm 2aM^{1/2}=0,
    \end{align}
    which gives 
    \begin{align}
        r_\gamma=M\qty(1+\frac{2}{\sqrt{3}}\lambda)+o(\lambda)
    \end{align}
    close to extremality.
\end{itemize}
Let us now compare these results to the outer horizon radius, which can be written as $r_+=M\qty(1+\lambda)$.
In the Schwarzschild black hole, $r_\textsf{ISCO}=6M$, $r_\textsf{IBCO}=4M$ and $r_\gamma=3M$ remain all larger than the horizon radius $r_+=2M$. These radii are monotonically decreasing functions of the black hole spin $a$, and both tend to $r=M$ for a maximally spinning black hole. This fact appears paradoxical, since the ISCO and the IBCO are timelike curves, but seem to be projected onto the light ring and the event horizon in the $\lambda\to 0$ limit, which are null submanifolds. We also notice that, close to extremality, the first subleading term in the radius expansion scales as $\lambda^{2/3}$ for the ISCO, whereas it scales as $\lambda$ for both the IBCO, the light ring and the event horizon.

We will resolve the geometry around these orbits by introducing a new system of coordinates $(\hat t,\hat r,\theta,\hat\varphi)$ defined as (see \cite{Compere:2012jk} for a review)
\begin{align}
    \hat t\triangleq \frac{t}{2M\kappa}\lambda^p,\qquad \hat r\triangleq\frac{\kappa}{M}\qty(r-r_+)\lambda^{-p},\qquad\hat\varphi\triangleq\varphi-\frac{t}{2M},
\end{align}
where $\kappa>0$ is a scale factor. The new coordinate system is tuned to be adimensional, comoving (the angular velocity $\dd\hat\varphi/\dd\hat t$ vanishes on the event horizon) and possesses a radial coordinate which is vanishing on the event horizon. The parameter $p$ shall be set to some specific value in order to resolve the geometry around some specific orbit. We will make the following choices:
\begin{itemize}
    \item $p=2/3$ will resolve the geometry around the ISCO, and gives rise to a new geometry, called \defining{Near Horizon Extremal Kerr} (NHEK);
    \item $p=1$ resolves the geometry around the IBCO, the light ring and the outer event horizon. It gives rise to the so-called \defining{near-NHEK} geometry.
\end{itemize}
Standing on the ISCO in the extremal limit $\lambda\to 0$, the proper radial distance to all the near-NHEK and the asymptotic orbits (characterized by $r>r_+$) becomes infinite. Moreover, Bardeen and Horowitz showed \cite{Bardeen:1999px} that NHEK and near-NHEK spacetimes turn on to be geodesically complete. Everything thus happens as if the original Kerr spacetime was splitting in three independent geometries, namely near-NHEK, NHEK and asymptotic ($r>r_+$) Kerr in this limit. It is then physically meaningful to study geodesic motion in (near-)NHEK spacetimes on their own, which is the subsequent task we will undertake.

\subsection{Near-horizon extremal Kerr (NHEK)}
\label{sec:NHEK}

Suitable coordinates for resolving the ISCO are defined as
\begin{align}
T&=\frac{ t }{2\kappa M}\lambda^{2/3},\qquad R =\frac{ \kappa\qty(r - r _+)}{M \lambda^{2/3}},\label{eq:cvn}\qquad
\Phi= \varphi -\frac{ t }{2M} .
\end{align}
Plugging \eqref{eq:cvn} into the Kerr metric \eqref{kerr_metric} and expanding the result in powers of $\lambda$ gives the NHEK spacetime in Poincaré coordinates
\begin{equation}
\dd s^2=2M^2\Gamma(\theta)\qty(-R^2\dd T^2+\frac{\dd R^2}{R^2}+\dd\theta^2+\Lambda^2(\theta)\qty(\dd\Phi+R\,\dd T)^2)+\mathcal{O}(\lambda^{2/3})\label{eq:NHEK_metric}
\end{equation}
where
\begin{align}
\Gamma(\theta)&\triangleq\frac{1+\cos^2\theta}{2},\qquad \Lambda(\theta) \triangleq\frac{2\sin\theta}{1+\cos^2\theta}.
\end{align}
In the $\lambda\to 0$ limit, NHEK metric is independent of the scale factor $\kappa$ and thus invariant under the scale transformation
\begin{align}
    R\to\kappa R,\qquad T\to\kappa^{-1}T.\label{NHEK_scale_Symmetry}
\end{align}
We therefore conventionally set $\kappa=1$ in all the NHEK formula.

In these NHEK coordinates, the Kerr ISCO is mapped onto the circular orbit of radius $R_\textsf{ISCO}=2^{1/3}$. However, considering the NHEK geometry on its own, the scale symmetry \eqref{NHEK_scale_Symmetry} makes all circular orbits equivalent, since their can be all mapped onto another. It is really the coordinate change Eq. \eqref{eq:cvn} (\textit{i.e.} the way we map Kerr spacetime to NHEK one) which gives a specific meaning to this peculiar radius.

The NHEK geometry admits a $\textsf{SL}(2,\mathbb R) \times \textsf{U}(1)$ symmetry generated by $\partial_\Phi$ and 
 \begin{align}
H_0 = T \partial_T - R \partial_R,\qquad H_+ = \partial_T,\qquad H_- = (T^2 + \frac{1}{R^2})\partial_T -2 T R \partial_R - \frac{2}{R}\partial_\Phi,
 \end{align}
with $H_0$ the generator of the scale symmetry Eq. \eqref{NHEK_scale_Symmetry}. The Killing tensor $K_{\mu\nu}$ \eqref{kerr_K_tensor} becomes reducible \cite{AlZahrani:2010qb,Galajinsky:2010zy} and can be expressed as
 \begin{align}
K^{\mu\nu} = M^2 g^{\mu\nu} + \mathcal C^{\mu\nu} + (\partial_\Phi)^\mu  (\partial_\Phi)^\nu
 \end{align}
where the $\textsf{SL}(2,\mathbb R)$ Casimir is given by
 \begin{align}
\mathcal C^{\mu\nu} \partial_\mu \partial_\nu = -H_0 H_0  + \frac{1}{2} (H_+  H_-  + H_- H_+).\label{Cas}
 \end{align}

We are interested in the Kerr geodesics that exist in the near-extremal limit within the NHEK geometry at leading order in $\lambda$. The NHEK angular momentum $L_0$ and Carter constant $Q_0$ are identical to their values defined in Boyer-Lindquist coordinates. The NHEK energy $E$ is related to the Boyer-Lindquist energy $E_0$ as
\begin{align}
 E_0 &=\frac{L_0}{2M}+\frac{\lambda^{2/3}}{2M}E. \label{rel1}
\end{align}
From now on, we will consider the leading high-spin limit; \textit{i.e.} we will neglect all $\mathcal O(\lambda^{2/3})$ corrections in \eqref{eq:NHEK_metric}.

\subsubsection{Geodesics}

In the NHEK geometry, Mino time is defined as $\lambda\triangleq\int^{\tilde\tau}\frac{\dd\tilde\tau'}{2M^2\Gamma(\theta(\tilde\tau'))}$, and the geodesic equations of motion simplify to 
\begin{subequations}\label{NHEK_geodesic_equations}
\boxedeqn{
\dv{T}{\lambda}&=\frac{E}{R^2}+\frac{L_0}{R},\label{eq:nhekT}\\
\dv{R}{\lambda}&=\pm_R\sqrt{v_R(R)},\label{eq:nhekR}\\
\dv{\cos \theta}{\lambda}&=\pm_\theta \sqrt{v_\theta(\cos^2\theta)},\label{eq:nhektheta}\\
\dv{\Phi}{\lambda}&=\frac{ L_0 }{\Lambda^2}-\frac{E}{R}-  L_0 ,\label{eq:nhekphi}
}{NHEK geodesic equations}
\end{subequations}
with
\begin{subequations}
    \begin{align}
v_R(R)&\triangleq E^2+2E  L_0  R+\frac{R^2}{4}\qty(3 L_0 ^2-4(Q_0+ M^2)),\label{eq:vR}\\
v_\theta(\cos^2\theta)&\triangleq Q_0 \sin^2\theta +\cos^2\theta  \sin^2\theta \qty(\frac{ L_0 ^2}{4}-M^2)- L_0 ^2\cos^2\theta\label{eq:vtheta}.
\end{align}
\end{subequations}
The limitation that $E$ remains real and finite implies from  \eqref{rel1} that we are only considering orbits with energy close to the extremal value $ L_0 /(2M)$. The ISSO angular momentum at extremality is equal to
\begin{align}
 L_{*}  = \frac{2}{\sqrt{3}} \sqrt{M^2 + Q_0}\label{ls}
\end{align}
It will play a key role in the following. On the equatorial plane ($\theta=\pi/2\Rightarrow Q_0=0$), the definition reduces to the $\ell_0$ used in Refs. \cite{Compere:2017hsi,Chen:2019hac}. We can also write down more simply 
 \begin{align}
v_R(R)= E^2+2E  L_0  R-\mathcal C R^2
 \end{align}
where $\mathcal C$ is the conserved quantity obtained from the $\textsf{SL}(2,\mathbb R)$ Casimir,
 \begin{align}
\mathcal C \triangleq \mathcal C^{\mu\nu} P_\mu P_\nu = Q - \frac{3}{4} L_0 ^2 + M^2 = \frac{3}{4} ( L_{*}^2 - L_0 ^2).\label{eqn:Carter}
 \end{align}
We also have 
\begin{equation}
v_R(R)=\left\lbrace\begin{array}{ll}
    -\mathcal C (R-R_+)(R-R_-), & \mathcal{C}\neq 0 ;\\
    2E L_0  (R-R_0), & \mathcal{C}=0 ,
\end{array}\right.
\end{equation}
with
 \begin{align}
R_\pm \triangleq \frac{E}{\mathcal C}  L_0  \pm \frac{\vert E \vert }{\vert \mathcal C\vert } \sqrt{ L_0 ^2+\mathcal C},\qquad R_0\triangleq -\frac{E}{2 L_0 }  . 
 \end{align}
The non-negative Carter constant $K_0$ is $K_0 = Q_0+\frac{ L_0 ^2}{4} > 0$, which implies that $\mathcal C > - L_0 ^2$ and that $R_- < R_+$ with $R_\pm$ both real. These equations all agree with Ref. \cite{Kapec:2019hro}. 

Similarly, defining $z\triangleq\cos^2\theta$, one can rewrite the polar potential as
\begin{equation}
    v_\theta(z)= - L_0 ^2 z+\qty(Q_0+\mathcal C_\circ z)(1-z)=\left\lbrace
    \begin{array}{ll}
        (Q_0+ L_0 ^2)(z_0-z) & \text{ for } \mathcal{C}_\circ=0\\
        \mathcal{C}_\circ(z_+-z)(z-z_-) & \text{ for } \mathcal{C}_\circ\neq 0
    \end{array}\right.
\end{equation}
where $\mathcal{C}_\circ$ is defined through the critical value of the angular momentum $ L_{\circ}$:
\begin{equation}
    \mathcal{C}_\circ\triangleq\frac{ L_0 ^2- L ^2_{\circ}}{4}, \qquad  L_{\circ}\triangleq 2M.
\end{equation}
The roots of the polar potential are given by
\begin{equation}
    z_0\triangleq\frac{Q_0}{Q_0+ L_0 ^2}\,\quad z_\pm\triangleq\Delta_\theta\pm \text{sign}(\mathcal C_\circ) \sqrt{\Delta_\theta^2+\frac{Q_0}{\mathcal{C}_\circ}},\quad\Delta_\theta\triangleq\frac{1}{2}\qty(1-\frac{Q_0+ L_0 ^2}{\mathcal{C}_\circ}).
\end{equation}

\subsubsection{Formal solution to the geodesic equations}

Using the same reasoning as for the Kerr geometry, the formal solutions to the geodesic equations are given by 
\begin{subequations}
    \begin{align}
\lambda_f  -\lambda_i&=  T^{(0)}_R(R_f)-T^{(0)}_R(R_i)  = \lambda_\theta(\theta_f)-\lambda_\theta(\theta_i)  , \label{eq:12}\\
T(\lambda_f) -T(\lambda_i)&= E \qty( T^{(2)}_R(R(\lambda_f))-T^{(2)}_R(R(\lambda_i))) \nonumber\\
&\quad+  L_0  \qty(  T^{(1)}_R(R(\lambda_f))-T^{(1)}_R(R(\lambda_i))), \\
\Phi(\lambda_f)-\Phi(\lambda_i) &=-\frac{3}{4}  L_0  (\lambda_f-\lambda_i) - E \qty( T^{(1)}_R(R(\lambda_f))-T^{(1)}_R(R(\lambda_i))) \nonumber\\
&~+  L_0  \big( \Phi_\theta(\theta(\lambda_f))-\Phi_\theta(\theta(\lambda_i))\big)\label{eq:NHEKphi}
\end{align}
\end{subequations}
where the three radial integrals are 
 \begin{align}
T^{(i)}_R(\lambda) & \triangleq  \sint \frac{ \dd R}{\pm_RR^i\sqrt{E^2+2E  L_0  R-\mathcal C R^2}} ,\qquad i = 0,1,2\label{eqn:radialNHEK}
 \end{align}
and the two polar integrals are 
\begin{subequations}
\begin{align}
\lambda_\theta(\lambda) & \triangleq \sint \frac{ \dd\cos \theta}{\pm_\theta\sqrt{v_\theta(\theta)}}, \\
\Phi_\theta(\lambda) & \triangleq \sint \frac{ \dd\cos \theta}{\pm_\theta\sqrt{v_\theta(\theta)}}\left( \frac{1}{\Lambda^2(\theta)} - \frac{1}{4}\right). 
 \end{align}
\end{subequations}
The notation was explained previously. We defined $\Phi_\theta$ such that it is zero for equatorial orbits (since $\Lambda(\pi/2)=2$). After integration, the equation \eqref{eq:12} can be inverted to give $R(\lambda)$ and $\theta(\lambda)$. We need to solve these five integrals as a function of the geodesic parameters $E, L_0 ,Q_0,s_\theta^i,s_R^i,T_i,R_i,\theta_i,\Phi_i$.

\subsection{Near-NHEK}\label{sec:nn}
\subsubsection{Metric and geodesic equations}
We now turn to the study of near-NHEK spacetime, which resolves a closer neighborhood of the black hole event horizon. We introduce the so-called \textit{near-NHEK} coordinates ($\hat t,\hat r,\theta,\hat\varphi$), related to Boyer-Lindquist coordinates through the relations
\begin{align}
\hat t&=\frac{ t }{2M\kappa}\lambda,\qquad 
\hat r=\frac{\kappa}{M}\qty( r - r _+)\lambda^{-1},\label{eq:cvnn}\qquad
\hat\varphi= \varphi -\frac{ t }{2M}.
\end{align}
At leading order in $\lambda$, the metric becomes
\begin{align}
\dd s^2&=2M^2\Gamma(\theta)\qty(-\hat r(\hat r+2\kappa)\dd \hat t^2+\frac{\dd \hat r^2}{\hat r(\hat r+2\kappa)}+\dd\theta^2+\Lambda^2(\theta)\qty(\dd\hat\varphi+(\hat r+\kappa)\dd \hat t)^2)\nonumber\\&\quad+\mathcal{O}(\lambda).
\end{align}
In these coordinates, the outer event horizon is located at $\hat r=0$, the light ring at $\hat r_\gamma=\frac{2}{\sqrt{3}}-1$ and the IBCO at $\hat r_\textsf{IBCO}=\sqrt{2}-1$.
Though the metric explicitly depends upon $\kappa$, no physical quantity depends upon it, since it is introduced from a coordinate transformation. The corresponding geodesic equations are now given by
\begin{subequations}\label{nNHEK_geodesic_equations}
    \boxedeqn{
\dv{\hat t}{\lambda}&=\frac{e+ L_0  (\hat r+\kappa)}{\hat r(\hat r+2\kappa)},\label{eq:nn1}\\
\dv{\hat\varphi}{\lambda}&=-\frac{e(\hat r+\kappa)+ L_0  \kappa^2}{\hat r(\hat r+2\kappa)}+ L_0 \qty(\frac{1}{\Lambda^2}-1),\label{eq:nn2}\\
\dv{\theta}{\lambda}&=\pm_\theta\sqrt{v_\theta(\cos^2 \theta)},\label{eq:nn3}\\
\dv{\hat r}{\lambda}&=\pm_{\hat r}\sqrt{v_{\hat r;\kappa}(\hat r)}\label{eq:nn4}
}{Near-NHEK geodesic equations}
\end{subequations}
where the radial potential can be written as
\begin{align}
v_{\hat r;\kappa}(\hat r)&\triangleq (e+ L_0  \kappa)^2+2e  L_0  \hat r+\frac{3}{4}\qty( L_0 ^2- L_{*} ^2)\hat r(\hat r+2\kappa )
\end{align}
and the angular potential is still as given in \eqref{eq:vtheta}. Although $\hat r$ is a meaningful radial coordinate (because of the horizon location at $\hat r=0$), it is convenient to introduce the shifted radial variable $R\triangleq \hat r+\kappa$ to get more elegant expressions. The symbol $R$ is also used in NHEK, but the context allows us to distinguish them. The generators of $\textsf{SL}(2,\mathbb R)\times \textsf{U}(1)$ are $\partial_{\hat{\varphi}}$ and 
\begin{equation}
    H_0=\frac{1}{\kappa}\partial_{\hat t},\qquad H_\pm=\frac{\text{exp}({\mp\kappa \hat t})}{\sqrt{R^2-\kappa^2}}\qty[\frac{R}{\kappa}\partial_{\hat t}\pm(R^2-\kappa^2)\partial_R-\kappa\partial_{\hat{\varphi}}].\label{vecs}
\end{equation}
The $\textsf{SL}(2,\mathbb R)$ Casimir $\mathcal C^{\mu\nu}\partial_\mu  \partial_\nu$ takes the form \eqref{Cas}, where the vectors are now given by \eqref{vecs}. 

The radial potential can be recast as 
 \begin{align}
v_{R;\kappa}(R)=\left\lbrace\begin{array}{ll}
    -\mathcal C(R-R_+)(R-R_-), & \mathcal C\neq 0 ;\\
    2e  L_{*} (R-R_0), & \mathcal{C}=0
\end{array}\right.
 \end{align}
where 
 \begin{align}
R_\pm \triangleq \frac{e  L_0 }{\mathcal C} \pm \frac{\sqrt{(\mathcal C+ L_0 ^2)(e^2+\kappa^2\mathcal{C})}}{|\mathcal{C}|},\qquad R_0\triangleq-\frac{e^2+\kappa^2 L_{*}^2}{2e L_{*}}.
 \end{align}
The near-NHEK energy $e$ is related to Boyer-Lindquist energy by
\begin{equation}
E_0=\frac{ L_0 }{2M}+\frac{\lambda}{2M\kappa}e.\label{eN}
\end{equation}
The near-NHEK and Boyer-Lindquist angular momenta and Carter constants $Q_0$ are equal. We define again the critical radius 
 \begin{align}
R_c = -\frac{e}{ L_0 }. 
 \end{align}
and the future orientation of the orbit again requires \eqref{consRc}. 

\subsubsection{Solutions to the equations of motion}
The NHEK and near-NHEK geodesic equations being very similar, this section will only briefly point out the similarities and the differences between the two cases. The formal solutions to the near-NHEK geodesic equations are
\begin{subequations}
    \begin{align}
    \lambda_f-\lambda_i&=t^{(0)}_{R;\kappa}(R_f)-t^{(0)}_{R;\kappa}(R_i)=\lambda_\theta(\theta_f)-\lambda_\theta(\theta_i),\\
    t(\lambda_f)-t(\lambda_i)&=e\qty(t^{(2)}_{R;\kappa}(R(\lambda_f))-t^{(2)}_{R;\kappa}(R(\lambda_i)))\nonumber\\
    &~+ L_0 \qty(t^{(1)}_{R;\kappa}(R(\lambda_f))-t^{(1)}_{R;\kappa}(R(\lambda_i))),\label{eqn:nnLambda}\\
    \hat\varphi(\lambda_f)-\hat\varphi(\lambda_i)&=-\frac{3}{4} L_0 (\lambda_f-\lambda_i)-e \qty(t^{(1)}_{R;\kappa}(R(\lambda_f))-t^{(1)}_{R;\kappa}(R(\lambda_i)))\nonumber\\
    &~-\kappa^2 L_0 \qty( t^{(2)}_{R;\kappa}(R(\lambda_f))- t^{(2)}_{R;\kappa}(R(\lambda_i)))\nonumber\\
    &~+ L_0 \qty(\Phi_\theta(\theta(\lambda_f))-\Phi_\theta(\theta(\lambda_i))) \label{eqn:nnPhi} 
\end{align}
\end{subequations}
where the polar integrals are the same as in NHEK (see above) and the radial ones are defined by
\begin{subequations}
    \begin{align}
    t^{(0)}_{R;\kappa}(\lambda)&\triangleq\sint\frac{\dd R}{\pm_R\sqrt{v_{R;\kappa}(R)}},\\
    t^{(i)}_{R;\kappa}(\lambda)&\triangleq\sint\frac{\dd R}{\pm_R\sqrt{v_{R;\kappa}(R)}}\frac{R^{2-i}}{R^2-\kappa^2},\qquad i=1,2.
\end{align}
\end{subequations}
Notice that NHEK geodesics equations can be recovered from near-NHEK ones by taking the formal limit $\kappa\to 0$; the normalization of the radial integrals has been chosen to satisfy $\lim_{\kappa\to 0} t^{(i)}_{R;\kappa}(R)=T^{(i)}(R)$ ($i=0,1,2$) as defined in \eqref{eqn:radialNHEK}. Therefore, the formal solutions to NHEK geodesic equations can also be recovered by taking the limit $\kappa\to 0$ in \eqref{eqn:nnLambda} and \eqref{eqn:nnPhi}.

\section{Concluding remarks}
We have now in our possession the geodesic equations for generic Kerr spacetime as well as for the associated NHEK and near-NHEK geometries, respectively given by the sets of equations \eqref{kerr_geodesics}, \eqref{NHEK_geodesic_equations} and \eqref{nNHEK_geodesic_equations}. These equations are all parametrized only in terms of constants quantities along the geodesics (energy, angular momentum and Carter constant) and of initial data of the motion.

Another remarkable fact to notice about these equations is that the radial and the polar motion are totally decoupled from each other. One can solve independently for both, before injecting the obtained solutions into the two remaining equations for obtaining the time and azimutal behaviour of the trajectory. This is also the reason why the polar and the radial part of the motion can be classified independently, as we will see in Chapter \ref{chap:classification}.

\chapter[Hamiltonian description of geodesic motion]{Hamiltonian description\\of geodesic motion}\label{chap:hamilton_kerr}

\vspace{\stretch{1}}

Before turning to the classification of Kerr geodesics, we will re-derive the generic Kerr geodesic equations \eqref{kerr_geodesics} from another perspective, namely that of Hamiltonian mechanics. 

Even if it might seems -- at first glance -- redundant with the previous chapter, this analysis will turn out to be of prime importance for the continuation of this work. The goal of this chapter is actually threefold: first, it will enable us to review the tools from Hamiltonian mechanics that will be extensively used in Part \ref{part:hamilton} of the thesis in a computationally simpler framework. Second, it will give us new perspectives on the results discussed so far. For instance, the resolvable character of geodesic equations in Kerr is a direct consequence of the fact that they form a completely (Liouville) integrable system, as can be seen directly within the framework of Hamiltonian mechanics. Finally, it will allow us to derive some practically insightful results: the action-angle formalism will be used to make explicit the tri-periodicity of Kerr (radially bounded) geodesic motion, and to compute its associated fundamental frequencies. This machinery reveals to be of prime importance in moderns developments, since it lies at the heart of perturbative schemes for describing the evolution of extreme mass ratio inspirals, such as two-timescale expansion \cite{Hinderer:2008dm,Pound:2021qin}.

Each section of this chapter first review the necessary concepts in an abstract way, before applying them to the problem of geodesic motion in Kerr spacetime.

\section{Hamiltonian description of geodesic motion in Kerr spacetime}

\subsubsection{Basic of Hamiltonian mechanics} 
We recall the basics of the geometrical (symplectic) formulation of Hamiltonian mechanics \cite{arnold1989mathematical}. The phase space of an Hamiltonian system possessing $N$ degrees of freedom is represented by a $2N$ differentiable manifold $\mathcal M$ equipped with a non-degenerated two-form $\mathbf \Omega$ (the \defining{symplectic form}) which is closed ($\dd\mathbf\Omega=0$). In this framework, the Hamiltonian $H$ is defined as a smooth function on $\mathcal M$. Hereafter, we will always consider time-independent Hamiltonians (\textit{i.e.}, conservative systems). The evolution of the system through its phase space is given by the integral curves of the Hamiltonian vector field
\begin{align}
v^\mathfrak{i}\triangleq\Omega^{\mathfrak i\mathfrak j}\nabla_\mathfrak{j} H.
\end{align}

To enable specific computations to be carried out, we have introduced some coordinates $(x^\mathfrak{i},p_\mathfrak{i})$ on the phase space. Here, $\mathfrak i,\mathfrak j,\ldots\in\qty{1,\ldots,N}$ stand here for phase space indices. 
Coordinates are said to be \defining{canonical} if $\mathbf\Omega=\dd p_\mathfrak{i}\wedge\dd x^\mathfrak{i}$. Prescribing a symplectic structure on phase space is equivalent to providing the algebra of Poisson brackets $\pb{\cdot}{\cdot}$ between all the (independent) coordinates. For canonical coordinates, one has $\pb{x^\mathfrak{i}}{p_\mathfrak{j}}=\delta^\mathfrak{i}_\mathfrak{j}$ and the Poisson brackets between two phase space functions $f$ and $g$ are
\begin{align}
    \pb{f}{g}\triangleq\sum_{\mathfrak{i}=1}^N\pdv{f}{q^\mathfrak{i}}\pdv{g}{p_\mathfrak{i}}-\pdv{f}{p_\mathfrak{i}}\pdv{g}{q^\mathfrak{i}}.\label{poisson_brackets}
\end{align}
Finally, the time evolution of any phase space quantity $f$ is given in term of Poisson brackets by
\begin{align}
    \dot F\triangleq\dv{F}{\tau}=\pb{F}{H}+\pdv{F}{\tau}.
\end{align}
In particular, $F$ is a constant of motion provided that $\pb{F}{H}=0$. Constants of motion are often called \defining{first integrals} of motion within the framework of Hamiltonian mechanics.

\subsubsection{Constrained Hamiltonian systems}
This is all for the basic tools we need from Hamiltonian mechanics. However, there is a subtlety that arises when building Hamiltonians for describing geodesic motion, and that our analysis shall account for: the presence of constraints. This subject is very technical, and we only bring here the elements useful for our purposes. The curious reader will consult fruitfully the book of M. Henneaux and C. Teitelboim \cite{Henneaux:1992ig}.

The appearance of constraints is more naturally introduced by starting from the Lagrangian viewpoint:
consider a theory described by the Lagrangian action
\begin{align}
    S_L\qty[q^\mathfrak{i}]=\int L(q^\mathfrak{i},\dot q^\mathfrak{i},\tau)\dd\tau.
\end{align}
The classical equations of motion are the Euler-Lagrange equations, which correspond to the extrema of $S_L$. They can be written in the enlightening way
\begin{align}
    \ddot q^\mathfrak{j}\pdv{L}{\dot q^\mathfrak{i}}{\dot q^\mathfrak{j}}=\pdv{L}{q^\mathfrak{i}}-\dot q^\mathfrak{j}\pdv{L}{\dot q^\mathfrak{i}}{q^\mathfrak{j}}.\label{euler_lagrange}
\end{align}

The starting point for switching to the Hamiltonian formulation is to define the conjugate momenta as
\begin{align}
    p_\mathfrak{i}\triangleq\pdv{L}{\dot q^\mathfrak{i}}.\label{momentaH}
\end{align}
This relation can be inverted to express the velocities in terms of the positions and the momenta provided that
\begin{align}
    \det \pdv{L}{\dot q^\mathfrak{i}}{\dot q^\mathfrak{j}}\neq 0.\label{hessian}
\end{align}
Notice that this is the very same condition that allows to invert Eq. \eqref{euler_lagrange} in order to express the accelerations uniquely in terms of the velocities and the positions. However, in many situations, this will not be the case and Eq. \eqref{hessian} will be vanishing. In that case, Eq. \eqref{momentaH} cannot be inverted. In others words, this means that the momenta are not all independent, but subjected to some constraints
\begin{align}
    \phi_A(q,p)\approx 0,\qquad A=1,\ldots,M.\label{csts}
\end{align}
Here, $M$ is an integer labelling the number of linearly independent constraints, and we use capital letters for indices running over the set of constraints. These constraints are referred to as \defining{primary constraints}, since they directly arise from the non-invertibility of the momentum-velocity relation. Notice that we use the weak equality symbol $\approx$ to emphasize that the constraints are not identically vanishing on phase space, but that they shall be enforced to be numerically vanishing. We call the submanifold of phase space defined by the constraints Eq. \eqref{csts} the \defining{primary constraint surface}. More generally, we say that two phase space functions $F$ and $G$ are \defining{weakly equal} if they are equal on the primary constraint surface. This is denoted $F\approx G$.

At the level of the equations of motion, the vanishing of Eq. \eqref{hessian} directly implies that one cannot express the accelerations solely in terms of the velocities and the positions. The EOMs solutions can then potentially contain arbitrary functions of time, which is the very signature of the presence of a gauge freedom in the theory.

The Hamiltonian is defined from the Lagrangian through the Legendre transformation
\begin{align}
    H= \dot q^\mathfrak{i}p_\mathfrak{i}-L.
\end{align}
Hamilton equations as well as the set of constraints \eqref{csts} can be obtained by varying the Hamiltonian action
\begin{align}
    S_H[q^\mathfrak{i},p_\mathfrak{i},u^A]=\int\qty(\dot q^\mathfrak{i}p_\mathfrak{i}-H-u^A\phi_A)\dd\tau.\label{hamiltonian_action}
\end{align}
Here, the functions $u^A$ take the role of Lagrange multipliers, and ensure the Legendre transformation to be invertible. The evolution equations for any phase space function $F$ can be conveniently written as
\begin{align}
    \dot F=\pb{F}{H}+u^A\pb{F}{\phi_A}.
\end{align}
Notice that all of this is equivalent to the reduced variational principle (with fewer variables) where the constraints have already been enforced, defined from
\begin{align}
    S_R[q^{\hat{\mathfrak{i}}},p_{\hat{\mathfrak{i}}}]=\int\qty(\dot q^{\hat{\mathfrak{i}}} p_{\hat{\mathfrak{i}}}-H)\dd\tau\label{reduced_action}
\end{align}
and subjected to the regularity conditions 
\begin{align}
    \phi_A=0,\qquad\delta\phi_A=0,\qquad\forall A\in 1,\ldots,M.
\end{align}

All this discussion only makes sense if the constraints are preserved under time evolution. This is the case provided that
\begin{align}
    \dot\phi_A=\pb{\phi_A}{H}+u^B\pb{\phi_A}{\phi_B}\stackrel{!}{\approx}0.
\end{align}
This equation will automatically be satisfied if
\begin{align}
    \pb{\phi_A}{H}\approx 0,\qquad\pb{\phi_A}{\phi_B}\approx 0,\qquad\forall A=1,\ldots,M.
\end{align}
Such constraints are known as \defining{first class constraints}. They are the only one that we will encounter in the continuation of this work. Since they are preserved under time evolution, these constraints can either be resolved at the level of the action or at the one of the equations of motion. Stated differently, one can replace the full phase space by the primary constraint surface, still using the same Poisson brackets than the ones defined on the original phase space. In this setup, the Lagrange multipliers $u^A$ are totally arbitrary functions.

\subsubsection{Constraints and gauge transformations} We will now discuss the strong link between first class primary constraints and gauge transformations, that is, transformations that do not alter the physical state of the system. Consider the evolution of some dynamical quantity $F$ from a time $\tau_1$ to a time $\tau_2=\tau_1+\delta \tau$. Since the Lagrange multipliers are arbitrary, one can make two different choices $u^A$ and $\tilde u^A$ at time $\tau_1$. We let the system evolve up to time $\tau_2$. The difference in its evolution -- corresponding to the two choices of multipliers -- is \cite{Henneaux:1992ig}
\begin{align}
    \delta F=\delta u^A\pb{F}{\phi_A},\qquad\delta u^A\triangleq\qty(u^A-\tilde u^A)\delta\tau.\label{gauge_generator}
\end{align}
The ambiguity in the evolution of $F$ is then physically irrelevant, since it is proportional to the physically irrelevant quantity $\delta u^A$. The physical state of the system at time $\tau_2$ is therefore not affected by the transformation Eq. \eqref{gauge_generator}. Importing the language of field theory, we say that \textit{the first class primary constraints generate gauge transformations}.

\subsubsection{Application to geodesic motion}
Let us go back to the problem of geodesic motion. For the sake of simplicity, we shall only consider timelike geodesics. They are the curves that extremize the proper distance
 \begin{align}
     S_L[x^\mu]=-\mu\int\sqrt{-\dot x_\mu\dot x^\mu} \dd\lambda,\qquad\dot x^\mu\triangleq\dv{x^\mu}{\lambda},\label{action_reparam}
 \end{align}
 with $\lambda$ an arbitrary time parameter and $\mu>0$ the mass of the particle. This problem is equivalently described by a $N=4$ Hamiltonian system. The phase space is spanned by $(x^\mu,p_\mu)$ (notice that the phase space indices $\mathfrak i,\ldots$ are now replaced by the spacetime indices $\mu,\ldots$), with
 \begin{align}
     p_\mu=\pdv{L}{\dot x^\mu}=\frac{\mu}{\sqrt{-\dot x_\alpha\dot x^\alpha}}\dot x_\mu.\label{geodesic_momentum}
 \end{align}
 The action Eq. \eqref{action_reparam} is invariant under arbitrary time reparametrizations $\lambda\to\sigma(\lambda)$. From Eq. \eqref{geodesic_momentum}, we see that it creates a dependence between the momenta, since
 \begin{align}
     p_\mu p^\mu=-\mu^2.
 \end{align}
 This leads to the existence of the so-called \defining{mass shell constraint}
 \begin{align}
     \mathcal H\triangleq p_\mu p^\mu+\mu^2\approx 0.\label{mass_shell}
 \end{align}
 The Hamiltonian is then obtained through the Legendre transformation
 \begin{align}
     H=\dot x^\mu p_\mu-L=x^\mu p_\mu+\mu\sqrt{\dot x_\mu\dot x^\mu}.
 \end{align}
 Contracting both sides of Eq. \eqref{geodesic_momentum} with $v^\mu$, we get the identity
 \begin{align}
     \dot x^\mu p_\mu=-\mu\sqrt{\dot x_\mu\dot x^\mu}=L.\label{identity_L}
 \end{align}
 We are therefore left with an identically vanishing Hamiltonian $H=0$, and the Hamiltonian action only contains the constraint,
 \begin{align}
     S_H[x^\mu,p_\mu,u]=\int\qty(\dot x^\mu p_\mu-u\mathcal H)\dd\lambda
 \end{align}
 with $u$ a Lagrange multiplier. Moreover, it is easy to show that $\mathcal H$ is indeed a first class constraint. 

Despite this Hamiltonian formulation is very elegant, having an identically vanishing Hamiltonian is inadequate for some of our subsequent purposes, \textit{e.g.} for studying the Hamilton-Jacobi equation. We therefore tackle the analysis from another viewpoint, namely by breaking from the start the reparametrization invariance. It is a textbook level statement \cite{carroll2003spacetime,Wald:1984rg,Misner1973} that the stationary points of the action \eqref{action_reparam} are equivalent to the stationary points of 
\begin{align}
    \bar S_L[x^\mu]=\frac{\mu}{2}\int\dot x_\mu\dot x^\mu\dd\tau
\end{align}
with $\tau$ being the proper time along the geodesic. The conjugate momenta are now
\begin{align}
    p_\mu=\mu \dot x^\mu.
\end{align}
Since the choice of $\tau$ as the proper time enforces the equality $\dot x_\mu\dot x^\mu=-1$ to hold, the constraint Eq. \eqref{mass_shell} is still present. However, the new Hamiltonian is non vanishing since 
\begin{align}
    \bar H=\dot x^\mu p_\mu-\frac{\mu}{2}\dot x_\mu\dot x^\mu=\frac{1}{2\mu}p_\mu p^\mu.
\end{align}
Expliciting all the dependencies, we get
\boxedeqn{
    \bar H(x^\alpha,p_\alpha)=\frac{1}{2\mu}g^{\mu\nu}(x^\alpha)p_\mu\, p_\nu.\label{H}
}{Hamiltonian for geodesic motion}
Again, the constraint $\mathcal H$ is first class. Evaluating Eq. \eqref{H} on the primary constraint surface gives simply
\begin{align}
   \bar H\approx-\frac{\mu}{2},\label{on_shell_H_geodesic}
\end{align}
and the on-shell value of the Hamiltonian is thus simply (minus one half of) the mass of the particle. Eq. \eqref{on_shell_H_geodesic} will turn on to admit a simple generalization for extended test bodies, as will be described in Chapter \ref{chap:covariant_H}.

\subsubsection{Hamilton and geodesic equations}
Let us convince ourselves that we have indeed worked out in a fancy way the geodesic equations, whose standard form is
\begin{align}\label{forced_geod_eq}
    \dv[2]{x^\mu}{\tau}+\Gamma^\mu_{\alpha\beta}\dv{x^\alpha}{\tau}\dv{x^\beta}{\tau} = 0.
\end{align}
A straightforward computation shows that Hamilton equations read
\begin{subequations}\label{geod_xp}
\begin{align}\label{geod_x}
    \dv{x^\nu}{\tau} &= g^{\nu\sigma}\frac{p_\sigma}{\mu},
    \\[1ex]\label{geod_p}
    \dv{p_\nu}{\tau} &= -\frac{1}{2\mu}g\tud{\sigma\rho}{,\nu}p_\sigma p_\rho.
\end{align}
\end{subequations}

Eq. \eqref{geod_x} is clearly a direct consequence of the definition of the impulsion, but Eq. \eqref{geod_p} requires more work. Taking the $\tau$-derivative of the 4-momentum $p_\nu$ we get
\begin{subequations}
    \begin{align}
    \dv{p_\nu}{\tau} =& \dv{}{\tau}\left(\mu\,g_{\mu\nu}\dv{x^\mu}{\tau}\right)\\
    =& \mu\,g_{\mu\nu,\alpha}\dv{x^\mu}{\tau}\dv{x^\alpha}{\tau} + \mu\,g_{\mu\nu}\dv[2]{x^\mu}{\tau}\\
    =& \mu\,g_{\mu\nu,\alpha}\dv{x^\mu}{\tau}\dv{x^\alpha}{\tau} + \mu\,g_{\mu\nu}\left[-\frac{1}{2}g^{\mu\alpha}\left(g_{\sigma\alpha,\rho} + g_{\rho\alpha,\sigma} - g_{\rho\sigma,\alpha}\right)\dv{x^\rho}{\tau}\dv{x^\sigma}{\tau} \right]\\
    =& \frac{1}{2\mu}g_{\sigma\rho,\nu}p^\sigma p^\rho  \\
    =& -\frac{1}{2\mu}g\tud{\sigma\rho}{,\nu}p_\sigma p_\rho.
\end{align}
\end{subequations}
In the third line we have made use of the geodesic equation \eqref{forced_geod_eq}, while in the last step we have used $p^\sigma p^\rho g_{\sigma\rho,\nu} = -p_\sigma p_\rho g\tud{\sigma\rho}{,\nu}$. This identity can be proven starting from the following
\begin{equation}
    p^\sigma p^\rho g_{\sigma\rho,\nu} = p^\sigma p^\rho \left(g^{\alpha\beta}g_{\alpha\sigma}g_{\beta\rho}\right)_{,\nu} = p_\alpha p_\beta g\tud{\alpha\beta}{,\nu} + p^\beta p^\rho g_{\beta\rho,\nu} + p^\sigma p^\alpha g_{\alpha\sigma,\nu}.
\end{equation}
Now, renaming the indices ($\alpha\rightarrow\beta$, $\sigma\rightarrow\rho$) and using the symmetry of the metric, the last two terms on the right-hand side are equal
\begin{equation}
    p^\sigma p^\rho g_{\sigma\rho,\nu} = p_\alpha p_\beta g\tud{\alpha\beta}{,\nu} + 2p^\beta p^\rho g_{\beta\rho,\nu}.
\end{equation}
Rearranging the terms we obtain the result
\begin{equation}
    -p^\sigma p^\rho g_{\sigma\rho,\nu} = p_\alpha p_\beta g\tud{\alpha\beta}{,\nu}.
\end{equation}
Gathering all the pieces leaves us with the second Hamilton equation \eqref{geod_p} and concludes the proof.

\subsubsection{Coordinate-time Hamiltonian} 
Actually, the Hamiltonian Eq. \eqref{H} is a \textit{generally covariant Hamiltonian}, since the time evolution is driven by an exterior parameter $\tau$ and not by the coordinate time $t$ itself. This formulation presents the advantage of being more symmetric with respect to all spacetime coordinates.

However, expressing the dynamics in terms of the time coordinate $t$ presents advantages in some situations, \textit{e.g.} for performing post-Newtonian expansions. In doing so, we will work directly with the reduced action Eq. \eqref{reduced_action}, thus reducing our phase space to be spanned only by $(x^i,p_i)$ with $i=1,2,3$. Let us choose the timelike coordinate such that $x^0=t$, and also take this $t$ to be the time parameter of the system. We therefore get
\begin{align}
    \dot x^0=\dv{t}{t}=1.
\end{align}
Starting again from the reparametrization invariant Lagrangian action Eq. \eqref{action_reparam} and using Eq. \eqref{identity_L}, we obtain the Hamiltonian through the Legendre transformation
\begin{align}
    H=\dot x^i p_i-L=\dot x^ip_i-\dot x^\mu p_\mu=-p_0.
\end{align}
Notice that, since we work on the reduced phase space, we have only to sum on the spatial indices in the Legendre transformation.

The last step is to rewrite $p_t$ by solving the constraint $\mathcal H\approx 0$. We write the metric in ADM coordinates \cite{poisson_2004}
\begin{align}
    g^{\mu\nu}=\mqty(g^{tt} & g^{ti} \\ g^{tj} & g^{ij})
    \triangleq\mqty(-\frac{1}{\alpha^2} & -\frac{\beta^i}{\alpha^2} \\ -\frac{\beta^j}{\alpha^2} & \gamma^{ij}-\frac{\beta^i\beta^j}{\alpha^2}),\qquad
    g_{\mu\nu}=\mqty(-\alpha^2+\beta_i\beta^i & -\beta_i\\-\beta_j & \gamma_{ij}),
\end{align}
where we respectively defined the lapse, the shift and the spatial metric as:
\begin{align}
    \alpha\triangleq\qty(-g^{tt})^{-1/2},\qquad \beta^i\triangleq\frac{g^{ti}}{g^{tt}},\qquad\gamma^{ij}\triangleq g^{ij}-\frac{g^{ti}g^{tj}}{g^{tt}}.
\end{align}
One has $\gamma^{ij}\gamma_{jk}=\delta^i_k$ and $\beta^i\triangleq \gamma^{ij}\beta_j$. Assuming that $p^t\propto\dv{x^0}{\lambda}$ is future-oriented, the constraint Eq. \eqref{mass_shell} yields
\begin{align}
    p_t\approx -\beta^ip_i-\alpha\sqrt{\mu^2+\gamma^{ij}p_ip_j}.
\end{align}
We therefore end up with the Hamiltonian
\begin{align}
    H\approx\beta^ip_i+\alpha\sqrt{\mu^2+\gamma^{ij}p_ip_j}.
\end{align}

\section{Complete (Liouville) integrability}
We now turn to the discussion of integrability. Remember that
\begin{definition}
An Hamiltonian system is \defining{completely integrable} (or Liouville integrable) in some open set $\mathcal U\subset\mathcal M$ provided that there exists $N$ linearly independent first integrals of motion $P_i=\qty(P_1=H,P_2,\ldots,P_N)$ that are linearly independent and in involution,
\begin{align}
    \pb{P_i}{P_j}=0\qquad\forall i,j=1,\ldots,N
\end{align}
at each point of $\mathcal U$.
\end{definition}
Notice that since $\pb{H}{H}=0$, the Hamiltonian is automatically a first integral, and we set $P_1=H$ by convention.

\defining{Liouville-Arnold theorem for integrable systems} \cite{arnold1989mathematical} states that a completely integrable system exhibits the following features:
\begin{enumerate}
    \item Its phase space is foliated by \defining{level sets}
\begin{align}
\mathcal M_\mathbf{p}\equiv \qty{x\in\mathcal M\,|\, P_i(x)=p_i\,|\,i=1,\ldots,N}
\end{align} 
which correspond to surfaces of constant $P_\alpha$ and are invariant under the Hamiltonian flow;
\item Its Hamilton equations can be integrated by quadratures, that is by a finite number of integrations and algebraic operations;
\item Finally, in the case of compact and connected level sets, one can switch to action-angle variables, see Section \ref{sec:aa} for a more precise description.
\end{enumerate}

In the case of geodesic motion in Kerr spacetime, the phase space is eight dimensional (that is, $N=4$). A set of four independent first integrals turns out to be the four conserved quantities studied in Chapter \ref{chap:kerr} (recall that $H\approx-\mu/2$, we can thus choose arbitrarily one of these two quantities). We set
\begin{align}
    P_\alpha=\qty(H,E_0,L_0,K_0).
\end{align}
Even if we already know that the last three quantities Poisson commute with $H$ (since they are constants of motion), we shall verify it as a consistency check. Using the identity (which comes from the metric compatibility condition)
\begin{align}
    \partial_\nu g^{\alpha\beta}=-2\Gamma^{(\alpha}_{\nu\lambda}g^{\beta)\lambda},
\end{align}
we have
\begin{align}
    \begin{split}
    \pb{H}{E_0}&=-\frac{1}{2\mu}\qty(\partial_\lambda g^{\alpha\beta}\xi^\mu\pb{x^\lambda}{p_\mu}+2g^{\alpha\beta} p_\alpha\pb{p_\beta}{x^\lambda}\partial_\lambda \xi^\mu)\\
    &=\frac{1}{\mu}\nabla_\mu\xi_\nu p^\mu p^\nu\\
    &=0,
    \end{split}
\end{align}
since $\xi^\mu$ is a Killing vector field. We use the exact same reasoning for the two others quantities to get
\begin{align}
    \pb{H}{L_0}=-\frac{1}{\mu}\nabla_\mu\eta_\nu p^\mu p^\nu=0,\qquad \pb{H}{K_0}=-\frac{1}{\mu}\nabla_\mu K_{\nu\rho}p^\mu p^\nu p^\rho=0.
\end{align}

It remains to show that $E_0,L_0$ and $K_0$ are mutually Poisson commuting. One has
\begin{align}
    \pb{E_0}{L_0}&=-\xi^\mu\partial_\mu\eta^\nu p_\nu+\eta^\mu\partial_\mu\xi^\nu p_\nu=0,
\end{align}
since $\partial_\mu\xi^\nu=\partial_\mu\eta^\nu=0$ are direct consequences of the explicit expressions Eq. \eqref{kerr_KV} of the Killing vectors. Moreover,
\begin{align}
    \pb{E_0}{K_0}&=\xi^\mu\partial_\mu K^{\alpha\beta} p_\alpha p_\beta-2 p_\mu\partial_\alpha\xi^\mu K^{\alpha\beta} p_\beta=0.
\end{align}
The last equality comes from the fact that the second term of the RHS is vanishing for the same reason used in the previous computation, while the vanishing of the first term arises from the stationarity of Kerr background, which ensures that $\xi^\mu\partial_\mu K^{\alpha\beta}=\partial_t K^{\alpha\beta}=0$. The computation of the last bracket $\pb{L_0}{K_0}$ is identical, except that one shall use the consequence of Kerr axial symmetry, $\eta^\mu\partial_\mu K^{\alpha\beta}=\partial_\varphi K^{\alpha\beta}=0$.

At the end of the day, we have formally proven that Kerr geodesic equations do form a completely integrable system. Notice that it is also the case for the (near-)NHEK spacetimes. We therefore have all the properties expected from Liouville-Arnold theorem. In particular, the equations of motion have already been formally integrated in the previous chapter. Section \ref{sec:aa} will discuss action-angle formulation of geodesic motion. As we will see, this will require some care, since level sets of Kerr geodesic motion are unbounded in the time direction.

\section{Hamilton-Jacobi equation}\label{sec:HJ_geodesic}

The Hamilton-Jacobi formulation of Kerr geodesic motion has been famously studied by Carter in the late sixties \cite{carter68}, and enabled him to discover the Carter constant. We reproduce this computation here, because it is the simplest way to reduce Kerr geodesic equations to the set of four first order ODEs Eqs. \eqref{kerr_geodesics}.

Let us first remind ourselves the Hamilton-Jacobi formulation for a time independent Hamiltonian system, $H(x^i,p_i,\tau)=H(x^i,p_i)$. In this case, $H$ is constant along the motion, and we denote $\alpha$ its constant value:
\begin{align}
    H(x^i,p_i)=\alpha.
\end{align}
For such a system, Hamilton's equations are equivalent to the Hamilton-Jacobi equation
\begin{align}
    H\qty(x^i,\pdv{S}{x^i})+\pdv{S}{\tau}=0.
\end{align}
where $S(x^i,\alpha,\tau)$ is referred to as \defining{Hamilton's principal function}.
Because the Hamiltonian is time independent, it necessarily takes the form
\begin{align}
    S(x^i,\alpha,\tau)=W(x^i,\alpha)-\alpha \tau
\end{align}
where $W(x^i,\alpha)$ is known as \defining{Hamilton's characteristic function}. Finally, we notice that the conjugate momenta are given by the partial derivatives of the Hamilton's principal/characteristic function with respect to the coordinates,
\begin{align}
    p_i=\pdv{S}{x^i}=\pdv{W}{x^i}.\label{p_from_HJ}
\end{align}
In many practical cases, this formulation is very powerful to separate the equations of motion in a set of independent ODEs \cite{Goldstein2001,arnold1989mathematical}.

Let us go back to geodesic motion in Kerr spacetime. Denoting\footnote{The notation $W^{(0)}$ will become meaningful in Chapter \ref{chap:HJ}, when we will study the Hamilton-Jacobi equation for spinning test bodies as a perturbation of the geodesic Hamilton-Jacobi problem.} $W^{(0)}\triangleq W/\mu$ and noticing that $\alpha=-\mu/2$ in our case, the Hamilton-Jacobi equation reduces to
\begin{align}
    g^{\mu\nu}W^{(0)}_{,\mu}W^{(0)}_{,\nu}+1=0.
\end{align}
We assume that it admits a separated solution of the form
\begin{align}
    W^{(0)}=-E_0 t+L_0\varphi+w_{0r}(r)+w_{0z}(z),
\end{align}
where $E_0$ and $L_0$ are (for now arbitrary) constants and $w_{0r}(r)$, $w_{0z}(z)$ functions that remain to be determined. Notice that we work here with $z\triangleq\cos\theta$ instead of the Boyer-Lindquist polar angle $\theta$ itself, because it allows to save a lot of computation time for computer-assisted symbolic checks of our equations (no trigonometric functions being involved). Moreover, Eq. \eqref{p_from_HJ} implies that
\begin{align}
    \hat p_t&=\pdv{W^{(0)}}{t}=-E_0,\qquad \hat p_\varphi=\pdv{W^{(0)}}{\varphi}=L_0\label{E0_L0}
\end{align}
and comparing these results with the discussion of Section \ref{sec:conserved_killing_geodesic}, we directly see that $E_0$ and $L_0$ are the quantities conserved along geodesics associated to Kerr's timelike and axial Killing vectors, thus explaining the notation chosen.

Using the explicit expression of the inverse Kerr metric \eqref{inverse_metric}, the Hamilton-Jacobi equation reads
\begin{align}
    \begin{split}
        &1-E_0^2+\frac{L_0^2}{(r^2+a^2)(1-z^2)}-\frac{2Mr\qty[E_0\qty(r^2+a^2)-aL_0]^2}{\Delta(r)\Sigma(r,z)\qty(r^2+a^2)}\\
    &\quad+\frac{\Delta(r)}{\Sigma(r,z)}\qty(w'_{0r})^2+\frac{1-z^2}{\Sigma(r,z)}\qty(w'_{0z})^2=0,
    \end{split}
\end{align}
where $'$ indicates a derivative of a single-variable function with respect to its argument.

A bunch of algebraic manipulations allows to show that this equation can be separated as
\begin{align}
    \begin{split}
        &\qty(1-E_0^2)r^2+L_0^2\frac{a^2}{r^2+a^2}-\frac{2Mr\qty[E_0(r^2+a^2)-aL_0]^2}{\Delta(r)(r^2+a^2)}+\Delta(r)\qty(w'_{0r})^2\\
    &=-\qty[\qty(1-E_0)^2a^2z^2+\frac{L_0^2}{1-z^2}+\qty(1-z^2)\qty(w'_{0z})^2].
    \end{split}
\end{align}
Since this equation is separated, the two sides of the equality shall be equal to a constant, that we will denote $-Q_0-L_0^2$. Another bunch of straightforward algebra then allows to turn the Hamilton-Jacobi equation in the two following ODEs:
\begin{subequations}
\begin{align}
    \qty(1-z^2)\qty(w'_{0z})^2&=Q_0-z^2\qty[a^2\qty(1-E_0^2)+\frac{L_0^2}{1-z^2}],\\
    \Delta^2(r)\qty(w'_{0r})^2&=-\qty(K_0+r^2)\Delta(r)+P_0^2(r).
\end{align}
\end{subequations}
Remind that we are using the shortcut notations
\begin{align}
    K_0\triangleq Q_0+\qty(L_0-aE_0)^2,\qquad P_0(r)\triangleq E_0\qty(r^2+a^2)-aL_0.
\end{align}
Solutions for the first order derivatives of the action with respect to $r$ and $z$ can now be written as
\begin{align}
    w'_{0r}(r)&=\pm_r\frac{\sqrt{R(r)}}{\Delta(r)},\qquad
    w'_{0z}(z)=\pm_\theta\frac{\sqrt{\Theta(z^2)}}{1-z^2}\label{HJ_potentials}
\end{align}
with the potentials $R(r)$ and $\Theta(z^2)$ defined in Eqs. \eqref{eq:kerr_vr} and \eqref{eq:kerr_vtheta}.
Moreover, since $w'_{0z}=\hat p_\theta$, the constant $Q_0$ takes the value\begin{align}
    Q_0=\hat p_\theta^2+\cos^2\theta\qty[a^2\qty(1-E_0^2)+\qty(\frac{L_0}{\sin\theta})^2],
\end{align}
which is in agreement with the value of Carter constant found in the previous chapter. However, $Q_0$ appears here as the separation constant of the geodesic Hamilton-Jacobi equation in Kerr spacetime, as in the original derivation of Carter \cite{carter68}.

The final step is to recover geodesic equations \eqref{kerr_geodesics} from our separated solution to the Hamilton-Jacobi equation. Using Eq. \eqref{p_from_HJ}, we notice that
\begin{align}
    \dv{x^\mu}{\tilde\tau}=g^{\mu\nu} W^{(0)}_{,\nu}.
\end{align}
This yields
\begin{subequations}
    \begin{align}
    \dv{t}{\tilde\tau}&=-g^{tt}E_0 + g^{t\varphi}L_0,\qquad
    &&\dv{r}{\tilde\tau}=g^{rr}w'_{0r}(r),\\
    \dv{z}{\tilde\tau}&=g^{zz} w'_{0z}(z),\qquad
    &&\dv{\varphi}{\tilde\tau}=- g^{\varphi t}E_0+g^{\varphi\varphi}L_0.
\end{align}
\end{subequations}
Replacing the derivatives of $w_{0r,z}$ by the solutions Eq. \eqref{HJ_potentials} and using the components of the inverse metric Eq. \eqref{inverse_metric}, we find that these equations are actually the geodesic equations in the form \eqref{kerr_geodesics}, as anticipated.

\section{Action-angle formulation}\label{sec:aa}

We now turn to the action-angle formulation of the geodesic motion. The philosophy is very simple: we will switch to a new set of variables $\qty(x^\mu,p_\mu)\to\qty(q^\mu,J_\mu)$ in which the geodesic (that is, Hamilton) equation take the simple form
\begin{align}
    \dv{q^\mu}{\tau}=\Omega^\mu(J_\alpha),\qquad\dv{J_\mu}{\tau}=0,
\end{align}
which thus make explicit the periodicity of the system with respect to some \textit{angle variables} $q^\mu$. Here, the \textit{action variables} $J_\mu$ are constants, and the \textit{fundamental frequencies} given by $\Omega^\mu$ only depend upon them. Moreover, as we will see, this machinery comes with a prescription for computation the fundamental frequencies from the Hamiltonian. The main advantage of this formulation relies on the fact that it allows -- for radially bounded geodesics -- to express any dynamical quantity $f(\tau)$ in a Fourier series \cite{Goldstein2001,Schmidt:2002qk}
\begin{align}
    f(\tau)&=\sum_{\boldsymbol k\in\mathbb Z^3}f_{\boldsymbol{k}}e^{i\Omega_{\boldsymbol{k}}\tau}=\sum_{\boldsymbol k\in\mathbb Z^3}f_{\boldsymbol{k}}e^{i\boldsymbol k\cdot\boldsymbol\Omega},
\end{align}
with $\boldsymbol k\cdot\boldsymbol\Omega=\Omega_{\boldsymbol{k}}\tau=k_r\Omega_r+k_\theta\Omega_\theta+k_\varphi\Omega_\varphi$. The dynamics is thus entirely specified by a discrete number of Fourier coefficients,
\begin{align}
    q_{\boldsymbol{k}}\triangleq\frac{1}{\qty(2\pi)^3}\int_{[0,2\pi]^3}\dd[3]\Omega f\qty(\frac{\boldsymbol k\cdot\boldsymbol\Omega}{\Omega_{\boldsymbol{k}}})e^{i\boldsymbol k\cdot\boldsymbol\Omega}.
\end{align}
This is this very fact that makes action-angle formalism to lies at the heart of the two timescale expansion scheme for describing self-forced motion around a Kerr black hole \cite{Hinderer:2008dm,Pound:2021qin}.

Even if it appears conceptually simple, the rigorous derivation of action-angle variables for geodesic motion in Kerr spacetime is quite technical, the main difficulty arising from the fact that the level sets foliating the phase space are non-compact in the time direction, thus forbidding to use Liouville-Arnold theorem directly. We will proceed in three steps: (i) introduce the ideas behind action-angle formulation on the simplest example possible, a $N=1$ Hamiltonian system; (ii) discuss the generalization of Liouville-Arnold theorem needed for dealing with non-compact level sets and finally (iii) apply this formalism to obtain action variables and fundamental frequencies in Boyer-Lindquist coordinates.

All along this section, we will remain very pictorial and skip most of the lengthy computations as well as too subtle technical details, our aim being to give a pedagogical overview of the subject and to obtain the explicit expressions of the fundamental frequencies. The interested reader may refer to \cite{Schmidt:2002qk,arnold1989mathematical,Hinderer:2008dm,Goldstein2001} for more detailed developments.

\subsection{A first glance at action-angle variables} 

Let us expose the philosophy of action-angle formalism on the simplest example possible. Consider a one dimensional conservative Hamiltonian system described by the conjugated variables $(x,p)$, which is periodic in $x$ (either librating or rotating).

Turning to Hamilton-Jacobi description, we notice that the associated Hamilton principal function possesses the right structure to be the generating function of a type II canonical transformation $(x,p)\to(X,\alpha)$, $\alpha$ being the constant value of the Hamiltonian. One has the usual relations for this kind of transformation:
\begin{align}
    p=\pdv{W(x,\alpha)}{x}\eval_{\alpha=\text{cst}},\qquad X=\pdv{W(x,\alpha)}{\alpha}\eval_{x=\text{cst}}.
\end{align}

Going from the original variables $(x,p)$ to the action-angle ones $(q,J)$ is achieved through a sequence of two canonical transformations, $(x,p)\to(X,\alpha)\to(q,J)$. The first one is the type II canonical transformation mentioned above, whereas the second transformation is performed by introducing a new variable $J$ (the \defining{action variable}) to replace $\alpha$ as the new (constant) momentum. This variable is defined as
\begin{align}
    J\equiv\frac{1}{2\pi}\oint p\,\dd x.
\end{align} 
Here, $\oint$ designs the integration over a complete period of the motion. One has
\begin{align}
    \alpha=H=H(J),\qquad W=W(x,J).
\end{align}
The \defining{angle variable} $q$ is then defined as being the generalized coordinate conjugated to $J$:
\begin{align}
    q\equiv\pdv{W}{J}\eval_{x=\text{cst}}.
\end{align}
The main gain obtained using action-angle formalism is that Hamilton's equations take the simple form
\begin{align}
    \dv{q}{t}=\Omega(J),\qquad\dv{J}{t}=0,
\end{align}
where 
\begin{align}
    \Omega(J)\equiv\pdv{H(J)}{J}.\label{freq}
\end{align}
They are trivially solved by
\begin{align}
    q(t)=\Omega(J)\, t+\beta, \qquad J(t)=\text{constant},
\end{align}
with $\beta$ being an arbitrary integration constant. Moreover, the frequencies associated with the periodic motion appear explicitly. Namely, the change in $q$ over a whole period $\Delta\tau$ of the motion is
\begin{align}
    \Delta q\equiv\oint\pdv{q}{x}\dd x=\oint\pdv[2]{W}{q\,}{J}\dd q=\dv{J}\oint\pdv{W}{x}\dd x=\dv{J}\oint p\,\dd x=2\pi.
\end{align}
It is then straightforward to notice that $\Omega={2\pi}/{\Delta\tau}$. Consequently, $\Omega$
can be interpreted as the \defining{fundamental frequency} of the periodic motion, and Eq. \eqref{freq} gives a prescription for computing the system's fundamental frequency from the Hamiltonian expressed in terms of the action variable $J$.

\subsection{Generalized action-angle variables for non-compact level sets}\label{sec:generalized_Liouville}

Our subsequent goal is to generalize this analysis up to a level enabling the treatment of Kerr geodesic motion. The relevant generalization will take the form of a generalized version of the Liouville-Arnold theorem, discussed in Section \ref{sec:generalized_Liouville}.

Recall that Liouville-Arnold theorem states that, for any system completely integrable in the neighbourhood of some compact and simply connected level set, one can define symplectic coordinates exhibiting the same features than the action-angle variables introduced above. Moreover, a prescription for defining them explicitly in a coordinate-invariant way can be formulated.

However, this result cannot be  applied directly to Kerr geodesic motion, because, in this case, the phase space is not bounded in the time direction, leading to the non-compactness of the level sets. Nevertheless, a generalization of the Arnold-Liouville result to systems possessing non-compact level-sets has been provided by Fiorani, Giachetta and Sardanashvily in 2003 \cite{Fiorani_2003}. Formally, the statement is the following: 
\begin{theorem}
Let be an Hamiltonian system possessing $N$ degrees of freedom and completely integrable in $\mathcal U\supset\mathcal M_\mathbf{p}$ for which the vector fields $v^a_\alpha\equiv \Omega^{ab}\nabla_bP_\alpha$ are complete on $\mathcal U$ and such that the level sets $\mathcal M_\mathbf{p}$ foliating $\mathcal U$ are all diffeomorphic one to another. Then:
\begin{itemize}
\item $\exists k\in\mathbb N$, with $0\leq k\leq N$, such that $\mathcal{M}_\mathbf{p}$ is diffeomorphic to $T^k\times \mathbb R^{N-k}$ (with $T^k$ the k-torus). Moreover, there exists some open set $\mathcal V\supset	\mathcal M_\mathbf{p}$ diffeomorphic to $T^k\times\mathbb R^{N-k}\times \mathcal B$, with $\mathcal B$ being an open ball.
\item There exist symplectic coordinates $(q^\alpha,J_\alpha)$ ($\alpha=1,\ldots,N$) whose first $k$ variables $q_\alpha$ are $2\pi$-periodic ($q_\alpha+2\pi\equiv q_\alpha$) and for which the first integrals $P_\alpha$ can be expressed as functions of the $J_\alpha$ only: $P_\alpha=P_\alpha(J_1,\ldots,J_N)$. The coordinates $(q^\alpha,J_\alpha)$ are called \textit{generalized action-angle variables}.
\end{itemize}
\end{theorem}

This theorem comes with a prescription for computing explicitly the action variables $J_\alpha$. Because the symplectic form is closed, there exists, in some neighbourhood of $\mathcal M_\mathbf{p}$, a one-form $\mathbf \Theta$ (the \defining{symplectic potential}) defined through $\dd\mathbf\Theta=\mathbf\Omega$. The set of inequivalent closed paths in the level set $\mathcal M_\mathbf{p}$ are given by a set of generators $\gamma_1,\ldots,\gamma_k$ of the fundamental homotopy group of $\mathcal M_\mathbf{p}$,
\begin{align}
\Pi_1\qty(\mathcal M_\mathbf{p})\simeq\qty(\mathbb Z_k,+_{\!\!\!\mod k}).
\end{align}
One then defines the generalized action variables as
\begin{align}
J_\alpha\equiv\frac{1}{2\pi}\oint_{\gamma_\alpha}\mathbf\Theta,\qquad\alpha=1,\ldots,k.
\end{align}
This definition turns out to be independent of the choice of the symplectic potential and of the generators. Moreover, this prescription is unique up to the redefinition of the origin of the angle variables
\begin{align}
q_\alpha\to q_\alpha+\pdv{Z(J_\beta)}{J_\alpha},\qquad J_\alpha\to J_\alpha\label{ambig2}
\end{align}
with $Z$ an arbitrary function of the action variables, and up to transformations of the form
\begin{align}
q_\alpha\to A_{\alpha\beta}q_\alpha,\qquad J_\alpha\to B_{\alpha\beta}J_\beta\label{ambig1}
\end{align}
where $A_{\alpha\beta}$, $B_{\alpha\beta}$ are constant real matrices with $A_{\alpha\beta}B_{\alpha\gamma}=\delta_{\beta\gamma}$ such that the $J_\alpha$'s are left invariant for $\alpha=1,\ldots,k$.
In this generalized action-angle formulation, Hamilton's equations take the simple form
\begin{align}\label{hamilton_general_AA}
\dv{q_\alpha}{t}=\Omega_\alpha({\mathbf J}),\qquad\dv{J_\alpha}{t}=0
\end{align}
with $\Omega_\alpha(\mathbf J)\equiv\pdv{H(\mathbf J)}{J_\alpha}$. Their solutions read
\begin{align}
q_\alpha(t)=\Omega_\alpha(\mathbf J_0)t+q_{\alpha 0},\qquad J_\alpha(t)=J_{\alpha0}.
\end{align}
For $\alpha=1,\ldots,k$, the functions $\Omega_\alpha(\mathbf J)$ are interpreted as the angular frequencies of the motion.

\subsection{Action-angle formulation of the geodesic motion in Kerr}

Given the results of the previous sections of this chapter, the generalized Arnold-Liouville theorem with $k=3$ and $N=4$ applies to \textit{radially bounded} Kerr geodesic motion. The generators of $\Pi_1(\mathcal M_\mathbf{p})$ can be constructed directly from the periodic motion $x^i(\tau)=(r,\theta,\varphi)$ in Boyer-Lindquist coordinates \cite{Schmidt:2002qk}. The ambiguity \eqref{ambig1} in the definition of the action variables takes the form\begin{align}
J_t\to\gamma J_t+v^iJ_i,\qquad J_i\to J_i.
\end{align}
Here, $\gamma\in\mathbb R$ and $v^i\in\mathbb R^3$ are arbitrary constants. It can be resolved by setting
\begin{align}
J_t\equiv\frac{1}{2\pi}\oint_{\gamma_t}\mathbf\Theta,
\end{align}
with $\gamma_t$ being a $2\pi$-lenght integral curve of the extension to $ \mathcal M$ of the timelike Killing vector field $\partial_t$  \cite{Hinderer:2008dm}. As shown in \cite{Hinderer:2008dm}, this definition is independent of the curve $\gamma_t$ chosen.

This scheme can be shown to realize a \textit{coordinate-invariant} definition of \textit{generalized action-angle variables} $(q^t,q^r,q^\theta,q^\varphi,J_t,J_r,J_\theta,J_\varphi)$ for bounded geodesic motion in Kerr spacetime. This definition is unique up to the (trivial) residual ambiguity
\begin{align}
q^\alpha\to q^\alpha+\pdv{Z(\mathbf J)}{J_\alpha},\qquad J_\alpha\to J_\alpha.
\end{align}

We will now explicit the action variables $J_\alpha$ and the frequencies $\Omega_\alpha$ in Boyer-Lindquist coordinates. A detailed computation can be found in \cite{Schmidt:2002qk} and is reviewed in \cite{Hinderer:2008dm}. Using the symplectic potential $\mathbf\Theta=p_\mu\dd x^\mu$, the action variables can be written in terms of the first integrals $P_\alpha=(H,E_0,L_0,Q_0)$ as
\begin{subequations}\label{JofP}
\begin{align}
    J_t&=\frac{1}{2\pi}\int_0^{2\pi}p_t\dd t=-\mu E_0,\qquad
    &&J_r=\frac{1}{2\pi}\oint\frac{\sqrt{R(r)}}{\Delta(r)}\dd r\\
    J_\theta&=\frac{1}{2\pi}\oint\sqrt{\Theta(\cos^2\theta)}\dd\theta,\qquad
    &&J_\phi=\frac{1}{2\pi}\oint p_\phi\dd\phi=\mu L_0.
\end{align}
\end{subequations}

Expressing the frequencies is slightly more cumbersome. One defines the integrals 
\begin{subequations}
    \begin{align}
    W&\triangleq\int_{r_1}^{r_2}\frac{r^2E_0(r^2+a^2)-2Mra(L_0-aE_0)}{\Delta(r)\sqrt{R(r)}}\dd r,\qquad
    X\triangleq\int_{r_1}^{r_2}\frac{\dd r}{\sqrt{R(r)}},\\
    Y&\triangleq\int_{r_1}^{r_2}\frac{r^2}{\sqrt{R(r)}}\dd r,\qquad
    Z\triangleq \int_{r_1}^{r_2}\frac{r\qty[L_0r-2M(L_0-aE_0)]}{\Delta(r)\sqrt{R(r)}}\dd r
\end{align}
\end{subequations}
with $r_{1,2}$ the two turning points of the radial motion (\textit{i.e.} the two largest roots of $R(r)$). One also denotes $z_\pm$ the roots of $\Theta(z)$ ($z_-<z_+$) with $z\triangleq\cos^2\theta$. Defining $\beta^2\triangleq a^2(1-E_0^2)$ and $k\triangleq\sqrt{z_-/z_+}$, one has
\begin{subequations}
    \boxedeqn{
    \Omega_t&=\frac{K(k)W+a^2z_+E_0\qty[K(k)-E(k)]X}{K(k)Y+a^2z_+\qty[K(k)-E(k)]X},\\
    \Omega_r&=\frac{\pi K(k)}{K(k)Y+a^2z_+\qty[K(k)-E(k)]X},\\
    \Omega_\theta&=\frac{\pi\beta\sqrt{z_+}X/2}{K(k)Y+a^2z_+\qty[K(k)-E(k)]X},\\
    \Omega_\phi&=\frac{K(k)Z+L_0\qty[\Pi(z_-,k)-K(k)]X}{K(k)Y+a^2z_+\qty[K(k)-E(k)]X}.
}{Fundamental frequencies of Kerr geodesic motion}
\end{subequations}
Here $K$, $E$ and $\Pi$ are the standard Legendre forms of the complete elliptic integrals. We use the conventions of \cite{olver10}, which are reviewed in Appendix \ref{app:ellipticFunctions}.

The most important take home message is that the action variables $J_\alpha$ are only functions of the first integrals of motion $P_\alpha$, $J_\alpha=J_\alpha(P_\beta)$ (see Eq. \eqref{JofP} for the explicit expressions). In terms of these generalized action-angle variables, the geodesic equations take the simple form
\begin{equation}\label{geod_aa_unperturbed}
    \dv{J_\alpha}{\tau} = 0, \quad \dv{q_\alpha}{\tau} = \omega_\alpha(J_\beta).
\end{equation}
These equations can be easily integrated to obtain
\begin{equation}
    J_\alpha(\tau) = J_{\alpha}(0), \qquad q_\alpha(\tau) = q_{\alpha}(0) + \omega_\alpha(J_\beta)\tau .
\end{equation}
The constancy of the action variables is a direct consequence of the conservation of the first integrals of the geodesic motion in Kerr spacetime.

\chapter[Classification of Kerr geodesics]{Classification\\of Kerr geodesics}\label{chap:classification}

\vspace{\stretch{1}}

This Chapter will be devoted to the classification of Kerr polar motion and (near-)NHEK radial motion, relying on the equations of motion obtained in Chapter \ref{chap:kerr}. As noticed before, radial and polar geodesic equations in these three spacetimes are decoupled one from another, and can thus be studied independently. Moreover, they both take the form
\begin{align}
    \dv{y}{\lambda}=\pm_y\sqrt{V(y)},
\end{align}
with $V(y)$ some polynomial \textit{effective potential}. Consequently, the study of the generic features of the motion (existence and location of turning points\ldots) reduces mainly to the study of the roots and the sign of the potential $V(y)$ for all the possible values of the parameters upon which it depends. We will also provide explicit solutions to the geodesic equations for each possible type of motion, thus verifying in a concrete way one of the main consequences of complete integrability of the system.

In this chapter, we will study Kerr polar motion in full generality, but leave the radial motion aside. Its study is actually much more involved, since the associated potential is of degree four, whereas it is only of degree two for the polar motion.  Instead, we will analyze in details the radial geodesic motion in both NHEK and near-NHEK spacetimes. This analysis turns out to be relevant for at least three reasons: (i) it will provide us with a deep comprehension of the peculiarities of the motion near highly spinning black holes, and in particular of the behaviour of their spherical geodesics. (ii) As we will see, any geodesic can be obtained by applying a well-chosen conformal transformation to a \textit{spherical} geodesic. This enables to use the computational scheme extensively applied in \cite{Compere:2017hsi} to generate the waveforms emitted by objects moving along geodesics in the near-horizon region. Such an analysis can allow to unravel ``smoking gun'' signatures for the existence of highly spinning black holes in Nature. Finally, (iii) this radial classification was also the work which paved the road to the one of generic Kerr radial motion, which was completed by G. Compère, Y. Liu and J. Long \cite{Compere:2021bkk}.

All along this chapter, we will adopt a different notation than the one used in all the rest of this thesis for Kerr geodesic conserved quantities. The main reason for this switch is to match with the conventions used in \cite{Compere:2020eat}. We also take the opportunity to restore the dependence of the equations upon $\mu$.  This enables our classification scheme to hold for null geodesics as well as for timelike ones. The temporary convention is given by the following dictionary
\begin{align}
    \begin{split}
        &\hat t \to t,\quad,\hat r\to r\quad,\hat\varphi\to\phi,\quad E_0\to \frac{\hat E}{\mu},\quad L_0\to\frac{\ell}{\mu},\\ &L_*\to\frac{\ell_*}{\mu},\quad L_\circ\to\frac{\ell_\circ}{\mu},\quad Q_0\to \frac{Q}{\mu^2}.
    \end{split}
\end{align}

\section{Classification of polar motion}\label{sec:polar_kerr}

This section aims to describe in full detail the classification of the polar geodesic motion around a Kerr black hole. The classification scheme used has been first introduced in \cite{Kapec:2019hro}, and subsequently completed in \cite{Compere:2020eat}. The taxonomy is depicted in Figs. \ref{fig:angularKerrTaxonomy} and \ref{fig:angularNullKerrTaxonomy}. The phenomenology of the polar behavior is principally governed by the sign of $Q$. The details of the classification are summarized in Tables \ref{table:taxPolarKerr1} and \ref{table:taxPolarKerrNull}. 

Let us recall that, due to mirror symmetry between the two hemispheres, the natural variable to describe the polar motion is $z\triangleq\cos^2\theta$. Let us define
\begin{equation}
     \epsilon _0(\hat E,\mu)  \triangleq a^2 (\hat E^2 - \mu^2), \label{defeps0}
\end{equation}
and
\begin{equation}\label{defzpm}
    z_\pm\triangleq\Delta_\theta\pm \text{sign}\, { \epsilon _0} \sqrt{\Delta_\theta^2+\frac{Q}{ \epsilon _0}},\qquad\Delta_\theta\triangleq\frac{1}{2}\qty(1-\frac{Q+\ell^2}{ \epsilon _0}),\qquad z_0\triangleq\frac{Q}{Q+\ell^2}. 
\end{equation}

The classification is determined by the roots of the polar potential \eqref{eq:kerr_vtheta}. Assuming $a \neq 0$, we can rewrite it as
\begin{equation}
    \Theta(z)=-\ell^2 z+(Q+ \epsilon _0 z)(1-z)= \left\lbrace\begin{array}{ll}
         \epsilon _0(z_+-z)(z-z_-), &  \epsilon _0\neq 0 ;\\
        (Q+\ell^2)(z_0-z), &  \epsilon _0=0.
    \end{array}\right.
\end{equation}

Our definition of the roots $z_\pm$ implies the ordering $z_-<z_+$ (and respectively $z_+<z_-$) for $ \epsilon _0>0$ (respectively $ \epsilon _0<0$). This is a convenient convention because in both cases the maximal angle will be related to $z_+$. The positivity of the polar potential implies that the poles $z=1$ ($\theta=0,\pi$) can only be reached if $\ell=0$. Note that when a geodesic crosses a pole, its $\varphi$ coordinates discontinuously jump by $\pi$. The invariance of the polar geodesic equation under $(\hat{E},\ell)\to(-\hat{E},-\ell)$ allows us to reduce the analysis to prograde $\ell\geq 0$ orbits. We distinguish the orbits with angular momentum $\ell \neq 0$ and without, $\ell = 0$:

\paragraph{I. Nonvanishing angular momentum $\ell\neq0$.} We must consider the following cases:
\begin{enumerate}
    \item {\textbf{$\boldsymbol{-(\ell-a\hat{E})^2\leq Q<0}$}} can only occur if $ \epsilon _0>0$, otherwise leading to $\Theta<0$. For $ \epsilon _0>0$, the motion is \textit{vortical}; \textit{i.e.}, it takes place only in one of the two hemispheres without crossing the equatorial plane and is bounded by
    \begin{equation}
        0<z_-\leq z\leq z_+<1.
    \end{equation}
This vortical motion can only occur provided $\ell^2\leq\qty(\sqrt{ \epsilon _0}-\sqrt{-Q})^2$. 
    
    \item {$\boldsymbol{Q>0}$} leads to motion crossing the equator and symmetric with respect to it, bounded by
    \begin{subequations}
    \begin{align}
       0 &\leq z \leq z_+<1\qquad( \epsilon _0\neq0),\\
        0 & \leq z \leq z_0<1\qquad( \epsilon _0=0).
    \end{align}
    \end{subequations}
    We will refer to such a motion as \textit{pendular}; 
    
    \item {$\boldsymbol{Q=0}$} allows us to write 
    \begin{equation}
        \Theta(z)= \epsilon _0\, z\, (1-\frac{\ell^2}{ \epsilon _0}-z).
    \end{equation}
    If $ \epsilon _0\leq 0$, the positivity of the polar potential enforces the motion to be equatorial. For $ \epsilon _0\geq 0 $, equatorial motion exists at $z=0$. For $ \epsilon _0\geq 0 $  and $\ell^2 \leq  \epsilon _0$, another motion exists bounded by
    \begin{equation}
        0 < z \leq 1-\frac{\ell^2}{ \epsilon _0}<1,
    \end{equation}
    which is a marginal case separating the pendular and vortical regimes; the motion then admits only one turning point and asymptotes to the equator both at future and at past times. Since we could not find a terminology for such a motion in the literature, we propose to call it \emph{equator-attractive}\footnote{This neologism accurately reflects the fact that the motion is polar and that the equator is an attractor. The terminology ``homoclinic'' is already used in the literature to refer to radial motion.}. In the special case where $z=0$ at the initial time, the motion remains $z=0$ at all times: it is \emph{equatorial}. 
\end{enumerate}

\vspace{4pt}
\noindent
\paragraph{II. Vanishing angular momentum $\ell=0$.} The polar potential reduces to
\begin{equation}
\Theta(z)=\left\lbrace\begin{array}{ll}
 \epsilon _0\qty(\frac{Q}{ \epsilon _0}+z)\qty(1-z), &  \epsilon _0\neq 0\\
Q\qty(1-z), &  \epsilon _0=0. 
\end{array}\right.
\end{equation}
We distinguish the following cases:
\begin{enumerate}
    \item {$\boldsymbol{ \epsilon _0=0}$} leads to motion over the whole polar range $0\leq z \leq 1$ for $Q>0$; we called it \textit{polar} motion. The only turning point is located at $z=1$. For $Q=0$, the potential vanishes identically and the polar angle remains constant; we call it \textit{azimuthal} motion; for $Q<0$ the potential is positive only if the motion takes place along the black hole axis $z=1$; we call it \textit{axial} motion.
    \item {$ \boldsymbol{\epsilon _0>0}$} leads to a \textit{polar} motion $0\leq z\leq1$ for {$Q>0$}. For $Q=0$ and $z=0$, the motion is \textit{equatorial}. For $Q=0$ and $z\neq0$, $z=0$ is an asymptotic attractor of the motion which only takes place in one of the hemispheres. It is therefore a special case of \textit{equator-attractive} motion where the turning point is at the pole $z=1$. For {$Q<0$}, the motion is either \textit{vortical} ({$0<-\frac{Q}{ \epsilon _0}\leq z \leq 1$}) for {$- \epsilon _0<Q<0$} or \textit{axial} with $z=1$ for {$Q\leq- \epsilon _0<0$}.
    \item {$\boldsymbol{ \epsilon _0<0}$} leads to a \textit{polar} motion $0\leq z\leq1$ for {$Q\geq- \epsilon _0>0$} and to a \textit{pendular} one ({$0\leq z \leq -\frac{Q}{ \epsilon _0}<1$}) for {$0<Q<- \epsilon _0$}. For $Q=0$, the motion is either \textit{equatorial} or \textit{axial} for the potential to be positive. For {$Q<0$}, the motion also has to take place along the axis.
\end{enumerate}
Let us finally notice that, for any value of $ \epsilon _0$ and $Q\geq -(a E_0 )^2$, an axial motion is always possible.

 \begin{table}[!hbt]
    \centering
    \begin{tabular}{|c|c|c|c|}
    \hline
        \textbf{Energy} & \textbf{Carter constant} & \textbf{Polar range} & \textbf{Denomination}\\\hline
        $ \epsilon _0<0$ ($|\hat{E}|<\mu$) & $-(\ell-a\hat{E})^2\leq Q<0$ & $\emptyset$ &\rule{0pt}{13pt} $\emptyset$\\\cline{2-4}
         & $Q=0$ & $z=0$ &\rule{0pt}{13pt} Equatorial$(\hat{E})$\\\cline{2-4}
          & $Q>0$ & $0\leq z\leq z_+<1$ & \rule{0pt}{13pt}Pendular$(\hat{E},Q)$\\\hline
           $ \epsilon _0=0$ ($|\hat{E}|=\mu$) & $-(\ell-a\hat{E})^2\leq Q<0$ & $\emptyset$ & \rule{0pt}{13pt}$\emptyset$\\\cline{2-4}
         & $Q=0$ & $z=0$ &\rule{0pt}{13pt} Equatorial${}_\circ$\\\cline{2-4}
          & $Q>0$ & $0\leq z\leq z_0<1$ &\rule{0pt}{13pt} Pendular${}_\circ(Q)$\\\hline
           $ \epsilon _0>0$ ($|\hat{E}|>\mu$) & $-(\ell-a\hat{E})^2\leq Q<0$ & $0<z_-\leq z \leq z_+<1$ &\rule{0pt}{13pt} Vortical$(\hat{E},Q)$\\\cline{2-4}
         & $Q=0$ and $ \epsilon _0\geq\ell^2$ & \begin{tabular}{c} \rule{0pt}{13pt}$0 < z \leq 1-\frac{\ell^2}{ \epsilon _0} <1$\\ ($ \text{sign}\,{\cos\theta}$ fixed)\end{tabular} &\begin{tabular}{c}\rule{0pt}{13pt} Equator-\\attractive$(\hat{E})$\end{tabular}\\\cline{3-4}
        & &$z=0$&Equatorial$(\hat{E})$  \\\cline{2-4}
               & $Q>0$ & $0\leq z\leq z_+<1$ &\rule{0pt}{13pt} Pendular$(\hat{E},Q)$\\\hline
         \end{tabular}
       \caption{Polar taxonomy of Kerr geodesics  with $\ell\geq 0$. The orbits with $\ell<0$ are obtained from $\ell>0$ by flipping the signs of both $\hat{E}$ and $\ell$.\\ }
    \label{table:taxPolarKerr1}
\end{table}

\begin{landscape}

\phantom{a}

\vspace{\stretch{1}}

\begin{table}[!htb]
    \centering
        \begin{tabular}{|c|c|c|c|}
\hline
        \textbf{Energy} &\textbf{ Carter constant} & \textbf{Polar range} & \textbf{Denomination}\\\hline
        $ \epsilon _0<0$ ($|\hat{E}|<\mu$) & $ -(a \hat E)^2 \leq Q<0$ & $z=1$ &\rule{0pt}{13pt} Axial${}^0(\hat{E},Q)$\\\cline{2-4}
        & $Q=0$ & $z=0,1$ &\rule{0pt}{13pt} \begin{tabular}{c}\rule{0pt}{13pt}Equatorial${}^0(\hat{E})$\\ Axial${}^0(\hat{E})$\end{tabular}\\\cline{2-4}
         & $0<Q<- \epsilon _0$ & $0\leq z\leq- \frac{Q}{ \epsilon _0}<1$ &\rule{0pt}{13pt} Pendular${}^0(\hat{E},Q)$\\\cline{2-4}
          & $0<- \epsilon _0\leq Q$ & $0\leq z\leq 1$ & \rule{0pt}{13pt}Polar${}^0(\hat{E},Q)$\\\hline
           $ \epsilon _0=0$ ($|\hat{E}|=\mu$) & $-(a\hat{E})^2\leq Q<0$ & $z=1$ & \rule{0pt}{13pt}Axial${}^0_\circ(Q)$\\\cline{2-4}
         & $Q=0$ & $z=\text{constant}$ &\rule{0pt}{13pt} Azimuthal${}^0_\circ$\\\cline{2-4}
          & $Q>0$ & $0\leq z\leq 1$ &\rule{0pt}{13pt} Polar${}^0_\circ(Q)$\\\hline
           $ \epsilon _0>0$ ($|\hat{E}|>\mu$) & $ -(a \hat E)^2 \leq Q\leq - \epsilon _0<0$ & $z=1$ &\rule{0pt}{13pt} Axial${}^0(\hat{E},Q)$\\\cline{2-4}
           & $- \epsilon _0<Q<0$ & $0<-\frac{Q}{ \epsilon _0}\leq z\leq 1$ &\rule{0pt}{13pt} Vortical${}^0(\hat{E},Q)$\\\cline{2-4}
         & $Q=0$  & $0 < z \leq  1$ ($ \text{sign}\,{\cos\theta}$ fixed)&\rule{0pt}{13pt} \begin{tabular}{c}Equator-\\attractive${}^0(\hat{E})$\end{tabular}\\\cline{3-4}
         &   & $z=0$ & \rule{0pt}{13pt}Equatorial${}^0(\hat{E})$ \\\cline{2-4}
                        
           & $Q>0$ & $0\leq z\leq1$ & \rule{0pt}{13pt}Polar${}^0(\hat{E},Q)$\\\hline
\rule{0 pt}{13 pt}$ \epsilon _0\in\mathbb{R}$ & $Q\geq - (a \hat E)^2$ & $z=1$ & Axial$^0( E_0 ,Q)$\\\hline
\end{tabular}
    \caption{Polar taxonomy of Kerr  geodesics  with $\ell =  0$.\vspace{20pt}}
    \label{table:taxPolarKerrNull}
\end{table}

\vspace{\stretch{1}}

\end{landscape}

\begin{figure}[!ph]\vspace{0pt}
    \centering
    \vspace{-30pt}
    \includegraphics[width=10.5cm]{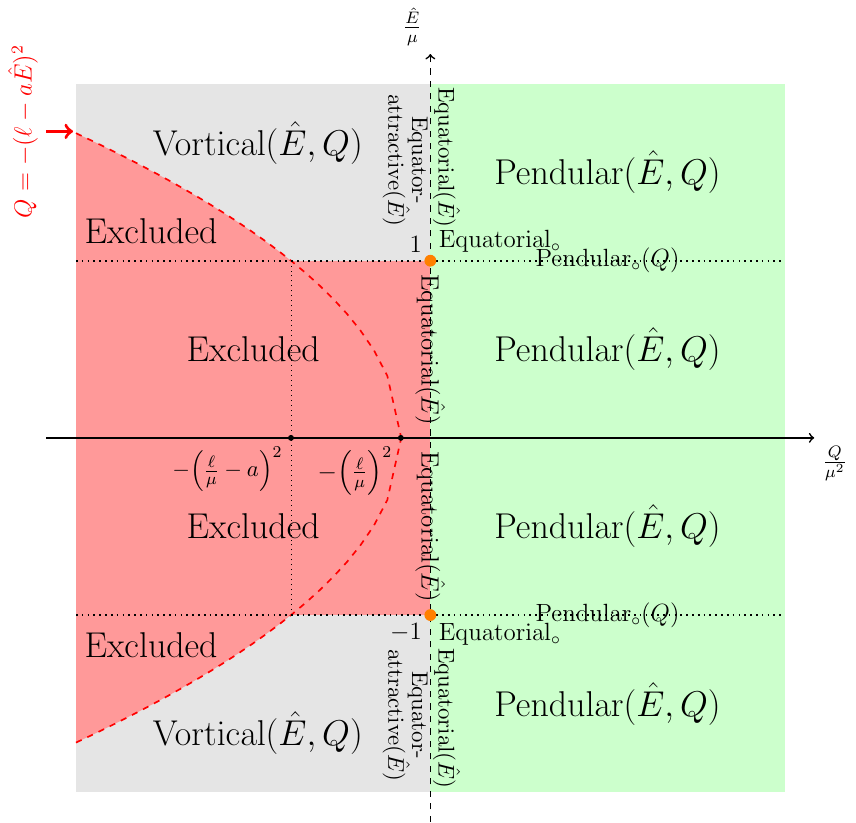}
    \caption{Polar taxonomy of $\ell\neq 0$ Kerr geodesics. Equator-attractive$(\hat{E})$ orbits become Equatorial$(\hat{E})$ orbits when the initial angle is at the equator. }\vspace{0pt}
    \label{fig:angularKerrTaxonomy}  
    \centering
    \includegraphics[width=10.5cm]{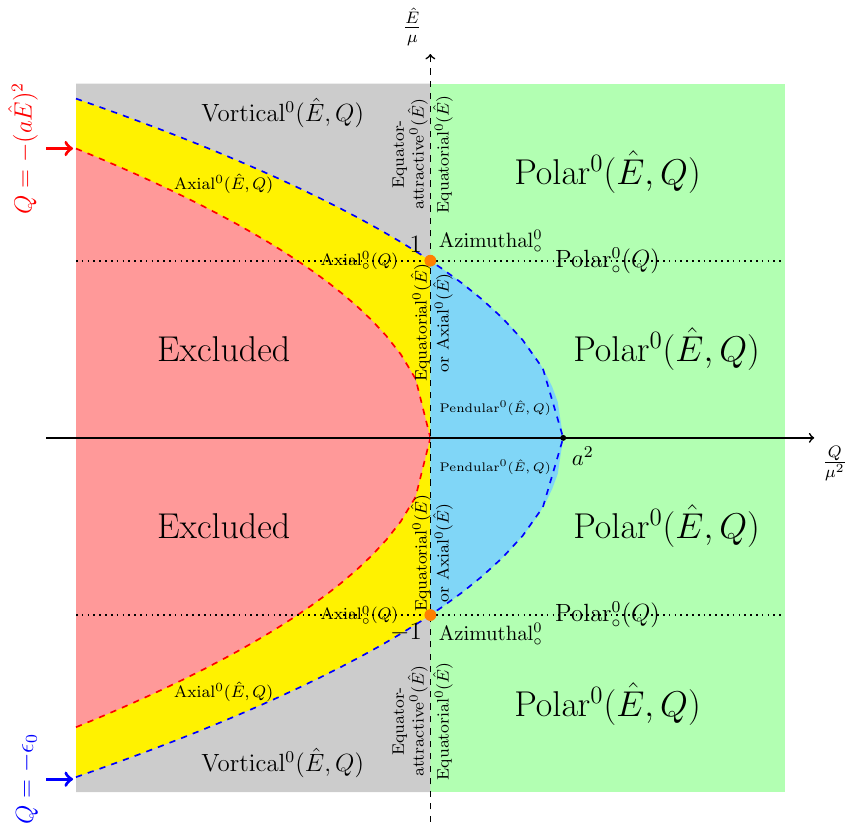}
    \caption{Polar taxonomy of $\ell= 0$ Kerr geodesics. In addition to the possible motions depicted in the figure, an axial motion is always possible for any value of $ \epsilon _0$ and $Q\geq-(a E_0 )^2$.}
    \label{fig:angularNullKerrTaxonomy}\vspace{0pt}
\end{figure}

\clearpage

\subsection{Solution to the polar integrals} 

After having classified the different types of motion allowed, we will provide manifestly real and positive explicit solutions in terms of elliptic integrals for each type of polar motion with $\ell \neq 0$, in line with the recent analysis \cite{Kapec:2019hro}. All such integrals will turn out to agree with Ref. \cite{Kapec:2019hro}, but our presentation will be slightly simpler.

The solution to the polar integrals \eqref{eqn:lambda} and \eqref{eqn:polarIntegrals} can be organized in terms of the categories of polar motion with $\ell \neq 0$: 
\begin{center}
    \begin{tabular}{|c|c|c|c|}
    \cline{2-4}
        \multicolumn{1}{c|}{ } & \textbf{Vortical} & \textbf{Equator-attractive} & \textbf{Pendular} \\\hline
       \rule{0pt}{13pt} $\boldsymbol{ \epsilon_0<0}$ & $\emptyset$ & $\emptyset$ & Pendular$(\hat E,Q)$\\\hline
       \rule{0pt}{13pt} $\boldsymbol{ \epsilon _0=0}$ & $\emptyset$ & $\emptyset$ & Pendular${}_*(Q)$\\\hline
       \rule{0pt}{13pt} $\boldsymbol{ \epsilon _0>0}$ & Vortical$(\hat E,Q)$ & Equator-attractive$(\hat E)$ & Pendular$(\hat E,Q)$\\\hline
    \end{tabular}
\end{center}

Each type of motion yields to a specific decomposition of the line integrals $\sintline$ in terms of basic integrals. In order to simplify the notations, we drop the ``$f$" indices labeling the final event and define $h \equiv  \text{sign}\,{ \cos\theta}$, $\theta_a \triangleq \arccos \sqrt{z_a}$ ($a=+,-,0$), as well as the initial and final signs $\eta_i$, $\eta$:
 \begin{align}
\eta_i \triangleq   - s_\theta^i  \,\text{sign}\,{\cos\theta_i}  ,\qquad \eta \triangleq  - (-1)^m s_\theta^i  \,\text{sign}\,{\cos\theta }.
 \end{align}
We are now ready to perform the explicit decomposition:
\begin{enumerate}
    \item \textbf{Pendular motion.}  We have $0< z_+ \leq 1$, and $\theta$ therefore belongs to the interval $\theta_+ \leq \theta \leq \pi - \theta_+$. The polar integral can be written (see Ref. \cite{Kapec:2019hro}) 
\begin{subequations}
    \begin{align}
\sint_{\cos \theta_i}^{\cos \theta} &= 2 m \left| \int_{0}^{\cos\theta_+} \right|  -\eta \left| \int_0^{\cos\theta}\right|+\eta_i \left| \int_0^{\cos \theta_i} \right|,\qquad \epsilon _0\neq 0,\label{int1}\\
\sint_{\cos \theta_i}^{\cos \theta} &= 2 m \left| \int_{0}^{\cos\theta_0} \right|  -\eta \left| \int_0^{\cos\theta}\right|+\eta_i \left| \int_0^{\cos \theta_i} \right|,\qquad \epsilon _0= 0.
\end{align}
\end{subequations}
It is useful to note that our definitions of the roots imply
\begin{equation}
     \epsilon _0\,z_-<0,\qquad \epsilon _0(z-z_-)>0,\qquad\frac{z_+}{z_-}\leq 1.
\end{equation}

    \item \textbf{Vortical motion.} We have $ \epsilon _0>0$ and $0 < z_-\leq \cos^2\theta \leq z_+ < 1$. The motion therefore never reaches the equator. The sign of $\cos\theta$ is constant and determines whether the motion takes place in the northern or the southern hemisphere. Without loss of genericity, let us focus on the northern hemisphere: $0 \leq \theta_+\leq \theta \leq \theta_- < \frac{\pi}{2}$; we denote again as $m$ the number of turning points at Mino time $\lambda$. The polar integral can be written (see Ref. \cite{Kapec:2019hro} and Appendix A of Ref. \cite{Gralla:2019ceu}):
 \begin{align}
\sint_{\cos \theta_i}^{\cos \theta} = \left( m -\eta_i \frac{1-(-1)^m}{2} \right) \left| \int_{\cos\theta_-}^{\cos\theta_+} \right|  -\eta  \left| \int_{\cos\theta_-}^{\cos\theta} \right|+ \eta_i \left| \int_{\cos\theta_-}^{\cos \theta_i} \right| .\label{int2}
 \end{align}
    
    \item \textbf{Equator-attractive motion.} This is a limit case of the vortical motion reached in the limit $z_-\to 0$, $z_+\to 2\Delta_\theta$. As detailed in Ref. \cite{Kapec:2019hro}, the turning point $z_-=0$ corresponds to a non-integrable singularity of the polar integrals and the motion exhibits consequently at most one turning point at $z_+=2\Delta_\theta$, leading to the line-integral decomposition
    \begin{equation}
    \sint_{\cos \theta_i}^{\cos \theta} = \eta\abs{\int_{\cos\theta_+}^{\cos\theta}}-\eta_i\abs{\int_{\cos\theta_+}^{\cos\theta_i}}.
    \end{equation}
\end{enumerate}
In all cases but the equator-attractive case, the polar motion is periodic. Denoting by $\Lambda_\theta$ its period, one can easily give an explicit formula for the number of turning points $m$ as a function of the Mino time:
\begin{equation}
    m(\lambda)=\left\lbrace\begin{array}{ll}
    \rule{0pt}{6pt}\left\lfloor \frac{2}{\Lambda_\theta}(\lambda-\lambda_i^\theta)+\frac{1}{2}\right\rfloor, & Q>0\\
    \rule{0pt}{16pt}\left\lfloor \frac{2}{\Lambda_\theta}(\lambda-\lambda_i^\theta)\right\rfloor + \left\lfloor \frac{2}{\Lambda_\theta}(\lambda_i^\theta-\lambda_i)\right\rfloor + \frac{3-s^i_\theta}{2}, & Q<0
    \end{array}\right. 
\end{equation}
with $\lambda_i^\theta\triangleq\lambda_i-s^i_\theta\int_{0}^{\cos\theta_i}\frac{\dd \cos \theta}{\sqrt{\Theta(\cos^2 \theta)}}$ and where the floor function is defined as $\left\lfloor x \right\rfloor\triangleq\max\qty{n\in\mathbb{Z}|n\leq x}$. For the equator-attractive case, one has simply $m(\lambda)=\theta(\lambda-\lambda_i^\theta)$ where $\theta$ is here the Heaviside step function.

The integrals introduced above are solved explicitly in Appendix \ref{app:basicIntegrals}. For each case, the corresponding solutions are detailed below and schematically depicted in Fig. \ref{fig:polarClassesFigure}.

\begin{figure}[!htb]
    \centering
    \begin{tabular}{cc}
      \includegraphics[width=4cm]{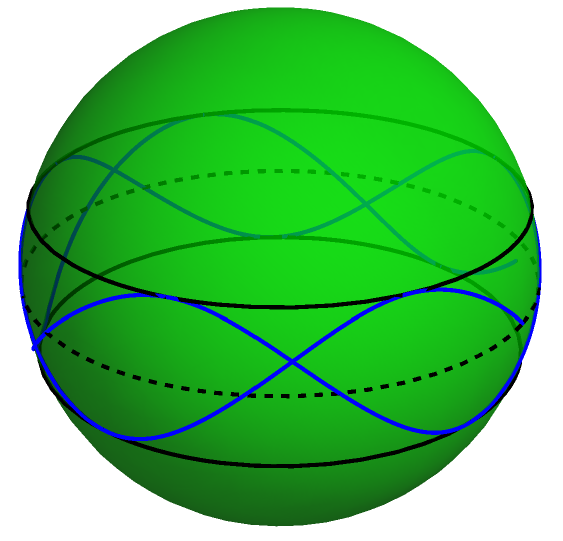} & \includegraphics[width=4cm]{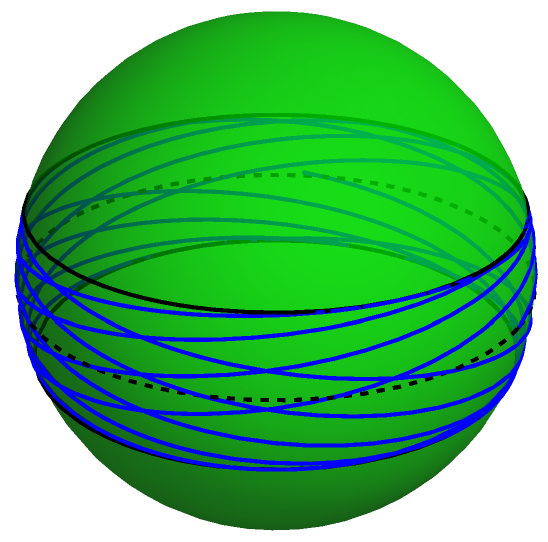} \\ (a) Pendular$(\hat E, Q)$ & (b) Pendular${}_\circ(Q)$\vspace{0.6cm}\\
        \includegraphics[width=4cm]{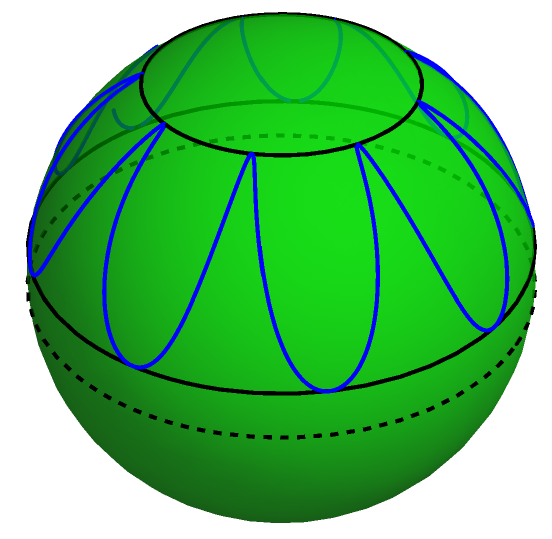} & \includegraphics[width=4cm]{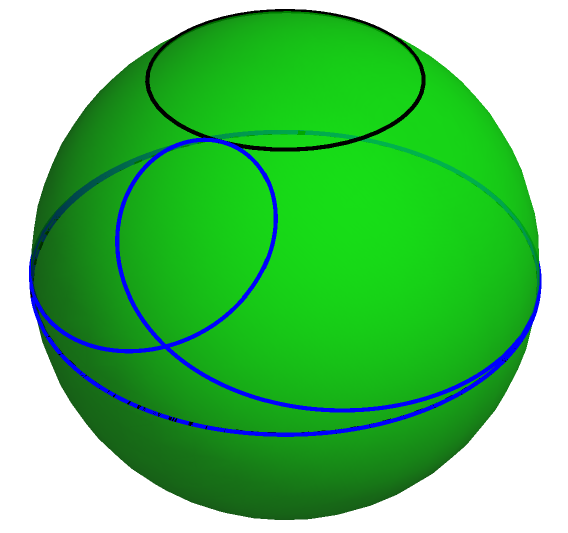} \\ (c) Vortical$(\hat E,Q)$ & (d) Equator-attractive$(\hat E)$
    \end{tabular}
    \caption{Angular taxonomy of $\ell\neq 0$ Kerr geodesics. The angular behavior is depicted in spherical coordinates on the unit sphere: the polar angle is $\theta(\lambda)$, and the azimuthal angle is the purely angular part of the Kerr azimuthal angle $(\ell- a \hat E) (\lambda-\lambda_i) + \ell \Phi_\theta (\lambda)$.}
    \label{fig:polarClassesFigure}
\end{figure}

\paragraph{Pendular$(\hat E,Q)$ motion.} The motion exhibits a positive Carter constant $Q$ and can occur for any $ \epsilon _0\neq 0$; our definition of the roots $z_\pm$ allows us to treat simultaneously the two cases $ \epsilon _0<0$ and $ \epsilon _0>0$, which is a simplification with respect to the analysis carried out in Ref. \cite{Kapec:2019hro}. The period of the polar motion (comprising two turning points) in Mino time is given by 
\begin{equation}
\Lambda_\theta = 4  \int_0^{\cos\theta_+} \frac{\text{d} \cos\theta}{\sqrt{\Theta(\cos^2\theta)}}\triangleq 4 I^{(0)}(\sqrt{z_+}) =\frac{4}{\sqrt{- \epsilon _0z_-}}K\qty(\frac{z_+}{z_-}).
\end{equation}
Using the basic integrals of Appendix \ref{app:ellipticFunctions}, one can write \eqref{eqn:lambda} as
\begin{align}
    \lambda-\lambda_i&=\frac{1}{\sqrt{- \epsilon _0z_-}}\left[2mK\qty(\frac{z_+}{z_-})+s^i_\theta(-1)^mF\qty(\Psi^+(\cos\theta),\frac{z_+}{z_-})\right.\nonumber\\
    &~\left.-s_\theta^iF\qty(\Psi^+(\cos\theta_i),\frac{z_+}{z_-})\right]\label{eqn:pendularLambda}
\end{align}
where we define $\Psi^+(x)\triangleq\arcsin\qty(\frac{x}{\sqrt{z_+}})$. 
Using \eqref{eqn:inversionEpsNeq0}, one can invert \eqref{eqn:pendularLambda} as
\begin{equation}
    \cos\theta=s^i_\theta(-1)^m\sqrt{z_+} \text{sn} \qty(\sqrt{- \epsilon _0 z_-}\qty(\lambda-\lambda_i^\theta)-2mK\qty(\frac{z_+}{z_-}),\frac{z_+}{z_-})
\end{equation}
where we introduce
\begin{align}
        \lambda_i^\theta&\triangleq\lambda_i-\frac{s^i_\theta}{\sqrt{- \epsilon _0z_-}}F\qty(\Psi^+(\cos\theta_i),\frac{z_+}{z_-}).
\end{align}
This expression matches with Eq. (38) of Ref. \cite{Fujita:2009bp}. Using the periodicity property \eqref{per} of the elliptic sine, we can further simplify it to 
\begin{equation}
\cos\theta(\lambda)=s^i_\theta \sqrt{z_+} \text{sn}\qty(  \sqrt{- \epsilon _0 z_-}(\lambda-\lambda_i^\theta), \frac{z_+}{z_-}).\label{eq:costh}
\end{equation}
It consistently obeys $\cos\theta(\lambda_i)=\cos\theta_i$ and $ \text{sign}\,{\cos\theta'(\lambda_i)}=s_\theta^i$. This formula agrees with (53) of Ref. \cite{Kapec:2019hro} but it is written in a simpler form. We also obtain
\begin{align}
T_\theta &=\frac{-2z_+}{\sqrt{- \epsilon _0 z_-}} \left[ 2 m E' \qty(\frac{z_+}{z_-}) +( \pm_\theta) E'\qty(\Psi^+(\cos\theta),\frac{z_+}{z_-})\right.\nonumber\\
&\left.~- s_\theta^i  E'\qty(\Psi^+(\cos\theta_i),\frac{z_+}{z_-} )\right], 
\end{align}
\begin{align}
\Phi_\theta &=  \frac{1}{\sqrt{- \epsilon _0 z_-}} \left[ 2 m \Pi \qty(z_+,\frac{z_+}{z_-}) +( \pm_\theta) \Pi\qty(z_+,\Psi^+(\cos\theta),\frac{z_+}{z_-})\right.\nonumber\\
&~\left.- s_\theta^i  \Pi\qty(z_+,\Psi^+(\cos\theta_i) ,\frac{z_+}{z_-})\right]-(\lambda-\lambda_i).
\end{align}
where $\lambda-\lambda_i$ is given by \eqref{eqn:pendularLambda}. All quantities involved are manifestly real. These final expressions agree with Ref. \cite{Kapec:2019hro}.

\paragraph{Pendular${}_\circ(Q)$ motion.} We now consider the critical case $|\hat E| = \mu$. The period of the polar motion is
\begin{equation}
    \Lambda_\theta=4 I^{(0)}(\sqrt{z_0})=2\pi\sqrt{\frac{z_0}{Q}}.
\end{equation}
In this critical case, \eqref{eqn:lambda} leads to 
\begin{equation}
    \lambda-\lambda_i=\sqrt{\frac{z_0}{Q}}\qty[m\pi+s^i_\theta(-1)^m\arcsin{\frac{\cos\theta}{\sqrt{z_0}}}-s^i_\theta\arcsin{\frac{\cos\theta_i}{\sqrt{z_0}}}],
\end{equation}
which can be simply inverted as
\begin{equation}
    \cos\theta=s_\theta^i\sqrt{z_0}\,\sin\qty(\sqrt{\frac{Q}{z_0}}(\lambda-\lambda_i^\theta)),\qquad\lambda_i^\theta\triangleq\lambda_i-\sqrt{\frac{z_0}{Q}}\arcsin{\frac{\cos\theta_i}{\sqrt{z_0}}}.
\end{equation}
The other polar integrals are

    \begin{align}
    T_\theta&=\frac{1}{2}\qty{z_0(\lambda-\lambda_i)-\sqrt{\frac{z_0}{Q}}\qty[(\pm_\theta)\cos\theta\sqrt{z_0-\cos^2\theta}-s^i_\theta\cos\theta_i\sqrt{z_0-\cos^2\theta_i}]},\nonumber\\
    \Phi_\theta&=\sqrt{\frac{z_0}{Q(1-z_0)}}\bigg[m\pi+(\pm_\theta)\arcsin\qty(\sqrt{\frac{1-z_0}{z_0}}\cot\theta)\nonumber\\
    &\quad-s^i_\theta\arcsin\qty(\sqrt{\frac{1-z_0}{z_0}}\cot\theta_i)\bigg]-(\lambda-\lambda_i).
\end{align}
 
\paragraph{Vortical$(\hat E,Q)$ motion.} 
The period in Mino time is given by 
\begin{equation}
\Lambda_\theta = 2 \left| \int_{\cos\theta_-}^{\cos\theta_+} \frac{\text{d} \cos\theta}{\sqrt{\Theta(\cos^2\theta)}} \right|=\frac{2}{\sqrt{ \epsilon _0 z_+}} K\qty(1-\frac{z_-}{z_+}).
\end{equation}
Using the basic integrals of Appendix \ref{app:basicIntegrals}, one has
\begin{align}
    \lambda-\lambda_i&=\frac{1}{\sqrt{ \epsilon _0z_+}}\left[\qty(m-h s^i_\theta\frac{1-(-1)^m}{2})K(\tilde m)-s^i_\theta(-1)^mF\qty(\Psi^-(\cos\theta),\tilde{m})\right.\nonumber\\
    &~\left.+s^i_\theta F\qty(\Psi^-(\cos\theta_i),\tilde m)\right]
\end{align}
where
\begin{equation}
    \tilde{m}\triangleq 1-\frac{z_-}{  z_+},\qquad\Psi^-(x)=\arcsin{\sqrt{\frac{z_+-x^2}{z_+-z_-}}}.
\end{equation}
Using the inversion formula \eqref{eqn:inversionFormula} and the periodicity property \eqref{eqn:perDn}, we obtain
\begin{equation}
    \cos\theta=h\sqrt{z_+} \text{dn}\qty(\sqrt{ \epsilon _0z_+}(\lambda-\lambda_\theta^i),\tilde m)
\end{equation}
with
\begin{equation}
    \lambda_i^\theta\triangleq\lambda_i+\frac{s^i_\theta h}{\sqrt{ \epsilon _0z_+}}F\qty(\Psi^-(\cos\theta_i),\tilde m).
\end{equation}
Again, one has $\cos\theta(\lambda_i)=\cos\theta_i$ and $ \text{sign}\,{\cos\theta'(\lambda_i)}=s_\theta^i$. The two other polar integrals are
\begin{align}
    T_\theta&=\sqrt{\frac{z_+}{ \epsilon _0}}\left[\qty(m-h s^i_\theta\frac{1-(-1)^m}{2})E(\tilde m)-(\pm_\theta)E\qty(\Psi^-(\cos\theta),\tilde{m})\right.\nonumber\\
    &~\left.+s^i_\theta E\qty(\Psi^-(\cos\theta_i),\tilde m)\right],\\
    \Phi_\theta&=\frac{1}{(1-z_+)\sqrt{ \epsilon _0z_+}}\left[\qty(m-h s^i_\theta\frac{1-(-1)^m}{2})\Pi\qty(\frac{z_--z_+}{1-z_+},\tilde m)\right.\nonumber\\&~\left.-(\pm_\theta)\Pi\qty(\frac{z_--z_+}{1-z_+},\Psi^-(\cos\theta),\tilde{m})+s^i_\theta \Pi\qty(\frac{z_--z_+}{1-z_+},\Psi^-(\cos\theta_i),\tilde m)\right]\nonumber\\& ~-(\lambda-\lambda_i)
\end{align}
in agreement with the results of Ref. \cite{Kapec:2019hro}.

\paragraph{Equator-attractive$(\hat E)$ motion.}
This is the only polar motion which is not periodic. One has
\begin{equation}
    \lambda-\lambda_i=\frac{h}{\sqrt{ \epsilon _0z_+}}\qty[-(\pm_\theta)\, \text{arctanh}\sqrt{1-\frac{\cos^2\theta}{z_+}}+s^i_\theta\, \text{arctanh}\sqrt{1-\frac{\cos^2\theta_i}{z_+}}]
\end{equation}
leading to
\begin{subequations}
    \begin{align}
    \cos\theta&=h\sqrt{z_+}\,\text{sech}\qty(\sqrt{ \epsilon _0z_+}(\lambda-\lambda_i^\theta)),\\
    \lambda_i^\theta&\triangleq\lambda_i+\frac{s^i_\theta h}{\sqrt{ \epsilon _0z_+}} \text{arctanh}\sqrt{1-\frac{\cos^2\theta_i}{z_+}}.
\end{align}
\end{subequations}
The polar integrals are
\begin{subequations}
    \begin{align}
    T_\theta&=\frac{h}{\sqrt{ \epsilon _0}}\qty[-(\pm_\theta)\sqrt{z_+-\cos^2\theta}+s^i_\theta\sqrt{z_+-\cos^2\theta_i}],\\
    \Phi_\theta&=\frac{h}{\sqrt{ \epsilon _0(1-z_+)}}\bigg[-(\pm_\theta)\arctan\sqrt{\frac{z_+-\cos^2\theta}{1-z_+}}\nonumber\\
    &\quad+s^i_\theta\arctan\sqrt{\frac{z_+-\cos^2\theta_i}{1-z_+}}\bigg].
    \end{align}
\end{subequations}
This agrees with the results of Ref. \cite{Kapec:2019hro}.

\section{Classification of near-horizon motion for high spin Kerr black holes}

In this section, we derive a complete classification of timelike and null geodesic trajectories lying in the near-horizon region of a quasi extremal Kerr black hole. We will provide explicit manifestly real analytic expressions for all geodesic trajectories. We will present the classification in terms of the geodesic energy, angular momentum, and Carter constant $Q$. We will also illustrate each radial motion in NHEK with a Penrose diagram. 

Partial classifications were performed in Refs. \cite{Compere:2017hsi} and \cite{Kapec:2019hro}. In Ref. \cite{Compere:2017hsi}, equatorial timelike prograde incoming (i.e. that originate from the Kerr exterior geometry) geodesics were classified. Such geodesics reach the spatial boundary of the near-horizon region at infinite past proper time and therefore physically reach the asymptotically flat Kerr region once the near-horizon is glued back to the exterior Kerr region. It turns out that bounded geodesics in the near-horizon Kerr region also arise in the study of gravitational waves since they correspond to the end point of the transition motion \cite{Compere:2019cqe}. Timelike outgoing geodesics originating from the white hole horizon and reaching the near-horizon boundary are also relevant for particle emission within the near-horizon region \cite{Gralla:2017ufe}. In addition, null outgoing  geodesics are relevant for black hole imaging around high-spin black holes \cite{Gralla:2019drh}. 

The generic non-equatorial geodesics were obtained in Ref. \cite{Kapec:2019hro}. In particular, real forms were obtained for each angular integral involved in geodesic motion. However, zero-measure sets of parameters were discarded. These zero-measure sets include in particular the separatrix between bounded and unbounded radial motion which plays a key role in EMRIs.  

 In the following, we do not make any assumption on the geodesic parameters. We will treat both timelike and null geodesics, prograde or retrograde, and with any boundary conditions. Without loss of genericity, we will consider future-directed orbits. Past-directed geodesics can be obtained from future-directed geodesics using the $\mathbb Z_2$ map: 
 \begin{align}
T \rightarrow -T,\quad \Phi \rightarrow -\Phi,\quad  E\rightarrow -E,\quad  \ell \rightarrow -\ell,\label{PTflip}
 \end{align}
which will play an important role in Sec. \ref{sec:classes}. We will denote it as the $\uparrow\!\downarrow$-flip.

\subsection{NHEK}
Future orientation of the geodesic is equivalent to $\dd T/\dd \lambda > 0$ or 
 \begin{align}
E+ L_0  R > 0. 
 \end{align}
Future-oriented geodesics with $ L_0  = 0$ have $E >0$. For $ L_0  \neq 0$, we define the critical radius as in \cite{Kapec:2019hro}:
 \begin{align}
R_c = -\frac{E}{ L_0 }. \label{defRc}
 \end{align}
Future-orientation of the orbit requires 
 \begin{align}\label{consRc}
R < R_c \quad \text{for} \quad  L_0  < 0, \quad \text{and} \quad R > R_c \quad \text{for} \quad  L_0  > 0. 
 \end{align} 

 \subsubsection{Polar behavior}

The results derived in Sec. \ref{sec:polar_kerr} in the context of generic Kerr still hold in the near-horizon high-spin limit which is obtained by the scaling limit $\lambda \rightarrow 0$ taken in the near-horizon coordinates \eqref{eq:cvn}. We anticipate that the results also hold in the distinct near-NHEK limit $\lambda \rightarrow 0$ taken in the near-horizon coordinates \eqref{eq:cvnn}. Due to the high-spin limit, the following substitution can be made:
\begin{align}
&a \mapsto M, \qquad \hat E \mapsto \frac{\ell}{2M},\qquad   \epsilon _0 \mapsto \mathcal{C}_\circ\triangleq\frac{\ell^2-\ell_\circ^2}{4}, \\
&\Theta(z) \mapsto v_\theta(z),\qquad  \hat\Phi_\theta-\frac{1}{4}\hat{T}_\theta \mapsto \Phi_\theta.
\end{align}
Notice that the dependence on $\hat E$ of $ \epsilon _0$ has been changed into a dependence in $\ell$, the Kerr energy being the same at zeroth order on $\lambda$ for all trajectories. Therefore, the quadratic term of the polar potential vanishes at the critical value $\ell_\circ$ of the angular momentum $\ell$.

One of the most striking features of the near-horizon polar motion is that $Q$ is non-negative as a consequence of the reality of polar motion, as noticed in Ref. \cite{Kapec:2019hro}: 
 \begin{proposition}\label{propQ}
\begin{equation}
    \forall z\in[0,1] :v_\theta(z)\geq0\Rightarrow Q\geq0.
\end{equation}
 \end{proposition}
\noindent \emph{Proof.} This property is a consequence of the dependence on $Q$ of $\mathcal{C}$ defined in \eqref{eqn:Carter}. Indeed,  using the fact that $z=\cos^2\theta\in[0,1]$ one can write
 \begin{align} 
 Q = \mathcal C + \frac{3}{4}\ell^2 -M^2 \mu^2 \geq \mathcal C + (1-\Lambda^{-2})\ell^2 -M^2 \mu^2 \geq v_\theta (z) \geq 0. 
 \end{align} 
A direct consequence is that the near-horizon polar motion cannot be vortical and is consequently either equatorial, pendular, polar or axial. 
We note that the condition $ \epsilon _0\geq\ell^2$ is never obeyed in the near-horizon case after using the definition \eqref{defeps0}, $a=M$ and $\hat E=\frac{\ell}{2M}$. The equator-attractive class is therefore discarded. The resulting polar classes are listed in Table \ref{table:taxPolar} and the phase space is represented in Fig.\ref{fig:angularTaxonomy}.

\subsubsection{Radial behavior} 
The radial behavior for generic inclined orbits can be solved using the equatorial results \cite{Compere:2017hsi} thanks to the following observation:
 \begin{proposition}\label{thm:equivNHEK}
The radial integrals  $T_R^{(i)}(R)$ $(i=0,1,2)$ only depend upon the NHEK energy $E$ and angular momentum $\ell$ while all the dependence upon the mass $\mu$ and Carter constant $Q$ is through $\ell_* = \frac{2}{\sqrt{3}}\sqrt{M^2 \mu^2 +Q}$. 
 \end{proposition}
This simple observation has far-reaching consequences. For any timelike geodesic with $Q \neq 0$, one could directly reuse the classification established in Ref. \cite{Compere:2017hsi}, modulo the substitution $\frac{2}{\sqrt{3}}M \mu \to \ell_*$ in every expression encountered. Moreover, null geodesics with $\mu = 0$ have $Q \geq 0$ from Proposition \ref{propQ}. We can therefore reuse the classification established in Ref. \cite{Compere:2017hsi} to classify null geodesics modulo the substitution $M \mu \to Q$ in every expression encountered. Overall, all radial integrals can be described in closed form for all cases by keeping the dependence upon $\ell_*$ or, equivalently, upon the Casimir invariant $\mathcal C$.

Since the equatorial taxonomy of Ref. \cite{Compere:2017hsi} did not consider bounded orbits and only considered $\ell > 0$, we will expand the taxonomy to the generic case. The generic classification can be achieved by studying the roots of $v_R$ and the range of $R$ where $v_R \geq 0$. We only consider orbits outside the horizon, $R>0$. There are three broad categories depending on the angular momentum: the supercritical case $\vert \ell  \vert > \ell_*$ or equivalently $\mathcal C < 0$, the critical case $\vert \ell \vert= \ell_*$ or equivalently $\mathcal C = 0$ and the subcritical case $0 \leq \vert \ell \vert  < \ell_*$ or $\mathcal C > 0$. 

The relative position of the critical radius \eqref{defRc} with respect to the roots of $v_R$ may restrict the allowed classes of future-oriented orbits. As a result of \eqref{consRc}, subcritical $\ell^2 <\ell^2_*$ orbits have either $R_+ < R_c$ for $\ell < 0$ or $R_c < 0$ for $\ell > 0$, and all orbits are future oriented. Critical orbits $\ell^2 =\ell^2_*$ have either $R_c < 0$ for $\ell = \ell_*$ or $R_c > R_0$ for $\ell = -\ell_*$. This restricts the classes of orbits. Supercritical orbits $\ell^2 > \ell^2_*$ with $E,\ell>0$ are future directed. Supercritical orbits with $E>0$, $\ell<0$ admit $R_- < R_c < R_+$, and only bounded orbits with $R \leq R_-$ are admissible. Finally, supercritical orbits with $E<0$ and $\ell > 0$ obey $R  \geq R_+ > R_c $ and are therefore deflecting. 

After a simple analysis, we reach the following taxonomy, displayed in Table \ref{table:taxNHEK} and in Fig. \ref{fig:taxonomyEquator}. In comparison with Ref. \cite{Compere:2017hsi},  the classes Outward$(E,\ell)$, Outward$_*(E)$,  Bounded$_>(E,$ $\ell)$, Bounded$^-_*(E)$, and Bounded$_<(E,\ell)$ are new, while all other classes with $\ell > 0$ appeared in Ref. \cite{Compere:2017hsi}. The class $\text{Osculating}(E,\ell)$ is now better called $\text{Def}\text{lecting}(E,\ell)$. The classes with $\ell = \pm \ell_*$ will be denoted with a subscript $_*$. The Spherical$_*$ orbit with $\ell = \ell_*$ is also the prograde ISSO. For $\ell\geq 0$, the conformal diagrams corresponding to those orbits are depicted in Fig. \ref{fig:penrose_NHEK} and their explicit forms are given in Appendix \ref{app:equatorial}. Past-oriented geodesics (not depicted) are obtained from a central symmetry around the origin $E=\ell=0$ as a result of the $\uparrow\!\downarrow$-flip \eqref{PTflip}. 

\begin{table}[!htb]
    \centering
    \begin{tabular}{|c|c|c|c|}
    \hline
        \textbf{Angular momentum} & \textbf{Carter constant} & \textbf{Polar range} & \textbf{Denomination}\\\hline
        $\ell=0$ ($\mathcal{C}_\circ=-\ell_\circ^2/4$) &  $Q=0$ & $z=0,~1$ & \begin{tabular}{c} \rule{0pt}{13pt}Equatorial${^0}$\\\rule{0pt}{13pt}Axial${}^0$\end{tabular}\\\cline{2-4}
         & $Q> 0$ & $z=1$ & \rule{0pt}{13pt}Axial${}^0(Q)$\\\hline
         $0<\ell<\ell_\circ$ ($-\frac{\ell_\circ^2}{4}<\mathcal{C}_\circ<0$) & $Q=0$ & $z=0$ & \rule{0pt}{13pt}Equatorial$(\ell)$\\\cline{2-4}
          & $Q>0$ & $0\leq z \leq z_+$ & \rule{0pt}{13pt}Pendular$(Q,\ell)$\\\hline
        $\ell=\ell_\circ$ ($\mathcal{C}_\circ=0$) &  $Q=0$ & $z=0$ &\rule{0pt}{13pt} \rule{0pt}{13pt}Equatorial${}_\circ$\\\cline{2-4}
         & $Q>0$ & $0\leq z \leq z_0$ & \rule{0pt}{13pt}Pendular${}_\circ(Q)$\\\hline
         $\ell>\ell_\circ$ ($\mathcal{C}_\circ>0$) & $Q>0$ & $0\leq z \leq z_+$ & \rule{0pt}{13pt}Pendular$(Q,\ell)$\\\hline
         \end{tabular}
    \caption{Polar taxonomy of near-horizon  geodesics  with $\ell\geq 0$. The orbits with $\ell<0$ are obtained from $\ell>0$ by flipping the sign of $\ell$ with the rest unchanged.}
    \label{table:taxPolar}
\end{table}

\begin{figure}[!hbt]
    \centering \vspace{-0.2cm}
    \includegraphics[width=11.8cm]{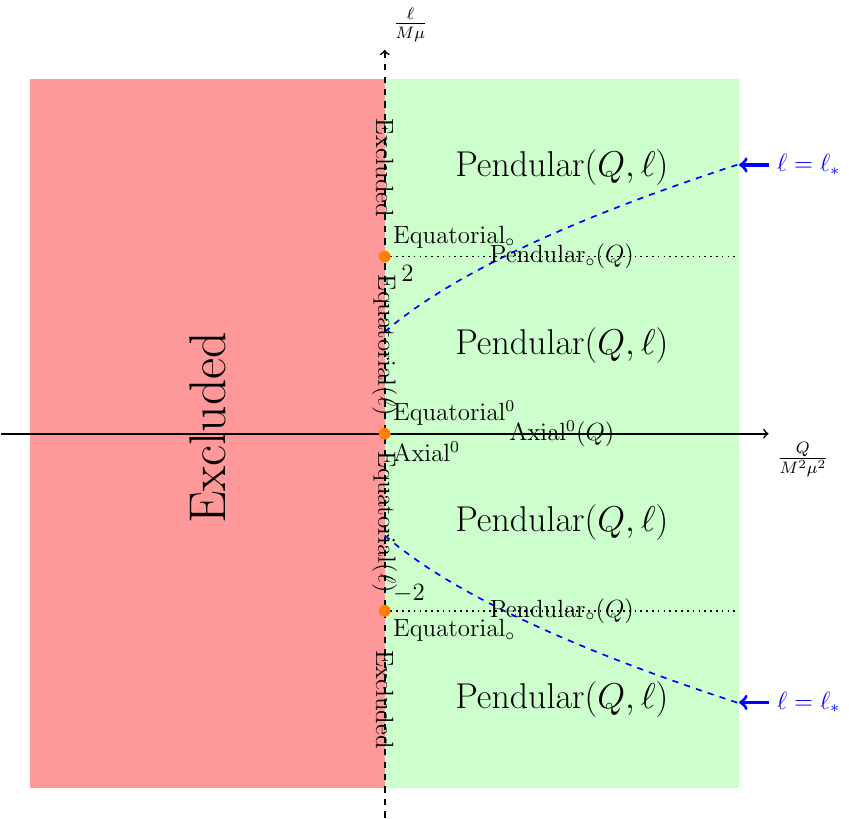}
    \vspace{-0.1cm}\caption{Polar taxonomy of near-horizon geodesics. For clarity, the scale is not respected on the horizontal axis. The dashed blue curves $\ell^2=\ell^2_*$ represent the position of the spherical orbits in parameter space. This figure contrasts with Figure \ref{fig:angularKerrTaxonomy}.}\vspace{-1cm}
    \label{fig:angularTaxonomy}
\end{figure}

\clearpage

\begin{table}[!tbh]    \centering
\begin{tabular}{|c|c|c|c|}\hline
\rule{0pt}{13pt}\textbf{Angular momentum (and Casimir)} & \textbf{Energy} & \textbf{Radial range} & \textbf{Denomination} \\ \hline
Supercritical: $\ell > \ell_*$ ($-\ell^2 < \mathcal C < 0$)& $E > 0$ & $0 \leq R \leq \infty$ & \begin{tabular}{c} \rule{0pt}{13pt}Plunging$(E,\ell)$\\\rule{0pt}{13pt}Outward$(E,\ell)$\end{tabular}\\\cline{2-4}
 & $E = 0$ & $0 < R \leq \infty$ & \rule{0pt}{13pt}Marginal$(\ell)$ \\ \cline{2-4}
& $E < 0$ & $R_+ \leq R \leq \infty$ & \rule{0pt}{13pt}Def\mbox{}lecting$(E,\ell)$ \\ \hline
Critical: $\ell = \ell_*$ ($\mathcal C = 0$) & $E > 0$ & $0 \leq R \leq \infty$ & \begin{tabular}{c} \rule{0pt}{13pt}Plunging$_*(E)$\\\rule{0pt}{13pt}Outward$_*(E)$\end{tabular} \\ \cline{2-4}
& $E = 0$ & $0 < R \leq \infty$ & \rule{0pt}{13pt}Spherical$_*$ (ISSO) \\  \hline
Subcritical: $0 \leq \ell^2 < \ell_*^2$ ($0 < \mathcal C \leq \frac{3 \ell_*^2}{4}$) & $E > 0$ & $0 \leq R \leq R_+$ & \rule{0pt}{13pt}Bounded$_<(E,\ell)$ \\ \hline
Critical: $\ell = -\ell_*$ ($\mathcal C = 0$) & $E > 0$ & $0 \leq R \leq R_0$ &\rule{0pt}{13pt} Bounded$^-_*(E)$ \\ \hline
Supercritical: $\ell < - \ell_*$ ($-\ell^2 < \mathcal C < 0$)& $E > 0$ & $0 \leq R \leq R_-$ &\rule{0pt}{13pt} Bounded$_>(E,\ell)$\\\hline 

\end{tabular}\caption{Radial taxonomy of future-oriented geodesics in NHEK. }\label{table:taxNHEK}
\end{table}
\begin{figure}[!hbt]
    \centering
    \includegraphics[width=0.8\textwidth]{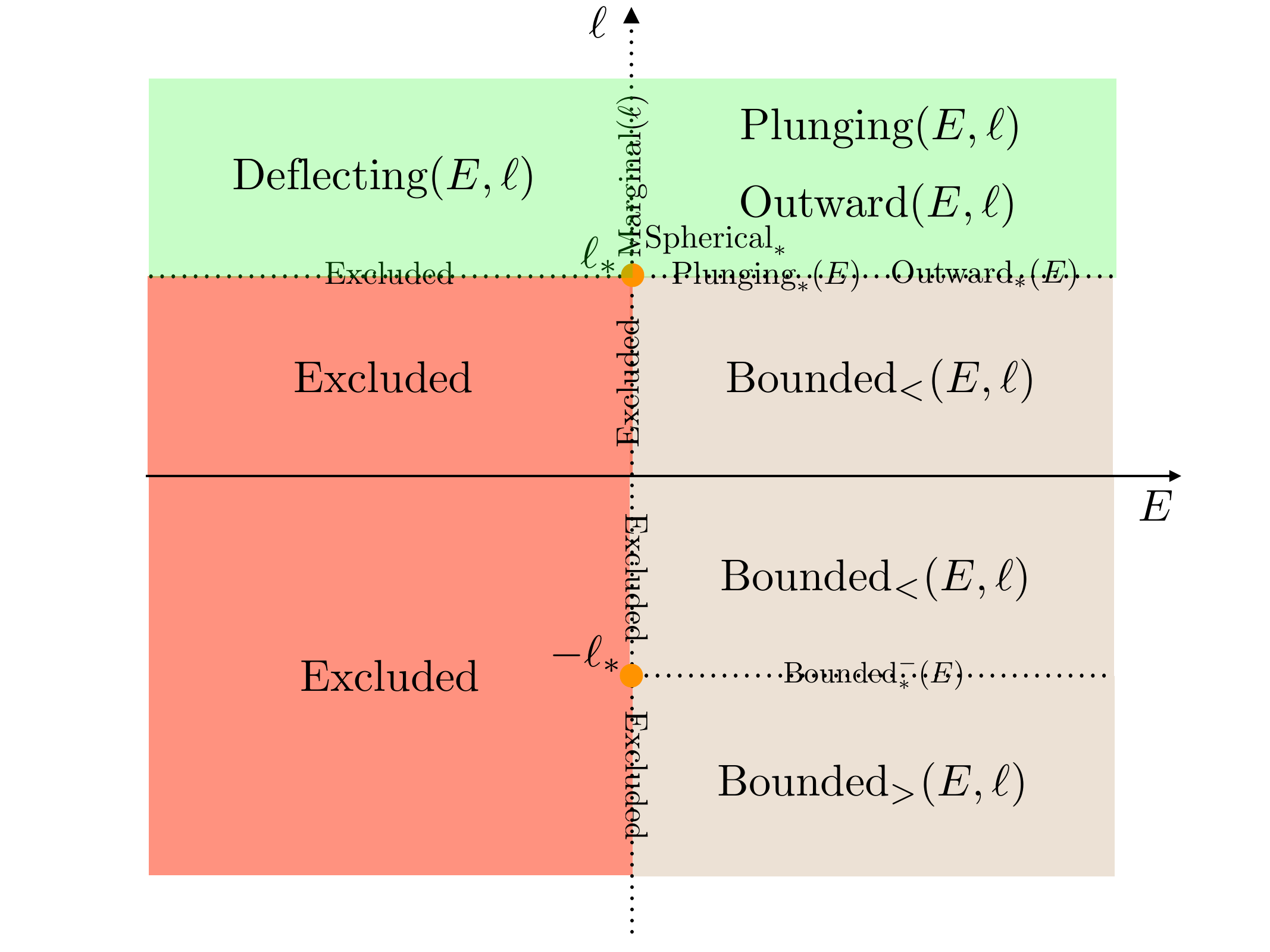} 
    \caption{Radial taxonomy of future oriented geodesics in NHEK. For equatorial geodesics, $\ell_*=\frac{2}{\sqrt{3}}M \mu$, while for orbits with inclination, $\ell_* = \frac{2}{\sqrt{3}} \sqrt{M^2\mu^2 + Q}$.}
    \label{fig:taxonomyEquator}
\end{figure}

\begin{figure}[!htbp]
    \centering
   \begin{tabular}{cccc}
    \includegraphics[width=3.5cm]{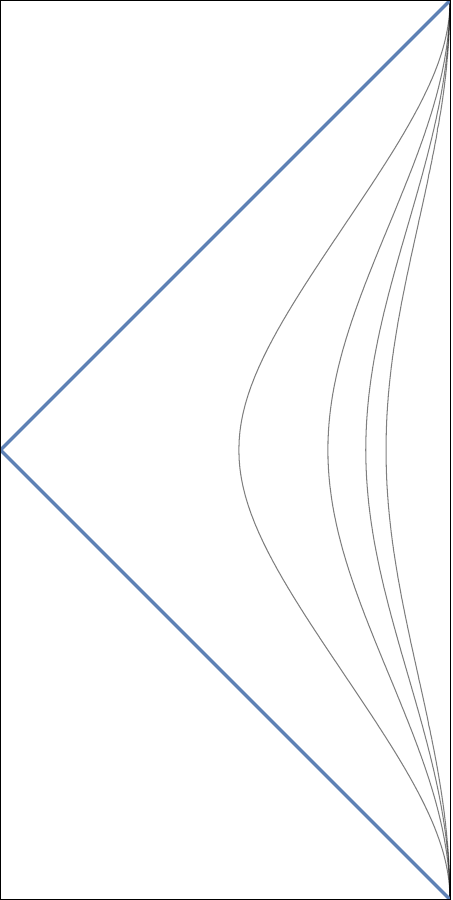} & \includegraphics[width=3.5cm]{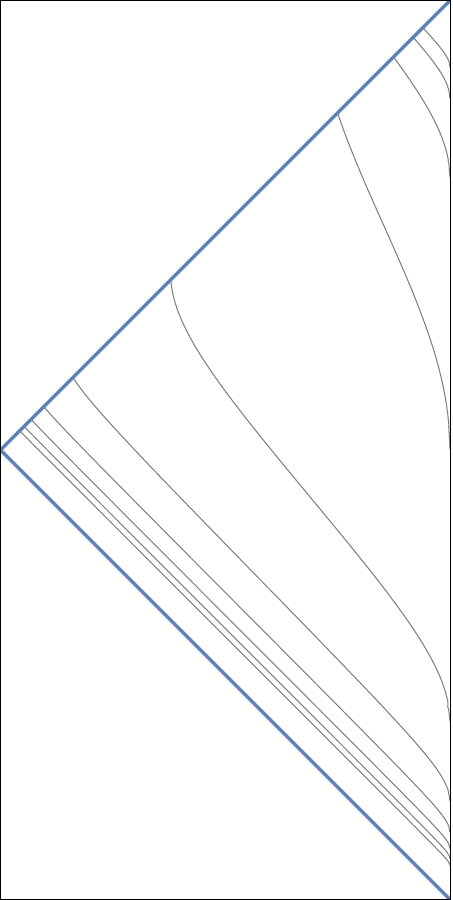} & \includegraphics[width=3.5cm]{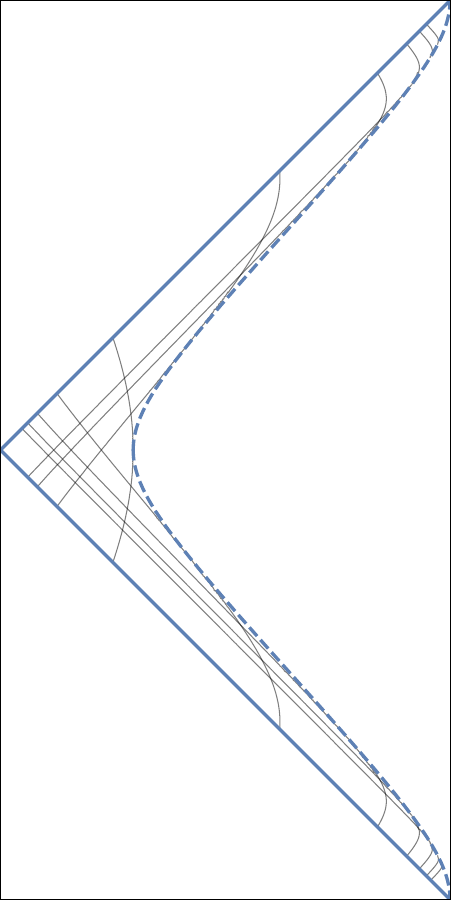} &
    \includegraphics[width=3.5cm]{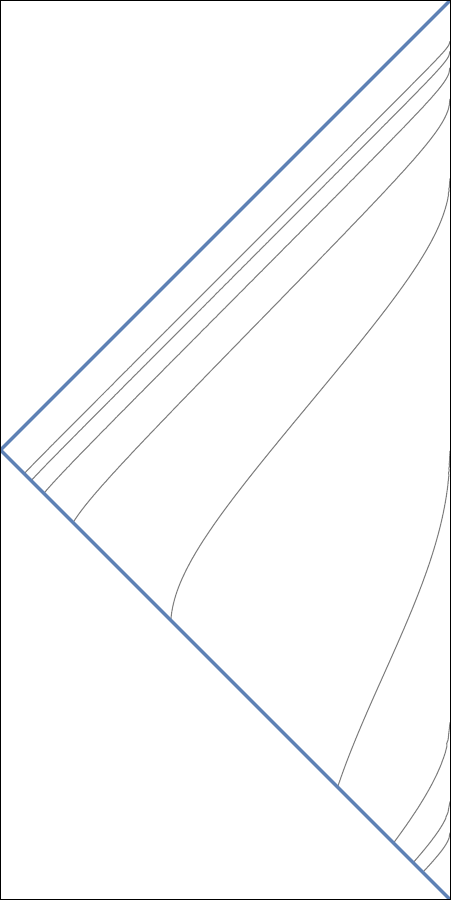}\\
    (a) Spherical${}_*(ISSO)$ & (b) Plunging${}_*(E)$ & (c) Bounded${}^-_*(E)$ & \begin{tabular}{c}
    (d) Outward$_*(E)$,\\
    Outward$(E,\ell)$
    \end{tabular} \\
    \rule{0pt}{30pt} & & & \\
    \includegraphics[width=3.5cm]{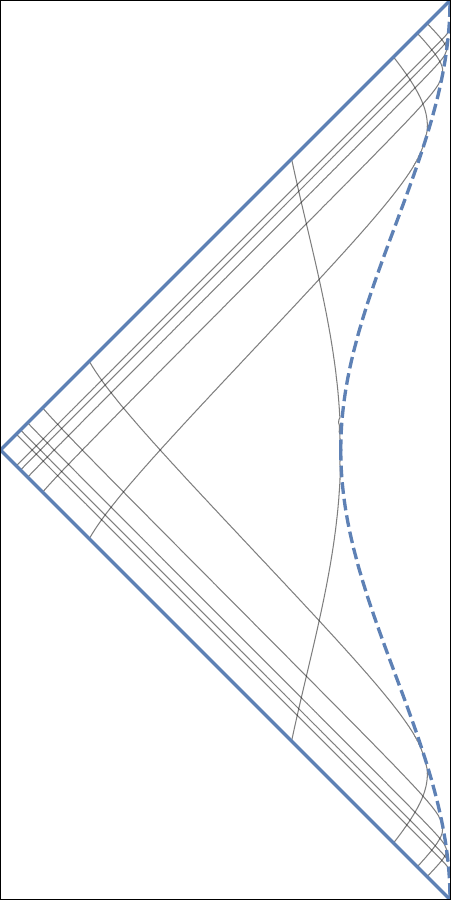} &
    \includegraphics[width=3.5cm]{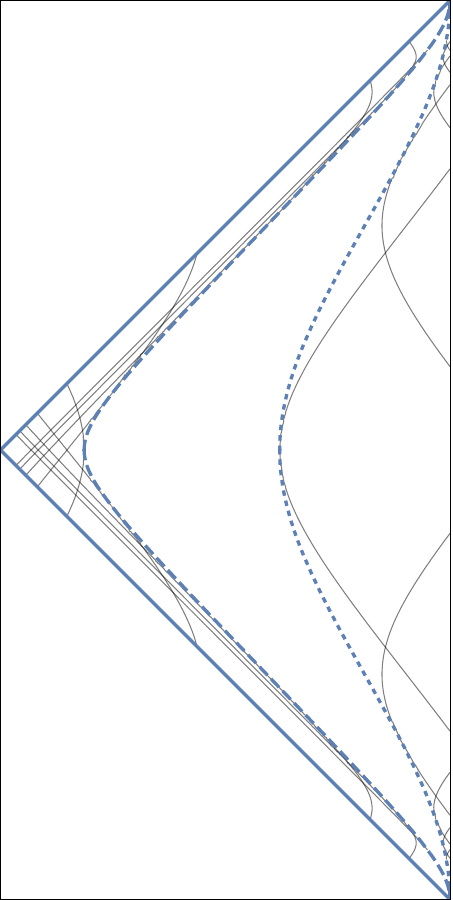} & \includegraphics[width=3.5cm]{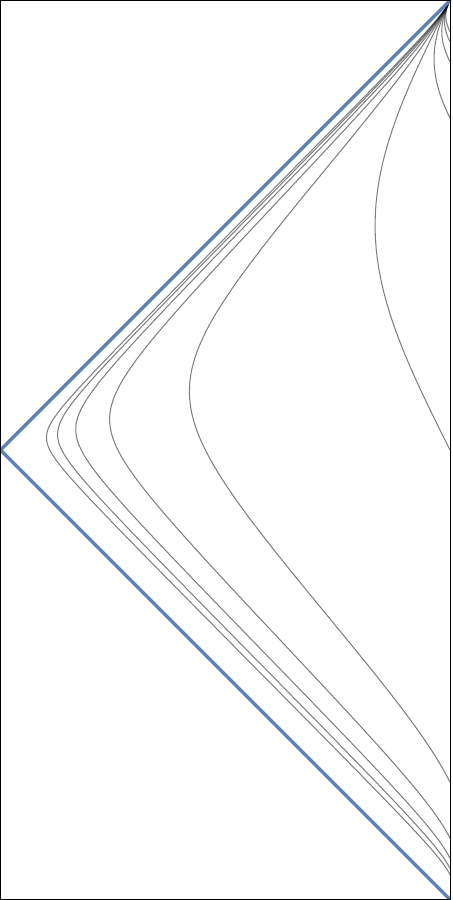} & \includegraphics[width=3.5cm]{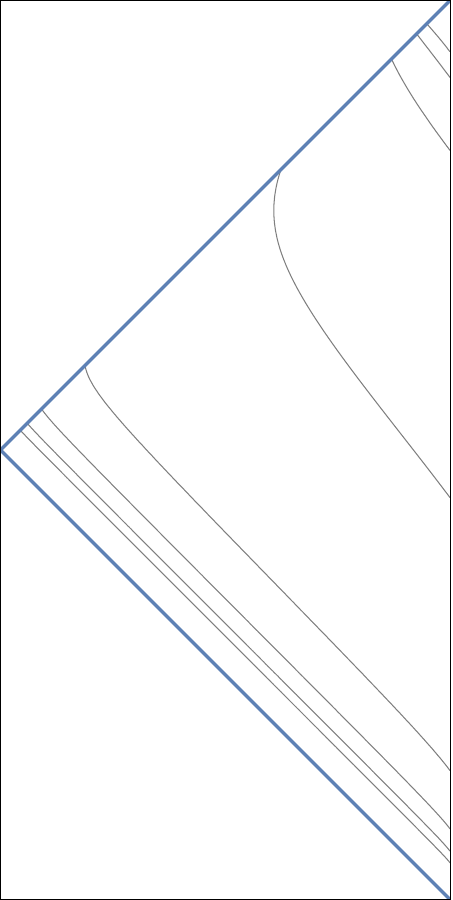} \\
     (e) Bounded${}_<(E,\ell)$ & \begin{tabular}{@{}c@{}}(f) Bounded${}_>(E,\ell)$, \\  Def\mbox{}lecting$(E,\ell)$\end{tabular} & (g) Marginal$(\ell)$ & (h) Plunging$(E,\ell)$\\
   \end{tabular}
    \caption{Taxonomy of NHEK geodesics depicted in the global NHEK conformal diagram. The upper (or respectively, lower) blue line represent the future (respectively, past) event horizon $R=0$ and the dashed/dotted lines are the roots of the radial potential. We used $M=1$, $E=\pm1$ and $\ell=\pm 2\ell_*$ (and $\ell=\pm \frac{1}{2}\ell_*$, respectively) for supercritical (and subcritical, respectively) trajectories.}
    \label{fig:penrose_NHEK}
\end{figure}

\begin{landscape}

\phantom{a}

\vspace{\stretch{1}}

\begin{table}[!bht]    \centering
\begin{tabular}{|c|c|c|c|}\hline
\begin{tabular}{c} \rule{0pt}{13pt} \textbf{Angular momentum}\\\textbf{(and Casimir)}\end{tabular} & \textbf{Energy} & \textbf{Radial range} & \textbf{Denomination} \\ \hline
Supercritical: $\ell > \ell_*$  & $e>-\kappa\sqrt{-\mathcal C}$ & $\kappa\leq R\leq\infty$ & \begin{tabular}{c}
\rule{0pt}{13pt}Plunging$(e,\ell)$ \\\rule{0pt}{13pt} Outward$(e,\ell)$ 
\end{tabular}\\ \cline{2-4}
($-\ell^2 < \mathcal C < 0$) & $e=-\kappa\sqrt{-\mathcal C}<0$ & $R=\frac{\kappa\ell}{\sqrt{-\mathcal C}}$ &\rule{0pt}{13pt} Spherical$(\ell)$ \\ \cline{2-4}
 & $ e<-\kappa\sqrt{-\mathcal C}<0$ & $R_+\leq R\leq\infty$ &\rule{0pt}{13pt} Def\mbox{}lecting$(e,\ell)$ \\ \hline
Critical: $\ell = \ell_*$ ($\mathcal C = 0$) & $e > 0$ & $\kappa \leq R \leq \infty$ & \begin{tabular}{c}
\rule{0pt}{13pt}Plunging$_*(e)$ \\
\rule{0pt}{13pt}Outward$_*(e)$
\end{tabular} \\ \cline{2-4}
 & $e = 0$ & $\kappa \leq R \leq \infty$ & \begin{tabular}{c}
\rule{0pt}{13pt}Plunging$_*$ \\
\rule{0pt}{13pt}Outward$_*$
 \end{tabular} \\ \cline{2-4}
 & $-\kappa\ell < e < 0$ & $\kappa \leq R \leq R_0$ &\rule{0pt}{13pt} Bounded$_*(e)$ \\ \hline
\begin{tabular}{c}
Subcritical: $0 \leq \ell^2 < \ell_*^2$\\($0 < \mathcal C \leq \frac{3 \ell_*^2}{4}$)\end{tabular} & $e>-\kappa\ell$ & $\kappa \leq R \leq R_+$ &\rule{0pt}{13pt} Bounded$_<(e,\ell)$ \\ \hline
Critical: $\ell = - \ell_*$ ($\mathcal C = 0$) & $e > -\kappa\ell>0$ & $\kappa \leq R \leq R_0$ &\rule{0pt}{13pt} Bounded$_*^-(e)$\\\hline
 \begin{tabular}{c}Supercritical: $\ell <- \ell_*$ \\ ($-\ell^2<\mathcal{C}<0$)\end{tabular}  & $e>-\kappa\ell>0$ & $\kappa\leq R\leq R_-$ & \rule{0pt}{13pt} Bounded$_>(e,\ell)$ \\\hline
\end{tabular}\caption{Taxonomy of future-directed geodesics in near-NHEK.}\label{table:taxNN}
\end{table}

\vspace{\stretch{1}}
    
\end{landscape}

\subsection{Near-NHEK}
The only difference between NHEK and near-NHEK geodesic solutions lies in the terms involving the radial coordinate. The proposition stating the equivalence relation between the equatorial and inclined radial parts of the geodesic motion takes the same form as in NHEK:
 \begin{proposition}\label{thm:equivnearNHEK}
For a given normalization $\kappa$, the radial integrals  $t^{(i)}_{R;\kappa}(R)$ $(i=0,1,2)$ only depend upon the near-NHEK energy $e$ and angular momentum $\ell$ while all the dependence upon the mass $\mu$ and Carter constant $Q$ is through $\ell_* = \frac{2}{\sqrt{3}}\sqrt{M^2 \mu^2 +Q}$. 
 \end{proposition}

As in NHEK, the radial taxonomy of Ref. \cite{Compere:2017hsi} is easily extended to bounded, outward and/or retrograde orbits by studying the roots and the sign of $v_{R;\kappa}(R)$. This leads to the classification displayed in Table \ref{table:taxNN} and Figure \ref{fig:taxonomyNN}. 

The future-orientation condition \eqref{consRc} implies $e> -\kappa \ell$ for each orbit that reaches the horizon at $R=\kappa$. In the case $\ell > \ell_*$ and $e < 0$, the condition $e \leq -\kappa \sqrt{-\mathcal C}$ implies $e+\kappa \ell \geq 0$ and therefore the parabola does not intersect the line. Past-oriented geodesics (not depicted here) are obtained from a central symmetry around the origin $e=\ell=0$ as a result of the $\uparrow\!\downarrow$-flip \eqref{PTflip}.  The explicit expressions of all near-NHEK geodesics are listed in Appendix \ref{app:explicitNN}. Also notice that the energy range of the Deflecting$(e,\ell)$ class has been corrected in this thesis with respect to the original derivation \cite{Compere:2020eat}, accordingly to the observation of \cite{Compere:2021bkk}.

\begin{figure}[!bth]
    \centering
    \includegraphics[width=\textwidth]{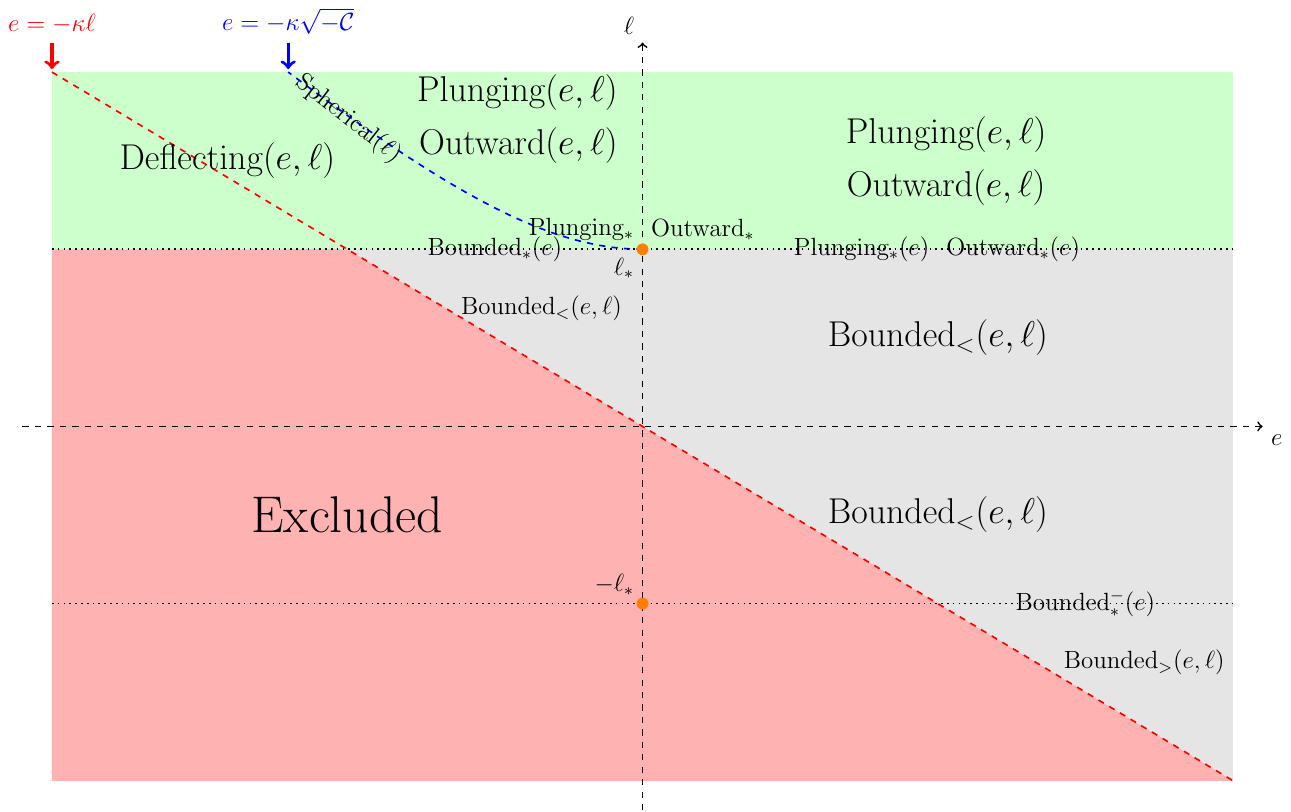} 
    \caption{Radial taxonomy of geodesics in near-NHEK. For equatorial geodesics, $\ell_*=\frac{2}{\sqrt{3}}M \mu$, while for orbits with inclination, $\ell_* = \frac{2}{\sqrt{3}} \sqrt{M^2\mu^2 + Q}$.}
    \label{fig:taxonomyNN}
\end{figure}

\subsection{High-spin features of geodesic motion}
\label{sec:univ}
Let us now discuss a few generic and universal features of near-horizon geodesic motion holding in the high-spin case.

\subsubsection{Radial motion}
A first straightforward conclusion one can derive from the analysis of the near-horizon radial geodesic motion is that
 \begin{proposition}\label{prop4}
All radially unbounded NHEK or near-NHEK geodesics are prograde and either critical or supercritical; \emph{i.e.}, they satisfy $\ell\geq\ell_*$.
 \end{proposition}
This feature of the near-horizon radial motion is directly visible in Figs. \ref{fig:taxonomyEquator} and \ref{fig:taxonomyNN} and leads to remarkable consequences concerning the polar behavior of such trajectories that we will derive in the following section.

The separatrix between bound and unbound motion is clearly visible in Figs. \ref{fig:taxonomyEquator} and \ref{fig:taxonomyNN}. It consists of the geodesic classes Plunging$_*(E)$ and Outward$_*(E)$ for NHEK and the geodesic classes Plunging$_*(e)$, Outward$_*(e)$, and Bounded$_*(e)$ for near-NHEK that each lie at the critical angular momentum line $\ell = \ell_*$.

\subsubsection{Polar motion} 
The polar motion of both NHEK and near-NHEK trajectories is bounded in an interval around the equator, $\theta_{\text{\text{min}}}\leq \theta \leq \pi - \theta_{\text{\text{min}}}$, where $\cos\theta_{\text{\text{min}}}=\sqrt{z_+}$ or  $\cos\theta_{\text{\text{min}}}=\sqrt{z_0}$. The maximal polar angle is determined for $\ell^2 \neq \ell^2_\circ = 4 M^2 \mu^2$ as
\begin{equation}
z_+(\ell ,Q ) =\frac{3\ell^2+4(Q+M^2\mu^2)-\sqrt{9\ell^4+16(M^2\mu^2-Q)^2+8\ell^2(3M^2\mu^2+5Q)}}{2(4M^2\mu^2-\ell^2)}\nonumber
\end{equation}
and for $\ell = \pm \ell_\circ$ as
 \begin{align}
z_0(Q) = \lim_{\ell \rightarrow \pm 2 M \mu} z_+ = \frac{Q}{Q+4 M^2 \mu^2}. 
 \end{align}
Remember that $Q \geq 0$ by consistency of polar motion. The asymptotic values are
\begin{align}
\lim_{\scriptsize{\begin{array}{l} Q\to 0 \\ \ell \text{ fixed} \end{array}}}z_+(\ell,Q)&=0, \qquad\qquad\qquad\qquad\qquad\;\;\;\;\;\;\;\; \lim_{\scriptsize{\begin{array}{l} Q\to \infty \\ \ell \text{ fixed} \end{array}}}z_+(\ell,Q)=1,\label{eq:asymptotics}\\
\lim_{\scriptsize{\begin{array}{l} \ell \to 0 \\ Q \text{ fixed} \end{array}}}z_+(\ell,Q)&=\left\lbrace
\begin{array}{cl}
\frac{Q}{M^2\mu^2} &\text{ if }Q<M^2\mu^2\\
1 & \text{ if }Q\geq M^2\mu^2
\end{array}\right. ,\qquad \lim_{\scriptsize{\begin{array}{l} \ell \to \infty \\ Q \text{ fixed} \end{array}}}z_+(\ell,Q) = 0.
\end{align}
 For fixed $\ell$, $z_+$ is a monotonic function of $Q$, and reciprocally $z_+$ is monotonic in $\ell$ at fixed $Q$. The pendular oscillation around the equatorial plane will explore a larger range of $\theta$ when $\theta_\text{\text{min}}$ is smallest or $z_+$ closer to 1, which occurs either for small $\ell$ and $Q \geq M^2 \mu^2$ or large $Q$. 
 
Now, one can check that for critical or supercritical angular momentum $\ell^2 \geq \ell^2_*(Q)$, one has $z_+ < 2 \sqrt{3}-3$ for $\ell \neq \ell_\circ(Q)$ and $z_0 < 2 \sqrt{3}-3$ for $\ell^2 = \ell^2_\circ$. The special angle 
 \begin{align}
\theta_{\text{VLS}} \triangleq \arccos{\sqrt{2\sqrt{3}-3}} \approx 47^\circ
 \end{align}
is in fact the velocity-of-light surface in the NHEK geometry \eqref{eq:NHEK_metric} (or near-NHEK geometry) defined as the polar angle such that $\partial_T$ is null. It obeys $\Lambda(\theta_{\text{VLS}}) = 1$. The polar region closer to either the north or south poles admits a timelike Killing vector, namely $\partial_T$. On the contrary, the polar region around the equator $\theta\in ]\theta_{\text{VLS}},\pi-\theta_{\text{VLS}}[$ does not admit a timelike Killing vector. The velocity-of-light surface separates these two polar regions. We have therefore proven the following property:
 \begin{proposition}\label{prop5}
All critical or supercritical orbits $\ell^2 \geq \ell^2_*$ in (near-)NHEK geometry lie in the polar region $\theta\in ]\theta_{\text{VLS}},\pi-\theta_{\text{VLS}}[$ where there is no timelike Killing vector. This applies in particular to all spherical orbits.
 \end{proposition}

The subcritical orbits $\ell^2 < \ell^2_*$ can explore all polar regions of the (near)-NHEK geometry. As a consequence of Propositions \ref{prop4} and \ref{prop5}, we have 
 \begin{proposition}\label{prop6}
All radially unbounded geodesics in (near-)NHEK geometry lie in the polar region $\theta\in ]\theta_{\text{VLS}},\pi-\theta_{\text{VLS}}[$ bounded by the velocity-of-light surface. 
 \end{proposition}
In particular, for null geodesics, this feature provides the ``NHEKline'' in the imaging of light sources around a nearly extreme Kerr black hole \cite{Bardeen:1972fi,Gralla:2017ufe}. In \cite{AlZahrani:2010qb,Porfyriadis:2016gwb}, Proposition \ref{prop6} was proven for null geodesics. Here, we show that it is a generic property of all timelike geodesics as well.

\section{Spherical geodesics}
\label{sec:spherical_properties}

The spherical (near-)NHEK geodesics take a distinguished role among all geodesics. First, a subclass of spherical geodesics in NHEK and near-NHEK constitute the innermost stable spherical orbits (ISSOs) and the innermost spherical bound orbits (ISBOs) in the high-spin limit, respectively.  Our first motivation is to fully characterize the ISSO, in order to generalize the analysis of the inspiral/merger transition performed around the equatorial plane in the high-spin limit  \cite{Compere:2019cqe,Burke:2019yek} to inclined orbits.

Second, as noticed in Ref. \cite{Compere:2017hsi}, the equatorial NHEK (resp. near-NHEK) orbits are the simplest representatives for each equivalence class of prograde incoming critical (respectively, supercritical) equatorial orbits under $\textsf{SL}(2,\mathbb R) \times \textsf{U}(1) \times \mathbb Z_2$ symmetry. We will show in Sec. \ref{sec:classes} that the spherical (near-)NHEK orbits are the simplest representatives for each equivalence class of arbitrary timelike (near-)NHEK geodesics under $\textsf{SL}(2,\mathbb R) \times \textsf{U}(1) \times (\mathbb Z_2)^3$ symmetry without any restriction. These two reasons justify the comprehensive study of the spherical geodesics. 

\subsection{Innermost stable spherical orbits} \label{sec:ISSO}

The ISSOs are defined as the last stable spherical orbits of Kerr. They are defined from the solutions to 
 \begin{align}\label{defISSO}
R(r) =R'(r)= R''(r) = 0
 \end{align}
where $R$ is defined in \eqref{eq:kerr_vr}. They admit a constant radius $ r $ and a fixed $\hat E$ and $\ell$, which can be obtained as solutions of polynomial equations which we will not give explicitly. There are two branches at positive $\hat E$ corresponding to prograde ($\ell \geq 0$) and retrograde orbits  ($\ell < 0$). For the Schwarzschild black hole, the parameters on the two branches of the ISSO are
 \begin{align}
r_{\text{ISSO}}=6M,\qquad \frac{\hat E_{\text{ISSO}}}{\mu}= \frac{2\sqrt{2}}{3},\qquad \frac{\ell_{\text{ISSO}}}{\mu M} = \pm \sqrt{12 - \frac{Q}{M^2\mu^2}},
 \end{align}
which implies the bound $Q \leq 12 M^2 \mu^2$.

For arbitary spin, the \textit{innermost stable circular orbit} (ISCO) is defined as the prograde ISSO equatorial orbit, i.e. restricted to $Q=0$ ($\theta = \frac{\pi}{2}$). The parameters are \cite{Bardeen:1972fi}
 \begin{align}
\frac{\hat E_{\text{ISCO}}}{M \mu} = \frac{1-2 /\tilde r_{\text{ISCO}}-\tilde a/\tilde r_{\text{ISCO}}^{3/2}}{\sqrt{1-3/\tilde r_{\text{ISCO}}-2\tilde a / \tilde r_{\text{ISCO}}^{3/2}}},\qquad \frac{\ell_{\text{ISCO}}}{\mu M} = \frac{2}{\sqrt{3 \tilde r_{\text{ISCO}}}} (3 \sqrt{\tilde r_{\text{ISCO}}} +2\tilde a),
 \end{align} 
where $\tilde a = a/M$ and 
\begin{subequations}
     \begin{align}\label{ISCO}
\tilde r_{\text{ISCO}} &\triangleq  \frac{r_{\text{ISCO}}}{M} =3+Z_2-\sqrt{(3-Z_1)(3+Z_1+2Z_2)}, \\
Z_1 & \triangleq  1+(1-\tilde a^2)^{1/3}[(1+\tilde a)^{1/3}+(1-\tilde a)^{1/3}],\qquad Z_2 \triangleq \sqrt{3\tilde a^2+(Z_1)^2}.
 \end{align}
\end{subequations}

\subsubsection{Minimal polar angle} 

In the generic case $\ell \neq 0$, the polar motion is pendular -- \textit{i.e.}, oscillating around the equator in the interval $[\theta_\text{\text{min}},\pi-\theta_\text{\text{min}}]$. The minimal angle as a function of the spin $a$ and ISCO radius $r_{\text{ISSO}}$ can simply be found by solving numerically the three equations \eqref{defISSO} that define the ISSO together with the condition that there is a polar turning point, $\Theta(\cos \theta_{\text{\text{min}}}) = 0$ where $\Theta(\cos^2\theta)$ is defined in \eqref{eq:kerr_vtheta}. The resulting minimal angle is displayed in Fig. \ref{fig:hugues} for a large range of spins including nearly extremal. This completes a similar plot drawn in Ref. \cite{Apte:2019txp} for spins far from extremality.

\begin{figure}[!hbt]
    \centering
    \includegraphics[width=13cm]{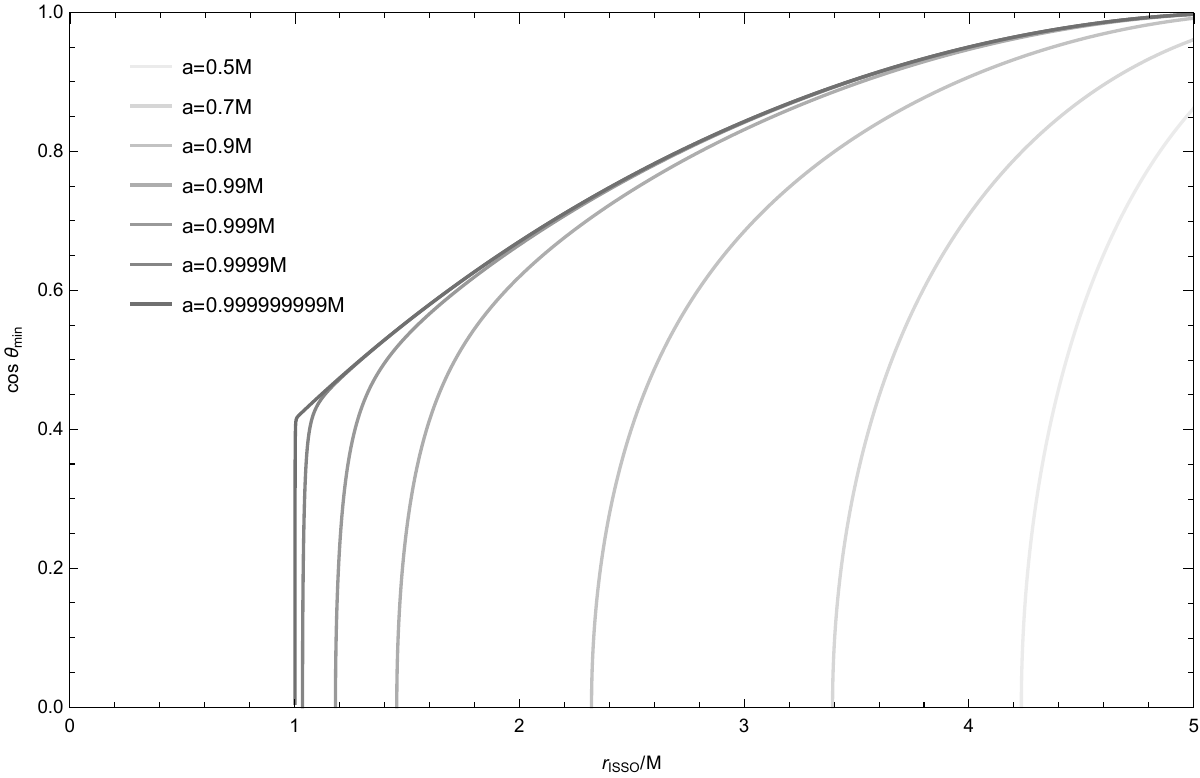}
    \caption{$\cos\theta_\text{\text{min}}$ as a function of ISSO radius for several black hole spins $a$.}
    \label{fig:hugues}
\end{figure}
\begin{figure}[!hbt]
    \centering
    \includegraphics[scale=0.5]{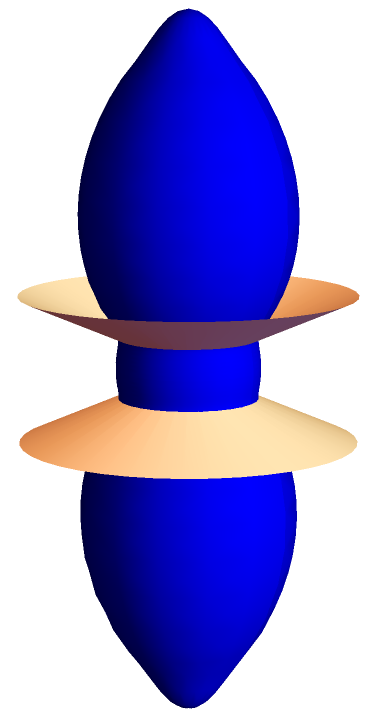}
    \caption{Euclidean embedding of the ISSO using the Boyer-Lindquist radius $r$, azimuthal angle $\varphi$ and polar angle $\theta$ for $a=0.9999$. A cone is drawn at the critical polar angle beyond which the ISSO lies in the NHEK region. In the extremal limit, this polar angle is $\theta \approx 65^\circ$.}
    \label{fig:cone}
\end{figure}

\clearpage

We note that for high-spins, the radius asymptotes to $r = M$ and the minimal angle reaches a critical value around $0.42$ radians or $65^\circ$. When the motion reaches regions sufficiently far from the equatorial plane, the ISSO radius increases steeply and leaves the near-horizon region $ r \simeq M$. Another graphical representation of this behavior is shown in Fig. \ref{fig:cone}. We will explain these features in the next section. 

\subsection{The NHEK spherical orbit and the high-spin ISSOs}

In the high-spin limit $a \rightarrow M$, the prograde ISSOs are characterized by the following Boyer-Lindquist energies and angular momentum:
 \begin{align}\label{ISSOext}
\hat E_{\text{ISSO}}= \frac{1}{\sqrt{3}M}\sqrt{M^2 \mu^2 +Q},\qquad \ell_{\text{ISSO}} = + 2M \hat E_{\text{ISSO}} 
 \end{align}
and the following Boyer-Lindquist radius:
 \begin{align} 
r_{\text{ISSO}} = M + M\left( \frac{Q+M^2 \mu^2}{-Q+\frac{M^2 \mu^2}{2}} \right)^{1/3} \lambda^{2/3}+\mathcal O(\lambda^{4/3}). \label{hatrISSO}
 \end{align}
Given the scaling in $\lambda$, for the range 
 \begin{align}
0 \leq Q \leq \frac{M^2 \mu^2}{2},\label{rangeQ}
 \end{align}
the ISSOs belong to the NHEK geometry and admit the NHEK radius 
 \begin{align}
R = R_{\text{ISSO}} \triangleq  \left( \frac{Q+M^2 \mu^2}{-Q+\frac{M^2 \mu^2}{2}} \right)^{1/3}. 
 \end{align}
In particular, the ISCO has the minimal radius $R_{\text{ISCO}}=2^{1/3}$. In terms of NHEK quantities, the orbits admit a critical angular momentum and a vanishing NHEK energy, 
 \begin{align}
\ell = \ell_* \triangleq \frac{2}{\sqrt{3}}\sqrt{Q+M^2 \mu^2},\qquad E = 0. 
 \end{align}
In the high-spin limit, the prograde ISSOs in the range \eqref{rangeQ} are therefore exactly the $\text{Spherical}_*(Q)$ orbits in the classification of Sec. \ref{sec:NHEK}.  The prograde ISSOs outside the range \eqref{rangeQ} and the retrograde ISSOs do not belong to the near-horizon geometry and will not be described here. 

In terms of polar behavior, $\text{Spherical}_*(Q)$ orbits are instances of $\text{Pendular}(Q,\ell_*)$ motion (except for $Q=0$, where they are just equatorial orbits). In the range \eqref{rangeQ}, they admit an $ \epsilon _0$ as defined in \eqref{defeps0} given by $ \epsilon _0 = \frac{Q-2 M^2 \mu^2}{3} < 0$, and the angular momentum lies below the value $\ell_\circ$:
 \begin{align}
\ell_* \leq \sqrt{2} M \mu < \ell_\circ. 
 \end{align}

The main property of $\text{Pendular}(Q,\ell_*)$ motion is that the polar angle $\theta$ is bounded in an interval around the equator (see \eqref{eq:costh} and \eqref{defzpm}) :
\begin{equation}
\theta\in[\theta_{\text{min}},\pi-\theta_{\text{min}}]
\end{equation}
where
\begin{equation}
\cos\theta_{\text{min}}=\sqrt{z_+}=\sqrt{\frac{Q}{\frac{3}{4}\ell^2_*+\sqrt{\frac{9}{16}\ell_*^4-\frac{\ell_*^2 Q}{2}+Q^2}}}.
\end{equation}

At fixed $M\mu$, $\theta_{\text{min}}(Q)$ is a monotonic function interpolating between the equator $\theta=90^\circ$ at $Q=0$ and $\theta_{\text{VLS}} \triangleq \arccos{\sqrt{2\sqrt{3}-3}} \approx 47^\circ$ for $Q \rightarrow \infty$. The special angle $\theta_{\text{VLS}}$ is the velocity-of-light surface in the NHEK geometry \eqref{eq:NHEK_metric} as described in Sec. \ref{sec:univ}. The ISSO therefore always lies in the region of NHEK spacetime around the equator, where there is no timelike Killing vector. This is depicted in Fig. \ref{fig:opening}.

\begin{figure}[!t]\center
\begin{tabular}{ccccc}\vspace{-20pt}
\includegraphics[width=0.2\textwidth]{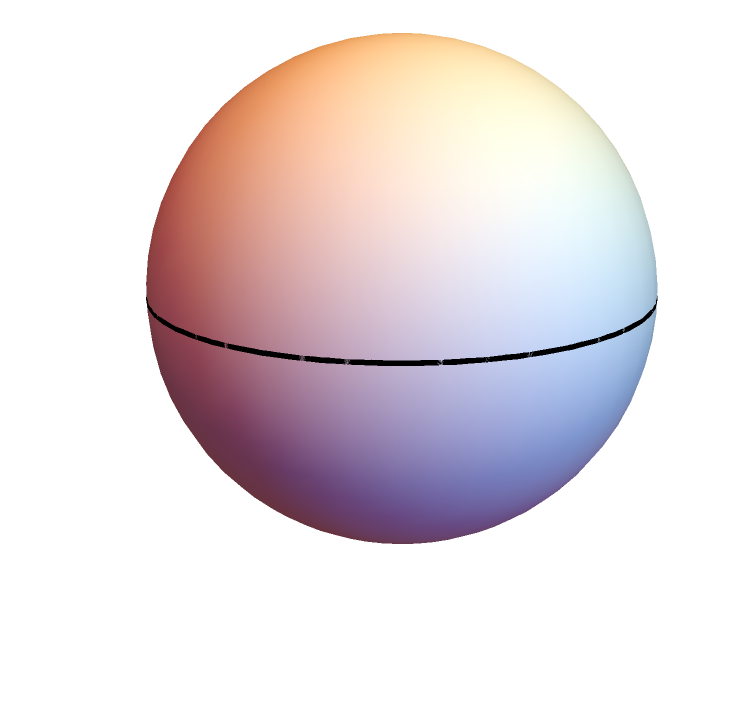} &  \includegraphics[width=0.2\textwidth]{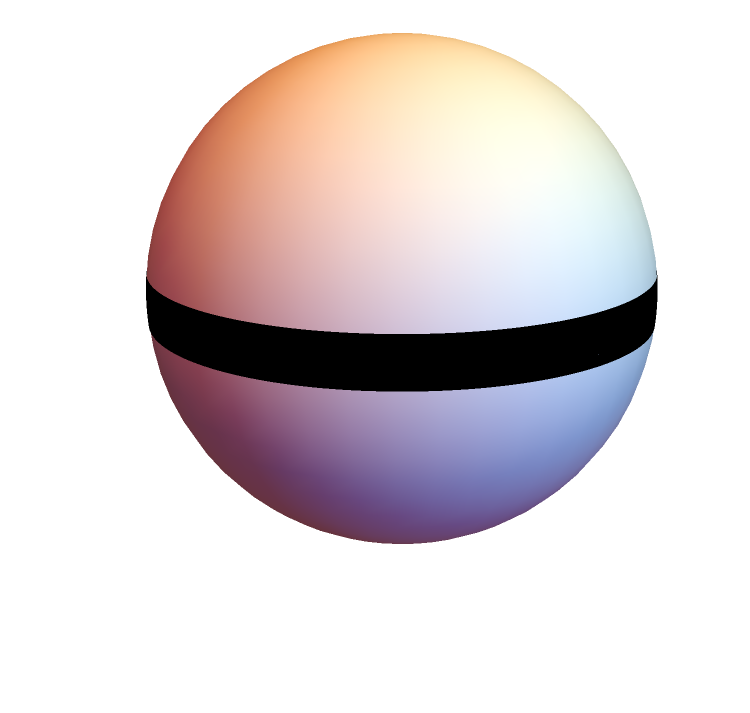} &  \includegraphics[width=0.2\textwidth]{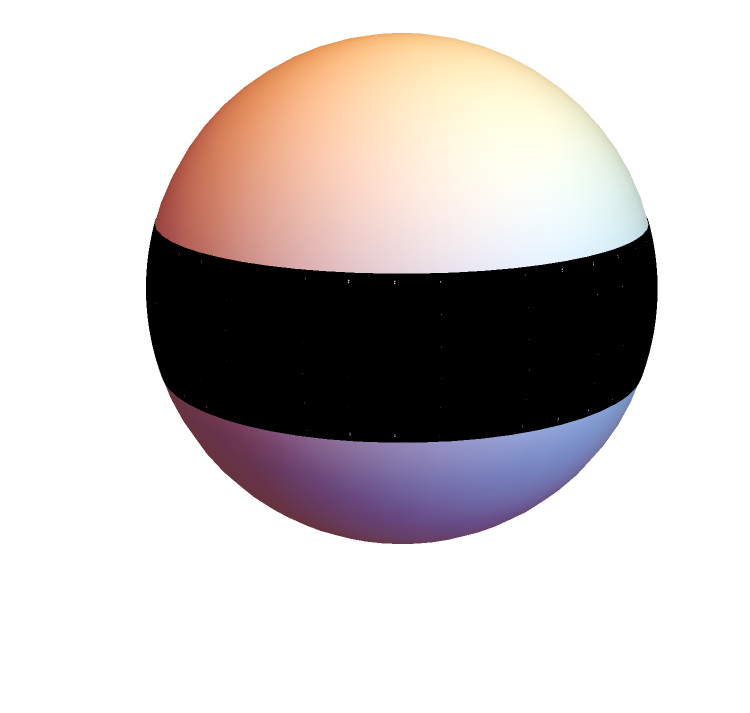} & \shortstack{\ldots\\\rule{0pt}{35pt}} &
\includegraphics[width=0.2\textwidth]{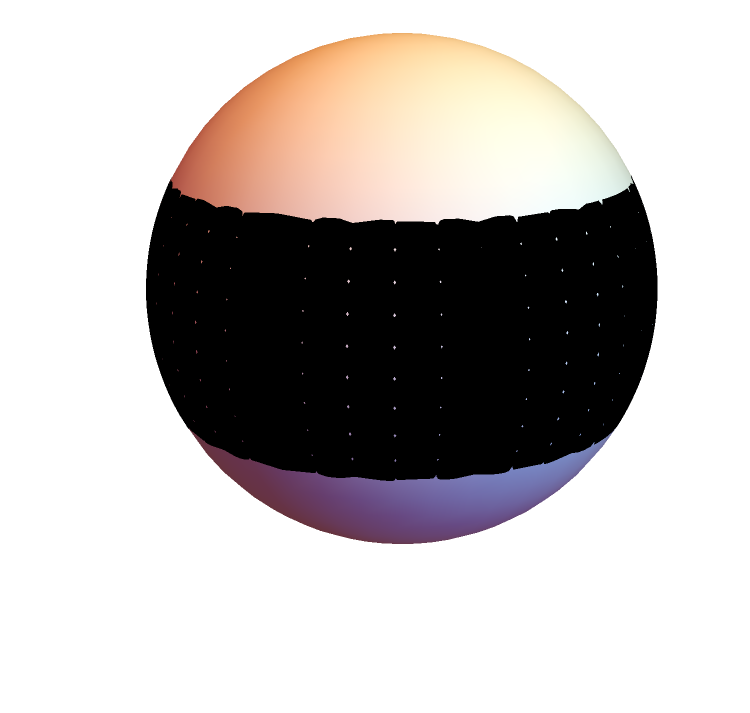}\\
\end{tabular}
\begin{tikzpicture}[scale=1.2]
\draw[thick] (0,0)--(7.5,0);
\draw[dashed,thick] (7.5,0)--(8.5,0);
\draw[->,thick] (8.5,0)--(10,0);
\node[right]() at (10,0) {$Q$};
\end{tikzpicture}
\caption{For increasing $Q \geq 0$, $\text{Spherical}_*(Q)$ orbits can explore a equator-centered band whose width becomes larger, finally reaching for $Q \rightarrow \infty$ the angular range $\theta\in[\theta_{\text{VLS}},\pi-\theta_{\text{VLS}}]$ bounded by the velocity-of-light surface. The prograde IBSOs lie in the near-NHEK region for $Q \leq 2 M^2\mu^2$, which further bounds the angular range.}
\label{fig:opening}
\end{figure}
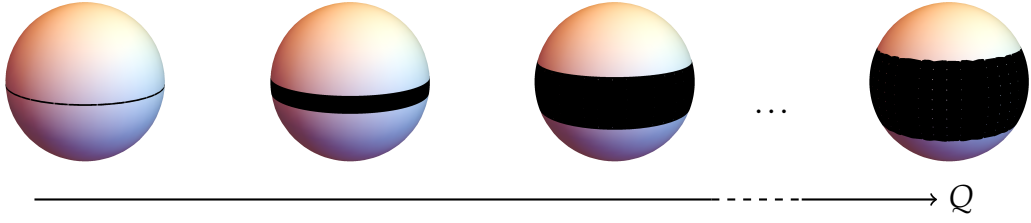

However, since the ISSO admits the range \eqref{rangeQ} due to its relationship to the asymptotically flat Boyer-Lindquist radius \eqref{hatrISSO}, the limiting angle is reached first for $Q=\frac{M^2 \mu^2}{2}$ at $\arccos{\sqrt{3-2\sqrt{2}}}\approx 65^\circ$. This explains the behavior depicted in Fig. \ref{fig:hugues}. This result was discovered simultaneously in \cite{Stein:2019buj,Compere:2020eat}. The limiting angle of the ISSO is given by $\arcsin{\sqrt{2(\sqrt{2}-1)}} = \arccos{\sqrt{3-2\sqrt{2}}} \approx 65^\circ$. 

\subsection{The near-NHEK spherical orbits and the high-spin IBSOs}

The innermost bound spherical orbits (IBSOs) are determined by the equations
 \begin{align}
R(r) = R'(r) = 0,\qquad \hat E = \mu.
 \end{align}
In the high-spin limit $\lambda \rightarrow 0$, the angular momentum and Boyer-Lindquist radius of the prograde IBSOs are given by 
\begin{subequations}
     \begin{align}
\ell &= \ell_\circ \left(1+\frac{\lambda}{\sqrt{2}}\sqrt{1-\frac{Q}{2M^2 \mu^2}}+\mathcal O(\lambda^2)\right), \\
r &= M \left(1+ \frac{\sqrt{2}\lambda}{\sqrt{1-\frac{Q}{2M^2 \mu^2}}}+\mathcal O(\lambda^2)\right)  \end{align}
\end{subequations}
where $\ell_\circ \equiv 2M \mu$. In particular, for $Q=0$ we recover the scaling of the innermost bound circular orbit (IBCO) \cite{Bardeen:1972fi}. Given the scaling $\sim \lambda$, the prograde IBCOs therefore lie in the near-NHEK region for all $Q < 2M^2 \mu^2$. Using \eqref{eq:cvnn}--\eqref{eN}, the angular momentum, near-NHEK energy and near-NHEK radius are given in the high-spin limit by 
\begin{subequations}
     \begin{align}
\ell &=  \ell_\circ,\\
\frac{e}{\kappa} &= -\sqrt{2M^2 \mu^2 -Q}, \\
\frac{r}{\kappa} &= \frac{\sqrt{2}\lambda}{\sqrt{1-\frac{Q}{2M^2 \mu^2}}}. 
 \end{align}
\end{subequations}
The prograde IBCOs in the range $0 \leq Q < 2M^2 \mu^2$ are described by instances of Spherical$(\ell)$ orbits. In terms of polar motion, $Q=0$ are equatorial and $Q>0$ are pendular of class Pendular$_\circ(Q)$; see Table \ref{table:taxPolar}. The polar range is determined as $\theta_{\text{\text{min}}} \leq \theta \leq \pi - \theta_{\text{\text{min}}}$ where
 \begin{align}
\theta_{\text{\text{min}}} = \arccos\sqrt{\frac{Q}{Q+\ell_\circ^2}}. 
 \end{align}
The maximal polar angle reachable within the near-NHEK region by IBSOs is obtained for the limiting value $Q = 2M^2 \mu^2$ at
 \begin{align}
\theta_{\text{\text{min}}} = \arccos\sqrt{1/3} = \arcsin{\sqrt{2/3}} \approx 55^\circ.
 \end{align} 
This critical angle was also previously obtained in Refs. \cite{Hod:2017uof,Stein:2019buj}. Finally, note that spherical photon orbits in the high-spin limit were also discussed in Refs. \cite{Yang:2012he,Hod:2012ax}.

\section{Conformal mappings between radial classes}\label{sec:classes}

The near-horizon region of near-extremal Kerr black holes admits four Killing vectors forming the group $\textsf{SL}(2,\mathbb R) \times \textsf{U}(1)$, hereafter denoted as the conformal group $G$. The geodesic equations are invariant under $G$ and the geodesics therefore transform under the action of $G$. Moreover, a group generated by four $\mathbb Z_2$ symmetries exists that preserve the geodesic equations. The subgroup preserving the domain $R > 0$ for NHEK (or $r > 0$ for near-NHEK) is generated by the $\uparrow\!\downarrow$-flip \eqref{PTflip}, which flips the geodesic orientation, and two additional $\mathbb Z_2$ transformations that preserve the geodesic orientation: namely, the parity flip
 \begin{align}
\theta \rightarrow \pi - \theta,\qquad  \Phi \rightarrow \Phi + \pi,\qquad s_\theta^i\rightarrow -s_\theta^i, \label{parity}
 \end{align}
and the $\rightleftarrows$-flip
   \begin{align}
    T \rightarrow -T,\qquad \Phi \rightarrow -\Phi,\qquad \lambda \rightarrow -\lambda,\qquad s_R^i \rightarrow -s_R^i,\qquad s_\theta^i \rightarrow -s_\theta^i .\label{Tphiflip}
 \end{align}
The last discrete transformation that we use as a basis is the \rotatebox[origin=c]{-45}{$\rightleftarrows$}-flip
 \begin{align}
R \rightarrow -R, \qquad \Phi \rightarrow -\Phi, \qquad \ell \rightarrow -\ell, \qquad s^i_R \rightarrow -s^i_R. 
 \end{align}

The parity transformation defined in \eqref{parity} leaves each motion invariant and will not be considered further. The $\rightleftarrows$-flip changes the boundary conditions of the geodesics, which may affect their denomination. It maps bounded orbits to bounded orbits, and deflecting orbits to deflecting orbits, but plunging orbits to outward orbits, as illustrated in Fig. \ref{fig:flipTPhi}. For bounded orbits, the part before the turning point is mapped to the part after the turning point, and vice-versa. The \rotatebox[origin=c]{-45}{$\rightleftarrows$}-flip can be used as follows: one first continues a geodesic defined in $R > 0$ beyond the horizon $R = 0$ and the resulting geodesic with $R < 0$ is then mapped to a geodesic in the $R > 0$ region using the \rotatebox[origin=c]{-45}{$\rightleftarrows$}-flip. Together with the action of \eqref{Tphiflip}, it allows us to map plunging orbits with $\ell > 0$ to bounded orbits with $\ell < 0$. This process is illustrated in Fig. \ref{fig:bounded_from_plunging}.

\begin{figure}[ht!]
    \centering
    \begin{tabular}{c|c}
      \includegraphics[width=7cm]{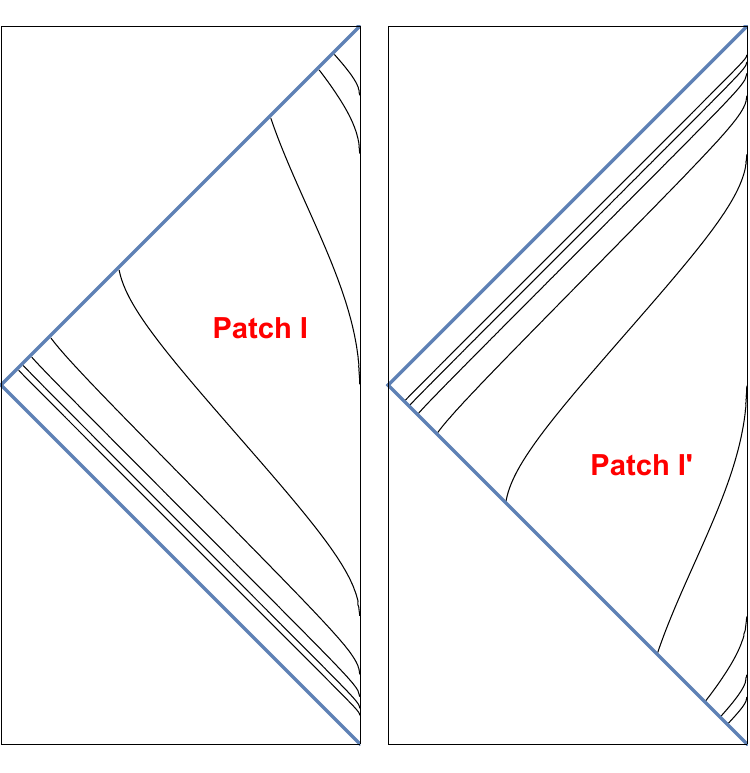}   & \includegraphics[width=7cm]{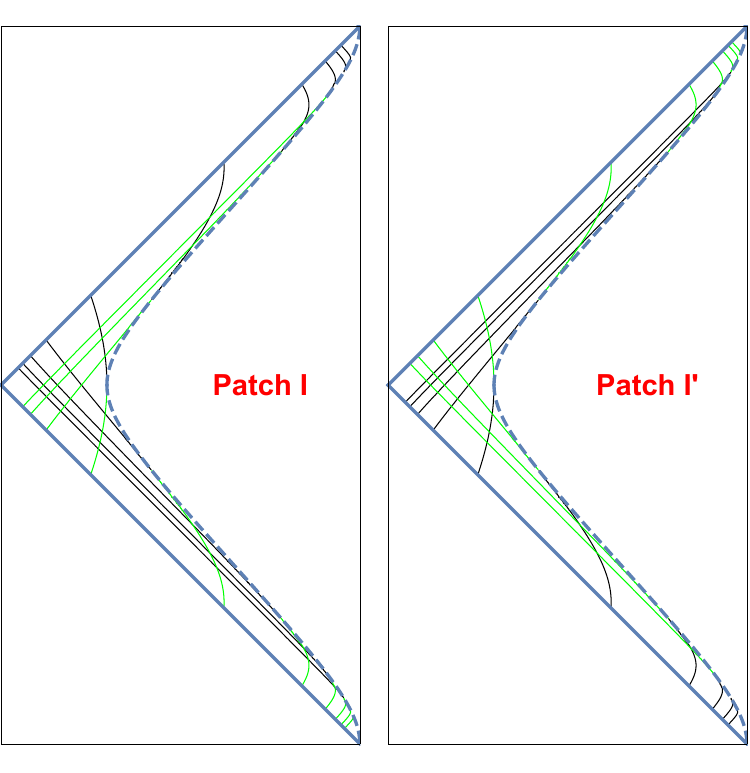}  \\
        (a) & (b)
    \end{tabular}
    \caption{Penrose diagram of NHEK spacetime depicting the action of the $\rightleftarrows$-flip on (a) plunging and (b) bounded geodesics. Under this transformation, a trajectory belonging to the patch I is mapped to an orbit of the patch I'. While plunging geodesics become outward ones, bounded motion remains bounded. The energy and angular momentum of the trajectory are unchanged.}
    \label{fig:flipTPhi}
\end{figure}

\begin{figure}
    \centering
    \begin{tabular}{ccc}
     \includegraphics[width=4cm]{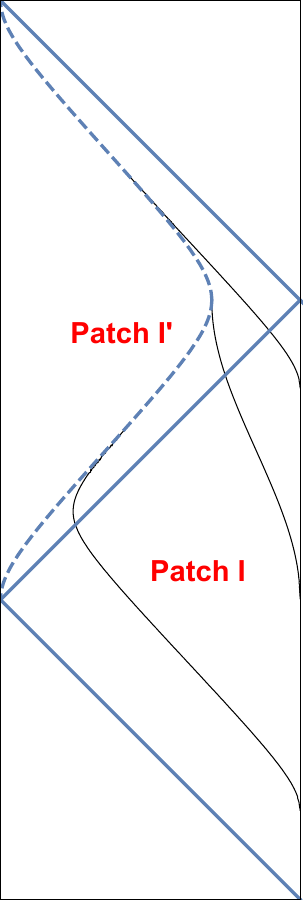} & \includegraphics[width=4cm]{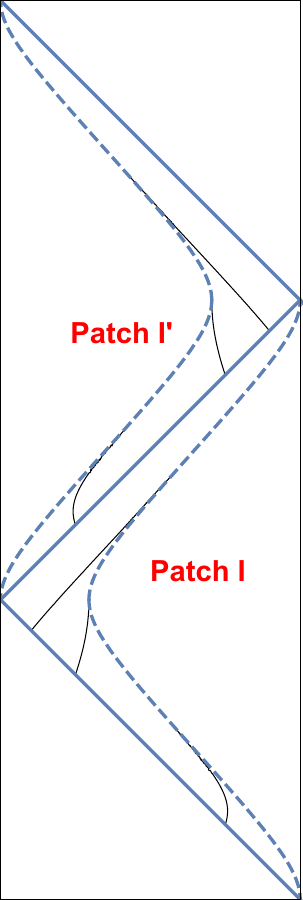} & \includegraphics[width=4cm]{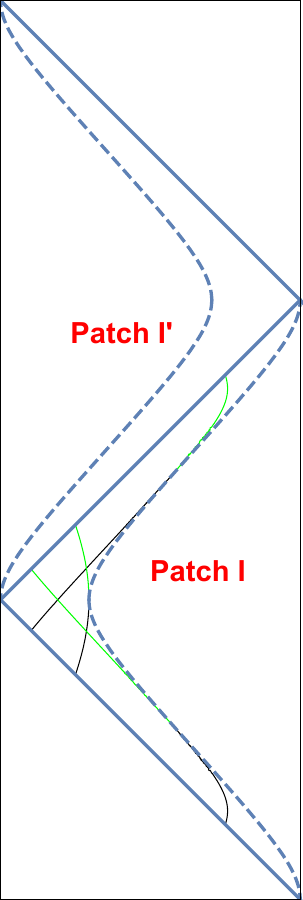}\\  (a) & (b) & (c) 
    \end{tabular}
    \caption{Penrose diagram representation of the construction of a critical NHEK bounded geodesic from a plunging one. (a) Continuation of the trajectory beyond the horizon in patch I' until the radial potential root (depicted with dashes); (b) \protect\rotatebox[origin=c]{-45}{$\rightleftarrows$}-flip which brings the part of the bounded geodesic before the turning point in the NHEK Poincar\'e patch I; (c) $\rightleftarrows$-flip, which maps the part of the bounded geodesic before the turning point to the part after the turning point.}
    \label{fig:bounded_from_plunging}
\end{figure}

The equivalence classes of equatorial critical and supercritical prograde timelike geodesics under the action of $\textsf{SL}(2,\mathbb R) \times \textsf{U}(1) \times \uparrow\!\downarrow$ symmetry were derived in  Ref. \cite{Compere:2017hsi} following earlier work \cite{Porfyriadis:2014fja,Hadar:2014dpa,Hadar:2015xpa,Hadar:2016vmk}. In this section, we will perform the decomposition of arbitrary geodesics into equivalence classes under the action of $\textsf{SL}(2,\mathbb R) \times \textsf{U}(1) \times \uparrow\!\downarrow \times \rightleftarrows \times$\rotatebox[origin=c]{-45}{$\rightleftarrows$}.

The Casimir $\mathcal{C}$ of $\textsf{SL}(2,\mathbb{R})$ cannot vary upon acting with $G \triangleq \textsf{SL}(2,\mathbb{R}) \times \textsf{U}(1)$ transformations. Moreover, the action of the group $G$ acts trivially on the polar coordinate $\theta$. These two properties imply that both $Q$ and $\ell$ are invariant under the action of $G$. In particular, critical, supercritical or subcritical geodesics form distinct classes under $G$. On the contrary, the (near-)NHEK energy $E$ (or $e$) can vary under conformal transformations. Conformal transformations can map NHEK to near-NHEK orbits, and vice-versa. As a result of Propositions \ref{thm:equivNHEK} and \ref{thm:equivnearNHEK}, null geodesics can be treated on the same footing as timelike geodesics.

A conformal transformation belonging to $\textsf{SL}(2,\mathbb{R}) \times \textsf{U}(1)$ maps (near)-NHEK spacetime parametrized by $(T,R,\theta,\Phi)$ to (near-)NHEK spacetime parametrized by $(\bar T,\bar R,\theta,\bar \Phi)$\footnote{We denote here without distinction NHEK and near-NHEK coordinates with capital letters.} where 
\begin{align}
\overline{T}&=\overline{T}(T,R),\nonumber\\
\overline{R}&=\overline{R}(T,R),\label{eq:conformal_tfo}\\
\overline{\Phi}&=\Phi+\delta\bar{\Phi}(T,R)\nonumber.
\end{align}
The geodesic equations in (near)-NHEK imply $T=T(R)$. Therefore, the action of conformal symmetries reduces to an action on the radial motion, leaving the polar motion unchanged. More precisely, in the decomposition of $\Phi(\lambda)$ \eqref{eq:NHEKphi}--\eqref{eqn:nnPhi} in terms of a radial part and a polar part, the polar part will remain untouched by conformal transformations.

It was shown in Ref. \cite{Compere:2017hsi} that each equivalence class of equatorial prograde critical (respectively, supercritical) geodesics with incoming boundary conditions under $G \times \uparrow\!\downarrow$ admits a distinguished simple representative, namely the NHEK (respectively, near-NHEK) circular orbits. After analysis, we obtain that each geodesic equivalence class under $G \times \uparrow\!\downarrow \times \rightleftarrows \times$\rotatebox[origin=c]{-45}{$\rightleftarrows$} admits a spherical orbit as the simplest representative as illustrated in Fig. \ref{fig:conformal_mappings}. Past directed geodesics must be considered as intermediate steps in order to relate each future directed geodesic to spherical geodesics.  Supercritical orbits ($\ell^2 > \ell^2_*$) admit the near-NHEK Spherical$(\ell)$ orbit as a representative and critical orbits ($\ell = \pm\ell_*$) admit the NHEK Spherical$_*$ orbit as a representative. No subcritical spherical geodesic exists. However, we introduce an analytically continued complex subcritical geodesic by continuing the radius $R_0 \mapsto i R_0$ and show that it generates the subcritical class.

The explicit formulas for the three categories of equivalence classes of orbits under $G \times \uparrow\!\downarrow \times \rightleftarrows \times$\rotatebox[origin=c]{-45}{$\rightleftarrows$} are given in the following sections. We will denote the final coordinates and orbital parameters reached by the conformal mappings with bars.

\subsection{Critical $\mathcal{C}=0$}
\paragraph{Spherical$_*$ $\Leftrightarrow$ Plunging$_*(E)$ (NHEK/NHEK).} The conformal mapping is given by
\begin{align}
    \bar T&= - \frac{R^2T}{R^2T^2-1},\nonumber\\
    \bar R&=\frac{R^2T^2-1}{R},\label{ct1}\\
    \bar\Phi&=\Phi+\log\frac{RT+1}{RT-1}-i\pi.\nonumber
\end{align}
It maps a (future-directed) NHEK spherical trajectory of radius $R_0$ to a (future-directed)  critical plunge of energy $\bar E=\frac{2\ell_*}{R_0} > 0$.

\paragraph{Spherical$_*$ $\Leftrightarrow$ Plunging$_*$ (NHEK/near-NHEK).} One performs the NHEK/near-NHEK diffeomorphism $(T,R,\theta,\Phi)\to(\bar t,\bar R,\theta,\bar\phi)$, whose explicit form is
\begin{align}
T&=-\exp\qty(-\kappa \bar t)\frac{\bar R}{\sqrt{\bar R^2-\kappa^2}}\nonumber,\\
    R&=\frac{1}{\kappa}\exp\qty(\kappa \bar t)\sqrt{\bar R^2-\kappa^2},\label{ct2}\\
    \Phi&=\phi-\frac{1}{2}\log\frac{\bar R-\kappa}{\bar R+\kappa}.\nonumber
\end{align}
Its inverse is
\begin{align}
    \bar t &= \frac{1}{\kappa}\log\frac{R}{\sqrt{R^2T^2-1}},\nonumber\\
    \bar R &=-\kappa RT,\\
    \bar\phi &= \Phi+\frac{1}{2}\log\frac{RT+1}{RT-1}\nonumber
\end{align}
for $R>0$ and $RT<-1$. The orbital parameters are related as
\begin{align}
    R_0=\frac{1}{\kappa}\exp\qty(\kappa t_0),\qquad\Phi_0=\phi_0-\frac{3}{4}.
\end{align}

\paragraph{Plunging$_*$ $\Leftrightarrow$ Outward$_*$ (near-NHEK/near-NHEK).} The orbits are related by the $\rightleftarrows$-flip \eqref{Tphiflip}.

\paragraph{Plunging$_*$ $\Leftrightarrow$ Plunging$_*(e)$ (near-NHEK/near-NHEK).}
The two (future-directed) orbits are related via the diffeomorphism 
\begin{align}
\bar t &= \frac{1}{2\kappa}\log\frac{\sqrt{R^2-\kappa^2}\cosh{\kappa t}-R}{\sqrt{R^2-\kappa^2}\cosh{\kappa t}+R}-\frac{i\pi}{\kappa},\nonumber\\
    \bar R &= \sqrt{R^2-\kappa^2}\sinh{\kappa t},\label{ct3}\\
    \bar \phi &= \phi + \frac{1}{2}\log\frac{R\sinh\kappa t+\kappa\cosh\kappa t}{R\sinh\kappa t-\kappa\cosh\kappa t}\nonumber.
\end{align}
The energy of the new trajectory is a function of the initial time $t_0$ of the former one:
\begin{equation}
    \bar e=\kappa^2\ell_*\exp\qty(-\kappa t_0) > 0.
\end{equation}

\paragraph{Plunging$_*(e)$ $\Leftrightarrow$ Outward$_*(e)$ (near-NHEK/near-NHEK).} The orbits are related by the $\rightleftarrows$-flip.

\paragraph{Plunging$_*(E)$ $\Leftrightarrow$ Bounded$_*^-(E)$ (NHEK/NHEK).} The critical bounded orbit is obtained from the plunging orbit by a continuation of the trajectory beyond the horizon ($R<0$) combined with $\mathbb{Z}_2$ flips. One must proceed in three steps:
\begin{enumerate}
    \item Continue the plunge defined from the physical domain $0\leq R\leq \infty$ to its whole domain of definition $R_0\leq R\leq\infty$ (\textit{i.e.}, up to the root of the radial potential $R_0=-\frac{e}{2\ell_*}$) and consider now only the part of the trajectory located beyond the horizon $R_0 \leq R\leq 0$.
    \item Apply the \rotatebox[origin=c]{-45}{$\rightleftarrows$}-flip to the latter part of the solution. This transformation restores the positivity of the radial coordinate. It preserves the time orientation of the geodesic but flips the sign of its angular momentum $\ell_*\to-\ell_*$. The new domain of definition of the trajectory is consequently $0\leq R \leq \frac{E}{2\ell_*}$.
    \item The procedure outlined above only leads to the part of the geodesic with $R'(\lambda)>0$, which is located before the turning point. As outlined in Appendix \ref{app:equatorial}, the part of a bounded trajectory located after the turning point can be obtained from the one located before it by a $\rightleftarrows$-flip. 
\end{enumerate}
This whole procedure is represented in Fig. \ref{fig:bounded_from_plunging}.

\paragraph{Plunging$_*(e)$ $\Leftrightarrow$ Bounded$_*^-(e)$ (near-NHEK/near-NHEK).} The mapping is similar to the one outlined above using the \rotatebox[origin=c]{-45}{$\rightleftarrows$}-flip. One subtlety is that one should start with the Plunging$_*(e)$ orbit with $e > \kappa \ell_*$ in order to obtain the future-directed Bounded$_*^-(e)$ orbit.

\paragraph{Plunging$_*(e)$ $\Leftrightarrow$ Bounded$_*(-e)$ (near-NHEK/near-NHEK).} We apply the \rotatebox[origin=c]{-45}{$\rightleftarrows$}-flip as outlined in the previous paragraph, but now choosing $0< e < \kappa \ell_*$. This leads to a retrograde past-directed bounded orbit. The future-directed prograde geodesic is then reached using the $\uparrow\!\downarrow$-flip.

\subsection{Supercritical $\mathcal{C}<0$}
\paragraph{Spherical$(\ell)$ $\Leftrightarrow$ Marginal$(\ell)$ (near-NHEK/NHEK).}
One applies the NHEK/near-NHEK diffeomorphism
\begin{align}
    T&=-\exp\qty(-\kappa \bar t)\frac{\bar R}{\sqrt{\bar R^2-\kappa^2}}\nonumber,\\
    R&=\frac{1}{\kappa}\exp\qty(\kappa \bar t)\sqrt{\bar R^2-\kappa^2},\label{eqn:diffeo}\\
    \Phi&=\phi-\frac{1}{2}\log\frac{\bar R-\kappa}{\bar R+\kappa}\nonumber
\end{align}
which maps the orbit Spherical$(\ell)$ on the past-directed Marginal$(-\ell)$ orbit. The future-directed Marginal$(\ell)$ orbit is recovered by composing this transformation with a $\uparrow\!\downarrow$-flip.

\paragraph{Marginal$(\ell)$ $\Leftrightarrow$ Plunging$(E,\ell)$ or Def\mbox{}lecting$(E,\ell)$ (NHEK/NHEK).}
One performs the transformation ($\zeta\neq 0$)
\begin{align}
	\bar T &=\frac{1}{\bar R}\frac{2R^2T\cos\zeta-(1+R^2(1-T^2))\sin\zeta}{2 R},\nonumber\\
    \bar R &=\frac{R^2(1+T^2)-1+(1+R^2(1-T^2))\cos\zeta+2R^2T\sin\zeta}{2R},\label{eqn:NHEK_shift}\\
    \bar\Phi &= \Phi+\log\frac{\cos\frac{\zeta}{2}R+\sin\frac{\zeta}{2}(RT+1)}{\cos\frac{\zeta}{2}R+\sin\frac{\zeta}{2}(RT-1)}.\nonumber
\end{align}
As outlined in Ref. \cite{Compere:2017hsi}, this mapping can be viewed as the action on Poincaré NHEK coordinates of a shift of the global NHEK time $\tau\to\tau-\zeta$. The energy of the final orbit is
\begin{align}
    \bar E=\sqrt{-\mathcal C}\qty(\sin\zeta+T_0(\cos\zeta-1)).
\end{align}
We directly see that any energy $E\neq 0$ can be reached by conveniently choosing the values of $T_0$ and $\zeta$. 

\paragraph{Plunging$(E,\ell)$ $\Leftrightarrow$ Outward$(E,\ell)$ (NHEK/NHEK).} The orbits are related by the $\rightleftarrows$-flip.

\paragraph{Plunging$(E,\ell)$ $\Leftrightarrow$ Bounded$_>(E,-\ell)$ (NHEK/NHEK).} The mapping consists in extending the radial range of the plunging orbit beyond the horizon, $R<0$, then using the \rotatebox[origin=c]{-45}{$\rightleftarrows$}-flip, which leads to the Bounded$_>(E,-\ell)$ orbit.

\paragraph{Spherical$(\ell)$ $\Leftrightarrow$ Plunging$(e,\ell)$ or Def\mbox{}lecting$(e,\ell)$ (near-NHEK/ near-NHEK).}
One uses the diffeomorphism ($\chi\neq\pm 1$)
\begin{align}
	t &= \frac{1}{\kappa}\log\frac{\sqrt{\bar R^2-\kappa^2}\cosh \kappa\bar t-\bar R}{\sqrt{ R^2-\kappa^2}}\nonumber,\\
     R &=\sqrt{\bar R^2-\kappa^2}\qty(\sinh \kappa\bar t+\chi\cosh \kappa\bar t)-\chi\bar R,\label{eqn:nn_shift}\\
    \phi &= \bar \phi -\frac{1}{2}\log\qty[\frac{\sqrt{\bar R^2-\kappa^2}-\bar R\cosh \kappa\bar t+\kappa\sinh \kappa\bar t}{\sqrt{\bar R^2-\kappa^2}-\bar R\cosh \kappa\bar t-\kappa\sinh \kappa\bar t}\frac{R+\kappa}{R-\kappa}]\nonumber.
\end{align}
This mapping can be seen as a NHEK global time shift written in near-NHEK coordinates; see Refs. \cite{Compere:2017hsi,Hadar:2016vmk}. The explicit inversion formula can be found in Ref. \cite{Hadar:2016vmk}. The energy of the new trajectory reads as
\begin{align}
    \bar e=\kappa\sqrt{-\mathcal{C}}\,\chi.
\end{align}
For $-\frac{\ell}{\sqrt{-\mathcal{C}}}<\chi<-1$, the orbit reached is future-directed and deflecting. The trajectory becomes plunging for $\chi>-1$. Note that for $\abs{\chi}>1$, $\bar t_0=-\frac{1}{2\kappa}\log\frac{1+\chi}{1-\chi}$ is complex and one has to perform an additional shift on $\bar t$ to make it real.

\paragraph{Plunging$(e,\ell)$ $\Leftrightarrow$ Outward$(e,\ell)$ (near-NHEK/near-NHEK).} The orbits are related by the $\rightleftarrows$-flip.

\paragraph{Plunging$(e,\ell)$ $\Leftrightarrow$ Bounded$_>(e,-\ell)$ (near-NHEK/near-NHEK).} The mapping consists in extending the radial range of the plunging orbit with $e > \kappa \ell$ beyond the horizon, $r<0$, then using the \rotatebox[origin=c]{-45}{$\rightleftarrows$}-flip, which leads to the Bounded$_>(e,-\ell)$ orbit.

\subsection{Subcritical $\mathcal{C}>0$}

There is no near-NHEK spherical geodesic for $\mathcal{C}>0$. We can nevertheless introduce the formal class of \textit{complex} spherical trajectories
\begin{align}
    \begin{split}
    t(\lambda) &= -i\frac{\ell}{R_0}\lambda,\\
    R(\lambda) &= i R_0,\qquad R_0\triangleq\frac{\kappa\ell}{\sqrt{\mathcal{C}}},\\
    \phi(\lambda) &= \phi_0-\frac{3}{4}\ell\lambda+\ell\Phi_\theta(\lambda)
    \end{split}
\end{align}
which is a formal (but nonphysical) solution of the near-NHEK geodesic equations, of complex near-NHEK ``energy'' $e=-i\kappa\sqrt{\mathcal{C}}$. We will denote this class of solutions as Spherical$_\mathbb{C}(\ell)$ and show that it can be used to generate all subcritical bounded trajectories by acting on it with properly chosen conformal transformations. The parametrized form of the orbit reads as
\begin{align}
    \begin{split}
        R &= i R_0,\\
    \phi(t) &= \phi_0-\frac{3}{4}iR_0t+\ell\Phi_\theta(\lambda(t)).
    \end{split}
\end{align}

\paragraph{Spherical$_\mathbb{C}(\ell)$ $\Leftrightarrow$ Bounded$_<(E,\ell)$.}
 One has to proceed in two steps, mimicking the procedure used to obtain the NHEK Plunging$(E,\ell)$ class:
 \begin{itemize}
     \item  We apply the near-NHEK/NHEK diffeomorphism \eqref{eqn:diffeo} to a Spherical$_\mathbb{C}(\ell)$ orbit, leading to another complex NHEK geodesic of null energy parametrized by
\begin{align}
    \begin{split}
    T(R) &= -\frac{i\ell}{\sqrt{C}R},\\
    \Phi(R) &= \Phi_0-\frac{3i\ell}{8\sqrt{C}}\log\frac{\mathcal{C}R^2}{\mathcal{C}+\ell^2}
    \end{split}
\end{align}
with the initial azimuthal angle $\Phi_0\triangleq\phi_0-\frac{3\pi\ell}{8\sqrt{C}}-\frac{1}{2}\log\qty(1-\frac{2\sqrt{C}}{\sqrt{C}+i\ell})$. We denote this class as Marginal$_\mathbb{C}(\ell)$.
    \item Second, we apply to the trajectory found above the global time shift \eqref{eqn:NHEK_shift}, but upgraded with an \textit{imaginary} parameter $\zeta\to i\zeta$. This leads to the Bounded$_<(E,\ell)$ class with orbital parameters
    \begin{subequations}
    \begin{align}
        \bar E &= \sqrt{\mathcal{C}}\sinh \zeta,\\
        \bar\Phi_0 &= \phi_0-\frac{3\pi\ell}{8 \sqrt{\mathcal C}}-\log\qty(\sqrt{\mathcal{C}}-i\ell)+\frac{3i\ell}{8\sqrt{\mathcal{C}}}\log\qty[\mathcal C(\mathcal{C}+\ell^2)\qty(1+\sqrt{\frac{\mathcal{C}+E^2}{\mathcal{C}}})^2]\nonumber\\
        &~-\frac{3\ell}{8\sqrt{C}}\log\qty[E^2(\mathcal{C}+\ell^2)]+\arctan\frac{\sqrt{\mathcal{C}}}{\ell}.
    \end{align}
    \end{subequations}
    Note that choosing $\zeta>0$ is sufficient to reach the full range of energies allowed for such a geodesic ($E>0$). Any geodesic of orbital parameters ($T_0,\tilde\Phi_0$) can  finally be obtained by performing the transformation $T\to T+T_0$, $\Phi\to\Phi-\bar\Phi_0+\tilde{\Phi}_0$, which also removes the unphysical imaginary part of the azimuthal coordinate.
 \end{itemize}

\paragraph{Spherical$_\mathbb{C}(\ell)$ $\Leftrightarrow$ Bounded$_<(e,\ell)$.} We apply to the Spherical$_\mathbb{C}(\ell)$ class the near-NHEK global time shift \eqref{eqn:nn_shift} upgraded with an imaginary parameter $\chi\to i\chi$ ($\chi\neq\pm1$), leading to a Bounded$_<(e,\ell)$ orbit of parameters 
\begin{align}
\begin{split}
\bar e &= \kappa\sqrt{\mathcal{C}}\,\chi,\\
\bar t_0 &= t_0+\frac{i}{\kappa}\arctan\frac{\kappa\sqrt{\mathcal{C}}}{e},\\
\bar\phi_0 &= \bar\phi_0(\phi_0,e,\ell,\mathcal{C},\kappa).
\end{split}
\end{align}
The explicit value of $\bar\phi_0$ is easily calculable, but too long to be reproduced here. To reach a manifestly real orbit of orbital parameters $(\tilde t_0,\tilde\phi_0)$, one has to perform the final shift
\begin{equation}
    t\to t-\bar{t}_0+\tilde t_0,\qquad\phi\to\phi-\bar\phi_0+\tilde\phi_0.
\end{equation}

\begin{landscape}
\begin{figure}
    \centering
    \begin{minipage}{0.95\textwidth}
    \includegraphics[width=\textwidth]{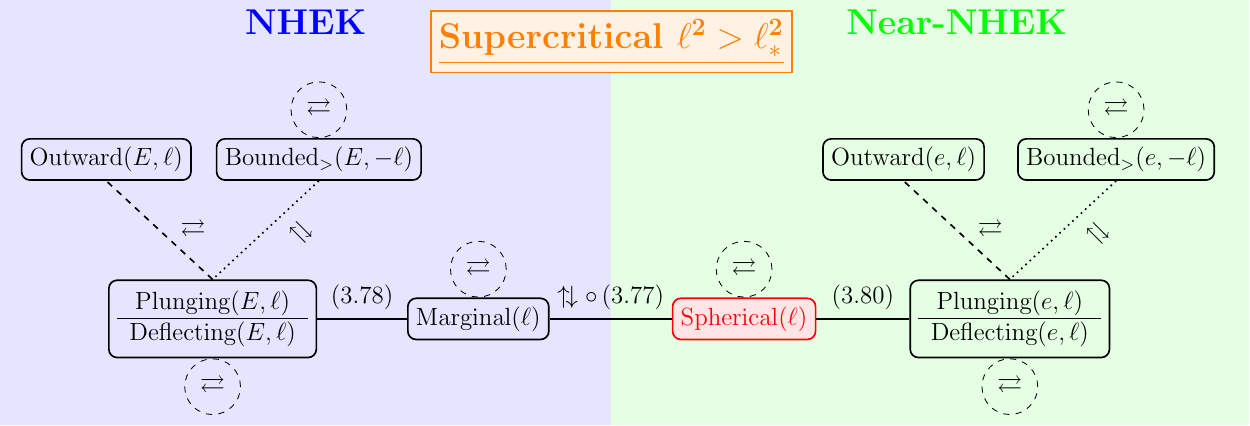}\\
    
    \includegraphics[width=\textwidth]{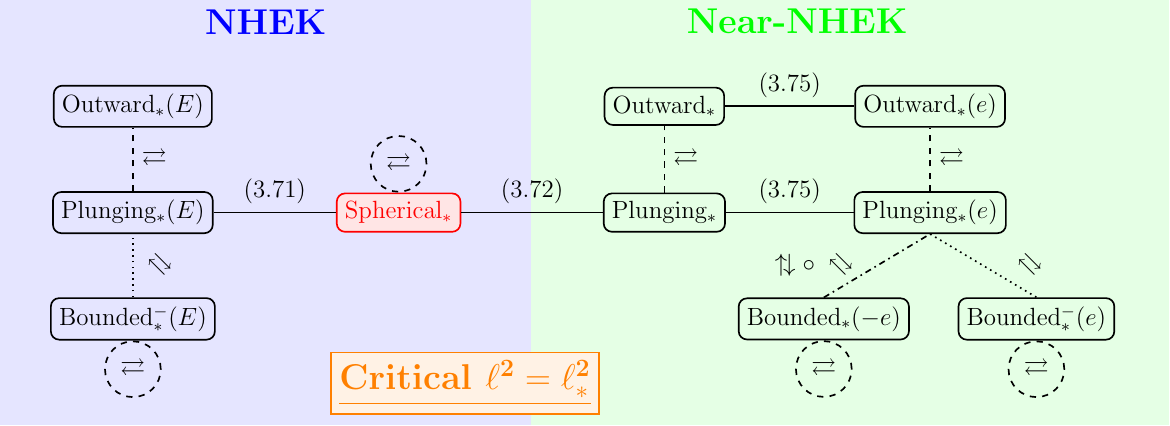}
    \end{minipage}
    \begin{minipage}{0.28\textwidth}
    \includegraphics[width=\textwidth]{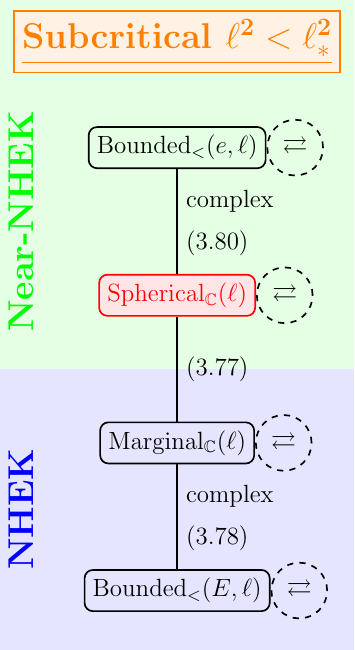}
    \end{minipage}
    
    \caption{Schematic overview of the action of the group  $\textsf{SL}(2,\mathbb{R})\times \textsf{U}(1) \times \uparrow\!\downarrow \times \rightleftarrows \times\!$\protect\rotatebox[origin=c]{-45}{$\rightleftarrows$} on near-horizon geodesics.
    }
    \label{fig:conformal_mappings}
\end{figure}{}
\end{landscape}

\partimage[width=\textwidth]{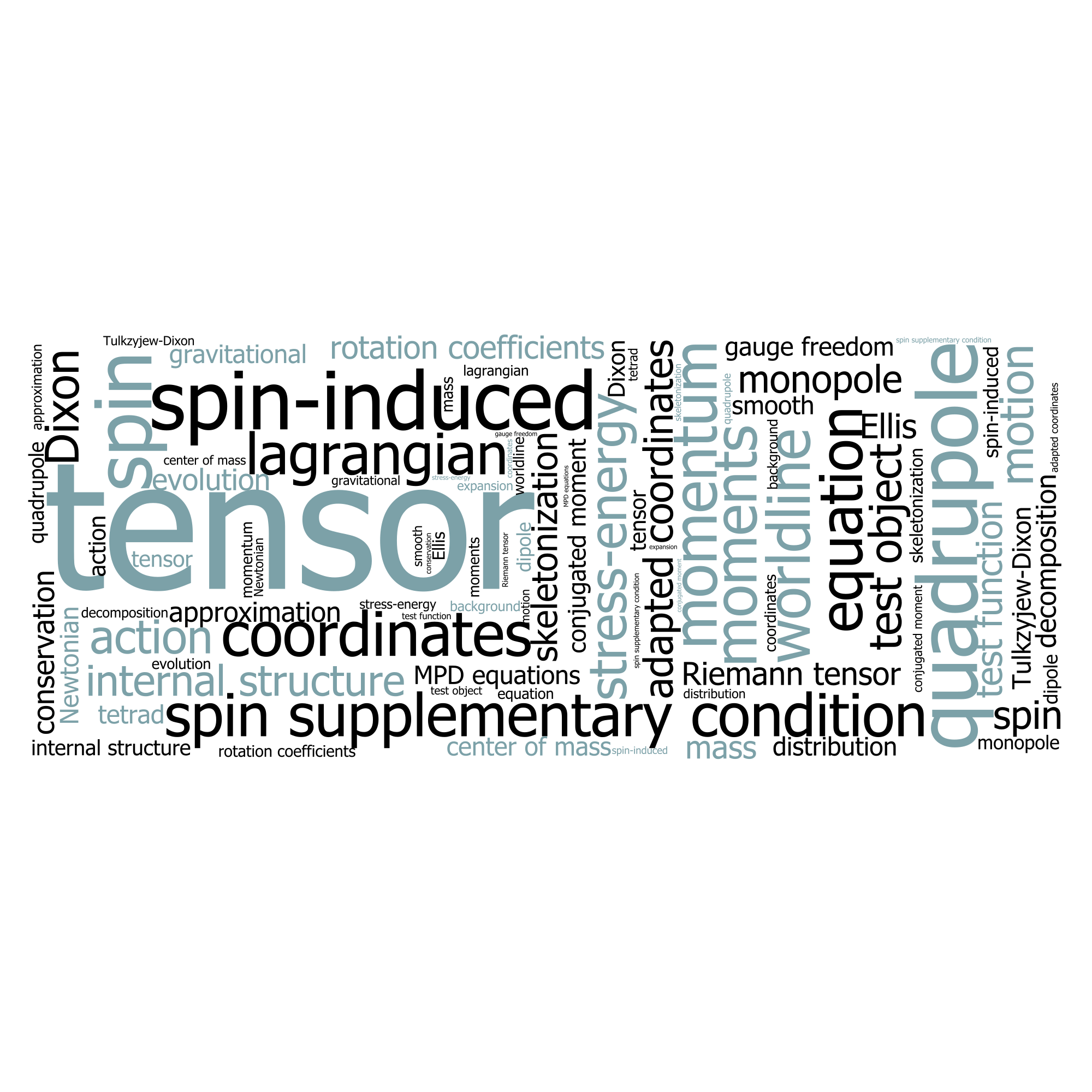}

\part[{\textsc{Test bodies in Curved Spacetime: Theoretical Foundations}}]{\textsc{Test bodies in Curved Spacetime:\newline Theoretical Foundations}}\label{part:spinnin_bodies}
{
\renewcommand{\thefigure}{II.1}
\begin{figure}
    \centering
    \includegraphics[height=7.5cm]{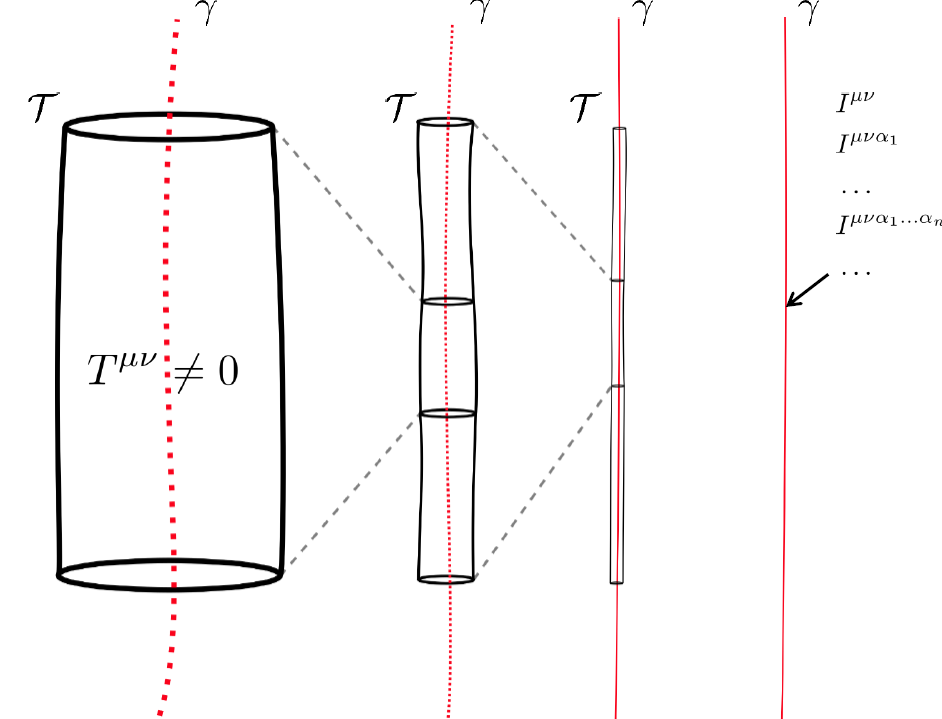}\hspace{0.5cm}\rule{0.5pt}{8cm}\hspace{0.5cm}
    \includegraphics[height=7.5cm]{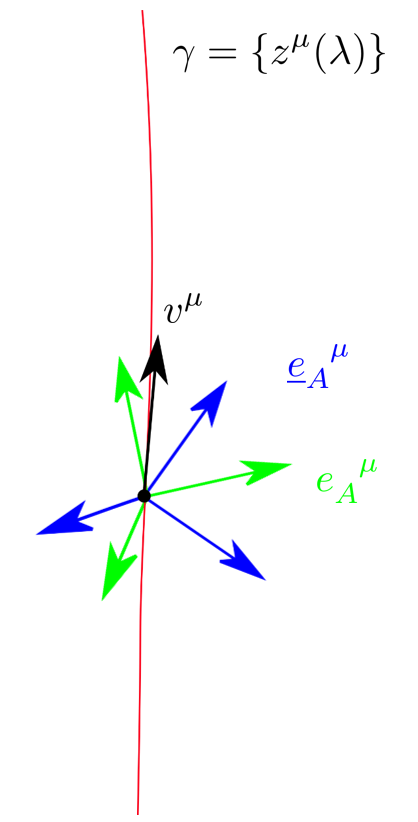}
    \caption{The two main pictures for the description of extended test bodies used in this text: the gravitational skeletonization (left) and the Lagrangian formulation (right).}
    \label{fig:approximations_schemes}
\end{figure}
}

\noindent 

\lettrine{L}{et} us consider the motion of a object described by some smooth stress-energy tensor $T_{\mu\nu}$ in a fixed background metric $g_{\mu\nu}$, thus neglecting self-force effects. Provided that the body has a finite spatial extension, its stress-energy tensor is supported on compact slices for any 3+1 decomposition of the spacetime. Such an object will be referred to as an \defining{extended test body}. We have in mind the motion of a ``small'' astrophysical object (stellar mass black hole or neutron star) around a hypermassive black hole. In this situation, the former can be viewed as a perturbation of the spacetime geometry created by the later. 

While the geodesic equations describe the motion of a structureless monopole test body in a fixed background spacetime, an important generalization is to allow the test body (while still having negligible mass and thus negligible influence on the gravitational field) to have a finite size and a nontrivial structure. All these effects -- departing from a bare geodesic motion -- are known as \defining{finite size effects}.

\subsubsection{Worldline description of extended test bodies}

In the case where the extended test body is \defining{compact}, that is, if its typical size $\ell$ is small compared to the radius of curvature $r$ of the background ($\ell\ll r$), there exists various equivalent approximation schemes for describing its motion in a somehow simpler way than considering its full stress-energy tensor. Both approaches end into characterizing the body by a centroid worldline $\gamma=\qty{z^\mu(\lambda)}$ along which a tower of gravitational multipole moments $I^{\mu\nu\alpha_1\ldots\alpha_n}$ replacing the original stress-energy tensor are defined, see Figure \ref{fig:approximations_schemes}. These moments can be understood as spatial integrals of $T^{\mu\nu}$,
\begin{equation*}
    I^{\mu\nu\alpha_1\ldots\alpha_n}\triangleq\int_{x^0=\text{cst}}\dd[3]{x}\sqrt{-g}T^{\mu\nu}\delta x^{\alpha_1}\ldots\delta x^{\alpha_n},
\end{equation*}
where $\delta x^\mu\triangleq x^\mu-z^{\mu}(\lambda)$.

The first of these schemes is known as the \defining{gravitational skeletonization}: the body is described by a distributional stress-energy tensor, which is non-vanishing only on a certain worldline, and contains the aforementioned tower of multipole moments. This tensor must be conserved within the background, $\nabla_\mu T^{\mu\nu}=0$, and one can show that it implies that the monopole $p_\mu$ and dipole $S_{\mu\nu}$ must evolve according to the \defining{Mathisson-Papapetrou-Dixon (MPD) equations} \cite{Mathisson:1937zz,1951RSPSA.209..248P,1970RSPSA.314..499D}
\begin{align*}
      \frac{\text{D}p^\mu}{\dd\tau}&=-\frac{1}{2}R\tud{\mu}{\nu\alpha\beta}v^\nu S^{\alpha\beta}+\ldots,\qquad
      \frac{\text{D}S^{\mu\nu}}{\dd\tau}=2 p^{[\mu}v^{\nu]}+\ldots,
\end{align*}
where the dots represent corrections due to the quadrupole and higher multipole moments. The monopole $p^\mu$ takes the interpretation of the linear momentum of the body, whereas the dipole $S^{\mu\nu}$ can be seen as its skew-symmetric spin tensor, describing the relativistic angular momentum of the body. This last object will play a central role in our description, which is the reason why we will sometimes refer to extended test bodies as \textit{spinning test bodies}. In terms of the original smooth stress-energy distribution, one can show that the two first moments are related to the original stress-energy tensor by
\begin{align*}
    p^\mu&\triangleq\int_{x^0=\text{constant}}\dd[3]{x}\sqrt{-g}\,T^{\mu 0}, \\
    S^{\mu\nu}&\triangleq\int_{x^0=\text{constant}}\dd[3]{x}\sqrt{-g}\,\qty(\delta x^\mu T^{\nu0}-\delta x^{\nu}T^{\mu 0}). 
\end{align*}

This approach has been investigated since the late thirties. The leading-order EOMs were first derived in the seminal works of M. Mathisson \cite{Mathisson:1937zz,mathisson_1940} and A. Papapetrou \cite{1951RSPSA.209..248P}. They have been subsequently generalized to higher multipolar orders by W.G. Dixon \cite{tul1959,dixon1973,dixon1974,dixon1979}. Despite its elegance and rigour, this approach appears to be quite long to perform and technically involved, discarding it from being a well-suited viewpoint for exposing comprehensively the problem in an introductory text like the present one.

We will instead follow the \defining{Lagrangian approach}, whose generic formulation in curved spacetime is due to I. Bailey and W. Israel in 1975 \cite{Bailey:1975fe}. Nevertheless, this method was pioneered for Special Relativity in earlier works, notably by A.J. Hanson and T. Regge (see \textit{e.g.} \cite{Westpfahl:1969ea,Hanson:1974qy}). It consists in formulating a generic action principle for the extended body modelled as a worldline, representing the motion of some ``center'' of its stress-energy distribution, and endowed with an orthonormal tetrad rigidly attached to it, whose evolution describes the orientation of the body. This approach also leads to the very same MPD equations. Having two equivalent descriptions of the same problem is extremely fruitful, since each of them turns out to be more appropriated for different purposes: as we will see, the skeletonization will be powerful for providing us with physical insights about the interpretation of multipole moments, while the Lagrangian approach will reveal particularly useful when we will turn to the Hamiltonian description of extended test bodies.

Others routes yielding the same equations of motion have also been followed. A supersymmetric description of classical spinning particles has been provided in 1993 by G.W. Gibbons, R.H. Rietdijk and J.W. van Holten \cite{Gibbons:1993ap}, and recently extended to include quadrupole effects \cite{Jakobsen:2021zvh}. Another recent (and somehow elegant) formulation accounting for the description of finite size effects is due to A. Harte \cite{Harte:2011ku,Harte:2014wya}, using the concept of ``generalized Killing vector fields''. Its main interest it that it allows naturally to account for the inclusion of gravitational back-reaction effects.

\subsubsection{Spin supplementary conditions}

There is a technical subtlety arising when studying the motion of test bodies described by the MPD equations. In order to obtain a closed system of equations, they have to be supplemented by an algebraic condition of the form $\mathcal V_\mu S^{\mu\nu}=0$, for some timelike vector field $\mathcal V^\mu$. Such conditions are known as \defining{spin supplementary conditions}. Physically, enforcing this kind of condition amounts to fix a choice of centroid worldline, setting to zero the mass-dipole moment in the proper frame whose timelike vector is aligned $\mathcal V^\mu$. The discussion of what is really happening is quite subtle, the main reason being that the notion of center of mass is observer-dependent in relativity.

\subsubsection{Truncation of the multipole expansion}

Another point of concern is to understand if the multipole expansion introduced above can be consistently truncated as some desired order, that is, if there exists a small parameter such that the magnitude of the successive multipoles decreases when the order of the multipoles increases. As we will see in Chapter \ref{chap:skeleton}, for compact objects, this small parameter will be the ratio between the typical size of the object and the typical radius of the background curvature, which is small by assumption.

As discussed in the introduction of the thesis, we will always truncate the expansion at the \textit{quadrupole} order. This is the first order in the multipole expansion where the internal structure of the body begins to matter. At pole-dipole order, the motion of finite size bodies is universal, in the sense that it is independent of the nature of the object.
Because we exclude the self-force in our description, our expansion will thus be valid at zeroth order in the mass ratio and at second order in the spin.

MPD equations promote the two first multipole $p^\mu$ and $S^{\mu\nu}$ to the rank of dynamical variables, but leave the higher order multipoles acting as sources. The latter shall consequently be prescribed depending on the internal structure of the test body. In this thesis, we will only be concerned with multipole moments induced by the proper rotation of the object, also known as \defining{spin-induced multipoles}. We therefore discard tidal and other type of contributions to the multipole structure. As we will see, this is the relevant description for modelling a binary black hole system evolving in vacuum, and this is the choice of multipole structure that will allow the existence of the largest number of conserved quantities along the motion.

Actually, for compact test bodies, they are several equivalent way of thinking about the truncation of the multipole expansion, which are consistent one to another, as we will check explicitly: (i) as a truncation of the number of multipoles that we use to describe the stress-energy tensor, which is the viewpoint of gravitational skeletonization that will be discussed in Chapter \ref{chap:skeleton}; (ii) as a truncation of the number of the derivatives of the background Riemann tensor upon which the action of the Lagrangian formulation can depend upon, as will be described in Chapter \ref{chap:EOM}; and finally (iii) as an expansion in integer powers of the magnitude of the spin dipole $\mathcal S$. This latter viewpoint turns out to be consistent with the two former for spin-induced multipoles, as will be reviewed in Chapter \ref{chap:quadrupole}.

\subsubsection{Plan of the text}
This part of the thesis is organized as follows: Chapter \ref{chap:EOM} will describe the Lagrangian formulation for extended test bodies in full generality, up to quadrupole order. In Chapter \ref{chap:skeleton}, we will discuss a particularly simple form of the gravitational skeletonization up to dipole order, which will enable to gain more intuition about the physical meaning of the monopole and the dipole moments. Our discussion will be however specialized to a specific choice of coordinates. Chapter \ref{chap:SSC} will describe the problem of enforcing the aforementioned spin supplementary conditions, as well as their physical interpretation. Finally, Chapter \ref{chap:quadrupole} will be devoted to spin-induced multipoles, and will focus on the explicit construction of the spin-induced quadrupole moment. 

\chapter{Lagrangian Formulation}
\label{chap:EOM}


This chapter discusses the Lagrangian formalism for spinning test bodies in General Relativity. In this text, we will always restrict our derivations up to quadrupole order. Nevertheless, higher orders can be reached, see \textit{e.g.} \cite{Marsat_2015}. This chapter mainly follows the excellent exposition of Marsat \cite{Marsat_2015}.

The core idea of the Lagrangian approach is to construct the most generic worldline Lagrangian action $S=\int L\,\dd\lambda$ describing the motion of a spinning test body in curved spacetime. As we will see, the form of the allowed Lagrangian $L$ can be highly constrained from very generic symmetry arguments.
Like in any classical mechanics problem, the equations of motion can then be derived from the associated first order variational principle $\delta S=0$ \cite{Goldstein2001,arnold1989mathematical}.

\section{Rotational degrees of freedom}
The two main questions one should ask for building an action are 
\begin{enumerate}
    \item What are the relevant degrees of freedom that should be introduced for describing a spinning body in curved spacetime ?
    \item What are the symmetries under which the action should be invariant?
\end{enumerate}
This section aims to tackle the first of them.

Let us denote $z^\mu(\lambda)$ the body's worldline. Here, $\lambda$ is an arbitrary ``time'' parameter describing the evolution of the motion. We also define the four-velocity
\begin{align}
    v^\mu\triangleq\dv{z^\mu}{\lambda}.
\end{align}
For any physical massive object, the four-velocity will be a timelike vector, $v_\mu v^\mu<0$. In the canonical language of Lagrangian mechanics, the four components $v^\mu$ will play the role of the velocities describing the position of the test body and associated to the coordinates $z^\mu$.

All along this discussion, the specific form of $\lambda$ as well as the normalization of the four-velocity will be left arbitrary at the level of the action; they will only be constrained later at the level of the equations of motion, setting $\lambda$ to be the proper time and consequently yielding the standard normalization $v_\mu v^\mu=-1$.

We are left with the problem of choosing the degrees of freedom that will account for the rotational orientation of the test body.
Following the proposition of Hanson and Regge for Special Relativity \cite{Hanson:1974qy}, the \textit{spin} (that is, the rotational) degrees of freedom of the body will be represented by an orthonormal tetrad $e\tdu{A}{\mu}(\lambda)$ rigidly attached to the body's worldline. Its orientation at any value of $\lambda$ will be measured thanks to the introduction of another \textit{background} orthonormal tetrad frame $\underline{e}\tdu{\underline A}{\mu}(x)$. At any point of the worldline, these two tetrads are related by a Lorentz transformation $\Lambda\tud{A}{\underline A}(\lambda)$:
\begin{align}
    \underline{e}\tdu{\underline A}{\mu}\qty(z(\lambda))=\Lambda\tud{A}{\underline A}(\lambda) e\tdu{A}{\mu}(\lambda).\label{tfo_tetrads}
\end{align}
Of course, we have the standard relations for Lorentz matrices
\begin{align}
    \Lambda\tdu{A}{\underline A}(\lambda)\Lambda_{B\underline A}(\lambda)=\eta_{AB},\qquad \Lambda_{A\underline A}(\lambda)\Lambda\tud{A}{\underline B}(\lambda)=\eta_{\underline A\,\underline B},
\end{align}
and for tetrad frames
\begin{subequations}
\begin{align}
    e\tdu{A}{\mu}(\lambda)e_{B\mu}(\lambda)&=\eta_{AB},\quad e\tdu{A}{\mu}(\lambda)e^{A\nu}(\lambda)=g^{\mu\nu}(z(\lambda)),\label{metric_e}\\
    \underline e\tdu{\underline A}{\mu}(x)\underline e_{\underline B\mu}(x)&=\eta_{\underline A\,\underline B},\quad \underline e\tdu{\underline A}{\mu}(x)\underline e^{\underline A\nu}(x)=g^{\mu\nu}(x).
\end{align}
\end{subequations}
The evolution of the body's tetrad will be described using the standard antisymmetric rotation coefficients $\Omega^{\mu\nu}$ (see \textit{e.g.} \cite{poisson_2004})
\begin{align}
    \frac{\text{D}e\tdu{A}{\mu}}{\dd\lambda}\triangleq-\Omega^{\mu\nu} e_{A\nu}\qquad\Leftrightarrow\qquad
    \Omega^{\mu\nu}\triangleq e^{A\mu}\frac{\text D e\tdu{A}{\nu}}{\dd\lambda}.
\end{align}
Here and in the remaining of this text, we use the notation $\text D/\dd\lambda\triangleq v^\alpha \nabla_\alpha$.

The Lorentz matrices $\Lambda\tud{\underline A}{A}(\lambda)$ encode all the informations regarding the orientation of the body's tetrad with respect to the background. As any homogeneous Lorentz transformation, they contain 6 degrees of freedom: three of them represent spatial rotations, and the three others relativistic boosts. Intuitively, one can see the three rotational degrees of freedom (DOFs) as being the spin ones, whereas the three boosts originate from the fact that one has not chosen yet the exact position of the worldline $z^\mu(\lambda)$ inside the body's worldtube. This ambiguity will be extensively discussed and resolved in Chapter \ref{chap:SSC}, by enforcing a so-called \textit{spin supplementary condition} (SSC).

\section{Constraining the action}
It is now time to write down an action for our theory. It seems natural to require the following symmetry requirements to hold \cite{Blanchet_2014}:
\begin{itemize}
    \item \textbf{Spacetime diffeomorphisms}: as any GR scalar expression, the action should be invariant under any generic spacetime diffeomorphism $x^\mu\to x^{\mu'}\qty(x^\mu)$. The Lagrangian should consequently be a tensorial scalar, in the sense that all the spacetime indices of the objects it is built from should be properly contracted between themselves;
    \item \textbf{Lorentz transformations}: the action should be invariant under local Lorentz transformations, which transform the body and the background tetrad as 
    \begin{align}
        e\tdu{A}{\mu}\to \Lambda\tud{A}{B}e\tdu{B}{\mu},\qquad\underline e\tdu{\underline A}{\mu}\to \bar\Lambda\tud{\underline A}{\underline B}\underline e\tdu{\underline B}{\mu}.
    \end{align}
    It amounts to require all the Lorentz indices of the tetrads to be properly contracted;
    \item \textbf{Reparametrization invariance}: the time parameter $\lambda$ being arbitrary, the action \eqref{action} should be invariant under any reparametrization $\lambda\to\lambda'(\lambda)$ of the trajectory.
\end{itemize}

In order to actually describe a spinning body, the Lagrangian should kinematically depend on the worldline velocity $v^\mu$ and on the rotation coefficients $\Omega^{\mu\nu}$, but not on the ``positions'' ($z^\mu$ and $e\tdu{A}{\mu}$) themselves, for the purpose of ensuring general covariance. Moreover, we forbid any dependence in the background structure $\underline e\tdu{\underline A}{\mu}$, so that our description depends only on degrees of freedom intrinsic to the body. The prescribed Lagrangian should account for finite size effects, that is, dynamical effects originating from the coupling between the body's spin and the background's curvature. The later is accounted for by the Riemann tensor and its derivatives. Notice that the background metric $g_{\mu\nu}$ is assumed to appear only for the purpose of contracting indices, thus allowing to construct scalars from the other tensorial objects in a natural way. We however forbid any dependence upon first derivatives of the metric (that is, upon Christoffel symbols). Derivatives of the metric are only allowed to enter in the action through the Riemann tensor and its derivatives.

Excluding any coupling with other external fields and given the discussion above, the generic action for an extended test body is then assumed takes the form:
\begin{align}
    S\qty[z^\mu,e\tdu{A}{\mu}]=\int_\gamma L\qty(v^\mu,\Omega^{\mu\nu},g_{\mu\nu}(z),R_{\mu\nu\rho\sigma}(z),\nabla_\lambda R_{\mu\nu\rho\sigma}(z),\ldots)\dd\lambda.\label{action}
\end{align}
The subscript $\gamma$ just refers to the fact that the integration over $\lambda$ is actually an integration over the worldline $\gamma$.

\subsubsection{\textbf{Homogeneity condition}} The next step will be to constrain the generic form of the action Eq. \eqref{action}. Actually, a very simple argument allows to provide a simple explicit -- but still non-uniquely fixed -- expression for the Lagrangian. As we have just mentioned, the action Eq. \eqref{action} should be invariant under any reparametrization of the trajectory; in particular, it should be invariant under a scaling $\lambda\to\Delta\lambda$ ($\Delta\neq 0)$. This implies that the Lagrangian must be \textit{homogeneously linear} in $v^\mu$ and $\Omega^{\mu\nu}$, which both scale as $\Delta^{-1}$ under this transformation. Euler's theorem on homogeneous functions\footnote{
Let $f:\mathbb R^n\to\mathbb R$ be a positively homogeneous function of degree $k\in\mathbb Z$, \textit{i.e.} 
\begin{align*}
    \forall \Delta>0:f(\Delta x_1,\ldots\Delta x_n)=\Delta^k f(x_1,\ldots,x_n)
\end{align*}
which is continuously differentiable in some open subset $\mathcal U\subset\mathbb R^n$.
Then, 
\begin{align*}
    k f(x_1,\ldots,x_n)=\sum_{i=1}^nx_i\pdv{f}{x_i}\qty(x_1,\ldots,x_n),\qquad\forall \qty(x_1,\ldots,x_n)\in\mathcal U.
\end{align*}
}
then implies that
\begin{align}
    L(v^\mu,\Omega^{\mu\nu},g_{\mu\nu},R_{\mu\nu\rho\sigma},\nabla_\lambda R_{\mu\nu\rho\sigma},\ldots)=\pdv{L}{v^\mu}v^\mu+\pdv{L}{\Omega^{\mu\nu}}\Omega^{\mu\nu}.
\end{align}
Defining the conjugate momenta\footnote{The factor $2$ in the definition of $S_{\mu\nu}$ is present for consistency with the conventions used in the literature.} (respectively referred to as the \textit{linear momentum} and the \textit{spin tensor})
\begin{align}
    p_\mu\triangleq\pdv{L}{v^\mu},\qquad S_{\mu\nu}\triangleq 2\pdv{L}{\Omega^{\mu\nu}},\label{momenta}
\end{align}
one can write
\begin{align}
    L=p_\mu v^\mu+\frac{1}{2}S_{\mu\nu}\Omega^{\mu\nu}.\label{generic_lagragian}
\end{align}

Be careful: we have \textit{not} provided a unique expression for the Lagrangian of our theory. The momenta $p_\mu$ and $S_{\mu\nu}$ remain here arbitrary functions of $v^\mu$, $\Omega^{\mu\nu}$,  the Riemann tensor and its derivatives. They will be fixed when a spin supplementary condition will be enforced, which will provide us with an explicit relation between the linear momentum $p_\mu$ and the four-velocity $v^\mu$. Nevertheless, the form of the Lagrangian \eqref{generic_lagragian} is relevant to mention for later purposes (\textit{e.g.} Hamiltonian description of the present problem); in particular, it is the same regardless to the finite size interactions allowed in the theory, \textit{i.e.} regardless to the dependence of $L$ in the Riemann tensor and its derivatives that is allowed. 

\subsubsection{\textbf{Quadrupole approximation}}
We will now constrain our theory by restricting the functional dependency of the Lagrangian in the (derivatives of the) Riemann tensor. In the continuation of this text, we will always restrict ourselves to the so-called \textit{quadrupole approximation} which consists into allowing $L$ to depends on the Riemann tensor, but not on its derivatives:
\begin{align}
    L=L\qty(v^\mu,\Omega^{\mu\nu},g_{\mu\nu},R_{\mu\nu\rho\sigma}).\label{L_quad}
\end{align}
More generically, allowing the Lagrangian to depend in the Riemann tensor up to its $n^\text{th}$ derivative will lead to the appearance of $2^{n+2}$-pole moments terms in the equations of motion. For an action of the form \eqref{action}, the multipole moments are only sourced by the spin of the body. For spin-induced moments, one can show (see Chapter \ref{chap:quadrupole}) that the $2^n$-pole moment scales as $\mathcal O\qty(\mathcal S^{n})$, where the spin magnitude $\mathcal S$ is defined as
\begin{align}
    \mathcal S^2\triangleq\frac{1}{2}S_{\mu\nu}S^{\mu\nu}.
\end{align}
The aforementioned approximation makes sense, because (i) the spin magnitude will turn out to be a constant of motion, regardless to the multipole order we are working with and because (ii) $\mathcal S$ can be assumed to be small in astrophysically realistic situations. It can consequently be used as an expansion parameter for setting up a perturbative treatment of the generic problem.

From a physical viewpoint, the quadrupole approximation amounts to consider deformations induced by the proper rotation of the object in our description up to quadrupole order, while neglecting higher order corrections.

Finally, let us notice that only the (mass) monopole moment $p_\mu$ and the (spin) dipole moment $S_{\mu\nu}$ are dynamical variables, because they are the only multipole moments present in the explicit form of the Lagrangian \eqref{generic_lagragian}. The higher moments are non-dynamical and entirely written in terms of these two first moments. They will act as sources in the equations of motion.

\section{Equations of motion}
Before deriving the equations of motion, it is useful to explicit some results concerning the first-order variations of the Lagrangian. First, in the quadrupole approximation, a generic variation of $L$ takes the form
\begin{align}
    \delta L=p_\mu\delta v^\mu+\frac{1}{2}S_{\mu\nu}\delta\Omega^{\mu\nu}+\pdv{L}{g_{\mu\nu}}\delta g_{\mu\nu}-\frac{1}{6}J^{\mu\nu\rho\sigma}\delta R_{\mu\nu\rho\sigma}.\label{variation}
\end{align}
Here, we have defined -- up to a proportionality factor -- the \defining{quadrupole moment} $J^{\mu\nu\rho\sigma}$ as the conjugate moment to the Riemann tensor:
\begin{align}
    J^{\mu\nu\rho\sigma}\triangleq -6\pdv{L}{R_{\mu\nu\rho\sigma}}.
\end{align}
Notice that the variation \eqref{variation} is independent of the explicit form \eqref{generic_lagragian} of the Lagrangian.

In particular, the action should be invariant under an infinitesimal change of coordinates $z^\mu\to z^\mu+\epsilon\xi^\mu$ ($\abs{\epsilon}\ll 1$). Particularizing the variation \eqref{variation} to this case and working in a locally inertial frame yields
\begin{align}
    \delta_\xi L&=\epsilon\qty(p^\nu v^\mu+S\tud{\nu}{\lambda}\Omega^{\mu\lambda}-2\pdv{L}{g_{\mu\nu}}+\frac{2}{3}J^{\nu\alpha\beta\gamma}R\tud{\mu}{\alpha\beta\gamma})\partial_\mu\xi_\nu.
\end{align}
This variation must vanish regardless to the value of $\xi^\mu$; therefore, the following constraint must hold:
\begin{align}
    p_\nu v_\mu+S_{\nu\lambda}\Omega\tdu{\mu}{\lambda}-2\pdv{L}{g_{\mu\nu}}+\frac{2}{3}J\tdu{\nu}{\alpha\beta\gamma}R_{\mu\alpha\beta\gamma}=0.
\end{align}
Taking the antisymmetric part of this expression allows to write
\begin{align}
    S^{\lambda[\mu}\Omega\tud{\nu]}{\lambda}=p^{[\mu}v^{\nu]}+\frac{2}{3}R\tud{[\mu}{\alpha\beta\gamma}J^{\nu]\alpha\beta\gamma},\label{S_Omega_terms}
\end{align}
which is valid in any frame. This last relation will become extremely useful in the following derivations.

\subsubsection{\textbf{Evolution equation for the spin tensor}}
The evolution equation for the spin tensor, also known as the \textit{precession equation}, is obtained by varying the action with respect to the body's tetrad $e\tdu{A}{\mu}$. In the background frame, the variation of the rotation coefficients takes the form
\begin{align}
    \delta\Omega^{\underline A\,\underline B}=e\tud{\underline A}{\mu}e\tud{\underline B}{\nu}\frac{\text{D}\,\delta\theta^{\mu\nu}}{\dd\lambda}+\Omega\tud{\underline A}{\underline C}\delta\theta^{\underline C\,\underline B}-\Omega\tud{\underline B}{\underline C}\delta\theta^{\underline C\,\underline A}.
\end{align}
For convenience, the variation of the tetrad has been entirely expressed in terms of the object
\begin{align}
    \delta\theta^{\underline A\,\underline B}\triangleq \Lambda^{A\underline A}\delta\Lambda\tdu{A}{\underline B}.
\end{align}

Plugging this result either in the explicit expression of the Lagrangian \eqref{generic_lagragian} or in the generic variation \eqref{variation}, one obtains
\begin{align}
    \delta_\theta L=\frac{1}{2}\qty(-\frac{\text{D}\, S^{\mu\nu}}{\dd\lambda}+S^{\nu\rho}\Omega\tud{\mu}{\rho}-S^{\mu\rho}\Omega\tud{\nu}{\rho})\delta\theta_{\mu\nu}.
\end{align}
Requiring this variation to be vanishing yields the evolution equation
\begin{align}
    \frac{\text{D}\, S^{\mu\nu}}{\dd\lambda}=S^{\nu\rho}\Omega\tud{\mu}{\rho}-S^{\mu\rho}\Omega\tud{\nu}{\rho}.
\end{align}

Before going further on, let us stress some points useful for the continuation of this work.
\begin{enumerate}
    \item A direct computations shows that the identity
    \begin{align}
        S\tud{\rho}{\mu}\Omega^{\mu\nu}S_{\nu\rho}=0
    \end{align}
    holds. It implies that the spin magnitude is conserved,
    \begin{align}
        \dv{\mathcal S}{\lambda}=0.
    \end{align}
    This conservation equation actually holds at any multipole order and is independent of the spin supplementary condition (see \textit{e.g.} \cite{Marsat_2015} and references therein).
    \item In the object's frame, the evolution equation becomes
    \begin{align}
        \frac{\text{D}\, S^{AB}}{\dd\lambda}=0.
    \end{align}
    The components of the spin tensor in the object's frame $S^{AB}$ are thus constant. This provides \textit{a posteriori} a justification to the statement that the tetrad $e\tdu{A}{\mu}$ is ``rigidly attached'' to the compact object.
    \item Using Eq. \eqref{S_Omega_terms}, one can eliminate the dependence of the precession equation in the rotation coefficients and make explicit its dependence in the Riemann tensor. A straightforward computation yields
    \boxedeqn{
        \frac{\text{D}\, S^{\mu\nu}}{\dd\lambda}=2 p^{[\mu}v^{\nu]}+\mathcal L^{\mu\nu},\qquad\mathcal L^{\mu\nu}\triangleq\frac{4}{3}R\tud{[\mu}{\alpha\beta\gamma}J^{\nu]\alpha\beta\gamma}.\label{MPD_spinning}
    }{MPD equation for the spin tensor}
    This is the standard form of the precession equation that can be found in the literature.
    \item Finally, contracting Eq. \eqref{MPD_spinning} with $v_\nu$, we remark that the momentum and the four velocity are not aligned anymore when spin is present, by contrast to the geodesic case:
    \begin{align}
        -v^2p^\mu=\mathfrak m v^\mu+p_\perp^\mu,\qquad p_\perp^\mu\triangleq \qty(-\frac{\text D S^{\mu\nu}}{\dd\lambda}+\mathcal L^{\mu\nu})v_\nu.
    \end{align}
    Here, $\mathfrak m\triangleq-v_\alpha p^\alpha$ denotes the body's mass in the frame attached to the worldline, also known as the \defining{kinetic mass}. The norm of the four-velocity can be set to $-1$ if the time evolution parameter is chosen to be the body proper time. The orthogonal component of the momentum $p^\mu_\perp$ can be expressed as a function of $x^\mu$, $v^\mu$ and $S^{\mu\nu}$ solely when a spin supplementary condition has been enforced, see Chapter \ref{chap:SSC}.
\end{enumerate}

\subsubsection{\textbf{Evolution equation for the linear momentum}}
The method for finding the evolution equation for the linear momentum is to vary the action with respect to the worldline. The procedure is the very same that the one which can be used for the derivation of the geodesic deviation equation \cite{Wald:1984rg,Misner1973}: let us consider an infinitesimal change of the worldline, parametrized by a displacement vector $\xi^\mu(\lambda)$ which is Lie-dragged along the worldline:
\begin{align}
    \mathcal L_v\xi^\mu=0\quad\Leftrightarrow\quad \xi^\lambda\nabla_\lambda v^\mu=v^\lambda\nabla_\lambda \xi^\mu.
\end{align}
In this case, the variation of the action takes the form
\begin{align}
    \delta_\xi S=\int_\gamma\delta_\xi L\,\dd\lambda=\int_\gamma\xi^\lambda\partial_\lambda L\,\dd\lambda=\int_\gamma\xi^\lambda\nabla_\lambda L\,\dd\lambda.
\end{align}
It is useful to notice that the following identities hold \cite{Marsat_2015}: 
\begin{subequations}
\begin{align}
    \xi^\lambda\nabla_\lambda v^\mu v^\mu&=\frac{\text{D}\,\xi^\mu}{\dd\lambda},\\
    \xi^\lambda \nabla_\lambda\Omega^{\mu\nu}&=-\frac{\text{D}}{\dd\lambda}\qty(e^{A\mu}\delta_\xi e\tdu{A}{\nu})+e^{A\mu}\frac{\text{D}\,\delta_\xi e\tdu{A}{\nu}}{\dd\lambda}-e^{A\nu}\frac{\text{D}\, \delta_\xi e\tdu{A}{\mu}}{\dd\lambda}-\xi^\alpha v^\beta R\tud{\mu\nu}{\alpha\beta}.
\end{align}
\end{subequations}
Here, we have defined
\begin{align}
    \delta_\xi e\tdu{A}{\mu}\triangleq\xi^\lambda\nabla_\lambda e\tdu{A}{\mu}
\end{align}
The value of this quantity is actually is arbitrary, since we are left with the freedom of choosing the way the tetrad is transported from the original worldline $z^\mu(\lambda)$ to the perturbed one $z^\mu(\lambda)+\xi^\mu(\lambda)$. Choosing the tetrad to be parallelly transported between the two worldlines allows to set
\begin{align}
    \delta_\xi e\tdu{A}{\mu}=0.
\end{align}
Gathering all the previous pieces, the evolution equation for the linear momentum can then be derived -- after integration by parts -- from the variational problem $\delta_\xi S=0$, yielding
    \boxedeqn{
    \frac{\text{D}\,p^\mu}{\dd\lambda}=-\frac{1}{2}R\tud{\mu}{\nu\alpha\beta}v^\nu S^{\alpha\beta}+\mathcal F^\mu,\qquad\mathcal F^\mu\triangleq-\frac{1}{6}J^{\alpha\beta\gamma\delta}\nabla^\mu R_{\alpha\beta\gamma\delta}.\label{MPD_impulsion}
    }{MPD equation for the linear momentum}
This is the standard evolution equation of the linear momentum at quadrupole order. Together with Eq. \eqref{MPD_spinning}, these equations are the Mathisson-Papapetrou-Dixon equations, restricted to quadrupole order. At higher orders, the structure of the equations remains the same, the contribution of the higher order multipoles being only contained in the force $\mathcal F^\mu$ and torque $ \mathcal L^{\mu\nu}$ terms \cite{Steinhoff_2010}.

\section{Stress-energy tensor}
An interesting computation to be performed at this stage of the discussion is to write out the stress-energy tensor of the theory. It is advantageously computed from the variation of the action with respect to the body's tetrad frame: first expressing the metric in terms of the tetrad in the action thanks to Eq. \eqref{metric_e}, the stress-energy tensor as be computed from
\begin{align}
    T_{\mu\nu}\triangleq\frac{1}{\sqrt{-g}}e_{a(\mu}\frac{\delta S}{\delta e\tdu{a}{\nu)}}.
\end{align}
The stress-energy tensor can be split as a sum over all the multipole orders involved:
\begin{align}
    T_{\mu\nu}=T^\text{pole}_{\mu\nu}+T^\text{dipole}_{\mu\nu}+T^\text{quad}_{\mu\nu}.
\end{align}
As usually when computing stress-energy tensors for point-like objects moving along worldlines, the action should be written as an integral over the spacetime by introducing a Dirac delta:
\begin{align}
    S=\int\dd[4]x\sqrt{-g}\int_\gamma\,L\qty(v^\mu,\Omega^{\mu\nu},e_{A\mu}e\tud{A}{\nu},R_{\mu\nu\rho\sigma})\delta_4(x,z)\dd\lambda.
\end{align}
Here, the symbol $\delta_4(x,z)$ stands for the diffeomorphism invariant Dirac distribution
\begin{align}
    \delta_4(x,z)\triangleq\frac{\delta^{(4)}(x-z)}{\sqrt{-g}},
\end{align}
where $\delta^{(4)}(x-z)$ is the standard four-dimensional Dirac distribution \cite{schwartz1966theorie}. After computation, the contributions of the right hand side are found to be given by
\begin{subequations}\label{T_lagrangian}
\begin{align}
    T^\text{pole}_{\mu\nu}&=\int_\gamma p_{(\mu}v_{\nu)}\delta_4(x,z)\dd\lambda,\\
    T^\text{dipole}_{\mu\nu}&=-\nabla_\lambda\int_\gamma S\tud{\lambda}{(\mu}v_{\nu)}\delta_4(x,z)\dd\lambda,\\
    T^\text{quad}_{\mu\nu}&=\frac{1}{3}\int_\gamma R\tdu{(\mu}{\alpha\beta\gamma}J_{\nu)\alpha\beta\gamma}\delta_4(x,z)\dd\lambda\nonumber\\
    &\quad-\frac{2}{3}\nabla_\lambda\nabla_\rho\int_\gamma J\tudu{\lambda}{(\mu\nu)}{\rho}\delta_4(x,z)\dd\lambda.
\end{align}
\end{subequations}
As expected, the stress energy tensor is not a function, but a distribution which is only non-vanishing on the body's worldline.

\chapter[Skeletonization of the stress-energy tensor]{Skeletonization\\of the stress-energy tensor}
\label{chap:skeleton}

\vspace{\stretch{1}}

Until now, we have obtained the Mathisson-Papapetrou-Dixon equations governing the motion of spinning test bodies in curved spacetime at quadrupole order, which are given by Eqs. \eqref{MPD_spinning} and \eqref{MPD_impulsion}. This was performed through writing down an action principle for our theory and deriving the associated equations of motion from the associated variational principle.

In this chapter, we will see that a totally different method allows to recover the very same equations of motion. It consists into replacing the smooth stress-energy tensor of the physical extended body by a distributional one, which is only supported on a single worldline encompassed in the body's worldtube. The equations of motion then follow from the stress-energy conservation equation. Justifying rigorously this approximation from first principles in GR is however very involved and technical. We refer the interested reader to the references mentioned in the introduction of Part \ref{part:spinnin_bodies} for more details (especially Dixon ones). In the present text, we will give some insights about the coherence of this approximation by comparing the (involved) GR situation to the (simpler) Newtonian one.

At quadrupole order, the computations associated to gravitational skeletonization turn out to be very cumbersome. This chapter aiming to provide a pedagogical introduction, we will restrict ourselves to the dipole order. Explicit computations for the quadrupole may be found in \cite{Steinhoff_2010,Gratus:2020cok}. 

This chapter is organized as follows: Section \ref{sec:newtonian_skeleton} reviews the gravitational skeletonization in Newtonian theory, which is then generalized to General Relativity in Section \ref{sec:GR_skeleton}. As we will see, they are several decompositions that can be chosen for performing the skeletonization. In Section \ref{sec:ellis_skeleton}, we use the \textit{Ellis decomposition} to recover the MPD equations at dipole order. The computations are carried out in a specific coordinates system, the \textit{adapted coordinates}, which enable to reduce dramatically the length and the technicality of the derivation.

\section{Invitation: gravitational skeleton in Newtonian gravity}
\label{sec:newtonian_skeleton}

\textit{This section is mainly based on \cite{Ramond_2021,Steinhoff_2010}.}

In order to acquire some feeling about the form of the Ansatz of the GR's gravitational skeleton of the stress-energy tensor, let us have a look at the equivalent problem in Newtonian gravity. The gravitational potential $U(t,\mathbf x)$ created by an object of mass density $\rho(t,\mathbf x)$ enclosed in a volume $\mathcal V\subset\mathbb R^3$ is solution of the Poisson equation
\begin{align}
    \Delta U(t,\mathbf x)=-4\pi\rho(t,\mathbf x).
\end{align}
This equation can be solved analytically, and its solution reads (up to a trivial additive constant)
\begin{align}
    U(t,\mathbf x)=\int_{\mathcal V}\dd[3]x'\,\frac{\rho(t,\mathbf x)}{\norm{\mathbf x-\mathbf x'}}.
\end{align}
This potential admit a convenient rewriting under the form of a multipole decomposition above an arbitrary point $\mathbf x_0\in\mathcal V$. For any $\mathbf x\notin \mathcal V$, one can write
\begin{align}
    U(t,\mathbf x)=\sum_{l=0}^{+\infty}\frac{\qty(-)^l}{l!}I^{i_1\ldots i_l}(t,\mathbf x_0)\partial_{i_1}\ldots\partial_{i_l}\norm{\mathbf x-\mathbf x_0}^{-1}\label{decomposition_newton_potential}
\end{align}
where we have defined the multipole moments
\begin{align}
    I^{i_1\ldots i_l}(t,\mathbf x_0)\triangleq \int_{\mathcal V}\dd[3]x\,\qty(x-x_0)^{i_1}\ldots\qty(x-x_0)^{i_l}\rho(t,\mathbf x).
\end{align}
The proof of Eq. \eqref{decomposition_newton_potential} is easily carried out by performing a Taylor expansion of $\norm{\mathbf x-\mathbf x'}$ with respect to $\mathbf x'$ above some point $\mathbf x_0\in\mathcal V$.

The gravitational skeletonization consists here in replacing the \textit{smooth} mass density distribution $\rho(t,\mathbf x)$ (which is supported on a finite-size region of space, $\text{supp}\,\rho\subseteq\mathcal V$) by a \textit{singular} mass density distribution -- say $\rho_\text{skel}$ -- which is supported on a single point of space, $\text{supp}\,\rho_\text{skel}= \qty{\mathbf x_0}\in\mathcal V$. The key result allowing such a skeletonization to be performed can be stated as follows:
\begin{proposition}{}{}
 Let $\mathbf x_0$ be a point of $\mathcal V$. For any point $\mathbf x\notin\mathcal V$ outside of the object, the distributional mass density
\begin{align}
    \rho_\text{skel}(t,\mathbf x)\triangleq\sum_{l=0}^{+\infty}\frac{(-)^l}{l!}I^{i_1\ldots i_l}(t,\mathbf x_0)\partial_{i_1}\ldots\partial_{i_l}\delta^{(3)}(\mathbf x-\mathbf x_0)\label{singular_mass_density}
\end{align}
generates the same potential $U(t,\mathbf x)$ as the smooth mass density $\rho(t,\mathbf x)$.
\end{proposition}
\begin{proof}
It is enough to recall ourselves that the identity 
\begin{align}
\Delta\norm{\mathbf x-\mathbf x_0}^{-1}=-4\pi\delta^{(3)}\qty(\mathbf x-\mathbf x_0).
\end{align}
holds in the sense of distributions \cite{schwartz1966theorie}.
\end{proof}
Of course, the explicit expression of $\rho_\text{skel}$ and all the related equations have to be understood in the sense of distributions (see \textit{e.g.} \cite{schwartz1966theorie} for a clear reminder of the meaning of this assertion).

In others words, when observed from outside, any localized gravitating object can be replaced by a particle located at a single point of spacetime and possessing an infinite tower of multipole moments. This replacement holds in the sense that the gravitational potentials generated by these two systems are identical as long as we remain outside of the object.

\section{General Relativist Skeletons}
\label{sec:GR_skeleton}

We now consider an extended body within the framework of General Relativity. This object is assumed to be described by a smooth stress-energy tensor supported on some worldtube $\mathcal T$. In the same spirit, we replace its smooth stress-energy tensor $T^{\mu\nu}(x)$ by a distributional stress-energy tensor $T^{\mu\nu}_\text{skel}(x)$ supported on a single timelike worldline $\gamma\subset\mathcal T$. By analogy with Eq. \eqref{singular_mass_density}, we assume this stress-energy tensor to take the form
\boxedeqn{
    T^{\mu\nu}_\text{skel}(x)=\sum_{l=0}^{+\infty}\frac{1}{l!}\int_\gamma\dd\lambda\, I^{\mu\nu \alpha_1\ldots\alpha_l}(z)\mathcal D^{(l)}_{\alpha_1\ldots\alpha_l}\delta_4(x,z).\label{GR_skeleton}
}{Gravitational skeleton for the stress-energy tensor}
Here, $\mathcal D^{(k)}_{\alpha_1\ldots\alpha_l}$ is some differential operator which contains at most $l$ derivatives. $z^\mu(\tau)$ are coordinates parametrizing the worldline $\gamma$ with respect to an affine time parameter $\lambda$. We denote the tangent vector to the worldline $v^\mu=\dv{z^\mu}{\lambda}$. For $l=0$, we use the conventions $\mathcal D^{(0)}=\text{Id}$ and $I^{\mu\nu\alpha_1\ldots\alpha_l}=I^{\mu\nu}$. At this level, the multipoles $I^{\mu\nu\alpha_1\ldots\alpha_l}$ are still arbitrary functions.

\subsubsection{Perturbative treatment}

One can show that it makes sense to treat the expansion \eqref{GR_skeleton} perturbatively, and consequently to truncate it at any desired order. By analogy with the non-relativistic case, we do expect the multipole to scale as $I^{\mu\nu \alpha_1\ldots\alpha_l}\sim \mu \ell^l$, with $\mu$ and $\ell$ being respectively the typical mass and size of the extended body under consideration. Moreover, we expect the multiple covariant derivative to scale as $\nabla_{\alpha_1\ldots\alpha_l}\sim r^{-l}$, with $r$ the typical curvature radius of the background metric. All in all, the $l^\text{th}$ term of the expansion \eqref{GR_skeleton} scales as $\mu\qty(\frac{\ell}{r})^l$. Then, if the object is assumed to be \textit{compact} in the sense mentioned in the introduction, we have $r\gg\ell$ and consequently $\frac{\ell}{r}\ll 1$. The expansion \eqref{GR_skeleton} can thus be truncated in a perturbative sense. A truncation at $l=0$ will correspond to the monopole approximation, $l=1$ to the pole-dipole (or simply dipole) one, $l=2$ to the quadrupole, etc.

\subsubsection{Dixon and Ellis representations}

To go further, we shall choose an explicit form for the operator $\mathcal D^{l}_{\alpha_1\ldots\alpha_l}$. Two equivalent choices have been used in the literature \cite{Gratus:2020cok}. The most common is the \defining{Dixon representation} \cite{dixon1973,dixon1974,dixon1979}
\begin{align}
    \mathcal D^{(k)}_{\alpha_1\ldots\alpha_k}&\triangleq\nabla_{\alpha_1}\ldots\nabla_{\alpha_k}.
\end{align}
In this text, we will instead make another choice, the \defining{Ellis representation}, defined by \cite{Ellis:1975rp}
\begin{align}
    \mathcal D^{(l)}_{\alpha_1\ldots\alpha_l}&=\left\lbrace\
    \begin{array}{ll}
      0   & \text{ if }l<N\\
      \partial_{\alpha_1}\ldots\partial_{\alpha_N} & \text{ if }l=l_\text{max} 
    \end{array}\right.,
\end{align}
with $l_\text{max}$ the order in $l$ at which the multipole expansion Eq. \eqref{GR_skeleton} is chosen to be truncated. 

Both of these representations have advantages and drawbacks, which are nicely reviewed in \cite{Gratus:2020cok}. For our purposes, it will be more convenient to work in the Ellis representation, since the computations involved at dipole order turn out to be less technical, and thus more suited for an introductory exposition.

\subsubsection{Reduction of the stress-energy tensor}

The above skeletonization amounts to replace the smooth object by a collection of multipole moments supported on a single worldline contained in the object worldtube. Actually proving the validity of the decomposition \eqref{GR_skeleton} would require to show that there exists some expressions of the multipole moments $I^{\mu\nu\alpha_1\ldots\alpha_l}$ such that both the smooth and the distributional stress-energy tensors generate the same spacetime curvature though the Einstein field equations $G_{\mu\nu}=8\pi T_{\mu\nu}$. This task being extremely involved, we will not attempt to tackle it in the present text, but rather consider \eqref{GR_skeleton} for granted. The interested reader would fruitfully refer to \cite{dixon1979} for a more formal exposition of the subject.

In this text, we will instead discuss the so-called \defining{reduction of the stress-energy tensor}: as it is always the case in GR, our stress-energy tensor must be conserved and consequently obeys \cite{Wald:1984rg,Misner1973}
\begin{align}
    \nabla_\mu T^{\mu\nu}=0.\label{T_conservation}
\end{align}

As we will show, this conservation equation highly constrains the form of the stress-energy tensor. At quadrupole order, one can show that it implies that there must exist a vector $p^\mu$, and antisymmetric tensor $S^{\mu\nu}$ and a rank four tensor $J^{\mu\nu\rho\sigma}$ exhibiting the same symmetries than the Riemann tensor such that \cite{Steinhoff_2010,Gratus:2020cok}
\begin{align}\label{stress_skeleton}
    T^{\mu\nu}(x)&=\int_\gamma\dd\lambda\,\qty[v^{(\mu} p^{\nu)}+\frac{1}{3}R\tud{(\mu}{\alpha\beta\gamma}J^{\nu)\alpha\beta\gamma}]\delta_4(x,z)+\nabla_\alpha\int_\gamma\dd\lambda\, v^{(\mu}S^{\nu)\alpha}\delta_4(x,z)\nonumber\\
    &\quad-\frac{2}{3}\nabla_\alpha\nabla_\beta\int_\gamma \dd\lambda\, J^{\alpha(\mu\nu)\beta}\delta_4(x,z).
\end{align}
Moreover, during the reduction process, we find additional constraints taking the form of differential equations for the quantities $p^\mu$, $S^{\mu\nu}$ and $J^{\mu\nu\rho\sigma}$ that turn out to be precisely the MPD equations. Notice that Eq. \eqref{stress_skeleton} agrees exactly with the form of the stress-energy tensor which has been found through the Lagrangian formulation, given in Eq. \eqref{T_lagrangian}.

\section{Ellis skeleton in adapted coordinates: to dipole order}
\label{sec:ellis_skeleton}

\textit{This section mainly follows the exposition of \cite{Gratus:2020cok}.}

Performing the explicit reduction of the stress-energy tensor at quadrupole order is computationally quite involved, and can be found in \cite{Steinhoff_2010,Gratus:2020cok}. As a proof of principle, we will instead perform the reduction at pole-dipole level and show that it leads to the MPD equations at the corresponding order. In the following, we will work with Ellis representation of multipoles. Explicitly, Ellis representation truncated at order $N$ takes the form
\begin{align}
    T^{\mu\nu}_\text{skel}(x)&=\frac{1}{N!}\int_\gamma\dd\lambda\, I^{\mu\nu\alpha_1\ldots\alpha_N}(z)\partial_{\alpha_1}\ldots\partial_{\alpha_N}\delta_{4}\qty(x,z)\label{ellis_decomp}
\end{align}
with $z^\mu=z^\mu(\lambda)$ the coordinates of the worldline. 

Since $T_\text{skel}^{\mu\nu}$ is a distributional stress-energy tensor, Eq. \eqref{ellis_decomp} shall be formally understood in the sense of distributions: for any symmetric test function $\phi_{\mu\nu}$, the quantity
\begin{align}
    \int_{\mathcal T}\dd[4]x\sqrt{-g} \,T_\text{skel}^{\mu\nu}(x)\phi_{\mu\nu}(x)
\end{align}
is a real number. Integrating distributional quantities against test functions will be the main tool we will use to show that the conservation of the stress-energy tensor leads to the MPD equations.

Notice that since $T^{\mu\nu}$ is symmetric and since the partial derivatives commute, the moment $I^{\mu\nu\alpha_1\ldots\alpha_N}$ must obey the algebraic symmetries
\begin{align}
    I^{\mu\nu\alpha_1\ldots\alpha_k}=I^{(\mu\nu)\alpha_1\ldots\alpha_k}=I^{\mu\nu(\alpha_1\ldots\alpha_k)}.
\end{align}
Finally, let us comment on two drawbacks of Ellis representation: (i) from their very definitions, the moments $I^{\mu\nu\alpha_1\ldots\alpha_N}$ are not tensors, since they are contracted on their last $N$ indices with partial derivatives, and since the total stress-energy tensor must be itself a tensorial object. Moreover, (ii) for a given order of truncation $N$, the full multipolar structure of the test object is represented by a single moment $I^{\mu\nu\alpha_1\ldots\alpha_N}$. There is thus no \textit{a priori} decomposition of this moment between a hierarchy of moments (monopole, dipole\ldots). As we will see in the following, such a split can be obtained by introducing a specific system of coordinates, the so-called \defining{adapted coordinates}.

\subsubsection{Ellis skeleton in adapted coordinates}

We now turn to coordinates adapted to the worldline, $x^\mu\to X^\mu=(X^0,\mathbf X)$, with $\mathbf X\triangleq\qty(X^1,X^2,X^3)$. They are required to satisfy
\begin{align}
    X^0\eval_\gamma&=\lambda,\quad X^i\eval_\gamma=0\qquad\Rightarrow\qquad v^0\eval_\gamma=\dv{X^0}{\lambda}\eval_\gamma=1,\quad v^i\eval_\gamma=\dv{X^i}{\lambda}\eval_\gamma=0.
\end{align}
An example of explicit construction of this type of coordinates are \textit{Fermi normal coordinates} \cite{Poisson:2011nh}. In what follows, we will systematically assume that such coordinates can be constructed over the worldtube $\mathcal T$ of the body. Using these coordinates, we can write
\begin{align}
    \begin{split}
    &\int_\mathcal{T}\dd[4]X\sqrt{-g}\, 
    T_\text{skel}^{\mu\nu}\phi_{\mu\nu}\\
    &=\frac{(-1)^N}{N!}\int_\gamma\dd\lambda\, I^{\mu\nu\alpha_1\ldots\alpha_N}\partial_{\alpha_1}\ldots\partial_{\alpha_N}\phi_{\mu\nu}\eval_{\gamma}\\
    &=\sum_{k=0}^N\frac{(-1)^N}{N!}\frac{N!}{k!(N-k)!}\int_\gamma\dd\lambda\, I^{\mu\nu i_1\ldots i_k0\ldots0}\partial_{i_1}\ldots\partial_{i_k}\partial^{N-k}_0\phi_{\mu\nu}\eval_\gamma\\
    &=\sum_{k=0}^N\frac{(-1)^k}{k!(N-k)!}\int_\gamma\dd\lambda\,\partial^{N-k}_0 I^{\mu\nu i_1\ldots i_k0\ldots0}\partial_{i_1}\ldots\partial_{i_k}\phi_{\mu\nu}\eval_\gamma.\label{adapted_test}
    \end{split}
\end{align}
It is therefore useful to define a new collection of moments 
\begin{align}
    \gamma_{(N)}^{i_1\ldots i_k}\triangleq\frac{1}{(N-k)!}\partial^{N-k}_0 I^{\mu\nu i_1\ldots i_k0\ldots0},\qquad k\leq N.\label{multipole_gamma}
\end{align}
They still satisfy the algebraic symmetries
\begin{align}
    \gamma_{(N)}^{\mu\nu i_1\ldots i_k}=\gamma_{(N)}^{(\mu\nu) i_1\ldots i_k}=\gamma_{(N)}^{\mu\nu (i_1\ldots i_k)}.
\end{align}

In terms of the new moments (and still in adapted coordinates), Ellis decomposition is equivalently given by
\boxedeqn{
    T^{\mu\nu}_\text{skel}=\frac{1}{\sqrt{-g}}\mathcal T^{\mu\nu},\qquad \mathcal T^{\mu\nu}\triangleq\sum_{k=0}^N\frac{1}{k!}\gamma_{(N)}^{\mu\nu i_1\ldots i_k}(\lambda)\partial_{i_1}\ldots\partial_{i_k}\delta^{(3)}\qty(\mathbf X).\label{ellis_adapted}
}{Ellis decomposition in adapted coordinates}
Notice that since $T^{\mu\nu}_\text{skel}$ is a tensor, $\mathcal T^{\mu\nu}$ is a tensor density of weight $-1$.
The proof of Eq. \eqref{ellis_adapted} consists into integrating this expression against an arbitrary, symmetric test function $\phi_{\mu\nu}$:
\begin{align}
    \begin{split}
    \int_\mathcal{T}\dd[4]X\sqrt{-g}\,T_\text{skel}^{\mu\nu}\phi_{\mu\nu}
    &=\sum_{k=0}^N\frac{1}{k!}\int_\gamma\dd\lambda\int\dd[3]X\,\gamma_{(N)}^{\mu\nu i_1\ldots i_k}\partial_{i_1}\ldots\partial_{i_k}\delta^{(3)}(\mathbf X)\phi_{\mu\nu}\\
    &=\sum_{k=0}^N\frac{(-1)^k}{k!}\int_\gamma\dd\lambda\,\gamma_{(N)}^{\mu\nu i_1\ldots i_k}\partial_{i_1}\ldots\partial_{i_k}\phi_{\mu\nu}.
    \end{split}
\end{align}
We therefore recover the expression obtained in Eq. \eqref{adapted_test} using the definition Eq. \eqref{multipole_gamma}.

Before turning to the derivation of the MPD equations using the conservation of the stress-energy tensor, let us make a couple of remarks. As we will check explicitly at the dipole level, the multipoles $\gamma$ are fully determined by the distribution of stress-energy, whereas the $I$s are not, thus leading to the appearance of a gauge freedom in their definition. This can be seen from Eq. \eqref{multipole_gamma}, since the construction of the $I$s from the $\gamma$s require to integrate with respect to time, thus leading to the appearance of arbitrary integration constants.

Moreover, as we will see when discussing the relation between distributional and smooth stress-energy tensors, the split of $I^{\mu\nu\alpha_1\ldots\alpha_N}$ in $k$ moments $\gamma^{\mu\nu i_1\ldots i_k}$ actually corresponds to a physical split between a monopole, a dipole, etc.
Notice that the price to pay for obtaining simple computations in Ellis representation is that we are forced to work in a specific coordinate system, the adapted coordinates. Moreover, the multipole decomposition tuned to this system is not explicitly covariant. Nevertheless, this framework allows to recover the results obtained with more involved approaches, \textit{e.g.} Dixon representation. Finally, notice that a coordinate-free approach to multipoles can be formulated, see \cite{Gratus:2020cok} for more details.

\subsubsection{Conservation equation for stress-energy tensor density}
In adapted coordinates, since the decomposition Eq. \eqref{ellis_adapted} is valid, it is more convenient to write the conservation of the stress-energy tensor Eq. \eqref{T_conservation} in terms of the tensor density $\mathcal T^{\mu\nu}$ introduced in Eq. \eqref{ellis_adapted}. Since $T^{\mu\nu}_\text{skel}$ is a tensor, the covariant derivative appearing in Eq. \eqref{T_conservation} can be expanded in terms of partial derivatives and Christoffel symbols. We get
\begin{align}
    \begin{split}
    \nabla_\mu T^{\mu\nu}_\text{skel}&=\partial_\mu T^{\mu\nu}_\text{skel}+2\Gamma^{(\mu}_{\nu\rho}T^{\nu)\rho}_\text{skel}\\
    &=\partial_\mu\qty(\frac{1}{\sqrt{-g}})\mathcal T^{\mu\nu}+\frac{1}{\sqrt{-g}}\qty(\partial_\mu \mathcal T^{\mu\nu }+2\Gamma^{(\mu}_{\nu\rho}\mathcal T^{\nu)\rho})\\
    &=\frac{1}{\sqrt{-g}}\qty(\partial_\mu \mathcal T^{\mu\nu}+\Gamma^{\nu}_{\mu\rho}\mathcal T^{\mu\rho}).
    \end{split}
\end{align}
The last equality uses the identity
\begin{align}
    \partial_\mu\qty(\sqrt{-g})=\sqrt{-g}\Gamma^\alpha_{\mu\alpha}.
\end{align}
At the end of the day, the conservation equation for the stress-energy tensor Eq. \eqref{T_conservation} is equivalent to the following equation for the stress-energy tensor density
\begin{align}
    \partial_\mu \mathcal T^{\mu\nu}+\Gamma^{\nu}_{\mu\rho}\mathcal T^{\mu\rho}=0.\label{density_conservation}
\end{align}

\subsubsection{MPD equations in adapted coordinates}
Before investigating the consequences of this conservation equation, it is useful to write down the form that the MPD equations take in adapted coordinates. Since we will derive only the pole-dipole equations in next section, we set the force and torque terms to zero in all the equations below.

Recalling that $v^\mu=\delta^\mu_0$ on the worldline, the evolution equation for the spin \eqref{MPD_spinning} takes the form
\begin{align}
    \nabla_0 S^{\mu\nu}=2p^{[\mu}v^{\nu]}.
\end{align}
Separating the spatial coordinates from the temporal one, it is equivalent to the two following equations
\begin{align}
    \nabla_0 S^{0i}&=-p^i,\qquad\nabla_0 S^{ij}=0.\label{spin_adapted}
\end{align}

Using the definition of the Riemann tensor, the evolution equation for the momentum \eqref{MPD_impulsion} takes the form
\begin{align}
    \begin{split}
    \nabla_0 p^\mu&=-\frac{1}{2}R\tud{\mu}{0\alpha\beta}S^{\alpha\beta}\\
    &=-\qty(\partial_\alpha\Gamma^{\mu}_{0\beta}+\Gamma^\mu_{\alpha\lambda}\Gamma^\lambda_{0\beta})S^{\alpha\beta}\\
    &=-\dot\Gamma^0_{0\mu}S^{0\mu}-\partial_i\Gamma^0_{0\mu}S^{i\mu}-\Gamma^\mu_{\alpha\lambda}\Gamma^{\lambda}_{0\beta}S^{\alpha\beta},\label{MPD_imp_adapted}
    \end{split}
\end{align}
where we use the notation $\dot{ }\triangleq\partial_0$.

\subsubsection{Recovering MPD equations from stress-energy conservation}

We will now truncate the multipole expansion to dipole order, and study the constraints enforced by the conservation equation \eqref{density_conservation}. At dipole order, Ellis decomposition takes the form
\begin{align}
    \mathcal T^{\mu\nu}=\gamma^{\mu\nu}\delta^{(3)}(\mathbf X)+\gamma^{\mu\nu i}\partial_i\delta^{(3)}(\mathbf X)\label{dipole_decomp}
\end{align}
To lighten the notations, we drop the subscript ``$(N)$'' from the multipoles $\gamma$ in the continuation of this section. For such a choice of stress-energy density, the conservation equation \eqref{density_conservation} is a distributional identity. It will be satisfied provided that
\begin{align}
    \int\dd[3]X\,\qty(\partial_\mu \mathcal T^{\mu\nu}+\Gamma^{\nu}_{\mu\rho}\mathcal T^{\mu\rho})\phi_\nu=0
\end{align}
for any arbitrary test function $\phi_\nu$. Inserting the decomposition Eq. \eqref{dipole_decomp} in this equation and integrating by parts yields
\begin{align}
    \begin{split}
    &\qty(\dot\gamma^{0\nu}+\Gamma^\nu_{\mu\rho}\gamma^{\mu\rho}-\partial_i\Gamma^\nu_{\mu\rho}\gamma^{\mu\rho i})\phi_\nu\\
    &\quad-\qty(\dot\gamma^{0\nu i}+\gamma^{i\nu}+\Gamma^\nu_{\mu\rho}\gamma^{\mu\rho i})\partial_i\phi_\nu+\gamma^{j\nu i}\partial_i\partial_j\phi_\nu=0
    \end{split}
\end{align}
Because this identity shall be valid for any choice of test function $\phi_\nu$, it is equivalent to the set of three equations
\begin{subequations}
\begin{align}
    &\gamma^{\nu (ij)}=0,\label{eq1}\\
    &\dot\gamma^{0\nu i}+\gamma^{i\nu}+\Gamma^\nu_{\mu\rho}\gamma^{\mu\rho i}=0,\label{eq2}\\
    &\dot\gamma^{0\nu}+\Gamma^\nu_{\mu\rho}\gamma^{\mu\rho}-\partial_i\Gamma^\nu_{\mu\rho}\gamma^{\mu\rho i}=0.\label{eq3}
\end{align}
\end{subequations}
We will now prove that this set of constraints is consistent with setting
\begin{subequations}\label{gamma_constrained}
\begin{align}
    \gamma^{\mu\nu}&=p^{(\mu}v^{\nu)}+v^\rho\Gamma^{(\mu}_{\rho\sigma}S^{\nu)\sigma} +\partial_0\qty(v^{(\mu}S^{\nu)0} ),\\
    \gamma^{\mu\nu i}&=v^{(\mu}S^{\nu) i}
\end{align}
\end{subequations}
in adapted coordinates, where $p^\mu$ is a vector and $S^{\mu\nu}$ an antisymmetric tensor. Notice that these relations can be inverted as
\begin{subequations}\label{inverted_moments}
\begin{align}
    S^{0i}&=\gamma^{00i},\\
    S^{ij}&=2\gamma^{0ij},\\
    p^\mu&=\gamma^{\mu0}-\Gamma^\mu_{0i}\gamma^{00i}.
\end{align}
\end{subequations}
The two first identities are straightforward to prove, whereas the latter is a little bit more involved. First, remark that one can write
\begin{align}
    \dot S^{\mu0}=p^\mu-p^0\delta^\mu_0+\Gamma^\mu_{0\rho} S^{0\rho}+\Gamma^0_{0\rho}S^{\rho\mu}.
\end{align}
We then obtain
\begin{align}
    \begin{split}
        \gamma^{\mu0}&=\frac{1}{2}\qty(p^\mu+p^0\delta^\mu_0+\Gamma^\mu_{0\rho}S^{0\rho}+\Gamma^0_{0\rho}S^{\mu\rho}+\dot S^{\mu0})\\
        &=p^\mu+\Gamma^\mu_{0\rho}S^{0\rho}\\
        &=p^\mu+\Gamma^\mu_{0i}\gamma^{00i},
    \end{split}
\end{align}
which completes the proof.

We will now plug the decomposition Eq. \eqref{gamma_constrained} into the constraint equations. Eq. \eqref{eq1} is automatically satisfied, since $S^{\mu\nu}$ is an antisymmetric tensor. Eq. \eqref{eq2} reduces to
\begin{align}
    \delta^{(\nu}_0\qty[\dot S^{0)i}+\dot S^{i)0}+p^{i)}]+\Gamma^{(i}_{0\rho}S^{0)\rho}+\Gamma^\nu_{0\rho}S^{\rho i}=0.
\end{align}
Setting $\nu=j$, we are left with
\begin{align}
    \dot S^{ij}+\Gamma^i_{\mu 0}S^{\mu j}+\Gamma^j_{\mu 0}S^{i\mu}=0\qquad\Leftrightarrow\qquad\nabla_0S^{ij}=0.
\end{align}
For $\nu=0$, we get
\begin{align}
    \dot S^{0i}+\Gamma^0_{\mu 0}S^{\mu i}+\Gamma^i_{\mu 0}S^{0\mu}+p^i=0\qquad\Leftrightarrow\qquad\nabla_0 S^{0i}=-p^i.\label{MPD_adapted_spin_0i}
\end{align}
These two equations are precisely the spin evolution equation in adapted coordinates, given by Eq. \eqref{spin_adapted}.

The last constraint Eq. \eqref{eq3} is the most involved to deal with. It reads
\begin{align}
    \begin{split}
    &\dot p^{(0}\delta^{\nu)}_0+\partial_0\qty(\Gamma^{(0}_{0\sigma}S^{\nu)\sigma})+\partial_0\qty(\delta^{(0}_0\dot S^{\nu)0})+\Gamma^\nu_{\mu0}\qty(p^\mu+\dot S^{\mu0})\\
    &\quad+\Gamma^\nu_{\mu\alpha}\Gamma^\mu_{0\beta}S^{\alpha\beta}-\partial_i\Gamma^\nu_{0\rho} S^{\rho i}=0.
    \end{split}
\end{align}
For $\nu=0$, direct algebra leads to
\begin{align}
    \nabla_0 p^0&=-\dot \Gamma^0_{0\mu}S^{0\mu}-\partial_i\Gamma^0_{0\mu}S^{i\mu}-\Gamma^0_{\mu\alpha}\Gamma^\mu_{0\beta}S^{\alpha\beta},
\end{align}
which precisely reproduce the component $\mu=0$ of Eq. \eqref{MPD_imp_adapted}. 

For $\nu=i$, one shall use Eq. \eqref{MPD_adapted_spin_0i} to replace the term $\frac{1}{2}\partial_0\dot S^{i0}$. After some algebra, we find
\begin{align}
    \nabla_0 p^i&=-\dot \Gamma^i_{0\mu}S^{0\mu}-\partial_j\Gamma^i_{0\mu}S^{j\mu}-\Gamma^j_{\mu\alpha}\Gamma^\mu_{0\beta}S^{\alpha\beta},
\end{align}
which again agrees with Eq. \eqref{MPD_imp_adapted}.

\subsubsection{Comparison with stress-energy tensor from Lagrangian formulation}
Let us summarize our findings. We have proven that, in adapted coordinates, the stress-energy conservation equation \eqref{T_conservation} enforces the pole-dipole stress energy tensor in Ellis formulation to be given by the decomposition Eq. \eqref{gamma_constrained}, where $p^\mu$ and $S^{\mu\nu}$ satisfy the pole-dipole MPD equations, and can thus be identified to the linear momentum and the spin dipole of the body.

In last chapter, we also found an explicit form for the stress-energy tensor in terms of the momentum and the spin, given by Eq. \eqref{T_lagrangian}. It is however not obvious that the two formulation agree one with another. To prove this, one can integrate the Lagrangian stress-energy tensor Eq. \eqref{T_lagrangian} truncated at dipole order against a test function, still working in adapted coordinates:
\begin{align}
    \begin{split}
    &\int\dd[4]X\sqrt{-g}\int_\gamma\dd\lambda\,\qty[p^{(\mu}v^{\nu)}\delta_4(X,Z)-\nabla_\lambda\qty(S^{\lambda(\mu}v^{\nu)}\delta_4(X,Z))]\phi_{\mu\nu}\\
    &=\int_\gamma\dd\lambda\, p^{(\mu}v^{\nu)}\phi_{\mu\nu}+\int_\gamma\dd\lambda\,S^{\lambda(\mu}v^{\nu)}\nabla_\lambda\phi_{\mu\nu}\\
    &=\int_\gamma\dd\lambda\,\qty[p^{(\mu}v^{\nu)}-S^{\lambda(\mu}\Gamma^{\nu)}_{\lambda\rho}v^\rho-\partial_0\qty(S^{0(\mu}v^{\nu)})]\phi_{\mu\nu}+\int_\gamma\dd\lambda\, S^{i(\mu}v^{\nu)}\partial_i\phi_{\mu\nu}\\
    &=\int_\gamma\dd\lambda\,\gamma^{\mu\nu}\phi_{\mu\nu}-\int_\gamma\dd\lambda\,\gamma^{\mu\nu i}\partial_i\phi_{\mu\nu},
    \end{split}
\end{align}
where all the quantities are understood to be evaluated on the worldline when integration over spacetime has been dropped. The last line precisely agrees with Ellis decomposition in adapted coordinates given in Eq. \eqref{ellis_adapted}, which proves the equivalence of the two frameworks.

\subsubsection{Link with smooth stress-energy distribution}

A last nice feature of Ellis representation in adapted coordinates is that it allows to easily relate the distributional stress-energy tensor $T^{\mu\nu}_\text{skel}$ to the smooth, physical stress-energy tensor $T^{\mu\nu}$ it represents. Since we are interested in compact objects, we can always assume that the stress-energy tensor has compact support on the spatial slices defined in adapted coordinates.

We can then define a one parameter family of regular ``squeezed stress-energy tensors'', given by
\begin{align}
    T_\epsilon^{\mu\nu}(\lambda,\mathbf X)\triangleq\frac{1}{\epsilon^3}T^{\mu\nu}\qty(\lambda,\frac{\mathbf X}{\epsilon}).
\end{align}
For $\abs{\epsilon}\ll 1$, one has
\begin{align}
    T^{\mu\nu}_\epsilon(\lambda,\mathbf X)=\tilde\gamma^{\mu\nu}(\lambda)\delta^{(3)}\qty(\mathbf X)+\epsilon\tilde\gamma^{\mu\nu i}(\lambda)\partial_i\delta^{(3)}\qty(\mathbf X)+\mathcal{O}\qty(\epsilon^2),
\end{align}
with\begin{align}
    \tilde\gamma^{\mu\nu}(\lambda)=\int\dd[3]X\sqrt{-g}\, T^{\mu\nu}(\lambda,\mathbf X),\qquad \tilde\gamma^{\mu\nu i}(\lambda)=\int\dd[3]X\sqrt{-g}\, X^i T^{\mu\nu}(\lambda,\mathbf X)
\end{align}
The proof proceeds by a simple change of variables $W^i=X^i/\epsilon$ and a Taylor expansion around $\mathbf X=\mathbf 0$:
\begin{align}
    \begin{split}
    &\int\dd[4]X\sqrt{-g}\,T^{\mu\nu}_\epsilon(\lambda,\mathbf X)\phi_{\mu\nu}(\lambda,\mathbf X)\\
    &=\int_\gamma\dd\lambda\int\dd[3]W\sqrt{-g}\,T^{\mu\nu}(\lambda,\mathbf W)\phi_{\mu\nu}(\lambda,\epsilon\mathbf W)\\
    &=\int_\gamma\dd\lambda\int\dd[3]W\sqrt{-g}\,T^{\mu\nu}(\lambda,\mathbf W)\phi_{\mu\nu}(\lambda,\mathbf 0)\\
    &\quad+\epsilon\int_\gamma\dd\lambda\int\dd[3]W\sqrt{-g}\,T^{\mu\nu}(\lambda,\mathbf W)W^i\partial_i\phi_{\mu\nu}(\lambda,\mathbf 0)+\mathcal O (\epsilon^2),
    \end{split}
\end{align}
which naturally leads to the result announced. This result is easily extended to higher orders. The moments $\tilde\gamma$ have now the real physical interpretation of a mono-pole, a dipole, etc. 

Notice that if we consider stress-energy distributions representing compact objects, there is no need to introduce an expansion parameter, since the moment $\tilde\gamma^{\mu\nu i_1\ldots i_k}$ will naturally scales as $\mu\qty(\ell/r)^k$, with $\ell/r\ll 1$, as discussed in the previous section. We can thus formally set $\epsilon=1$ in the development above, the consistency of the truncation of the Taylor series being granted by the existence of this new small parameter. In this case, the squeezed stress-energy tensor is equal to the physical one, and 
\begin{align}
    \tilde\gamma^{\mu\nu i_1\ldots i_k}=\gamma_{(N)}^{\mu\nu i_1\ldots i_k}.
\end{align}
For compact objects, there is thus a natural relationship between the distributional moments and the smooth, physical stress-energy tensor.

More concretely, we can provide expressions for the ``physical'' moments using Eq. \eqref{inverted_moments}. For the spin tensor, one has
\begin{align}
    \begin{split}
    S^{0i}&=\gamma^{00i}\\
    &=\int\dd[3]X\sqrt{-g}\, X^iT^{00}\\
    S^{ij}&=2\gamma^{0ij}\\
    &=2\int\dd[3]X\sqrt{-g} X^{i}T^{j0}\\
    &=\int\dd[3]X\sqrt{-g}\qty(X^i T^{j0}-X^j T^{i0}).
    \end{split}
\end{align}
Gathering these two results, we find back the formula mentioned in the introduction,
\boxedeqn{
    S^{\mu\nu}=\int_{X^0=\text{cst}}\dd[3]X\sqrt{-g}\qty(\delta X^\mu T^{\nu0}-\delta X^\nu T^{\mu 0}),\label{multipole_spin}
}{Spin tensor from the stress-energy tensor}
since, in adapted coordinates, $\delta X^0=0$ and $\delta X^i=X^i$. The same game can be played for the momentum $p^\mu$.
We obtain
\boxedeqn{
    p^\mu&=\int\dd[3]X\sqrt{-g}\, T^{\mu 0}-\Gamma^\mu_{0\nu}S^{0\nu}.\label{multipole_momentum}
}{Linear momentum from the stress-energy tensor}
In flat spacetime, this equation reduces to the standard result
\begin{align}
    p^\mu=\int_{X^0=\text{cst}}\dd[3]X\sqrt{-g}T^{\mu 0}.
\end{align}

We end here this discussion. The main point to highlight is that Ellis skeletonization in adapted coordinates has allowed to provide a more concrete viewpoint on physical significance of the linear momentum and the spin tensor, which can now be understood as ``physical'' multipole moments of the smooth stress-energy tensor.

\chapter{Spin supplementary conditions}\label{chap:SSC}


Let us look back at what has been achieved so far. In Chapter \ref{chap:EOM}, we have derived the explicit form of the equations of motion for extended test bodies through the Lagrangian formulation, whereas Chapter \ref{chap:skeleton} has allowed us to gain more insights about the physical meaning of the dynamical variables $p^\mu$, $S^{\mu\nu}$\ldots We have thus now in our possession \textit{ten} differential equations for describing the motion of test bodies in curved spacetime, namely the MPD equations \eqref{MPD_spinning} and \eqref{MPD_impulsion}. However, the motion of a test body is described by \textit{fourteen} dynamical quantities $v^\mu$, $p^\mu$ and $S^{\mu\nu}=S^{[\mu\nu]}$. The system of equations is consequently not closed, and we are left with four extra dynamical quantities. One extra condition can be enforced by setting the time parameter to the proper time of the worldline, $\lambda=\tau$, yielding $v_\mu v^\mu=-1$. Nevertheless, we are still left with three missing constraints.

In both description studied so far, these missing constraints carry a clear physical interpretation:
\begin{itemize}
    \item From the Lagrangian viewpoint, they arise from the fact that the orientation of the body is entirely described by the means of the rotational degrees of freedom of the Lorentz matrix appearing in Eq. \eqref{tfo_tetrads}. As we shall see in Section \ref{sec:ssc:lagrange}, the boost degrees of freedom turn out to be redundant with the linear momentum $p^\mu$, which is expected since this latter quantity already describes the evolution of the position of the ``center'' of the body along the worldline. Therefore, allowing for degrees of freedom boosting the center position does not add any new physical degree of freedom to the system. These boost degrees of freedom are gauge degrees of freedom, and they can be fixed without affecting the physical state of the system.
    \item From the skeletonization perspective, we have not yet specified along \textit{which} worldline (belonging to the body worldtube) the multipole moments were defined. This fact accounts for the degrees of freedom that remain to be specified. A natural choice is to set the worldline to describe the evolution of the position of the body's center of mass (COM) at each instant of the time evolution. However, as we will review in Section \ref{sec:ssc:sr}, the relativistic notion of center of mass is observer-dependent, and we shall therefore specify with respect to \textit{which} observer our choice has been made.
\end{itemize}

In both cases, we have three unnecessary degrees of freedom, which can be fixed by enforcing an algebraic relation of the form
\begin{align}
    \mathcal V_\mu S^{\mu\nu}=0,\label{covariant_ssc}
\end{align}
where $\mathcal V^\mu$ is some normalized timelike vector. Such a constraint is known as a covariant \defining{spin supplementary condition} (SSC). This chapter will be devoted to the study of such constraints, and to the understanding of their emergence from the physical prescriptions discussed above. Remark that Eq. \eqref{covariant_ssc} actually contains only three independent constraints, since $\mathcal V_\mu\mathcal V_\nu S^{\mu\nu}=0$ follows automatically from the antisymmetry of the spin tensor.

This chapter is structured as follows: Section \ref{sec:ssc:sr} will review the problem of center of mass in Special Relativity. Section \ref{sec:ssc:skeleton} and \ref{sec:ssc:lagrange} will respectively discuss spin supplementary conditions from the skeleton and the Lagrangian perspectives. In Section \ref{sec:ssc:ortho}, we will relate the presence of SSCs to an orthogonal split of the spin tensor with respect to the timelike direction $\mathcal V^\mu$. In particular, this will enable us to show that, when a SSC has been enforced, one can always express the content of the spin tensor in terms of a single spin vector. Section \ref{sec:ssc:various} will discuss the main SSCs that have been studied in the literature, that is the various choices of $\mathcal V^\mu$ that can be made.

\section{Invitation: center of mass in special relativity}\label{sec:ssc:sr}

\begin{figure}
    \centering
    \includegraphics[width=\textwidth]{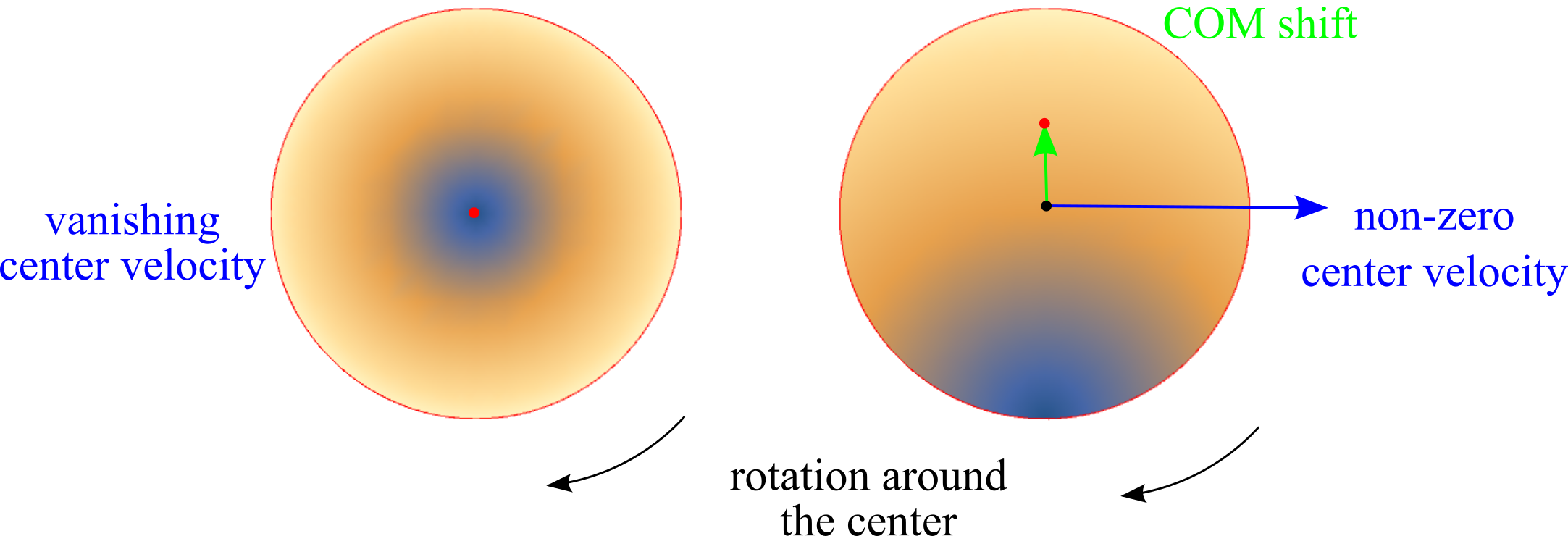}
    \caption{The same special relativist's homogeneous body spinning around its geometrical center seen by two different observers: on the left side, the observer is at rest with respect to the body's rotation axis. The (relativist) mass density of the body is then symmetric with respect to the axis, leading to a center of mass which coincides with its geometrical center. On the right side, the same body seen by a boosted observer (blue arrow). With respect to it, the velocities of the matter lumps of the upper hemisphere are higher than the ones of the lower hemisphere. The relativist mass of the upper hemisphere is thus greater than the one of the lower, leading to a shift of the center of mass (green arrow). The background color depicts the velocity distribution inside the body.}
    \label{fig:center_SR}
\end{figure}

The notion of center of mass becomes an observer dependent notion in relativity. This will be crucial for understanding the status of SSCs from the perspective of gravitational skeletonization (fixation of the worldline). We will illustrate this fact in the simpler framework of Special Relativity (SR). Consider a spinning body in SR described by a center position $x^\mu$ and an intrinsic (spin) angular momentum $S^{\mu\nu}$. Its linear momentum is denoted $p^\mu$. The total angular momentum of the system is \cite{Hanson:1974qy,Steinhoff:2011sya}
\begin{align}
    J^{\mu\nu}=x^\mu p^\nu- x^\nu p^\mu+S^{\mu\nu}.
\end{align}

Because of the Poincaré invariance of the system, $J^{\mu\nu}$ and $p^\mu$ should be conserved with respect to Poincaré transformations. However, the choice of the center $x^i$ is \textit{a priori} arbitrary. Therefore, for a change of center position
\begin{align}
    x^i\to x^i+\delta x^i,
\end{align}
one shall have the transformation rules
\begin{align}
    S^{ij}\to S^{ij}+\delta x^ip^j-\delta x^j p^i,\qquad S^{0i}\to S^{0i}+\delta x^i E\label{tfo_SR}
\end{align}
with $E\triangleq p^0$ the energy of the body. From the transformation rule Eq. \eqref{tfo_SR}, we can make two observations: (i) the components $S^{0i}$ of the spin tensor take the interpretation of the mass dipole of the object, relative to the center $x^i$. (ii) Unlike in classical mechanics, the \defining{center of mass} (that is, the choice of center for which the mass dipole vanishes) becomes an observer dependent notion, as graphically depicted in Figure \ref{fig:center_SR}. However, because of the transformation rule Eq. \eqref{tfo_SR}, one can always choose the center $x^i$ such that the mass dipole vanishes for \textit{one} specific timelike observer of velocity $\mathcal V^\mu$. These observer-dependent COMs are sometimes referred to as \defining{centroids} in the literature. In its proper frame, the four-velocity of an observer becomes
\begin{align}
    \mathcal V^{\hat{\mu}}=\qty(1,\mathbf 0).
\end{align} Therefore, the condition $S^{\hat 0\hat i}=0$ can be covariantly written 
\begin{align}
    \mathcal V_\mu S^{\mu\nu}=0,
\end{align}
which correspond to the generic canonical form of the generic covariant SSC Eq. \eqref{covariant_ssc}.

\section{SSCs in the skeleton formulation: fixation of the worldline}\label{sec:ssc:skeleton}

Having gained some intuition from SR, let us now turn back to our original problem. In the skeleton formulation of test bodies, we want to fix the representative worldline of the body. The most physically meaningful way of proceeding is to choose the worldline tight to a given timelike centroid. As we will see, this still corresponds to choose $S^{\hat 0 \hat \mu}$ to be vanishing in the proper frame of the timelike observer of four-velocity $\mathcal V^\mu$ characterizing the centroid. Actually, from the multipole decomposition of the spin tensor Eq. \eqref{multipole_spin}, we can write
\begin{align}
    \begin{split}
    S^{0\mu}&=-\int_{x^0=\text{cst}}\dd[3]x\sqrt{-g}\,\qty(x^\mu-x^\mu_\text{c}) T^{00}\\
    &=x^\mu_\text{c}\int_{x^0=\text{cst}}\dd[3]x\sqrt{-g}\, T^{00}-\int_{x^0=\text{cst}}\dd[3]x\sqrt{-g}\, x^\mu T^{00}\\
    &=x^\mu_\text{c}\qty(p^0+\Gamma^0_{0\nu}S^{0\nu})-\int_{x^0=\text{cst}}\dd[3]x\sqrt{-g}\,x^\mu T^{00},
    \end{split}
\end{align}
where the last equality follows from Eq. \eqref{multipole_momentum}. Now, asking that $S^{\hat 0\hat \mu}=0$ in the proper frame of a given observer amounts to set
\begin{align}
    x^{\hat \mu}_\text{c}=\frac{1}{p^{\hat 0}}\int_{x^{\hat 0}=\text{cst}}\dd[3]\hat x\sqrt{-g}\, x^{\hat \mu}T^{\hat 0 \hat 0}
\end{align}
in that frame, \textit{i.e.} to set the worldline $x_\text{c}^{\hat \mu}$ to be the center of mass of the body with respect to that observer. For the very same reason as in SR, the condition $S^{\hat 0\hat \mu}=0$ can be written covariantly under the form Eq. \eqref{covariant_ssc}.

We have therefore reached the following conclusion: \textit{enforcing a given covariant SSC Eq. \eqref{covariant_ssc} amounts to set the representative worldline of the body to be its center of mass with respect to a timelike observer of four-velocity $\mathcal V^\mu$.}

\section{SSCs in the Lagrangian formulation: spin gauge symmetry}\label{sec:ssc:lagrange}
We now turn to the analysis of spin supplementary conditions from the perspective of Lagrangian formulation. This discussion will be of prime relevance for turning to the Hamiltonian formulation, as will be undertaken in Part \ref{part:hamilton} of this thesis. As already discussed in Chapter \ref{chap:hamilton_kerr} for the geodesic case, turning to the Hamiltonian description first requires to check whether all the momenta are independent or not, that is if the system is subjected to constraints. The mass-shell constraint will be discussed in Chapter \ref{chap:symplectic}, we will here focus on the spin degrees of freedom contained in the Lorentz matrix\footnote{From now on in this thesis, we drop the underlined indices for the sake of lisibility of the present discussion.} $\Lambda\tud{A}{B}$. One can show that the degrees of freedom implemented in the Lorentz matrix are not all independent from the linear momentum $p^\mu$. Since we have the freedom to perform local Lorentz transformations without affecting the physical state of the system, let us consider the specific transformation
\begin{align}
    L\tud{A}{B}\triangleq\delta^A_B-2\hat p^A\Lambda_{0B}-\frac{w^A w_B}{\hat p_C w^C},
\end{align}
where $\hat p^A\triangleq p^A/\mu$ and $w^A\triangleq\hat p^A+\Lambda\tud{A}{0}$. Under this transformation, the $\Lambda\tud{A}{0}$ components of the matrix transform as
\begin{align}
    \Lambda\tud{A}{0}&\to L\tud{A}{B}\Lambda\tud{B}{0}=\Lambda\tud{A}{0}+2\hat p^A-w^A=\hat p^A.
\end{align}
The other components transform as
\begin{align}
    \Lambda\tud{A}{i}\to\Lambda\tud{A}{i}-\Lambda\tud{B}{i}\frac{p_B w^A}{p_Cw^C}.
\end{align}
Therefore, using the freedom of performing local Lorentz transformations, we were able to bring the $\Lambda\tud{A}{0}$ components of the Lorentz matrix to $\hat p^A=\underline e\tud{A}{\mu}\hat p^\mu$. They are consequently redundant degrees of freedom and yield to the presence of constraints when turning to Hamiltonian formulation. As we will discuss in more details after having established the symplectic structure on phase space for spinning test bodies (see Chapter \ref{chap:symplectic}), the corresponding constraints will read
\begin{align}
    \phi^\mu\triangleq w_\alpha S^{\alpha\mu}=\qty(\hat p_\alpha+\Lambda_{0\alpha })S^{\alpha\mu}\approx 0.
\end{align}
This is very similar to the spin supplementary condition Eq. \eqref{covariant_ssc}. Different SSCs will correspond to different choices of $\Lambda_{0\alpha}$. Moreover, these constraints turn out to be first class, and the redundant degrees of freedom that we want to get rid of by enforcing a specific SSC are genuine gauge degrees of freedom of the system \cite{Steinhoff:2014kwa}. Working under a specific SSC will correspond to work within a specific choice of gauge fixation, or, from the Hamiltonian perspective, on a given constraint surface lying in the phase space. As already mentioned, more details will be provided in Chapter \ref{chap:symplectic}.

\section{SSCs and orthogonal decomposition of the spin tensor}
\label{sec:ssc:ortho}

\subsubsection{Orthogonal decomposition}
One can further explore the structure obtained when a covariant SSC has been enforced as follows. Let $\mathcal V^\mu$ be a timelike unit vector. In full generality, we can perform an orthogonal decomposition of the spin tensor with respect to that direction, and write \cite{Witzany:2018ahb,Ramond_2021}
\begin{align}
    S^{\mu\nu}&=-\epsilon^{\mu\nu\rho\sigma}\mathcal V_\rho S_\sigma+2 D^{[\mu}\mathcal V^{\nu]}.\label{decomposition_spin}
\end{align}
This relation inverts as
\begin{align}
    S^\mu=\frac{1}{2}\epsilon^{\mu\nu\rho\sigma}\mathcal V_\nu S_{\rho\sigma},\qquad D^\mu=\mathcal V_\alpha S^{\alpha\mu}.\label{inverse_rel}
\end{align}
Notice that $\mathcal V^\mu$ is orthogonal to both $S^\mu$ and $D^\mu$.

Now, identifying the vector $\mathcal V^\mu$ with the one of Eq. \eqref{covariant_ssc} shows that enforcing a covariant spin supplementary condition amounts to set
\begin{align}
    D^\mu=0
\end{align}
in the decomposition Eq. \eqref{decomposition_spin}. Therefore, \textit{when a covariant spin supplementary condition is enforced, the spin tensor is solely described in terms of a spin vector $S^\mu$.} Moreover, the vector $D^\mu$ takes the natural interpretation of the mass dipole of the system. One can pass from one description to another thanks to
\begin{align}
    S^\mu=\frac{1}{2}\epsilon^{\mu\nu\rho\sigma}\mathcal V_\nu S_{\rho\sigma}\qquad\Leftrightarrow\qquad S^{\mu\nu}=-\epsilon^{\mu\nu\rho\sigma}\mathcal V_\rho S_\sigma.
\end{align}

\subsubsection{SSC in adapted tetrad frame}
Any covariant spin supplementary condition takes a very simple form when the spin tensor components are expressed in a well-chosen background tetrad. Consider a background tetrad $\underline e\tdu{A}{\mu}$ whose timelike leg is given by
\begin{align}
    \underline e\tdu{0}{\mu}=\mathcal V^\mu.
\end{align}
This does not fix the tetrad uniquely, since its spatial legs can still be defined up to an arbitrary $\mathsf{SO}(3)$ transformation. In this tetrad, the components of the spin tensor are defined as $S^{AB}=\underline e\tud{A}{\mu}\underline e\tud{B}{\nu}S^{\mu\nu}$. A direct consequence of the SSC Eq. \eqref{covariant_ssc} is that
\begin{align}
    S^{0A}=\mathcal V_\mu \underline e\tud{A}{\nu}S^{\mu\nu}=0.
\end{align}
Therefore, \textit{when a covariant SSC is enforced, the only non-vanishing ``adapted tetrad'' components of the spin tensor are the purely spatial ones, $S^{IJ}$ (with $I,J=1,2,3$).}

The tetrad components of the vectors $D^\mu$ and $S^\mu$ read
\begin{subequations}
\begin{align}
    D^A&=\qty(0,\mathbf D),&&\mathbf D\triangleq\qty(S^{10},S^{20},S^{30}),\\
    S^A&=\qty(0,\mathbf S),&&\mathbf S\triangleq\qty(S^{23},S^{31},S^{12}).
\end{align}
\end{subequations}
Therefore, \textit{enforcing a covariant SSC amounts to set $\mathbf D=\mathbf 0$ in this formulation.}  Notice that the spacetime norms of the vectors $S^\mu$ and $D^\mu$ coincide with the (euclidean) norms of the three-vectors $\mathbf S$ and $\mathbf D$,
\begin{align}
    S^\mu S^\nu g_{\mu\nu}= S^I S^J\delta_{IJ}\triangleq\norm{\mathbf D}^2,\qquad S^\mu S^\nu g_{\mu\nu}= 
    D^I D^J\delta_{IJ}\triangleq\norm{\mathbf D}^2.
\end{align}

\subsubsection{The two Casimir invariants}
As will be found when studying Hamiltonian formulation of test bodies, there exist two Casimir invariants
\begin{align}
    \mathcal S^2\triangleq\frac{1}{2}S_{\mu\nu}S^{\mu\nu},\qquad \mathcal S_*^2\triangleq\frac{1}{8}\epsilon_{\mu\nu\rho\sigma}S^{\mu\nu}S^{\rho\sigma},
\end{align}
whose Poisson brackets with any other dynamical variable vanish.
Notice that $\mathcal S^2$ is precisely the spin magnitude introduced before.
Their explicit expressions become particularly enlightening when expressed in terms of the tetrad components of the vectors $S^A$ and $D^A$. Straightforward algebra allows to show that
\begin{subequations}
\begin{align}
    \mathcal S^2&= S_{0A}S^{0A}+\frac{1}{2}S_{IJ}S^{IJ}=\norm{\mathbf S}^2-\norm{\mathbf D}^2,\\
    \mathcal S_*^2&=\mathbf S\cdot\mathbf D.
\end{align}
\end{subequations}
Here, the dot stands for the usual euclidean scalar product. When a covariant SSC is enforced, the Casimir invariants reduce to
\begin{align}
    \mathcal S^2=\norm{\mathbf S}^2=S_\mu S^\mu,\qquad \mathcal S^2_*=0.
\end{align}
Therefore, $\mathcal S^2$ takes the simple interpretation of being the norm of the spin vector, whereas $\mathcal S_*^2$ is set to zero when a SSC is enforced. 

\section{Review of the various SSCs}
\label{sec:ssc:various}

In this final section, we will review the main covariant SSCs studied in the literature, \textit{i.e.} the possible choices of vector $\mathcal V^\mu$ in Eq. \eqref{covariant_ssc}. It is interesting to wonder how many independent degrees of freedom we are left with. This is a priori not clear, since the vector $\mathcal V^\mu$ is unspecified and can include additional degrees of freedom. This problem should be discussed separately for each specific spin supplementary condition. A brief review of each SSC is provided below, more details may be found in the references mentioned in the text or in the more recent papers \cite{Kyrian:2007zz,Witzany:2018ahb}. A pictorial overview is provided in Table \ref{tab:sscs}.

\subsubsection{Corinaldesi-Papapetrou (CP) and Newton-Wigner (NW)}
These two SSCs can be written together as
\boxedeqn{
    \qty(\xi_\mu+\alpha\hat p_\mu)S^{\mu\nu}=0,
}{Corinaldesi-Papapetrou/Newton-Wigner spin supplementary conditions}
where $\alpha=0$ for the CP condition \cite{1951RSPSA.209..248P} and $\alpha=1$ for the NW one \cite{RevModPhys.21.400}. Here, $\xi^\mu$ is an external timelike vector field that shall be specified. This kind of SSC amounts to choose some ``laboratory frame'' characterized by $\xi^\mu$ for expressing the multipole moments of the objects. We are left with two independent degrees of freedom in the spin tensor. The main advantage of this SSC is that it allows to formulate the motion under the form of a coordinate-time Hamiltonian system whose spin vector satisfies the canonical $\mathsf{SO}(3)$ algebra \cite{Vines:2016unv,Hanson:1974qy,Barausse:2009aa}. Its main drawback -- which will make it useless for our purposes -- is that it is not compatible with a covariant Hamiltonian formulation, since it is not preserved by the covariant Poisson bracket structure that we will derive in the next part of the thesis. 

\subsubsection{Mathisson-Pirani (MP)}
The Mathisson-Pirani \cite{mathisson_1940,Pirani:1956tn} SSC reads
\boxedeqn{
    v_\mu S^{\mu\nu}=0.
}{Mathisson-Pirani spin supplementary condition}
Unlike the CP/NW conditions, it does not not depends upon exterior structures. It corresponds to compute the moments in a frame $\mathcal V^\mu\propto v^\mu$ with respect to which the COM of the body is at rest. However, the spin tensor contains now four independent degrees of freedom, since the MP condition can be shown to hold provided that $S^{\mu\nu}$ is degenerate in \textit{some} timelike direction. Moreover, using this SSC, many initial data can give rise to the same physical situation, see \cite{Costa:2017kdr} for a detailed discussion.

\subsubsection{Kyrian-Semerák (KS)}
The KS condition provides a SSC for which the linear momentum is proportional to the four-velocity, $p^\mu=\mu v^\mu$. It reads
\boxedeqn{
    w_\mu S^{\mu\nu}=0
}{Kyrian-Semerák spin supplementary condition}
where $w^\mu$ is an \textit{arbitrary} unit timelike vector which is parallel-transported along the worldline, $\frac{\text D w^\mu}{\dd\lambda}=0$. Physically, it corresponds to first choose an arbitrary frame for computing the moments, and then parallel transport this frame along the worldline. The spin tensor is left with four independent degrees of freedom. Arbitrariness of the frame can be a problem, see \cite{Kyrian:2007zz,Vines:2016unv,Costa:2017kdr,Levi:2015msa} for details.

\subsubsection{Tulczyjew-Dixon (TD)}
The TD condition reads \cite{Tulczyjew:1959aa,dixon1973}
\boxedeqn{
    p_\mu S^{\mu\nu}=0.
}{Tulczyjew-Dixon spin supplementary condition}
Using this SSC, the components of the spin tensor and the linear momentum are not anymore all linearly independent, thus adding less dimensions to the spin sector of the phase space. This is the spin supplementary condition that we will use in the continuation of this thesis, since it is the one that provide the greatest number of advantages for our purposes, see Table \ref{tab:sscs}. The mains benefits of using TD SSC were nicely summarized by P. Ramond \cite{Ramond:2022vhj}. We briefly summarize the main ones here:
\begin{enumerate}
    \item TD SSC is compatible with covariant Hamiltonian formulation (like the MP and KS ones, but unlike the CP/NW one), which will be of prime importance for studying the associated Hamilton-Jacobi equation in Chapter \ref{chap:HJ};
    \item It does not require any background or external structure in its formulation, but only depends upon the intrinsic properties of the body;
    \item It reduces the number of degrees of freedom by the largest amount. We are left with only two DOFs for parametrizing the spin tensor, as for the CP/NW SSCs;
    \item TD condition is the only SSC that was proven to allow to define a unique center of mass worldline for the test body, see \cite{Harte_2015,Schattner:1979vp,Schattner:1979vn};
    \item It provides us with a dynamical mass $\mu$ which is conserved at pole-dipole order;
    \item Finally, it is the spin supplementary condition which naturally arises when deriving MPD equations from the perspective of Harte's generalized Killing vectors. This originates from the uniqueness of the worldline that the TD SSC defines.
\end{enumerate}

We end our short guided tour of the SSCs zoo here. We shall consider again SSCs from the Hamiltonian perspective in the last part of the thesis. Chapter \ref{chap:building} will derive more identities and properties that hold when the TD SSC is obeyed.

\begin{landscape}
\begin{table}
    \centering
    \renewcommand{\arraystretch}{1.2}
    \begin{tabular}{|c|c|c|c|c|}\hline
      \multirow{2}{*}{\textsc{Name}} & \textsc{Corinaldesi-Papapetrou/} & \textsc{Mathisson-}& \multirow{2}{*}{\textsc{Kyrian-Semerak}} & \textsc{Tulczyjew-} \\
       & \textsc{Newton-Wigner} & \textsc{Pirani} & & \textsc{Dixon} \\\hline
      \textsc{Abbreviation}& \textsc{CP/NW}& \textsc{MP} & \textsc{KS} & \textsc{TD}\\\hline
      \multirow{2}{*}{\textsc{Frame} $\mathcal V^\mu$}& 
        $\xi^\mu+\alpha\hat p^\mu$
      &  \multirow{2}{*}{$v^\mu$} & $w^\mu$ timelike & \multirow{2}{*}{$p^\mu$}\\
       & ($\alpha=0$ CP, $\alpha=1$ NW) & & s.t.  $\text D w^\mu/\dd\tau=0$ & \\\hline
      \textsc{Independent components} & \multirow{2}{*}{$2$} & \multirow{2}{*}{$4$} & \multirow{2}{*}{$4$} & \multirow{2}{*}{$2$} \\
      \textsc{of the spin tensor} &  & & & \\\hline
      \textsc{Requires additional} & \multirow{2}{*}{yes} & \multirow{2}{*}{no} & \multirow{2}{*}{ no} & \multirow{2}{*}{no} \\
      \textsc{background structure?} & & & & \\\hline
      \textsc{Compatible with covariant} & \multirow{2}{*}{no} & \multirow{2}{*}{yes} & \multirow{2}{*}{yes} & \multirow{2}{*}{yes} \\
      \textsc{Hamiltonian formulation?} & & & & \\\hline
    \end{tabular}
    \caption{The main spin supplementary conditions used in the literature.}
    \label{tab:sscs}
\end{table}
\end{landscape}

\chapter[Beyond the Pole-Dipole approximation: spin-induced multipoles]{Beyond the Pole-Dipole approximation:\\spin-induced multipoles}\label{chap:quadrupole}
\chaptermark{\textsc{Spin-induced multipoles}}

\vspace{\stretch{1}}

We have now in our possession a closed system of equations for describing the motion of extended test bodies in generic curved spacetime, up to quadrupole order included. However, as already mentioned in Chapter \ref{chap:EOM}, unlike the linear momentum and the spin dipole who both play the role of dynamical variables, the quadrupole moment $J^{\mu\nu\rho\sigma}$ should be prescribed. The choice of this prescription will encode the various physical effects leading to the appearance of a quadrupole moment, and depends on the internal structure of the body. Because our main concern is the description of a binary black hole system, we will only consider here the so-called \defining{spin-induced quadrupole}, that is, the quadrupole moment induced by the rotation of the test body. Other type of effects, such as tidal deformations \cite{Porto:2008tb} will be discarded here, since Kerr black holes can be shown to exhibit zero tidal deformability, see \textit{e.g.} \cite{LeTiec:2020bos}. This choice is also coherent with our underlying idea of understanding the motion of spinning test bodies as a perturbative expansion in $\mathcal S$ built upon the geodesic motion, since the presence of tidal-type quadrupole terms would disable us to recover the geodesic equations by taking the $\mathcal S\to 0$ limit of the MPD equations.

As it has been developed and applied in a number of works using various techniques (see e.g.\ \cite{Porto:2005ac,Porto:2008jj,Steinhoff:2011sya,Steinhoff:2014kwa,Marsat_2015,Bohe:2015ana,Levi:2015msa,Bini:2015zya,Vines:2016unv}), the form of $J^{\mu\nu\rho\sigma}$ appropriate to describe a spin-induced quadrupole moment is given (in a reparametrization invariant form) by
\boxedeqn{
J^{\mu\nu\rho\sigma}=\kappa\frac{3 p\cdot v}{(p^2)^2}p^{[\mu}S^{\nu]\lambda}S^{[\rho}{}_{\lambda}p^{\sigma]}+\os{4}.\label{quad_intro}
}{Spin-induced quadrupole moment}
At $\os{2}$, it is therefore unique up to a response coefficient $\kappa$ which controls the magnitude of the quadrupolar deformation, proportional to the square of the spin.  Typical values for $\kappa$ for a neutron star are in the range 4 to 8 \cite{Laarakkers:1997hb}, while for a Kerr black hole $\kappa$ is exactly equal to one, $\kappa_\text{BH}=1$. A broader introduction on the history of quadrupole and higher moments as well as their relations with PN expansion can be found in \cite{Marsat_2015}.

Actually, as implicitly stated in Eq. \eqref{quad_intro}, this quadrupole moment is not exact, but admits corrections that scale as $\mathcal O\qty(\mathcal S^4)$, which are consequently discarded in the present text. As will be shown in the discussion, two of our main perspectives on the approximation scheme used to truncate the multipole expansion (either considering it as a truncation in the number of derivatives of the Riemann tensor allowed in the Lagrangian or as a truncation of the expansion in the spin magnitude) will reveal to be consistent one to another when solely spin-induced multipole moments are taken into account.

In this short chapter, we aim to provide a pedagogical, self-contained derivation of the spin-induced quadrupole, expanding the procedure described in \cite{Marsat_2015}. We will always assume that Tulczyjew-Dixon spin supplementary condition holds, as well as consider the background metric to be Ricci-flat. Notice that choosing the evolution parameter to be the proper time of the body, Eq. \eqref{quad_intro} reduces to the simpler expression
\begin{align}
    J^{\mu\nu\rho\sigma}= -\frac{3\kappa}{\mu} v^{[\mu}\Theta^{\nu][\rho}v^{\sigma]}+\os{4},\quad\text{where }\quad\Theta^{\alpha\beta}\triangleq S^{\alpha\lambda}S\tud{\beta}{\lambda}.\label{spin_induced_Q7}
\end{align}

\section{Spin-induced multipoles from dimensional analysis}

We aim to derive the most generic form of the spin-induced quadrupole term appearing in the MPD equations \eqref{MPD_spinning} and \eqref{MPD_impulsion}. Recall that the quadrupole tensor is defined in the Lagrangian formulation as 
\begin{align}
    J^{\alpha\beta\gamma\delta}\triangleq -6\pdv{L}{R_{\alpha\beta\gamma\delta}}.
\end{align}
This can be straightforwardly extended to the $2^{l+2}$-pole moment, which will be defined as
\begin{align}
    J^{\alpha\beta\gamma\delta\mu_1\ldots\mu_l}\triangleq C_l \pdv{L}{\qty(\nabla\mu_{1}\ldots\nabla_{\mu_l}R_{\alpha\beta\gamma\delta})},
\end{align}
where $C_l$ is an arbitrary normalization constant (for the quadrupole, $C_0=-6$).

Without loss of generality, we can assume that it originates from a term present in the Lagrangian \eqref{L_quad} taking the form
\begin{align}
    L\ni\frac{1}{C_l}\nabla_{\mu_1}\ldots\nabla_{\mu_l}R_{\alpha\beta\gamma\delta}J^{\alpha\beta\gamma\delta\mu_1\ldots\mu_l}.\label{quadrupole_Lagragian}
\end{align}
Notice that from its very definition, it is clear that the tensor $J^{\alpha\beta\gamma\delta\mu_1\ldots\mu_l}$ should possess the same algebraic symmetries than the Riemann tensor in its four first indices $\alpha\beta\gamma\delta$. 

Thanks to dimensional analysis, we will constrain the form of such multipole tensors. We will always consider \defining{spin-induced multipoles}, in the sense that they can only depend on $p^\mu$, $v^\mu$, $S^{\mu\nu}$, $g_{\mu\nu}$ and $R_{\alpha\beta\gamma\delta}$ and that they should be vanishing when the spin vanishes itself. Under the Tulczyjew-Dixon spin supplementary condition, the relation between four-velocity and linear momentum allows to replace the dependence in the latter by a dependence in the mass $\mu$.

We will subsequently show that the form of this spin-induced quadrupole tensor is necessarily given by Eq. \eqref{spin_induced_Q7}, that is, it is unique at $\os{2}$ up to an overall coefficient depending on the body's nature.

\subsubsection{Dimensional analysis}
To further constrain the form of the spin-induced quadrupole, we will use dimensional analysis. Restoring the dimensional character of the gravitational constant but still setting the speed of light to one ($G\neq c=1$), units of length are still equal to units of time, but not to units of energy (which are the same than units of mass), \textit{i.e.} $L=T\neq E=M$. The dimensions of the relevant quantities for the following are given by \cite{Gratus:2020cok}
\begin{align}
    &[g_{\mu\nu}]=[v^\mu]=1,\qquad [\mu]=[p^\mu]=[L]=M,\qquad [\nabla_\mu]=L^{-1},\nonumber\\
    &[R\tud{\mu}{\nu\rho\sigma}]=L^{-2},\qquad[S^{\mu\nu}]=ML,\qquad [J^{\alpha\beta\gamma\delta\mu_1\ldots\mu_l}]=ML^{l+2}.\label{dimensions}
\end{align}
As stated above, spin-induced multipole moments should only depend on the mass, the four-velocity, the metric and the Riemann tensor. In this text, we will by assumption exclude any non-trivial dependence in the Levi-Civita tensor, since this would break the parity \cite{Marsat_2015}. Schematically (\textit{i.e.} without writing explicitly the various contractions between indices), we can assume\footnote{Actually, this only amounts to assume that the corresponding term of the Lagrangian is smooth. It consequently admits an expansion in an integer power series, which is nothing but a polynomial containing infinitely many terms. The only important point is that these terms should take the form given in Eq. \eqref{generic_term}.} that Eq. \eqref{quadrupole_Lagragian} will take the form of a polynomial whose terms look like
\begin{align}
    \mu ^{N_1}\qty[\qty(\nabla_\lambda)^{N_2} \qty(R\tud{\mu}{\nu\rho\sigma})^{N_3}]  \qty(S^{\alpha\beta})^{N_4}\qty(v^\kappa)^{N_5},\label{generic_term}
\end{align}
with $N_1\in\mathbb Z,~N_2,N_3,N_4,N_5\in\mathbb N$ and $N_4>0$. These choices for the coefficients $N_i$ can be understood as follows:
\begin{itemize}
    \item We allow for negative powers of the mass $\mu$ to appear, in order to control the dimension of the expression;
    \item Since we are considering spin-induced multipoles, they shall shall consistently tend to zero when $\mathcal S\to 0$, yielding $N_4>0$.
\end{itemize}
Moreover, the metric does not appear in Eq. \eqref{generic_term} since it is viewed as a dimensionless quantity that is only used to lower and raise indices.
Performing a dimensional analysis of the expression Eq. \eqref{generic_term} with the dimensions given in Eq. \eqref{dimensions}, it is easy to show that the following relations between the power coefficients hold:
\begin{align}
    N_1=\frac{N_4-N_2}{2},\qquad N_3=1-N_4,\qquad N_2,N_4\text{ and }N_5\text{ arbitrary}.
\end{align}
The parameter $N_2$ controls the number of covariant derivatives of the Riemann tensor appearing in the expression, that is, the order of the approximation from the Lagrangian perspective. The number of copies of the four-velocity appearing in the expression is always left unconstrained from the dimensional analysis viewpoint, since it is a dimensionless quantity. As we will show in next section, this $N_5$ will actually be constrained by the number of non-equivalent, non-trivial contractions that can be made between the indices of Eq. \eqref{generic_term}.

The first non-trivial solutions to our dimensional analysis are schematically depicted in Table \ref{tab:multipoles}. We directly notice that, in the quadrupole approximation (that is, for $N_2=0$), the non-trivial term containing the smallest number of copies of the spin tensor corresponds to $N_3=1$, $N_1=-1$ and $N_4=2$, which thus scales as $RSSv^{N_5}/\mu$. All the other consistent terms will include at least four copies of the spin tensor (the next one being $RRSSSSv^{N_5}/\mu^3$), and they can consequently be discarded in our analysis, since we are only interested in (at most) quadratic terms in the spin.

\begin{table}[ht]
    \centering\renewcommand{\arraystretch}{1.5}
    \begin{tabular}{c|ccc}
         & $\propto R$ ($N_2=0$) & $\propto\nabla R$ ($N_2=1$) & $\propto\nabla\nabla R$ ($N_2=2$) \\\hline
        $\mathcal O\qty(\mathcal S^2)$ ($N_4=2$) & $RSS v^{N_5}/\mu$ & $0$ & $0$ \\
        $\mathcal O\qty(\mathcal S^3)$ $(N_4=3$) & $0$ & $\nabla R SSSv^{N_5}/\mu^2$& $0$ \\
        $\mathcal O\qty(\mathcal S^4)$ ($N_4=4$) & $RRSSSSv^{N_5}/\mu^3$ & $0$ & $\nabla\nabla RSSSSv^{N_5}/\mu^3$\\
    \end{tabular}
    \caption{The first terms resulting from the dimensional analysis described in the text.}
    \label{tab:multipoles}
\end{table}

\section{Explicit form of the spin-induced quadrupole}

We will now focus on the $\mathcal O\qty(\mathcal S^2)$ contribution to the quadrupole, which reads schematically
\begin{align}
    \frac{1}{\mu}R\tud{\mu}{\nu\rho\sigma}S^{\alpha\beta}S^{\gamma\delta}v^{\alpha_1}\ldots v^{\alpha_{N_5}}.\label{expr_no_contr}
\end{align}
Finding the most generic expression for this spin-induced quadrupole amounts to list all the independent, non-trivial contractions that can be made between the free indices of Eq. \eqref{expr_no_contr}. Since we only consider Ricci-flat spacetimes, the Riemann tensor reduces to its trace-free part, namely the Weyl tensor \cite{Wald:1984rg}. Consequently, its indices cannot be contracted together. Moreover, it is not possible to contract the four-velocity with the spin tensor, since the TD SSC enforces the relation $v_\mu S^{\mu\nu}=\os{3}$ to hold. Finally, contracting the four-velocity with itself leads to trivial contributions, since $v_\mu v^\mu=-1$ when the evolution is parametrized by the proper time. At the end of the day, we are only left with three (seemingly independent) combinations of the desired form that produce a well-defined scalar that can contribute to the Lagrangian:
\begin{align}
    L_\text{quad}=\frac{C_1}{\mu}R_{\alpha\beta\gamma\delta}v^{\alpha}\Theta^{\beta\gamma}  v^{\delta}+ \frac{C_2}{\mu}R_{\alpha\beta\gamma\delta}S^{\alpha\beta}S^{\gamma\delta}+\frac{C_3}{\mu}R_{\alpha\beta\gamma\delta}S^{\alpha\gamma}S^{\beta\delta}+\mathcal O\qty(\mathcal S^4),\label{quadrupole_3_coeffs}
\end{align}
with $C_1,C_2,C_3\in\mathbb R$. 

We will prove now that these three terms are not independent and that there is only one relevant combination that can be written out of them. Obviously, from the algebraic Ricci identity $R_{\alpha[\beta\gamma\delta]}=0$, one has
\begin{align}
    R_{\alpha\beta\gamma\delta}S^{\alpha\beta}S^{\gamma\delta}=2R_{\alpha\beta\gamma\delta}S^{\alpha\gamma}S^{\beta\delta}.
\end{align}
The two last terms of Eq. \eqref{quadrupole_3_coeffs} are consequently the same. Moreover, the Weyl tensor $C_{\mu\nu\rho\sigma}$ can be shown to obey the identity \cite{Marsat_2015}
\begin{align}
    C_{\alpha\beta\gamma\delta}=4g_{\rho[\alpha}C_{\beta]\mu\nu[\gamma}\delta_{\delta]}^\rho v^\mu v^\nu+2C_{\alpha\beta\mu[\gamma}v_{\delta]}v^\mu+2C_{\gamma\delta\mu[\alpha}v_{\beta]}v^\mu.
\end{align}
Therefore, one can show that
\begin{align}
    C_{\alpha\beta\gamma\delta}S^{\alpha\gamma}S^{\beta\delta}=-2C_{\alpha\beta\gamma\delta}v^\alpha\Theta^{\beta\gamma}v^\delta.
\end{align}
The quadrupole term of the Lagrangian then reduces to
\begin{align}
    L\ni L_\text{quad}=\frac{C_1-2\qty(2C_2+C_3)}{\mu}R_{\alpha\beta\gamma\delta}v^{\alpha}\Theta^{\beta\gamma}  v^{\delta}+\mathcal O\qty(\mathcal S^4),
\end{align}
which gives rise to the spin-induced quadrupole tensor announced in Eq. \eqref{quad_intro}, provided we rename the overall constant $C_1-2\qty(2C_2+C_3)\triangleq 3\kappa$ in order to agree with the conventions widely used in the literature. One can check explicitly that this spin-induced quadrupole tensor obeys all the algebraic symmetries of the Riemann tensor.

We have therefore demonstrated that the form of the spin-induced quadrupole is unique, up to an overall coupling coefficient $\kappa$ whose numerical value depends on the nature of the test body. Discussing the value taken by $\kappa$ goes beyond the scope of this thesis, and is actually quite involved, see \textit{e.g.} \cite{Marsat_2015} and references therein. 

\partimage[width=\textwidth]{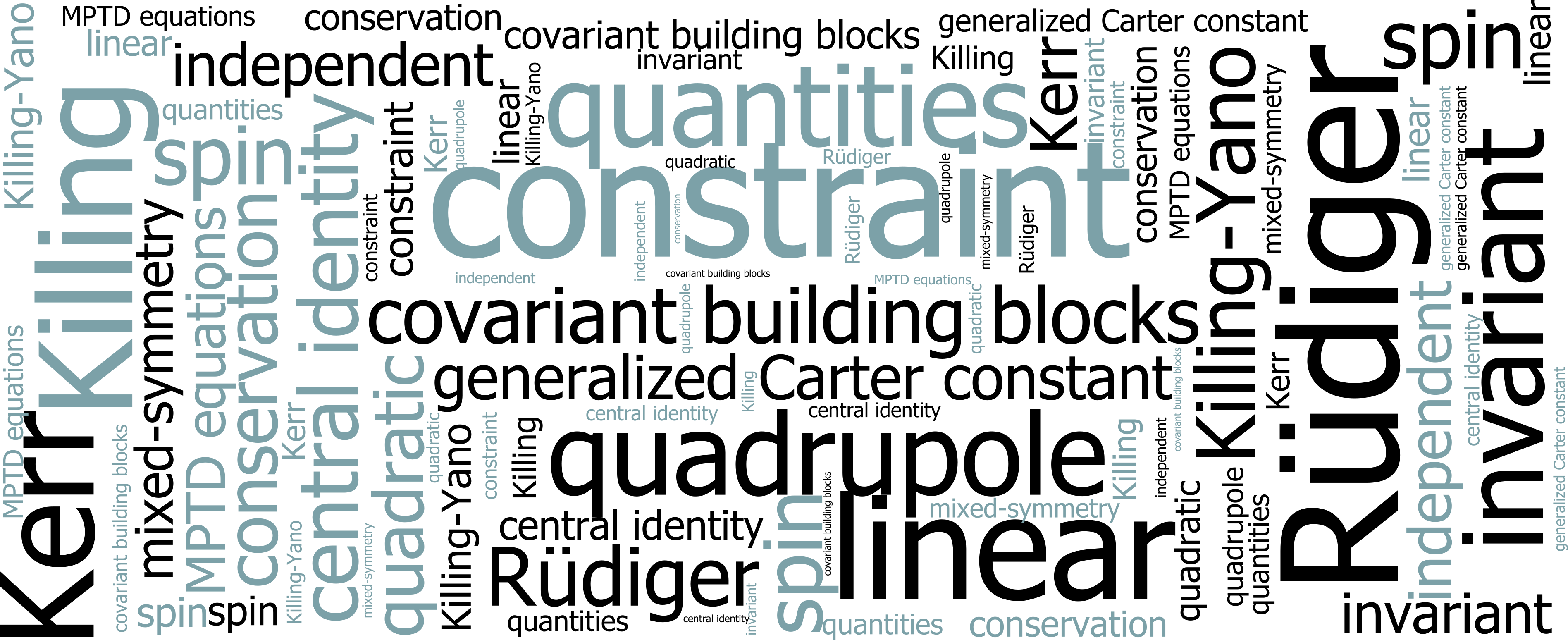}

\part[\textsc{Test Bodies in Kerr Spacetime: Conserved Quantities}]{\textsc{Test Bodies in Kerr Spacetime:\\Conserved Quantities}}\label{part:conserved_quantities}
\noindent

{
\renewcommand{\thefigure}{III.1}
\begin{figure}
\centering
\begin{tikzpicture}[scale=0.7]
\fill[black!10] (0,0) rectangle (6,12.1);
\fill[black!15] (6,0) rectangle (12,12.1);
\fill[black!20] (12,0) rectangle (18,12.1);

\node[below] () at (3,12) {\Large $\mathcal O\qty(\mathcal S^0)$};
\node[below] () at (9,12) {\Large $\mathcal O\qty(\mathcal S^1)$};
\node[below] () at (15,12) {\Large $\mathcal O\qty(\mathcal S^2)$};

\fill[blue!30, rounded corners, draw=blue,ultra thick] (1,10) rectangle (17,8.5);
\node[blue] () at (9,9.25) {\Large $\mathcal S$};

\draw[ultra thick, red, ->] (11,7.25)--(13,7.25);

\fill[red!30, rounded corners, draw= red,ultra thick] (1,8) rectangle (11,6.5);
\node[red] () at (6,7.25) {\Large $\mu$};

\fill[red!20, rounded corners, draw=red,ultra thick] (13,8) rectangle (17,6.5);
\node[red] () at (15,7.25) {\Large $\tilde\mu$};

\draw[ultra thick, green, ->] (5,5.25)--(7,5.25);

\fill[green!30, rounded corners, draw=green,ultra thick] (1,6) rectangle (5,4.5);
\node[green] () at (3,5.25) {\Large $\mathcal E_0,\mathcal L_0$};

\fill[green!20, rounded corners, draw=green,ultra thick] (7,6) rectangle (17,4.5);
\node[green] () at (12,5.25) {\Large $\mathcal E,\mathcal L$};

\draw[ultra thick, orange, ->] (5,3.25)--(7,3.25);
\draw[ultra thick, orange, ->] (11,3.25)--(13,3.25);

\fill[orange!30, rounded corners, draw=orange,ultra thick] (1,4) rectangle (5,2.5);
\node[orange] () at (3,3.25) {\Large $\mathcal K_0$};

\fill[orange!20, rounded corners, draw=orange,ultra thick] (7,4) rectangle (11,2.5);
\node[orange] () at (9,3.25) {\Large $\mathcal Q_R$};

\fill[orange!10, rounded corners, draw=orange, dashed,ultra thick] (13,4) rectangle (17,2.5);
\node[orange] () at (15,3.25) {\Large $\mathcal Q_\text{BH}^{(2)}$};

\fill[cyan!30, rounded corners] (7,2) rectangle (17,0.5);
\node[cyan] () at (12,1.25) {\Large $\mathcal Q_Y$};
\draw[cyan,rounded corners,ultra thick] (12,2)--(7,2)--(7,0.5)--(12,0.5);
\draw[cyan,rounded corners,dashed,ultra thick] (12,0.5)--(17,0.5)--(17,2)--(12,2); 

\end{tikzpicture}
    \caption{Conserved quantities for the motion of spinning test bodies in Kerr spacetime. The figure summarizes the status at $\mathcal O\qty(S^0)$ (geodesic motion), $\mathcal O\qty(S^1)$ (linearized MPTD equations) and $\mathcal O\qty(S^2)$ (quadratic MPTD equations endowed with a spin-induced quadrupole). The dashed contours mean that the conservation only holds provided the quadrupole coupling is the one of a test black hole, $\kappa=1$.}
    
    \label{fig:spinning_conserved_quantities}
\end{figure}
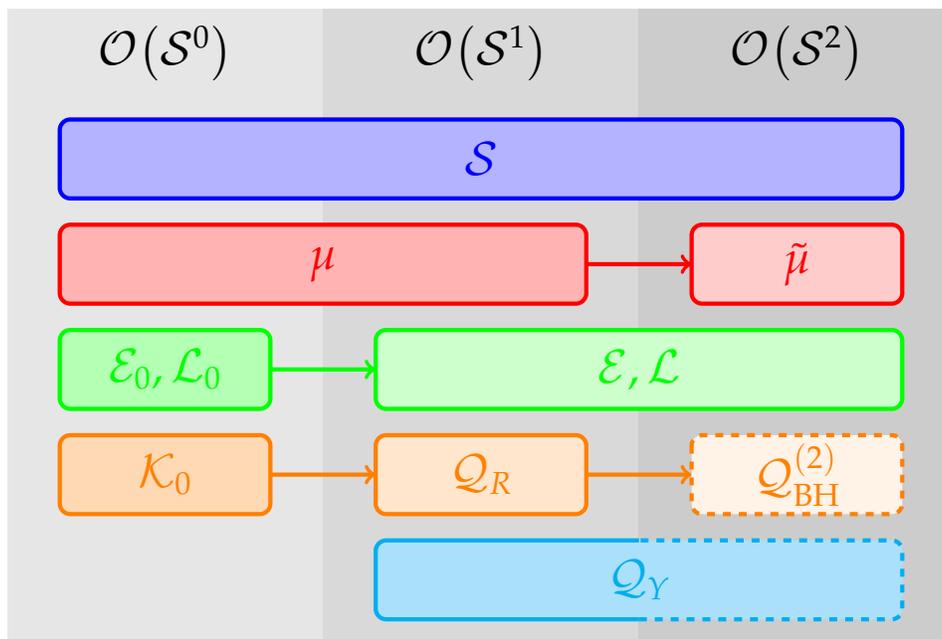
}

\lettrine{W}{e} have now in our possession of a system of equations for describing the motion of spinning test bodies in any curved spacetime, the MPD equations. In the continuation of this text, and for the reasons discussed in Chapter \ref{chap:SSC}, we will always supplement them by the Tulczyjew-Dixon condition. The total system of equations obtained will be referred to as the \defining{MPTD equations}. 

The MPTD equations take the form of a set of first order partial differential equations for the linear momentum $p^\mu$ and the spin $S^\mu$. However, the very goal of anyone wanting to solve the MPTD equations is to obtain the position of the test body as a function of the proper time, $z^\mu(\tau)$. With respect to the positions $x^\mu$, the MPTD equations are a set of \textit{second order} PDEs. As it is always the case when studying such a dynamical system, much information can be obtained  if one is able to build \textit{first integrals} of motion (or invariants or conserved quantities), \textit{i.e.} functions of the dynamical variables $\mathcal Q(p^\alpha,S^\alpha)$ that are constant along the motion:
\begin{align*}
    \dot{\mathcal{Q}}(p^\alpha,S^\alpha)=0,\qquad \dot{}\triangleq\dv{\tau}.\label{cons}
\end{align*}
As always in General Relativity, the existence of first integrals of motion will be strongly related with the presence of symmetries of the background spacetime ({i.e.} the existence of Killing vectors or Killing(-Yano) tensors). Building such conserved quantities will be the main concern of this part of the thesis. 

Because what we have in mind is the description of the dynamics of extreme mass-ratio inspirals, we will treat the problem perturbatively, order by order in the spin magnitude $\mathcal S$. The conservation equation will thus be required to hold only up to some given order in $\mathcal S$, giving rise to quantities which are only ``quasi-conserved''. Nevertheless, we will often refer to them as conserved quantities. The reader should keep this remark in mind for what follows. We will push our analysis up to second order in $\mathcal S$. This is the first order at which the universal character of the MPTD dynamics is broken, since the explicit form of the quadrupole moment will depend upon the internal structure of the test body.

Moreover, we will take the force and torque terms to originate only from the spin-induced quadrupole moment introduced in Chapter \ref{chap:quadrupole}. This choice provide an exact description for the secondary object being a black hole, since in that case it will not exhibit tidal deformability, see \textit{e.g.} \cite{Charalambous:2021mea} and references therein. However, for the secondary being another type of compact object (\textit{e.g.} a neutron star), tidal-type contributions will also arise at quadrupole order, see \textit{e.g.} \cite{Porto:2016pyg}. We expect these contributions to break the conservation of the so-called ``hidden'' constants of motion discussed below. As we will see, this expectation will turn out to be consistent with our findings, since such constants of motion will be shown to only exist for the spin-induced quadrupole coupling $\kappa$ taking the value $\kappa=1$, which is precisely the one expected for a black hole. The various conserved quantities that can be built are summarized in Figure \ref{fig:spinning_conserved_quantities}, and briefly discussed below. Most of them will take the form of (perturbative) deformations of the quantities conserved along the geodesic motion.

\subsubsection{``Explicit'' constants of motion}

As it has been proven in Chapter \ref{chap:EOM}, the spin magnitude $\mathcal S$ is exactly conserved for any choice of spin supplementary condition. Moreover, up to $\os{1}$, the dynamical mass $\mu^2=-p_\alpha p^\alpha$ is also invariant. At second order in $\mathcal S$, $\mu$ is not conserved anymore, but we can still define a mass-like quantity $\tilde\mu$ which is quasi-conserved, see Section \ref{sec:MPT}.

As shown by Dixon \cite{1970RSPSA.314..499D,1977GReGr...8..197E}, and as was central to his construction of the multipole moments, if the background has a Killing vector $\xi^\mu$, then the quantity
\begin{equation*}
\mathcal C_\xi=p_\mu\xi^\mu+\frac{1}{2}S^{\mu\nu}\nabla_\mu\xi_\nu
\end{equation*}
is exactly conserved along any worldline when $p^\mu$ and $S^{\mu\nu}$ are evolved by the MPD equations, to all orders in the multipole expansion, for arbitrary quadrupole and higher moments.  (See also, \textit{e.g.}, the earlier derivation by Souriau for the pole-dipole system \cite{Souriau:1974}, and the insightful exposition by Harte \cite{Harte:2014wya}.) Specializing to a background Kerr spacetime, the existence of its two Killing vector fields leads to the appearance of two conserved quantities, denoted $\mathcal E$ and $\mathcal L$. They will naturally take the interpretation of energy and (projection onto the black hole axis of the) angular momentum of the test body. 

One is then naturally lead to wonder whether the hidden symmetry of Kerr spacetime leads to conserved quantities for the MPD dynamics, including a generalization of the Carter constant to the case of spinning extended test bodies.

\subsubsection{``Hidden'' constants: to linear order in $\mathcal S$}

This question was answered for the case of the pole-dipole MPTD equations by R. R\"udiger \cite{doi:10.1098/rspa.1981.0046,doi:10.1098/rspa.1983.0012}. First he showed that the quantity
\begin{align*}
    \mathcal Q_Y={}Y^*_{\mu\nu}S^{\mu\nu}
\end{align*}
built from the Killing-Yano tensor $Y_{\mu\nu}$ of Kerr is conserved, up to remainders quadra-tic in the spin tensor and quadrupolar corrections. Here $Y^*_{\mu\nu}\triangleq\tfrac{1}{2}\epsilon_{\mu\nu\alpha\beta}Y^{\alpha\beta}$ is the dual of the Kerr Killing-Yano tensor $Y^{\mu\nu}$. 

He further showed that there is indeed a generalized Carter constant of the form 
\begin{equation*}\label{Q_dipole_intro}
\mathcal Q_R=Y\tdu{\mu}{\lambda}Y_{\nu\lambda}p^{\mu} p^{\nu}+4 \xi^\lambda\epsilon_{\lambda \mu\sigma[\rho}Y\tud{\sigma}{\nu]}S^{\mu\nu}p^\rho
\end{equation*}
which is conserved along the MPTD equations, up to remainders quadratic in the spin tensor and quadrupolar corrections that were not determined. A generalization of this result to the Kerr-Newman (charged spinning black hole) spacetime was independently discovered by Gibbons \emph{et al.}~\cite{Gibbons:1993ap} using a supersymmetric description of spinning particle dynamics.    The existence of these constants of motion has been shown by Witzany to allow the separation of a Hamilton-Jacobi equation for the pole-dipole system in Kerr, leading to analytic expressions for the fundamental frequencies of the motion \cite{Witzany:2019nml} using a Hamiltonian formalism for spinning test bodies  \cite{Witzany:2018ahb}.

The unanswered question of the existence of other independent quasi-conserved quantities for the MPTD equations led to an undetermined status of the role of integrability and by opposition, chaos, in the dynamics of spinning test bodies around Kerr. While chaos has been established to appear at second order in the spin \cite{Zelenka:2019nyp}, numerical simulations suggest that no chaos occurs at linear order in the spin \cite{Kunst:2015tla,Ruangsri:2015cvg,Zelenka:2019nyp}. From the  solution of the Hamilton-Jacobi equations at linear order in the spin, one can infer that chaotic motion at linear order is negligible \cite{Witzany:2019nml}. In \cite{Compere:2021kjz}, we have related the existence of new quasi-conserved quantities homogeneously linear in the spin to the existence of a new tensorial structure on the background, that we will refer to as a \textit{mixed-symmetry Killing tensor}.
This result applies to the Kerr background and more generally to Ricci-flat spacetimes admitting a Killing-Yano tensor. We have demonstrated that under the assumption of stationarity and axisymmetry, no such non-trivial structure exists on the Schwarzschild background and no non-trivial mixed-symmetry Killing tensor on Kerr can be constructed from deformations of trivial mixed-symmetry Killing tensors on Schwarzschild.

\subsubsection{``Hidden constants'': quadratic order in $\mathcal S$}

In \cite{Compere:2023alp}, we have explored whether such ``hidden constants'' exist for test bodies with spin-induced quadrupole moments moving in a Kerr background. As the central results of that paper, we established that two quantities, $\mathcal Q_Y$ and $\mathcal Q^{(2)}_\text{BH}$, are conserved up to cubic-in-spin or octupolar corrections,
\begin{equation*}
\frac{\dd \mathcal Q_Y}{\dd\tau}=\mathcal O(\mathcal S^3),
\qquad
\frac{\dd \mathcal Q^{(2)}_\text{BH}}{\dd\tau}=\mathcal O(\mathcal S^3),
\end{equation*}
along the motion of a ``quadrupolar test black hole'', governed by the MPTD equations endowed with the spin-induced quadrupole term \eqref{spin_induced_Q}, with $\kappa=1$, in a background Kerr spacetime, for arbitrary orbital and spin orientations.  The first quantity $\mathcal Q_Y$ is R\"udiger's linear-in-spin constant, unmodified.
The second quantity $\mathcal Q^{(2)}_\text{BH}$ is quadratic in $p^\mu$ and $S^{\mu\nu}$ and generalizes R\"udiger's constant $\mathcal Q_R$ to the quadrupolar order for a test black hole; it is given explicitly by
\begin{align*}
\begin{split}
\mathcal{Q}^{(2)}_{\text{BH}}&=
Y\tdu{\mu}{\lambda}Y_{\nu\lambda}p^{\mu} p^{\nu}+4 \xi^\lambda\epsilon_{\lambda \mu\sigma[\rho}Y\tud{\sigma}{\nu]}S^{\mu\nu}p^\rho
\\
&\quad+\bigg[
g_{\mu\rho}\Big(\xi_\nu \xi_\sigma-\frac{1}{2}g_{\nu\sigma}\xi^2\Big)-\frac{1}{2}Y_\mu{}^\lambda\Big(Y_\rho{}^\kappa R_{\lambda\nu\kappa\sigma}+\frac{1}{2}Y_\lambda{}^\kappa R_{\kappa\nu\rho\sigma}\Big)
\bigg]S^{\mu\nu}S^{\rho\sigma}.
\end{split}
\end{align*}

\subsubsection{Plan of the text}
This part of the text is structured as follows: Chapter \ref{chap:building} will review the MPTD equations endowed with the spin-induced quadrupole term, and derive a number of useful identities. Rüdiger's procedure for building conserved quantities through resolving explicitly the conservation equations will then be discussed, and illustrated on two simple examples. Subsequently, Chapters \ref{chap:first_order_solutions} and \ref{chap:second_order_solution} will apply this procedure, respectively at first and second order in the spin magnitude $\mathcal S$. We shall warn the reader that these two last chapters consist into long and very technical derivations, which only lead to the results exposed in this introduction.

\chapter{Building conserved quantities}\label{chap:building}


This chapter aims to introduce in a somehow pedagogical way the general scheme introduced by R. Rüdiger in the early eighties \cite{doi:10.1098/rspa.1981.0046,doi:10.1098/rspa.1983.0012} for finding quantities that are conserved for the motion of a test body driven by the MPD equations. Even if it is conceptually simple, the actual computations required by this scheme turn out to be very cumbersome. This is the reason why this chapter introduces the procedure conceptually, before applying it on the two simplest examples one can find: generic conserved quantities for geodesic motion on a generic curved background, and linear (in the sense defined in Section \ref{sec:invariants}) conserved quantities for the MPD linearized equations. The ``real'' computations will be delayed to the two following chapters. 

\section{Review: MPD equations with spin-induced quadru-pole under the Tulczyjew-Dixon SSC}
\label{sec:MPT}

Before tackling this program, let us recall ourselves that MPD equations only form a closed system of equations when a spin supplementary condition is enforced, see Chapter \ref{chap:SSC}. In this part of the thesis, we will always consider MPD equations endowed with the spin-induced quadrupole term and supplemented by the Tulczyjew-Dixon SSC. The full system of equations will be referred to as the MPTD equations. This section aims to both review these equations and their properties, introduce some useful notations and derive relations that will become of prime importance for the continuation of this work. This section is somehow redundant with Part \ref{part:spinnin_bodies} of the text, but aims to gather all the results relevant for the forthcoming computations.

\subsection{Mathisson-Papapetrou-Dixon equations}
As it has been extensively discussed in Part \ref{part:spinnin_bodies} of the thesis, the motion of an extended test body over a curved background is described -- in General Relativity -- by the Mathisson-Papapetrou-Dixon (MPD) equations
\begin{subequations}\label{MPD}
\begin{align}
    \frac{\text{D}p^\mu}{\dd\tau}&=-\frac{1}{2}R\tud{\mu}{\nu\alpha\beta}v^\nu S^{\alpha\beta}+\mathcal F^\mu,\label{MPD_1}\\
    \frac{\text{D}S^{\mu\nu}}{\dd\tau}&=2 p^{[\mu}v^{\nu]}+\mathcal L^{\mu\nu}.\label{MPD_2}
\end{align}
\end{subequations}
where we defined the tangent vector to the worldline as $v^\mu\triangleq\dv{x^\mu}{\lambda}$ ($\lambda$ being any affine parameter) and where $\ldv{}{\lambda}\triangleq v^\mu\nabla_\mu$ is the covariant derivative along the worldline. $\mathcal F^\mu$ and $\mathcal L^{\mu\nu}$ are force and torque terms that shall be specified, depending upon the multipole moments the body is described with. We also introduce the notations
\begin{subequations}
    \begin{align}
    \mathfrak m&\triangleq -p^\mu v_\mu, \\
    \mu^2&\triangleq-p^\mu p_\mu, \\
    \mathcal S^2&\triangleq\frac{1}{2}S^{\mu\nu}S_{\mu\nu}. 
\end{align}
\end{subequations}
Here, $\mu^2$ is the \defining{dynamical rest mass} of the object, \textit{i.e.} the mass of the object measured by an observer in a frame where the spatial components of the linear momentum $p^i$ do vanish; $\mathfrak m$ will be referred to as the \defining{kinetic mass} and $\mathcal S$ as the \defining{spin magnitude}.

At this point, let us emphasize that the linear momentum is no aligned with the velocity, and thus not tangent to the worldline, since contracting Eq. \eqref{MPD_2} with $v^\nu$ yields
\begin{equation}
    p^\mu=\frac{1}{v^2}\qty(v_\alpha\ldv{S^{\mu\alpha}}{\lambda}-\mathfrak m\,v^\mu).\label{pAsv} 
\end{equation}

The dynamical and kinetic masses are in general \textit{not} constants of motion: in fact, one can show that
\begin{align}
    \dv{\mathfrak m}{\lambda}&=-\frac{1}{v^2}\ldv{v_\alpha}{\lambda} v_\beta \ldv{S^{\alpha\beta}}{\lambda}, \\
    \dv{\mu}{\lambda}&=-\frac{1}{\mu \,\mathfrak m}p_\alpha\ldv{p_\beta}{\lambda}\ldv{S^{\alpha\beta}}{\lambda}.\label{evolMu} 
\end{align}
As discussed in Chapter \ref{chap:EOM}, the spin magnitude is exactly conserved for any choice of multipoles:
\begin{equation}
    \dv{(\mathcal S^2)}{\lambda}=0. 
\end{equation}

\subsection{The spin supplementary condition}
We supplement the MPD equations by enforcing the covariant  \textit{Tulczyjew-Dixon} \textit{spin supplementary condition} \begin{equation}
    S^{\mu\nu}p_\nu=0. \label{ssc}
\end{equation}
The TD SSC forms a set of three additional constraints since a contraction with $p_\mu$ leads to a trivial identity. In what follows, we will choose the affine parameter driving the evolution as the body's proper time, $\lambda=\tau$.  This enforces the four-velocity to be normalized,
\begin{equation}\label{v2}
    v^\mu v_\mu=-1, 
\end{equation}
which thereby guarantees its timelike nature along the evolution of the system. These conditions consequently close our system of equations which is then call the MPTD equations. These conditions fix uniquely the worldline and allow to invert Eq. \eqref{pAsv} in order to express the four-velocity as a function of the linear momentum. 

\subsection{Spin-induced quadrupole approximation} 
We will consider only spin-induced multipole moments and work in the quadrupole approximation, \textit{i.e.} neglecting octupole and higher moments. This is the relevant approximation for addressing spin-squared interactions: considering only spin-indu-ced multipole terms, the $2^n$-pole scales as $\mathcal O(\mathcal S^n)$, with $\mathcal S^2\triangleq\frac{1}{2}S_{\mu\nu}S^{\mu\nu}$.

At the level of the equations of motion, this corresponds to choose the force and torque given by 
\begin{align}
    \mathcal F^\mu&=-\frac{1}{6}J^{\alpha\beta\gamma\delta}\nabla^\mu R_{\alpha\beta\gamma\delta},\qquad
    \mathcal L^{\mu\nu}=\frac{4}{3}R\tud{[\mu}{\alpha\beta\gamma}J^{\nu]\alpha\beta\gamma}.
\end{align}
The quadrupole tensor $J^{\mu\nu\rho\sigma}$ possesses the same algebraic symmetries as the Riemann tensor. 

We further particularize our setup by considering only a quadrupole moment that is induced by the spin of the body, discarding the possible presence of some intrinsic quadrupole moment. This \textit{spin-induced quadrupole} was shown to take the form given in Eq. \eqref{quad_intro}, see \textit{e.g.} \cite{Porto:2005ac,Porto:2008jj,Steinhoff:2011sya,Steinhoff:2014kwa,Marsat_2015,Bohe:2015ana,Levi:2015msa,Bini:2015zya,Vines:2016unv}. Specialized to the case $v_\alpha v^\alpha=-1$ and at leading order it reduces to 
\begin{align}
    J^{\mu\nu\rho\sigma}=\frac{3\kappa}{\mu}v^{[\mu}S^{\nu]\lambda}S\tdu{\lambda}{[\rho}v^{\sigma]} = -\frac{3\kappa}{\mu} v^{[\mu}\Theta^{\nu][\rho}v^{\sigma]},\qquad\text{where }\Theta^{\alpha\beta}\triangleq S^{\alpha\lambda}S\tud{\beta}{\lambda}.\label{spin_induced_Q}
\end{align}
Here $\kappa$ is a free coupling parameter that equals 1 for a Kerr black hole and takes another value if the test-body is another compact object, \textit{e.g.} a neutron star.

\subsection{The spin vector}

Having imposed the Tulczyjew-Dixon condition, all the information contained in the spin-dipole tensor can be recast into a \textit{spin vector} $S^\mu$. Indeed, if one defines\footnote{The convention chosen here differs from the one of \cite{doi:10.1098/rspa.1981.0046,Batista:2020cto} by a global `$-$' sign, but agrees with our previous publications \cite{Compere:2021kjz,Compere:2023alp}.}
\begin{equation}
    S^\alpha\triangleq\frac{1}{2}\epsilon^{\alpha\beta\gamma\delta}\hat p_\beta S_{\gamma\delta} 
\end{equation}
where $\hat p^\mu\triangleq \frac{p^\mu}{\mu}$ (yielding $\hat p^2=-1$), one can invert the previous relation in order to rewrite $S^{\alpha\beta}$ in terms of $S^\alpha$:
\begin{equation}\label{defSmu}
    S^{\alpha\beta}=-\epsilon^{\alpha\beta\gamma\delta} \hat p_\gamma S_\delta. 
\end{equation}
This can be easily checked thanks to the identity \cite{Wald:1984rg}
\begin{equation}
    \epsilon^{\alpha_1\ldots\alpha_j\alpha_{j+1}\ldots\alpha_n}\epsilon_{\alpha_1\ldots\alpha_j\beta_{j+1}\ldots\beta_n}=-(n-j)!\,j!\,\delta^{[\alpha_{j+1}\ldots\alpha_n]}_{\beta_{j+1}\ldots\beta_n}\label{contraction}
\end{equation}
which is valid in any n-dimensional Lorentzian manifold. We make use of the shortcut notation $\delta^{\mu_1\ldots\mu_N}_{\nu_1\ldots\nu_N}\triangleq\delta^{\mu_1}_{\nu_1}\ldots\delta^{\mu_N}_{\nu_N}$. By definition, the spin vector is automatically orthogonal to the linear momentum:
\begin{align}
    p_\mu S^\mu=0.
\end{align}
Finally, also notice that the spin parameter is simply the squared norm of the spin vector, $\mathcal S^2=S^\alpha S_\alpha$ .

\subsection{Conservation of the spin, mass; relation between four-velocity and linear momentum}
Differentiating the SSC Eq. \eqref{ssc} yields
\begin{align}
    \mu^2 v^\mu-\mathfrak m p^\mu=\frac{1}{2}S^{\mu\nu}R_{\nu\lambda\rho\sigma}v^\lambda S^{\rho\sigma}-\mathcal L^{\mu\nu}p_\nu- S^{\mu\nu}\mathcal F_\nu.
\end{align}
Contracting this equation with $v_\mu$ provides us with
\begin{align}
    \mu^2=\mathfrak m^2+\mathcal O(\mathcal S^3).\label{masses}
\end{align}
The expression of the linear momentum in terms of the four-velocity reads 
\begin{align}
    p^\mu=\mu v^\mu-\frac{1}{2\mu}S^{\mu\nu}R_{\nu\lambda\rho\sigma}v^\lambda S^{\rho\sigma}+\mathcal L^{\mu\nu}v_\nu+\mathcal O(\mathcal S^3).
\label{p_of_v}
\end{align}
In the quadrupole approximation, the dynamical mass $\mu$ is no longer conserved at $\mathcal O(\mathcal S^2)$, since
\begin{align}
    \dv{\mu}{\tau}=-v_\mu\mathcal F^\mu+\mathcal O(\mathcal S^3).
\end{align}
However, notice that, provided we assume\footnote{This condition is automatically satisfied for the spin-induced quadrupole.}
\begin{align}
    \frac{\text{D}}{\dd\tau}J^{\alpha\beta\gamma\delta}=\mathcal O(\mathcal S^3),
\end{align}
one can still define a mass-like quantity, given by
\begin{align}
    \tilde\mu\triangleq \mu-\frac{1}{6}J^{\alpha\beta\gamma\delta}R_{\alpha\beta\gamma\delta},\label{mass_like}
\end{align}
which is quasi-conserved independently of the spin, namely
\begin{align}
    \dv{\tilde\mu}{\tau}=\mathcal O(\mathcal S^3).
\end{align}
Moreover, one can perturbatively invert \eqref{p_of_v} to obtain an expression of the four-velocity in terms of the linear momentum and the spin:
\begin{align}
    v^\mu = \hat p^\mu +(D^{\mu\nu}-\frac{1}{\mu}\mathcal L^{\mu\nu})\hat p_\nu+\mathcal O(\mathcal S^3),
    \label{v_of_p}
\end{align}
with
\begin{align}
    \hat p^\alpha\triangleq\frac{p^\alpha}{\mu}=\frac{p^\alpha}{\tilde\mu}+\mathcal O(\mathcal S^2),\qquad D\tud{\mu}{\nu}&\triangleq\frac{1}{2\mu^2}S^{\mu\lambda}R_{\lambda\nu\rho\sigma}S^{\rho\sigma}.
\end{align}
Eq. \eqref{v_of_p} will play a central role when we will work out the conservation equations in the following chapters. 

\subsection{Projector and Hodge dualities: useful relations}
\subsubsection{Projector on the hypersurface orthogonal to $p^\mu$} 
We introduce
\begin{equation}
    \Pi^\mu_\nu\triangleq\delta^\mu_\nu+\hat p^\mu \hat p_\nu, 
\end{equation}
the projector onto the hypersurface orthogonal to $p^\mu$. It can be directly checked from the definition that $\Pi^\mu_\nu$ satisfies the properties:
\begin{subequations}
    \begin{align}
    \text{(I)}&\quad \Pi^\mu_\nu\,\Pi^\nu_\rho=\Pi^\mu_\rho, \\
    \text{(II)}&\quad\Pi^\mu_\nu\, p^\nu=0, \\
    \text{(III)}&\quad\Pi^\mu_\alpha\,S^{\alpha\nu}=S^{\mu\nu}. 
\end{align}
\end{subequations}

\subsubsection{Hodge duality} 
Given any tensor $\mathbf A$, one can define the left and the right Hodge duals as
\begin{subequations}
    \begin{align}
    {}^*\!A_{\mu\nu\alpha_1\ldots\alpha_p}&\triangleq\frac{1}{2}\epsilon\tdu{\mu\nu}{\rho\sigma}A_{\rho \sigma\alpha_1\ldots\alpha_p}, \\
    A^*_{\alpha_1\ldots\alpha_p\mu\nu}&\triangleq\frac{1}{2}\epsilon\tdu{\mu\nu}{\rho\sigma}A_{\alpha_1\ldots\alpha_p\rho\sigma}. 
\end{align}
\end{subequations}
They correspond, respectively, to the Hodge dualization on the two first, resp. the two last, indices of $\mathbf A$. The definition of the bidual $\mathbf{ {}^*\! A^*}$ follows directly from the definitions above. Finally, given the product of two antisymmetrized vectors, one can similarly define
\begin{equation}
    l^{[\mu}m^{\nu]*}\triangleq\frac{1}{2}\epsilon^{\mu\nu\rho\sigma}l_\rho m_\sigma. 
\end{equation}
The dual tensors obey the following properties:
\begin{subequations}
    \begin{align}
    \text{(I)}&\quad A^*_{\mu\nu}={}^*\!A_{\mu\nu}, \\
    \text{(II)}&\quad A^{**}_{\mu\nu}={}^{**}\!A_{\mu\nu}=-A_{\mu\nu}, \\
    \text{(III)}&\quad{}^*\!A_{[\mu\nu]}B^{[\mu\nu]}=A_{[\mu\nu]}{}^*\!B^{[\mu\nu]}. 
\end{align}
\end{subequations}
Moreover, for any geometry we have
\begin{equation}
\mbox{}^*R\tud{\alpha}{\mu\alpha\nu}=\frac{1}{2}\epsilon\tudu{\alpha}{\beta}{\beta\gamma}R_{[\beta\gamma\alpha]\nu}=0.
\end{equation}
Notice also that Bianchi's identities $R_{\alpha\beta [\mu\nu ; \sigma]}=0$ can be equivalently written as
\begin{align}
R\tudud{*}{\alpha\beta\gamma}{\sigma}{;\sigma}= 0.   
\end{align}

\subsubsection{Identities for the MPTD equations} 
Let us turn to the more specific context of the MPTD theory. Using all the previous definitions, it is not complicated to check that the following properties hold:
\begin{subequations}
    \begin{align}
    \text{(I)}&\quad S^{\alpha\beta}=2S^{[\alpha}\hat p^{\beta]*},  \label{subsS}\\
    \text{(II)}&\quad \Pi^{\alpha[\beta}S^{\gamma]*}=S^{\alpha[\beta}\hat p^{\gamma]}. 
\end{align}
\end{subequations}
One can show that the first identity implies the following decomposition for $\Theta^{\alpha\beta}$:
\begin{align}
    \Theta^{\alpha\beta}&=\Pi^{\alpha\beta}\mathcal S^2-S^\alpha S^\beta.
\end{align}

Moreover, one has the identity
\begin{align}
    D\tud{\mu}{\alpha}p^\alpha&=\frac{1}{2\mu}S^{\mu[\nu}\hat p^{\alpha]}R_{\nu\alpha\rho\sigma}S^{\rho\sigma} \nonumber\\
    &=\frac{1}{\mu}\Pi^{\mu[\nu}S^{\alpha]*}R_{\nu\alpha\rho\sigma}S^{[\rho}\hat p^{\sigma]*} \nonumber\\
    &=\frac{1}{\mu}\Pi^{\mu\alpha}{}^*\!R^*_{\alpha\beta\gamma\delta}S^\beta S^\gamma \hat p^\delta,
\end{align}
as well as
\begin{subequations}
\begin{align}
    \mathcal F^\mu&=\frac{\kappa}{2\mu}\hat p^\alpha\Theta^{\beta\gamma}\hat p^\delta\nabla^\mu R_{\alpha\beta\gamma\delta}\label{F}+\mathcal O(\mathcal S^3) ,\\ 
    \mathcal L^{\mu\nu}&=\frac{2\kappa}{\mu}R\tud{\nu}{\alpha\beta\gamma}v^{[\mu}\Theta^{\alpha]\beta}v^\gamma-\qty(\mu\leftrightarrow\nu),\label{L}\\
    \mathcal L^{\mu\nu}v_\nu&=\frac{\kappa}{\mu}\qty(\hat p^\mu \hat p^\nu R_{\nu\alpha\beta\gamma}+R\tud{\mu}{\alpha\beta\gamma})\Theta^{\alpha\beta} \hat p^\gamma+\mathcal O(\mathcal S^4)\nonumber\\
    &\quad = \frac{\kappa}{\mu} \Pi^{\mu\nu} R_{\nu \alpha\beta\gamma}\Theta^{\alpha\beta}\hat p^\gamma +\mathcal O(\mathcal S^4).\label{Lv}
\end{align}
\end{subequations}

\subsection{Independent dynamical variables}

Let us now summarize the independent dynamical variables of the MPTD system. Under the SSC \eqref{ssc}, one can write
\begin{subequations}
\begin{align}
    \mu&=\mu(p^\alpha)=\sqrt{-p_\alpha p^\alpha}, \\
    \mathfrak m & = \mathfrak m (p^\alpha,S^\alpha), \\
    S^{\mu\nu}&=S^{\mu\nu}(p^\alpha,S^\alpha)=-\epsilon^{\mu\nu\alpha\beta}\hat p_\alpha S_\beta, \\
    v^\mu&=v^\mu(p^\alpha, S^\alpha).
\end{align}
\end{subequations}
The explicit expression for $v^\mu$ has been worked out in Eq. \eqref{v_of_p}. Consequently, the system can be fully described in terms of the dynamical variables $x^\mu$, $p^\mu$ and $S^\mu$. However, the four components of the spin vector $S^\mu$ are not independent, since they are subjected to the orthogonality condition $p^\mu S_\mu=0$. This fact leads to complications when we will seek to build invariants, as will be detailed in Section \ref{sec:invariants}. To overcome this difficulty, we introduce the \textit{relaxed spin vector} $s^\alpha$ from
\begin{equation}
    S^\alpha\triangleq\Pi^\alpha_\beta s^\beta\label{relaxed_spin} 
\end{equation}
where the part of $\mathbf s$ aligned with $\mathbf p$ is left arbitrary, but is assumed (without loss of generality) to be of the same order of magnitude. It ensures the relation $\mathcal O(s)=\mathcal O(\mathcal S)$ to hold (where $s^2\triangleq s_\alpha s^\alpha$). While working out the conservation constraints, one will often encounter the spin vector antisymmetrized with $\mathbf p$, in expressions of the type $p^{[\mu}S^{\nu]}$. In that case, we can directly write $p^{[\mu}S^{\nu]}=p^{[\mu}s^{\nu]}$, thereby replacing the (constrained) $S^\alpha$ by the (independent) variables $s^\alpha$. 

\subsection{MPTD equations at linear order in the spin}
As discussed in the introduction, the complexity of the constraint equations will motivate us to proceed perturbatively: one will first seek for solutions in the linearized theory, thus neglecting all $\mathcal O\qty(\mathcal S^2)$ terms. The solutions valid a second order in $\mathcal S$ will then be build as perturbations on the top of the previously found linear solutions. It is therefore relevant to understand how MPTD equations simplify at linear order in $\mathcal S$.

Neglecting all $\mathcal O(\mathcal S^2)$ terms, Eq. \eqref{p_of_v} leads to the usual relation between the linear momentum and the four-velocity,
\begin{align}
    p^\mu=\mu v^\mu.\label{pv}
\end{align}
Once linearized in the spin, the MPD equations \eqref{MPD} reduce to
\begin{subequations}\label{linearized_MPTD}
    \begin{align}
    \frac{\text D p^\mu}{\dd\tau}&=f^{(1)\mu}_{S}\triangleq-\frac{1}{2\mu}R\tud{\mu}{\nu\alpha\beta}p^\nu S^{\alpha\beta},\label{MPT_lin_1}\\
    \frac{\text D S^{\mu\nu}}{\dd\tau}&=0\qquad\Leftrightarrow\qquad\frac{\text D S^{\mu}}{\dd\tau}=0,\label{MPT_lin_2}
\end{align}
\end{subequations}
which are, respectively, the forced geodesic equation with force $f^{(1)\mu}_{S}=\mathcal O\qty(\mathcal S^1)$ and the parallel transport equation of the spin tensor/vector studied \textit{e.g.} in \cite{Ruangsri:2015cvg,vandeMeent:2019cam}.

\section{Rüdiger's procedure}\label{sec:invariants}

In two papers published in the early 80s, R. Rüdiger described a scheme for constructing quantities conserved along the motion driven by the MPTD equations \cite{doi:10.1098/rspa.1981.0046,doi:10.1098/rspa.1983.0012}. The basic guideline followed in his scheme was to enforce directly the conservation equation on a generic Ansatz for the conserved quantity, and to subsequently solve the constraints obtained. We summarize here this procedure in an abstract way, before applying it to simple concrete examples in the next section. 

\begin{itemize} 
    \item \textbf{Step 1: postulate an Ansatz for the conserved quantity.} The conserved quantity should be a function of the dynamical variables $p^\mu$ and $S^{\mu}$. It is therefore a function $\mathcal Q\qty(x^\mu,S^{\mu},p^\mu)$. Assuming its analyticity, it can be expanded as
\begin{align}
    \mathcal Q\qty(x^\mu,S^\alpha,p^\mu)=\sum_{\substack{s,p\geq0\\s+p>0}}\mathcal Q^{[s,p]}\qty(x^\mu,S^\alpha,p^\mu)\label{ansatz_all_orders}
\end{align}
with
\begin{align}
    \mathcal Q^{[s,p]}\qty(x^\mu,S^\alpha,p^\mu)
    \triangleq \mathcal Q^{[s,p]}_{\alpha_1\ldots\alpha_{s}\mu_1\ldots\mu_p}(x^\mu)S^{\alpha_1}\ldots S^{\alpha_{s}}p^{\mu_1}\ldots p^{\mu_p}.\label{ansatz_term}
\end{align}
Expressions like this one -- that is, tensorial quantities fully contracted with occurrences of the linear momentum and the spin -- will often appear in the following computations. It is useful to enable a distinction between them by introducing a grading allowing the counting of the number of occurrences of both the spin vector $S^{\mu}$ and the linear momentum $p^\mu$, which is provided by the notation $[s,p]$. More generally, we define:

\begin{definition}
A fully-contracted expression of the type
\begin{align}
    T_{\alpha_1\ldots\alpha_s\mu_1\ldots\mu_p}\ell_s^{\alpha_1}\ldots\ell_s^{\alpha_s}\ell_p^{\mu_1}\ldots\ell_p^{\mu_p}\label{fully_contracted}
\end{align}
where $\ell_s^\alpha=S^\alpha,s^\alpha$ (the relaxed spin vector $s^\mu$ will be defined below) and $\ell_p^\mu=p^\mu,\hat p^\mu$ is said to be of grading $[s,p]$. Equivalently, $s$ (resp. $p$) will be referred to as the spin (resp. momentum) grading of this expression.
\end{definition}

Since we have only included the quadrupole term in the equations of motion but neglected all the $\mathcal{O}\qty(\mathcal S^3)$ terms, it is not self-consistent to look at quantities which are conserved beyond second order in the spin magnitude. We therefore restrict our analysis to Ansätze that contain terms of of spin grading at most equal to two. Historically, Rüdiger didn't consider the full set of possible Ansätze originating from this discussion, but only the two restricted cases
\begin{subequations}\label{ansatze}
    \begin{align}
    \mathcal Q^{(1)}&\triangleq\sum_{p=1}Q^{[s,p]}\triangleq X_\mu p^\mu+W_{\mu\nu}S^{\mu\nu},\label{linear_ansatz}\\
    \mathcal Q^{(2)}&\triangleq\sum_{p=2}Q^{[s,p]}\triangleq K_{\mu\nu}p^\mu p^\nu+L_{\mu\nu\rho}S^{\mu\nu}p^\rho+M_{\mu\nu\rho\sigma}S^{\mu\nu}S^{\rho\sigma}. \label{quadratic_ansatz}
\end{align}
\end{subequations}
We will refer to them as respectively the \textit{linear} and the \textit{quadratic} invariants in $p^\mu$. They are homogeneous in the number of occurrences of $p^\mu$ and $S^{\mu\nu}$ they contain. As long as we consider the MPD equations at linear order in the spin magnitude or at quadratic order with the quadrupole coupling of the test body being the one of a black hole ($\kappa=1$), it turns out that considering only these two types of ansatzes will be enough to derive a complete set of conserved quantities. However, a more general ansatz will be necessary to consider arbitrary quadrupole couplings ($\kappa\neq 1$), as discussed in Section \ref{sec:NS}.

\item \textbf{Step 2: write down the conservation equation.} 
Because we always consider the MPTD equations truncated up to some order in $\mathcal S$, it is consistent to require our quantities to be only \defining{quasi-conserved}. For the MPDT equations including terms up to $\mathcal O\qty(\mathcal S^n)$ included, we only require the conservation to hold up at order $n+1$ in the spin magnitude:
\begin{align}
    \dot{\mathcal{Q}}\triangleq v^\lambda\nabla_\lambda \mathcal Q\stackrel{!}{=}\mathcal O\qty(\mathcal S^{n+1}).
\end{align}

\item \textbf{Step 3: expand the conservation equation using the equations of motion.}
The next step is to plug the explicit form of the Ansatz chosen in the conservation equation, and to use the MPD equations \eqref{MPD} to replace the covariant derivatives of the linear momentum and of the spin tensor. The occurrences of the four-velocity are replaced by the means of Eq. \eqref{v_of_p}.

\item \textbf{Step 4: express the conservation equation in terms of independent variables.}
As already discussed above, the presence of the SSC make the variables $p^\mu$, $S^{\alpha\beta}$ not independent among themselves. We turn to an independent set of variables in two steps: (i) we use the relation $S^{\alpha\beta}=2S^{[\alpha}\hat p^{\beta]*}$ to replace all the spin tensors $S^{\mu\nu}$ by the spin vectors $S^{\mu}$ and (ii) we replace the occurrences of the spin vector by the \textit{relaxed spin vector} $s^\alpha$ defined through
\begin{align}
    S^\alpha=\Pi^\alpha_\beta s^\beta.
\end{align}
It allows to relax the residual constraint $S_\mu p^\mu=0$ by considering a spin vector possessing a non-vanishing component along the direction of the linear momentum. Physical quantities will be independent of this component. It is introduced in order to decouple the conservation equation. For convenience, we scale the unphysical component of the relaxed spin vector such that $s_\alpha s^\alpha\sim S_\alpha S^\alpha=\mathcal S^2$. Notice that we have the useful identity
\begin{align}
    S^{[\alpha}p^{\beta]}=s^{[\alpha}p^{\beta]}\qquad\Rightarrow\qquad S^{\alpha\beta}=2s^{[\alpha}\hat p^{\beta]*}.
\end{align}

\item \textbf{Step 5: infer the independent constraints.} The conservation equation takes now the form of a sum of fully-contracted expressions of the type \eqref{fully_contracted}, involving only the \textit{independent} dynamical variables $p^\mu$ and $s^\alpha$:
\begin{align}
    \dot{\mathcal{Q}}=\sum_{\substack{s,p\geq 0\\s+p>0}}T^{[s,p]}_{\alpha_1\ldots\alpha_s\mu_1\ldots\mu_p}s^{\alpha_1}\ldots s^{\alpha_s} p^{\mu_1}\ldots p^{\mu_p}\stackrel{!}{=}\mathcal O\qty(\mathcal S^{n+1}).
\end{align}
The conservation equation is then equivalent to the requirement that all the terms of different gradings $[s,p]$ vanish independently:
\begin{align}
    T^{[s,p]}_{\alpha_1\ldots\alpha_s\mu_1\ldots\mu_p}s^{\alpha_1}\ldots s^{\alpha_s} p^{\mu_1}\ldots p^{\mu_p}\stackrel{!}{=}\mathcal O\qty(\mathcal S^{n+1}).
\end{align}
$s^\alpha$ and $p^\mu$ being arbitrary, this is equivalent to the \emph{constraint equations}
\begin{align}
    T^{[s,p]}_{(\alpha_1\ldots\alpha_s)(\mu_1\ldots\mu_p)}\stackrel{!}{=}\mathcal O\qty(\mathcal S^{n+1}).\label{generic_tensor_cst}
\end{align}

\item \textbf{Step 6: find a solution and prove uniqueness.} This final step is non-systematic. For the simplest cases (linear invariant with black hole quadrupole coupling, quadratic invariant at first order in the spin magnitude), it will be sufficient to work only with the tensorial constraints \eqref{generic_tensor_cst}. However, for more involved cases (\textit{e.g.} quadratic invariant at second order in $\mathcal S$), the tensorial relations will become so cumbersome that turning to another formulation of the problem will appear to be fruitful. This will be the purpose of the covariant building blocks for Kerr introduced in Section \ref{sec:CBB}.
\end{itemize}

\section{The simplest examples}
In this final section, we apply Rüdiger's procedure to the two simplest examples: generic polynomial invariants for geodesic motion and linear invariants for linearized MPTD equation.

\subsection{Generic polynomial invariants for geodesic motion }
In the geodesic case, the linear momentum is tangent to the worldline, $p^\mu=\mu v^\mu$ and geodesic equations are simply the $\mathcal O\qty(\mathcal S^0)$ MPTD equations:
\begin{equation}
    \frac{D p^\mu}{\dd\tau}=0,\qquad\frac{D}{\dd\tau}\triangleq v^\alpha\nabla_\alpha.\label{geodesic_equation}
\end{equation}

In most GR textbooks and lectures, the problem is tackled from the perspective ``\textit{symmetry implies conservation}'': one first introduces the notion of Killing vector fields ($\nabla_{(\alpha}\xi_{\beta)}=0$) and \textit{then} prove the well-known property stating that, given any geodesic of linear momentum $p^\mu$ and a Killing vector field $\xi^\mu$ of the background spacetime, the quantity $C_\xi\triangleq \xi_\alpha p^\alpha$ is constant along the geodesic. One also shows that this property generalizes in the presence of a Killing tensor, and that the invariant mass $\mu^2=-p_\alpha p^\alpha$ is also constant along the geodesic trajectory. Along a geodesic, the conservation equation for a quantity $C(p^\alpha)$ takes the form
\begin{equation}
    \dot{{C}}(p^\alpha)=0\quad\Leftrightarrow\quad p^\mu\nabla_\mu  C(p^\alpha)=0.\label{conservation_equation}
\end{equation}

In applying Rüdiger algorithm, we will tackle the problem in the opposite way (``\textit{conservation requires symmetry}''): given an arbitrary geodesic, is it possible to construct invariants of motion that are polynomial quantities of the linear momentum, \textit{i.e.} that are composed of monomials of the form
\begin{equation}
 C_\mathbf{K}^{(n)}\triangleq K_{\alpha_1\ldots\alpha_n}p^{\alpha_1}\ldots p^{\alpha_n}\label{geodesic_invariant}
\end{equation}
where, at this point, $\mathbf K$ is an arbitrary, by definition totally symmetric tensor of rank $n$? We will show that requiring the conservation of $ C_\mathbf{K}^{(n)}$ will require either $\mathbf K$ to be a Killing vector/tensor or either that $C_\mathbf{K}^{(2)}$ is the invariant mass. 

In order to work out the most general constraint on $\mathbf K$, we plug the definition of $C_\mathbf{K}$ \eqref{geodesic_invariant} into the conservation equation \eqref{conservation_equation}. Using the geodesic equation \eqref{geodesic_equation} and relabelling the indices, one  gets
\begin{equation}
p^\mu\nabla_\mu K_{\alpha_1\ldots\alpha_n} p^{\alpha_1}\ldots p^{\alpha_n}=0.
\end{equation}
The crucial point is that the dynamical variables $p^\alpha$ are \textit{independent} among themselves. The above relation must hold for any values of the independent $p^\alpha$, yielding the general constraint
\begin{equation}
    \nabla_{(\mu}K_{\alpha_1\ldots \alpha_n)}=0.\label{genK}
\end{equation}
The only possible cases for solving this constraint are the following:
\begin{itemize}
    \item for $n=1$, $K_\mu$ must be a Killing vector, $\nabla_{(\mu}K_{\nu)}=0$;
    \item for $n=2$, either $K_{\mu\nu}$ must be a rank-2 Killing tensor ($\nabla_{(\mu}K_{\nu\rho)}=0$), either one takes $K_{\mu\nu}=g_{\mu\nu}$ which leads to the conservation of the invariant mass, $ C_\mathbf{g}^{(2)}=-\mu^2$;
    \item for any $n\geq3$, $\mathbf K$ must be a rank-n Killing tensor.
\end{itemize}
Before turning to the spinning body case, let us make a couple of remarks:
\begin{enumerate}
    \item As stated above, the viewpoint adopted here is reversed with respect to the `traditional' one: we have proven that the existence of conserved quantities along geodesic trajectories that are \textit{polynomial} in the linear momentum require the existence of symmetries of the background spacetime (except for the invariant mass $\mu$ which is always conserved).
    \item Any linear combination of the invariants defined above remains of course invariant. Nevertheless, the conservation can be checked separately at each order in $\mathbf p$ because the application of the conservation condition \eqref{conservation_equation} doesn't change the order in $\mathbf p$ of the terms contained in the resulting expression.
    \item The invariant related to a Killing tensor is  relevant only if the latter is irreducible, \textit{i.e.} if it cannot be written as the product of Killing vectors. Otherwise, the invariant at order $n$ in $\mathbf p$ is just a product of invariants of lower order. 
\end{enumerate}

\subsection{Linear invariants for the linearized MPTD equations}

We now turn to an example which exhibits a non-trivial dependence in the spin. We will seek for a quantity conserved for the linearized MPTD equations \eqref{linearized_MPTD} of the form given in Eq. \eqref{linear_ansatz}, that is
\begin{align}
    \mathcal Q^{(1)}\triangleq X_\mu p^\mu + W_{\mu\nu}S^{\mu\nu}.
\end{align}
Notice that, by construction, $W_{\mu\nu}$ shall be an antisymmetric tensor. Using Eqs. \eqref{linearized_MPTD}, the conservation equation $\dot{\mathcal{Q}}^{(1)}=\mathcal{O}\qty(\mathcal S^2)$ becomes
\begin{align}
    v^\lambda\nabla_\lambda X_\mu p^\mu+v^\lambda\nabla_\lambda W_{\mu\nu}S^{\mu\nu}-\frac{1}{2}X_{\mu}R\tud{\mu}{\nu\alpha\beta}v^\nu S^{\alpha\beta}+2W_{\mu\nu}p^\mu v^\nu=\os{2}.
\end{align}
Expressing this equation in terms of the independent variables $\hat p^\mu$ and $s^\alpha$ yields
\begin{align}
    \mu\nabla_\mu X_\nu\hat p^\mu\hat p^\nu-\qty(2\nabla_\mu W^*_{\nu\alpha}-X_\lambda R\tud{*\lambda}{\mu\nu\alpha})s^\alpha\hat p^\mu\hat p^\nu=\os{2}.
\end{align}
It is equivalent to the set of two independent constraints
\begin{subequations}\label{test_cst}
    \begin{align}
    [0,2]:&\quad \nabla_{(\mu}X_{\nu)}=\os{2},\label{test_cst1}\\
    [1,2]:&\quad 2\nabla_{(\mu}W^*_{\nu)\alpha}-X^\lambda R_{*\lambda(\mu\nu)\alpha}=\os{2}.\label{test_cst2}
\end{align}
\end{subequations}
The development above can be summarized as follows:
\begin{proposition}
For any pair $(X_\mu,W_{\mu\nu})$ satisfying the constraint equations \eqref{test_cst} and assuming the linearized MPTD equations \eqref{linearized_MPTD} are obeyed, the quantity $\mathcal Q^{(1)}$ given in Eq. \eqref{linear_ansatz} will be conserved up to first order in the spin parameter, \textit{i.e.} $\dot{\mathcal{Q}}^{(1)}=\mathcal O\qty(\mathcal S^2)$.
\end{proposition}

The set of equations \eqref{test_cst} admits two independent solutions:
\begin{itemize}
    \item For $X^\mu\neq 0$, Eq. \eqref{test_cst1} is satisfied provided that $X^\mu$ is a Killing vector. Recall that for any Killing vector $X^\mu$, one has the Kostant formula \cite{carroll2003spacetime}
\begin{align}
    \nabla_\mu\nabla_\nu X_\alpha=X_\lambda R\tud{\lambda}{\mu\nu\alpha}
\end{align}
and the second constraint Eq. \eqref{test_cst2} will admit a solution if the stronger constraint
\begin{align}
    2\nabla_{\mu}\qty(W_{\nu\alpha}-\frac{1}{2}\nabla_{\nu}X_\alpha)=\os{2}
\end{align}
holds. It is naturally solved by
\begin{align}
    W_{\nu\alpha}=\frac{1}{2}\nabla_\nu X_\alpha.
\end{align}
At the end of the day, we have shown that the quantity
\boxedeqn{
    \mathcal C_X=X_\mu p^\mu+\frac{1}{2}\nabla_\mu X_\nu S^{\mu\nu}\label{killing_conserved_spin}
}{Conserved quantity for Killing vectors}

\noindent
is conserved up to $\os{2}$ corrections along the motion generated by the linearized MPTD equations \eqref{linearized_MPTD}. This is the generalization to the case of spinning objects of the quantity conserved for Killing vector along geodesics. Actually, one can show that this quantity is exactly conserved at any order of the multipole expansion, regardless to the nature of the multipoles considered \cite{dixon1979}.

\item The constraints equations \eqref{test_cst} admit another independent solution, which corresponds to $X^\mu=0$. In that case, the second equation \eqref{test_cst2} reduces to
\begin{align}
    \nabla_{(\mu}W^*_{\nu)\alpha}=\os{2}.
\end{align}
It is satisfied provided that $W^*_{\mu\nu}$ is a Killing-Yano tensor. Denoting $Y_{\mu\nu}=W^*_{\mu\nu}$, the corresponding conserved quantity is
\boxedeqn{
    \mathcal Q_Y=Y^*_{\mu\nu}S^{\mu\nu}.\label{linear_rudiger_invariant}
}{Rüdiger linear invariant}
This is \defining{Rüdiger linear invariant}, which was first unravelled by R. Rüdiger in 1981 \cite{doi:10.1098/rspa.1981.0046}.
\end{itemize}

\chapter[{Solving the constraints: first order in the spin}]{Solving the constraints:\\first order in the spin}
\label{chap:first_order_solutions}
\chaptermark{\textsc{First order in the spin}}

\vspace{\stretch{1}}

In this chapter, we will be concerned with the construction of quadratic invariants for the linearized MPTD equations \eqref{linearized_MPTD}. 
It is thus consistent to consider a quadratic invariant that is at most linear in $\mathcal S$:\begin{equation}
    \mathcal Q^{(2)}\triangleq K_{\mu\nu}p^\mu p^\nu+L_{\mu\nu\rho}S^{\mu\nu}p^\rho.\label{quadratic_invariant}
\end{equation}
The tensors $K_{\mu\nu}$ and $L_{\mu\nu\rho}$, which are by definition independent of $p^\mu$ and $S^{\mu\nu}$, satisfy the algebraic symmetries $K_{\mu\nu}=K_{(\mu\nu)}$ and $L_{\mu\nu\rho}=L_{[\mu\nu]\rho}$.

In Ricci-flat spacetimes that admit a Killing-Yano tensor on their background, we will show that the conservation equation Eq. \eqref{quadratic_invariant} admits a solution which is a generalization of the usual Carter constant, already found by R. Rüdiger in 1983 \cite{doi:10.1098/rspa.1983.0012}. The uniqueness of this construction will be proven in Kerr spacetime.

This chapter is organized as follows. In Section \ref{sec:cst_lin}, we explicit the set of constraints that must be fulfilled for such an invariant at most linear in the spin to exist. Conservation at linear order requires to solve only two constraints. We simplify the second, most difficult, constraint in Section \ref{sec:C2}. 
We subsequently particularize our setup in Section \ref{sec:KY} to spacetimes admitting a Killing-Yano (KY) tensor. After deriving some general properties of KY tensors, we will prove a cornerstone result for the continuation of our work, which we will refer to as the \textit{central identity}. Building on all previous sections, we will solve the aforementioned constraint for Ricci-flat (vacuum) spacetimes possessing a KY tensor in Section \ref{sec:unicity}. This will enable us to study in full generality the quasi-invariants for the MPTD equations that are quadratic in the combination of spin and momentum. On the one hand, we recover Rüdiger's results \cite{doi:10.1098/rspa.1981.0046,doi:10.1098/rspa.1983.0012}. On the other hand, we prove that the existence of any further quasi-invariant, which is then necessarily homogeneously linear in the spin, reduces to the existence of a non-trivial mixed-symmetry Killing tensor on the background. The significance of this result is examined for spinning test bodies in Kerr spacetime in the final Section \ref{sec:Kerr}. We show that a  stationary and axisymmetric non-trivial mixed-symmetry Killing tensor does not exist on the Kerr geometry. Consequently, an additional independent quasi-constant of motion for the linearized MPTD equations does not exist. As will be detailed later on in Chapter \ref{chap:integrability}, the linearized MPTD integrals of motion are not in involution, which implies that the system is not integrable in the sense of Liouville.

\section{Constraints for a quadratic invariant linear in \texorpdfstring{$\mathcal S$}{}: pole-dipole MPTD equations, to quadratic order in the spin}\label{sec:cst_lin}

Despite we will hereafter focus on the linearized MPTD equations, we will derive the generic constraints equation for the quantity Eq. \eqref{quadratic_invariant} to be conserved for the full pole-dipole MPTD equations, up to second order in the spin. Even the pole-dipole equations (that is, the MPTD equations with the force and torque terms being set to zero) only make sense at linear order in the spin $\mathcal S$, it is computationally relevant to derive here the constraint they generate at $\os{2}$. This will allow us (in the next chapter) to disentangle the terms originating from the pole-dipole sector of the equations from the ones originating from the quadrupole sector. The main advantage taken will be to significantly reduce the computational load that one will encounter.

\subsubsection{Useful identity for the four-velocity}
Let us start by deriving a useful identity. It can be shown \cite{1977GReGr...8..197E,Suzuki:1997tg,Batista:2020cto} that the following relation holds exactly:
\begin{align}
    v^\mu=\frac{\mathfrak m}{\mu^2}\qty(p^\mu+\frac{D\tud{\mu}{\alpha}p^\alpha}{1-\frac{d}{2}}),\qquad d&\triangleq D\tud{\alpha}{\alpha}.\label{velocity}
\end{align}
Moreover, one can show that\begin{equation}
    d=-\frac{2}{\mu^2}{}^*\!R^*_{\alpha\beta\gamma\delta} S^\alpha \hat p^\beta S^\gamma \hat p^\delta. 
\end{equation}
Putting all the pieces together, one can rewrite Eq. \eqref{velocity} as
\begin{align}
\begin{split}
    \frac{\mu^2}{\mathfrak m}\qty(1-\frac{d}{2})v^\mu&=\qty(1-\frac{d}{2})p^\mu+D\tud{\mu}{\alpha}p^\alpha \\
    &=p^\mu-\frac{1}{\mu}\qty(\hat p^\mu\hat p^\alpha-\Pi^{\mu\alpha}){}^*\!R^*_{\alpha\beta\gamma\delta}S^\beta S^\gamma \hat p^\delta\\
    &=p^\mu+\frac{1}{\mu}g^{\mu\alpha}{}^*\!R^*_{\alpha\beta\gamma\delta}S^\beta S^\gamma \hat p^\delta \\
    &=p^\mu+\frac{1}{\mu}{}^*\!{R^*}\tud{\mu}{\beta\gamma\delta}S^\beta S^\gamma \hat p^\delta.\label{velocity2}
\end{split}
\end{align}
This relation will be a fundamental building block of the forthcoming computations.
 
\subsubsection{Reduction of the conservation equation} 
Using the MPTD equations \eqref{MPD} and the expression \eqref{velocity2} for the four-velocity, the conservation equation for the quantity \eqref{quadratic_invariant} can be written
\begin{align}
\begin{split}
    \dot{\mathcal Q}^{(2)}
    &=\Xi\qty(p^\lambda+\frac{1}{\mu}\sRs\tud{\lambda}{\kappa\theta\sigma}S^\kappa S^\theta \hat p^\sigma)
    \bigg[\qty(\nabla_\lambda K_{\mu\nu}-2L_{\lambda\mu\nu})p^\mu p^\nu \\
    &\quad+ \qty(\nabla_\lambda L_{\alpha\beta\mu}-K_{\mu\rho}R\tud{\rho}{\lambda\alpha\beta})S^{\alpha\beta}p^\mu-\frac{1}{2}L_{\alpha\beta\rho}R\tud{\rho}{\lambda\gamma\delta}S^{\alpha\beta}S^{\gamma\delta}\bigg]\overset{!}{=}\os{3}.
\end{split}
\end{align}
Here, the coefficient $\Xi\triangleq\frac{\mathfrak m^2}{\mu^2\qty(1-d/2)}$ is non-vanishing. Let us introduce the tensors
\begin{subequations}
\begin{align}
    U_{\alpha\beta\gamma}&\triangleq\nabla_\gamma K_{\alpha\beta}-2L_{\gamma(\alpha\beta)},\\
    V_{\alpha\beta\gamma\delta}&\triangleq\nabla_\delta L_{\alpha\beta\gamma}-K_{\lambda\gamma}R\tud{\lambda}{\delta\alpha\beta}+\frac{2}{3}K_{\lambda\rho}{R\tud{\lambda}{\delta[\alpha}}^\rho g_{\beta]\gamma},\\
    W_{\alpha\beta\gamma\delta\epsilon}&\triangleq-\frac{1}{2}L_{\alpha\beta\lambda}R\tud{\lambda}{\gamma\delta\epsilon}.
\end{align}
\end{subequations}
By construction, they obey the algebraic symmetries $U_{\alpha\beta\gamma}=U_{(\alpha\beta)\gamma}$, $V_{\alpha\beta\gamma\delta}=V_{[\alpha\beta]\gamma\delta}$, $W_{\alpha\beta\gamma\delta\epsilon}=W_{[\alpha\beta]\gamma\delta\epsilon}$. Notice that the orthogonality conditions $S^{\alpha\beta}p_\beta=0$ imply
\begin{align}
    \frac{2}{3}K_{\lambda\rho}{R\tud{\lambda}{\delta[\alpha}}^\rho g_{\beta]\gamma}S^{\alpha\beta}p^\gamma=0.
\end{align}
With these notations, he conservation equation simplifies to
\begin{align}
\begin{split}
    \dot{\mathcal Q}^{(2)}&={\Xi}\qty(p^\lambda+\frac{1}{\mu}\sRs\tud{\lambda}{\kappa\theta\sigma}S^\kappa S^\theta \hat p^\sigma)\\
    &\quad\times\qty[U_{\mu\nu\lambda}p^\mu p^\nu+V_{\alpha\beta\mu\lambda}S^{\alpha\beta}p^\mu+W_{\alpha\beta\lambda\gamma\delta}S^{\alpha\beta}S^{\gamma\delta}]\overset{!}{=}\os{3}.
\end{split}
\end{align}

We will now go through a number of steps in order to express this condition in terms of the independent variables $s_\alpha$ and $\hat p^\mu$.
First, let us expand all terms and express the spin-related quantities in terms of the independent variables $s^\alpha$. For this purpose, we will make use of the aforementioned identities
\begin{align}
    p^{[\alpha}S^{\beta]}&=p^{[\alpha}s^{\beta]},\qquad S^{\alpha\beta}=2S^{[\alpha}\hat p^{\beta]*}=2s^{[\alpha}\hat p^{\beta]*},\qquad S^\alpha=\Pi^\alpha_\beta s^\beta.\label{id}
\end{align}
The conservation equation becomes
\begin{align}
    \begin{split}
    \dot{\mathcal Q}^{(2)}&={\frac{\Xi}{\mu}}\Bigg[\mu^4 U_{\mu\nu\rho}\hat p^\mu \hat p^\nu \hat p^\rho +2 \mu^3\sV_{\alpha\mu\nu\rho}s^\alpha \hat p^\mu \hat p^\nu \hat p^\rho\\
    &\quad +\mu^2\qty(4\sWs_{\alpha\mu\nu\beta\rho}+\sRs\tud{\lambda}{\kappa\alpha\rho}U_{\mu\nu\lambda}\Pi^\kappa_\beta) s^\alpha s^\beta \hat p^\mu \hat p^\nu \hat p^\rho\\
    &\quad+ 2\mu \sRs\tud{\lambda}{\kappa\alpha\rho}\sV_{\gamma\mu\nu\lambda}\Pi^\kappa_\beta s^\alpha s^\beta s^\gamma \hat p^\mu \hat p^\nu \hat p^\rho \\
    &\quad+4 \sRs\tud{\lambda}{\kappa\alpha\rho} \sWs_{\delta\mu\lambda\gamma\nu}\Pi^\kappa_\beta s^\alpha s^\beta s^\gamma s^\delta \hat p^\mu \hat p^\nu \hat p^\rho\Bigg]\overset{!}{=}\os{3}.
    \end{split}
\end{align}
Second, we will remove the projectors. One has the identity
\begin{align}
    \sRs\tud{\lambda}{\kappa\alpha\rho}\Pi_\beta^\kappa s^\alpha s^\beta=- I\tud{\lambda\alpha\beta}{\rho\sigma\kappa}s_\alpha s_\beta \hat p^\sigma \hat p^\kappa
\end{align}
where we have defined
\begin{align}
    I\tud{\lambda\alpha\beta}{\rho\sigma\kappa}\triangleq {\sRs\tud{\lambda}{\kappa\rho}}^\alpha\delta_\sigma^\beta+\sRs\tud{\lambda\alpha\beta}{\rho}g_{\sigma\kappa}.
\end{align}
The proof is easily carried out, using the fact that $\hat p_\mu \hat p^\mu=-1$: 
\begin{align}
     \begin{split}
     \sRs\tud{\lambda}{\kappa\alpha\rho}\Pi_\beta^\kappa s^\alpha s^\beta&= \sRs\tud{\lambda}{\kappa\alpha\rho}\qty(\delta^\kappa_\beta+\hat p^\kappa \hat p_\beta) s^\alpha s^\beta\\
     &=\qty[\qty(-\hat p^\sigma \hat p^\kappa g_{\sigma\kappa}) \sRs\tud{\lambda}{\beta\alpha\rho}+\sRs\tud{\lambda}{\kappa\alpha\rho}\hat p^\kappa \delta_\beta^\sigma\hat p_\sigma]s^\alpha s^\beta\\
     &=-\qty(\sRs{\tud{\lambda}{\kappa\rho}}^\alpha\delta^\beta_\sigma+\sRs\tud{\lambda\alpha\beta}{\rho}g_{\sigma\kappa})s_\alpha s_\beta \hat p^\sigma \hat p^\kappa.
     \end{split}
\end{align}
Using this identity, the conservation equation finally reads as
\begin{align}
   \begin{split}
    \dot{\mathcal Q}^{(2)}&={\frac{\Xi}{\mu}}\Bigg[\mu^4 U_{\mu\nu\rho}\hat p^\mu \hat p^\nu \hat p^\rho+2\mu^3 \sV\tud{\alpha}{\mu\nu\rho}s_\alpha\hat p^\mu \hat p^\nu \hat p^\rho\\
    &\quad-\mu^2\qty(I\tud{\lambda\alpha\beta}{\rho\sigma\kappa}U_{\mu\nu\lambda}+4\sWs\tudud{\alpha}{\mu\nu}{\beta}{\rho}\, g_{\sigma\kappa})s_\alpha s_\beta \hat p^\mu \hat p^\nu \hat p^\rho \hat p^\sigma \hat p^\kappa\\
    &\quad-2\mu I\tud{\lambda\alpha\beta}{\rho\sigma\kappa}\sV\tud{\gamma}{\mu\nu\lambda}s_\alpha s_\beta s_\gamma \hat p^\mu \hat p^\nu \hat p^\rho \hat p^\sigma \hat p^\kappa\\
    &\quad-4 I\tud{\lambda\alpha\beta}{\rho\sigma\kappa}\sWs\tudud{\gamma}{\mu\lambda}{\delta}{\nu}\,s_\alpha s_\beta s_\gamma s_\delta \hat p^\mu \hat p^\nu \hat p^\rho \hat p^\sigma \hat p^\kappa\Bigg] \overset{!}{=}\os{3}.\label{conservation_eqn}
   \end{split}
\end{align}
 
Because the variables $\hat p^\mu$ and $s_\alpha$ are independent, Eq. \eqref{conservation_eqn} is equivalent to the following set of three constraints, each of them arising at a different order in the spin parameter:
\begin{subequations}
\begin{align}
    [0,3]:&\quad U_{(\mu\nu\rho)}=\os{3},\label{C1}\\
    [1,3]:&\quad\sV\tud{\alpha}{(\mu\nu\rho)}=\os{3},\label{C2}\\
    [2,5]:&\quad I\tud{\lambda(\alpha\beta)}{(\mu\nu\rho}U_{\sigma\kappa)\lambda}+4\sWs\tudud{(\alpha}{(\mu\nu}{\beta)}{\rho}\,g_{\sigma\kappa)}=\os{3},\label{C3}
\end{align}
\end{subequations}

Notice that the $[0,3]$ constraint \eqref{C1} simply reduces to
\begin{equation}
    \nabla_{(\alpha}K_{\beta\gamma)}=0,
\end{equation}
\textit{i.e.} $K_{\mu\nu}$ must be a Killing tensor of the background spacetime.

The $[1,3]$ constraint \eqref{C2} is more difficult to work out. In Section \ref{sec:C2}, we will proceed to a clever rewriting of this constraint, which will then be particularized to spacetimes admitting a Killing-Yano tensor in Section \ref{sec:KY}. Section \ref{sec:unicity} will aim to solve it generally. Finally, all these results will be particularized to a Kerr background in Section \ref{sec:Kerr}. The $[2,5]$ constraint Eq. \eqref{C3} will be a fundamental building block for computations of Chapter \ref{chap:second_order_solution}.

\section{Conservation equation at linear order in the spin}\label{sec:C2}
We will now proceed to the aforementioned rewriting of the constraint \eqref{C2} by introducing a new set of variables. The three first parts of this section are devoted to the derivation of preliminary results, that will be crucial for working out the main result.

\subsection{Dual form of $\mathbf V$}
We want to compute the dual form of the tensor
\begin{align}
    V_{\alpha\beta\gamma\delta}&\triangleq\nabla_\delta L_{\alpha\beta\gamma}-K_{\lambda\gamma}R\tud{\lambda}{\delta\alpha\beta}+\frac{2}{3}K_{\lambda\rho}{R\tud{\lambda}{\delta[\alpha}}^\rho g_{\beta]\gamma}
\end{align}
with respect to its two first indices. One has
\begin{align}
    \sV_{\alpha\beta\gamma\delta}=\nabla_\delta \sL_{\alpha\beta\gamma}-K_{\lambda\gamma} {R^*}\tud{\lambda}{\delta\alpha\beta}+\frac{2}{3}K_{\lambda\rho}R\tudu{\lambda}{\delta[\alpha}{\rho}g_{\beta]*\gamma}.
\end{align}
The last term of this equality can be written as
\begin{align}
    \begin{split}
        \frac{2}{3}K_{\lambda\rho}R\tudu{\lambda}{\delta[\alpha}{\rho}g_{\beta]*\gamma}&=\frac{1}{3}K^{\lambda\rho}\epsilon\tdu{\alpha\beta}{\mu\nu}R_{\lambda\delta\mu\rho}g_{\nu\gamma}
    =\frac{1}{3}K^{\lambda\rho}\epsilon\tdud{\alpha\beta}{\mu}{\gamma}R_{\lambda\delta\mu\rho}\\
    &=\frac{1}{3}K^{\lambda[\mu}\epsilon\tdud{\alpha\beta}{\nu]}{\gamma}R^{**}_{\lambda\delta\mu\nu}
    =\frac{1}{3}K^{\lambda[\mu}\epsilon\tdud{\alpha\beta}{\nu]*}{\gamma}R^*_{\lambda\delta\mu\nu}\\
    &=-\frac{1}{6}\epsilon^{\sigma\mu\nu\rho}\epsilon_{\sigma\alpha\beta\gamma}K_{\lambda\rho}R\tud{*\lambda}{\delta\mu\nu}\\
    &=\frac{1}{3}\qty(K_{\lambda\gamma}R\tud{*\lambda}{\delta\alpha\beta}+K_{\lambda\alpha}R\tud{*\lambda}{\delta\beta\gamma}+K_{\lambda\beta}R\tud{*\lambda}{\delta\gamma\alpha})\\
    &=\frac{1}{3}K_{\gamma\lambda}R\tud{*\lambda}{\delta\alpha\beta}+\frac{2}{3}R\tud{*\lambda}{\delta\gamma[\alpha}K_{\beta]\lambda}.
    \end{split}
\end{align}
This finally yields
\begin{align}
    \sV_{\alpha\beta\gamma\delta}=\nabla_\delta\sL_{\alpha\beta\gamma}-\frac{2}{3}K_{\lambda\gamma}R\tud{*\lambda}{\delta\alpha\beta}+\frac{2}{3}R\tud{*\lambda}{\delta\gamma[\alpha}K_{\beta]\lambda}.
\end{align}

\subsection{Rüdiger variables}
 
Following Rüdiger \cite{doi:10.1098/rspa.1983.0012}, let us introduce
\begin{align}
    \tilde X_{\alpha\beta\gamma}\triangleq L_{\alpha\beta\gamma}-\frac{1}{3}\qty(\lambda_{\alpha\beta\gamma}+g_{\gamma[\alpha}\nabla_{\beta]}K),\label{tildeX}
\end{align}
where we have made use of the notations
\begin{align}
    \lambda_{\alpha\beta\gamma}\triangleq 2\nabla_{[\alpha}K_{\beta]\gamma},\qquad K\triangleq K^\alpha_{\;\;\alpha}.\label{def_lambda}
\end{align}
The irreducible parts $X_\alpha$ and $X_{\alpha\beta\gamma}$ of $\tilde X_{\alpha\beta\gamma}$ are defined through the relation
\begin{align}
    \tilde X_{\alpha\beta\gamma}\triangleq X_{\alpha\beta\gamma}+\epsilon_{\alpha\beta\gamma\delta}X^\delta,\qquad \text{ with }X_{[\alpha\beta\gamma]}\overset{!}{=}0.\label{irreducible_parts}
\end{align}
They provide an equivalent description, since Eq. \eqref{irreducible_parts} can be inverted as
\begin{align}
    X_{\alpha\beta\gamma}&=\tilde X_{\alpha\beta\gamma}-\tilde X_{[\alpha\beta\gamma]},\\
    X^\alpha&=\frac{1}{6}\epsilon^{\alpha\beta\gamma\delta}\tilde X_{\beta\gamma\delta}.
\end{align}
Finally, a simple computation shows that the dual of $\mathbf{\tilde X}$ is given by
\begin{align}
    {}^*\tilde X_{\alpha\beta\gamma}={}^*X_{\alpha\beta\gamma}-2g_{\gamma[\alpha}X_{\beta]}.\label{dual_tilde_X}
\end{align}
 
\subsection{The structural equation}
This third preliminary part will be devoted to the proof of the \textit{structural equation} \cite{doi:10.1098/rspa.1983.0012}
\boxedeqn{
    \nabla_\delta\lambda_{\alpha\beta\gamma}=2\qty(R\tud{\lambda}{\delta\alpha\beta}K_{\gamma\lambda}-R\tud{\lambda}{\delta\gamma[\alpha}K_{\beta]\lambda})+\mu_{\alpha\beta\gamma\delta}\label{structural}
}{Structural equation}
with
\begin{align}
    \begin{split}
        \mu_{\alpha\beta\gamma\delta}&\triangleq\frac{1}{2}\bigg[K_{\beta\gamma;(\alpha\delta)}+K_{\alpha\delta;(\beta\gamma)}-K_{\alpha\gamma;(\beta\delta)}-K_{\beta\delta;(\alpha\gamma)}\\
        &\quad-3\qty(K_{\lambda[\alpha}R\tud{\lambda}{\beta]\gamma\delta}+K_{\lambda[\gamma}R\tud{\lambda}{\delta]\alpha\beta})\bigg].\label{mu}
    \end{split}
\end{align}
Here, $\mu_{\alpha\beta\gamma\delta}$ possesses the same algebraic symmetries than the Riemann tensor.
We remind the reader that $\lambda_{\alpha\beta\gamma}\triangleq 2\nabla_{[\alpha}K_{\beta]\gamma}$. We will use indifferently the notations $\nabla_\alpha \mathbf T$ or $\mathbf T_{;\alpha}$ for the covariant derivative of a tensor $\mathbf T$.
The proof goes as a lengthy rewriting of the original expression:
\begin{align}
    \begin{split}
    \nabla_\delta\lambda_{\alpha\beta\gamma}&=\nabla_\delta\nabla_\alpha K_{\beta\gamma}-\nabla_\delta\nabla_\beta K_{\alpha\gamma}\\
    &=\nabla_{(\delta}\nabla_{\alpha)} K_{\beta\gamma}+\nabla_{[\delta}\nabla_{\alpha]} K_{\beta\gamma}-\nabla_{(\delta}\nabla_{\beta)} K_{\alpha\gamma}-\nabla_{[\delta}\nabla_{\beta]} K_{\alpha\gamma}\\
    &=\frac{1}{2}\nabla_{(\alpha}\nabla_{\delta)}K_{\beta\gamma}-\frac{1}{2}\nabla_{(\beta}\nabla_{\delta)}K_{\alpha\gamma}+\frac{1}{2}\qty(\nabla_{(\alpha}\nabla_{\delta)}K_{\beta\gamma}-\nabla_{(\beta}\nabla_{\delta)}K_{\alpha\gamma})\\
    &\quad+\frac{1}{2}\comm{\nabla_\delta}{\nabla_\alpha}K_{\beta\gamma}-\frac{1}{2}\comm{\nabla_\delta}{\nabla_\beta}K_{\alpha\gamma}.
    \end{split}
\end{align}
We proceed to the following rewriting of twice the quantity in brackets contained in the above expression:
\begin{align}
    \begin{split}
    &\nabla_\alpha\nabla_\delta K_{\beta\gamma}+\nabla_\delta\nabla_\alpha K_{\beta\gamma}-\nabla_\beta\nabla_\delta K_{\alpha\gamma}-\nabla_\delta\nabla_\beta K_{\alpha\gamma}\\
    &=2\nabla_\alpha\nabla_\delta K_{\beta\gamma}-2\nabla_\beta\nabla_\delta K_{\alpha\gamma}+\comm{\nabla_\delta}{\nabla_\alpha}K_{\beta\gamma}-\comm{\nabla_\delta}{\nabla_\beta}K_{\alpha\gamma}\\
    &=2\qty(\nabla_\beta\nabla_\alpha K_{\gamma\delta}+\nabla_\beta\nabla_\gamma K_{\alpha\delta}-\nabla_\alpha\nabla_\beta K_{\gamma\delta}-\nabla_\alpha\nabla_\gamma K_{\beta\delta})\\&\quad+\comm{\nabla_\delta}{\nabla_\alpha}K_{\beta\gamma}-\comm{\nabla_\delta}{\nabla_\beta}K_{\alpha\gamma}\\
    &=2\qty(\nabla_{(\beta}\nabla_{\gamma)}K_{\alpha\delta}-\nabla_{(\alpha}\nabla_{\gamma)}K_{\beta\delta})+\comm{\nabla_\delta}{\nabla_\alpha}K_{\beta\gamma}-\comm{\nabla_\delta}{\nabla_\beta}K_{\alpha\gamma}\\
    &\quad-2\comm{\nabla_\alpha}{\nabla_\beta}K_{\gamma\delta}-\comm{\nabla_\alpha}{\nabla_\gamma}K_{\beta\delta}+\comm{\nabla_\beta}{\nabla_\gamma}K_{\alpha\delta}.
    \end{split}
\end{align}
This yields
\begin{align}
    \nabla_\delta\lambda_{\alpha\beta\gamma}&=\frac{1}{2}\qty(K_{\beta\gamma;(\alpha\delta)}+K_{\alpha\delta;(\beta\gamma)}-K_{\alpha\gamma;(\beta\delta)}-K_{\beta\delta;(\alpha\gamma)})\nonumber\\
    &\quad+\bigg(\frac{3}{4}\comm{\nabla_\delta}{\nabla_\alpha}K_{\beta\gamma}-\frac{3}{4}\comm{\nabla_\delta}{\nabla_\beta}K_{\alpha\gamma}-\frac{1}{2}\comm{\nabla_\alpha}{\nabla_\beta}K_{\gamma\delta}\\
    &\quad-\frac{1}{4}\comm{\nabla_\alpha}{\nabla_\gamma}K_{\beta\delta}+\frac{1}{4}\comm{\nabla_\beta}{\nabla_\gamma}K_{\alpha\delta}\bigg)\label{structural_int}.
\end{align}
Let us denote $\heart$ the quantity between parentheses in the last equation. It can be rearranged in the following way:
\begin{align}
\begin{split}
        4\heart&=3\comm{\nabla_\delta}{\nabla_\alpha}K_{\beta\gamma}-3\comm{\nabla_\delta}{\nabla_\beta}K_{\alpha\gamma}-2\comm{\nabla_\alpha}{\nabla_\beta}K_{\gamma\delta}\\
    &\quad-\comm{\nabla_\alpha}{\nabla_\gamma}K_{\beta\delta}+\comm{\nabla_\beta}{\nabla_\gamma}K_{\alpha\delta}\\
    &=\qty(3 R\tud{\lambda}{\gamma\delta\beta}-R\tud{\lambda}{\delta\beta\gamma})K_{\alpha\lambda}+\qty(R\tud{\lambda}{\delta\alpha\gamma}-3R\tud{\lambda}{\gamma\delta\alpha})K_{\beta\lambda}\nonumber\\
    &\quad+\qty(3R\tud{\lambda}{\alpha\delta\beta}-3R\tud{\lambda}{\beta\delta\alpha}+2R\tud{\lambda}{\delta\alpha\beta})K_{\gamma\lambda}+\qty(2R\tud{\lambda}{\gamma\alpha\beta}+R\tud{\lambda}{\beta\alpha\gamma}-R\tud{\lambda}{\alpha\beta\gamma})K_{\delta\lambda}\\
    &=5R\tud{\lambda}{\delta\alpha\beta}K_{\gamma\lambda}+4\qty(R\tud{\lambda}{\delta\alpha\gamma}K_{\beta\lambda}-R\tud{\lambda}{\delta\beta\gamma}K_{\alpha\lambda})\\
    &\quad+3\qty(R\tud{\lambda}{\alpha\gamma\delta}K_{\beta\lambda}-R\tud{\lambda}{\beta\gamma\delta}K_{\alpha\lambda}+R\tud{\lambda}{\gamma\alpha\beta}K_{\delta\lambda})\\
    &=5R\tud{\lambda}{\delta\alpha\beta}K_{\gamma\lambda}+8R\tud{\lambda}{\delta[\alpha|\gamma}K_{|\beta]\lambda}-6K_{\lambda[\alpha}R\tud{\lambda}{\beta]\gamma\delta}\\
    &\quad+3 R\tud{\lambda}{\gamma\alpha\beta}K_{\delta\lambda}\underbrace{-3R\tud{\lambda}{\delta\alpha\beta}K_{\gamma\lambda}+3R\tud{\lambda}{\delta\alpha\beta}K_{\gamma\lambda}}_{=0}\\
    &=8R\tud{\lambda}{\delta\alpha\beta}K_{\gamma\lambda}-8 R\tud{\lambda}{\delta\gamma[\alpha}K_{\beta]\lambda}-6K_{\lambda[\alpha}R\tud{\lambda}{\beta]\gamma\delta}-6 K_{\lambda[\gamma}R\tud{\lambda}{\delta]\alpha\beta}.
\end{split}
\end{align}
Consequently,
\begin{align}
    \heart=2\qty(R\tud{\lambda}{\delta\alpha\beta}K_{\gamma\lambda}- R\tud{\lambda}{\delta\gamma[\alpha}K_{\beta]\lambda})-\frac{3}{2}\qty(K_{\lambda[\alpha}R\tud{\lambda}{\beta]\gamma\delta}+ K_{\lambda[\gamma}R\tud{\lambda}{\delta]\alpha\beta}).
\end{align}
Inserting this result into Eq. \eqref{structural_int} gives the structural equation \eqref{structural} and consequently concludes the proof.

We will end this section by working out the dual form of the structural equation \eqref{structural}. One has
\begin{align}
    \begin{split}
    \nabla_\delta{}^*\!\lambda_{\alpha\beta\gamma}&=2R\tud{*\lambda}{\delta\alpha\beta}K_{\gamma\lambda}-\epsilon\tdu{\alpha\beta}{\mu\nu}R\tud{\lambda}{\delta\gamma\mu}K_{\nu\lambda}+{}^*\!\mu_{\alpha\beta\gamma\delta}\\
    &=2R\tud{*\lambda}{\delta\alpha\beta}K_{\gamma\lambda}+\epsilon\tdu{\alpha\beta}{\mu\nu}R\tud{**\lambda}{\delta\gamma\mu}K_{\nu\lambda}+{}^*\!\mu_{\alpha\beta\gamma\delta}.
    \end{split}
\end{align}
It is now easier to compute
\begin{align}
    \begin{split}
        \nabla_\delta{}^*\!\lambda\tdu{\alpha\beta}{\gamma}&=2R\tud{*\lambda}{\delta\alpha\beta}K\tud{\gamma}{\lambda}-\frac{1}{2}\epsilon_{\mu\alpha\beta\nu}\epsilon^{\mu\gamma\rho\sigma}R\tud{*\lambda}{\delta\rho\sigma}K\tud{\nu}{\lambda}+{}^*\!\mu\tdud{\alpha\beta}{\gamma}{\delta}\\
    &=2R\tud{*\lambda}{\delta\alpha\beta}K\tud{\gamma}{\lambda}+3\,\delta^{[\gamma}_\alpha\delta^\rho_\beta\delta_\nu^{\sigma]}R\tud{*\lambda}{\delta\rho\sigma}K\tud{\nu}{\lambda}+{}^*\!\mu\tdud{\alpha\beta}{\gamma}{\delta},
    \end{split}
\end{align}
which yields
\begin{align}
    \begin{split}
    \nabla_\delta{}^*\!\lambda_{\alpha\beta\gamma}&=2R\tud{*\lambda}{\delta\alpha\beta}K_{\gamma\lambda}+3\,g_{\gamma[\alpha|}R\tud{*\lambda}{\delta|\beta\nu]}K\tud{\nu}{\lambda}+{}^*\!\mu_{\alpha\beta\gamma\delta}\\
    &=2R\tud{*\lambda}{\delta\alpha\beta}K_{\gamma\lambda}+\big(R\tud{*\lambda}{\delta\alpha\beta}K_{\gamma\lambda}+R\tudu{*\lambda}{\delta\beta}{\rho}K_{\lambda\rho}g_{\alpha\gamma}\\
    &\quad-R\tudu{*\lambda}{\delta\alpha}{\rho}K_{\lambda\rho}g_{\beta\gamma}\big) +{}^*\!\mu_{\alpha\beta\gamma\delta}.
    \end{split}
\end{align}
Rearranging the different terms leads to the final expression
\boxedeqn{
\nabla_\delta{}^*\!\lambda_{\alpha\beta\gamma}=3R\tud{*\lambda}{\delta\alpha\beta}K_{\gamma\lambda}-2R\tudu{*\lambda}{\delta[\alpha}{\rho}g_{\beta]\gamma}K_{\lambda\rho} +{}^*\!\mu_{\alpha\beta\gamma\delta}.\label{covDlambda}
}{Structural equation (dual form)}

\subsection{Reduction of the second constraint}
We will now gather the results obtained in the three previous subsections to express the constraint \eqref{C2} in terms of the irreducible variables introduced above. Let us remind that Eq. \eqref{C2} reads as
\begin{align}
    \sV\tud{\alpha}{(\beta\gamma\delta)}=0.
\end{align}
Using Eqs. \eqref{tildeX}, \eqref{dual_tilde_X} and \eqref{covDlambda}, we can rewrite
\begin{align}
    \sV_{\alpha\beta\gamma\delta}&=\nabla_\delta{}^*\!\tilde X_{\alpha\beta\gamma}+\frac{1}{3}\nabla_\delta\qty({}^*\!\lambda_{\alpha\beta\gamma}+g_{\gamma[\alpha}\nabla_{\beta]*}K)-\frac{2}{3}K_{\lambda\gamma}R\tud{*\lambda}{\delta\alpha\beta}+\frac{2}{3}R\tud{*\lambda}{\delta\gamma[\alpha}K_{\beta]\lambda}\nonumber\\
    &=\nabla_\delta\qty({}^*\!X_{\alpha\beta\gamma}-2g_{\gamma[\alpha}X_{\beta]})+\frac{1}{3}\underbrace{\qty(R\tud{*\lambda}{\delta\alpha\beta}K_{\gamma\lambda}+2R\tud{*\lambda}{\delta\gamma[\alpha}K_{\beta]\lambda})}_{\triangleq\,\clubsuit}\nonumber\\
    &\quad+\frac{1}{3}\nabla_\delta\qty( g_{\gamma[\alpha}\nabla_{\beta]*}K)+\frac{1}{3}{}^*\!\mu_{\alpha\beta\gamma\delta}-\frac{2}{3}R\tudu{*\lambda}{\delta[\alpha}{\rho}g_{\beta]\gamma}K_{\lambda\rho}.
\end{align}
On the one hand, we have
\begin{align}
    \clubsuit&=R\tud{*\lambda}{\delta\alpha\beta}K_{\gamma\lambda}+R\tud{*\lambda}{\delta\gamma\alpha}K_{\beta\lambda}-R\tud{*\lambda}{\delta\gamma\beta}K_{\alpha\lambda}
    =R\tud{*\lambda}{\delta\beta\gamma}K_{\alpha\lambda}+2R\tud{*\lambda}{\delta\alpha[\beta}K_{\gamma]\lambda}.
\end{align}
And on the other hand, we can write
\begin{align}
    \nabla_\delta\qty( g_{\gamma[\alpha}\nabla_{\beta]*}K)&=\frac{1}{2}\epsilon\tdu{\alpha\beta}{\mu\nu}g_{\gamma\mu}\nabla_\delta\nabla_\nu K
    =\frac{1}{2}\epsilon\tdu{\alpha\beta\gamma}{\mu}\nabla_\delta\nabla_\mu K.
\end{align}
Gathering these expressions leads to
\begin{align}
    \begin{split}
    \sV_{\alpha\beta\gamma\delta}&=\nabla_\delta{}^*\!X_{\alpha\beta\gamma}-2g_{\gamma[\alpha|}\nabla_\delta X_{|\beta]}-\frac{2}{3}R\tudu{*\lambda}{\delta[\alpha}{\rho}g_{\beta]\gamma}K_{\lambda\rho}+\frac{1}{3}R\tud{*\lambda}{\delta\beta\gamma}K_{\alpha\lambda}\\
    &\quad+\frac{2}{3}R\tud{*\lambda}{\delta\alpha[\beta}K_{\gamma]\lambda}+\frac{1}{6}\epsilon\tdu{\alpha\beta\gamma}{\mu}\nabla_\delta\nabla_\mu K+\frac{1}{3}{}^*\!\mu_{\alpha\beta\gamma\delta}.
    \end{split}
\end{align}
\noindent
When symmetrizing the three last indices, the four last terms of this expression vanish. The constraint \eqref{C2} takes the final form
\boxedeqn{
    \begin{split}
        &{}^*\!X_{\alpha(\beta\gamma;\delta)}+X_{\alpha;(\beta}g_{\gamma\delta)}-g_{\alpha(\beta}X_{\gamma;\delta)}\\
        &\quad+\frac{1}{3}\qty(g_{\alpha(\beta}R\tudu{*\lambda}{\gamma\delta)}{\rho}-R\tudu{*\lambda}{\alpha(\beta}{\rho}g_{\gamma\delta)})K_{\lambda\rho}=0.\label{reduced_C2}
    \end{split}
    }{First order constraint for quadratic conserved quantities}
Our principal motivation being the motion of spinning particles in Kerr spacetime, we will now focus on spacetimes possessing a Killing-Yano (KY) tensor. It will turn out that the constraint \eqref{reduced_C2} can still be dramatically simplified in such a framework.

\section{Spacetimes admitting a Killing-Yano tensor}\label{sec:KY}
We now particularize our analysis to spacetimes equipped with a Killing-Yano (KY) tensor, \textit{i.e.} a rank-2, antisymmetric tensor $Y_{\mu\nu}=Y_{[\mu\nu]}$ obeying the Killing-Yano equation:
\boxedeqn{\nabla_{(\alpha}Y_{\beta)\gamma}=0.\label{KYdef}
}{Killing-Yano equation}
In this case, the constraint \eqref{C1} is automatically fulfilled, because
\begin{equation}\label{defK}
    K_{\alpha\beta}\triangleq Y\tdu{\alpha}{\lambda}Y_{\lambda\beta}
\end{equation}
is a Killing tensor.

The aim of this section is twofold. First, we will review some general properties of KY tensors useful for the continuation of our analysis. Even if most of them were previously mentioned in the literature \cite{doi:10.1098/rspa.1981.0046,doi:10.1098/rspa.1981.0056}, the goal of the present exposition is to provide a self-contained summary of these results and of their derivations. Second, we will work out an involved identity that will become a cornerstone for solving the constraint \eqref{reduced_C2}. This so-called \textit{central identity} is the generalization of a result mentioned by Rüdiger in \cite{doi:10.1098/rspa.1983.0012}. Rüdiger only quickly sketched the proof of his identity, while we aim here to provide a more pedagogical derivation of this central result.

\subsection{Some general properties of KY tensors}
Let us review some basic properties of KY tensors, sticking to our conventions and simplifying the notations used in the literature \cite{doi:10.1098/rspa.1981.0046,doi:10.1098/rspa.1981.0056}.

\subsubsection{An equivalent form of the KY equation}
The symmetries $\nabla_{(\alpha}Y_{\beta)\gamma}=\nabla_\alpha Y_{(\beta\gamma)}=0$ ensure the quantity $\nabla_\alpha Y_{\beta\gamma}$ to be totally antisymmetric in its three indices. Consequently, there exists a vector $\boldsymbol\xi$ such that
\boxedeqn{\nabla_\alpha Y_{\beta\gamma}=\epsilon_{\alpha\beta\gamma\lambda}\xi^\lambda.\label{KY_equiv}}{Killing-Yano equation (equivalent form)}
The value of $\boldsymbol \xi$ can be found by contracting Eq. \eqref{KY_equiv} with $\epsilon^{\mu\alpha\beta\gamma}$ and making use of the contraction formula for the Levi-Civita tensor \eqref{contraction}. We obtain
\begin{equation}
\xi^\alpha=-\frac{1}{3}\nabla_\lambda Y^{*\lambda\alpha}.\label{xi}
\end{equation}
Notice that Eq. \eqref{KY_equiv} with $\boldsymbol \xi$ given by Eq. \eqref{xi} is totally equivalent to the Killing-Yano equation.

\subsubsection{Dual KY equation}
Let us derive the equivalent to the Killing-Yano equation for the dual of the KY tensor $Y^*_{\mu\nu}$. One has
\begin{align}
    \begin{split}
    \nabla_\alpha Y^{*\beta\gamma}&=\frac{1}{2}\epsilon^{\beta\gamma\mu\nu}\nabla_\alpha Y_{\mu\nu}
    =\frac{1}{2}\epsilon^{\mu\nu\beta\gamma}\epsilon_{\mu\nu\alpha\lambda}\xi^\lambda\\
    &=-2\delta^{[\beta}_\alpha\delta^{\gamma]}_\lambda\xi^\lambda
    =-2\delta^{[\beta}_\alpha\xi^{\gamma]},
    \end{split}
\end{align}
leading to
\boxedeqn{\nabla_\alpha Y^*_{\beta\gamma}=-2g_{\alpha[\beta}\xi_{\gamma]}.\label{KY_bis}
}{Conformal Killing-Yano equation}
This equation is nothing but the conformal Killing-Yano equation for the dual tensor ${Y}^*_{\mu\nu}$. This proves Proposition \ref{prop:killing_quantities}.1.

\subsubsection{Integrability conditions for the KY equation}
We will work out some necessary conditions for the tensor $\mathbf Y$ to be a Killing-Yano tensor, \textit{i.e.} relations that must hold for $\mathbf{Y}$ to satisfy the KY equation. We will refer to them as \textit{integrability conditions} for the Killing-Yano equation.

We begin by proving some preliminary results:
\begin{lemma}\label{lemma:divless}
One has
\begin{align}
    \nabla^\alpha\xi_\alpha=0.
\end{align}
\end{lemma}
\begin{proof}
We proceed by applying the Ricci identity to the expression:
\begin{align}
    \begin{split}
    \nabla^\alpha\xi_\alpha&=-\frac{1}{3}\nabla_\alpha\nabla_\lambda Y^{*\lambda\alpha}
    =-\frac{1}{6}\comm{\nabla_\alpha}{\nabla_\lambda}Y^{*\lambda\alpha}\\
    &=\frac{1}{6}g^{\lambda\mu}g^{\alpha\nu}\qty(R\tud{\rho}{\mu\alpha\lambda}Y^*_{\rho\nu}+R\tud{\rho}{\nu\alpha\lambda}Y^*_{\mu\rho})\\
    &=\frac{1}{3}R^{\alpha\beta}Y^*_{\alpha\beta}
    =0.
    \end{split}
\end{align}
\end{proof}

\begin{lemma}\label{lemma:sym}
For any antisymmetric tensor $A_{\alpha\beta}=A_{[\alpha\beta]}$, one has
\begin{align}
    R\tudu{\alpha}{\mu\nu}{\beta}A^{\mu\nu}=R\tudu{[\alpha}{\mu\nu}{\beta]}A^{\mu\nu}.
\end{align}
\end{lemma}
\begin{proof}
The proof is straightforward:
\begin{align}
    \begin{split}
    R\tudu{\beta}{\mu\nu}{\alpha}A^{\mu\nu}&=R\tdud{\nu}{\alpha\beta}{\mu}A^{\mu\nu}=R\tudu{\alpha}{\mu\nu}{\beta}A^{\nu\mu}\\
    &=-R\tudu{\alpha}{\mu\nu}{\beta}A^{\mu\nu}.
    \end{split}
\end{align}
\end{proof}
\begin{lemma}\label{lemma:indices_R}
For any antisymmetric tensor $A_{\alpha\beta}=A_{[\alpha\beta]}$, we have
\begin{align}
    A_{\mu\nu}R\tudu{\mu}{\alpha\beta}{\nu}=-\frac{1}{2}A^{\mu\nu}R_{\alpha\beta\mu\nu}.
\end{align}
\end{lemma}
\begin{proof}
The proof is again pretty simple:
\begin{align}
    \begin{split}
    A_{\mu\nu}R\tudu{\mu}{\alpha\beta}{\nu}&=A^{\mu\nu}\qty(R_{\alpha\beta\nu\mu}+R_{\alpha\nu\mu\beta})\\
    &=-A^{\mu\nu}R_{\alpha\beta\mu\nu}+A^{\mu\nu}R_{\nu\alpha\beta\mu}.
    \end{split}
\end{align}
The conclusion is reached using Lemma \ref{lemma:sym}.
\end{proof}

We can now work out the first integrability condition. One has
\begin{align}
    \begin{split}
    \nabla_\alpha\xi_\beta&=-\frac{1}{3}\nabla_\alpha\nabla_\lambda Y\tud{*\lambda}{\beta}\\
    &=-\frac{1}{3}\nabla^\lambda\nabla_\alpha Y^*_{\lambda\beta}-\frac{1}{3}\comm{\nabla_\alpha}{\nabla_\lambda}Y\tud{*\lambda}{\beta}\\
    &=\frac{2}{3}\nabla^\lambda\qty(g_{\alpha[\lambda}\xi_{\beta]})-\frac{1}{3}g^{\lambda\rho}\comm{\nabla_\alpha}{\nabla_\lambda}Y^*_{\rho\beta}\\
    &=\frac{1}{3}\nabla_\alpha\xi_\beta-\frac{1}{3}g_{\alpha\beta}\nabla^\lambda\xi_\lambda+\frac{1}{3}\qty(R\tud{\lambda}{\alpha}Y^*_{\lambda\beta}+R\tudu{\lambda}{\beta\alpha}{\rho}Y^*_{\rho\lambda}).
    \end{split}
\end{align}
This gives rise to the Killing-Yano equation \defining{first integrability condition}:
    \boxedeqn{\nabla_\alpha\xi_\beta=\frac{1}{2}\qty(R\tud{\lambda}{\alpha}Y^*_{\lambda\beta}+R\tudu{\lambda}{\alpha\beta}{\rho}Y^*_{\lambda\rho})\label{int_1bis}}{Killing-Yano first integrability condition}
or, equivalently, using Lemma \ref{lemma:indices_R}:
\begin{align}
    \nabla_\alpha\xi_\beta=\frac{1}{2} R\tud{\lambda}{\alpha}Y^*_{\lambda\beta}-\frac{1}{4}R_{\alpha\beta\mu\nu}Y^{*\mu\nu}.\label{int_1}
\end{align}
Symmetrizing this equation gives the \defining{reduced form}
    \boxedeqn{\nabla_{(\alpha}\xi_{\beta)}=\frac{1}{2}Y^*_{\lambda(\alpha}R\tud{\lambda}{\beta)}.\label{int_1_reduced}}{Reduced Killing-Yano first integrability condition}
In particular, for Ricci-flat spacetimes, $\xi^\mu$ is a Killing vector. This proves proposition \ref{prop:killing_quantities}.5.

A second integrability condition can be written as follows.
Taking the derivative of the defining equation $\nabla_{(\alpha}Y_{\beta)\gamma}=0$ yields
\begin{align}
    \nabla_\alpha\nabla_\beta Y_{\gamma\delta}+\nabla_\alpha\nabla_\gamma Y_{\beta\delta}=0.
\end{align}
Out of this equation, we can write the three (equivalent) identities
\begin{align}
    \begin{split}
        \nabla_\alpha\nabla_\beta Y_{\gamma\delta}+\nabla_\gamma\nabla_\alpha Y_{\beta\delta}+\comm{\nabla_\alpha}{\nabla_\gamma}Y_{\beta\delta}&=0,\\
    \nabla_\gamma\nabla_\alpha Y_{\beta\delta}+\nabla_\beta\nabla_\gamma Y_{\alpha\delta}+\comm{\nabla_\gamma}{\nabla_\beta}Y_{\alpha\delta}&=0,\\
    \nabla_\beta\nabla_\gamma Y_{\alpha\delta}+\nabla_\alpha\nabla_\beta Y_{\gamma\delta}+\comm{\nabla_\beta}{\nabla_\alpha}Y_{\gamma\delta}&=0.
    \end{split}
\end{align}
Summing the first and the third and subtracting the second leads to
\begin{align}
    \begin{split}
    2\nabla_\alpha\nabla_\beta Y_{\gamma\delta}&=\comm{\nabla_\gamma}{\nabla_\alpha}Y_{\beta\delta}+\comm{\nabla_\gamma}{\nabla_\beta}Y_{\alpha\delta}+\comm{\nabla_\alpha}{\nabla_\beta}Y_{\gamma\delta}\\
    &=R\tud{\lambda}{\delta\beta\gamma}Y_{\alpha\lambda}+R\tud{\lambda}{\delta\alpha\gamma}Y_{\beta\lambda}+R\tud{\lambda}{\delta\beta\alpha}Y_{\gamma\lambda}\\
    &\quad-\qty(R\tud{\lambda}{\beta\gamma\alpha}+R\tud{\lambda}{\alpha\gamma\beta}+R\tud{\lambda}{\gamma\alpha\beta})Y_{\lambda\delta}\\
    &=2R\tud{\lambda}{\alpha\beta\gamma}Y_{\lambda\delta}+R\tud{\lambda}{\delta\beta\gamma}Y_{\alpha\lambda}+R\tud{\lambda}{\delta\alpha\gamma}Y_{\beta\lambda}+R\tud{\lambda}{\delta\beta\alpha}Y_{\gamma\lambda}.
    \end{split}
\end{align}
This gives rise to the \defining{second integrability condition}
\boxedeqn{\nabla_\alpha\nabla_\beta Y_{\gamma\delta}=R\tud{\lambda}{\alpha\beta\gamma}Y_{\lambda\delta}+\frac{1}{2}\qty(R\tud{\lambda}{\delta\beta\gamma}Y_{\alpha\lambda}+R\tud{\lambda}{\delta\alpha\gamma}Y_{\beta\lambda}+R\tud{\lambda}{\delta\beta\alpha}Y_{\gamma\lambda}).\label{int_2}
}{Killing-Yano second integrability condition}
A reduced form can be obtained by contracting the equation above with $g^{\gamma\delta}$. We obtain
\begin{align}
    0&=R\tudu{\lambda}{\alpha\beta}{\rho}Y_{\lambda\rho}+\frac{1}{2}\qty(R\tud{\lambda}{\beta}Y_{\alpha\lambda}+R\tud{\lambda}{\alpha}Y_{\beta\lambda}+R\tud{\lambda\rho}{\beta\alpha}Y_{\rho\lambda}).
\end{align}
 
Using Lemma \ref{lemma:indices_R} shows that the first and the fourth terms of the right-hand side of this relation cancel. We obtain the \defining{reduced form of the second integrability condition}
\boxedeqn{
    Y_{\lambda(\alpha}R\tud{\lambda}{\beta)}=0.
    }{Reduced Killing-Yano second integrability condition}
We can also symmetrize the indices $\gamma\delta$ in Eq. \eqref{int_2} to obtain the \defining{symmetrized second integrability condition}
\boxedeqn{R\tud{\lambda}{\alpha\beta(\gamma}Y_{\delta) \lambda}-\frac{1}{2}Y_{\lambda (\gamma}R\tud{\lambda}{\delta)\alpha\beta}+R\tud{\lambda}{(\gamma\delta)(\alpha}Y_{\beta) \lambda}=0.\label{int3}
}{Symmetrized Killing-Yano integrability condition}

\subsection{The central identity}
Our so-called central identity will consists into a clever rewriting of the expression $K_{\lambda(\beta}R\tud{*\lambda}{\gamma\delta)\alpha}$, and was first derived in \cite{Compere:2020eat}. Its reduced version first introduced by Rüdiger \cite{doi:10.1098/rspa.1983.0012} is a rewriting of the contracted expression $K_{\lambda\rho}R\tudu{*\lambda}{\alpha\beta}{\rho}$.

\subsubsection{The KY scalar}
Let us define the scalar quantity
\begin{equation}
    \mathcal Z\triangleq\frac{1}{4}Y^*_{\alpha\beta}Y^{\alpha\beta}.\label{defZ}
\end{equation}
Its first covariant derivative takes the form
\begin{align}
    \nabla_\mu\mathcal Z&=\frac{1}{2}\nabla_\mu Y^*_{\alpha\beta}Y^{\alpha\beta}
    =\xi_{[\alpha}g_{\beta]\mu}Y^{\alpha\beta}
    =\xi_\alpha Y\tud{\alpha}{\mu}.\label{DZ}
\end{align}
Its second covariant derivative can be expressed as (notice that $\nabla_\alpha\nabla_\beta\mathcal Z=\nabla_\beta\nabla_\alpha\mathcal Z$)
\begin{align}
    \nabla_\mu\nabla_\nu\mathcal Z&=\nabla_\mu\xi_\alpha Y\tud{\alpha}{\nu}+\xi^\alpha\epsilon_{\mu\alpha\nu\rho}\,\xi^\rho
    =\nabla_\mu\xi_\alpha Y\tud{\alpha}{\nu}=\nabla_\nu\xi_\alpha Y\tud{\alpha}{\mu},\label{DDZ}
\end{align}
where the last equality follows from symmetry of the right-hand side.

The following identity is also useful:
\begin{align}
    Y_{\alpha[\beta}Y_{\gamma]\delta}=-\frac{1}{2}Y_{\beta\gamma}Y_{\alpha\delta}-\frac{1}{2}\mathcal Z \,\epsilon_{\alpha\beta\gamma\delta}.
    \label{Y_indices}
\end{align}
The proof consists into noticing that the combination $Y_{\alpha\beta}Y_{\gamma\delta}-Y_{\alpha\gamma}Y_{\beta\delta}+Y_{\beta\gamma}Y_{\alpha\delta}$ is totally antisymmetric in all its indices. It must consequently be proportional to the Levi-Civita tensor:
\begin{align}
    Y_{\alpha\beta}Y_{\gamma\delta}-Y_{\alpha\gamma}Y_{\beta\delta}+Y_{\beta\gamma}Y_{\alpha\delta}=\mathcal A\, \epsilon_{\alpha\beta\gamma\delta},
\end{align}
where the constant $\mathcal A$ remains to be determined. This is achieved by contracting the equation above with $\epsilon^{\alpha\beta\gamma\delta}$, yielding
\begin{align}
    3\epsilon^{\alpha\beta\gamma\delta}Y_{\alpha\beta}Y_{\gamma\delta}=-4!\,\mathcal A.
\end{align}
Using the definition of $\mathcal Z$ leads to
\begin{align}
    \mathcal A=-\mathcal Z,
\end{align}
which gives the desired result.

\subsubsection{Derivation of the central identity}
The trick for deriving the central identity is to define the tensor
\begin{align}
    \mathcal T_{\mu\nu\rho\sigma}\triangleq\epsilon_{\mu\alpha\beta\sigma}\nabla^\alpha\nabla^\beta Y_{\nu\lambda}Y\tud{\lambda}{\rho},\label{generalized_T}
\end{align}
to perform two different rewritings of this expression and finally to equate them. 
Notice that we recover the tensor used in Rüdiger's proof \cite{doi:10.1098/rspa.1983.0012} by contracting the last two indices of Eq. \eqref{generalized_T}. 
 
First, applying the Ricci identity to Eq. \eqref{generalized_T} and making use of Eq. \eqref{Y_indices} yields 
\begin{align}
    \begin{split}
    \mathcal T_{\mu\nu\rho\sigma}&=\frac{1}{2}\epsilon_{\mu\alpha\beta\sigma}\comm{\nabla^\alpha}{\nabla^\beta}Y_{\nu\lambda}Y\tud{\lambda}{\rho}
    =R\tud{*\kappa}{\nu\mu\sigma}K_{\kappa\rho}-R\tud{*\kappa\lambda}{\mu\sigma}Y_{\nu[\kappa}Y_{\lambda]\rho}\\
    &=-R\tud{*\kappa}{\nu\sigma\mu}K_{\kappa\rho}+\frac{1}{2}R\tud{*\kappa\lambda}{\mu\sigma}\qty(Y_{\kappa\lambda}Y_{\nu\rho}+\mathcal Z\epsilon_{\nu\kappa\lambda\rho}).
    \end{split}
\end{align}
Second, using Eqs. \eqref{xi}, \eqref{KY_bis} and Lemma \ref{lemma:divless}, we rewrite $\mathcal T_{\mu\nu\rho\sigma}$ as follows:
\begin{align}
\begin{split}
\mathcal T\tudu{\mu}{\nu\rho}{\sigma}&=-\frac{1}{2}\epsilon^{\mu\alpha\beta\sigma}\epsilon_{\nu\lambda\gamma\delta}\nabla_\alpha\nabla_\beta Y^{*\gamma\delta}Y\tud{\lambda}{\rho}\\
&=-12\delta^{[\sigma}_\nu Y\tud{\mu}{\rho}\nabla_\alpha\nabla_\beta Y^{*\alpha\beta]}\\
&=-\delta^\sigma_\nu\qty(Y\tud{\mu}{\rho}\nabla_\alpha\nabla_\beta Y^{*\alpha\beta}+Y\tud{\alpha}{\rho}\nabla_\alpha\nabla_\beta Y^{*\beta\mu}+Y\tud{\beta}{\rho}\nabla_\alpha\nabla_\beta Y^{*\mu\alpha})\\
&\quad + \delta^\mu_\nu\qty(Y\tud{\sigma}{\rho}\nabla_\alpha\nabla_\beta Y^{*\alpha\beta}+Y\tud{\alpha}{\rho}\nabla_\alpha\nabla_\beta Y^{*\beta\sigma}+Y\tud{\beta}{\rho}\nabla_\alpha\nabla_\beta Y^{*\sigma\alpha})\\
&\quad-\delta^\alpha_\nu\qty(Y\tud{\sigma}{\rho}\nabla_\alpha\nabla_\beta Y^{*\mu\beta}+Y\tud{\mu}{\rho}\nabla_\alpha\nabla_\beta Y^{*\beta\sigma}+Y\tud{\beta}{\rho}\nabla_\alpha\nabla_\beta Y^{*\sigma\mu})\\
&\quad+\delta^\beta_\nu\qty(Y\tud{\sigma}{\rho}\nabla_\alpha\nabla_\beta Y^{*\mu\alpha}+Y\tud{\mu}{\rho}\nabla_\alpha\nabla_\beta Y^{*\alpha\sigma}+Y\tud{\alpha}{\rho}\nabla_\alpha\nabla_\beta Y^{*\sigma\mu})\\
&=\delta^\sigma_\nu\qty(3Y\tud{\alpha}{\rho}\nabla_\alpha\xi^\mu+2Y\tud{\beta}{\rho}\delta^{[\mu}_\beta\nabla_\alpha\xi^{\alpha]}){ -}\delta^\mu_\nu\qty(3Y\tud{\alpha}{\rho}\nabla_\alpha\xi^\sigma+2Y\tud{\beta}{\rho}\delta_\beta^{[\sigma}\nabla_\alpha\xi^{\alpha]})\\
&\quad-3Y\tud{\sigma}{\rho}\nabla_\nu\xi^\mu+3Y\tud{\mu}{\rho}\nabla_\nu\xi^\sigma+2Y\tud{\beta}{\rho}\delta_\beta^{[\sigma}\nabla_\nu\xi^{\mu]}\\
&\quad-2\qty(Y\tud{\sigma}{\rho}\delta_\nu^{[\mu}\nabla_\alpha\xi^{\alpha]}+Y\tud{\mu}{\rho}\delta_\nu^{[\alpha}\nabla_\alpha\xi^{\sigma]}+Y\tud{\alpha}{\rho}\delta_\nu^{[\sigma}\nabla_\alpha\xi^{\mu]})
\\
&=\delta^\sigma_\nu Y\tud{\lambda}{\rho}\nabla_\lambda\xi^\mu-\delta^\mu_\nu Y\tud{\lambda}{\rho}\nabla_\lambda\xi^\sigma+ Y\tud{\mu}{\rho}\nabla_\nu\xi^\sigma-Y\tud{\sigma}{\rho}\nabla_\nu\xi^\mu.
\end{split}
\end{align}
Equating the two expressions of $\mathcal{T}_{\mu\nu\rho\sigma}$ obtained leads to
\begin{align}
   \begin{split}
    R\tud{*\lambda}{\nu\mu\sigma}K_{\rho\lambda}&=-\frac{1}{2}R\tud{*\kappa\lambda}{\mu\sigma}\qty(Y_{\kappa\lambda}Y_{\nu\rho}+\mathcal Z\epsilon_{\nu\kappa\lambda\rho})-g_{\mu\nu}Y\tud{\lambda}{\rho}\nabla_\lambda\xi_\sigma\\
    &\quad+g_{\nu\sigma}Y\tud{\lambda}{\rho}\nabla_\lambda\xi_\mu-Y_{\sigma\rho}\nabla_\nu\xi_\mu+Y_{\mu\rho}\nabla_\nu\xi_\sigma.\label{central_relation}
   \end{split}
\end{align}
This is the cornerstone equation for deriving both the central identity and its reduced form.

\subsubsection{Non-reduced form}
Fully symmetrizing Eq. \eqref{central_relation} in $(\mu\nu\rho)$ leads to
\begin{align}
    K_{\lambda(\beta}R\tud{*\lambda}{\gamma\delta)\alpha}=-Y\tud{\lambda}{(\beta}g_{\gamma\delta)}\xi_{\alpha;\lambda}+g_{\alpha(\beta}Y\tud{\lambda}{\gamma}\xi_{\delta);\lambda}-Y_{\alpha(\beta}\xi_{\gamma;\delta)}.
\end{align}
We make use of the reduced integrability condition \eqref{int_1_reduced} to write the last term as
\begin{align}
    \begin{split}
    Y_{\alpha(\beta}\xi_{\gamma;\delta)}&=\frac{1}{2}Y_{\alpha(\beta|}Y^*_{\lambda|\gamma}R\tud{\lambda}{\delta)}\\
    &=\frac{1}{2}Y_{\alpha(\beta}Y\tud{*\lambda}{\gamma}G_{\delta)\lambda}+\frac{R}{4}Y_{\alpha(\beta}Y^*_{\gamma\delta)}\\
    &=\frac{1}{2}Y_{\alpha(\beta}Y\tud{*\lambda}{\gamma}G_{\delta)\lambda}\label{last}
    \end{split}
\end{align}
where $G_{\alpha\beta}\triangleq R_{\alpha\beta}-\frac{R}{2}g_{\alpha\beta}$ is the Einstein tensor. This gives rise to the central identity:
\boxedeqn{
K_{\lambda(\beta}R\tud{*\lambda}{\gamma\delta)\alpha}=-Y\tud{\lambda}{(\beta}g_{\gamma\delta)}\xi_{\alpha;\lambda}+g_{\alpha(\beta}Y\tud{\lambda}{\gamma}\xi_{\delta);\lambda}-\frac{1}{2}Y_{\alpha(\beta}Y\tud{*\lambda}{\gamma}G_{\delta)\lambda}.
\label{generalized_central}
}{Central identity}

\subsubsection{Reduced form} 
Rüdiger's reduced form of the central identity can be derived by contracting Eq. \eqref{central_relation} with $g^{\rho\sigma}$ and using Eqs. \eqref{int_1_reduced} and \eqref{DDZ}:
\begin{align}
    R\tudu{*\lambda}{\nu\mu}{\rho}K_{\lambda\rho}&=-\frac{1}{2}R\tudu{*\kappa\lambda}{\mu}{\rho}\qty(Y_{\kappa\lambda}Y_{\nu\rho}+\mathcal Z\epsilon_{\nu\kappa\lambda\rho})-g_{\mu\nu}Y^{\lambda\rho}\nabla_\lambda\xi_\rho+Y\tud{\lambda}{\nu}\nabla_\lambda\xi_\mu+Y_{\mu\rho}\nabla_\nu\xi^\rho\nonumber\\
    &=-\frac{1}{2}R\tudu{*\kappa\lambda}{\mu}{\rho}\qty(Y_{\kappa\lambda}Y_{\nu\rho}+\mathcal Z\epsilon_{\nu\kappa\lambda\rho})+g_{\mu\nu}Y^{\lambda\rho}\nabla_\rho\xi_\lambda-Y\tud{\lambda}{\nu}\nabla_\mu\xi_\lambda\nonumber\\
    &\quad+Y\tud{\lambda}{\nu}Y^*_{\rho(\lambda}R\tud{\rho}{\mu)}+Y_{\mu\rho}\nabla_\nu\xi^\rho\\
    &=-\frac{1}{2}R\tudu{*\kappa\lambda}{\mu}{\rho}\qty(Y_{\kappa\lambda}Y_{\nu\rho}+\mathcal Z\epsilon_{\nu\kappa\lambda\rho})-2\nabla_\mu\nabla_\nu\mathcal Z+g_{\mu\nu}\Delta\mathcal Z+Y\tud{\lambda}{\nu}Y^*_{\rho(\lambda}G\tud{\rho}{\mu)}.\nonumber
\end{align}
The first term of the above equation can be simplified by noticing that, on the one hand, making use of Lemma \ref{lemma:indices_R},
\begin{align}
\begin{split}
\frac{1}{2}R\tudu{*\kappa\lambda}{\mu}{\rho}Y_{\kappa\lambda}Y_{\nu\rho}&=
\frac{1}{4}\epsilon\tdu{\mu}{\alpha\beta\lambda}R\tud{\sigma\rho}{\alpha\beta}Y_{\rho\sigma}Y_{\lambda\nu}
=-\frac{1}{4}\epsilon\tdu{\mu}{\alpha\beta\lambda}R\tud{\sigma\rho}{\alpha\beta}Y^{**}_{\rho\sigma}Y_{\lambda\nu}\\
&=-\frac{1}{8}\epsilon_{\mu\alpha\beta\lambda}\epsilon^{\rho\sigma\gamma\delta}R\tud{\alpha\beta}{\sigma\rho}Y^*_{\gamma\delta}Y\tud{\lambda}{\nu}
=3R\tud{\alpha\beta}{[\alpha\mu}Y^*_{\beta\lambda]}Y\tud{\lambda}{\nu}\\
&=\frac{1}{4}Y\tud{\lambda}{\nu}\qty(4R\tud{\beta}{\mu}Y^*_{\beta\lambda}+4R\tud{\beta}{\lambda}Y^*_{\mu\beta}-2R\tud{\alpha\beta}{\mu\lambda}Y^*_{\alpha\beta}-2RY^*_{\mu\lambda})\\
&=\frac{1}{4}Y\tud{\lambda}{\nu}\qty[4\qty(R\tudu{\alpha}{\mu\lambda}{\beta}Y^*_{\alpha\beta}+R\tud{\beta}{\mu}Y^*_{\beta\lambda})+4R\tud{\beta}{\lambda}Y^*_{\mu\beta}-2RY^*_{\mu\lambda}]\\
&\overset{\eqref{int_1}}{=}2Y\tud{\lambda}{\nu}\nabla_\mu\xi_\lambda+R\tud{\beta}{\lambda}Y\tud{\lambda}{\nu}Y^*_{\mu\beta}-\frac{1}{2}RY\tud{\lambda}{\nu}Y^*_{\mu\lambda}\\
&=2\nabla_\mu\nabla_\nu\mathcal Z+G\tud{\beta}{\lambda}Y\tud{\lambda}{\nu}Y^*_{\mu\beta}.
\end{split}
\end{align}
On the other hand, the term $\frac{1}{2}R\tudu{*\kappa\lambda}{\mu}{\rho}\mathcal Z\epsilon_{\nu\kappa\lambda\rho}$ can be reduced thanks to the identity
\begin{align}
    \frac{1}{4}\mathcal Z\epsilon^{\mu\alpha\beta\lambda}\epsilon_{\lambda\rho\sigma\nu}R\tud{\sigma\rho}{\alpha\beta}&=\frac{3}{2}\mathcal Z \delta^\mu_{[\rho}R\tud{\sigma\rho}{\sigma\nu]}
    =\frac{1}{2}\mathcal Z\qty(2 R\tud{\mu}{\nu}-\delta^\mu_\nu R)
    =\mathcal Z G\tud{\mu}{\nu}.
\end{align}
Putting all pieces together, we obtain the reduced central identity
\boxedeqn{
    \begin{split}
     R\tudu{*\lambda}{\nu\mu}{\rho}K_{\lambda\rho}&=-4\nabla_\mu\nabla_\nu\mathcal Z+g_{\mu\nu}\Delta \mathcal Z+Y\tud{*\lambda}{(\rho}G_{\mu)\lambda}Y\tud{\rho}{\nu}\\
     &\quad-G\tud{\beta}{\lambda}Y\tud{\lambda}{\nu}Y^*_{\mu\beta}-\mathcal Z G_{\mu\nu}.\label{central_id_general}
    \end{split}
}{Reduced central identity}
This is the generalization of Rüdiger's reduced central  identity \cite{doi:10.1098/rspa.1983.0012} to non Ricci-flat spacetimes. Notice that the occurrence of Einstein tensor in the relation Eq. \eqref{central_id_general} that we derived suggests that Rüdiger's quadratic invariant will admit a generalization to the Kerr-Newmann spacetime which also admits a Killing-Yano tensor, once the Einstein tensor is replaced with the electromagnetic stress-energy tensor. This remains to be investigated.

\subsubsection{Central identity in Ricci-flat spacetimes} 
Let us now particularize our analysis to vacuum spacetimes, i.e. Ricci-flat spacetimes $R_{\alpha\beta}=G_{\alpha\beta}=0$. This includes in particular the astrophysically relevant Kerr spacetime. 

Using the reduced integrability condition \eqref{int_1_reduced}, the central identity \eqref{generalized_central} becomes
\begin{align}
   \begin{split}
    K_{\lambda(\beta}R\tud{*\lambda}{\gamma\delta)\alpha}&=-Y\tud{\lambda}{(\beta}g_{\gamma\delta)}\xi_{\alpha;\lambda}+g_{\alpha(\beta}Y\tud{\lambda}{\gamma}\xi_{\delta);\lambda}\\
    &=Y\tud{\lambda}{(\beta}g_{\gamma\delta)}\xi_{\lambda;\alpha}-g_{\alpha(\beta}Y\tud{\lambda}{\gamma|}\xi_{\lambda;|\delta)}.
   \end{split}
\end{align}
Making use of Eq. \eqref{DDZ}, it takes the final form
    \boxedeqn{
    K_{\lambda(\beta}R\tud{*\lambda}{\gamma\delta)\alpha}=\nabla_\alpha\nabla_{(\beta}\mathcal Z g_{\gamma\delta)}-g_{\alpha(\beta}\nabla_\gamma\nabla_{\delta)}\mathcal Z,\label{central_id_gen_RF}
    }{Central identity in Ricci-flat spacetimes}
which does not appear in Rüdiger \cite{doi:10.1098/rspa.1983.0012}. 
The reduced central identity \eqref{central_id_general} becomes
\boxedeqn{
    R\tudu{*\lambda}{\mu\nu}{\rho}K_{\lambda\rho}=-4\nabla_\mu\nabla_\nu\mathcal Z+g_{\mu\nu}\Delta\mathcal Z,\label{central_id}
    }{Reduced central identity in Ricci-flat spacetimes}
as obtained by Rüdiger \cite{doi:10.1098/rspa.1983.0012}.

\section{Solutions to the \texorpdfstring{$\mathcal O(\mathcal S)$}{} constraint in Ricci-flat spacetimes}\label{sec:unicity}

In this Section, one will gather the results obtained in the two previous sections of this chapter. Our aim will be to solve the constraint \eqref{reduced_C2} in Ricci-flat spacetimes admitting a KY tensor. The Ricci-flatness assumption will enable us to make use of the simple expressions provided by the central identity \eqref{central_id_gen_RF} and its reduced form \eqref{central_id}.

\subsection{Simplification of the constraint}
Plugging Eq. \eqref{central_id} into Eq. \eqref{reduced_C2} leads -- after a few easy manipulations -- to the constraint
\begin{align}
    {}^*\,X_{\alpha(\beta\gamma;\delta)}+\nabla_{(\beta|}\qty[X_\alpha+\frac{4}{3}\nabla_\alpha\mathcal Z]g_{|\gamma\delta)}-g_{\alpha(\beta}\nabla_{\delta}\qty[X_{\gamma)}+\frac{4}{3}\nabla_{\gamma)}\mathcal Z]=0.\label{final_constraint_KY}
\end{align}
It is straightforward to see that it admits the following non-trivial solution
\begin{align}
    X_{\alpha\beta\gamma}=0,\qquad X_{\alpha}=-\frac{4}{3}\nabla_\alpha\mathcal Z.\label{sol_C2}
\end{align}
As we will see later, this solution will lead to Rüdiger's invariant. In what follows, we will seek a more general solution to Eq. \eqref{final_constraint_KY}, which does not assume $X_{\alpha\beta\gamma}=0$. The first step is to simplify more the constraint \eqref{final_constraint_KY}. We begin by rewriting it in the much simpler form
\begin{align}
    \qty[{}^* X_{\alpha(\beta\gamma}+Y_\alpha g_{(\beta\gamma}-g_{\alpha(\beta}Y_\gamma]_{;\delta)}=0\label{constr_red}
\end{align}
by introducing the shifted variable
\begin{align}
    Y_\alpha\triangleq X_\alpha+\frac{4}{3}\nabla_\alpha\mathcal Z.
\end{align}
We subsequently rewrite Eq. \eqref{constr_red} in order to remove the dual operator from the tensor $\mathbf X$. Let us define the Hodge dual $\tilde Y_{\alpha\beta\gamma}$ of $Y_\alpha$:
\begin{align}
    \tilde Y_{\alpha\beta\gamma}\triangleq\epsilon_{\mu\alpha\beta\gamma}Y^\mu\qquad\Leftrightarrow\qquad Y_\alpha\triangleq-\frac{1}{6}\epsilon_{\alpha\mu\nu\rho}\tilde Y^{\mu\nu\rho}.
\end{align}
Contracting Eq. \eqref{constr_red} with $\epsilon^{\alpha\mu\nu\rho}$ and using the usual properties of the Levi-Civita tensor, we get
\begin{align}
-3\delta^{[\mu}_{(\beta}X\tud{\nu\rho]}{\gamma;\delta)}+\tilde Y\tud{\mu\nu\rho}{;(\beta}g_{\gamma\delta)}-\epsilon\tdu{(\beta}{\mu\nu\rho}Y_{\gamma;\delta)}=0.
\end{align}
Now, using the fact that
\begin{align}
    \epsilon^{\beta\mu\nu\rho}Y_{\gamma;\delta}=4\delta_\gamma^{[\beta}\tilde Y\tud{\mu\nu\rho]}{;\delta},
\end{align}
the last term of the previous equation reads
\begin{align}
    -\epsilon\tdu{(\beta}{\mu\nu\rho}Y_{\gamma;\delta)}=-g_{(\beta\gamma}\tilde Y\tud{\mu\nu\rho}{;\delta)}+3\delta^{[\mu}_{(\gamma}\tilde Y\tud{\nu\rho]}{\beta;\delta)}.
\end{align}
Putting all pieces together, Eq. \eqref{final_constraint_KY} becomes equivalent to
\begin{align}
    \delta_{(\beta}^{[\mu}\qty(X\tud{\nu\rho]}{\gamma}-\tilde Y\tud{\nu\rho]}{\gamma})_{;\delta)}=0.\label{c_red_int}
\end{align}
It follows from Eq. \eqref{tildeX} and from the definition of $L_{\alpha\beta\gamma}=L_{[\alpha\beta]\gamma}$ that $X_{\alpha\beta\gamma}$ is antisymmetric on its two first indices, $X_{\alpha\beta\gamma}=X_{[\alpha\beta]\gamma}$.
Let us decompose $X_{\alpha\beta\gamma}$ into its traces and trace-free parts:
\begin{align}
    X_{\alpha\beta\gamma}=A_\alpha g_{\beta\gamma}+B_{\beta}g_{\alpha\gamma}+C_\gamma g_{\alpha\beta}+D^\text{tf}_{\alpha\beta\gamma}.\label{defXD}
\end{align}
Note that the constraint $X_{[\alpha\beta\gamma]}=0$ reduces to $D^\text{tf}_{[\alpha\beta\gamma]}=0$.
Moreover, since $X_{\alpha\beta\gamma}=X_{[\alpha\beta]\gamma}$, it implies that one can set $C_\gamma=0$. Plugging the above decomposition into Eq. \eqref{c_red_int} and using the identity $\delta_{(\alpha}^{[\mu}\delta_{\beta)}^{\nu]}=0$, the constraint \eqref{c_red_int} becomes
\begin{align}
    \delta^{[\mu}_{(\beta}\qty(D\tud{\text{tf}\,\nu \rho]}{\gamma}-\tilde Y\tud{\nu\rho]}{\gamma})_{;\delta)}=0.\label{tf_constraint}
\end{align}
This implies that the 1-forms $A_\alpha$ and $B_\beta$ determining the trace part of $X_{\alpha\beta\gamma}$ are left unconstrained. However, these trace parts produce terms into the conserved quantity \eqref{quadratic_invariant} containing a $p_\mu S^{\mu\nu}$ factor, which vanish due to the spin supplementary condition \eqref{ssc}.
All in all, we can, without loss of generality, set $A_\alpha=B_\beta=0=C_\gamma$, \textit{i.e.} consider that $X_{\alpha\beta\gamma}$ reduces to its traceless part. The constraint equation can be written in the short form
\begin{align}
    \delta_{(\beta}^{[\mu}W\tud{\nu\rho]}{\gamma;\delta)}=0,\label{final_cst}
\end{align}
where we have defined $W_{\alpha\beta\gamma}\triangleq D^{\text{tf}}_{\alpha\beta\gamma}-\tilde Y_{\alpha\beta\gamma}$, which is by definition traceless and antisymmetric in its two first indices. 
Contracting the constraint \eqref{final_cst} with $-\frac{1}{2} \epsilon_{\mu\nu\rho\alpha}$ leads to the equivalent condition
\begin{align}
    ^*\!W_{\alpha(\beta\gamma;\delta)}=0,\label{final_cst_star}
\end{align}
where
\begin{align}
     ^*\!W_{\alpha\beta\gamma}=\,^*\!D^\text{tf}_{\alpha\beta\gamma}-2g_{\gamma[\alpha}Y_{\beta]}.\label{decomp_star}
\end{align}

Given a solution to the simplified constraint \eqref{final_cst_star}, the quasi-invariant $K_{\mu\nu}p^\mu p^\nu + L_{\alpha\beta\gamma}S^{\alpha\beta}p^\gamma$ 
can be constructed using Eqs. \eqref{tildeX}, \eqref{irreducible_parts} and \eqref{defXD} which leads to 
\begin{equation}
  L_{\alpha\beta\gamma}=D^{\text{tf}}_{\alpha\beta\gamma}  +\frac{1}{3}\lambda_{\alpha\beta\gamma}+\epsilon_{\alpha\beta\gamma\delta}(Y^\delta - \frac{4}{3}\nabla^\delta Z).\label{recon}
\end{equation}

\subsection{The R\"udiger quasi-invariant $\mathcal Q_R$}

Let us first recover R\"udiger's quasi-invariant \cite{doi:10.1098/rspa.1983.0012}.
In terms of our new variables, Rüdiger's solution \eqref{sol_C2} is simply
\begin{align}
    D^{\text{tf}}_{\alpha\beta\gamma}=0,\qquad Y_{\alpha}=0. \label{sol_new_var}
\end{align}
Substituting in \eqref{recon}, it leads to the quasi-invariant
\begin{align}
    \mathcal Q_R &=K_{\mu\nu}p^\mu p^\nu+\frac{1}{3}\lambda_{\mu\nu\rho}S^{\mu\nu}p^\rho -\frac{4}{3}\epsilon_{\mu\nu\rho\sigma}S^{\mu\nu}p^\rho \nabla_\sigma \mathcal Z.\label{Rudinv}
\end{align}
Let us work on each term of the right-hand side separately:
\begin{itemize}
    \item We recall the definition 
$L_\alpha\triangleq Y_{\alpha\lambda}p^\lambda$.  The definition \eqref{defK} leads to 
    \begin{align}
        K_{\mu\nu}p^\mu p^\nu=L_\alpha L^\alpha .
    \end{align}
     \item Using the definitions \eqref{def_lambda} and \eqref{defK} and the properties \eqref{KY_equiv} and \eqref{DZ} we have
      
    \begin{align}
        \begin{split}
        \lambda_{\mu\nu\rho}S^{\mu\nu}p^\rho &= 2\nabla_\mu K_{\nu\rho}S^{\mu\nu}p^\rho\\
        &=2\qty(\epsilon_{\mu\nu\lambda\sigma}\epsilon^{\mu\nu\alpha\beta}Y\tud{\lambda}{\rho}+\epsilon_{\mu\lambda\rho\sigma}\epsilon^{\mu\nu\alpha\beta}Y\tdu{\nu}{\lambda})\xi^\sigma\hat p_\alpha S_\beta p^\rho\\
        &=-2\qty(4Y\tud{\lambda}{\rho}\xi^\sigma \hat p_{[\lambda}S_{\sigma]}p^\rho+6Y\tdu{[\lambda}{\lambda}\hat p_\rho S_{\sigma]}\xi^\sigma p^\rho)\\
        &=2\qty(3 Y\tud{\lambda}{\rho}p^\rho\xi^\sigma\hat p_\sigma S_\lambda+Y\tdu{\sigma}{\lambda}S_\lambda\xi^\sigma \hat p_\rho p^\rho-\underbrace{Y\tdu{\sigma}{\lambda}\hat p_\lambda\xi^\sigma S_\rho p^\rho}_{=0})\\
        &=-2\mu\xi^\sigma Y\tdu{\sigma}{\lambda}S_\lambda+6\mu^{-1} L_\alpha S^\alpha \xi_\beta p^\beta\\
        &=-2\mu S^\alpha\partial_\alpha\mathcal Z+6\mu^{-1} L_\alpha S^\alpha \xi_\beta p^\beta.\label{lambdaterm}
        \end{split}
    \end{align}
      
Using \eqref{id}, the factor $L_\alpha S^\alpha$ can be written in a much more enlightening way:
\begin{align}
\begin{split}
    L_\alpha S^\alpha&=Y_{\alpha\lambda}p^\lambda S^\alpha=Y^*_{\alpha\lambda}p^{[\alpha}S^{\lambda]*}=-\frac{\mu}{2}Y^*_{\alpha\lambda}S^{\alpha\lambda} \\ 
    &=-\frac{\mu}{2}\mathcal Q_Y,\label{LS}
\end{split}
\end{align}
where $\mathcal Q_Y$ is Rüdiger's linear invariant \cite{doi:10.1098/rspa.1981.0046}, which is also conserved up to linear order in the spin. 

\item Using the definition of the spin vector \eqref{defSmu} together with the spin supplementary condition allows to write
    \begin{align}
        \begin{split}
        \epsilon_{\mu\nu\rho\sigma} S^{\mu\nu}p^\rho \nabla^\sigma\mathcal Z&=4\delta_\rho^{[\alpha}\delta_\sigma^{\beta]}\nabla^\sigma\mathcal Z \hat p_\alpha S_\beta p^\rho \\ 
        &=4\hat p_{[\rho}S_{\sigma]}\nabla^\sigma \mathcal Z p^\rho  \\ 
        &=2S_\sigma\nabla^\sigma\mathcal{Z}\hat p_\rho p^\rho\\
        &=-2\mu S^{\alpha}\partial_\alpha\mathcal{Z}.\label{third}
        \end{split}
    \end{align}

\end{itemize}

Putting all the pieces together, we get the following form for $\mathcal Q_R$, 
\boxedeqn{
\mathcal Q_R=L_\alpha L^\alpha+2\mu S^\alpha\partial_\alpha\mathcal Z-\mathcal Q_Y\,\xi_\beta p^\beta.\label{QR}
}{Rüdiger quadratic invariant}
which is identically the quadratic R\"udiger invariant \cite{doi:10.1098/rspa.1983.0012}. For Ricci-flat spacetimes, $\xi^\mu$ is a Killing vector and $\xi_\alpha p^\alpha$ can be upgraded to an invariant with a $\mathcal O(\mathcal S^1)$ correction. This implies that the last term in \eqref{QR} is trivial since it is a product of quasi-invariant at linear order in $\mathcal S$.

\subsection{Trivial solutions to the algebraic constraints}
Let us now discuss more general solutions to the simplified algebraic constraints \eqref{final_cst_star}. By linearity, we can substract the R\"udiger quasi-invariant \eqref{Rudinv} from the definition of quasi-invariants \eqref{quadratic_invariant} where $L_{\alpha\beta\gamma}$ is given in Eq. \eqref{recon}. Given a solution to the simplified algebraic constraints \eqref{final_cst_star}, one therefore obtains a quasi-invariant of the form
\begin{align}
    \begin{split}
    \mathcal Q(D_{\alpha\beta\gamma}^{\text{tf}},Y_\alpha)&=\qty(D^\text{tf}_{\alpha\beta\gamma}+\epsilon_{\alpha\beta\gamma\lambda}Y^\lambda)S^{\alpha\beta}p^\gamma\\
    &=2\qty(\,^*\!D^\text{tf}_{\alpha\beta\gamma}+\,^*\!\epsilon_{\alpha\beta\gamma\lambda}Y^\lambda)S^\alpha\hat p^\beta p^\gamma\\
    &=2\mu\,^*\!D^\text{tf}_{\alpha(\beta\gamma)}S^\alpha\hat p^\beta\hat p^\gamma-2\mu S_\alpha Y^\alpha\label{inv_star}
    \end{split}
\end{align}
after using Eq. \eqref{subsS}. Such quasi-invariant is homogeneous in $S$. 

Looking at \eqref{final_cst_star}, it is appealing to first attempt to generate a quasi-invariant through solving the stronger algebraic constraint
\begin{align}
    ^*W_{\alpha(\beta\gamma)}=0.\label{strong_cst}
\end{align}
However, this procedure only leads to identically vanishing quasi-invariants. Indeed, by virtue of Eq. \eqref{decomp_star}, one then has
\begin{align}
    ^*D^\text{tf}_{\alpha(\beta\gamma)}=g_{\alpha(\gamma}Y_{\beta)}-g_{\beta\gamma}Y_\alpha.
\end{align}
This yields
\begin{align}
    \mathcal Q(D_{\alpha\beta\gamma}^{\text{tf}},Y_\alpha)&=-2\mu S_\alpha Y^\alpha \hat p_\beta \hat p^\beta-2\mu S_\alpha Y^\alpha=0.\label{trivid}
\end{align}
This implies that new non-trivial quasi-invariants can only be generated by making a non-trivial use of the symmetrized covariant derivative $_{;\delta)}$ in Eq. \eqref{final_cst_star}.

\subsection{Invariants homogeneously linear in $\mathcal S$}

Non-trivial invariants of the form \eqref{inv_star} can only be generated by finding a solution to the differential constraint \eqref{final_cst_star}. As we will see in this section, this problem and the form of the generated invariant can be recast in a very simple form. The first step is to express the invariant \eqref{inv_star} as a function of $^*W_{\alpha\beta\gamma}$ only. Contracting Eq. \eqref{decomp_star} with $g^{\beta\gamma}$ yields
\begin{align}
    Y_{\alpha}=\frac{1}{3}\,^*W\tdu{\alpha\lambda}{\lambda}.\label{Y}
\end{align}
Rearranging Eq. \eqref{decomp_star} and making use of Eq. \eqref{Y}, we get
\begin{align}
    ^*D_{\alpha\beta\gamma}^\text{tf}=\,^*W_{\alpha\beta\gamma}+\frac{2}{3}g_{\gamma[\alpha}\,^*W\tdu{\beta]\lambda}{\lambda}.
\end{align}
Plugging these two expressions into Eq. \eqref{inv_star} and using the orthogonality condition $\hat p_\alpha S^\alpha=0$ leads to the simple expression
\begin{align}
    \mathcal Q(W_{\alpha\beta\gamma})=2\mu \,^*W_{\alpha(\beta\gamma)}S^\alpha\hat p^\beta\hat p^\gamma.\label{inv_star_final}
\end{align}
Now, because the dualization is an invertible operation, the giving of $^*W_{\alpha(\beta\gamma)}$ is equivalent to the giving of the symmetrized part in its two last indices of a rank-3 tensor $N_{\alpha\beta\gamma}$ antisymmetric in its two first indices, such that
\begin{align}
    N_{\alpha\beta\gamma}\equiv\,^*W_{\alpha\beta\gamma},
\end{align}
which obeys
\begin{align}
    N_{\alpha(\beta\gamma;\delta)}=0.\label{sfinalc}
\end{align}
Note that one \textit{cannot} impose that $N_{\alpha\beta\gamma}$ is also symmetric in $\beta\gamma$ otherwise it would vanish because one would have $N_{\alpha\beta\gamma}=-N_{\beta\alpha\gamma}=-N_{\beta\gamma\alpha}=N_{\gamma\beta\alpha}=N_{\gamma\alpha\beta}=-N_{\alpha\gamma\beta}=-N_{\alpha\beta\gamma}$.

The tensor $T_{\alpha\beta\gamma}\triangleq N_{\alpha(\beta\gamma)}$ obeys the cyclic identity 
\begin{align}\label{cyc}
T_{\alpha\beta\gamma}+T_{\beta\gamma\alpha}+T_{\gamma\alpha\beta}=0
\end{align} 
but since $T_{\alpha\beta\gamma}$ is not symmetric in its two first indices, it is not totally symmetric and the condition defining the mixed-symmetry tensor \eqref{sfinalc} is \emph{distinct} from the condition defining a Killing tensor $K_{(\alpha\beta\gamma;\delta)}=0$ where $K_{\alpha\beta\gamma}$ is totally symmetric $K_{(\alpha\beta\gamma)}=K_{\alpha\beta\gamma}$. In terms of representation of the permutation group, $T_{\alpha\beta\gamma}$ is a $\{2,1\}$ Young diagram.

Consequently, the following statement holds:
 \begin{proposition}
Let $(\mathcal M,g_{\mu\nu})$ be a (3+1)-dimensional Ricci-flat spacetime admitting a Killing-Yano tensor such that there exists on $\mathcal M$ a mixed-symmetry Killing tensor, \emph{i.e.}, a rank-3 tensor $T_{\alpha\beta\gamma}$ which is a $\{2,1\}$ Young tableau, \emph{i.e.}, built as $T_{\alpha\beta\gamma}= N_{\alpha(\beta\gamma)}$ such that $N_{\alpha\beta\gamma}=N_{[\alpha\beta]\gamma}$,  satisfying the (differential) constraint
\begin{align}
    T_{\alpha(\beta\gamma;\delta)}=0.\label{finalc}
\end{align}
Then, the quantity
\begin{align}
    \mathcal{N}\triangleq T_{\alpha\beta\gamma}S^\alpha\hat p^\beta\hat p^\gamma\label{cons_qty}
\end{align}
is a (homogeneously linear in $\mathcal S$) quasi-invariant for the linearized MPT equations on $\mathcal M$, \emph{i.e.}
\begin{align}
    \dv{\mathcal{N}}{\tau}=\mathcal O(\mathcal S^2).
\end{align}
\end{proposition}

\subsection{Trivial mixed-symmetry Killing tensors}

If a mixed-symmetry Killing tensor $T_{\alpha\beta\gamma}$ is found, the only point to be addressed before claiming the existence of a new quasi-invariant is to check its non-triviality,   \textit{i.e.},  it should not be the product of two others quasi-invariants. We define a \emph{trivial mixed-symmetry Killing tensor} $T_{\alpha\beta\gamma}$ as mixed-symmetry Killing tensor that generates a trivial quasi-invariant, \textit{i.e.} an invariant which can be written as the product of other quasi-invariants.\footnote{It remains to be investigated if such a tensor necessarily takes the form of a sum of direct products of tensors such that Eq. \eqref{finalc} holds. We have not found any simple argument for proving this assertion.} A \emph{non-trivial mixed-symmetry Killing tensor} is a mixed-symmetry Killing tensor which is not trivial. 

If the spacetime admits a Killing-Yano tensor and a Killing vector, we can construct a trivial mixed-symmetry Killing tensor of the form 
\begin{equation}
T_{\alpha\beta\gamma}=Y_{\alpha(\beta}\xi_{\gamma )} .\label{Ntrivial}
\end{equation}
built as $T_{\alpha\beta\gamma} = N^{(1)}_{\alpha(\beta\gamma)}=N^{(2)}_{\alpha(\beta\gamma)}$ from either 
\begin{equation}
N^{(1)}_{\alpha\beta\gamma}=Y_{\alpha\beta}\xi_\gamma    , \quad \text{or}\quad 
N^{(2)}_{\alpha\beta\gamma}=Y_{\alpha\gamma}\xi_\beta - Y_{\beta\gamma}\xi_\alpha . 
\end{equation}
It is straightforward from Eq. \eqref{KYdef} and the Killing equation that they obey \eqref{finalc}.

Now, for this $T_{\alpha\beta\gamma}$ the associated quasi-invariant is the following product of quasi-invariants: 
\begin{equation}
\mathcal Q=-\frac{1}{2} \mathcal Q_Y \; \xi^\alpha \hat p_\alpha ,
\end{equation}
where the linear R\"udiger quasi-invariant $\mathcal Q_Y$ is defined in Eq. \eqref{QR} and obeys Eq. \eqref{LS}.

The question of existence of quasi-invariants beyond the ones found by R\"udiger therefore amounts to determine the existence of non-trivial mixed-symmetry Killing tensors. 

In order to gain some intuition about mixed-symmetry Killing tensors, it is useful to derive the general such tensor in Minkowski spacetime. For a Riemann flat spacetime, the covariant derivative becomes a coordinate derivative in Minkowskian coordinates (in any dimension). We can therefore consider the index $\alpha$ in Eq. \eqref{finalc} as a parametric index and consider the list of 2-components objects $K_{(\alpha)\beta\gamma}=T_{\alpha\beta\gamma}$ symmetric under the exchange of $\beta\gamma$, parametrized by $\alpha$. The constraint \eqref{finalc} is then equivalent to the Killing tensor equation for each $K_{(\alpha)\beta\gamma}$, $\alpha$ fixed. Since all Killing tensors in Minkowski spacetime are direct products of Killing vectors, we can write $K_{(\alpha)\beta\gamma}$ as a sum of terms of the form $K_{(\alpha)\beta\gamma}=M_{(\alpha)(i)(j)}\xi^{(i)}_\beta \xi^{(j)}_\gamma$ where $(i),(j)$ label the independent Killing vectors. In order to respect the symmetries of a mixed-symmetry tensor we can write in particular the resulting $T_{\alpha\beta\gamma}$ as a linear combination of terms  $X_{\alpha(\beta}\xi_{\gamma)}$ where $\xi^\gamma$ is a Killing vector. The symmetry properties of a  mixed-symmetry tensor imply that $X_{\alpha\beta}=X_{[\alpha\beta]}$. The constraint \eqref{finalc} finally reduces to the condition that $X_{\alpha\beta}$ is a Killing-Yano tensor. We have therefore proven that any mixed-symmetry tensor in Minkowski spacetime takes the trivial form \eqref{Ntrivial}. 

Now, Kerr spacetime admits a non-trivial Killing tensor and curvature and further analysis is required. We will address the existence of a mixed-symmetry tensor for Kerr spacetime in the following.

\section{Linear invariants for Kerr spacetime}\label{sec:Kerr}

We will now derive the explicit expressions for the various invariants discussed above particularized to Kerr spacetime, which is a Ricci-flat spacetime admitting an irreducible Killing-Yano tensor as well as two Killing vectors. In order to simplify the expressions obtained, we introduce the tetrad\footnote{Following \cite{Ruangsri:2015cvg}, we've introduced the symbol ``$\stackrel{.}{=}$'', whose meaning is ``the tensorial object of the left-hand side is represented in Boyer-Lindquist coordinates by the components given in the right-hand side''. Similarly, we introduce the symbol ``$\stackrel{,}{=}$'', bearing the same meaning, but in the tetrad basis.} \cite{Mino:1997bx}
\begin{subequations}
\begin{align}
    e\tud{0}{\mu}&\stackrel{.}{=}\mqty(\sqrt{\frac{\Delta}{\Sigma}},&0,&0,&-a\sin^2\theta\sqrt{\frac{\Delta}{\Sigma}}),\\
    e\tud{1}{\mu}&\stackrel{.}{=}\mqty(0,&\sqrt{\frac{\Sigma}{\Delta}},&0,&0),\\
    e\tud{2}{\mu}&\stackrel{.}{=}\mqty(0,&0,&\sqrt{\Sigma},&0),\\
    e\tud{3}{\mu}&\stackrel{.}{=}\mqty(-\frac{a}{\sqrt{\Sigma}}\sin\theta,&0,&0,&\frac{r^2+a^2}{\sqrt{\Sigma}}\sin\theta).
\end{align}
\end{subequations}
We use lower-case Latin indices to denote tetrad components, \textit{e.g.} 
\begin{align}
    e\tud{a}{\mu}=\qty(e\tud{0}{\mu},e\tud{1}{\mu},e\tud{2}{\mu},e\tud{3}{\mu}).
\end{align}
The tetrad components $V^a$ of any vector $\mathbf V$ are given by $V^a=e\tud{a}{\mu} V^\mu$. Tetrad indices are lowered and raised with the Minkowski metric $\eta_{ab}=\text{diag}\mqty(-1,&1,&1,&1)$. We choose the convention $\epsilon^{t r \theta \varphi }=+1$. In this tetrad basis, the Kerr Killing-Yano tensor and its dual take the elegant form
\begin{subequations}
\begin{align}
    &\frac{1}{2} Y_{ab}\dd x^a\wedge\dd x^b=a\cos\theta\,\dd x^1\wedge\dd x^0+r\,\dd x^2\wedge\dd x^3,\label{Yddt}\\
    &\frac{1}{2}Y^*_{ab}\dd x^a\wedge\dd x^b=r\,\dd x^1\wedge\dd x^0+a\cos\theta\,\dd x^3\wedge\dd x^2.\label{YStarddt}
\end{align}
\end{subequations}

\subsection{Known quasi-conserved quantities}

We now explicit the various quasi-invariants for the linearized MPTD equations on Kerr spacetime. We recall that the norms of the vector-variables
\begin{align}
    \mu^2=-p_a p^a,\qquad S^2\triangleq S_a S^a
\end{align}
are conserved along the motion. In what follows, we will work with the dimensionful quantities $\mathcal E_0$, $\mathcal L_0$, $\mathcal K_0$. They are related to the proper quantities $E_0$, $L_0$ and $K_0$ through
\begin{align}
    \mathcal E_0=\mu E_0,\qquad\mathcal L_0=\mu L_0,\qquad\mathcal K_0=\mu^2 K_0.
\end{align}
This is totally equivalent, since the ``proper'' versions of the constants of motion are obtained from the non-proper ones through the replacements 
\begin{align}
    p^\mu\to \hat p^\mu=\frac{p^\mu}{\mu},\qquad S^{\mu\nu}\to \frac{S^{\mu\nu}}{\mu},\qquad S^\mu\to\frac{S^\mu}{\mu}.
\end{align}

\subsubsection{Invariants generated by Killing fields} 
From the existence of the two Killing fields $\xi^\mu$ and $\eta^\mu$, one can construct two linear invariants, namely the energy $\mathcal E\triangleq -\mathcal C_\xi$ and the projection of the angular momentum along the direction of the black hole spin, $\mathcal L \triangleq\mathcal C_\eta$ where $\mathcal C_\xi$ is defined in Eq. \eqref{killing_conserved_spin}. Their explicit expressions are given by \cite{Saijo:1998mn}
\begin{subequations}\label{spin_EL}
\begin{align}
    \mathcal E&=\mathcal E_0+\frac{M}{\Sigma^2}\qty[\qty(r^2-a^2\cos^2\theta)S^{10}-2ar\cos\theta S^{32}],\\
    \mathcal L&=\mathcal L_0 + \frac{a \sin^2\theta}{\Sigma^2}\qty[(r-M)\Sigma+2Mr^2]S^{10}+\frac{a\sqrt{\Delta}\sin\theta\cos\theta}{\Sigma}S^{20}\nonumber\\
    &\quad+\frac{r\sqrt{\Delta}\sin\theta}{\Sigma}S^{13}
    +\frac{\cos\theta}{\Sigma^2}\qty[(r^2+a^2)^2-a^2\Delta\sin^2\theta]S^{23}.
\end{align}
\end{subequations}
Here, $\mathcal E_0$ and $\mathcal L_0$ denote respectively the geodesic energy and angular momentum
\begin{subequations}
\begin{align}
\mathcal E_0&=-\xi_\mu p^\mu=-p_t=\sqrt{\frac{\Delta}{\Sigma}}p^0+\frac{a\sin\theta}{\sqrt{\Sigma}}p^3,\label{energy}\\
\mathcal L_0 &= \eta_\mu p^\mu=p_\varphi= a\sin^2\theta\sqrt{\frac{\Delta}{\Sigma}}p^0+(a^2+r^2)\frac{\sin\theta}{\sqrt{\Sigma}}p^3.\label{angular_momentum}
\end{align}
\end{subequations}

\subsubsection{Linear Rüdiger quasi-invariant} 
An elegant expression for $\mathcal Q_Y$ \eqref{linear_rudiger_invariant} is given through expressing the spin tensor components in the tetrad frame.  Using Eq. \eqref{YStarddt}, we find that Rüdiger linear quasi-invariant is given by the simple expression
\boxedeqn{
    \mathcal Q_Y=2\qty(rS^{10}+a\cos\theta S^{32}).\label{linear_rudiger_kerr}
}{Linear Rüdiger invariant in Kerr spacetime}

\subsubsection{Quadratic Rüdiger quasi-invariant} 
An explicit expression for Rüdiger quadratic invariant \eqref{quadratic_rudiger} is obtained by evaluating the three scalar products $L_\alpha L^\alpha$, $S^\alpha\partial_\alpha\mathcal Z$ and $\xi_\alpha p^\alpha$. First, the product $L_\alpha L^\alpha$ reduces to the usual Carter constant $\mathcal K_0$,
\begin{align}
    L_\alpha L^\alpha = K_{\alpha\beta}p^\alpha p^\beta = \mathcal K_0.
\end{align}

Second, a direct computation reveals that the Killing vector $\xi^\alpha$ \eqref{xi} reduces to the timelike Kerr Killing vector:
\begin{align}
    \xi^\alpha\stackrel{.}{=}\mqty(1,&0,&0,&0).
\end{align}
The product $\xi_\alpha p^\alpha$ is consequently equal to minus the geodesic energy $\mathcal E_0$:
\begin{align}
    \xi_\alpha p^\alpha=-\mathcal E_0.
\end{align}

Only the factor $S^\alpha\partial_\alpha\mathcal Z$ remains to be computed. The Killing-Yano scalar $\mathcal Z$ \eqref{defZ} reads
\begin{align}
    \mathcal Z=-ar\cos\theta.
\end{align}
The simplest expression for $S^\alpha\partial_\alpha\mathcal Z$ is provided through expressing the components of the spin vector in Boyer-Lindquist coordinates:
\begin{align}
    S^\alpha\partial_\alpha\mathcal Z&=a(r \sin\theta\, S^\theta - \cos\theta\, S^r).
\end{align}
In terms of the spin tensor, we get the more involved expression
\begin{align}
    S^\alpha\partial_\alpha\mathcal Z&= -\frac{3 a}{\mu\sqrt{\Sigma}}(\sqrt{\Delta}\cos\theta p^{[0}S^{23]}+r \sin\theta p^{[0}S^{13]}). \label{SDZ_tetrad}
\end{align}
Rüdiger's quadratic invariant \eqref{Rudinv} consequently takes the form
\boxedeqn{
    \mathcal Q_R=\mathcal K_0+2\mu a \qty(r\sin\theta S^\theta-\cos\theta S^r)+\mathcal Q_Y \mathcal E_0.\label{quadratic_rudiger_kerr}
}{Rüdiger quadratic invariant in Kerr spacetime}

In summary, we have in our possession  four quasi-constants of motion (in addition to $\mu^2$ and $\mathcal S$): $\mathcal E$, $\mathcal L$, $\mathcal Q_Y$ and $\mathcal Q_R$ that are conserved along the flow generated by the MPTD equations at linear order in $\mathcal S$. They consequently form a set of four linearly independent first integrals for the linearized MPTD equations which, however, are \textit{not} in involution, as will be later proven in Chapter \ref{chap:integrability}.

\subsection{Looking for mixed-symmetry tensors}
Let us now address the question of the existence of a non-trivial mixed-symmetry Killing tensor, i.e. a tensor $T_{\alpha\beta\gamma}$ obeying the symmetries of a $\{2,1\}$ Young tableau and obeying Eq. \eqref{finalc}, on Kerr spacetime. It is highly relevant because such the existence of such a tensor is in one-to-one correspondence with the existence of another independent quasi-conserved quantity that could, possibly, be in involution with the others first integrals and consequently leading the system to be integrable.
All the computations mentioned below being very cumbersome, they will not be reproduced here but are encoded in a \textit{Mathematica} notebook, which is available on simple request.
 
From the symmetry in $(\beta\gamma)$ alone, there are $4 \times 10=40$ components in $T_{\alpha\beta\gamma}$ in 4 spacetime dimensions but they are not all independent because $\mathbf T$ is defined from $\mathbf N$ which is antisymmetric in its first indices. From the cyclic identity \eqref{cyc} we deduce the following: (i) The components of $T_{\alpha\beta\gamma}$ with all indices set equal to $i=1, \dots, 4$ is 0, $T_{iii}=0$ (4 identities); (ii) In the presence of two distinct components $i,j=1, \dots , 4$ we have $2T_{iij}+T_{jii}=0$ (12 identities); (iii) In the presence of three distinct components $i,j,k=1,\dots ,4$ we have $T_{ijk}+T_{jki}+T_{kij}=0$ (4 identities). There are no further algebraic symmetries. There are therefore 20 independent components, which we can canonically choose to be the union of the sets of components $\mathcal T_2$ (of order $\vert \mathcal T_2 \vert=12$) and $\mathcal T_3$ (of order $\vert \mathcal T_3 \vert=8$) that are defined as 
\begin{subequations}
\begin{align}
\mathcal T_2 &=\{ T_{ijj} \,\vert\, i,j=1,\dots , 4, \, j \neq i \}, \\ 
    \mathcal T_3 &= \{T_{123},T_{213},T_{124},T_{214},T_{134},T_{314},T_{234},T_{324} \}. 
\end{align}
\end{subequations}
The constraints \eqref{finalc} are 64 equations which can be splitted as follows: (i) 4 equations with 4 distinct indices $T_{i(jk;l)}=0$; (ii) 12 equations with 3 distinct indices of type $T_{i(ij;k)}=0$; (iii) 24 equations with 3 distinct indices of type $T_{i(jj;k)}=0$; (iv) 12 equations with 2 distinct indices of type $T_{i(jj;j)}=0$ and (v) 12 equations with 2 distinct indices of type $T_{i(ii;j)}=0$. 

Let us now specialize to the Kerr background and impose stationarity and axisymmetry, $\partial_t T_{\alpha\beta\gamma}=\partial_\varphi T_{\alpha\beta\gamma}=0$. In that case, the 4 equations $T_{r(tt;\varphi)}=0$, $T_{r(\varphi\varphi;t)}=0$, $T_{\theta (tt;\varphi)}=0$, $T_{\theta(\varphi\varphi;t)}=0$ and the 6 equations $T_{t(\varphi\varphi;\varphi)}=0$, $T_{r(\varphi\varphi;\varphi)}=0$, $T_{\theta(\varphi\varphi;\varphi)}=0$, $T_{r(tt;t)}=0$, $T_{\theta(tt;t)}=0$, $T_{\varphi(tt;t)}=0$  are algebraic, and can be algebraically solved for 10 out of the 20 variables. There are two additional combinations of the remaining equations that allow to algebraically solve for 2 further variables. After removing redundant equations, we can finally algebraically reduce the system of 64 equations in 20 variables to 13 partial differential equations in 8 variables.

Further specializing to the Schwarzschild background, we found the general solution to the thirteen equations. There are exactly two regular solutions which are both trivial mixed-symmetry tensors of the form \eqref{Ntrivial} with either $\xi = \partial_t$ or $\xi = \partial_\varphi$. Note that if we relax axisymmetry, there are two further trivial mixed-symmetry tensors \eqref{Ntrivial} with $\xi$ given by the two additional $\text{\textsf{SO}}(3)$ vectors.

For Kerr, we did not find the general solution of the 13 partial differential equations in 8 variables. However, we obtained the most general perturbative deformation in $a$ of the 2-parameter family of trivial mixed-symmetry tensors of the form \eqref{Ntrivial} with $\xi = \partial_t$ and $\partial_\varphi$, assuming stationarity and axisymmetry. We obtained that the most general deformation is precisely a linear combination of the two trivial mixed-symmetry tensors of the form \eqref{Ntrivial} with $\xi = \partial_t$ and $\partial_\varphi$. Moreover, we also checked that there does not exist a consistent linear deformation of a linear combination of the two $\varphi$-dependent $\text{\textsf{SO}}(3)$ Schwarzschild trivial mixed-symmetry tensors. Since we physically expect continuity of the quasi-conserved quantities between Kerr and Schwarzschild, we ruled out the existence of new stationary and axisymmetric quasi-conserved quantities of the MPTD equations. 

\chapter[Solving the constraints: second order in the spin]{Solving the constraints:\\ second order in the spin}\label{chap:second_order_solution}
\chaptermark{\textsc{Second order in the spin}}

\vspace{\stretch{1}}

In this chapter, we will look again at the conservation equations for the two Ansätze Eqs. \eqref{ansatze}. This time, we will consider the full MPTD equations \eqref{MPD} up to second order in $\mathcal S$. We assume the force and torque terms to be generated by the presence of a spin-induced quadrupole term, given by Eq. \eqref{spin_induced_Q}.

We will heavily base ourselves on the quantities conserved at linear order that were discussed in the previous chapter, building our second order conserved quantities as perturbations of the former. Before starting, two important points shall be raised: first, by contrast to the linearized case, we will have to specialize to Kerr background to be able to work out explicit solutions to the conservation equations. As we will see, a crucial feature for the analytic solving of the problem is the existence of \textit{covariant building blocks} for all Kerr tensorial quantities. In this formulation, all differential relations will reduce to algebraic identities. This will enable to turn the task of solving the constraint equations to a purely algebraic, enumerative problem. Second, the constraints will now explicitly depend upon the quadrupole coupling $\kappa$. This reflects the fact that the universality of the motion of spinning test bodies present at first order is now broken, since $\kappa$ depends upon the nature of the body. Moreover, the value of the coupling $\kappa$ will appear to be of prime importance: one will only be able to find generalizations of the constants $\mathcal Q_Y$ and $\mathcal Q_R$ if the body's quadrupole coupling is itself the one of a Kerr black hole, $\kappa=1$. 

We will proceed as follows. We start by developing the constraint equations for the existence of the conserved quantities in tensorial form in Section~\ref{sec:ConstrainEqs}. 
In Section~\ref{sec:CBB}, we discuss covariant algebraic and differential relations amongst basic fields, ``covariant building blocks,'' characterizing the Kerr geometry, which play a central role in our solutions to the constraint equations.  We use these to reduce the tensorial constraint equations to a system of scalar equations in Section~\ref{sec:Reduction}, and we derive our solutions for the special black-hole case $\kappa=1$ in Section~\ref{sec:BH}.  In Section~\ref{sec:NS}, we investigate the case $\kappa\ne 1$, for non-black-hole bodies such as neutron stars with spin-induced quadrupoles, concluding that there is no solution to the constraints for $\kappa\ne 1$.  Finally we summarize our findings in Section~\ref{sec:outlooks}.

\section{Constraint equations: tensorial formulation}\label{sec:ConstrainEqs}
In this section, we will apply Rüdiger's procedure to derive the tensorial constraint equations that should be obeyed for ensuring the conservation of the two quantities Eqs. \eqref{ansatze}. Form now on, we take the background spacetime to be \textsf{Kerr}$(M,a)$.

\subsection{Linear constraint}
Let us begin by looking again at the linear invariant \eqref{linear_ansatz}. Recall that is takes the form
\begin{align}
    \mathcal Q^{(1)}\triangleq X_\mu p^\mu+W_{\mu\nu}S^{\mu\nu}
\end{align}
with $W_{\mu\nu}$ a skew-symmetric tensor. 

Applying Rüdiger's procedure, the conservation equation $\dot{\mathcal{ Q}}^{(1)}=\mathcal O\qty(\mathcal S^3)$ reduces to the following set of equations:
\begin{subequations}\label{linear_cst}
\begin{align}
    [0,2]:&\qquad\nabla_\mu X_\nu\hat p^\mu\hat p^\nu=\mathcal O\qty(\mathcal S^3),\label{[0,2]}\\
    [1,2]:&\qquad\nabla_\mu W^*_{\alpha\nu}s^\alpha\hat p^\mu\hat p^\nu-\frac{1}{2}X^\lambda R^*_{\lambda\nu\beta\rho}s^\beta\hat p^\nu\hat p^\rho=\mathcal O\qty(\mathcal S^3),\label{[1,2]}\\
    [2,2]:&\qquad\frac{\kappa}{2\mu}X^\lambda\nabla_\lambda R_{\nu\alpha\beta\rho}s^\alpha s^\beta\hat p^\nu\hat p^\rho+Y_{\mu\nu}\mathcal L^{*\mu\nu}=\mathcal O\qty(\mathcal S^3),\label{[2,2]}\\
    [2,4]:&\qquad\qty(\nabla_\lambda X_\mu-2 W_{\lambda\mu})\qty(\mu D\tud{\lambda}{\nu}-\mathcal L\tud{\lambda}{\nu})\hat p^\mu\hat p^\nu=\mathcal O\qty(\mathcal S^3).\label{[2,4]}
\end{align}
\end{subequations}
We therefore have the following proposition:
\begin{proposition}
For any pair $(X_\mu,W_{\mu\nu})$ satisfying the constraint equations \eqref{linear_cst} and assuming the MPD equations \eqref{MPD} are obeyed, the quantity $\mathcal Q^{(1)}$ \eqref{linear_ansatz} will be conserved up to second order in the spin parameter, \textit{i.e.} $\dot{\mathcal{Q}}^{(1)}=\mathcal O\qty(\mathcal S^3)$.
\end{proposition}
The two first equations \eqref{[0,2]}-\eqref{[1,2]} arise at zeroth and first order in $\mathcal S$, and are identical to Eqs. \eqref{test_cst} encountered at linear order. The conserved quantities $\mathcal C_X$ and $\mathcal Q_Y$ found at linear order are therefore the possible candidates for being conserved at this order. We shall examine the two cases independently:
\begin{itemize} 
    \item For $\mathcal C_X$, Eq. \eqref{[2,2]} is the only constraint which is not trivially vanishing. also holds for this value of $Y_{\alpha\beta}$. Using the explicit form of the torque for the spin-induced quadrupole allows to prove that this equation is well obeyed regarless to the value of $\kappa$.
    \item In the case $X_\mu=0$, the constraint equations \eqref{linear_cst} reduce to
        \begin{subequations}\label{linear_cst2}
        \begin{align}
            [1,2]:&\qquad\nabla_\mu Y_{\alpha\nu}s^\alpha\hat p^\mu\hat p^\nu=\mathcal O\qty(\mathcal S^3),\label{[1,2]Q}\\
            [2,2]:&\qquad Y_{\mu\nu}\mathcal L^{*\mu\nu}=\mathcal O\qty(\mathcal S^3),\label{[2,2]Q}\\
            [2,4]:&\qquad W_{\lambda\mu}\qty(\mu D\tud{\lambda}{\nu}-\mathcal L\tud{\lambda}{\nu})\hat p^\mu\hat p^\nu=\mathcal O\qty(\mathcal S^3).\label{[2,4]Q}
        \end{align}
        \end{subequations}
       Eq. \eqref{[1,2]Q} is unchanged and still enforces $Y_{\mu\nu}$ to be a Killing-Yano tensor. The two additional conditions Eqs. \eqref{[2,2]Q}-\eqref{[2,4]Q} are more involved and will be discussed in Section \ref{sec:linCBB}.
\end{itemize}

\subsection{Quadratic constraint}
We now turn to the constraint equations for the quadratic invariants. It is useful to write the Ansatz for the quadratic quantity Eq. \eqref{quadratic_ansatz} as a $\os{2}$ correction added to Rüdiger invariant $\mathcal Q_R$. In doing so, we do not loose in generality, since $\mathcal Q_R$ is the more generic non-trivial (stationary and axisymmetric) quadratic invariant that can be built in Kerr spacetime. We therefore set
\begin{align}
    \mathcal Q^{(2)}=\mathcal Q_R+ \mathcal Q^\text{quad}\label{quasi_inv}
\end{align}
where
\begin{align}
    \mathcal Q_R\triangleq K_{\mu\nu}p^\mu p^\nu+L_{\mu\nu\rho}S^{\mu\nu}p^\rho,\qquad
    \mathcal Q^\text{quad}\triangleq M_{\alpha\beta\gamma\delta}S^{\alpha\beta}S^{\gamma\delta}.
\end{align}
Recall that Rüdiger invariant $\mathcal Q_R$ is defined by
\begin{align}
    K_{\mu\nu}&=Y\tdu{\mu}{\lambda}Y_{\nu\lambda},\quad
    L_{\alpha\beta\gamma}=\frac{2}{3}\nabla_{[\alpha}K_{\beta]\gamma}+\frac{4}{3}\epsilon_{\alpha\beta\gamma\delta}\nabla^\delta\mathcal Z,\label{R2}
\end{align}
with $Y_{\mu\nu}$ Kerr Killing-Yano tensor Eq. \eqref{kerr_KY}. 
Moreover, the tensor $M_{\alpha\beta\gamma\delta}$ admits the same algebraic symmetries than the Riemann tensor\begin{align}
    M_{\alpha\beta\gamma\delta}=M_{[\alpha\beta]\gamma\delta}=M_{\alpha\beta[\gamma\delta]}=M_{\gamma\delta\alpha\beta}.
\end{align}

The conservation conditions at zeroth and first order in $\mathcal S$ are left unchanged by the presence of quadrupolar terms in the MPTD equations, and still give rise to Eqs. \eqref{C1} and \eqref{C2}.  The presence of quadrupole terms in the MPD equations \eqref{MPD} will only appear at quadratic order in $\mathcal S$. Actually, we are left with only one constraint, which is of grading $[2,3]$. The derivation of this quadratic constraint is too long to be provided in the main text and can be found instead in Appendix \ref{app:quadratic_constraint}. We have now demonstrated the following proposition:
\begin{proposition}
Any tensor $N_{\alpha\beta\gamma\delta}$ possessing the same algebraic symmetries as the Riemann tensor and satisfying the constraint equation
\begin{align}
\begin{split}
    &\bigg[4\nabla_\mu N_{\alpha\nu\beta\rho}-2\kappa\nabla_{[\alpha}\mathcal M^{(1)}_{\vert\mu\vert \nu]\beta\rho} +\kappa\qty(g_{\alpha\mu}Y_{\lambda\nu}-g_{\mu\nu}Y_{\lambda\alpha})\xi_\kappa\,^*\!R\tud{\lambda\kappa}{\beta\rho}\\
    &+\qty(2\kappa Y_{\alpha\mu}\xi_\lambda+\qty(2-\kappa)\qty(Y_{\lambda\mu}\xi_\alpha+Y_{\alpha\lambda}\xi_\mu)+{3\kappa}g_{\alpha\mu}\nabla_\lambda\mathcal Z)\,^*\!R\tud{\lambda}{\nu\beta\rho}\\
    &-3\kappa g_{\mu\nu}\nabla_\lambda\mathcal Z\,^*\!R\tud{\lambda}{\alpha\beta\rho}+(3\kappa-2)\nabla_\mu\mathcal Z R^*_{\nu\alpha\beta\rho}\bigg]s^\alpha s^\beta \hat p^\mu \hat p^\nu \hat p^\rho
    \stackrel{!}{=}\mathcal O( \mathcal S^3),\label{developped_cst}
\end{split}
\end{align}
where\footnote{This notation $\mathcal M^{(1)}_{\alpha\beta\gamma\delta}$ will become clearer later on.}
\begin{align}
    \mathcal M^{(1)}_{\alpha\beta\gamma\delta}&\triangleq K_{\alpha\lambda}R\tud{\lambda}{\beta\gamma\delta}\label{M1}
\end{align}
gives rise to a quantity
\begin{align}
    \mathcal Q^{(2)}=\mathcal Q_R+M_{\alpha\beta\gamma\delta}S^{\alpha\beta}S^{\gamma\delta},\qquad M_{\alpha\beta\gamma\delta}\triangleq\, ^*\!N^*_{\alpha\beta\gamma\delta}.\label{explicit_quantity}
\end{align}
which is conserved up to second order in the spin parameter for the  MPD equations with spin-induced quadrupole \eqref{MPD}, \textit{i.e.} $\dot{\mathcal Q}^{(2)}=\mathcal O\qty(\mathcal S^3)$.
\end{proposition}

Our next goal with be to find a way to disentangle the $\kappa=1$ and the $\kappa\neq 1$ problems. Without loss of generality, we set
\begin{align}
N_{\alpha\beta\gamma\delta}={\triangleq N^\text{BH}_{\alpha\beta\gamma\delta}}+(\kappa-1) N^{\text{NS}}_{\alpha\beta\gamma\delta}, 
\end{align}
Because $\kappa$ is \textit{a priori} arbitrary, the constraint \eqref{developped_cst} turns out to be equivalent to the two independent equations
\begin{align}
    \begin{split}
    &\bigg[4\nabla_\mu N^{\text{BH}}_{\alpha\nu\beta\rho}-2\nabla_{[\alpha}\mathcal M^{(1)}_{\vert\mu\vert \nu]\beta\rho} +\qty(g_{\alpha\mu}Y_{\lambda\nu}-g_{\mu\nu}Y_{\lambda\alpha})\xi_\kappa\,^*\!R\tud{\lambda\kappa}{\beta\rho}\\
    &+\big(2 Y_{\alpha\mu}\xi_\lambda+\qty(Y_{\lambda\mu}\xi_\alpha+Y_{\alpha\lambda}\xi_\mu)+{3}g_{\alpha\mu}\nabla_\lambda\mathcal Z\big)\,^*\!R\tud{\lambda}{\nu\beta\rho}\\
    &-3 g_{\mu\nu}\nabla_\lambda\mathcal Z\,^*\!R\tud{\lambda}{\alpha\beta\rho}+\nabla_\mu\mathcal Z R^*_{\nu\alpha\beta\rho}\bigg]s^\alpha s^\beta \hat p^\mu \hat p^\nu \hat p^\rho
    =\mathcal O( \mathcal S^3)\label{developped_cst_BH}
    \end{split}
\end{align}
and
\begin{align}
    \begin{split}
    &\bigg[4\nabla_\mu N^{\text{NS}}_{\alpha\nu\beta\rho}-2\nabla_{[\alpha}\mathcal M^{(1)}_{\vert\mu\vert \nu]\beta\rho} +\qty(g_{\alpha\mu}Y_{\lambda\nu}-g_{\mu\nu}Y_{\lambda\alpha})\xi_\kappa\,^*\!R\tud{\lambda\kappa}{\beta\rho}\\
    &+\big(2 Y_{\alpha\mu}\xi_\lambda-\qty(Y_{\lambda\mu}\xi_\alpha+Y_{\alpha\lambda}\xi_\mu)+{3}g_{\alpha\mu}\nabla_\lambda\mathcal Z\big)\,^*\!R\tud{\lambda}{\nu\beta\rho}\\
    &-3 g_{\mu\nu}\nabla_\lambda\mathcal Z\,^*\!R\tud{\lambda}{\alpha\beta\rho}+3\nabla_\mu\mathcal Z R^*_{\nu\alpha\beta\rho}\bigg]s^\alpha s^\beta \hat p^\mu \hat p^\nu \hat p^\rho=\mathcal O( \mathcal S^3).\label{developped_cst_NS}
    \end{split}
\end{align}
In the continuation, we will refer to these two problems are respectively the ``black hole problem'' ($\kappa=1$) and the ``neutron star problem'' ($\kappa\neq 1$). Their resolutions are independent and will be addressed separately. Notice that the overall quasi-conserved quantity is given by
\begin{align}
\mathcal Q^{(2)}=\mathcal Q_\text{R}+\mathcal Q_{\text{BH}}+(\kappa-1)\,\mathcal Q_\text{NS}. \label{conserved_splitting}
\end{align}
The contributions $\mathcal Q_{\text{BH}}$ and $ \mathcal Q_\text{NS}$ can be directly computed from the corresponding $N_{\alpha\beta\gamma\delta}$ tensor through Eq. \eqref{explicit_quantity}. 

\section{Kerr covariant formalism: generalities}\label{sec:CBB}

In this section, we will show that the very structure of Kerr spacetime allows us to reduce the  differential constraint equations \eqref{linear_cst}-\eqref{developped_cst_BH}-\eqref{developped_cst_NS} to \textit{purely algebraic} relations. It is then possible to find a unique non-trivial solution to the black hole constraint \eqref{developped_cst_BH}, as will be demonstrated in Section \ref{sec:BH}. It also enables to provide an algebraic way for solving the  $\kappa\neq 1$ linear and quadratic problems (\textit{i.e.} Eq. \eqref{linear_cst} and \eqref{developped_cst_NS}, respectively), as will be discussed in Section \ref{sec:NS}.

\subsection{Covariant building blocks for Kerr}
In Kerr spacetime, the constraint equations \eqref{linear_cst}-\eqref{developped_cst_BH}-\eqref{developped_cst_NS} can be \textit{fully} expressed in terms of the basic tensors that live on the manifold (that is the metric $g_{\mu\nu}$, the Levi-Civita tensor $\epsilon_{\mu\nu\rho\sigma}$ and the Kronecker symbol $\delta^\mu_\nu$) and of three additional tensorial structures that we will refer to as \defining{Kerr's covariant building blocks}: the timelike Killing vector field $\xi^\mu$, the complex scalar
\begin{align}\label{defR}
    \mathcal R\triangleq r+ia\cos\theta
\end{align}
and the 2-form
\begin{align}
    N_{\alpha\beta}\triangleq -i G_{\alpha\beta\mu\nu}l^\mu n^\nu.
\end{align} Here, $l^\mu$ and $n^\nu$ are the two principal null directions of Kerr given in Eq. \eqref{kinnersley} and $G\tdu{\alpha\beta}{\gamma\delta}$ is (four times) the projector
\begin{align}
    G\tdu{\alpha\beta}{\gamma\delta}\triangleq 2\delta^{[\gamma}_\alpha\delta^{\delta]}_\beta-i\epsilon\tdu{\alpha\beta}{\gamma\delta}.
\end{align}
Notice that we have the property
\begin{equation}
N_{\alpha\beta}=\frac{2}{\xi^2} ( \nabla_{[\alpha}\mathcal R\xi_{\beta]^*}+i\nabla_{[\alpha}\mathcal R\xi_{\beta]}).     
\end{equation}
The Killing-Yano and Riemann tensors can be written algebraically in terms of these objects:
\begin{align}
    Y_{\alpha\beta}=-\frac{1}{2}\mathcal R N_{\alpha\beta}+c.c.,\qquad R_{\alpha\beta\gamma\delta}=M\Re\qty(\frac{3N_{\alpha\beta}N_{\gamma\delta}-G_{\alpha\beta\gamma\delta}}{\mathcal R^3}).
\end{align}
Moreover, they obey  the following closed differential relations,
\begin{align}
    i\nabla_\alpha\mathcal R= N_{\alpha\beta}\xi^\beta,\quad i\nabla_\gamma\qty(\mathcal R N_{\alpha\beta})=G_{\alpha\beta\gamma\delta}\xi^\delta,\quad i\nabla_\alpha\xi_\beta=-\frac{M}{2}\qty(\frac{N_{\alpha\beta}}{\mathcal R^2}-\frac{\bar N_{\alpha\beta}}{\bar{\mathcal{R}}^2}).\label{diff_to_alg}
\end{align}
All the derivatives appearing in the constraints can consequently be expressed in terms of purely algebraic relations between the covariant building blocks.

\subsection{Some identities}

Let us first derive some useful identities. Many of them can be found in \cite{floyd_phdthesis,Yasui:2011pr}. We have the algebraic identities
\begin{align}
    N_{\alpha\beta}N\tud{\beta}{\gamma}=-g_{\alpha\gamma},\quad N_{\alpha\beta}\bar N^{\alpha\beta}=0.
\end{align}
Notice that this first relation yields
\begin{align}
    N_{\alpha\beta}N^{\alpha\beta}=4.
\end{align}
Both $N_{\alpha\beta}$ and $G\tdu{\alpha\beta}{\gamma\delta}$ are self-dual tensors:
\begin{align}
    N^*_{\alpha\beta}=i N_{\alpha\beta},\qquad ^*G\tdu{\alpha\beta}{\gamma\delta}={G^*}\tdu{\alpha\beta}{\gamma\delta}=i G\tdu{\alpha\beta}{\gamma\delta}.
\end{align}
This leads to the relations
\begin{align}
    ^*R_{\alpha\beta\gamma\delta}=R^*_{\alpha\beta\gamma\delta}=-M\Im\qty(\frac{3N_{\alpha\beta}N_{\gamma\delta}-G_{\alpha\beta\gamma\delta}}{\mathcal R^3}),\qquad \bar N^*_{\alpha\beta}=-i\bar N_{\alpha\beta}.
\end{align}
Given the identities just derived, the only non-trivial contraction of the 2-form that can be written is
\begin{align}
    h_{\mu\nu}\triangleq N\tdu{\mu}{\alpha}\bar N_{\nu\alpha}.
\end{align}
It is a real, symmetric and traceless tensor:
\begin{align}
    h_{\mu\nu}=h_{(\mu\nu)}=\bar h_{\mu\nu},\qquad h^\mu_\mu=0.
\end{align}
Using the previous identities, one shows that
\begin{align}
    \mathcal Z=-\frac{1}{2}\Im\qty(\mathcal R^2).
\end{align}
This yields
\begin{align}
    \nabla_\alpha\mathcal Z=-\Re\qty(\mathcal R \xi^\lambda N_{\lambda\alpha}).
\end{align}
The Killing tensor can be written as
\begin{align}
    K_{\mu\nu}=\frac{1}{2}\qty(\Re\qty(\mathcal R^2)g_{\mu\nu}+\abs{\mathcal R^2}h_{\mu\nu}).
\end{align}
Its trace is simply
\begin{align}
    K&=2\Re\qty(\mathcal R^2).
\end{align}
Other useful identities include
\begin{align}
   \begin{split}
    &N_{\lambda\kappa}G\tud{\lambda\kappa}{\beta\rho}=4N_{\beta\rho},\qquad N_{\lambda\kappa}\bar G\tud{\lambda\kappa}{\beta\rho}=\bar N_{\lambda\kappa}G\tud{\lambda\kappa}{\beta\rho}=0,\\
    &\bar N_{\lambda\kappa}\bar G\tud{\lambda\kappa}{\beta\rho}=4\bar N_{\beta\rho},\qquad h_{\alpha\lambda}N\tud{\lambda}{\beta}=\bar N_{\alpha\beta}.\label{lastId}
   \end{split}
\end{align}

\subsection{Basis of contractions}

Our goal is now to rewrite Eqs. \eqref{linear_cst}-\eqref{developped_cst_BH}-\eqref{developped_cst_NS} as \textit{scalar} (that is, fully-contracted) equations involving only contractions between the Kerr covariant building blocks and the dynamical variables $s^\alpha$ and $\hat p^\alpha$. We define
\begin{align}
    \mathcal S^2&\triangleq s_\alpha s^\alpha,\qquad
    \mathcal P^2\triangleq -\hat p_\alpha \hat p^\alpha,\qquad
    \mathcal A\triangleq s_\alpha\hat p^\alpha.
\end{align}
We will naturally set $\mathcal P^2=1$ at the end of the computation, but we find useful to keep this quantity explicit in the intermediate algebra. Since $\mathcal A$ is unphysical, it will disappear from any physical expression but it will appear in intermediate computations. 

We further define the following quantities at least linear in either $\hat p^\mu$ or $s^\mu$, 

\begin{align}
    \begin{split}
    A&\triangleq N_{\lambda\mu}\xi^\lambda\hat p^\mu,\qquad
    B\triangleq N_{\alpha\mu}s^\alpha \hat p^\mu,\qquad
    C\triangleq N_{\lambda\alpha}\xi^\lambda s^\alpha,\qquad D\triangleq h_{\lambda\alpha}\xi^\lambda s^\alpha,\\
    E&\triangleq -\xi_\alpha \hat p^\alpha,\qquad
    E_s\triangleq-\xi_\alpha s^\alpha,\qquad
    F\triangleq h_{\lambda\mu}\xi^\lambda \hat p^\mu,\qquad G\triangleq h_{\alpha\mu}s^\alpha\hat p^\mu,\\
    H&\triangleq h_{\mu\nu}\hat p^\mu\hat p^\nu,\qquad
    I\triangleq h_{\alpha\beta} s^\alpha s^\beta. \label{ABCD}
    \end{split}
\end{align}
We further define: 
\begin{equation}
J \triangleq \xi^\alpha (h_{\alpha\beta}+g_{\alpha\beta})\xi^\beta \stackrel{.}{=} \frac{2a^2\sin^2\theta}{r^2+a^2 \cos^2\theta}  
\end{equation}
Notice that we can now use $J$ as a Kerr covariant substitute for $a$.

Because of the algebraic identities derived above, these scalars form a spanning set of scalars built from contractions among the Kerr covariant building blocks. Any higher order contraction between building blocks will reduce to a product of the ones provided in the above list with coefficients that may depend upon $M$, $a$ and $\mathcal R$. The quantities $A,B,C$ are complex while the others are real. Notice that we don't have to include $\xi^2$ in our basis of building blocks, since
\begin{align}
    \xi^2=-1+2M\Re\qty(\mathcal R^{-1}).
\end{align}
It can consequently be written in terms of the other quantities. In what follows, we will always consider quantities built from \eqref{ABCD}, the complex scalar $\mathcal R$, the mass $M$ and $J$.

\subsection{A $\mathbb Z_2$ grading}

It is possible to further restrict the number of combinations of the previously defined building blocks appearing in the equations. In order to achieve this goal, let us define a $\mathbb Z_2$ grading $\{\cdot\}$ as follows. We note that the determining equations for the covariant building blocks for Kerr from Eq. \eqref{defR} to Eq. \eqref{lastId} are invariant under the following $\mathbb Z_2$ grading: 
\begin{align}
    \begin{split}
        \{g_{\alpha\beta}\}&=\{M\}=\{x^\mu\}=\{\nabla_\mu\}=\{\mathcal R\}\\
        &=\{\mathcal Z\}=\{ G_{\alpha\beta\gamma\delta}\}=\{h_{\mu\nu}\}=\{K_{\mu\nu}\}=+1\\
        \{N_{\alpha\beta}\}&=\{\xi^\alpha\}=\{Y_{\alpha\beta}\}=-1.
    \end{split}
\end{align}
Further assigning $\{s^\mu\}=\{p^\mu\}=+1$, we deduce that 
\begin{align}
    \begin{split}
        \{A\}&=\{C\}=\{G\}=\{H\}=\{I\}=+1,\\ \{B\}&=\{D\}=\{E\}=\{E_s\}=\{F\}=-1. 
    \end{split}
\end{align}
Since the constraints \eqref{developped_cst_BH}, \eqref{developped_cst_NS} have grading $+1$, the odd quantities will have to be combined in pairs in order to build a solution to the constraint.

We define the $(s,p)^\pm$ grading of an expression as the $s$ number of $s^\alpha$ and $p$ number of $\hat p^\alpha$ factors in the expression with the sign $\pm$ indicating the $\mathbb Z_2$ grading. The complete list of the lowest $s+p=1$ and $s+p=2$ grading spanning elements is given in Table \ref{tab:spbasis}. The list of spanning elements of grading $(s,p)^\pm$ for $s+p \geq 3$ is obtained iteratively by direct product of the lower order basis elements. For example, the independent real terms of grading $(2,1)^+$ are obtained from $(2,0)^+ \times (0,1)^+$, $(2,0)^- \times (0,1)^-$, $(1,1)^+ \times (1,0)^+$ and $(1,1)^- \times (1,0)^-$ with duplicated elements suppressed.

\begin{table}[!htb]
    \centering
    \renewcommand{\arraystretch}{1.3}
    \begin{tabular}{|c|c|c|}\hline
     $(s,p)^\pm$ & \textsc{Spanning set} & \rule{0pt}{13pt}\shortstack{\textsc{Real dimension}} \\ \hline
        \rule{0pt}{13pt}$(1,0)^+$ &  $C$ & 2 \\
        $(1,0)^-$ & $D$,\; $E_s$ &2 \\
        $(0,1)^+$& $A$ & 2\\
        $(0,1)^-$& $E$,\; $F$ &2 \\
    $(2,0)^+$ &$I$,\; $\mathcal S$,\; $(1,0)^+ \times (1,0)^+$,\; $(1,0)^- \times (1,0)^-$ & 8\\
$(2,0)^-$& $(1,0)^+ \times (1,0)^-$& 4\\
$(1,1)^+$& $G$, \;$\mathcal A$, \; $(1,0)^+ \times (0,1)^+$ , \; $(1,0)^- \times (0,1)^-$  & 10\\
$(1,1)^-$& $B$, \; $(1,0)^+ \times (0,1)^-$,\; $(1,0)^- \times (0,1)^+$   & 10\\
$(0,2)^+$& $H$,\; $\mathcal P$,\;  $(0,1)^+ \times (0,1)^+$,\; $(0,1)^- \times (0,1)^-$  &8\\
$(0,2)^-$&  $(0,1)^+ \times (0,1)^-$ & 4\\        
\hline    \end{tabular}
    \caption{Spanning set of elements with $(s,p)^\pm$ grading such that $s+p=1$ or $s+p=2$ and their real dimensions.}
    \label{tab:spbasis}
\end{table}

Of prime importance for solving the linear and quadratic constraint equations will be the elements of gradings $(2,2)^+$ and $(2,3)^+$. Their respective spanning sets contain $118$ and $284$ elements, which are explicitly listed in the \textit{Mathematica} notebooks appended to this chapter, available at \url{https://github.com/addruart/generalizedCarterConstant}.

\subsection{The $\alpha$-$\omega$ basis}

We note that the covariant building blocks all depend on $\mathcal R$ through real and imaginary parts of expressions containing fractions of $\mathcal R$ and $\bar{\mathcal R}$. We find therefore natural to define the objects ($n,p \in\mathbb Z$ and $K=1,A,B,C,\ldots,J$ or any combination of these objects):
\begin{align}
    \alpha_K^{(n,p)}\triangleq\Re\qty(\frac{K\bar{\mathcal R}^n}{\mathcal R^p}),\qquad\omega_K^{(n,p)}\triangleq\Im\qty(\frac{K\bar{\mathcal R}^n}{\mathcal R^p}).
\end{align}
They satisfy the following properties:
\begin{align}
    \begin{split}
    \alpha_{iK}^{(n,p)}&=-\omega_K^{(n,p)},\qquad
    \omega_{iK}^{(n,p)}=\alpha_K^{(n,p)},\\
    \alpha_{\bar K}^{(n,p)}&=\alpha_K^{(-p,-n)},\qquad
    \omega_{\bar K}^{(n,p)}=-\omega_K^{(-p,-n)}.\label{ao_prop}
    \end{split}
\end{align}
Moreover, one has
\begin{align}
    \begin{split}
    &\abs{\mathcal R}^2\alpha_K^{(n,p)}=\alpha_K^{(n+1,p-1)},\qquad
    \Re\qty(\mathcal R^2)\alpha_K^{(n,p)}=\frac{1}{2}\qty[\alpha_K^{(n,p-2)}+\alpha_K^{(n+2,p)}],\\
    &
    \alpha_{\mathcal R^k K}^{(n,p)}=\alpha_K^{(n,p-k)},\qquad
    \alpha_{\bar{\mathcal R}^kK}^{(n,p)}=\alpha_K^{(n+k,p)}.
    \end{split}
\end{align}
The same properties hold with $\omega$ instead of $\alpha$. Finally, $\alpha$ and $\omega$ are linear in their subscript argument with respect to real-valued functions.
Let us denote $\ell^\mu=\hat p^\mu$ or $s^\mu$. Then, for any $T\triangleq T_{\mu_1\ldots\mu_k}\ell^{\mu_1}\ldots\ell^{\mu_k}$, we define the operator $\hat\nabla$ as
\begin{align}
    \hat \nabla T\triangleq \hat p^\lambda\nabla_\lambda\qty(T_{\mu_1\ldots\mu_k})\ell^{\mu_1}\ldots\ell^{\mu_k}.
\end{align}
Making use of the identities
\begin{align}
    \hat\nabla \mathcal R^n=in\mathcal R^{n-1}A,\qquad\hat\nabla\bar{\mathcal R}=-in\bar{\mathcal R}^{n-1}\bar A,
\end{align}
we get the following relations:
\begin{align}
    \begin{split}
    \hat\nabla\alpha_K^{(n,p)}&=\alpha_{\hat\nabla K}^{(n,p)}+n\omega_{K\bar A}^{(n-1,p)}+p\omega_{KA}^{(n,p+1)},\\
    \hat\nabla\omega_K^{(n,p)}&=\omega_{\hat\nabla K}^{(n,p)}-n\alpha_{K\bar A}^{(n-1,p)}-p\alpha_{KA}^{(n,p+1)}.\label{deriv_ao}
    \end{split}
\end{align}

We use dimensions such that $G=c=1$. Given the large amount of definitions just provided, we find useful to summarize the mass dimensions $[\cdot]$ of all quantities in order to keep track of the powers of the mass $M$ that can arise. We have the following mass dimensions  $[\nabla_\mu]=-1$, $[g_{\mu\nu}]=[\xi^\alpha]=[N_{\alpha\beta}]=[G_{\alpha\beta\mu\nu}]=0$, $[x^\mu]=[M]=[\mathcal R]=[Y_{\alpha\beta}]=1$ and $[K_{\alpha\beta}]=[K]=2$. We deduce 
\begin{eqnarray}
[X]=0,\qquad [\alpha_X^{(n,p)}]=[\omega_X^{(n,p)}]=n-p, 
\end{eqnarray}
where $X$ is any function of the set $A,B,C,D,E,E_s,F,G,H,I,J$ defined in Eqs. \eqref{ABCD}. 

\section{Kerr covariant formalism: reduction of the constraints}\label{sec:Reduction}

\subsection{Linear constraint with $X^\mu=0$}
\label{sec:linCBB}
We will now reduce the linear constraint equations in the case where $Y_{\mu\nu}$ is the Kerr Killing-Yano tensor. Using the explicit form of $\mathcal L_{\mu\nu}$ as defined in \eqref{L} and expressing the quantity $   \mathcal L^*_{\mu\nu} Y^{\mu\nu}$ in terms of covariant building blocks we find after evaluation
\begin{align}
    \mathcal L^*_{\mu\nu} Y^{\mu\nu}=0,
\end{align}
and therefore the $[2,2]$ constraint is automatically fulfilled. The $[2,4]$ constraint can be rewritten
\begin{align}
    -\mu Y_{\alpha\beta}p^{[\alpha}D^{\beta]*\lambda}\hat p_\lambda+Y_{\alpha\beta}p^{[\alpha}\mathcal L^{\beta]*\lambda}\hat p_\lambda=\mathcal O\qty(\mathcal S^3).
\end{align}
A direct computation shows that
\begin{align}
    \begin{split}
        \mu Y_{\alpha\beta}p^{[\alpha}D^{\beta]*\lambda}\hat p_\lambda&=-3M\qty(\mathcal A H+\mathcal P^2 G)\omega_B^{(1,3)},\\
    Y_{\alpha\beta}p^{[\alpha}\mathcal L^{\beta]*\lambda}\hat p_\lambda&=3\kappa M\qty(\mathcal A H+\mathcal P^2 G)\omega_B^{(1,3)}.
    \end{split}
\end{align}
Using these identities and defining $\kappa\triangleq 1+\delta\kappa$, the $[2,4]$ constraint takes the very simple form
\begin{align}
    -3M\delta\kappa\qty(\mathcal AH+\mathcal P^2G)\omega_B^{(1,3)}=\mathcal O\qty(\mathcal S^3).
\end{align}
It is automatically fulfilled for the test body being a Kerr black hole, because $\delta\kappa=0$ in this case. However, if the test body is a neutron star, $\delta\kappa\neq 0$ and the $[2,4]$ constraint is not obeyed anymore. Therefore, $\mathcal Q_Y$ is not anymore a constant of motion at second order in the spin in the NS case.

A way to enable the $[2,4]$ constraint to be solvable in the neutron star case is to supplement the Ansatz for the conserved quantity with a term
\begin{align}
    \mathcal Q^{(1)}_\text{NS}=\delta\kappa M_{\alpha\beta\mu\gamma\delta}S^{\alpha\beta}S^{\gamma\delta}p^\mu.
\end{align}
The conservation equation will then acquire a correction given by
\begin{align}
    \dot{\mathcal{Q}}^{(1)}_\text{NS}&=\delta\kappa v^\lambda\nabla_\lambda\qty(M_{\alpha\beta\mu\gamma\delta}S^{\alpha\beta}S^{\gamma\delta}p^\mu)\\
    &=\delta\kappa\hat p^\lambda\nabla_\lambda M_{\alpha\beta\mu\gamma\delta}S^{\alpha\beta}S^{\gamma\delta}p^\mu+\mathcal O\qty(\mathcal S^3)\\
    &=4\delta\kappa\hat p^\lambda\nabla_\lambda N_{\alpha\beta\mu\gamma\delta}s^\alpha\hat p^\beta s^\gamma\hat p^\delta p^\mu+\mathcal O\qty(\mathcal S^3)
\end{align}
where $N_{\alpha\beta\mu\gamma\delta}={}^\star M_{\alpha\beta\mu\gamma\delta}^\star$. In our scalar notation, it corresponds to supplement the $[2,4]$ constraint with a term $4\delta\kappa\hat \nabla N$, with $N$ being of grading $[2,3]$. The constraint to be solved then takes the simple form
\begin{align}
    \hat\nabla N=\frac{3M}{4}\qty(\mathcal A H+\mathcal P^2 G)\omega_B^{(1,3)}.
\end{align}

It is useful to summarize the discussion by the two following statements:
\begin{main_result}
Rüdiger's linear invariant $\mathcal Q_Y=Y_{\alpha\beta}^* S^{\alpha\beta}$ is still conserved for the MPTD equations at second order in the spin magnitude for spin-induced quadrupoles, \textit{i.e.} $\dot{\mathcal{Q}}_Y=\mathcal O\qty(\mathcal S^3)$ provided that $\delta\kappa=0$, \textit{i.e.} if the test body possesses the multipole structure of a black hole.
\end{main_result}

\begin{preliminary_result}
Any tensor $N_{\alpha\beta\mu\gamma\delta}$ possessing the algebraic symmetries $N_{\alpha\beta\mu\gamma\delta}=N_{\gamma\delta\mu\alpha\beta}=N_{[\alpha\beta]\mu\gamma\delta}=N_{\alpha\beta\mu[\gamma\delta]}$ and satisfying the constraint equation
\begin{align}
\hat\nabla N=\frac{3M}{4}\qty(\mathcal A H+\mathcal P^2 G)\omega_B^{(1,3)}\label{lin_cst}
\end{align}
will give rise to a quantity
\begin{align}
    \mathcal Q^{(1)}=\mathcal Q_Y+\delta\kappa M_{\alpha\beta\mu\gamma\delta}S^{\alpha\beta}S^{\gamma\delta}p^\mu,\qquad M_{\alpha\beta\mu\gamma\delta}=\,^*\!N^*_{\alpha\beta\mu\gamma\delta}
\end{align}
which is conserved up to second order in the spin parameter for the spin-induced quadrupole MPTD equations \eqref{MPD}, \textit{i.e.} $\dot{\mathcal{Q}}^{(1)}=\mathcal O\qty(\mathcal S^3)$, regardless to the value taken by $\delta\kappa$.
\end{preliminary_result}

\subsection{Quadratic constraint}

\subsubsection{Some identities}
Before going further on, it is useful to notice that all the covariant building blocks combinations that will appear in our equations will not be linearly independent. Actually, a direct computation shows that
\begin{subequations}\label{lin_dep_terms}
\begin{align}
    &2\abs{B}^2+2\mathcal AG+\mathcal P^2I-\mathcal S^2H=0,\\
    &\qty(\mathcal A H+\mathcal P^2G)\omega^{(1,3)}_C=\qty(\mathcal AG+\mathcal P^2I)\omega^{(1,3)}_A-\qty(\mathcal A F+\mathcal P^2 D)\omega^{(1,3)}_B\nonumber\\
    &\quad+\qty(\mathcal A^2+\mathcal P^2\mathcal S^2)\omega^{(1,3)}_{\bar A}+\qty(\mathcal AE+\mathcal P^2E_s)\omega^{(1,3)}_{\bar B},\\
    &\omega^{(1,3)}_{\bar AB^2}=-\abs{B}^2\omega^{(1,3)}_A-\qty(\mathcal AF-EG+E_sH+\mathcal P^2D)\omega_B^{(1,3)}.
\end{align}
\end{subequations}
Moreover, let us mention the identities
\begin{subequations}
    \begin{align}
    &\omega^{(0,k)}_K\alpha_L^{(n,p)}=\frac{1}{2}\qty[\omega_{KL}^{(n,p+k)}-\omega_{\bar KL}^{(n-k,p)}],\\
    &\Re\qty(\mathcal R^2)\Im\qty(\frac{K}{\mathcal R^4})=\frac{1}{2}\qty(\omega^{(0,2)}_K+\omega^{(2,4)}_K),\\
    &\abs{\mathcal R}^2\Im\qty(\frac{K}{\mathcal R^4})=\omega^{(1,3)}_K,
    \end{align}
\end{subequations}
which will be useful in the forthocoming computations.

\subsubsection{Reducing the $\mathcal M^{(1)}$ contribution}
Our goal is here to compute the contribution 
\begin{align}
    \text{DM}\triangleq 2\nabla_{[\alpha|}\mathcal M^{(1)}_{\mu|\nu]\beta\rho}s^\alpha s^\beta\hat p^\mu\hat p^\nu\hat p^\rho
\end{align}
in some details, as a proof of principle of the computations to follow, which will not be developed in full details. Noticing the identity
\begin{align}
    \nabla_\mu K_{\alpha\beta}=2\epsilon_{\lambda\rho\mu(\alpha}Y\tud{\lambda}{\beta)}\xi^\rho,\label{DK}
\end{align}
we get
\begin{align}
    \nabla_\mu \mathcal M^{(1)}_{\alpha\nu\beta\rho} &=K_{\alpha\lambda}\nabla_\mu R\tud{\lambda}{\nu\beta\rho}-\nabla_\nu\mathcal Z R^*_{\mu\alpha\beta\rho}+\qty(Y_{\mu\nu}\xi_\lambda+g_{\mu\nu}\nabla_\lambda\mathcal{Z}+Y_{\lambda\mu}\xi_\nu)R\tud{*\lambda}{\alpha\beta\rho}\nonumber\\
    &\quad-\qty(2\xi_\nu Y_{\lambda\alpha}-2\xi_\lambda Y_{\nu\alpha}+g_{\alpha\nu}\nabla_\lambda\mathcal Z)R\tud{*\lambda}{\mu\beta\rho}\label{DM}\\
    &\quad-\qty(2g_{\mu\nu}Y_{\kappa\alpha}-g_{\alpha\nu} Y_{\kappa\mu})\xi_\lambda R\tud{*\lambda\kappa}{\beta\rho}.\nonumber
\end{align}
Making use of Eq. \eqref{DM}, one can show that
\begin{align}
    \text{DM}&=\bigg[2K_{\mu\lambda}\nabla_{[\alpha}R\tud{\lambda}{\nu]\beta\rho}-\nabla_\nu\mathcal Z R^*_{\alpha\mu\beta\rho}-\qty(2\xi_\nu Y_{\lambda\mu}+g_{\mu\nu}\nabla_\lambda\mathcal Z)R\tud{*\lambda}{\alpha\beta\rho}\nonumber\\
    &\quad+\qty(Y_{\lambda\alpha}\xi_\mu+\xi_\alpha Y_{\lambda\mu}+g_{\mu\alpha}\nabla_\lambda\mathcal Z)R\tud{*\lambda}{\nu\beta\rho}\\
    &\quad-\qty(g_{\mu\alpha}Y_{\kappa\nu}-g_{\mu\nu}Y_{\kappa\alpha})\xi_\lambda R\tud{*\lambda\kappa}{\beta\rho}\bigg]s^\alpha s^\beta\hat p^\mu\hat p^\nu\hat p^\rho.\nonumber
\end{align}
Using the various identities derived above, the relations of Appendix \ref{app:CBB_identities} and performing some simple algebra, one can express this contribution in terms of linearly independent quantities as
\begin{align}
    \text{DM}&=\frac{M}{4}\qty(\mathcal A^2+\mathcal P^2\mathcal S^2)\qty(5\omega^{(0,2)}_A+4\omega^{(1,3)}_{\bar A}+3 \omega^{(2,4)}_A)\nonumber\\
    &\quad+\frac{M}{2}\qty(\mathcal AE+\mathcal P^2E_s)\qty(\omega^{(0,2)}_B-\omega^{(1,3)}_{\bar B}-3\omega^{(2,4)}_B)\\
    &\quad-\frac{3M}{2}\qty(2\mathcal A F-EG+E_sH+2\mathcal P^2D)\omega^{(1,3)}_B+\frac{9M}{4}\omega^{(0,2)}_{AB^2}+\frac{15M}{4}\omega^{(2,4)}_{AB^2}.\nonumber
\end{align}

\subsubsection{The black hole constraint equation}
Making use of the notations introduced above, the constraint equation \eqref{developped_cst_BH} can be written as
\begin{align}
    4\hat\nabla N^\text{BH}+\text{DM}-\Upsilon=\mathcal O\qty(\mathcal S^3),\label{scalar_like_constraint}
\end{align}
where
\begin{align}
    \Upsilon&\triangleq\bigg[\qty(g_{\alpha\mu}Y_{\lambda\nu}-g_{\mu\nu}Y_{\lambda\alpha})\xi_\kappa\,^*\!R\tud{\lambda\kappa}{\beta\rho}+\qty(2 Y_{\alpha\mu}\xi_\lambda+\qty(Y_{\lambda\mu}\xi_\alpha+Y_{\alpha\lambda}\xi_\mu)+3g_{\alpha\mu}\nabla_\lambda\mathcal Z)\,^*\!R\tud{\lambda}{\nu\beta\rho}\nonumber\\
    &\quad-3 g_{\mu\nu}\nabla_\lambda\mathcal Z\,^*\!R\tud{\lambda}{\alpha\beta\rho}+\nabla_\mu\mathcal Z R^*_{\nu\alpha\beta\rho}\bigg]s^\alpha s^\beta \hat p^\mu \hat p^\nu \hat p^\rho.
\end{align}
Using the scalar basis introduced above and the identities \eqref{lin_dep_terms} yields
\begin{align}
    \Upsilon&=\frac{M}{2}\qty(\mathcal A^2+\mathcal P^2\mathcal S^2)\qty[3\omega^{(0,2)}_A{+}\omega^{(1,3)}_{\bar A}]+\frac{M}{2}\qty(\mathcal AE+\mathcal P^2E_s)\qty[8\omega^{(0,2)}_B-3\omega^{(1,3)}_{\bar B}]\nonumber\\
    &\quad+\frac{3M}{2}\omega^{(0,2)}_{AB^2}+\frac{9M}{2}\abs{B}^2\omega^{(1,3)}_A-\frac{3M}{2}\qty(\mathcal AF+\mathcal P^2D)\omega^{(1,3)}_B.
\end{align}
In summary, Eq. \eqref{scalar_like_constraint} can be written as 
\begin{align}
    \hat \nabla N^\text{BH}=\Upsilon_\text{BH},\label{BH_scalar_constraint}
\end{align}
with the source term
\begin{align}
\Upsilon_\text{BH}=\frac{\Upsilon-\text{DM}}{4}.
\end{align}

\subsubsection{The neutron star constraint equation}
Repeating the very same procedure, the neutron star constraint \eqref{developped_cst_NS} reduces to the scalar-like equation
\begin{align}
  \hat\nabla N^\text{NS}=\Upsilon_\text{NS}
    \label{source_BH}
\end{align}
with source
\begin{align}
\Upsilon_\text{NS}&=-\frac{3M}{16}\qty(\mathcal A^2+\mathcal P^2\mathcal S^2)\qty(\omega_A^{(0,2)}+\omega_A^{(2,4)}+2\omega_{\bar A}^{(1,3)})\nonumber\\
&\quad+\frac{3M}{8}\qty(\mathcal AE+\mathcal P^2E_s)\qty(\omega_B^{(0,2)}+\omega_B^{(2,4)})-\frac{15M}{16}\qty(\omega_{AB^2}^{(0,2)}+\omega_{AB^2}^{(2,4)})\label{source_NS}\\
&\quad-\frac{15M}{8}\omega_{\bar AB^2}^{(1,3)}-\frac{3}{4}M\qty(\mathcal AF+\mathcal P^2D)\omega_B^{(1,3)}.\nonumber
\end{align}

\section{Solution for the quadratic invariant in the black hole case}
\label{sec:BH}

We will now try to find a quadratic conserved quantity for the $\delta\kappa=0$ case. This corresponds to find a solution to the black hole constraint equation \eqref{BH_scalar_constraint}. In order to reach this goal, we will postulate an Ansatz for the fully-contracted quantity $N$ appearing in the left-hand side of Eq. \eqref{BH_scalar_constraint} and then use the covariant building blocks formulation to constrain the Ansatz coefficients. 

Notice that we proceed here by postulating an Ansatz because we known from the start what the solution will look like; it allows to present in only a few pages a fully analytical derivation of the conserved quantity. However, if it has not been the case, we could have proceeded from scratch, by writing explicitly all the terms possessing the good grading to appear in the conserved quantity, and then fixing the coefficients by numerical evaluation. This procedure will be applied in the neutron star case, see next section for more details.

\subsection{The Ansatz}
Let us consider the following Ansatz
\begin{align}
    N^\text{BH}_{\alpha\beta\gamma\delta}\triangleq\sum_{A=1}^4 \Lambda_A N_{\alpha\beta\gamma\delta}^{(A)}\label{ansatz}
\end{align}
where $\Lambda_A$ are arbitrary coefficients and where
\begin{align}
    N_{\alpha\beta\gamma\delta}^{(A)}\triangleq \,^*\!\mathcal M^{*(A)}_{\alpha\beta\gamma\delta}.
\end{align}
The quantity $\mathcal M^{(1)}_{\alpha\beta\gamma\delta}$ has been defined in Eq. \eqref{M1}, and we introduce
\begin{align}
    \mathcal M^{(2)}_{\alpha\beta\gamma\delta}&\triangleq Y\tud{\lambda}{\alpha}Y\tud{\sigma}{\gamma}R_{\lambda\beta\sigma\delta},\qquad
        \mathcal M^{(3)}_{\alpha\beta\gamma\delta}\triangleq g_{\alpha\gamma}\xi_\beta\xi_\delta,\qquad
    \mathcal M^{(4)}_{\alpha\beta\gamma\delta}\triangleq g_{\alpha\gamma}g_{\beta\delta}\xi^2.
\end{align}

Using the identities derived in Appendix \ref{app:CBB_identities}, one can show that the directional derivatives of the $N^{(A)}$ are given by
\begin{subequations}
\begin{align}
    \hat\nabla N^{(1)}&=-\frac{M}{4}\qty(\mathcal A^2+\mathcal P^2\mathcal S^2)\qty(\omega^{(0,2)}_A+2\omega^{(1,3)}_{\bar A}+3\omega^{(2,4)}_A)\nonumber\\
    &\quad+\frac{M}{2}\qty(\mathcal AE+\mathcal P^2E_s)\qty(5\omega^{(0,2)}_B-2\omega^{(1,3)}_{\bar B}+3\omega^{(2,4)}_B)\quad+\frac{9M}{2}\abs{B}^2\omega^{(1,3)}_A\nonumber\\
    & +\frac{3M}{2}\qty(\mathcal AF-EG+E_sH+\mathcal P^2D)\omega^{(1,3)}_B-\frac{9M}{4}\omega^{(0,2)}_{AB^2}-\frac{15M}{4}\omega^{(2,4)}_{AB^2},\label{DN1_Kerr}\\
    \hat\nabla N^{(2)}&=-\frac{M}{4}\qty(\mathcal A^2+\mathcal P^2\mathcal S^2)\omega^{(0,2)}_A+\frac{M}{2}\qty(\mathcal A E+\mathcal P^2 E_s)\omega^{(0,2)}_B-\frac{3M}{4}\omega^{(0,2)}_{AB^2},\\
    \hat\nabla N^{(3)}&=\frac{M}{2}\qty[\qty(\mathcal S^2\mathcal P^2+\mathcal A^2)\omega^{(0,2)}_A+\qty(\mathcal A E+\mathcal P^2 E_s)\omega^{(0,2)}_B],\label{DN3_Kerr}\\
    \hat\nabla N^{(4)}&=M\qty(\mathcal S^2\mathcal P^2+\mathcal A^2)\omega^{(0,2)}_A.\label{DN4_Kerr}
\end{align}
\end{subequations}

\subsection{Solution to the constraint}

We will now look for a solution to the black hole constraint equation \eqref{BH_scalar_constraint} using the Ansatz \eqref{ansatz}, \textit{i.e.} we are seeking for specific values of the parameters $\Lambda_A$ such that Eq. \eqref{scalar_like_constraint} is fulfilled. More explicitly, one therefore requires
\begin{align}
    \sum_{A=1}^4\Lambda_A \text{DN}^{(A)}-\Upsilon_\text{BH}\stackrel{!}{=}0.
\end{align}
The left-hand side of this equation takes the form of a first order polynomial, homogeneous in the ten \textit{linearly independent} elements (as it can be shown through a direct computation)
\begin{align}
    \begin{split}
    &\omega^{(0,2)}_A,\quad\omega^{(0,2)}_B,\quad\omega^{(0,2)}_{AB^2},\quad\omega^{(1,3)}_A,\quad\omega^{(1,3)}_B,\\
    &\omega^{(1,3)}_{\bar A},\quad\omega^{(1,3)}_{\bar B},\quad\omega^{(2,4)}_A,\quad\omega^{(2,4)}_B,\quad\omega^{(2,4)}_{AB^2}.
    \end{split}
\end{align}
Because all the combinations of these elements implied in the constraint equation are linearly independent, all the coefficients appearing in front of these expressions   should vanish independently. 

\begin{table}[b]
    \centering
    \renewcommand{\arraystretch}{1.3}
    \begin{tabular}{|c|c|c|c|c|c|c|c|c|c|c|}
    \hline
 &\rule{0pt}{15pt} $\omega^{(0,2)}_A$ & $\omega^{(0,2)}_B$& $\omega^{(0,2)}_{AB^2}$ & $\omega^{(1,3)}_A$  & $\omega^{(1,3)}_B$ & $\omega^{(1,3)}_{\bar A}$ & $\omega^{(1,3)}_{\bar B}$ & $\omega^{(2,4)}_A$ & $\omega^{(2,4)}_B$ & $\omega^{(2,4)}_{AB^2}$\\
 \hline
 $\text{DN}^{(1)}$& \checkmark& \checkmark&\checkmark &\checkmark &\checkmark &\checkmark &\checkmark &\checkmark &\checkmark &\checkmark \\\hline
 $\text{DN}^{(2)}$&\checkmark &\checkmark & \checkmark & & & & & & & \\\hline
 $\text{DN}^{(3)}$& \checkmark & \checkmark & & & & & & & & \\\hline
 $\text{DN}^{(4)}$& \checkmark & & & & & & & & & \\\hline
 $\Upsilon_\text{BH}$& \checkmark &\checkmark  &\checkmark  &\checkmark  &\checkmark  &\checkmark  &\checkmark  & \checkmark & \checkmark & \checkmark\\\hline
\end{tabular}
    \caption{Structure of the distribution of the different types of contractions in the various contribution to the black hole constraint equation.}
    \label{tab:terms}
\end{table}

All the terms do not appear in all the contributions, as depicted in Table \ref{tab:terms}. In order to fix the values of the Ansatz coefficients, let us proceed along the following sequence:
\begin{itemize} 
    \item  {\textbf{$\omega^{(1,3)}_A$ term}:} this contribution reads
    \begin{align}
        3M\qty(6\Lambda_1-\frac{3}{2})\abs{B}^2\omega^{(1,3)}_A.
    \end{align}
    One therefore requires
    \begin{align}
      \Lambda_1=\frac{1}{4}\label{L1_kappa}\, .
    \end{align}
    \item  {\textbf{$\omega^{(1,3)}_{\bar A}$, $\omega^{(1,3)}_{\bar B}$, $\omega^{(1,3)}_B$, $\omega^{(2,4)}_A$, $\omega^{(2,4)}_B$ and $\omega^{(2,4)}_{AB^2}$ terms}:} their coefficients consistently vanish when \eqref{L1_kappa} is fulfilled.
    \item  {\textbf{$\omega^{(0,2)}_{AB^2}$ term}:} this contribution reads
    \begin{align}
        3M\qty(-3\Lambda_1-\Lambda_2+\frac{1}{4})\omega^{(0,2)}_{AB^2}.
    \end{align}
    Using \eqref{L1_kappa}, this yields
    \begin{align}
       \Lambda_2=-\frac{1}{2}.\label{L2}
    \end{align}
    \item  {\textbf{$\omega^{(1,3)}_{\bar A}$, $\omega^{(1,3)}_{\bar B}$ and $\omega^{(1,3)}_B$ terms:}} their coefficients consistently vanish when Eqs. \eqref{L1_kappa} and \eqref{L2} are fulfilled.
    \item  {\textbf{$\omega^{(0,2)}_{B}$ term}:} this contribution reads
    \begin{align}
        M\qty(\mathcal AE+\mathcal P^2 E_s)\qty(10\Lambda_1+2\Lambda_2+2\Lambda_3-\frac{7}{2})\omega^{(0,2)}_B.
    \end{align}
    Using \eqref{L1_kappa} and \eqref{L2}, this yields
    \begin{align}
       \Lambda_3=1.\label{L3}
    \end{align}
    \item  {\textbf{$\omega^{(0,2)}_{A}$ term}:} this contribution reads
    \begin{align}
        M\qty(\mathcal A^2+\mathcal P^2 \mathcal S^2)\qty(-\Lambda_1-\Lambda_2+2\Lambda_3+4\Lambda_4-\frac{1}{4})\omega^{(0,2)}_A.
    \end{align}
    Using Eqs. \eqref{L1_kappa}, \eqref{L2} and \eqref{L3}, this finally yields
    \begin{align}
      \Lambda_4=-\frac{1}{2}.\label{L4}
    \end{align}
\end{itemize}
In conclusion, the Ansatz \eqref{ansatz} gives a coherent solution to the constraint equation \eqref{BH_scalar_constraint} only if 
\begin{align}
    \Lambda_1=\frac{1}{4},\quad\Lambda_2=-\frac{1}{2},\quad\Lambda_3=1,\quad\Lambda_4=-\frac{1}{2}.\label{coeff_values}
\end{align}
More explicitly, it corresponds to set
\begin{align}
    M_{\alpha\beta\gamma\delta}=g_{\alpha\gamma}\qty(\xi_\beta \xi_\delta-\frac{1}{2}g_{\beta\delta}\xi^2)-\frac{1}{2}Y\tdu{\alpha}{\lambda}\qty(Y\tdu{\gamma}{\kappa}R_{\lambda\beta\kappa\delta}+\frac{1}{2}Y\tdu{\lambda}{\kappa}R_{\kappa\beta\gamma\delta}).\label{sol_M}
\end{align}

\subsection{Uniqueness of the solution}

We now address the uniqueness to the non-trivial solution \eqref{sol_M} to the constraint \eqref{BH_scalar_constraint} derived above.
If one adds an additional piece to our Ansatz, it will satisfy an homogeneous equation since all source terms have been cancelled by the Ansatz. Demonstrating uniqueness of the non-trivial solution \eqref{sol_M} therefore amounts to prove that 
\begin{align}
    \nabla_{(\mu}N_{\nu|(\alpha\beta)|\rho)}=0\label{homogeneous_eqn}
\end{align}
does not admit any non-trivial solution in Kerr spacetime. We call such a tensor field a Young tableau $\qty{2,2}$ Killing tensor. A trivial Killing tensor is defined as a Killing tensor which is given by a cross-product. Such a trivial Killing tensor would add to the quadratic conserved quantity a product of conserved quantities that are already defined. There is only trivial Killing tensor of symmetry type $\qty{2,2}$, namely $N_{\mu\nu\alpha\beta}=Y_{\mu\nu}Y_{\alpha\beta}$ which correspond to add the product $(\mathcal Q_Y)^2$ to the quadratic conserved quantity. 
We checked explicitly by solving the partial differential equations analytically using a \textit{Mathematica} notebook that no non-trivial such tensor exists in a perturbative series expansion in $a$ around $a=0$ assuming that it only depends upon $r$ and $\theta$.

\subsection{Summary of the results}
Let us summarize the results we have obtained about the quadratic invariants. Our discussion can be compactified in the two following propositions:
\begin{main_result}
The quadratic invariant 
\boxedeqn{
\begin{split}
\mathcal Q^{(2)}_\text{BH}&=\mathcal Q_R+\bigg[g_{\alpha\gamma}\qty(\xi_\beta \xi_\delta-\frac{1}{2}g_{\beta\delta}\xi^2)\\
&\quad-\frac{1}{2}Y\tdu{\alpha}{\lambda}\qty(Y\tdu{\gamma}{\kappa}R_{\lambda\beta\kappa\delta}+\frac{1}{2}Y\tdu{\lambda}{\kappa}R_{\kappa\beta\gamma\delta})\bigg]S^{\alpha\beta}S^{\gamma\delta}\label{quadratic_rudiger}
\end{split}
}{Generalized quadratic invariant}
is conserved for the MPTD equations at second order in the spin magnitude for spin-induced quadrupole, \textit{i.e.} $\dot{\mathcal{Q}}^{(2)}_\text{BH}=\mathcal O\qty(\mathcal S^3)$ provided that $\delta\kappa=0$, \textit{i.e.} if the test body possesses the multipole structure of a black hole. Here, $\mathcal Q_R=  K_{\mu\nu}p^\mu p^\nu+L_{\mu\nu\rho}S^{\mu\nu}p^\rho$ with $L_{\mu\nu\rho}$ given in Eq. \eqref{R2} is Rüdiger's quadratic invariant \cite{doi:10.1098/rspa.1983.0012,Compere:2021kjz}.
\end{main_result}

\begin{preliminary_result}
Any tensor $N_{\alpha\beta\gamma\delta}$ possessing the same algebraic symmetries than the Riemann tensor and satisfying the constraint equation
\begin{align}
\hat\nabla N=\Upsilon_\text{NS},
\end{align}
where the source term $\Upsilon_\text{NS}$ is given in Eq. \eqref{source_NS}
will give rise to a quantity
\begin{align}
    \mathcal Q^{(2)}=\mathcal Q^{(2)}_\text{BH}+\delta\kappa M_{\alpha\beta\gamma\delta}S^{\alpha\beta}S^{\gamma\delta},\qquad M_{\alpha\beta\gamma\delta}=\,^*\!N^*_{\alpha\beta\gamma\delta}
\end{align}
which is conserved up to second order in the spin parameter for the spin-induced quadrupole MPTD equations \eqref{MPD}, \textit{i.e.} $\dot {\mathcal{Q}}^{(2)}=\mathcal O\qty(\mathcal S^3)$, regardless to the value taken by $\delta\kappa$.
\end{preliminary_result}

We finally notice that $\Upsilon_\text{NS}\eval_{a=0}=0$ in the Schwarzschild case by explicit evaluation of \eqref{source_NS}. A direct consequence is that the deformation of Rüdiger's quadratic invariant constructed in the black hole case is still quasi-conserved for arbitrary $\kappa$:
\begin{main_result}
In Schwarzschild spacetime ($a=0$), the deformation of Rüdiger's quadratic invariant $\mathcal Q_\text{BH}^{(2)}$
given in Eq. \eqref{quadratic_rudiger} 
is still conserved for the MPTD equations up to $\mathcal O(\mathcal S^3)$ corrections for arbitrary ($\kappa \in \mathbb R$)  spin-induced quadrupole.
\end{main_result}
Notice that the conservation does not hold for Rüdiger's linear invariant $\mathcal Q_Y$. The Kerr case ($a \neq 0$) will be further discussed in Section \ref{sec:NS}.

\section{Neutron star case around Kerr: a no-go result}
\label{sec:NS}

We summarize in Table \ref{tab:sources} the three constraint equations discussed previously. They all take the form 
\begin{align}
    \hat\nabla N=\Upsilon,
\end{align}
with $\Upsilon$ being of grading $[s,p]^+$. It implies that $N$ should be of grading $[s,p-1]^+$. 

\begin{table}[!bt]
\centering
\renewcommand{\arraystretch}{1.3}
\begin{tabular}{|c|c|c|c|c|}
\hline
\textsc{problem} & \textsc{source term $\Upsilon$} & \textsc{equation} & $\Upsilon\eval_{a=0}$ & \textsc{grading}\\\hline
NS linear invariant & $\Upsilon_\text{lin}$ & \eqref{lin_cst} & $\neq 0$ & $[2,4]^+$ \\\hline
BH quadratic invariant & $\Upsilon_\text{BH}$ & \eqref{source_BH} & $\neq 0$& $[2,3]^+$ \\\hline
NS quadratic invariant & $\Upsilon_\text{NS}$ & \eqref{source_NS} & $0$ & $[2,3]^+$\\\hline
\end{tabular}
\caption{Source terms for the various constraint equations $\hat\nabla N=\Upsilon$ studied in the thesis.}
\label{tab:sources}
\end{table}

Let $\qty{K_\mathfrak{a}}$ be a basis of linearly independent and dimensionless functions build from the (manifestly real) functions $\mathcal A, \mathcal P^2,\mathcal S^2,\Re A,\Im A,\ldots,\Re C,\Im C,D,\ldots,I$. Given that $N$ is dimensionless and given the structure of the source terms $\Upsilon$, we propose the following Ansatz,  
\begin{align}
    N=\sum_\mathfrak{a}\sum_{(k,l)\in\mathbb Z^2}K_\mathfrak{a}M^{l}\qty(C_\mathfrak{a}^{(k,l)}f_\mathfrak{a}^{(k,l)}(J)\alpha_1^{(k,k+l)}+D_\mathfrak{a}^{(k,l)}g^{(k,l)}_\mathfrak{a}(J)\omega_1^{(k,k+l)}).\label{generic_ansatz}
\end{align}
Here, $C_\mathfrak{a}^{k,l}$ and $D_\mathfrak{a}^{k,l}$ are numerical coefficients and $f_\mathfrak{a}^{k,l}(J)$ and $g_\mathfrak{a}^{k,l}(J)$ are smooth functions of $J$. 

We can work with dimensionless quantities by first introducing the dimensionless variables
\begin{align}
    \tilde r=\frac{r}{M},\qquad\tilde a=\frac{a}{M}.
\end{align}
We notice that the $K_\mathfrak{a}$'s are left unchanged and do not depend anymore on $M$, whereas $\mathcal R\triangleq M\tilde {\mathcal R}$, with $\tilde {\mathcal R}\triangleq \tilde r+i\tilde a$. This yields
\begin{align}
    \begin{split}
    \alpha^{(n,p)}_1=M^{n-p}\Re\qty(\frac{\bar{\tilde{\mathcal{R}}}}{\tilde{\mathcal{R}}})\triangleq M^{n-p}\tilde{\alpha}^{(n,p)},\\
    \omega^{(n,p)}_1=M^{n-p}\Im\qty(\frac{\bar{\tilde{\mathcal{R}}}}{\tilde{\mathcal{R}}})\triangleq M^{n-p}\tilde{\omega}^{(n,p)}.
    \end{split}
\end{align}
Each derivative of the term present in the Ansatz scales as $M^{-1}$ times a manifestly dimensionless quantity. All the source terms appearing earlier can be written as $\Upsilon=M^{-1}\tilde\Upsilon$, with $\tilde{\Upsilon}$ being an dimensionless quantity. 
This implies that Eq. \eqref{generic_ansatz} reduces to
\begin{align}
   N=\sum_\mathfrak{a}\sum_{(k,l)\in\mathbb Z^2} K_\mathfrak{a}\qty(C_\mathfrak{a}^{k,l}f_\mathfrak{a}^{k,l}(J)\tilde\alpha^{(k,k+l)}+D_\mathfrak{a}^{k,l}g^{k,l}_\mathfrak{a}(J)\tilde\omega^{(k,k+l)}),\label{generic_ansatz_dimless} 
\end{align}
which contain only terms that are explicitly independent of $M$. We can further define $\tilde \nabla = M \hat \nabla$ the dimensionless derivative operator and the constraints take the dimensionless form $\tilde\nabla N = \tilde \Upsilon$.

\subsection{Perturbative expansion in $a$ of the constraint equations}

Instead of addressing the non-linear problem in $a$ we will perform a perturbative series in $a$. For any smooth function $f$ of $a$, we define
\begin{align}
    \qty(f)_n\triangleq\dv[n]{f}{a}\eval_{a=0}.
\end{align}
The constraint equation then becomes an infinite hierarchy of equations
\begin{align}
\qty(\tilde{\nabla} N)_n=\tilde\Upsilon_n,\qquad\forall n\geq 0.
\end{align}

Let us describe the $n=0$ and $n=1$ equations. 
Since $J\propto a^2$, the functions $f_\mathfrak{a}^{k,l}(J)$ and $g_\mathfrak{a}^{k,l}(J)$ do not contribute  and can be set to one without loss of generality.

\subsubsection{$n=0$ equation} 

Noticing the identities
\begin{align}
    \begin{split}
    &\tilde\omega^{(k,k+l)}\eval_{a=0}=\tilde\alpha^{(k,k+l+1)}_A\eval_{a=0}=\tilde\alpha^{(k-1,k+l)}_{\bar A}\eval_{a=0}=0,
    \\
    &\tilde\alpha^{(k,k+l)}\eval_{a=0}=\tilde r^{-l},\quad\tilde\omega^{(k,k+l+1)}_A\eval_{a=0}=-p_r \tilde r^{-(l+1)},\\
    &\quad\tilde\omega^{(k-1,k+l)}_{\bar A}\eval_{a=0}=p_r \tilde r^{-(l+1)}
    \end{split}
\end{align}
and making use of Eq. \eqref{deriv_ao}, the $n=0$ constraint becomes
\begin{align}
    \sum_\mathfrak{a}\sum_{(k,l)\in\mathbb Z^2} C_\mathfrak{a}^{(k,l)}\qty[\qty({\tilde\nabla K_\mathfrak{a}})_0\tilde{r}^{-l}-l \qty(K_\mathfrak{a})_0 p_r \tilde r^{-(l+1)}]=\qty(\tilde\Upsilon)_0.
\end{align}
It does not depend on the terms involving $\omega$'s contributions. Moreover, denoting
\begin{align}
C^{(l)}_\mathfrak{a}\triangleq\sum_{k\in\mathbb Z}C^{(k,l)}_\mathfrak{a},
\end{align}
this equation can be further simplified to
\begin{align}
\sum_\mathfrak{a}\sum_{l\in\mathbb Z} C^{(l)}_\mathfrak{a}\qty[\qty({\tilde\nabla K_\mathfrak{a}})_0-l \qty(K_\mathfrak{a})_0 p_r \tilde r^{-1}]\tilde r^{-l}=\qty(\tilde\Upsilon)_0.
\label{eq_LO}
\end{align}

\subsubsection{$n=1$ equation} 

Following an identical procedure and denoting
\begin{align}
    D_\mathfrak{a}^{(l)}\triangleq\sum_{k\in\mathbb Z}\qty(2k+l)D_\mathfrak{a}^{(k,l)},
\end{align}
the $n=1$ constraint equation can be shown to take the form
\begin{align}
    \begin{split}
    &\sum_\mathfrak{a}\sum_{l\in\mathbb Z}\bigg\lbrace C_\mathfrak{a}^{(l)}\qty[\qty({\tilde\nabla K_\mathfrak{a}})_1-\qty(K_\mathfrak{a})_1lp_r\tilde r^{-1}]\tilde r^{-l}\\
    &-D_\mathfrak{a}^{(l)}\qty[\qty({\tilde\nabla K_\mathfrak{a}})_0x+\qty(K_\mathfrak{a})_0\qty(p_\theta-\qty(l+1)p_rx\tilde r^{-1})]\tilde r^{-(l+1)}\bigg\rbrace=(\tilde\Upsilon)_1.\label{eq_NLO}
    \end{split}
\end{align}

\subsubsection{Numerical evaluation}

Eqs. \eqref{eq_LO} and \eqref{eq_NLO} have been numerically evaluated using \textit{Mathematica} in order to try to fix the values of the coefficients $C_\mathfrak{a}^{(l)}$ and $D_\mathfrak{a}^{(l)}$ that would enable a possible solution to the neutron star cases.  We have only looked for ``polynomial'' solutions to these equations, \textit{i.e.} solutions for which the coefficients of the ansatz are non-vanishing only over a finite interval $[l_\text{min},l_\text{max}]$. Given the size of the expressions involved, the only computationally reasonable solving method available to us was the following: let us denote $N$ the number of terms present in the left-hand side of \eqref{eq_LO} (resp. \eqref{eq_NLO}) for a given $[l_\text{min},l_\text{max}]$. Eq. \eqref{eq_LO} (resp. \eqref{eq_NLO}) was then evaluated $N+1$ times at different random values of its variables and parameters, resulting into a linear system of $N+1$ algebraic equations in $N$ variables (the coefficients $C_\mathfrak{a}^{(l)}$ and $D_\mathfrak{a}^{(l)}$) that was then solved using the built-in numerical equation solver of \textit{Mathematica}.

This procedure has been proof-tested by reproducing the coefficients corresponding to the black hole quadratic invariant from the source term $\Upsilon_\text{BH}$. It has then been used to attempt to find a solution to both the linear and the quadratic neutron star problems, with $[l_\text{min},l_\text{max}]=[-10,100]$. No solution has been found, discarding \textit{a priori} the existence of polynomial-type solutions. 

In the appended \textit{Mathematica} notebooks\footnote{\url{https://github.com/addruart/generalizedCarterConstant}}, the interval of $l$ is reduced to $[l_\text{min},l_\text{max}]=[0,5]$ in order to reduce the computational time for the interested reader. The four notebooks related to this section are:
\begin{itemize} 
    \item \texttt{\lstinline{Quadratic_BH_LO_final.nb}}: check of the numerical evaluation of the $n=0$ equation: reproduction of the black hole quadratic invariant from $\Upsilon_\text{BH}$;
    \item \texttt{\lstinline{Linear_NS_LO_final.nb}}: attempt of finding a polynomial solution to the $n=0$ equation for the neutron star linear problem (source term $\Upsilon_\text{lin}$);
    \item \texttt{\lstinline{Quadratic_BH_NLO_final.nb}}: check of the numerical evaluation of the $n=1$ equation: reproduction of the black hole quadratic invariant from $\Upsilon_\text{BH}$;
    \item \texttt{\lstinline{Quadratic_NS_NLO_final.nb}}: attempt of finding a polynomial solution to the $n=1$ equation for the neutron star quadratic problem (source term $\Upsilon_\text{NS}$).
\end{itemize}

\section{Summary}
\label{sec:outlooks}

We have completed our exploration of the conserved quantities for the MPTD equations (that is, the MPD equations endowed with the TD spin supplementary condition) in Kerr spacetime. Our discussion is graphically summarized in Fig. \ref{fig:spinning_conserved_quantities}. Notice that by ``conservation'', we really mean ``quasi-conservation'', \textit{i.e.} we only require the conservation equation to hold up the some given order in $\mathcal S$. The magnitude $\mathcal S$ of the spin is always conserved. At linear order in the spin $\mathcal S$, the dynamical mass $\mu$ of the body is also conserved. Moreover, the three additional quantities conserved along geodesic motion (namely $\mathcal E_0$, $\mathcal L_0$ and $\mathcal K_0$) can be deformed to construct quantities $\mathcal E$, $\mathcal L$ and $\mathcal Q_R$ which are still conserved. They are respectively given by Eqs. \eqref{spin_EL} and \eqref{quadratic_rudiger_kerr}. Moreover, a new invariant $\mathcal Q_Y$ appears, homogeneously linear in $\mathcal S$. It is given by Eq. \eqref{linear_rudiger_kerr}. For historical reasons, $\mathcal Q_Y$ and $\mathcal Q_R$ are respectively referred to as the linear and the quadratic Rüdiger invariants.

At second order in the spin magnitude and for spin-induced quadrupole with black hole-type coupling ($\kappa=1$), the mass $\mu$ is no more conserved, but one can still define a mass-like quantity $\tilde\mu$ which is still conserved, see Eq. \eqref{mass_like}. Our others finding are as follows: (i) $\mathcal E$ and $\mathcal L$ are still constants of motion. This property can be shown to hold at any order of the multipole expansion of the equations of motion \cite{dixon1979}; (ii) the linear Rüdiger invariant $\mathcal Q_Y$ is still conserved and (iii) a deformation of Rüdiger's quadratic invariant denoted $\mathcal Q_\text{BH}^{(2)}$ exists such that the deformed Rüdiger's quadratic invariant is also quasi-conserved. It is given by Eq. \eqref{quadratic_rudiger}. Finally, (iv) the conservation of the deformed quadratic invariant can only be extended to arbitrary coupling ($\kappa\neq 1$) around the  Schwarzschild spacetime ($a=0$).

All our attempts to find solutions to the constraint equations in the case of an arbitrary spin-induced coupling in generic Kerr spacetime have failed. Let us notice that, even if someone would succeed in solving them, the quasi-invariants obtained would not be of direct astrophysical  interest. This is due to the fact that, except in the special case where the test body is itself a black hole, the spin-induced term is not the only contribution to the quadrupole. Tidal-type contributions will also arise, breaking the quasi-conservation obtained for spin-induced quadrupole only.  

The existence of these constants of motion has interesting consequences when we study the motion of spinning test bodies from the Hamiltonian perspective. First, as it was the case for geodesics, the existence of these constants is actually strongly related to the separability of the Hamilton-Jacobi equation for spinning bodies. Second, having in our possession a set of conserved quantities will enable us to address the status of integrability of the motion of spinning test bodies in Kerr spacetime. These topics will be the ones that will be treated in the final part of this thesis.

\partimage[width=\textwidth]{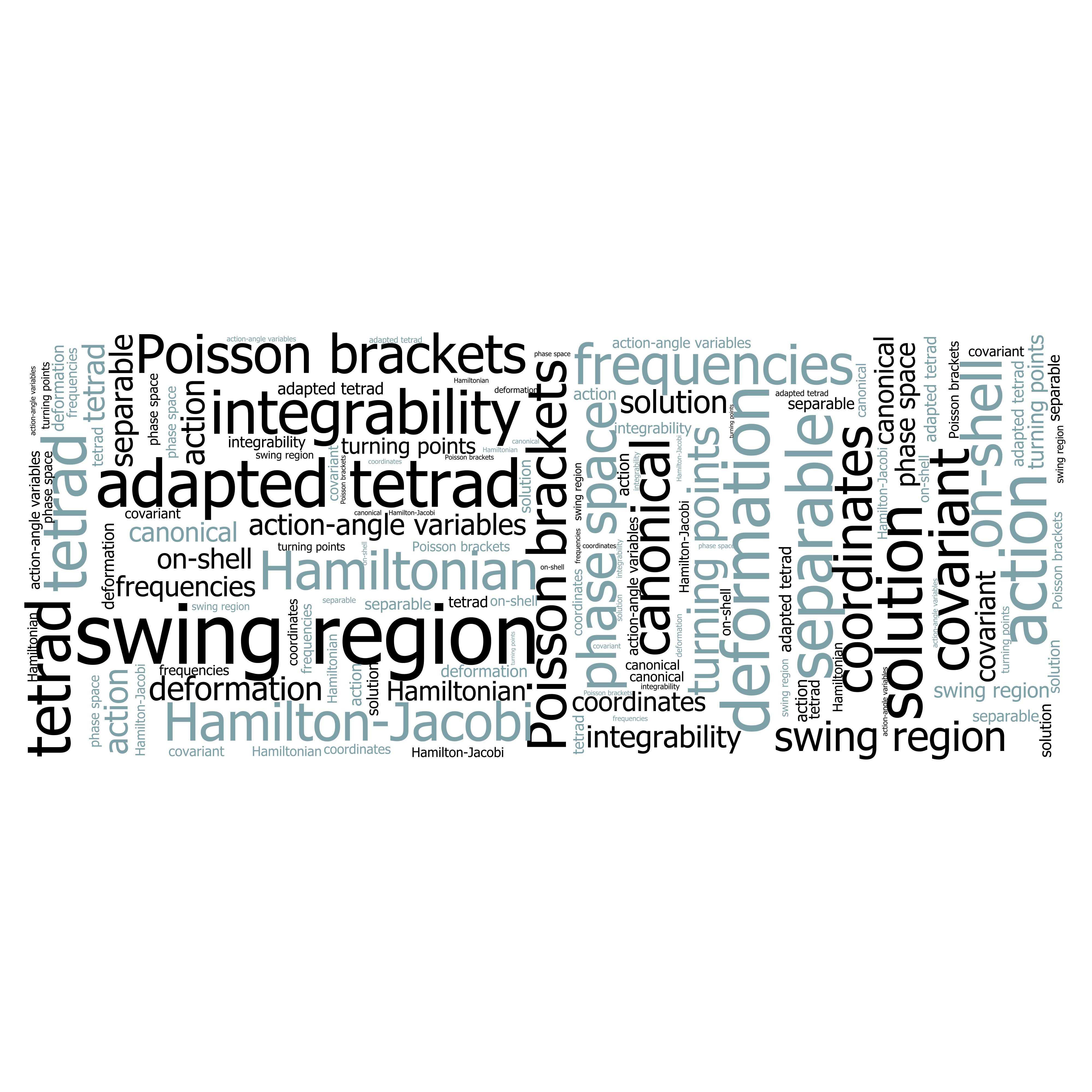}

\part[\textsc{Hamiltonian Description of Extended Test Bodies}]{\textsc{Hamiltonian Description\\of Extended Test Bodies}}\label{part:hamilton}

\noindent

\lettrine{A}{} more comprehensive understanding of some properties of test body motion in Kerr spacetime can be achieved by turning to an Hamiltonian formulation of the problem. In particular, it will enable us to discuss the integrability properties of the motion and the relation between the conserved quantities derived in Part \ref{part:conserved_quantities} and the solution of the associated Hamilton-Jacobi equation.

As it was the case for geodesic motion (and was extensively described in Chapter \ref{chap:hamilton_kerr}), there are two main formulations of test body motion in curved spacetime. The first one is based on a 3+1 decomposition of the background spacetime and uses the coordinate time $t$ for evolving the motion. Its associated phase space is 12-dimensional, and is spanned by the variables $\qty{x^i,p_i,S^{\mu\nu}}$. Historically, this is the Hamiltonian formulation of test bodies that has been the most widely used in the literature, particularly because it is very convenient for performing post-Newtonian expansions of results obtained in the MPD framework. See \textit{e.g.} \cite{Steinhoff_2010,Barausse_2009,Vines:2016unv} and references therein for more details.

Even if this formulation is well-suited for numerous applications, it is not well adapted to some others. One of these is the study of the structure of the associated Hamilton-Jacobi equation, and the computation of the related spin-induced shifts in the orbital turning points and in the fundamental frequencies of the motion. However, this loophole can be overcome by constructing covariant Hamiltonians, whose time evolution is parametrized by the proper time of the body or by any other external time parameter. These covariant Hamiltonians will drive the evolution on a larger, 14-dimensional, phase space parametrized by $(x^\mu,p_\mu,S^{\mu\nu})$. The construction of such Hamiltonians valid at all orders in $\mathcal S$ but restricted to the pole-dipole MPD equations was discussed for various spin supplementary conditions by W. Witzany, J. Steinhoff and G. Lukes-Gerakopoulos in \cite{Witzany:2018ahb}.
They also discussed the construction of canonical coordinates on phase space. Their derivation is valid provided that an arbitrary covariant SSC of the form Eq. \eqref{covariant_ssc} has been enforced.

This part of the thesis is voluntarily a little bit more sketchy than the previous ones, and present some partial results and conjectures, leaving significant room for subsequent work to be performed. The text is structured as follows. Chapter \ref{chap:symplectic} introduces a phase space compatible with a covariant Hamiltonian description of test bodies in generic curved spacetime. Both non-symplectic and symplectic coordinates systems are described, and the associated Poisson brackets algebras are derived. Finally, spin supplementary conditions viewed as constraints between the phase space variables are discussed from the action perspective. Chapter \ref{chap:covariant_H} reviews and generalizes the method of \cite{Witzany:2018ahb} for building covariant Hamiltonians, by including quadrupole moment but restricting to quadratic order in $\mathcal S$, as physical coherence of the multipole expansion requires. Using this technique, we present a fully covariant Hamiltonian which generates the MPD equations endowed with spin-induced quadrupole moment under the Tulczyjew-Dixon spin supplementary condition. To our knowledge, this result does not appear in the present form in the literature. The two final chapters of the thesis are devoted to applications of the Hamiltonian formalism. Chapter \ref{chap:integrability} reviews the status of integrability of test body motion in both Schwarzschild and Kerr spacetimes, at linear and quadratic orders in the spin magnitude $\mathcal S$. As claimed in \cite{Compere:2020eat} and numerically corroborated by \cite{Kunst:2015tla}, we argue that MPD equations in Kerr spacetime are not integrable already at linear order in $\mathcal S$, by explicitly computing the Poisson brackets between the maximal set of linearly independent conserved quantities found in Part \ref{part:conserved_quantities}. Finally, Chapter \ref{chap:HJ} reviews Witzany's analysis of Hamilton-Jacobi equation for test bodies in Kerr spacetime \cite{Witzany:2019nml}. Strictly speaking, the solution of the Hamilton-Jacobi equation is not anymore separable when spin has been turned on. However, at linear order in $\mathcal S$, the terms breaking the separability are negligible everywhere in the phase space, but near the turning points of the associated zeroth-order geodesic motion. In this so-called ``swing region'' where it is separable (to be precisely defined in Eq. \eqref{swing_region}), the solution takes a form very similar to the one worked out by Carter for geodesic motion (\textit{cf.} Chapter \ref{chap:hamilton_kerr}). In particular, Rüdiger's deformation of the Carter constant still plays the role of the separation constant of the Hamilton-Jacobi problem. The main advantage of this Hamilton-Jacobi formulation is that it allows to compute easily the spin-induced shifts in the orbital turning points and in the fundamental frequencies of the motion, with respect to a purely geodesic trajectory \cite{Witzany:2019nml}. The chapter ends by discussing a partial solution at quadratic order in $\mathcal S$, valid in the swing region of phase space, and by conjecturing the status of separability of the Hamilton-Jacobi problem at this order.

\chapter{Symplectic formulation}
\label{chap:symplectic}


This chapters aims to set the stage for the Hamiltonian description of spinning test bodies. Its goal will be twofold, namely (i) discussing the phase space structure and its associated symplectic structure (Poisson brackets algebra) and (ii) understanding which are the constraints the phase space variables are subjected to. The discussion is organized as follows: Section \ref{sec:phase_space} will discuss the structure of phase space for spinning test bodies, and how an associated Poisson brackets structure can be build. Section \ref{sec:symplectic} will review two ways of constructing symplectic coordinates on the phase space. Finally, we turn to the discussion of constraints, which is the subject of Section \ref{sec:csts}. 

In what follows, we start again from the abstract form of the action obtained in Chapter \ref{chap:EOM}:
\begin{align}
    S=\int\dd\lambda\qty(p_\mu v^\mu+\frac{1}{2}S_{\mu\nu}\Omega^{\mu\nu}),
\end{align}
where $p_\mu$ and $S_{\mu\nu}$ denote the momenta conjugated to $v^\mu$ and $\Omega^{\mu\nu}$.

\section{Phase space and non-symplectic coordinates}\label{sec:phase_space}

We will take the phase space $\mathcal M$ of spinning test bodies to be spanned by the variables $\qty{x^\mu,p_\mu,S^{\mu\nu}}$. It is therefore 14-dimensional, and has the simple structure
\begin{align}
    \mathcal M\simeq\mathbb R^4\times\mathbb R^4\times \mathbb R^6.
\end{align}
In this section, we aim to derive a symplectic structure on $\mathcal M$ that is compatible with the Hamiltonian description of test bodies.  

In order to reach this aim, a guideline is to recall that one of the keypoints of Hamiltonian description is to replace the velocities $\dot q^\mathfrak{i}$ by the conjugate momenta $P_\mathfrak{i}\triangleq\pdv{L}{\dot q^\mathfrak{i}}$. In terms of the coordinates $q^\mathfrak{i}$ and of the momenta $P_\mathfrak{j}$, Eq. \eqref{poisson_brackets} implies that the Poisson brackets algebra is simply\footnote{From this point, we understand that all the Poisson brackets that are not explicitly written are vanishing.}
\begin{align}
    \pb{q^\mathfrak{i}}{P_\mathfrak{j}}=\delta^\mathfrak{i}_\mathfrak{j}.\label{PB_Darboux}
\end{align}
In that case, we will say that $q^\mathfrak{i}$, $P_\mathfrak{i}$ are symplectic variables, or equivalently that the Poisson brackets algebra is in Darboux form.

Our strategy for computing the Poisson brackets algebra between the usual phase space variables consists in two steps: (i) identify the variables that play the role of the coordinates $q^\mathfrak{i}$ and of their associated velocities $\dot q^\mathfrak{i}$ in the Lagrangian formulation of test bodies introduced in Chapter \ref{chap:EOM}. As we will see, it will allow us to find explicit expressions for the form of the canonically conjugate momenta. Subsequently, (ii) we will use our findings as well as the canonical algebra Eq. \eqref{PB_Darboux} to obtain the standard Poisson brackets between the variables $x^\mu,p_\mu$ and $S^{\mu\nu}$.

\subsection{Coordinates and velocities}
The identification of coordinates and velocities in the Lagrangian can be performed in two steps \cite{Barausse_2009}:
\begin{enumerate}
    \item Up to this point, we have written our Lagrangian action as
\begin{align}
    S=\int\dd\lambda\,L\qty(v^\mu,\Omega^{\mu\nu},g_{\mu\nu},R_{\mu\nu\rho\sigma},\nabla_\alpha R_{\mu\nu\rho\sigma},\ldots).
\end{align}
For a stationary background, the metric, the Riemann tensor and all its covariant derivatives are only functions of the spacetime coordinates $x^\mu$. Therefore, the action above can be formally written
\begin{align}
    S=\int\dd\lambda\, L\qty(x^\mu,v^\mu,\Omega^{\mu\nu}).
\end{align}
    \item Given an arbitrary background tetrad $\underline e\tdu{ A}{\mu}$, the worldline tetrad $e\tdu{A}{\mu}$ can be expressed as
    \begin{align}
     e\tdu{A}{\mu}\qty(\lambda)=\qty({\Lambda^{-1}})\tud{B}{A}(\lambda)\underline e\tdu{B}{\mu}\qty(z(\lambda)),  
    \end{align}
    where the Lorentz transformation matrix\footnote{Recall that we dropped the underline in the background tetrad indices.} $\Lambda\tud{A}{B}$ has been introduced in Eq. \eqref{tfo_tetrads}. As mentioned earlier, this matrix depends on six (reals) parameters. Denoting them collectively $\qty{\phi^I}$ ($I=1,\ldots,6$), the worldline tetrad can be seen as a function of the worldline position $x^\mu$ and of the Lorentz parameters $\phi\triangleq\qty{\phi^I}$, which depend both on $\lambda$:
    \begin{align}
        e\tdu{A}{\mu}(\lambda)=e\tdu{A}{\mu}\qty(\phi(\lambda),x(\lambda)).
    \end{align}
    This implies that the rotation coefficients can be written as
    \begin{align}
        \Omega^{\mu\nu}=e^{A\mu}\qty(\phi,x)\qty(\dv{\phi^I}{\lambda}\pdv{e\tdu{A}{\nu}}{\phi^I}\qty(\phi,x)+v^\lambda e\tdud{A}{\nu}{,\lambda}\qty(\phi,x))+\Gamma^\nu_{\alpha\beta}g^{\mu\alpha}v^\beta.\label{rotation_coeff_expansion}
    \end{align}
    It follows that the rotation coefficients are only functions of $x^\mu$, $\phi^I$ and its first derivative:
    \begin{align}
        \Omega^{\mu\nu}=\Omega^{\mu\nu}\qty(x,\phi,\dv{\phi}{\lambda}).
    \end{align}
\end{enumerate}
From the above discussion, our action can be written as
\begin{align}
    S=\int\dd\lambda\, \bar 
    L\qty(x^\mu,v^\mu,\phi^I,\dv{\phi^I}{\lambda}).
\end{align}
The notation $\bar L$ has been introduced to make explicit the new functional dependence of the Lagrangian, which will be enlightening when computing the conjugate momenta. We also introduce the shortcut notation $\dot{\phi}^I\triangleq\dv{\phi^I}{\lambda}$. The configuration space of the system is spanned by the generalized coordinates $\mathbf q\triangleq\qty{x^\mu,\phi^I}$ and its action is explicitly written in terms of the coordinates $\mathbf q$ and their associated velocities $\dot{\mathbf{q}}$:
\begin{align}
    S=\int\dd \lambda\,\bar L\qty(x^\mu,v^\mu,\phi^I,\dot{\phi}^I,t).
\end{align}

\subsection{Conjugate momenta}
The next step is to write down the momenta conjugated to the variables $\mathbf q$. As revealed by the above discussion, the Lagrangian can be seen either as a function $L\qty(x^\mu,v^\mu, \Omega^{\mu\nu},\lambda)$ or as a function $\bar L\qty(x^\mu,v^\mu,\phi^I,\dot\phi^I,t)$. These function must satisfy
\begin{align}
    L\qty(x^\mu,v^\mu, \Omega^{\mu\nu},\lambda)\eval_{\Omega^{\mu\nu}=\Omega^{\mu\nu}\qty(x^\mu,\phi^I,\dot{\phi^I})}=\bar L\qty(x^\mu,v^\mu,\phi^I,\dot\phi^I,\lambda).
\end{align}
Varying the two sides of this equation gives respectively
\begin{align}
    \delta\bar L=\pdv{\bar L}{x^\mu}\delta x^\mu+P_\mu\delta v^\mu+\pdv{\bar L}{\phi^I}\delta\phi^I+P_{\phi^I}\delta\dot{\phi}^I
\end{align}
with
\begin{align}
    P_\mu\triangleq\pdv{\bar L}{v^\mu},\qquad P_{\phi^I}\triangleq\pdv{\bar L}{\dot{\phi}^I}.
\end{align}
and 
\begin{align}
    \begin{split}
    \delta L&= \pdv{L}{x^\mu}\delta x^\mu+\pdv{L}{v^\mu}\delta v^\mu+\qty(\pdv{L}{\Omega^{\mu\nu}}\delta\Omega^{\mu\nu})\eval_{\Omega^{\mu\nu}=\Omega^{\mu\nu}\qty(x^\mu,\phi^I,\dot{\phi^I})}\\
    &=\qty(\pdv{L}{x^\mu}+\pdv{L}{\Omega^{\mu\nu}}\pdv{\Omega^{\mu\nu}}{x^\mu})\delta x^\mu+\qty(\pdv{L}{v^\mu}+\pdv{L}{\Omega^{\mu\nu}}\pdv{\Omega^{\mu\nu}}{v^\mu})\delta v^\mu\\
    &\quad+\pdv{L}{\Omega^{\mu\nu}}\pdv{\Omega^{\mu\nu}}{\phi^I}\delta\phi^I+\pdv{L}{\Omega^{\mu\nu}}\pdv{\Omega^{\mu\nu}}{\dot{\phi}^I}\delta\dot{\phi}^I.
    \end{split}
\end{align}
Comparing the two variations and using the definition \eqref{momenta} yields
\begin{align}
    P_\mu&=p_\mu+\frac{1}{2}S_{\mu\nu}\pdv{\Omega^{\mu\nu}}{v^\mu},\qquad
    P_{\phi^I}=\frac{1}{2}S_{\mu\nu}\pdv{\Omega^{\mu\nu}}{\dot{\phi}^I}.\label{canonical_coordinates}
\end{align}
Using Eq. \eqref{rotation_coeff_expansion}, one can write, after some straightforward algebra
\begin{subequations}\label{conj}
\begin{align}
    P_\mu&=p_\mu-\frac{1}{2}\omega_{\mu AB}S^{AB},\qquad \omega_{\mu AB}\triangleq \underline e\tdu{A}{\lambda}\underline e_{B\lambda;\mu}\label{conj_pos}\\
    P_{\phi ^I}&=\frac{1}{2}S_{\mu\nu}\lambda^{AB}_I\underline e\tdu{A}{\mu}\underline e\tdu{B}{\nu}, \qquad
    \lambda^{AB}_I\triangleq{\Lambda^{-1}}\tdu{C}{A}\pdv{ {\Lambda^{-1}}^{CB}}{\phi^I}.\label{conj_spin}
\end{align}
\end{subequations}
Notice that we have introduced the \defining{connection 1-forms} $\omega_{\mu AB}$ following the notation of Wald \cite{Wald:1984rg}. These quantities satisfy $\omega_{\mu AB}=-\omega_{\mu BA}$. In terms of the coordinates $x^\mu,\phi^I$ and of their conjugate momenta $P_\mu,P_{\phi^I}$, the Poisson brackets algebra is simply of the Darboux form \eqref{PB_Darboux}, which reads explicitly
\begin{align}
    \pb{x^\mu}{P_\nu}=\delta^\mu_\nu,\qquad \pb{\phi^I}{P_{\phi^J}}=\delta^I_J.\label{darboux_spin}
\end{align}
The expressions found for the canonically conjugate momenta are rather abstract, and more work will be required to end up with explicit formulae for these momenta, as well as to disentangle the independent variables. Nevertheless, this abstract formulation will already allow us to compute the Poisson algebra between the conventional phase space variables. 

\subsection{Quasi-symplectic coordinates}
The Poisson algebra takes a particularly simple form if one choose the phase space coordinates to be $\qty{x^\mu,P_\mu,S^{AB}}$, where $S^{AB}=\underline e\tud{A}{\mu}\underline e\tud{B}{\nu}S^{\mu\nu}$ are the components of the spin tensor expressed in an arbitrary background orthonormal tetrad $\underline e\tdu{A}{\mu}$. For these coordinates, the only non-vanishing brackets are (see \cite{Witzany:2018ahb} and references therein)
\begin{subequations}\label{quasi_symplectic}
\boxedeqn{
\pb{x^\mu}{P_\nu}&=\delta^\mu_\nu,\\ \pb{S^{AB}}{S^{CD}}&=S^{AC}\eta^{BD}+S^{BD}\eta^{AC}-S^{AD}\eta^{BC}-S^{BC}\eta^{AD}.
}{Poisson brackets for quasi symplectic coordinates}
\end{subequations}
Following \cite{Ramond_2021}, the form of this algebra leads us to refer to these coordinates as \defining{quasi-symplectic coordinates}: they behave as symplectic coordinates for the ``position sector'' of the phase space, covered by the variables $x^\mu$ and $P_\mu$, but not for the ``spin sector'', covered by $S^{AB}$. Notice however that the brackets between the tetrad spin tensor components are a realization of the $\mathfrak{so}(1,3)$ (\textit{i.e.} Lorentz) algebra, which originates from the definition of $S^{AB}$ as being expressed in an orthonormal tetrad frame.

The rest of this section will be devoted to the proof of the algebra Eq. \eqref{quasi_symplectic}. We begin by deriving a couple of useful identities, following \cite{Hanson:1974qy,Barausse_2009}. First, let us assume -- without loss of generality -- that the inverse of Eq. \eqref{conj_spin} takes the form
\begin{align}
    S^{\mu\nu}=\underline e\tdu{A}{\mu}\underline e\tdu{B}{\nu}\rho^{AB}_I(\phi) P_{\phi^I}.\label{inverse_spin}
\end{align}
Plugging this equation in Eq. \eqref{conj_spin}, we get
\begin{align}
    P_{\phi^I}\qty(\delta_{IJ}-\frac{1}{2}\rho^{AB}_I\lambda_{JAB}P_{\phi^J})=0,
\end{align}
which yields
\begin{align}
    \rho^{AB}_I\lambda_{JAB}=2\delta_{IJ}.\label{id_one}
\end{align}
Now, proceeding the other way around and plugging Eq. \eqref{conj_spin} into Eq. \eqref{inverse_spin}, we get
\begin{align}
    S^{\alpha\beta}\underline e\tdu{A}{\mu}\underline e\tdu{B}{\nu}\underline e_{C\alpha}\underline e_{D\beta}\qty(\eta^{AC}\eta^{BD}-\frac{1}{2}\rho^{AB}_I\lambda^{CD}_I)=0,
\end{align}
which yields, since it must be valid for any antisymmetric spin tensor $S^{\alpha\beta}$,
\begin{align}
    \rho_I^{AB}\lambda^{CD}_I=\eta^{AC}\eta^{BD}-\eta^{AD}\eta^{BC}.\label{id_two}
\end{align}
Finally, by differentiating directly the definition Eq. \eqref{conj_spin} of $\lambda_I^{AB}$, we get the third identity
\begin{align}
    \pdv{\lambda^{AB}_I}{\phi^J}-\pdv{\lambda^{AB}_J}{\phi^I}=\lambda^{AC}_I\lambda\tdu{JC}{B}-\lambda^{AC}_J\lambda\tdu{IC}{B}\label{id_three}.
\end{align}

Combining the three identities \eqref{id_one}, \eqref{id_two} and \eqref{id_three} yields, after a bunch of algebra
\begin{align}
    \rho^{AB}_J\pdv{\rho_I^{CD}}{\phi^J}-\rho^{CD}_J\pdv{\rho_I^{AB}}{\phi^J}=-\rho^{AC}_I\eta^{BD}-\rho^{BD}_I\eta^{AC}+\rho^{AD}_I\eta^{BC}+\rho_I^{BC}\eta^{AD},\label{func_id}
\end{align}
which can be recognized as another realization of the Lorentz algebra.

We can now turn to the computation of the Poisson brackets between the components of the spin tensor. Since $S^{AB}=\rho^{AB}_IP_{\phi^I}$, we get
\begin{align}
\begin{split}
    \pb{S^{AB}}{S^{CD}}&=-\rho^{AB}_I\pdv{\rho^{CD}_J}{\phi^K}\pb{\phi^K}{P_{\phi^I}}P_{\phi^J}-\qty(A\leftrightarrow C,B\leftrightarrow D)\\
    &=-\rho^{AB}_I\pdv{\rho^{CD}_J}{\phi^I}P_{\phi^J}-\qty(A\leftrightarrow C,B\leftrightarrow D)\\
    &=\qty(\rho^{AC}_I\eta^{BD}+\rho^{BD}_I\eta^{AC}-\rho^{AD}_I\eta^{BC}-\rho_I^{BC}\eta^{AD})P_{\phi^J}\\&=S^{AC}\eta^{BD}+S^{BD}\eta^{AC}-S^{AD}\eta^{BC}-S^{BC}\eta^{AD}.
\end{split}
\end{align}
The second equality was obtained using the Poisson algebra Eq. \eqref{darboux_spin}, whereas the third one has made use of the identity Eq. \eqref{func_id}. The vanishing of the others brackets between quasi-symplectic coordinates is straightforward to check, and we end up with the algebra announced in Eq. \eqref{quasi_symplectic}.

\subsection{Non-symplectic coordinates}
From the quasi-symplectic Poisson algebra, we can descent to the algebra for the non-symplectic coordinates $\qty{x^\mu,p_\mu,S^{\mu\nu}}$, which appear as the most ``basic'' coordinates for parametrizing the phase space, and which have been used in a wide range of works (see again \cite{Witzany:2018ahb} and references therein). In terms of these variables, the Poisson algebra reads
\begin{subequations}\label{non_cov_brackets}
\boxedeqn{
    \pb{x^\mu}{p_\nu}&=\delta^\mu_\nu,\\
    \pb{p_\mu}{p_\nu}&=-\frac{1}{2}R_{\mu\nu\alpha\beta}S^{\alpha\beta},\label{pb:pp}\\
    \pb{S^{\mu\nu}}{p_\kappa}&=2\Gamma^{[\mu}_{\lambda\kappa}S^{\nu]\lambda},\label{pb:sp}\\
    \pb{S^{\mu\nu}}{S^{\rho\sigma}}&=g^{\mu\rho}S^{\nu\sigma}-g^{\mu\sigma}S^{\nu\rho}+g^{\nu\sigma}S^{\mu\rho}-g^{\nu\rho}S^{\mu\sigma}.\label{pb:ss}
}{Poisson brackets for non-symplectic coordinates}
\end{subequations}
The derivation of this algebra is rather straightforward from the quasi-symplectic algebra Eq. \eqref{quasi_symplectic} and the definitions of the conjugate momenta Eq. \eqref{conj}. It is however quite long and not particularly enlightening, and will therefore not be reproduced here. The only keypoint to mention is that the appearance of the Riemann tensor in Eq. \eqref{pb:pp} comes from the possibility of writing the Riemann tensor solely in terms of the $\omega_{\mu AB}$, see Eqs. (3.4.20) and (3.4.21) of \cite{Wald:1984rg}.

\subsection{Casimir invariants}
From the Poisson algebra Eq. \eqref{quasi_symplectic}, it is easy to prove that the two quantities 
\begin{align}
    \mathcal S^2\triangleq\frac{1}{2}S_{AB}S^{AB},\qquad\mathcal S_*^2\triangleq\frac{1}{8}\epsilon_{ABCD}S^{AB}S^{CD}\label{casimirs}
\end{align}
have vanishing Poisson brackets with any variable $x^\mu$, $P_\mu$ and $S^{AB}$. Their Poisson brackets with any phase space function $F$ are therefore vanishing, 
\begin{align}
    \pb{\mathcal S^2}{F}=\pb{\mathcal S^2_*}{F}=0.
\end{align}
Taking $F=H$ implies that they are conserved for the evolution of any dynamical system of Hamiltonian $H$ evolving on $\mathcal M$ endowed with the set of Poisson brackets Eq. \eqref{quasi_symplectic}.

\section{Symplectic coordinates}\label{sec:symplectic}

We now turn to the construction of \textit{fully} symplectic coordinates on the phase space. Since the quasi-symplectic coordinates are already symplectic for the position sector of the phase space, only the spin sector remains to be discussed. This sector is six dimensional, since it is parametrized by $S^{AB}$. Given the existence of the two Casimir invariants \eqref{casimirs}, we are only missing two pairs of canonically conjugate coordinates for fully describing the spin sector.

There exist currently two approaches for building such symplectic coordinates for the spin sector: the first one, due to W. Witzany, J. Steinhoff and G. Lukes-Gerakopoulos \cite{Witzany:2018ahb} amounts to parameterize in a clever way the Lorentz matrix \eqref{tfo_tetrads} encoding the spin degrees of freedom of the system. The procedure for deriving these symplectic coordinates will be discussed in details below, and is valid provided that a covariant spin supplementary condition of the form \eqref{covariant_ssc} has been enforced. The second path has been recently introduced by P. Ramond \cite{Ramond:2022vhj}, and consists into explicitly solving the differential equations constraining the spin sector variables to be symplectic. He found variables which are linear combinations of Witzany's ones, but his derivation  is more generic, since it is totally independent from any spin supplementary condition.

\subsection{Decomposition of the Lorentz matrix: Witzany \textit{et al.} coordinates}

In this section, we will derive canonical coordinates for the spin sector of the problem following Witzany \textit{et al.} \cite{Witzany:2018ahb}. \textit{We temporarily introduce again the underline notation for the background tetrad indices}, since both background and object tetrads will be involved in the discussion.

This construction restricts ourselves to non-trivial spin tensors subjected to an arbitrary covariant SSC of the form Eq. \eqref{covariant_ssc}, thus possessing one timelike degenerate direction and one spatial non-degenerate one. This restriction on the form of the spin tensor allows to tune the object tetrad so that
\begin{align}
    S_{AB}=\mqty(0&0&0&0\\0&0&\mathcal S&0\\0&-\mathcal S&0&0\\0&0&0&0).\label{tetrad_spin_tensor}
\end{align}

We will now construct canonical coordinates in a very pedestrian way: dynamical variables $(\phi_i,\chi^i)$ will be canonical coordinates for the spin sector of the problem provided that the spin term $\frac{1}{2}S_{\mu\nu}\Omega^{\mu\nu}$ present in the action can be written as
\begin{align}
    \frac{1}{2}S_{\mu\nu}\Omega^{\mu\nu}\stackrel{!}{=}\phi_i\dot\chi^i
\end{align}
with $\dot{ }\triangleq\dv{\lambda}$.
We will now exhibit such coordinates. First, remark that a straightforward computation yields
\begin{equation}
    \frac{1}{2}S_{\mu\nu}\Omega^{\mu\nu}=\frac{1}{2}S_{\underline A\,\underline B}\Omega^{\underline A\,\underline B}=-\frac{1}{2}S_{AB}\Lambda\tud{A}{\underline A}\dv{\Lambda^{\underline AB}}{\lambda}.\label{spin_term}
\end{equation}
Said with words, the dynamics of the spin sector of the problem can be entirely expressed in terms of the Lorentz matrices linking the body and the background tetrads. The second step is to explicit such a transformation in terms of its six Lorentz parameters. To do it in a clever way, let us notice that the spin tensor \eqref{tetrad_spin_tensor} is invariant with respect to rotations and boosts in the $z$ direction. Therefore, two out of the six Lorentz parameters will play the role of gauge degrees of freedom, whereas the four other will be physically relevant. To take advantage of this situation, let us decompose our generic Lorentz transformation as follows \cite{Witzany:2018ahb}:
\begin{align}
    \Lambda=R(\alpha,\boldsymbol n_z)B(v_z,\boldsymbol n_z)B(u,\boldsymbol n_\psi)R(-\vartheta,\boldsymbol n_\phi).\label{lorentz_tfo}
\end{align}
Here, the directions of the two rightmost boost and rotation are chosen to lie in the $x-y$ plane,
\begin{align}
    \boldsymbol n_\psi\triangleq\mqty(\sin\psi&\cos\psi&0),\qquad\boldsymbol n_\phi\triangleq\mqty(\sin\phi&\cos\phi&0).
\end{align}
All the parameters of the Lorentz transformations $\alpha,v_z,u,\psi,\theta,\phi$ are considered as depending on the time parameter $\lambda$. The explicit rotation and boost matrices are respectively given by
\begin{align}
    R(\theta,\boldsymbol u)=
    {\scriptsize
    \mqty(1&0&0&0\\
    0 & \cos\vartheta+u_x^2(1-\cos\vartheta) & u_x u_y(1-\cos\vartheta)-u_z\sin\vartheta & u_xu_z(1-\cos\vartheta)+u_y\sin\vartheta\\
    0 & u_yu_x(1-\cos\vartheta)+u_z\sin\vartheta & \cos\vartheta+u_y^2(1-\cos\vartheta) & u_yu_z(1-\cos\vartheta)-u_x\sin\vartheta\\
    0 & u_zu_x(1-\cos\vartheta)-u_y\sin\vartheta & u_zu_y(1-\cos\vartheta)+u_x\sin\vartheta & \cos\vartheta+u_z^2(1-\cos\vartheta))}
\end{align}
and
\begin{align}
    B(v,\boldsymbol u)=\mqty(\gamma&-\gamma v u_x & -\gamma v u_y & -\gamma vu_z\\
    -\gamma vu_x & 1+(\gamma-1)u_x^2 & (\gamma-1) u_xu_y & (\gamma-1)u_xu_z\\
    -\gamma vu_y & (\gamma-1) u_yu_x & 1+(\gamma-1) u_y^2 & (\gamma-1)u_yu_z\\
    -\gamma vu_z & (\gamma-1) u_zu_x & (\gamma-1)u_zu_y & 1+(\gamma-1)u_z^2)
\end{align}
with $\boldsymbol u=(u_x,u_y,u_z)$ a normalized spatial vector ($u_x^2+u_y^2+u_z^2=1$) and $\gamma=\frac{1}{\sqrt{1-v^2}}$.

Using the form of the spin tensor \eqref{tetrad_spin_tensor}, Eq. \eqref{spin_term} becomes
\begin{align}
    \frac{1}{2}S_{\mu\nu}\Omega^{\mu\nu}=-\mathcal S\Lambda\tud{1}{\underline A}\dv{\Lambda^{\underline A 2}}{\lambda}.
\end{align}
Plugging the explicit Lorentz transformation \eqref{lorentz_tfo} in this expression leads to
\begin{align}
    \frac{1}{2}S_{\mu\nu}\Omega^{\mu\nu}=\mathcal S\qty[\dot\alpha+\frac{\cos\vartheta-1}{\sqrt{1-u^2}}\dot\phi+\qty(\frac{1}{\sqrt{1-u^2}}-1)\dot\psi].
\end{align}
The first term being a total derivative, it can be dropped out of the action. By inspection of this expression, we therefore find that pairs of canonical coordinates covering the spin sector are given by $(\phi,A)$ and $(\psi,B)$, where the conjugate momenta $A$ and $B$ are respectively given by
\begin{align}
    A&=\mathcal S\frac{\cos\vartheta-1}{\sqrt{1-u^2}},\qquad
    B=\mathcal S\qty(\frac{1}{\sqrt{1-u^2}}-1).
\end{align}
As expected, the expression above do not depend upon the gauge degrees of freedom $\alpha$ and $v_z$. In terms of the canonical coordinates, the components of the spin tensor in the background tetrad $S_{\underline A\,\underline B}=\Lambda\tud{A}{\underline A}\Lambda\tud{B}{\underline B} S_{AB}$ read
\begin{subequations}\label{background_canonical}
\boxedeqn{
    S_{\underline 0\,\underline 1}&=-\mathcal D\qty[A\cos\qty(2\phi-\psi)+\qty(A+2B+2\mathcal S)\cos\psi],\\
    S_{\underline 0\,\underline 2}&=\mathcal D\qty[A\sin\qty(2\phi-\psi)+\qty(A+2B+2\mathcal S)\sin\psi],\\
    S_{\underline 0\,\underline 3}&=2\mathcal D\mathcal E\cos\qty(\phi-\psi),\\
    S_{\underline 1\,\underline 2}&=A+B+\mathcal S,\\
    S_{\underline 1\,\underline 3}&=\mathcal E\sin\phi,\\
    S_{\underline 2\,\underline 3}&=\mathcal E\cos\phi.
}{Spin tensor in Witzany's symplectic coordinates}
\end{subequations}
where
\begin{align}
    \mathcal D\triangleq-\frac{\sqrt{B\qty(B+2\mathcal S)}}{2\qty(B+\mathcal S)},\qquad\mathcal E\triangleq\sqrt{-A\qty(A+2B+2\mathcal S)}.
\end{align}
The relations \eqref{background_canonical} can be inverted to obtain $A,B,\phi$ and $\psi$ as a function of the background tetrad components of the spin tensor. Using these relations, one can check explicitly that we have indeed build canonical coordinates, since the only non-vanishing Poisson brackets are
\begin{align}
    \pb{x^\mu}{P_\nu}=\delta^\mu_\nu,\qquad\pb{\phi}{A}=\pb{\psi}{B}=1.
\end{align}
Notice that this parametrization can be checked to consistently imply $\mathcal S^2_*=\mathbf S\cdot\mathbf D=\mathbf 0$, which shall be obeyd since a covariant SSC is assumed to hold.

\subsection{Ramond's coordinates}
In 2022, P. Ramond showed \cite{Ramond:2022vhj} that the spin tensor can be replaced (without any assumption regarding the existence of a spin supplementary condition) by the two Casimir invariants \eqref{casimirs} together with two pairs of symplectic coordinates $(\sigma,\pi_\sigma)$ and $(\zeta,\pi_\zeta)$. In terms of these coordinates, the background tetrad components of the spin tensor are \cite{Ramond:2022vhj}
\begin{subequations}
\boxedeqn{
    S_{\underline 0\,\underline 1}&=Y\pi_\sigma\sin\zeta\cos\sigma+Y\pi_\zeta\cos\zeta\sin\sigma+XZ\cos\sigma,\\
    S_{\underline 0\,\underline 2}&=Y\pi_\sigma\sin\zeta\sin\sigma-Y\pi_\zeta\cos\zeta\cos\sigma+XZ\sin\sigma,\\
    S_{\underline 0\,\underline 3}&=Z\pi_\sigma-XY\sin\zeta,\\
    S_{\underline 1\,\underline 2}&=\pi_\sigma,\\
    S_{\underline 1\,\underline 3}&=X\sin\sigma,\\
    S_{\underline 2\,\underline 3}&=X\cos\sigma,
}{Spin tensor in Ramond's symplectic coordinates}
\end{subequations}
where
\begin{align}
    X\triangleq\sqrt{\pi_\zeta^2-\pi_\sigma^2},\quad Y\triangleq\sqrt{1-\frac{\mathcal S^2}{\pi_\zeta^2}-\frac{\mathcal S^4_*}{\pi_\zeta^2}},\quad Z\triangleq\frac{\mathcal S^2_*}{\pi^2_\zeta}.
\end{align}
They are related to Witzany coordinates by a simple affine map
\begin{align}
    \qty(\sigma,\pi_\sigma)=\qty(\phi,A+B+\mathcal S),\qquad \qty(\zeta,\pi_\zeta)=\qty(\psi-\phi+\frac{\pi}{2},B+\mathcal S)
\end{align}
and carry a nice physical interpretation in terms of respective orientations of the 3-vectors $\mathbf S$ and $\mathbf D$ and the spatial legs of the background tetrad frame, see \cite{Ramond:2022vhj} for explicit details.

\section{Constraints}\label{sec:csts}

\begin{table}[t]
    \centering
    \renewcommand{\arraystretch}{1.5}
    \begin{tabular}{|c|c|}\hline
       \textsc{Value of} $\Lambda\tdu{0}{A}$  & \textsc{SSC} \\\hline
        $\hat p^A$ & TD\\
        $2\hat p^0\delta^A_0-\hat p^A$ & CP\\
        $\delta^A_0$ & NW\\\hline
    \end{tabular}
    \caption{Various choices of $\Lambda\tdu{0}{A}$ allowing to enforce the main SSC used in the literature at the level of the constrained action.}
    \label{tab:gauges}
\end{table}

This section will briefly discuss the constraints that shall be taken into account when turning to the Hamiltonian description, and how they can be implemented at the level of the action. This is a rather technical subject, based upon results and ideas introduced in the context of geodesic motion in Chapter \ref{chap:hamilton_kerr}. We only review here the main ideas discussed in the literature. Technical exposition of the subject may be found in \cite{Steinhoff_2010,Steinhoff:2014kwa,Levi:2015msa,Vines:2016unv}.

Two kinds of constraints shall be enforced to ensure the independence of the conjugate momenta of Hamiltonian formulation: the first one comes from the existence of redundant degrees of freedom in the spin sector and is directly related to spin supplementary conditions, as discussed in Section \ref{sec:ssc:lagrange}. It was first introduced in the case of a flat background in \cite{Steinhoff:2014kwa}, and subsequently expanded to generic curved background, see \cite{Levi:2015msa,Vines:2016unv}. Moreover, the fact that the generic action described by Eq. \eqref{generic_lagragian} is reparametrization invariant leads to the existence of a mass-shell constraint of the same type as in the geodesic case, as already pointed out in \cite{Hanson:1974qy}. Altogether, these constraints can be implemented at the level of the action by considering
\begin{align}
    S= \int\dd\lambda\qty( p_\mu v^\mu+\frac{1}{2} S_{\mu\nu}\Omega^{\mu\nu}- H_\text{D})
\end{align}
with
\begin{align}
    H_\text{D}=\chi^\mu  \mathcal C_\mu +\frac{\lambda}{2}\mathcal H
\end{align}
and
\begin{align}
\mathcal C_\mu&\triangleq S_{\mu\nu}\qty(\hat{{p}}^\nu+\Lambda\tdu{0}{\nu}),\qquad\mathcal H\triangleq p^2+\mu^2.
\end{align}
Here, $\chi^\mu$ and $\lambda$ play the role of Lagrange multipliers enforcing the constraints
\begin{align}
    \mathcal C_\mu\approx0,\qquad \mathcal H\approx 0\label{constraints}
\end{align}
to hold on-shell. 

An important point to notice is that $\mathcal C_\mu\approx 0$ is \textit{not} the realization of any specific SSC. Actually, $\Lambda\tdu{0}{A}$ plays the role of a gauge field, whose fixing corresponds to a specific choice of a particular covariant SSC. Choosing a specific SSC thus correspond to a gauge fixation. Examples of values of $\Lambda\tdu{0}{A}$ leading to some of the most common SSCs are given in Table \ref{tab:gauges}. 

When going beyond the linear order in the spin magnitude $\mathcal S$, one of the main differences arising with respect to geodesic motion is that the mass $\mu^2$ appearing in the mass-shell constraint $\mathcal H\approx 0$ is not a constant anymore, as was discussed in Section \ref{sec:MPT}. For $\mu^2$ depending in $x^\mu$ only through the metric and the Riemann tensor, we get back the quadrupole approximation of Chapter \ref{chap:quadrupole}, with this time \cite{Vines:2016unv}
\begin{align}
    J^{\mu\nu\rho\sigma}=\frac{3 p_\alpha v^\alpha}{p^2}\pdv{\mu^2}{R_{\mu\nu\rho\sigma}}.
\end{align}

In the case of the spin-induced quadrupole together with the TD SSC, this mass-shell constraint can be rewritten in a more enlightening way. As discussed in Section \ref{sec:MPT} of the present thesis, there exists a mass-like quantity $\tilde \mu$ which is conserved in that case. From its definition Eq. \eqref{mass_like}, it is straightforward to see that one can perform the rewriting 
\begin{align}
    \tilde \mu^2&=-\qty(g^{\mu\nu}-\frac{\kappa}{\mu^2}\Theta^{\alpha\beta} R\tudu{\mu}{\alpha\beta}{\nu})p_\mu p_\nu+\mathcal O\qty(\mathcal S^3).\label{tilde_mu_metric}
\end{align}
Therefore, at $\os{2}$, one can replace the mass-shell condition defined above by
\begin{align}
    \tilde g^{\mu\nu}p_\mu p_\nu\approx-\tilde\mu^2,
\end{align}
where we defined
\begin{align}
    \tilde g_{\mu\nu}\triangleq g_{\mu\nu}-\frac{\kappa}{\mu^2}\Theta^{\alpha\beta} R_{\mu\alpha\beta\nu}.
\end{align}

Several formal developments regarding the constrained formulation have been worked out in the literature. In the case of a flat background, J. Steinhoff showed that the constraints \eqref{constraints} were first class among themselves for a wide class of functional dependence of $\mu^2$ \cite{Steinhoff:2014kwa}.
The relation between shifts of the representative worldline and the enforcement of covariant spin supplementary conditions viewed as gauge fixations was explored in \cite{Steinhoff:2014kwa,Vines:2016unv}. More recently, P. Ramond discussed in quite details the first class nature of TD condition for a generic curved background up to linear order in the spin expansion \cite{Ramond:2022vhj}.

In the following, we will remain at a more pedestrian level. The next chapter will be devoted to the construction of covariant Hamiltonians for extended test bodies endowed with a spin-induced quadrupole. We will work under the TD condition and explicitly check that the proposed Hamiltonians will preserve it under time evolution. A more careful analysis of the exact status of constraints at second order in the spin expansion remains to be performed. This task will be left for subsequent works.

\chapter{Generally covariant Hamiltonians}\label{chap:covariant_H}

\vspace{\stretch{1}}

This chapter will be devoted to the construction of covariant Hamiltonians describing the motion of extended test bodies over the 14-dimensional phase space introduced in the previous chapter. We will work up to second order in the spin magnitude included and restricting to the spin-induced quadrupole term, under the Tulczyjew-Dixon spin supplementary condition. The Hamiltonians that we will derive will therefore drive the motion only on the constraint surface where this condition is enforced.

The most rigorous approach for deriving an Hamiltonian would have been to apply a Legendre transformation to the Lagrangian Eq. \eqref{generic_lagragian}, while taking care of the associated constraints. Since this Lagrangian is reparametrization invariant, one can convince ourselves that we would have ended up with an identically vanishing Hamiltonian as it was the case for the reparametrization invariant formulation of geodesic motion. The associated Hamiltonian action principle would have therefore be only composed of the constraints, which would have been rather not well-suited to our subsequent purposes.

The construction of the Hamiltonians presented in this chapter will therefore follow another path, based on the more heuristic approach of W. Witzany \textit{et al.} \cite{Witzany:2018ahb}. The original derivation was only aimed at reproducing pole-dipole MPD equations, including all orders in the spin magnitude. However, as argued in \cite{Ramond:2022vhj}, such an exact treatment of the pole-dipole is physically inconsistent with respect to the multipole expansion of the test body structure, since the equations of motion expanded at order $\os{n}$ must account for the presence of the $2^n$-pole moment. The treatment of \cite{Witzany:2018ahb} is therefore physically consistent when linearized in the spin magnitude. Otherwise, the Hamiltonian must account for the presence of higher order multipole moments. The treatment of the spin-induced quadrupole will be the task that will be carried out in this chapter.

The basic idea of Witzany \textit{et al.} is to notice that the TD momentum-velocity relation Eq. \eqref{v_of_p} was actually providing us with the partial derivative of the (unknown) Hamiltonian with respect to the linear momentum. If one manage to integrate this expression, we end up with a candidate expression for being the Hamiltonian driving the motion on the constraint surface. However, there is \textit{a priori} no guarantee that the guessed Hamiltonian will reproduce correctly the MPD equations and preserves the SSC. This shall be checked \textit{a posteriori}. Section \ref{sec:H:conditions} will derive computationally practical sufficient conditions that shall be obeyed for the validity of the Hamiltonian to be granted. Section \ref{sec:H:constr} will be devoted to the derivation of a candidate Hamiltonian, whose validity will be checked in Section \ref{sec:H:check}. The linearized Hamiltonian is finally discussed in Section \ref{sec:H:linear}, and we end the discussion with a couple of remarks. Notice that the procedure used in this chapter is quite generic, and can be extended to derive covariant Hamiltonians valid for other covariant spin supplementary conditions and/or including higher order multipole moments.

\section{Hamilton and MPD equations}\label{sec:H:conditions}
The Hamilton equations for some generic Hamiltonian $H(x^\mu,p_\mu,S^{\mu\nu})$ are directly obtained from the Poisson brackets algebra Eq. \eqref{non_cov_brackets}, and are given by Eq. (36) of \cite{Witzany:2018ahb}:
\begin{subequations}\label{hamilton_spin}
\begin{align}
    &\dv{x^\mu}{\lambda}=\pdv{H}{p_\mu},\label{hamilton_1}\\
    &\dv{p_\nu}{\lambda}+\pdv{H}{x^\nu}-\pdv{H}{S^{\mu\kappa}}\qty(\Gamma^\mu_{\nu\gamma}S^{\gamma\kappa}+\Gamma^\kappa_{\nu\gamma}S^{\mu\gamma})=-\frac{1}{2}R_{\nu\omega\lambda\chi}\pdv{H}{p_\omega}S^{\lambda\chi},\\
    &\dv{S^{\gamma\kappa}}{\lambda}+\Gamma^\gamma_{\nu\lambda}S^{\lambda\kappa}+\pdv{H}{p_\nu} \qty(\Gamma^\kappa_{\nu\lambda}S^{\gamma\lambda})\nonumber\\
    &\quad=\pdv{H}{S^{\mu\nu}}\qty(g^{\gamma\mu}S^{\kappa\nu}-g^{\gamma\nu}S^{\kappa\mu}+g^{\kappa\nu}S^{\gamma\mu}-g^{\kappa\mu}S^{\gamma\nu}).
\end{align}
\end{subequations}
Here, $\lambda$ stands for an arbitrary time parameter. A natural question to be asked is then: for which choice(s) of the Hamiltonian $H$ will these equation reproduce the MPD equations?  Comparing Hamilton equations \eqref{hamilton_spin} with MPD equations \eqref{MPD_1} and \eqref{MPD_2}, one finds that the former will reduce to the latter provided that the following conditions hold:
\begin{subequations}\label{MPD_conditions}
\begin{align} 
    &\pdv{H}{x^\nu}-\pdv{H}{S^{\mu\kappa}}\qty(\Gamma^\mu_{\nu\gamma}S^{\gamma\kappa}+\Gamma^\kappa_{\nu\gamma}S^{\mu\gamma})\approx -\Gamma^\alpha_{\beta\nu}\pdv{H}{p_\beta}p_\alpha-\mathcal F^\nu,\\
    &\pdv{H}{S^{\mu\nu}}\qty(g^{\gamma\mu}S^{\kappa\nu}-g^{\gamma\nu}S^{\kappa\mu}+g^{\kappa\nu}S^{\gamma\mu}-g^{\kappa\mu}S^{\gamma\nu})\approx p^\gamma\pdv{H}{p_\kappa}-p^\kappa\pdv{H}{p_\gamma}+\mathcal{L}^{\gamma\kappa},\\
    &\pb{S^{\mu\nu}p_\nu}{H}\approx 0.
\end{align}
\end{subequations}
The last condition originates from the fact that we will seek for an Hamiltonian valid for the TD SSC, which must be preserved under time evolution of the system. Recall that the weak-equality symbol ``$\approx$'' present in the three equalities above is defined as follows:

\begin{definition}
Two  quantities $F$ and $G$ valued on the phase space are said to be equal ``on-shell'' (which is denoted by the weak equality $F\approx G$) when they are equal on the constraint surface, \textit{i.e.} the submanifold of the phase space where the constraints present in the action principle have been enforced. In the present context,
\begin{itemize}
    \item the TD spin supplementary condition $S^{\mu\nu}p_\nu\approx0$ holds;
    \item the mass shell condition $p_\mu p_\nu g^{\mu\nu}+\mu^2\approx 0$ (or equivalently $p_\mu p_\nu \tilde g^{\mu\nu}+\tilde\mu^2\approx 0$) is enforced.
\end{itemize}
\end{definition}
Notice that here, we do not require $\mu$ to be a constant. As discussed in the previous part of this thesis, the invariant mass $\mu$ is conserved for the pole-dipole linearized MPD equation, but not for the spin-induced quadrupole MPD equations. In both cases, it amounts to enforce an additional algebraic constraint between the dynamical variables.

\section{Construction of the Hamiltonian}\label{sec:H:constr}
We will now build an Hamiltonian generating the MPD equations and including the spin-induced quadrupole term. From now on, we set the time parameter $\lambda$ to the proper time of the body $\tau$. Our starting point is the relation between the four-velocity and the linear momentum for the TD condition:
\begin{align}
    v^\mu\approx\frac{1}{\mu}\qty[p^\mu+\qty(D^{\mu\nu}-\frac{1}{\mu}\mathcal L^{\mu\nu})p_\nu]+\mathcal O\qty(\mathcal S^3).\label{TD_velocity}
\end{align}
We recall the identities
\begin{subequations}
\begin{align}
    D\tud{\mu}{\nu}&\approx\frac{1}{2\mu^2}S^{\mu\lambda}R_{\lambda\nu\rho\sigma}S^{\rho\sigma},\\
    \mathcal L^{\mu\nu}p_\nu&\approx\frac{\kappa}{\mu}\Pi^{\mu\nu}R_{\nu\alpha\beta\gamma}\Theta^{\alpha\beta}p^\gamma+\mathcal O\qty(\mathcal S^4).
\end{align}
\end{subequations}
Using Eq. \eqref{hamilton_1} and the relations above, the relation \eqref{TD_velocity} takes the explicit form
\begin{align}
    \pdv{H_0}{p_\mu}&\approx\frac{1}{\mu}p^\mu+\frac{1}{2\mu^3}S^{\mu\lambda}R_{\lambda\nu\alpha\beta}S^{\alpha\beta}p^\nu\nonumber\\
    &\quad-\frac{\kappa}{\mu^3}R\tud{\mu}{\alpha\beta\gamma}\Theta^{\alpha\beta} p^\gamma-\frac{\kappa}{\mu^5}p^\mu p^\nu R_{\nu\alpha\beta\gamma}\Theta^{\alpha\beta}p^\gamma
    +\mathcal O\qty(\mathcal S^3).\label{HQ_deriv_p}
\end{align}
We have here denoted the Hamiltonian $H_0$ because, as we will see later, this construction will lead to an Hamiltonian that is vanishing on-shell. Now, still following \cite{Witzany:2018ahb}, for any arbitrary functions $G^\mu$ and $F$ of the dynamical variables $x^\mu$, $p_\mu$ and $S^{\mu\nu}$, the following relations hold \textit{on-shell}:
\begin{subequations}
\begin{align}
    \pdv{p_\mu}\qty[\frac{1}{2}\qty(g^{\alpha\beta}p_\alpha p_\beta+\mu^2)F]&\approx F p^\mu,\\
    \pdv{p_\mu}\qty[G_\alpha p_\beta S^{\alpha\beta}]&\approx G_\alpha S^{\alpha\mu},\\
    \pdv{p_\mu}\qty[S^{\nu\lambda}R_{\lambda\rho\alpha\beta}S^{\alpha\beta}p^\rho]&=S^{\nu\lambda}R\tdud{\lambda}{\mu}{\alpha\beta}S^{\alpha\beta},\\
    \pdv{p_\mu}\qty[\frac{1}{2}p_\nu p_\rho R\tudu{\nu}{\alpha\beta}{\rho}\Theta^{\alpha\beta}]&\approx R\tud{\mu}{\alpha\beta\rho}\Theta^{\alpha\beta}p^\rho.
\end{align}
\end{subequations}
Making use of these identities and treating $\mu$ as being independent of the dynamical variables, one can infer the following form for the Hamiltonian:
\begin{align}
    \begin{split}
    H_0&=\frac{1}{2\mu}\qty[\qty(g^{\mu\nu}+2D^{\mu\nu}-\frac{\kappa}{\mu^2}\qty(\frac{g^{\rho\sigma}p_\rho p_\sigma}{\mu^2}+2)R\tudu{\mu}{\alpha\beta}{\nu}\Theta^{\alpha\beta})p_\mu p_\nu+\mu^2]\\
    &\quad+\mathcal O\qty(\mathcal S^3).\label{quadrupole_hamiltonian}
    \end{split}
\end{align}
Actually, this expression shall in principle be supplemented by an integration constant depending upon the variables $x^\mu$ and $S^{\mu\nu}$. However, this constant can be set to zero since the Hamiltonian \eqref{quadrupole_hamiltonian} already satisfies the conditions for generating the MPD equations, as will be checked in the next section. 

The final step is to rewrite our candidate Hamiltonian in a way that allows to replace the dependence in the non-constant dynamical mass $\mu$ by the constant shifted mass-like quantity $\tilde \mu$. This will allow huge simplifications while checking the equations of motion. A conventional Taylor expansion can be used to show that
\begin{align}
\frac{1}{\mu}=\frac{1}{\tilde\mu}\qty(1+\frac{\kappa}{2\tilde\mu^2}p^\alpha\Theta^{\beta\gamma}p^\delta R_{\alpha\beta\gamma\delta})+\mathcal O\qty(\mathcal S^3).   
\end{align}
Using this expression as well as Eq. \eqref{tilde_mu_metric}, we end up with
\begin{align}
    \begin{split}
    H_0&=\frac{1}{2\tilde\mu}\qty[\qty(\tilde g^{\mu\nu}+2D^{\mu\nu}-\frac{\kappa}{2\tilde\mu^2}\qty(\frac{g^{\rho\sigma} p_\rho p_\sigma}{\tilde\mu^2}+1)R\tudu{\mu}{\alpha\beta}{\nu}\Theta^{\alpha\beta})p_\mu p_\nu+\tilde\mu^2]\\
    &\quad+\mathcal O\qty(\mathcal S^3).\label{quadratic_hamiltonian_tilde}
    \end{split}
\end{align}
Recalling that $p_\mu D^{\mu\nu}\approx 0$, we directly see that this expression does vanish on-shell,
\begin{align*}
    H_0\approx 0+\mathcal O\qty(\mathcal S^3).
\end{align*}
Notice that it is possible to obtain an Hamiltonian which does not vanish on-shell by shifting $H_0$ by a constant quantity. This operation leaves the dynamics of the system invariant, since Hamilton's equations only depend on the derivatives of the Hamiltonian with respect to the dynamical variables. Choosing the value of the shift to be equal to $-\tilde \mu/2$, we end up with a new Hamiltonian $H\triangleq H_0-\frac{\tilde\mu}{2}$ equal to
\begin{align}
\boxed{
    H=\frac{1}{2\tilde\mu}\qty[\qty(\tilde g^{\mu\nu}+2D^{\mu\nu}-\frac{\kappa}{2\tilde\mu^2}\qty(\frac{g^{\rho\sigma} p_\rho p_\sigma}{\tilde\mu^2}+1)R\tudu{\mu}{\alpha\beta}{\nu}\Theta^{\alpha\beta})p_\mu p_\nu]+\mathcal O\qty(\mathcal S^3).
    }\label{HQ_offshell}
\end{align}
On-shell, this expression simply reduces to
\begin{align}
    \boxed{
    H\approx\frac{1}{2\tilde \mu}\tilde g^{\mu\nu}p_\mu p_\nu\approx-\frac{\tilde\mu}{2}.
    }\label{HQ_onshell}
\end{align}
The Hamiltonian then reduces to (minus one half of) the conserved mass $\tilde\mu$, which is the exact same situation than in the geodesic case (where we had $\bar H\approx-\frac{\mu}{2}$, $\mu$ being conserved along the motion), as was discussed in Chapter \ref{chap:hamilton_kerr}.

\section{Check of Hamilton equations}\label{sec:H:check}
We will now plug our candidate Hamiltonians in Eqs. \eqref{MPD_conditions}, which will enable to prove that they will indeed generate the MPD equations. Actually, since the terms of the equations that are (in)dependent of $\kappa$ shall vanish independently, one can split these equations in two parts, one containing the terms proportional to $\kappa$ (originating from the quadrupole terms of the EOMs) and the other the ones that are independent of it (originating from the pole dipole part of the EOMs). It is therefore relevant to split the Hamiltonian as
\begin{align}
    H=H_\textsc{PD}+\kappa H_\textsc{Q},\label{H_decompos}
\end{align}
with $H_\textsc{PD}=H\eval_{\kappa=0}$ and  $H_\textsc{Q}=\pdv{\kappa} H\eval_{\kappa=0}$. We find
\begin{subequations}
    \begin{align}
    H_\textsc{PD}&=\frac{1}{2\tilde \mu}\qty(g^{\mu\nu}+2D^{\mu\nu})p_{\mu}p_\nu,\\
    H_\textsc{Q}&=-\frac{1}{4\tilde \mu^3}\qty(\frac{g^{\rho\sigma}p_\rho p_\sigma}{\tilde \mu^2}+3)R\tudu{\mu}{\alpha\beta}{\nu}\Theta^{\alpha\beta}p_\mu p_\nu
\end{align}
\end{subequations}
and the verification of Eqs. \eqref{MPD_conditions} will be valid for both Hamiltonians. Using decomposition \eqref{H_decompos}, 
Eqs. \eqref{MPD_conditions} split in two sets of three equations, one corresponding to the pole-dipole part and one to the quadrupole part:
\begin{subequations}
    \begin{align}
    &\pdv{H_\textsc{PD}}{x^\nu}+2\pdv{H_\textsc{PD}}{S^{\mu\rho}}\Gamma^{[\mu}_{\nu\lambda}S^{\rho]\lambda}\approx-\Gamma^\mu_{\nu\rho}\pdv{H_\textsc{PD}}{p_\rho}p_\mu,\label{PD1}\\
    &\pdv{H_\textsc{PD}}{S^{\mu\nu}}\qty[g^{\rho\mu}S^{\sigma\nu}+\qty(\text{usual permutations})]\approx p^\rho\pdv{H_\textsc{PD}}{p_\sigma}-\qty(\rho\leftrightarrow\sigma),\label{PD2}\\
    &2S^{\mu[\rho}p^{\sigma]}\pdv{H_\textsc{PD}}{S^{\rho\sigma}}-\Gamma^\nu_{\lambda\kappa}S^{\mu\lambda}\pdv{H_\textsc{PD}}{p_\kappa}p_\nu-S^{\mu\nu}\pdv{H_\textsc{PD}}{x^\nu}-\tilde\mu^2 D\tud{\mu}{\nu}\pdv{H_\textsc{PD}}{p_\nu}\approx 0,\label{PD3}\\
    &\pdv{H_\textsc{Q}}{x^\nu}+2\pdv{H_\textsc{Q}}{S^{\mu\rho}}\Gamma^{[\mu}_{\nu\lambda}S^{\rho]\lambda}\approx-\Gamma^\mu_{\nu\rho}\pdv{H_\textsc{Q}}{p_\rho}p_\mu-\mathcal F_\nu,\label{Q1}\\
    &\pdv{H_\textsc{Q}}{S^{\mu\nu}}\qty[g^{\rho\mu}S^{\sigma\nu}+\qty(\text{usual permutations})]\approx \qty[p^\rho\pdv{H_\textsc{Q}}{p_\sigma}-\qty(\rho\leftrightarrow\sigma)]+\mathcal L^{\rho\sigma},\label{Q2}\\
    &2S^{\mu[\rho}p^{\sigma]}\pdv{H_\textsc{Q}}{S^{\rho\sigma}}-\Gamma^\nu_{\lambda\kappa}S^{\mu\lambda}\pdv{H_\textsc{Q}}{p_\kappa}p_\nu-S^{\mu\nu}\pdv{H_\textsc{Q}}{x^\nu}-\tilde\mu^2 D\tud{\mu}{\nu}\pdv{H_\textsc{Q}}{p_\nu}\approx 0.\label{Q3}
\end{align}
\end{subequations}
We will now check these equations one by one.

\paragraph{Pole-dipole equations.}
Let us begin with the pole dipole equations, and compute the relevant derivatives of $H_\textsc{PD}$. Two useful preliminary identities are
\begin{align}
    \partial_\nu g_{\alpha\beta}=2\Gamma^\lambda_{\nu(\alpha}g_{\beta)\lambda},\qquad\partial_\nu g^{\alpha\beta}=-2\Gamma^{(\alpha}_{\nu\lambda}g^{\beta)\lambda}.
\end{align}
These equations are straightforward consequences of the metric compatibility condition $\nabla_\nu g_{\alpha\beta}=0$. Consequently,
\begin{subequations}
    \begin{align}
    2\tilde \mu\pdv{H_\textsc{PD}}{x^\nu}&=\qty(\partial_\nu g^{\alpha\beta}+\frac{1}{\tilde\mu^2}S^{\alpha\lambda}\partial_\nu R\tdud{\lambda}{\beta}{\gamma\delta}S^{\gamma\delta})p_\alpha p_\beta\nonumber\\
    &\approx \partial_\nu g^{\alpha\beta}p_\alpha p_\beta
    =-2\Gamma^{\alpha}_{\nu\beta}p_\alpha p^\beta,\\
    2\tilde\mu\pdv{H_\textsc{PD}}{p_\rho}&=2\qty(g^{\rho\nu}p_\nu+D^{\rho\nu}p_\nu+p_\mu D^{\mu\rho})\nonumber\\
    &\approx 2\qty(p^\rho+D^{\rho\nu}p_\nu),\\
    2\tilde\mu\pdv{H_\textsc{PD}}{S^{\mu\rho}}&\approx\frac{1}{\tilde\mu^2}p_\mu R_{\rho\sigma\gamma\delta}S^{\gamma\delta}p^\sigma.
\end{align}
\end{subequations}
These identities can now be used to check Eqs. \eqref{PD1}-\eqref{PD3}:
\begin{itemize}
    \item \textit{Eq. \eqref{PD1}:} one has
    \begin{align}
        2\tilde\mu\,\qty(\text{LHS})&\approx-2\Gamma^\alpha_{\nu\beta}p_\alpha p^\beta -2p_\mu \Gamma^\mu_{\nu\lambda}D^{\lambda\sigma}p_\sigma\approx 2\tilde\mu\,\qty(\text{RHS});
    \end{align}
    \item \textit{Eq. \eqref{PD2}:} we obtain
    \begin{align}
        2\tilde\mu\,\qty(\text{LHS})&\approx
        p^{[\rho}D^{\sigma]\lambda}p_\lambda
        \approx 2\tilde\mu\,\qty(\text{RHS});
    \end{align}
    \item \textit{Eq. \eqref{PD3}:} a similar computation shows that
    \begin{align}
        2\tilde\mu\,\qty(\text{LHS})&\approx
       0
        \approx 2\tilde\mu\,\qty(\text{RHS}).
    \end{align}
\end{itemize}

\paragraph{Quadrupole equations.} In a similar fashion, we show that
\begin{subequations}
    \begin{align}
    -4\tilde \mu^3\pdv{H_\textsc{Q}}{x^\nu}&\approx -\frac{2}{\tilde \mu^2}\Gamma^\rho_{\nu\lambda}p_\rho p^\lambda R\tudu{\alpha}{\beta\gamma}{\delta}\Theta^{\beta\gamma} p_\alpha p_\delta+\partial_\nu R\tudu{\alpha}{\beta\gamma}{\delta}\Theta^{\beta\gamma} p_\alpha p_\delta\nonumber\\
    &\quad+2\Gamma^\kappa_{\nu\lambda}S^{\beta\lambda}S\tud{\gamma}{\kappa}R\tudu{\alpha}{\beta\gamma}{\delta}\Theta^{\beta\gamma} p_\alpha p_\delta,\\
    -4\tilde\mu^3\pdv{H_\textsc{Q}}{p_\rho}&\approx\frac{2}{\tilde\mu^2}p^\rho R\tudu{\alpha}{\beta\gamma}{\delta}\Theta^{\beta\gamma} p_\alpha p_\delta+2R\tudu{\rho}{\beta\gamma}{\delta}\Theta^{\beta\gamma} p_\delta,\\
    -4\tilde \mu^3\pdv{H_\textsc{Q}}{S^{\mu\nu}}&\approx 2 R\tudu{\alpha}{\mu\gamma}{\delta}S\tud{\gamma}{\nu} p_\alpha p_\delta.
\end{align}
\end{subequations}
These relations can be directly used to verify equations \eqref{Q1}-\eqref{Q3}. The computation is quite straightforward, and we will not reproduce it here because it is not particularly enlightening. The only technical point involved shows up while checking \eqref{Q1}. It consists in noticing that
\begin{align}
    \partial_\nu R\tudu{\alpha}{\beta\gamma}{\delta}\Theta^{\beta\gamma} p_\alpha p_\delta&=\bigg(\nabla_\nu R\tudu{\alpha}{\beta\gamma}{\delta}
    -2\Gamma^\alpha_{\nu\lambda}R\tudu{\lambda}{\beta\gamma}{\delta}+2\Gamma^\lambda_{\nu\beta}R\tudu{\alpha}{\lambda\gamma}{\delta}
    \bigg)\Theta^{\beta\gamma} p_\alpha p_\delta.
\end{align}
The first term of the RHS of this equation will cancel with the force term $\mathcal F_\nu$, whereas the two other will respectively cancel with other terms involving Christoffel symbols. Our candidate Hamiltonian therefore reproduces correctly the MPD equations endowed with the spin-induced quadrupole term and preserves the TD spin supplementary condition along time evolution.

\section{Linearized Hamiltonian}\label{sec:H:linear}
When discussing the integrability of MPD equations in next chapter, it will be of prime importance to know which Hamiltonian can be used to generate the linearized pole-dipole motion. It is given by the linearized version of \eqref{HQ_offshell}, which is simply
\begin{align}
    H_{\text{lin}}=\frac{1}{2\mu}g^{\mu\nu}p_\mu p_\nu\approx-\frac{\mu}{2}.\label{lin_H}
\end{align}
Therefore, when expressed in the non-symplectic coordinates, the linearized MPD equations are generated by the same Hamiltonian than the geodesic motion, the only (but huge) difference lying in the phase space and the Poisson brackets algebra being considered. This difference can be made manifest in the Hamiltonian by turning to  quasi-symplectic coordinates. In terms of these variables, the linearized Hamiltonian takes the form
\begin{align}
    H_\text{lin}=\frac{1}{2\mu}\qty(g^{\mu\nu} P_\mu P_\nu+P^\mu \underline e\tdu{A}{\lambda}\underline e_{B\lambda;\mu}S^{AB}).\label{lin_H_quasi_sympl}
\end{align}
This linearized Hamiltonian is standard and has already been extensively used in the literature, see \cite{Witzany:2018ahb,witzany2019spinperturbed,Ramond_2021} and references therein.

\section{Concluding remarks}

We are now in possession of a covariant Hamiltonian generating the MPD equations under the TD condition in the spin-induced quadrupole approximation. The Hamiltonian valid at quadrupole order is given in Eq. \eqref{HQ_offshell}, whereas its pole-dipole linearized version is provided in Eq. \eqref{lin_H}. These Hamiltonian will be at the heart of the two applications discussed in the last chapters of this thesis: the discussion of the non-integrability of MPD equations in Kerr spacetime and the investigation of the associated Hamilton-Jacobi equation. 


\chapter{Non-integrability of the MPD equations in Kerr spacetime}
\label{chap:integrability}
\chaptermark{\textsc{Non-integrability in Kerr spacetime}}

In this chapter, we will explore the first application of the covariant Hamiltonian formalism for extended bodies in Kerr spacetime: the discussion of (Liouville) integrability of the motion. Integrability of MPD equations in Kerr and Schwarzschild spacetimes can be summarized by the following table \cite{Ramond:2022vhj}:
\begin{table}[h!]
    \centering
    \renewcommand{\arraystretch}{1.5}
    \begin{tabular}{c|cc}
       & \textsc{Schwarzschild} ($a=0$) & \textsc{Kerr} ($a\neq 0$) \\\hline
       $\os{}$  & integrable & non-integrable \\
       $\os{2}$  & non-integrable & non-integrable \\
    \end{tabular}
    \caption{Status of the integrability for test bodies in Schwarzschild and Kerr spacetimes. Table after \cite{Ramond:2022vhj}.}
    \label{tab:integrable}
\end{table}

The status of the proofs of the statements formulated in Table \ref{tab:integrable} differs between the different cells: in Schwarzschild at linear order in $\mathcal S$, the proof of integrability of MPD equations has been provided in several works, both by using non-covariant Hamiltonians \cite{Kunst:2015tla} and covariant ones \cite{Ramond:2022vhj}. This is the only case where Liouville integrability still holds when departing from geodesic motion. Still in Schwarzschild, but at quadratic order in $\mathcal S$, the breaking of integrability is suggested from numerical studies, see \cite{Zelenka:2019nyp} and references therein. However, an analytical study of the breaking of integrability in this case is still missing.

Turning to Kerr spacetime, the breaking of integrability at linear order in the spin magnitude $\mathcal S$ has been demonstrated both from numerical perspective \cite{Kunst:2015tla} and from analytical one \cite{Compere:2021kjz}. The status as second order simply follows from the non-integrability already present at first order.

In this chapter, we will review the analytical argument leading to the conclusion that linearized MPD equations in Kerr spacetime do not form an integrable system. In Part \ref{part:conserved_quantities} of this thesis, we have shown that -- up to $\os{2}$ corrections -- the only non-trivial, polynomial quantities conserved for linearized MPD equations were the dynamical mass $\mu$ and the four invariants
\begin{align}
    \mathcal I_A&=\qty(\mathcal E,\,\mathcal L,\,\mathcal Q_Y,\,\mathcal Q_R).\label{lin_qties}
\end{align}
Using the symplectic structure introduced in Chapter \ref{chap:symplectic}, a direct computation of the Poisson brackets between these quantities show that they fail to be in involution, since
\begin{align}
    \pb{\mathcal Q_Y}{\mathcal Q_R}=\os{1}.
\end{align}
This is the only non-vanishing bracket at linear order in $\mathcal S$, and one can convince ourselves that there is no hope to find another combination of conserved quantities which are Poisson-commuting among themselves, thus suggesting that linearized MPD equations are not integrable in Kerr spacetime.

This Chapter is organized as follows: Section \ref{sec:hamilton} will discuss the computation of the Poisson brackets between the conserved quantities $\mathcal I_A$, while Section \ref{sec:NI} briefly explores the consequences of perturbative non-integrability of a dynamical system.

\section{Non-integrability of the linearized MPD equations in Kerr spacetime}\label{sec:hamilton}

From the discussion of Chapters \ref{chap:symplectic} and \ref{chap:covariant_H}, the linearized MPD equations endowed with the Tulzcyjew-Dixon spin supplementary condition can be described by a $N=5$ Hamiltonian system, whose evolution is driven by the Hamiltonian
\begin{align}
    H_\text{lin}=\frac{1}{2\mu}\qty(g^{\mu\nu} P_\mu P_\nu+P^\mu \underline e\tdu{A}{\lambda}\underline e_{B\lambda;\mu}S^{AB})\approx-\frac{\mu}{2}.
\end{align}
The counting of the independent degrees of freedom is advantageously performed using the symplectic coordinates introduced in Chapter \ref{chap:symplectic}. Without loss of generality, we will choose Witzany's coordinates for the discussion. We start from the full 14-dimensional phase space ($N=7$ Hamiltonian system). Turning to symplectic coordinates shows that the motion can be parametrized by only 12 variables ${x^\mu,P_\mu,\phi,A,\psi,B}$, two two extra degrees of freedom corresponding to the Casimir invariants $\mathcal S^2$ and $\mathcal S^2_*$. We are then left with a 12-dimensional phase space ($N=6$ Hamiltonian system). Moreover, expressing the TD spin supplementary condition in an adapted background tetrad frame (see Section \ref{sec:ssc:ortho}) allows to rewrite it as
\begin{align}
    S^{0A}=0
\end{align}
in that frame. Comparing this requirement with Eq. \eqref{background_canonical} implies that, under the TD SSC, the Hamiltonian will not depend upon the coordinate $\psi$, which is thus cyclical. Its conjugated moment $B$ is therefore conserved, and we end up with a final 10-dimensional phase space ($N=5$ Hamiltonian system), as announced.

We now stand in a comfortable position for discussing the integrability of the linearized MPD equations in Kerr spacetime. Since the Hamiltonian $H_\text{lin}\approx-\frac{\mu}{2}$ is itself a constant of the motion, the system will be integrable only if the four independent constants of the motion given in Eq. \eqref{lin_qties} are in involution:
\begin{align}
    \pb{\mathcal I_A}{\mathcal I_B}=0,\qquad\forall A,B\in 1,\ldots,N.
\end{align}

\begin{table}[t!]
    \begin{center}
        \begin{tikzpicture}[xscale=.75,yscale=.85]
        \cell{1.5}{0}{$\mathcal S$}{black!10};
        \cell{3}{0}{$\mathcal E$}{black!15};
        \cell{4.5}{0}{$\mathcal L$}{black!10};
        \cell{6}{0}{$\mathcal Q_Y$}{black!15};
        \cell{7.5}{0}{$\mathcal Q_R$}{black!10};
        
        \cell{0}{-.75}{$\mu$}{black!15};
        \cell{0}{-1.5}{$\mathcal S$}{black!10};
        \cell{0}{-2.25}{$\mathcal E$}{black!15};
        \cell{0}{-3}{$\mathcal L$}{black!10};
        \cell{0}{-3.75}{$\mathcal Q_Y$}{black!15};
        
        \cell{1.5}{-.75}{$0$}{green!60};
        \cell{3}{-.75}{$0$}{green!60};
        \cell{4.5}{-.75}{$0$}{green!60};
        \cell{6}{-.75}{$\mathcal O(\mathcal S^2)$}{green!20};
        \cell{7.5}{-.75}{$\mathcal O(\mathcal S^2)$}{green!20};
        
        \cell{3}{-1.5}{$0$}{green!60};
        \cell{4.5}{-1.5}{$0$}{green!60};
        \cell{6}{-1.5}{$0$}{green!60};
        \cell{7.5}{-1.5}{$0$}{green!60};
        
        \cell{4.5}{-2.25}{$0$}{green!60};
        \cell{6}{-2.25}{$0$}{green!60};
        \cell{7.5}{-2.25}{$\mathcal O(\mathcal S^2)$}{green!20};
        
        \cell{6}{-3}{$0$}{green!60};
        \cell{7.5}{-3}{$\mathcal O(\mathcal S^2)$}{green!20};
        
        \cell{7.5}{-3.75}{$\mathcal O(\mathcal S)$}{red!60};
        
        
        \cell{1.5}{-1.5}{}{black!5};
        
        \cell{1.5}{-2.25}{}{black!5};
        \cell{3}{-2.25}{}{black!5};
        
        \cell{1.5}{-3}{}{black!5};
        \cell{3}{-3}{}{black!5};
        \cell{4.5}{-3}{}{black!5};
        
        \cell{1.5}{-3.75}{}{black!5};
        \cell{3}{-3.75}{}{black!5};
        \cell{4.5}{-3.75}{}{black!5};
        \cell{6}{-3.75}{}{black!5};
        
        \end{tikzpicture}
    \end{center}
    \caption{Poisson brackets between the first integrals $\mathcal I_A$ of linearized MPTD equations. The bracket that is generally non-vanishing at order $\mathcal O(\mathcal S)$ is represented in red, the other ones in green.}
    \label{tab:pb}
\end{table}
The results of the computations of these brackets are summarized in Table \ref{tab:pb}.

We begin by deriving an useful result about Poisson brackets: let $f(X)$ be an analytic function of some dynamical quantity $X$ such that the coefficients $f_n$ of the Taylor expansion $f(X)=\sum_{n=0}^{+\infty}f_n\frac{X^n}{n!}$ are constant, and let $Y$ be a dynamical quantity such that $\pb{X}{Y}\neq 0$. Then, from the Leibniz rule for Poisson brackets we have $\pb{X^n}{Y}=nX^{n-1}\pb{X}{Y}$ and one has the chain rule property
\begin{align}
\begin{split}
        \pb{f(X)}{Y}&=\sum_{n=0}^{+\infty}\frac{f_n}{n!}\pb{X^n}{Y}\\
    &=\sum_{n=1}^{+\infty}f_n\frac{X^{n-1}}{(n-1)!}\pb{X}{Y}\\
    &=\dv{f}{X}\pb{X}{Y}.
\end{split}
\end{align}
This relation is easily generalized to any analytic function of $n$ variables $X^\alpha$ ($\alpha=1,\ldots,n$):
\begin{align}
    \pb{f(X^\alpha)}{Y}&=\pdv{f}{X^\lambda}\pb{X^\lambda}{Y}.\label{pb:chain_rule}
\end{align}

\subsection{$\pb{\mathcal I_{A}}{\mathcal E}$-type brackets}
 
\subsubsection{$\mathbf{\mathcal I_{A}=\mathcal L}$}
Using the identities  $\nabla_\alpha \xi^\mu=\Gamma^\mu_{\alpha t}$ and $\nabla_\alpha \eta^\mu=\Gamma^\mu_{\alpha\varphi}$, the bracket reads
\begin{align}
    \begin{split}
    \pb{\mathcal L}{\mathcal E}&=\pb{p_t}{p_\varphi}+\frac{1}{2}\nabla_\alpha \xi_\beta \pb{S^{\alpha\beta}}{p_\varphi}\\
    &\quad-\frac{1}{2}\nabla_\alpha \eta_\beta \pb{S^{\alpha\beta}}{p_t}-\frac{1}{4}\nabla_\alpha \eta_\beta  \nabla_\gamma  \xi_\delta \pb{S^{\alpha\beta}}{S^{\gamma\delta}}\\
    &=-\frac{1}{2}R_{t\varphi \alpha\beta}S^{\alpha\beta}+\qty(\nabla_\alpha \eta_\beta  \Gamma ^\alpha_{\lambda  t}-\nabla_\alpha \xi_\beta \Gamma ^\alpha_{\lambda \varphi})S^{\lambda \beta}+\nabla_\alpha \eta^\lambda \nabla_\lambda  \xi_\beta  S^{\alpha\beta}\\
    &=-\frac{1}{2}R_{t\varphi \alpha\beta}S^{\alpha\beta}-\nabla_\alpha \eta^\lambda \nabla_\lambda  \xi_\beta  S^{\alpha\beta}\\
    &=0.
    \end{split}
\end{align}
The last equality follows from the fact that the axisymmetry of Kerr spacetime together with the definition of the Riemann tensor enforce the relation
\begin{align}
    R_{t\varphi \alpha\beta}S^{\alpha\beta}&=2\Gamma^\alpha_{t\lambda}\Gamma^\lambda_{\varphi \beta}S\tdu{\alpha}{\beta}=-2\nabla_\alpha \eta^\lambda\nabla_\lambda \xi_\beta S^{\alpha\beta}
\end{align}
to hold.

\subsubsection{$\mathcal I_{A}=\mathcal Q_Y$}

The axisymmetric character of Kerr spacetime allows to consider any background tetrad $\underline e\tdu{A}{\mu}$ to be independent of $t$ and $\varphi$. This yields $\pb{p_{t,\varphi}}{\underline e\tdu{A}{\mu}}=-\partial_{t,\varphi}\underline e\tdu{A}{\mu}=0$. Consequently,
\begin{align}
    \pb{p_{t,\varphi}}{S^{AB}}=\pb{p_{t,\varphi}}{S^{\mu\nu}}\underline e\tdu{A}{\mu}\underline e\tdu{B}{\nu}.\label{ss_tetrad}
\end{align}
Using this last equation and the fact that $\mathcal Q_Y$ commute with $x^\mu$, we get
\begin{align}
    \begin{split}
    \pb{\mathcal Q_Y}{\mathcal E}&=\pb{p_t}{\mathcal Q_Y}+\frac{1}{2}\nabla_A \xi_B\pb{S^{AB}}{\mathcal Q_Y}\\
    &\overset{\eqref{ss_tetrad}}{=}-4\bigg[r\qty(\Gamma^{[1}_{At}-\nabla_A \xi^{[1})S^{0]A}+a\cos\theta\qty(\Gamma^{[3}_{At}-\nabla_A\xi^{[3})S^{2]A}\bigg]\\
    &=0.
    \end{split}
\end{align}

\subsubsection{$\mathcal I_{\hat A}=\mathcal Q_R$} Using the identity
\begin{align}
    \pb{\mathcal E_0}{\mathcal E}&=\frac{1}{2}\nabla_\alpha \xi_\beta\Gamma^\alpha_{\lambda t}S^{\lambda\beta}=-\frac{1}{2}\nabla_\beta \xi_\alpha\nabla_\lambda \xi^\alpha S^{\lambda\beta}=0
\end{align}
and the chain rule \eqref{pb:chain_rule}, one has
\begin{align}
    \begin{split}
    \pb{\mathcal Q_R}{\mathcal E}&=2p_\mu K^{\mu\nu}\pb{p_\nu}{\mathcal E}-2\mu\pb{S^\alpha\partial_\alpha\mathcal Z}{\mathcal E}\\
    &=-p_\mu K^{\mu\nu}\qty(2\pb{p_\nu}{p_t}+\pb{p_\nu}{\nabla_\alpha \xi_\beta S^{\alpha\beta}})\\
    \hspace{-12pt}&\quad+\mu\qty(2\pb{S^\alpha\partial_\alpha\mathcal Z}{p_t}+\pb{S^\alpha\partial_\alpha\mathcal Z}{\nabla_\mu \xi_\nu S^{\mu\nu}}).
    \end{split}
\end{align}
Finally, making use of the identity
\begin{align}
  R_{t\beta\mu\nu}S^{\mu\nu}&=\partial_\beta\qty(\nabla_\mu  \xi_\nu)S^{\mu\nu}+2\nabla_\alpha \xi_\nu\Gamma^\alpha_{\beta\lambda}S^{\nu\lambda},
\end{align}
the two first Poisson brackets of this expression can be shown to cancel mutually, and we are left with
\begin{align}
    \begin{split}
    \hspace{-6pt}\pb{\mathcal Q_R}{\mathcal E}&=\frac{1}{2}\epsilon^{\alpha\beta\gamma\delta}S_{\gamma\delta}\,\partial_\alpha\mathcal Z\qty[R_{t\beta\mu\nu}-\partial_\beta\qty(\nabla_\mu \xi_\nu)]S^{\mu\nu}\\
    &=\epsilon^{\alpha\beta\gamma\delta}\partial_\alpha\mathcal Z S_{\gamma\delta}\nabla_\rho \xi_\nu\Gamma^\rho_{\lambda\beta}S^{\nu\lambda}\\
    &=\mathcal O(\mathcal S^2).
    \end{split}
\end{align}

\subsection{$\pb{\mathcal I_{A}}{\mathcal L}$-type brackets}

The computations are identical to the $\pb{\mathcal I_{A}}{\mathcal E}$-type case, but with $\xi^\alpha\to \eta^\alpha$. We consequently find
\begin{align}
    \pb{\mathcal Q_Y}{\mathcal L}&=0,\qquad
    \pb{\mathcal Q_R}{\mathcal L}=\mathcal O(\mathcal S^2).
\end{align}

\subsection{$\pb{\mathcal I_{A}}{\mathcal Q_Y}$-type brackets}

The final bracket to be computed takes the form
\begin{align}
    \begin{split}
    \pb{\mathcal Q_R}{\mathcal Q_Y}&=2p_\mu K^{\mu\nu}\pb{p_\nu}{\mathcal Q_Y}-2\mu\,\partial_\alpha\mathcal Z\pb{S^\alpha}{\mathcal Q_Y}\\
    &=-2p_\mu K^{\mu\nu}\qty(\partial_\nu Y^*_{\alpha\beta} S^{\alpha\beta}+2\Gamma^\alpha_{\rho\nu}S^{\beta\rho}Y^*_{\alpha\beta})\\
    &\quad-4\epsilon\tud{\alpha\beta}{\gamma\delta}\partial_\alpha\mathcal Z p_\beta Y\tud{*\gamma}{\nu}S^{\delta\nu}+\mathcal O(\mathcal S^2)\\
    &=-2p_\mu K^{\mu\nu}\nabla_\nu Y^*_{\alpha\beta}S^{\alpha\beta}-2\epsilon^{\alpha\beta\gamma\delta}\partial_\alpha\mathcal Z p_\beta \epsilon_{\gamma\nu\rho\sigma}Y^{\rho\sigma}S\tdu{\delta}{\nu}+\mathcal O(\mathcal S^2)\\
    &=4p_\mu S^{\alpha\beta}\qty(K\tud{\mu}{\alpha}\xi_\beta+Y\tud{\mu}{\alpha}Y\tud{\lambda}{\beta}\xi_\lambda)+\mathcal O(\mathcal S^2).
    \end{split}
\end{align}
This expression is generally non-vanishing at order $\mathcal O(\mathcal S)$.

After analysis, it is thus found that only the Poisson bracket $\pb{\mathcal Q_Y}{\mathcal Q_R}$ is non-vanishing at order $\mathcal O(\mathcal S)$, as displayed in Table \ref{tab:pb}.
The four linearly independent first integrals $\mathcal I_A$ are consequently \textit{not} in involution at the linear level, and the linearized MPTD equations do not form an integrable system in the sense of Liouville.

\section{Non-integrability: is that so bad?}\label{sec:NI}

Contrarily to geodesic motion, extended test body motion is not anymore integrable in Kerr spacetime, even at linear order in $\mathcal S$. As discussed in the introduction of this thesis, one of the main benefits of integrability of geodesic motion was that it allowed to turn to action-angle formulation (see Chapter \ref{chap:hamilton_kerr}), thus making explicit the fundamental frequencies of the bounded motion, which can subsequently be used for building models of inspiralling self-forced motion \cite{Hinderer:2008dm}. Actually, since the integrability is only broken by a small non-integrable perturbation, we can expect that not so much has been lost! In particular, one should expect the very powerful result known as the \defining{Kolmogorov–Arnold–Moser (KAM) theorem} to apply. The exact statement of this theorem is rather involved \cite{jose:1998} and its description goes beyond the scope of this thesis. An intuitive formulation can be found in \cite{Goldstein2001}, and takes the following form:
\begin{theorem}[Kolmogorov–Arnold–Moser]
If the bounded motion of an integrable Hamiltonian $H_0$ is disturbed by a small perturbation $\Delta H$ that makes the total Hamiltonian $H=H_0+\Delta H$ non-integrable and provided that:
\begin{enumerate}
    \item the perturbation $\Delta H$ is small, and
    \item the fundamental frequencies of $H_0$ are incommensurate,
\end{enumerate}
then the motion remains confined on a N-torus, excepted for a negligible set of initial conditions.
\end{theorem}
In our case, $H_0$ would be the geodesic Hamiltonian while $\Delta H$ is the first order correction induced by the spin. There is no rigorous proof that KAM theorem applies to the system of MPD equations in Kerr, but typical features of KAM behaviour were numerically observed in Schwarzschild spacetime \cite{Zelenka:2019nyp}. Such features are also expected to show up in Kerr spacetime. The main consequence of KAM theorem is that the appearance of chaotic motion is confined to a negligible part of phase space, namely the one whose initial conditions lead some of the geodesic fundamental frequencies to be commensurate one to another (resonances). In the remaining part of phase space, the motion still exhibits a periodic-like behaviour, and one expects to still be able to compute the associated fundamental frequencies. Still in a perturbative spirit, these frequencies may be though as perturbations of the geodesic fundamental frequencies discussed in Chapter \ref{chap:hamilton_kerr}.

Actually, the spin-induced corrections to the fundamental frequencies of the motions were computed by W. Witzany at first order in $\mathcal S$ \cite{witzany2019spinperturbed} by studying the Hamilton-Jacobi formulation of the spinning test body problem in Kerr spacetime, to which the next chapter is devoted. In principle, the knowledge of these frequency shifts would allow to account for the finite-size nature of test bodies in evolution schemes describing self-forced EMRI's motion, such as the two-timescale expansion \cite{Hinderer:2008dm} (see also \cite{Pound:2021qin}).

\chapter{Hamilton-Jacobi equation}\label{chap:HJ}

\vspace{-12pt}

In this final chapter, we analyze the Hamilton-Jacobi equation for extended test bodies in Kerr spacetime and its relation with the constants of motion discussed in Part \ref{part:conserved_quantities} of the present thesis. As before, we restrict to the spin-induced quadrupole approximation and work under the TD spin supplementary condition.

Contrarily to the geodesic case, Hamilton-Jacobi equation for spinning test bodies is not separable in Kerr spacetime. Actually, W. Witzany showed that \cite{Witzany:2019nml}, at linear order in the spin magnitude, the terms that break the separability can be neglected everywhere except near the radial and polar turning points of the associated (zeroth order in the spin) geodesic motion. In this so-called ``swing region'', a perturbatively separable solution can be provided, and the separation constant arising in the analysis turns out to be precisely Rüdiger's extension of Carter constant. Moreover, an explicit solution to Hamilton-Jacobi equation valid both in the swing region and near the turning points can be built. 

This chapter is structured as follows: Section \ref{sec:HJ:HJ} will derive the explicit form of the Hamilton-Jacobi equation. Section \ref{sec:HJ:witzany} will discuss Witzany's solution at linear order in the spin magnitude. Section \ref{sec:HJ:separation} will discussed the interpretation of the separations constants in terms of the constants of motion of Part \ref{part:conserved_quantities}. Finally, Section \ref{sec:HJ:second} will provide some preliminary results concerning the extension of the solution at second order in the spin, namely a form of an explicit solution in the swing region valid regardless the value of the quadrupole coupling constant $\kappa$. Moreover, a conjecture concerning the expected form of the solution for the peculiar value of the coupling $\kappa=1$ will be proposed.

\section[Hamilton-Jacobi equation in the spin-induced quadrupole case]{Hamilton-Jacobi equation in the spin-induced qua-drupole case}\label{sec:HJ:HJ}

The method for building Hamilton-Jacobi equation has been reviewed in the geodesic case in Chapter \ref{chap:hamilton_kerr}. We shall first turn to symplectic coordinates. In order to stick to the original derivation \cite{Witzany:2019nml}, we will choose Witzany's coordinates $(x^\mu,P_\mu,\phi,A,\psi,B)$. In terms of these variables, the Hamiltonian \eqref{quadrupole_hamiltonian} reads\footnote{Again, we have dropped the underlines in the background tetrad indices.}
\begin{align}
    \begin{split}
    H&=\frac{\tilde\mu}{2}\bigg[ g^{\mu\nu} U_\mu U_\nu+U^\mu \omega_{\mu AB}S^{AB}+\frac{1}{4}\qty(\omega_{\mu AB}S^{AB})^2\\
    &\quad+\kappa\qty(\underline e\tdu{A}{\mu}s^{A B}R\tdud{ B}{\nu}{CD}s^{CD}-\qty(g^{\rho\sigma}U_\rho U_\sigma+2)R\tudu{\mu}{AB}{\nu}\theta^{AB})U_\mu U_\nu+1\bigg]\\
    &\quad+\mathcal O\qty(\mathcal S^3).
    \end{split}
\end{align}
In order to get rid of the constant $\tilde\mu$ in the equations, we have introduced the rescaled quantities
\begin{align}
    U_\mu\triangleq \frac{P_\mu}{\tilde\mu},\qquad s^{\mu\nu}\triangleq\frac{S^{\mu\nu}}{\tilde\mu},\qquad\theta^{\mu\nu}\triangleq\frac{\Theta^{\mu\nu}}{\tilde\mu^2}.
\end{align}

We now stand in a comfortable position to write down the Hamilton-Jacobi equation. Here, the canonical coordinates are $(x^\mu,\phi,\psi)$, and $(U_\mu,A,B)$ their respective conjugate momenta. The Hamilton-Jacobi equation is a first order PDE for Hamilton's characteristic function $W(x^\mu,\phi,\psi)$ (abbreviated as ``the action'' in the continuation of the chapter) which takes the explicit form
\begin{align}
    \begin{split}
    &1+g^{\mu\nu}W_{,\mu}W_{,\nu}+W^{,\mu}\omega_{\mu AB}s^{AB}+\frac{1}{4}\qty(\omega_{\mu AB}s^{AB})^2\\
    &+\kappa\qty[\underline e\tdu{A}{\mu}s^{AB}R\tdud{B}{\nu}{CD}s^{CD}-\qty(g^{\rho\sigma}W_{,\rho} W_{,\sigma}+2)R\tudu{\mu}{AB}{\nu}\theta^{AB}]W_{,\mu} W_{,\nu}=\mathcal O\qty(\mathcal S^3)\label{hamilton_jacobi}
    \end{split}
\end{align}
From here, we always understand all the occurrences of the spin tensor to be computed using Eq. \eqref{background_canonical} together with the substitutions $A\to W_{,\phi}$ and $B\to W_{,\psi}$. These substitutions originate from the fact that, while deriving the Hamilton-Jacobi equation, one shall substitute the conjugate momenta by the derivatives of the action with respect to their respective coordinates. See Chapter \ref{chap:hamilton_kerr} for more details.

We will now attempt to solve this equation perturbatively, order by order in the spin.

\section{Witzany's solution at first order}\label{sec:HJ:witzany}
We will now look at the first order in the spin Hamilton-Jacobi equation, which has been worked out recently by Witzany \cite{Witzany:2019nml}. At this order, Eq. \eqref{hamilton_jacobi} reduces to
\begin{align}
    1+g^{\mu\nu}W_{,\mu}W_{,\nu}+W^{,\mu}\omega_{\mu AB}s^{AB}=\mathcal O\qty(\mathcal S^2).
\end{align}
As we will see, the status of the separability of this equation is rather subtle at this order in the perturbative development.

\subsection{Geodesic-adapted tetrad}
The first step is to chose a specific form for the background tetrad $\underline e\tdu{\underline A}{\mu}$. We will use Witzany's ``geodesic adapted" tetrad \cite{Witzany:2019nml}, which has been independently studied by Marck in 1983 \cite{Marck1983SolutionTT}. This tetrad is constructed as follows: its timelike leg is taken tangent to a geodesic congruence of parameters $(E_c,L_c,K_c)$ (see Section \ref{sec:conserved_killing_geodesic}), \textit{i.e.} $\underline e_{0 \mu}=u_{c\mu}$ with
\begin{align}
    &u_{ct}=-E_c,\qquad &&u_{cr}=\pm_r\frac{\sqrt{R(r;E_c,L_c,K_c)}}{\Delta},\nonumber\\
    &u_{c x}=\pm_x\frac{\sqrt{\Theta_x(x;E_c,L_c,K_c)}}{1-x^2},\qquad &&u_{c\varphi}=L_c.
\end{align}
The three spatial legs of the tetrad are defined as
\begin{subequations}
\begin{align}
    \underline e_{1\mu}&=\frac{1}{N_{(1)}}\qty(K_{\mu\nu}+K_c g_{\mu\nu})u_c^\nu,\\
    \underline e_{2\mu}&=\frac{1}{N_{(2)}}\qty(K_{\mu\nu}-\frac{K_c^{(2)}}{K_c}g_{\mu\nu})Y\tud{\nu}{\lambda}u^\lambda_c,\\
    \underline e_{3\mu}&=\frac{1}{\sqrt{K_c}}Y_{\mu\nu}u^\nu_c,
\end{align}
\end{subequations}
with
\begin{subequations}
\begin{align}
    K_{\mu\nu}&\triangleq Y_{\mu\lambda}Y\tdu{\nu}{\lambda},\\
    K_c^{(2)}&\triangleq K_{\mu\nu}K\tud{\nu}{\rho}u_c^\mu u_c^\rho,\qquad && N_{(1)}^2\triangleq K_c^{(2)}+K_c^2,\\
    K_c^{(3)}&\triangleq K_{\mu\nu}K\tud{\nu}{\rho}K\tud{\rho}{\sigma}u_c^\mu u_c^\sigma,\qquad && N_{(2)}^2\triangleq K_c^{(3)}-\frac{\qty(K_c^{(2)})^2}{K_c}
\end{align}
\end{subequations}
and where $Y_{\mu\nu}$ is Kerr's Killing-Yano tensor.
A direct computation shows that this tetrad is indeed orthonormal. Moreover, it is defined up to the signs $\pm_{r,x}$ that shall be specified and depend on the specific geodesic that is followed. Finally, we notice that the tetrad is only well-defined away from the turning points of the geodesic congruence, where one has $\underline e_{0 r}=\underline e_{0\theta}=0$. This observation will be essential for the following, and will enforce us to solve the problem two times, both near and away from the turning points of the congruence.

A crucial property of this tetrad is the remarkable simple form of the projection of the connection 1-form onto its timelike leg. As for any orthogonal tetrad, one has of course
\begin{align}
    \omega_{0AB}\triangleq\underline e_{A\mu}\underline e\tdud{B}{\mu}{;\nu}\underline e\tdu{0}{\nu}=-\underline e_{B\mu}\underline e\tdud{A}{\mu}{;\nu}\underline e\tdu{0}{\nu}=-\omega_{0BA},
\end{align}
but the only non-vanishing component of this projection is given by
\begin{align}
    \omega_{012}=-\omega_{021}=\frac{\sqrt{K_c}}{\Sigma}\qty(\frac{P(E_c,L_c)}{r^2+K_c}+a\frac{L_c-aE_c(1-x^2)}{K_c-a^2x^2}).\label{non_vanishing_projection}
\end{align}
Recall that $P(E,L)\triangleq E(r^2+a^2)-aL$. As one can notice, except for the factor $1/\Sigma$ which is easy to get rid of, the RHS of this equation is separated with respect to the variables $r$ and $x$. This fact, together with the experience we have gained while discussing the geodesic problem in Chapter \ref{chap:hamilton_kerr}, will enable us to provide efficiently a solution to the first order problem away from the radial and polar turning points of the congruence.

\subsection{Solution in the swing region}
Let us first formalize the idea of ``being away from the turning points of the congruence''. Denoting respectively $y_t$ ($y=r,x$) the turning points of the geodesic congruence, one has 
\begin{align}
    R(r)\propto (r-r_t),\qquad \Theta_x(x)\propto (x-x_t)
\end{align}
near the turning points. As detailed in next section, this will lead to a $\propto(y-y_t)^{-1/2}$ divergence in the spin connection terms involved in the Hamilton-Jacobi equation.

In order to avoid this complication during a first time, we will seek a solution in the so-called \defining{swing region}, defined by
\begin{align}
    \abs{r-r_t}\gg\mathcal S,\qquad r\abs{x-x_t}\gg\mathcal S.\label{swing_region}
\end{align}
In this region, the only terms of the Hamilton-Jacobi equation that will scale as powers of the spin parameters are those that involve explicitly occurrences of the spin tensor.

We will look for a solution to the Hamilton-Jacobi equation that takes the form
\begin{align}
    W^{(1\sw)}=W^{(0)}+\mathcal O\qty(\mathcal S).
\end{align}
At first order, Eq. \eqref{hamilton_jacobi} becomes
simply
\begin{align}
    1+g^{\mu\nu}W^{(1\sw)}_{,\mu}W^{(1\sw)}_{,\nu}+W^{(0),\mu}\omega_{\mu AB}s^{AB}=\mathcal O\qty(\mathcal S^2).
\end{align}

To go further on, we set the constants of motion of the geodesic congruence to be close to the ones of the associated zeroth order geodesic motion,
\begin{align}
    E_c-E_0\sim L_c-L_0\sim K_c-K_0\sim\mathcal O\qty(\mathcal S)
\end{align}
and we choose the same signs $\pm_{r,x}$ for both the zeroth order action and the geodesic congruence. It allows to write
\begin{align}
    W^{(0)}_{,\mu}=u_{c\mu}+\mathcal O\qty(\mathcal S)=\underline e_{0\mu}+\mathcal O\qty(\mathcal S).
\end{align}
Using Eq. \eqref{non_vanishing_projection} allows to rewrite the spin connection term as
\begin{align}
    \begin{split}
    &W^{(0),\mu}\omega_{\mu AB}s^{AB}=2\frac{\sqrt{K_c}}{\Sigma}\qty(\frac{P(E_c,L_c)}{r^2+K_c}+a\frac{L_c-aE_c(1-x^2)}{K_c-a^2x^2})\\
    &\quad\times\qty(W^{(1sw)}_{,\phi}+W^{(1sw)}_{,\psi}+s)+\mathcal O\qty(\mathcal S).
    \end{split}
\end{align}

Moreover, since we assume that the TD condition holds and given the choice of tetrad and the discussion of Section \ref{sec:ssc:ortho} about SSCs in adapted tetrads, it is rather easy to convince ourselves that 
\begin{align}
    S^{0A}=\mathcal O\qty(\mathcal S^2).
\end{align}
Comparing this results with Eq. \eqref{background_canonical}, one notices that the explicit form of the spin tensor is independent of $\psi$, which is consequently a cyclic coordinate. In particular, one has $W^{(1\sw)}_{,\psi}=\mathcal O\qty(\mathcal S^2)$, and the action is therefore independent of $\psi$ at this order of the perturbative expansion.

Mimicking Carter's procedure for the geodesic case, we therefore make the following Ansatz for the swing region action:
\begin{align}
    W^{(1\sw)}=-E_\so t+L_\so\varphi+\qty(s_{\parallel}-s) \phi+w_{1r}(r)+w_{1x}(x).
\end{align}
The separation constants are chosen to be labelled with subscripts `so' (spin-orbit). Plugging this Anstaz into the Hamilton-Jacobi equation yields
\begin{align}
    \begin{split}
    &a\Sigma\qty(1+g^{\mu\nu}W^{(1\sw)}_{,\mu}W^{(1\sw)}_{,\nu})\\
    &+2s_{\parallel}\sqrt{K_c}\qty(\frac{P(E_c,L_c)}{r^2+K_c}+a\frac{L_c-aE_c(1-x^2)}{K_c-a^2x^2})=\mathcal O\qty(\mathcal S^2).
    \end{split}
\end{align}

A direct comparison of this equation with the one arising in the zeroth order computation allows us to use the very same procedure for the separation, with one additional term arising from the connection appearing in each separated part:
\begin{subequations}
\begin{align}
    \qty(1-x^2)\qty(w'_{1x})^2&=Q_\so-x^2\qty[a^2\qty(1-E_\so^2)+\frac{L_\so^2}{1-x^2}]\nonumber\\
    &\quad-2as_{\parallel}\sqrt{K_c}\frac{L_c-aE_c(1-x^2)}{K_c-a^2x^2},\label{eq2_rud}\\
    \Delta^2\qty(w'_{1r})^2&=-\qty(K_\so+r^2)\Delta+P^2(E_\so,L_\so)\nonumber\\
    &\quad-2s_{\parallel}\Delta\sqrt{K_c}\frac{P(E_c,L_c)}{K_c+r^2}.
\end{align}
\end{subequations}
This provides the full first order solution in the swing region. For now, $E_\so$, $L_\so$, $K_\so$ and $s_{\parallel}$ are just some separation constants, and $Q_\so$ is defined in the same way than $Q_0$. Their physical meaning will be investigated in the last section of the present chapter.

\subsection{Solution near the turning points}

Near the turning points $y_t$ of the geodesic congruence, the geodesic-adapted tetrad becomes ill-defined since some of its legs vanish. Nevertheless, this problem can be overcome in the following way: first, notice that the swing region solution remains valid even if we shift the constants $E_c$, $L_c$ and $K_c$ by an $\mathcal O\qty(\mathcal S)$ quantity. Therefore, one can always choose these constants to avoid the turning points by an $\mathcal O\qty(\mathcal S)$ distance. Nevertheless, the price to pay is that some other quantities in the Hamilton-Jacobi equation will scale with the spin parameter near the turning points. The various terms that acquire such a scaling are listed in Table \ref{tab:spin_scaling}.

\begin{table}[h!]
   \begin{center}
     \begin{tabular}{c|cc}
     & \textsc{turning point region} & \textsc{swing region}\\\hline
     \rule{0pt}{15pt}Connection terms $\omega_{\mu AB}$ & $\sim 1/\sqrt{y-y_t}\sim 1/\sqrt{\mathcal S}$ & $\sim 1$\\
     \rule{0pt}{15pt}Swing region action $W^{(1\sw)}$ & $\sim\sqrt{\mathcal S}$ & $\sim\mathcal S$\\
     \rule{0pt}{15pt}Correction terms $\delta_y W^{(1\tp)}$ & $\sim \mathcal S^2/\sqrt{y-y_t}\sim \mathcal S^{3/2}$ & $\sim\mathcal S^2$
    \end{tabular}
\end{center}
    \caption{Scalings in the spin both near and far way from the turning points of the congruence.}
    \label{tab:spin_scaling}
\end{table}

The main consequence of this new scaling behaviour is that Hamilton-Jacobi equation will now include some new terms. These corrections will break its separability. Nevertheless, a solution can still be provided, as we will demonstrate now. Explicitly, the equation to be solved becomes \cite{Witzany:2019nml}
\begin{align}
    \begin{split}
    &1+g^{\mu\nu}W^{(1)}_{,\mu}W^{(1)}_{,\nu}+\omega_{0AB}s^{AB}\\
    &+\sum_{y=r,x}\qty[\qty(W^{(1)}_{,y}-\underline e_{0 y})\omega_{y AB} s^{AB}+\frac{1}{4}\qty(\omega_{y AB}s^{AB})^2]g^{yy}=\mathcal O \qty(\mathcal S^2).
    \end{split}
\end{align}
Since the first line of this equation is already solved by the swing region solution, it is reasonable to make the following Ansatz near the turning points:
\begin{itemize}
    \item The action behaves as in the swing region with respect to $\phi$ and $\psi$. Therefore, one can replace the occurrences of the spin tensor $s^{AB}$ in the equations with
    \begin{align}
        \tilde s^{AB}\triangleq s^{AB}\qty(A= s_{1\parallel}-s,B= 0).
    \end{align}
    The only non-vanishing components up to $\mathcal O\qty(\mathcal S^2)$ are
    \begin{subequations}
    \begin{align}
        \tilde s^{12}&=-\tilde s^{21}=s_{1\parallel},\\
        \tilde s^{13}&=-\tilde s^{31}=\sqrt{s^2-s_{1\parallel}^2}\sin\phi,\\
        \tilde s^{23}&=-\tilde s^{32}=\sqrt{s^2-s^2_{1\parallel}}\cos\phi.
    \end{align}
    \end{subequations}
    \item The action valid both near and far away from the turning points is assumed to take the form
    \begin{align}
        W^{(1)}=W^{(1\sw)}+\sum_{y=r,x}\delta_y W^{(1\tp)}.
    \end{align}
    Here, the corrections $\delta_y W^{(1\tp)}$ are assumed to scale as $\mathcal S^2/\sqrt{y-y_t}$. They are therefore relevant in the turning points region, whereas they are only contained in higher-order terms in the swing regions (see Table \ref{tab:spin_scaling}).
\end{itemize}
Plugging this Ansatz into the Hamilton-Jacobi equation allows to  reduce our problem to the following sets of two equations
\begin{align}
    \begin{split}
    &\qty(\delta_y W^{(1\tp)}_{,y})^2+2 S'_{1y}\delta_y W^{(1\tp)}_{,y}+\omega_{yAB} s^{AB}\qty(S'_{1y}+\delta_y W^{(1\tp)}_{,y}-\underline e_{0 y})\\
    &+\frac{1}{4}\qty(\omega_{yAB}\tilde s^{AB})^2=0,\qquad\text{ for } y=r,x.
    \end{split}
\end{align}
These equations are easily solved by quadrature for $\delta_y W^{(1t)}_{,y}$. When the dust settles, we obtain
\begin{align}
    \delta_y W^{(1\tp)}_{,y}=-w'_{1y}-\frac{1}{2}\omega_{yAB}\tilde s^{AB}\pm_y\sqrt{\qty(w_{1y}')^2+\underline e_{0 y}\omega_{yAB} \tilde s^{AB}}.
\end{align}
This solution can be integrated as
\begin{align}
    \delta_y W^{(1\tp)}=-w_{1y}+\int\qty(-\frac{1}{2}\omega_{yAB}\tilde s^{AB}\pm_y\sqrt{\qty(w_{1y}')^2+\underline e_{0 y} \omega_{yAB}\tilde s^{AB}})\dd y+C_y,
\end{align}
where $C_r=C_r(\phi,x)$ and $C_x=C_x(\phi,r)$ are integration constant that will be set to zero in the continuation of this work. 

At the end of the day, we obtain the following form of the action at first order in the spin, which is valid both near and far away of the turning points of the congruence:
\boxedeqn{
    \begin{split}
    W^{(1)}&=-E_\so t+L_\so\varphi+\qty(s_{\parallel}-s)\phi\\
    &\quad+ \sum_{y=r,x}\int\qty(-\frac{1}{2}\omega_{yAB} s^{AB}\pm_y\sqrt{\qty(w_{1y}')^2+\underline e_{0 y}\omega_{yAB} \tilde s^{AB}})\dd y
    \end{split}
}{Hamilton-Jacobi equation solution at first order in the spin}

\section{The separation constants}\label{sec:HJ:separation}
We will now discuss the physical interpretation of the separation spin-orbit separation constants $E_\so, L_\so, K_\so$ and $s_\parallel$ that showed up in our discussion. As we will show, these constants are directly related to the constants of motion arising from the symmetries of the background geometry that were studied in Part \ref{part:conserved_quantities} of the present text.

\subsection{Spin-orbit energy and angular momentum}

First notice that, since
\begin{align}
    W_{,\mu}=U_{\mu}=\hat p_\mu-\frac{1}{2}\omega_{\mu AB}\tilde s^{AB},
\end{align}
one has the cornerstone identity
\begin{align}
    \hat p_\mu=W_{,\mu}+\frac{1}{2}\omega_{\mu AB}\tilde s^{AB}.\label{eq_corner}
\end{align}
Using Eq. \eqref{E0_L0} to express the quantities $E_0$ and $L_0$, we directly get the following relation for the spin-orbit energy and angular momentum:
\begin{subequations}\label{so_E_L}
\begin{align}
    E_\so&=E_0+\frac{1}{2}\omega_{t AB}\tilde s^{AB},\\
    L_\so&=L_0-\frac{1}{2}\omega_{\varphi AB}\tilde s^{AB}.
\end{align}
\end{subequations}
A simple computation can show that these quantities are precisely equal to the conserved quantities for MPD equations related to the existence of Killing vector fields. Recall that, for any Killing vector $\xi^\mu$ of the background, the quantity
\begin{align}
    \mathcal C_\xi\triangleq\xi_\mu\hat p^\mu+\frac{1}{2}\nabla_\mu\xi_\nu s^{\mu\nu}
\end{align}
is conserved at any order of the multipole expansion. Let us now particularize this result to Kerr's timelike and axial Killing vectors $\xi=\partial_t,\partial_\varphi$. In this case, one has
\begin{align}
    \nabla_\mu \xi_\nu&=g_{\nu\rho}\qty(\partial_\mu\xi^\rho+\Gamma^\rho_{\mu\sigma}\xi^\sigma)
    =\Gamma_{\nu\mu\sigma}\xi^\sigma,
\end{align}
where the last equality follows from the fact that the components $\qty(\partial_t)^\mu=\qty(1,0,0,0)$ and $\qty(\partial_\varphi)^\mu=\qty(0,0,0,1)$ are constant in Boyer-Lindquist coordinates. Making use of this result and of identity (7.142) of \cite{Nakahara:206619} to write the Christoffel symbols in terms of the connection 1-forms, we get the final expression
\begin{align}
    \mathcal C_\xi=\xi_\mu\hat p^\mu-\frac{1}{2}\omega_{\mu AB}\xi^\mu \tilde s^{AB}.
\end{align}
Using this result to compare $E=-C_{\partial_t}$ and $L=C_{\partial_\varphi}$ with the spin-orbit separation constants \eqref{so_E_L}, we directly find that
\begin{align}
    E=E_\so,\qquad L=L_\so.
\end{align}

\subsection{Spin projection $s_{\parallel}$}
We now turn to the spin separation constant $s_\parallel$. One the one side, one can write
\begin{align}
    s_\parallel&=\tilde s^{12}
    =\frac{1}{2}\epsilon^{30AB}\tilde s_{AB}
    =\frac{1}{2}\underline e_{3\mu}\underline e_{0\nu}\epsilon^{\mu\nu\rho\sigma}\tilde s_{\rho\sigma}
    =\frac{1}{2\sqrt{K_c}}Y_{\mu\lambda} u_c^\lambda u_{c\nu}\epsilon^{\mu\nu\rho\sigma}\tilde s_{\rho\sigma}.
\end{align}
On the other side,
\begin{align}
    \begin{split}
    Y_{\mu\lambda} u_c^\lambda u_{c\nu}\epsilon^{\mu\nu\rho\sigma}\tilde s_{\rho\sigma}&=-Y^{**}_{\mu\lambda} u_c^\lambda u_{c\nu}\epsilon^{\mu\nu\rho\sigma}\tilde s_{\rho\sigma}\\
    &=-\frac{1}{2}\epsilon_{\mu\lambda\alpha\beta}\epsilon^{\mu\nu\gamma\delta}Y^{*\alpha\beta}u_c^\lambda u_{c\nu}\tilde s_{\gamma\delta}\\
    &=-Y^*_{\mu\nu}\tilde s^{\mu\nu}\\
    &=-\mathcal Q_Y,
    \end{split}
\end{align}
where $\mathcal Q_Y$ is the Rüdiger's linear invariant discussed in Part \ref{part:conserved_quantities}. Comparing the two previous equations and still choosing $K_c=K_0+\mathcal O \qty(\mathcal S^2)$, we get
\begin{align}
    \boxed{
    s_\parallel=-\frac{1}{2}\frac{\mathcal Q_Y}{\sqrt{K_0}}+\mathcal O \qty(\mathcal S^3).}
\end{align}

Two remarks can be made before closing this discussion. First, let us define the \textit{specific angular momentum vector}
\begin{align}
    \ell^\mu\triangleq Y^{\mu\nu}u_{c\nu}.
\end{align}
The terminology makes sense, since Carter's constant obeys $K_0=\ell_\alpha \ell^\alpha$ and can be roughly interpreted as some sort of ``orbital angular momentum squared") (see \textit{e.g.} \cite{Witzany:2019nml}). Moreover, since
\begin{align}
    \mathcal Q_Y=-2\ell_\alpha s^\alpha.
\end{align}
We get the enlightening expression
\begin{align}
    s_\parallel=\frac{\ell_\alpha \tilde s^\alpha}{\sqrt{\ell_\alpha \ell^\alpha}}+\mathcal O \qty(\mathcal S^2).
\end{align}
The separation constant $s_\parallel$ can thus be understood as the projection of the spin vector $\tilde s^\mu$ onto the direction of the specific orbital angular momentum vector $\ell^\mu$.

Second, a simple computation allows to show that the background tetrad components of the normalized spin vector are
\begin{align}
    \tilde s^{\underline A}=\qty(0,\,\sqrt{s^2-s_\parallel^2}\cos\phi,\,-\sqrt{s^2-s^2_\parallel}\sin\phi,\,s_\parallel)+\mathcal O\qty(\mathcal S^2)
\end{align}
We therefore see that the cases $s_\parallel=\pm s$ correspond to the spin being totally (anti-)aligned with the angular orbital momentum, since $\underline e\tdu{3}{\mu}\propto \ell_\mu$. Moreover, the canonical coordinate $\phi$ can be viewed as a precession angle of the spin around this axis.

\subsection{Generalized Carter constant $K_\so$}
The evaluation of the separation constant $K_\so$ is the most tedious, and will not be discussed in full details here, since the computations are not particularly enlightening. The basic idea consists, on the one side, to make use of Eqs. \eqref{eq2_rud} and \eqref{eq_corner} to write a closed form expression for $K_\so$ in terms of $r$, $x$, $p_r$, $p_x$ and of the constants of motion. On the other side, Rüdiger's quadratic invariant can be written
\begin{align}
    \mathcal Q_R=K_0+2 s^\alpha\partial_\alpha \mathcal Z+E_0\mathcal Q_Y.
\end{align}
Computing explicitly the second term of this expression allow to show that the two constants do coincide \cite{Witzany:2019nml},
\begin{align}
    K_\so=\mathcal Q_R.
\end{align}
The separation constant of the Hamilton-Jacobi equation at first order in the spin expansion is therefore precisely Rüdiger's quadratic invariant.

To summarize the discussion, all the constants appearing in the swing region separation of Hamilton-Jacobi equation at linear order in the spin magnitude $\mathcal S$ are in direct correspondence with the conserved quantities found by explicitly solving the conservation constraint equations of Part \ref{part:conserved_quantities}.

\section{Solution at second order in the swing region}\label{sec:HJ:second}
Given the discussion of last section, it is natural to wonder what happens at second order in the spin. This section will provide partial results at this order, and we will end by conjecturing the expected relation between Hamilton-Jacobi formulation and conserved quantities of Part \ref{part:conserved_quantities}.

At second order in $\mathcal S$, Hamilton-Jacobi equation is given by the following mess
\begin{align}
    \begin{split}
    &1+g^{\mu\nu}W_{,\mu}^{(2\sw)}W_{,\nu}^{(2\sw)}+W^{(1\sw),\mu}\omega_{\mu AB}s^{AB}+\frac{1}{4}\qty(\omega_{\mu AB}s^{AB})^2\\
    &+\kappa\qty[\underline e\tdu{A}{\mu}s^{AB}R\tdud{B}{\nu}{CD}s^{CD}-\qty(g^{\rho\sigma}W^{(0)}_{,\rho} W^{(0)}_{,\sigma}+2)R\tudu{\mu}{AB}{\nu}\theta^{AB}]W^{(0)}_{,\mu} W^{(0)}_{,\nu}\\
    &=\mathcal O\qty(\mathcal S^3).\label{os2_HJ}
    \end{split}
\end{align}
The second line contains only terms involving occurrences of the Riemann tensor. Using the TD SSC $s^{0A}=\mathcal O\qty(\mathcal S^2)$, it can be simplified as follows:
\begin{align}
    \begin{split}
    &\qty[\underline e\tdu{A}{\mu}s^{AB}R\tdud{B}{\nu}{CD}s^{CD}-\qty(g^{\rho\sigma}W^{(0)}_{,\rho} W^{(0)}_{,\sigma}+2)R\tudu{\mu}{AB}{\nu}\theta^{AB}]W^{(0)}_{,\mu} W^{(0)}_{,\nu}\\
    &=\qty[\underline e\tdu{A}{\mu}s^{AB}R\tdud{B}{\nu}{CD}s^{CD}-\qty(g^{\rho\sigma}\underline e_{0\rho}\underline e_{0\sigma}+2)R\tudu{\mu}{AB}{\nu}\theta^{AB}]\underline e_{0\mu}\underline e_{0\nu}+\mathcal O\qty(\mathcal S^3)\\
    &=R_{0A0B}\theta^{AB}+\mathcal O\qty(\mathcal S^3).
    \end{split}
\end{align}

In order to solve this equation in the swing region, we mimic the first order turning region procedure. We make the following Ansatz
\begin{align}
    W^{(2\sw)}=W^{(1\sw)}+\sum_{y=r,x}\delta_yW^{(2\sw)}
\end{align}
and still consider the replacement of the spin tensor by its tilded version. Moreover, we choose the congruence constants as
\begin{align}
    E_c-E_0\sim L_c-L_0\sim K_c-K_0\sim\mathcal O\qty(\mathcal S^2)\quad\Leftrightarrow\quad
    W^{(0)}_{,\mu}=\underline e_{0\mu}+\mathcal O\qty(\mathcal S^2).
\end{align}
Hamilton-Jacobi equation then becomes
\begin{align}
\begin{split}
    &1+g^{\mu\nu}W^{(1\sw)}_{,\mu}W^{(1\sw)}_{,\nu}+W^{(0),\mu}\omega_{\mu AB} \tilde s^{AB}\\
    &+\sum_{y=r,x}\qty[\qty(\delta_y W^{(2\sw)}_{,y})^2+2w'_{1y}\delta_y W^{(2\sw)}_{,y}+F^{(2\sw)}_y]g^{yy}=\mathcal O\qty(\mathcal S^3).\label{HJ_swing_2}
\end{split}
\end{align}
Here, $F^{(2\sw)}_y=F^{(2\sw)}_y(\phi,r,x)$ denote two functions that satisfy to the algebraic equation
\begin{align}
    \begin{split}
    &\sum_{y=rx}g^{yy}F^{(2\sw)}_y=\frac{1}{4}\qty(\omega_{\mu AB}\tilde s^{AB})^2\\
    &+\qty(W^{(1\sw),\mu}
    -W^{(0),\mu})\omega_{\mu AB}\tilde s^{A B}+\kappa R_{0A0B}\tilde \theta^{AB}.\label{source_function}
    \end{split}
\end{align}
This equation is always solvable, and its simplest solution reads
\begin{align}
    F^{(2\sw)}_y=\frac{g_{yy}}{2}\,\qty(\text{RHS}),
\end{align}
where `RHS' denotes the RHS of Eq. \eqref{source_function}. Nevertheless, this splitting is in principle tunable to improve the properties of the solution. Eq. \eqref{HJ_swing_2} is satisfied provided that, for both $y=r$ and $x$:
\begin{align}
    \qty(\delta_y W^{(2\sw)}_{,y})^2+2w'_{1y}\delta_y W^{(2\sw)}_{,y}+F_y^{(2\sw)}=0.
\end{align}
The solution of this second order algebraic equation is simply
\begin{align}
    \delta_y W^{(2\sw)}_{,y}=-w'_{1y}\pm_y\sqrt{\qty(w'_{1y})^2-F_y^{(2\sw)}(\phi,r,x)},
\end{align}
which integrates as
\begin{align}
    \delta_y W^{(2\sw)}=-w_{1y}\pm_y\int\sqrt{\qty(w'_{1y})^2-F_y^{(2\sw)}(\phi,r,x)}\dd y.
\end{align}
Here, we have set the possible integration constants to zero. The full second order swing region solution is
\begin{align}
    \begin{split}
    W^{(2\sw)}&=-E_\so t+L_\so\varphi+\qty(s_{\parallel}-s)\phi\\
    &\quad+\sum_{y=r,x}\int\pm_y\sqrt{\qty(w'_{1y})^2-F_y^{(2\sw)}(\phi,r,x)}\dd y.
    \end{split}
\end{align}
This solution is no more separated for generic values of $\kappa$. However, for the specific case $\kappa=1$, the existence of a deformation of Rüdiger's constant $\mathcal Q^{(2)}_\text{BH}$ and the fact that -- both at zeroth and first order -- the constants $K_0$ and $K_\so$ play the role of separation constants of the associated Hamilton-Jacobi equations lead us to propose the following conjecture:
\begin{conjecture}\label{conject}
In Kerr spacetime, at quadratic order in the spin parameter $\mathcal S$, Hamilton-Jacobi equation Eq. \eqref{os2_HJ} endowed with TD spin supplementary condition and spin-induced quadrupole moment is separable in the swing region of phase space. The separated solution has the same form as at linear order, but the separation constant is now the quasi invariant $\mathcal Q^{(2)}_\text{BH}$ given in Eq. \eqref{quadratic_rudiger}. All the others constants appearing in the solution ($E_\so$, $L_\so$ and $s_\parallel$) are left unchanged.
\end{conjecture}

On the technical level, the proof of this conjecture shall consist in the possibility of finding separated forms for the source functions of Eq. \eqref{source_function}, namely
\begin{align}
    F^{(2\sw)}_y(\phi,r,x)=F^{(2\sw)}_y(y).
\end{align}
Identities such as (3.4.20) of \cite{Wald:1984rg} allow in principle to express the Riemann tensor only in terms of quadratic combinations of the connection 1-form and of its first covariant derivative. We believe that this fact will allow to dramatically simplify the RHS of Eq. \eqref{source_function}, the Riemann term then having the right $\kappa=1$ coefficient for simplifying with the quadratic connection term, yielding the separability property announced. This conjecture remains to be investigated in full details.

Finally, notice that the knowledge of the solution to Hamilton-Jacobi equation allows to compute the spin-induced shifts in both the fundamental frequencies and the orbital turning points of the motion. This has been discussed by W. Witzany in \cite{Witzany:2019nml}, up to first order in $\mathcal S$. The computation of the corresponding shifts at $\os{2}$ would first require a more complete understanding of the solution of Hamilton-Jacobi equation at quadratic order in $\mathcal S$. This remains to be investigated in coming works.

\pagestyle{conclusion}

\addcontentsline{toc}{part}{\textsc{Conclusion and Outlook}}
\part*{\textsc{Conclusion and Outlook}}

\lettrine{O}{ur} journey aiming at understanding the motion of extended test bodies in Kerr spacetime is about to end. Before closing definitively the discussion, let us briefly summarize the main outcomes of this thesis and the principal directions that remain to be explored.

\subsubsection{Summary of the main results}

The first part of this text was devoted to a review of the main features of geodesic motion in Kerr spacetime, including a discussion of Hamiltonian formulation and of classification of timelike geodesics. The geodesic equations are easily derived by solving the associated Hamilton-Jacobi equation, which turns out to be separable. They form a completely integrable system of equations, and the motion is characterized by four constants: the mass of the body, its energy and angular momentum and the Carter constant, the latter playing the role of the separation constant of the Hamilton-Jacobi problem. The integrable nature of motion allows to both (i) provide an action-angle formulation for radially bounded geodesic motion, thus making explicit its tri-periodic character and (ii) solve explicitly the equations of motion by quadratures. Moreover, radial and polar parts of the motion decouple one from the other, and can be solved independently. This allowed us to provide two original results: the classification of polar timelike geodesic motion in generic Kerr spacetime, and the one of radial motion for the near-horizon geodesics of high spin black holes.

The next part was aimed at pedagogically reviewing standard results of the literature. We derived the MPD equations, which describe the motion of a finite size test body in curved spacetime, by modelling it as a worldline endowed with a collection of multipole moments accounting for its internal structure. This multipole expansion is consistent with an expansion of the motion in integer powers of the spin magnitude above a spinless, geodesic motion. Moreover, for bodies spinning at astrophysically realistic rates, it can be consistently truncated at any desired order. In this thesis, it has been chosen to work up to quadrupole order included, the presence of a quadrupole moment originating only from the proper rotation of the body. The explicit form of such a spin-induced quadrupole moment was explicitly derived. Moreover, it was shown that obtaining a closed system of equations required to enforced additional algebraic constraints, known as spin supplementary conditions. This amounts to specify the worldline which is chosen for defining the body's multipole structure. Several spin supplementary conditions have been discussed, and the so-called Tulczyjew-Dixon condition (TD SSC) was argued to be the best suited for the continuation of the work.

Still tackling the problem in a perturbative spirit, the question of finding conserved quantities for the motion of extended test bodies was addressed both at linear and quadratic order in the (conserved) spin magnitude in the third part of the text. Energy and angular momentum can be deformed to obtain quantities which are conserved at any order of the multipole expansion. Mass is conserved at linear order, and a mass-like conserved quantity can still be defined at quadratic order. Extending the Carter constant to spinning bodies is more subtle. Its generalization at first order is due to Rüdiger, who also unrevealed the existence of another quantity, homogeneously linear in the spin. In Kerr spacetime, we demonstrated the uniqueness of this construction at first order in the spin. At quadratic order and under the spin-induced quadrupole approximation, it was shown that one can still find a generalized Carter-type constant and that the homogeneously linear invariant was still conserved provided that the quadrupole coupling has the value of the test body being itself a Kerr black hole. For other values of the coupling, we argued that it was impossible to find polynomial deformations of the known conserved quantities in Kerr spacetime.

The last part of the thesis was aimed as depicting some aspects of the Hamiltonian description of test body motion in curved spacetime. We first reviewed a phase space a different coordinates systems (symplectic or not) adapted to a covariant Hamiltonian description of the problem. Expanding on the results known at pole-dipole level in the literature, we then presented the construction of a Hamiltonian reproducing the MPD equations endowed with a spin-induced quadru-pole moment under the TD spin supplementary condition, which is valid up to quadratic order in the spin magnitude. Two applications of this formalism to Kerr spacetime were finally discussed. First, we argued that the MPD equations were not anymore integrable in Kerr spacetime already at linear order, contrarily to the geodesic ones. Second, we reproduced Witzany's analysis, which showed that all the conserved quantities previously found were in direct one-to-one correspondence with the constants appearing in the (almost) separable solution to the Hamilton-Jacobi equation at first order in the spin magnitude.

\subsubsection{Outlook}

Various extensions of this work would be interesting to explore. Some of the main ones will be briefly detailed below.

\begin{itemize}
    \item \textbf{Conserved quantities in Kerr spacetime}

    At linear order in the spin magnitude, the occurrence of Einstein's tensor in the central identity \eqref{generalized_central} that we derived suggests that Rüdiger's quadratic invariant will admit a generalization to the Kerr-Newmann spacetime which also admits a Killing-Yano tensor, once the Einstein tensor is replaced with the electromagnetic stress-energy tensor. This remains to be investigated.

    Moreover, even if we used the Tulczyjew-Dixon spin supplementary condition to derive the conserved quantities, the reasoning could be in principle translated to others choices of SSCs. At linear order in $\mathcal S$, the transposition to MP and KS conditions is trivial, since the momentum and the four-velocity are related through Eq. \eqref{pv} for any of these three conditions. At quadratic order in $\mathcal S$, the translation of our results from TD to other SSCs would in principle be facilitated by the formalism describing transitions between SSCs derived in \cite{Steinhoff:2014kwa,Vines:2016unv}.

    It would also be enlightening to explore the relationship between the new constant $\mathcal Q^{(2)}_\text{BH}$ for test (that is, non-backreacting) black holes in Kerr spacetime and those for arbitrary-mass-ratio binary black holes at second-post-Newtonian order (including the spin-induced quadrupole effects) found by Tanay \emph{et al.}\ in Ref.~\cite{Tanay:2020gfb}, ensuring the integrability of the system at that order.

    Finally, the covariant building blocks formalism introduced in the derivation may be applied for solving others problems in Kerr spacetime. It is very powerful, since it turns all differential identities into algebraic ones, and consequently allows for a purely enumerative approach for solving the original problems. Moreover, if the structure of the 4-dimensional Kerr covariant building blocks can be generalized to higher dimensions and more generic spacetimes, it could bring new insights on hidden symmetries of such spacetimes \cite{Frolov:2017kze}.

    \item \textbf{Higher multipoles}
    
    One could wonder whether deformations of the conserved quantities that have studied in this thesis still exist at higher orders in the multipole expansion for black hole-type couplings. The structure of the spin-induced octupole moment has already been well established \cite{Marsat_2015}, but the very form of higher order spin-induced moments as well as the corresponding equations of motion remains unclear \cite{Steinhoff_2010,Marsat_2015}. The assumption of the existence of deformations of the present quadrupole invariants could potentially serves as a guide for investigating test black holes motion at higher orders.

    \item \textbf{Hamilton and Hamilton-Jacobi formulations}

    The main direction concerning the Hamiltonian sector of this thesis would be to prove (or invalidate) Conjecture \ref{conject}, which would enable to understand the link between the existence of the quasi-conserved quantities and the swing-region separability of the associated Hamilton-Jacobi equation at second order in the spin magnitude for black hole coupling, thus pushing the analysis of Witzany \cite{Witzany:2019nml} to the next order in the multipole expansion. This would also enable to compute the corresponding shifts in the fundamental frequencies of the action-angle variables description of the finite size particle, which are of direct relevance for modeling EMRIs involving spinning secondaries.

    The (marginal) appearance of chaos in the motion at linear order in $\mathcal S$ was suggested in \cite{Kunst:2015tla}, but ruled out by the conclusions of \cite{Witzany:2019nml}. Notice that a direct comparison between analytical results and the outcomes of numerical investigations \cite{Ruangsri:2015cvg,Kunst:2015tla,Zelenka:2019nyp} still deserves caution due to the actual hypotheses under which the analysis is carried out. Moreover, Witzany's solution to Hamilton-Jacobi equation \cite{Witzany:2019nml} did not allow for a separation of the orbital equations of motion at linear order in $\mathcal S$. These results are consistent with our findings that Liouville integrability does not hold due to a single non-vanishing Poisson bracket between the two R\"udiger quasi-invariants, even though the explicit relationship between these statements remains to be deepened. A more careful investigation using a symplectic Hamiltonian formulation (in the spirit of the one of P. Ramond for Schwarzschild spacetime \cite{Ramond:2022vhj}) of the motion in Kerr at linear order in $\mathcal S$ would enable to better understand the breaking of the integrability, as well as provide more definitive answers regarding the relationships between non-integrability, structure of the Hamilton-Jacobi equation solution, and the (non-)appearance of chaos.
\end{itemize}

\vspace{\stretch{1}}
\noindent
Here ends this thesis. We have been pleased to share a part of the enthusiasm we have experienced while digging into the entangled but fascinating world of Kerr black holes, still impressed by the possibility of being able to carry out so far exact analytic computations. We sincerely hope that the attentive reader may have found, well hidden behind all the formal statements, some insightful shards of beauty.

\vspace{\stretch{0.25}}

\begin{center}
    *\\
    *\hspace{1cm}*
\end{center}

\vspace{\stretch{1}}

\pagestyle{plain}

\appendix

\addcontentsline{toc}{part}{\textsc{Appendices}}
\part*{\textsc{Appendices}}

\chapter[Elliptic integrals and Jacobi functions]{Elliptic integrals\\and Jacobi functions}\label{app:ellipticFunctions}

\vspace{\stretch{1}}

In this appendix, we set our conventions for the elliptic integrals and Jacobi functions used in the main text, following \cite{olver10}.

The \defining{incomplete elliptic integrals} of the first, second, and third kind are defined as
\begin{subequations}
\begin{align}
F(\varphi , m) &\triangleq  \int_0^\varphi \frac{\text{d} \theta}{\sqrt{1- m \sin^2 \theta}} = \int_0^{\sin \varphi} \frac{\text{d} t}{\sqrt{(1-t^2)(1-m t^2)}}, \\
E(\varphi , m) &\triangleq  \int_0^\varphi d\theta \sqrt{1- m \sin^2\theta}= \int_0^{\sin \varphi} \text{d} t \sqrt{\frac{1-m t^2}{1-t^2}}, \\
\Pi(n, \varphi ,m)& \triangleq \int_0^\varphi \frac{1}{1- n \sin^2 \theta}\frac{\text{d} \theta}{\sqrt{1- m \sin^2 \theta}}\nonumber\\
&\quad= \int_0^{\sin\varphi} \frac{1}{1- n t^2} \frac{\text{d} t}{\sqrt{(1-m t^2)(1-t^2)}},
\end{align}
\end{subequations}
respectively. We also define $E'(\varphi, m ) = \partial_m E(\varphi,m) = \frac{1}{2m} [E(\varphi,m) - F(\varphi , m)] $. 

The \defining{complete elliptic integrals} of the first, second, and third kind are defined as
\begin{subequations}
    \begin{align}
K(m) &\triangleq  F(\frac{\pi}{2},m), \\
E(m) &\triangleq  E(\frac{\pi}{2},m), \\
\Pi(n,m) &\triangleq  \Pi(n,\frac{\pi}{2},m),
\end{align}
\end{subequations}
respectively, and $E'(m)=\partial_m E(m)$. 

\defining{Jacobi functions} are defined as the inverse of the incomplete elliptic integrals of the first kind. More precisely, one can invert $u = F(\varphi,m)$ into $\varphi =\text{am}(u,m)$ in the interval $-K(m) \leq u \leq K(m)$. The elliptic sinus, elliptic cosinus, and delta amplitude are defined as 
\begin{subequations}
    \begin{align}
\text{sn}(u,m) &\triangleq \sin\text{am}(u,m), \qquad \text{cn}(u,m) \triangleq \cos\text{am}(u,m),\\
 \text{dn}(u,m) &\triangleq  \sqrt{1-m (\sin\text{am}(u,m))^2},
\end{align}
\end{subequations}
respectively. They have the periodicity ($k,l \in \mathbb Z$)
\begin{subequations}
    \begin{align}
\text{sn}(u + 2 k K(m)+2ilK(1-m) ,m)  &= (-1)^k \text{sn}(u,m) ,\label{per}\\
\text{dn}(u + 2 k K(m)+2ilK(1-m) ,m)  &= (-1)^l \text{sn}(u,m)\label{eqn:perDn}
\end{align}
\end{subequations}
and obey the properties
\begin{subequations}
    \begin{align}
\text{cn}^2(u,m)+\text{sn}^2(u,m) &=1 ,\qquad \text{dn}^2(u,m)+m\,\text{sn}^2(u,m) =1, \\
\frac{\partial}{\partial u}\text{am}(u,m) &=\text{dn}(u,m), \\
\frac{\partial}{\partial u}\text{sn}(u,m) &=\text{cn}(u,m)\text{dn}(u,m),\\
\frac{\partial}{\partial u}\text{cn}(u,m) &= - \text{sn}(u,m)\text{dn}(u,m),\\
\frac{\partial}{\partial u}\text{dn}(u,m) &= - m\, \text{sn}(u,m)\text{cn}(u,m).
\end{align}
\end{subequations}

\chapter{Explicit solutions to Kerr, NHEK and near-NHEK geodesic equations: technical details}
\chaptermark{\textsc{Geodesics: technical details}}

\vspace{\stretch{1}}

\section{Elementary Polar Integrals}\label{app:basicIntegrals}

For the pendular and equator-attractive cases, one has to compute the following integrals:
\begin{align}
    \hat I^{(0)}(x)&\triangleq\int_{x_0}^x\frac{\dd t}{\sqrt{\Theta(t^2)}},\quad\hat I^{(1)}(x)\triangleq\int_{x_0}^x\frac{t^2\dd t}{\sqrt{\Theta(t^2)}}\nonumber\\
    \hat I^{(2)}(x)&\triangleq\int_{x_0}^x\frac{\dd t}{\sqrt{\Theta(t^2)}} \frac{1}{1-t^2} \label{eqn:polarInt}
\end{align}
where $-1 \leq x \leq 1$. In the main text, $x$ will be substituted by $\cos\theta$, where $\theta$ is a polar angle.

\subsection{$\epsilon_0=0$}
In this case, $\Theta(t^2)=\sqrt{\frac{z_0}{Q}}(z_0-t^2)$. Choosing $x_0=0$, one finds directly
\begin{subequations}
    \begin{align}
    \hat I^{(0)}(x)&=\sqrt{\frac{z_0}{Q}}\arcsin{\frac{x}{\sqrt{z_0}}},\label{eqn:I0epsEq0}\\
    \hat I^{(1)}(x)&=\frac{1}{2}\sqrt{\frac{z_0}{Q}}\qty[z_0\arcsin{\frac{x}{\sqrt{z_0}}}-x\sqrt{z_0-x^2}],\\
    \hat I^{(2)}(x)&=\sqrt{\frac{z_0}{Q(1-z_0)}}\arcsin\sqrt{\frac{x}{z_0}\frac{1-z_0}{1-x}}
\end{align}
\end{subequations}
and the particular values
\begin{subequations}
    \begin{align}
    \hat I^{(0)}\qty(\sqrt{z_0})&=\frac{\pi}{2}\sqrt{\frac{z_0}{Q}},\\
    \hat I^{(1)}\qty(\sqrt{z_0})&=\frac{\pi}{4}\sqrt{\frac{z_0}{Q}},\\
    \hat I^{(2)}\qty(\sqrt{z_0})&=\frac{\pi}{2}\sqrt{\frac{z_0}{Q(1-z_0)}}.
\end{align}
\end{subequations}
One inverts \eqref{eqn:I0epsEq0} as
\begin{equation}
    x=\sqrt{z_0}\sin\qty(\frac{Q}{z_0}\hat I^{(0)}).
\end{equation}

\subsection{$\epsilon_0\neq0$, $z_-\neq0$}
In this case, $\Theta(t^2)=\epsilon_0(t^2-z_-)(z_+-t^2)$, where $t=\cos\theta$. Instead of solving separately the pendular and vortical cases as in Ref. \cite{Kapec:2019hro}, we will introduce a formal notation enabling us to treat both cases simultaneously. Let us define
\begin{equation}
    q\triangleq\,\text{sign}\, Q,\qquad z_{\pm1}\triangleq z_\pm,\qquad m\triangleq\left\lbrace\begin{array}{cc}
    \frac{z_+}{z_-},  & q=+1 \\
    1-\frac{z_-}{z_+},     & q=-1
    \end{array}\right.
\end{equation}
and
\begin{equation}
    y(t)\triangleq\left\lbrace\begin{array}{cc}
    \frac{t}{\sqrt{z_+}},  & q=+1  \\
   \,\text{sign}\, t \sqrt{\frac{z_+-t^2}{z_+-z_-}},     & q=-1
    \end{array}\right.,\qquad\Psi^q(x)\triangleq\arcsin{y(x)}.
\end{equation}
Pendular motion corresponds to $q = 1$, which has $0 \leq t^2 \leq z_+ <1$ and $\epsilon_0 z_- <0$, while vortical motion corresponds to $q = -1$, which has $0< z_- \leq t^2 \leq z_+ <1$ and $\epsilon_0 >0$. One can then rewrite
\begin{equation}
    \frac{\dd t}{\sqrt{\epsilon_0(t^2-z_-)(z_+-t^2)}}=\frac{q}{\sqrt{-q\epsilon_0z_{-q}}}\frac{\dd y}{\sqrt{(1-y^2)(1-my^2)}}.
\end{equation}
Both factors of the right side of this equation are real, either for pendular or for vortical motions. The lower bound of the integral will be chosen as $x_0=0$ for $Q\geq0$ and $x_0=\,\text{sign}\,{x}\sqrt{z_+}$ for $Q<0$. Then $y(x_0)=0$ for both values of $q$. This allows us to solve directly the integrals in terms of elliptic integrals (see Appendix \ref{app:ellipticFunctions} for definitions and conventions used):
\begin{subequations}
    \begin{align}
    \hat I^{(0)}(x)&=\frac{q}{\sqrt{-q\epsilon_0z_{-q}}}F\qty(\Psi^q(x),m)\label{eqn:I0epsNeq0},\\
    \hat I^{(1)}(x)&=\left\lbrace\begin{array}{ll}
    \frac{-2z_+}{\sqrt{-\epsilon_0 z_-}}E'\qty(\Psi^+(x),m), & q=+1\\
   - \sqrt{\frac{z_+}{\epsilon_0}}E\qty(\Psi^-(x),m),& q=-1
    \end{array}\right., \label{eqn:I1epsNeq0}\\
    \hat I^{(2)}(x)&=\left\lbrace\begin{array}{ll}
    \frac{1}{\sqrt{-\epsilon_0 z_-}}\Pi\qty(z_+,\Psi^+(x),m), & q=+1,\\
    \frac{-1}{(1-z_+)\sqrt{\epsilon_0 z_+}}\Pi\qty(\frac{z_--z_+}{1-z_+},\Psi^-(x),m), & q=-1
    \end{array}\right. .\label{eqn:I2epsNeq0}
\end{align}
\end{subequations}
For $x=\sqrt{z_q}$, $\Psi^q(\sqrt{z_q})=\frac{\pi}{2}$ for both $q= \pm 1$ and the incomplete elliptic integrals are replaced by complete ones:
\begin{subequations}
    \begin{align}
    \hat I^{(0)}(\sqrt{z_q})&=\frac{q}{\sqrt{-q\epsilon_0z_{-q}}}K\qty(m),\\
    \hat I^{(1)}(\sqrt{z_q})&=\left\lbrace\begin{array}{ll}
    \frac{-2z_+}{\sqrt{-\epsilon_0 z_-}}E'\qty(m), & q=+1\\
  -  \sqrt{\frac{z_+}{\epsilon_0}}E\qty(m),& q=-1
    \end{array}\right.,\\
    \hat I^{(2)}(\sqrt{z_q})&=\left\lbrace\begin{array}{ll}
    \frac{1}{\sqrt{-\epsilon_0 z_-}}\Pi\qty(z_+,m), & q=+1,\\
    \frac{-1}{(1-z_+)\sqrt{\epsilon_0 z_+}}\Pi\qty(\frac{z_--z_+}{1-z_+},m), & q=-1
    \end{array}\right. .
\end{align}
\end{subequations}
For $q=-1$, all integrals \eqref{eqn:I0epsNeq0}, \eqref{eqn:I1epsNeq0}, and \eqref{eqn:I2epsNeq0} vanish when evaluated at $x=\pm\sqrt{z_+}$ because $\Psi^-(\sqrt{z_+})=0$. Finally, \eqref{eqn:I0epsNeq0} can be inverted as
\begin{equation}
    x=y^{-1}\qty(\,\text{sn}\,\qty(\sqrt{-q\epsilon_0z_{-q}}\hat I^{(0)},m))\label{eqn:inversionEpsNeq0}
\end{equation}
with $y^{-1}$ the inverse function of $y$:
\begin{equation}
    y^{-1}(t)=\left\lbrace\begin{array}{ll}
        \sqrt{z_+}t, & q=+1  \\
         \,\text{sign}\, t\sqrt{z_+(1-mt^2)},& q=-1
    \end{array}\right.
\end{equation}
leading to the explicit formula
\begin{equation}
    x=\left\lbrace\begin{array}{ll}
        \sqrt{z_+}\,\text{sn}\,\qty(\sqrt{-\epsilon_0z_{-}}\hat I^{(0)},m), & q=+1 \\
        \,\text{sign}\, x \sqrt{z_+}\,\text{dn}\,\qty(\sqrt{\epsilon_0z_{+}}\hat I^{(0)},m),& q=-1.
    \end{array}\right.\label{eqn:inversionFormula}
\end{equation}
In the context of this paper, we will always have $x=\cos\theta$. 

\subsection{$\epsilon_0>0$, $z_-=0$}
This case is relevant for equator-attractive orbits. In this case, the potential reduces to $\Theta(t^2)=\epsilon_0\,t^2(z_+-t^2)$. One needs the following integrals (see also Ref. \cite{Kapec:2019hro}):
\begin{subequations}
    \begin{align}
    \hat{\mathcal{I}}^{(0)}&\triangleq\int_x^{\sqrt{z_+}}\frac{\dd t}{\sqrt{\Theta(t^2)}}=\frac{1}{\sqrt{\epsilon_0z_+}}\,\text{arctanh}\,\sqrt{1-\frac{x^2}{z_+}},\\
    \hat{\mathcal{I}}^{(1)}&\triangleq\int_x^{\sqrt{z_+}}\frac{t^2\dd t}{\sqrt{\Theta(t^2)}}=\sqrt{\frac{z_+-x^2}{\epsilon_0}},\\
    \hat{\mathcal{I}}^{(2)}&\triangleq\int_x^{\sqrt{z_+}}\frac{\dd t}{\sqrt{\Theta(t^2)}}\qty(\frac{1}{1-t^2}-1)=\frac{1}{\sqrt{\epsilon_0(1-z_+)}}\arctan\sqrt{\frac{z_+-x^2}{1-z_+}}
\end{align}
\end{subequations}
for any $x\leq \sqrt{z_+}$. All such integrals are obviously vanishing when evaluated at $x=\sqrt{z_+}$. Note that $ \hat{\mathcal{I}}^{(2)}$ contains a $-1$ integrand for simplicity of the final answer.

\section{Explicit form of prograde NHEK geodesics}
\label{app:equatorial}

We provide here the explicit motion in Mino time and the parametrized form of all classes of future-oriented geodesics in NHEK in the region outside the horizon $R > 0$. Without loss of generality, we can use the existence of the $\rightleftarrows$ and \rotatebox[origin=c]{-45}{$\rightleftarrows$}-flips to restrict ourselves to the subclasses of NHEK future-oriented geodesics, determined as
\begin{itemize}
    \item {\textit{Prograde ($\ell\geq 0$) geodesics;}} 
    \item {\textit{Partially ingoing geodesics,}} i.e., trajectories whose radial coordinate is decreasing on at least part of the total motion.
\end{itemize}

Because the near-horizon motion admits at most one turning point, we will only encounter spherical, plunging, marginal, bounded and deflecting motions.

\paragraph{Generic integrals.} One must provide explicit solutions to equations \eqref{eq:nhekT} to \eqref{eq:nhekphi}. The main point one has to deal with consists of solving the radial integrals \eqref{eqn:radialNHEK}. Notice that the primitives \begin{equation}
\mathbb{T}^{(i)}(R)=\int\frac{\dd R}{R^i\sqrt{E^2+2E\ell R-\mathcal{C}R^2}},\qquad i=0,1,2
\end{equation}
can be directly integrated as
\begin{subequations}
    \begin{align}
\mathbb{T}^{(0)}(R)&=\left\lbrace
\begin{array}{ll}
\frac{1}{\sqrt{-\mathcal{C}}}\log\qty[E\ell-\mathcal{C} R+\sqrt{-\mathcal{C}}\sqrt{v_R(R)}],  & \qquad\mathcal{C}\neq 0 \\
\frac{\sqrt{E^2+2E\ell_*R}}{E\ell_*},  & \qquad\mathcal{C}=0,~E\neq 0,
\end{array}
\right.\\
\mathbb{T}^{(1)}(R)&=\frac{\log R-\log\qty[E+\ell R+\sqrt{v_R(R)}]}{E}\qquad E\neq 0,\\
\mathbb{T}^{(2)}(R)&=-\frac{1}{E}\qty(\frac{\sqrt{v_R(R)}}{ER}+\ell\,\mathbb{T}^{(1)}(R)),\qquad E\neq 0.
\end{align}
\end{subequations}

We treat each type of motion cited above separately. We consider a geodesic path linking two events as described in the main text. We consider here a future-oriented path, $\Delta T>0$ and $\Delta\lambda>0$. Let us proceed systematically:
\begin{itemize}
    \item \textit{Spherical} motion has $R$ constant, and the radial integrals are ill defined and irrelevant. One can directly integrate the basic geodesic equations in this case.
    
    \item For \textit{plunging} motion, we will fix the final conditions at the final event $(T_f,R_f,$ $\theta_f,\Phi_f)$ such that $R_f<0$ (We will choose $R_f$ to be a root of the radial potential in all cases where they are real. Otherwise, we will chose it for convenience.) We will drop the subscript $i$ of the initial event. We denote here $\Delta\lambda=\lambda_f-\lambda$, $\Delta T=T_f-T(\lambda)$,\ldots and
    \begin{equation}
        T^{(j)}_R(R)=\mathbb{T}^{(j)}(R)-\mathbb T^{(j)}(R_f),\qquad j=0,1,2.
    \end{equation}
    
    \item For \textit{bounded} motion, we identify the physically relevant root of the radial potential $R_{\text{turn}}>0$ at $T_{\text{turn}}$ (which is an integration constant) that represents the turning point of the motion. We consider the motion either before or after the turning point: 
    \begin{itemize}[label=$\rightarrow$]
        \item ${\boldsymbol{T(\lambda)>T_{\text{turn}}}:}$ We choose the initial event as the turning point and the final one as $(T_>(\lambda),R_>(\lambda),$ $\theta_>(\lambda),\Phi_>(\lambda)).$ This leads to
        \begin{equation}
         T^{(j)}_R(R_>)=\mathbb{T}^{(j)}(R_{\text{turn}})-\mathbb T^{(j)}(R_>),\qquad j=0,1,2.   
        \end{equation}
        \item ${\boldsymbol{T(\lambda)<T_{\text{turn}}}:}$ We choose the initial event as $(T_<(\lambda),R_<(\lambda),\theta_<(\lambda),$ $\Phi_<(\lambda))$ and the final event as the turning point. This leads to
        \begin{equation}
         T^{(j)}_R(R_<)=\mathbb{T}^{(j)}(R_{\text{turn}})-\mathbb T^{(j)}(R_<),\qquad j=0,1,2.
        \end{equation}
    \end{itemize}
    The motion being symmetric in $R$ with respect to $R_{\text{turn}}$, one must only determine $T_>(R)$ and $\Phi_>(R)$, because
    \begin{equation}
        T_<(R)=-T_>(R),\qquad\Phi_<(R)=-\Phi_>(R).\label{eqn:symmetry}
    \end{equation}
    \item For \textit{deflecting} motion, there is also a turning point and \eqref{eqn:symmetry} remains true. This case is treated similarly to the bounded case but with a minus-sign change. We get here
    \begin{align}
        T_{R}^{(j)}(R_>)&=\mathbb T^{(j)}(R_>)-\mathbb T^{(j)}(R_{\text{turn}}),\nonumber\\
        T^{(j)}_R(R_<)&=\mathbb T^{(j)}(R_<)-\mathbb{T}^{(j)}(R_{\text{turn}}).
    \end{align}
    
    \item Finally, for \textit{marginal} motion, there is no turning point and the integral is immediate.

\end{itemize}
In what follows, for the motions with one turning point, we will only make explicit the motion after the turning point, $T>T_{\text{turn}}$, and we will denote $T=T_>$ and $\Phi=\Phi_>$ in order to simplify the notations.

\vspace{10pt}
\paragraph{\bf Spherical$_*$ (ISSO).} The explicit form is
\begin{subequations}
    \begin{align}
 T(\lambda) &= T_0 + \frac{\ell_*}{R_0}(\lambda-\lambda_0)  , \\
 R(\lambda) &= R_0 , \\
\Phi(\lambda) &=\Phi_0 - \frac{3}{4}\ell_* (\lambda-\lambda_0) + \ell_* \Phi_\theta(\lambda-\lambda_0). 
\end{align}
\end{subequations}
The parametrized form is 
\begin{subequations}
    \begin{align}
R & = R_0 , \\
\Phi & =  \Phi_0 -\frac{3}{4}R_0 (T-T_0) + \ell_*  \Phi_\theta\qty( \frac{R_0}{\ell_*} (T-T_0)). 
\end{align}
\end{subequations}

\paragraph{\bf Marginal$(\ell)$.} The Casimir obeys here $\mathcal{C}<0$; we denote $q\triangleq\sqrt{-\mathcal{C}}$ and consider the initial condition $R(\lambda_m)=R_m$. The explicit form of the solution is given (as a function of the Mino time) by
\begin{subequations}
    \begin{align}
 T(\lambda) &= T_0+\frac{\ell}{R_m q}\exp\qty(q(\lambda-\lambda_m)), \\
 R(\lambda) &= R_m \exp\qty(-q(\lambda-\lambda_m)) , \\
\Phi(\lambda) &= \Phi_0 -\frac{3}{4}\ell(\lambda-\lambda_m)+\ell\Phi_\theta(\lambda-\lambda_m). 
\end{align}
\end{subequations}
The parametrized form is
\begin{subequations}
    \begin{align}
T(R) & =  T_0+\frac{\ell}{q\,R}, \\
\Phi(R) & =  \Phi_0+\frac{3}{4}\frac{\ell}{q}\log\frac{R}{R_m}+\ell\Phi_\theta\qty(\lambda\qty(\frac{R}{R_m})). 
\end{align}
\end{subequations}
Here, the constants $T_0$ and $\Phi_0$ remain arbitrary.

\paragraph{\bf Plunging$_*(E)$.} The energy satisfies $E>0$ and the initial condition is $R(\lambda_0)=R_0=-\frac{E}{2\ell_*}$, leading to
\begin{align}
\lambda(R) -\lambda_0= - \frac{1}{\ell_*}  \sqrt{1+\frac{2 \ell_* R}{E}}. 
\end{align}
The explicit form is
\begin{subequations}
    \begin{align}
T(\lambda) &= \frac{2  \ell_*^2 (\lambda - \lambda_0)}{E(1-\ell_*^2 (\lambda-\lambda_0)^2)}, \\
 R(\lambda) &= \frac{E}{2\ell_*} \left( \ell_*^2 (\lambda-\lambda_0)^2 - 1 \right) , \\
\Phi(\lambda) &= -\frac{3}{4}\ell_* (\lambda-\lambda_0) +2 \, \,\text{arctanh}\, (\ell_* (\lambda - \lambda_0)) \nonumber\\
&+ \ell_* \Phi_\theta(\lambda-\lambda_0) . 
\end{align}
\end{subequations}
The parametrized form is 
\begin{subequations}
    \begin{align}
T(R) & = \frac{1}{R} \sqrt{1+\frac{2 \ell_* R}{E}}, \\
\Phi(R) & = \frac{3}{4} \sqrt{1+\frac{2 \ell_* R}{E}}-2 \, \,\text{arctanh}\,\sqrt{1+\frac{2 \ell_* R}{E}}+\ell_* \Phi_\theta (\lambda(R)-\lambda_0). 
\end{align}
\end{subequations}

\paragraph{\bf Plunging$(E,\ell)$.} The parameters obey $\mathcal C < 0$ and $E>0$. We denote $q = \sqrt{-\mathcal C}$ and $\lambda_-$ the value of the Mino time such that $R(\lambda_-)= R_-$. The orbit is given by
\begin{align}
\lambda(R) - \lambda_- = \frac{1}{q} \log\left(  \frac{E \sqrt{\mathcal C+\ell^2}}{E \ell - \mathcal C R +q \sqrt{v_R}} \right)=\frac{i}{q}\arccos\qty(\frac{E\ell-\mathcal{C}R}{E\sqrt{\mathcal{C}+\ell^2}}).
\end{align}
The last form is not explicitly real, but it allows us to find easily
\begin{subequations}
    \begin{align}
T(\lambda) &= \frac{q}{E} \frac{\abs{\sinh (q(\lambda-\lambda_-))}}{\frac{\ell}{\sqrt{\mathcal C + \ell^2}}-\cosh (q(\lambda-\lambda_-))}, \\
R(\lambda) &= \frac{E}{\mathcal C} \left( \ell - \sqrt{\mathcal C + \ell^2} \cosh (q (\lambda - \lambda_-))\right), \\
\Phi(\lambda) &= - \frac{3 \ell}{4}(\lambda-\lambda_-) +2\, \,\text{arctanh}\, \left( \frac{\ell + \sqrt{\mathcal C + \ell^2}}{q} \tanh (\frac{q}{2}(\lambda-\lambda_-))\right)\nonumber\\
&+\ell\, \Phi_\theta(\lambda-\lambda_-) . 
\end{align}
\end{subequations}
The parametrized form can be simplified as
\begin{subequations}
    \begin{align}
T(R) & = \frac{\sqrt{v_R}}{E R},\label{eqn:paramNHEKR} \\
\Phi(R) & = \Phi_0-\log\frac{E+\ell R+\sqrt{v_R(R)}}  {R}\nonumber\\
&~+\frac{3\ell}{4q}\log\qty(E\ell-\mathcal{C}R+q\sqrt{v_R(R)})+\ell\Phi_\theta(\lambda(R)-\lambda_-)  \label{eqn:paramNHEKphi}
\end{align}
\end{subequations}
where $\Phi_0\triangleq \log\frac{E+\ell R_-}{R_-}-\frac{3\ell}{4q}\log\qty(E\ell-\mathcal C R_-)$. The orbit start at $R=+\infty$ at $\lambda=-\infty$ and reaches the black hole horizon $R=0$ at $\lambda = \lambda_- - \lambda_H$, where $\lambda_H \equiv \frac{1}{q}\text{arccosh}(\frac{\ell}{\sqrt{\mathcal C + \ell^2}})$. It never reaches $R_- < 0$. 

\paragraph{\bf Bounded$_<(E,\ell)$.}
We have $\mathcal C > 0$ and $E>0$, leading to
\begin{align}
\lambda(R) -\lambda_+ = \frac{1}{q}\text{arccos}\left( \frac{\mathcal C R - E \ell}{E \sqrt{\mathcal C+\ell^2}} \right).
\end{align}
We normalized Mino time such that $R(\lambda_+) = R_+$ and denoted $q \triangleq \sqrt{\mathcal C}$. The orbit starts at $\lambda = \lambda_+$ at the turning point $R_+$ and plunges inside the black hole $R=0$ at $\lambda = \lambda_++\frac{1}{q}\text{arccos}(\frac{-\ell}{\sqrt{\mathcal C+\ell^2}})$. We have
\begin{subequations}
    \begin{align}
T (\lambda) &=  \frac{q}{E} \frac{\sin (q(\lambda-\lambda_+))}{\frac{\ell}{\sqrt{\mathcal C + \ell^2}}+\cos (q(\lambda-\lambda_+))} , \\
R(\lambda) &= \frac{E}{\mathcal C} \left( \ell + \sqrt{\mathcal C + \ell^2} \cos (q (\lambda-\lambda_+) )\right) , \\
\Phi(\lambda)  &= - \frac{3 \ell}{4}(\lambda-\lambda_+) +2\,\,\text{arctanh}\, \left( \frac{\ell - \sqrt{\mathcal C + \ell^2}}{q} \tan (\frac{q}{2}(\lambda-\lambda_+))\right)\nonumber\\
&+\ell \Phi_\theta(\lambda-\lambda_+). 
\end{align}
\end{subequations}
The parametrized form is given by \eqref{eqn:paramNHEKR} and \eqref{eqn:paramNHEKphi}, but with $q\to iq$. One can write a manifestly real form of the azimutal coordinate by shifting the initial value $\Phi_0$, leading to
    \begin{align}
    \Phi(R)=\Phi_0'-\log\frac{E+\ell R+\sqrt{v_R(R)}}  {R}+\frac{3\ell}{4q}\arctan\frac{q\sqrt{v_R}}{E\ell-\mathcal C R}+\ell\Phi_\theta(\lambda(R)-\lambda_-)
\end{align}
with $\Phi'_0\triangleq\log\frac{E+\ell R_-}{R_-}.$

\paragraph{Def\mbox{}lecting$(E,\ell)$.} We have $\mathcal C < 0$ and $E<0$, leading to ($q\triangleq\sqrt{-\mathcal{C}}$)
\begin{equation}
    \lambda(R)-\lambda_+=-\frac{i}{q}\arccos{\frac{E\ell-\mathcal{C}R}{E\sqrt{\mathcal{C}+\ell^2}}}.\label{eq45}
\end{equation}
The initial condition $R(\lambda_+)=R_+$ corresponds to the minimal radius reached by the trajectory. The orbit starts and ends at $R=+\infty$ at Mino time $\lambda=\pm\infty$. We have
\begin{subequations}
    \begin{align}
    T(\lambda)&=\frac{q}{E}\frac{\sinh{q(\lambda-\lambda_+)}}{\frac{\ell}{\sqrt{\mathcal{C}+\ell^2}}-\cosh{q(\lambda-\lambda_+)}}\\
    R(\lambda)&=\frac{E}{\mathcal{C}}\qty(\ell-\sqrt{\mathcal{C}+\ell^2}\cosh{q(\lambda-\lambda_+)})\\
    \Phi(\lambda)&=-\frac{3\ell}{4}(\lambda-\lambda_+)+2\,\,\text{arctanh}\, \left( \frac{\ell - \sqrt{\mathcal C + \ell^2}}{q} \tanh (\frac{q}{2}(\lambda-\lambda_+))\right)
    \nonumber\\
&+\ell \Phi_\theta(\lambda-\lambda_+). 
\end{align}
\end{subequations}
We finally have the parametrized form 
\begin{subequations}
    \begin{align}
    T(R) & = - \frac{\sqrt{v_R(R)}}{E R},\\
    \Phi(R)& = \Phi_0+\log\frac{E+\ell R+\sqrt{v_R(R)}}  {R}\nonumber\\
&~-\frac{3\ell}{4q}\log\qty(E\ell-\mathcal{C}R+q\sqrt{v_R(R)})+\ell\Phi_\theta(\lambda(R)-\lambda_+)
\end{align}
\end{subequations}
with $\Phi_0\triangleq-\log\frac{E+\ell R_+}{R_+}+\frac{3\ell}{4q}\log\qty(E\ell-\mathcal{C}R_+)$.

\section{Explicit form of prograde near-NHEK geodesics}
\label{app:explicitNN}

We follow the same procedure as the one in Appendix \ref{app:equatorial}. As before, we only focus, without loss of generality, on future-oriented partially ingoing prograde orbits. For convenience, we also include one class of retrograde bounded geodesics. 

\paragraph{Spherical$(\ell)$.} The explicit form reads as
\begin{subequations}
    \begin{align}
    R&=R_0=\frac{\kappa\ell}{\sqrt{-\mathcal C}},\\
    t(\lambda)&=t_0+\frac{\ell}{R_0}(\lambda-\lambda_0),\\
    \phi(\lambda)&=\phi_0-\frac{3}{4}\ell(\lambda-\lambda_0)+\ell\Phi_\theta(\lambda-\lambda_0).
\end{align}
\end{subequations}
The parametrized form is
\begin{subequations}
    \begin{align}
    R&=R_0=\frac{\kappa\ell}{\sqrt{-\mathcal C}},\\
    \phi(t)&=\phi_0-\frac{3}{4}R_0(t-t_0)+\ell\Phi_\theta(\lambda-\lambda_0).\label{eq46}
\end{align}
\end{subequations}
Note that $R_0\geq\qty(\frac{2}{\sqrt{3}}-1)\kappa$.

\paragraph{Plunging$_*$.} One has
\begin{equation}
    \lambda-\lambda_i=-\frac{R-R_i}{\kappa\ell_*}.
\end{equation}
The explicit form is
\begin{subequations}
    \begin{align}
    R(\lambda)&=R_i-\kappa\ell_*(\lambda-\lambda_i),\\
    t(\lambda)&=-\frac{1}{2\kappa}\log\qty[1+\kappa\ell_*(\lambda-\lambda_i)\frac{\kappa\ell_*(\lambda-\lambda_i)-2R_i}{R_i^2-\kappa^2}],\\
    \phi(\lambda)&=-\frac{3}{4}\ell_*(\lambda-\lambda_i)+\frac{1}{2}\log\frac{1-\frac{\kappa\ell_*(\lambda-\lambda_i)}{R_i-\kappa}}{1-\frac{\kappa\ell_*(\lambda-\lambda_i)}{R_i+\kappa}}+\ell_*\Phi_\theta(\lambda-\lambda_i).
\end{align}
\end{subequations}
The parametrized form is
\begin{subequations}
    \begin{align}
    t(R)&=-\frac{1}{2\kappa}\log\frac{R^2-\kappa^2}{R_i^2-\kappa^2},\\
    \phi(R)&= \phi_i+\frac{3}{4\kappa}R+\frac{1}{2}\log\frac{R-\kappa}{R+\kappa}+\ell_*\Phi_\theta(\lambda(R)-\lambda_i)
\end{align}
\end{subequations}
with $\phi_i\triangleq -\frac{3R_i}{4\kappa}-\frac{1}{2}\log\frac{R_i-\kappa}{R_i+\kappa}$. The geodesic starts from $R=+\infty$ at $\lambda=-\infty$ and reaches the horizon at Mino time $\lambda_H=\lambda_i+\frac{1}{\ell_*}\qty(\frac{R_i}{\kappa}-1)$.

\paragraph{Bounded$_*(e)$ and Plunging$_*(e)$.} The orbital parameters satisfy $\mathcal C=0$ and $e<0$ (bounded) or $e>0$ (plunging). The potential is simply $v_{R;\kappa}(R)=e^2+\kappa^2\ell_*^2+2e\ell_* R$ and the initial conditions are imposed at $R(\lambda_0)=R_0$, where $R_0$ is the (unique) root of the radial potential. One directly obtains
\begin{equation}
    \lambda(R)-\lambda_0=-\frac{\sqrt{v_{R;\kappa(R)}}}{e\ell_*},
\end{equation}
leading to
\begin{subequations}
    \begin{align}
    R(\lambda)&=R_0+\frac{e\ell_*}{2}(\lambda-\lambda_0)^2,\\
    t(\lambda)&= -\frac{1}{\kappa}\,\,\text{arctanh}\,\frac{2\kappa e \ell_*^2(\lambda-\lambda_0)}{\kappa^2\ell_*^2+e^2\qty[\ell^2_*(\lambda-\lambda_0)^2-1]},\\
    \phi(\lambda)&=-\frac{3}{4}\ell_*(\lambda-\lambda_0)+\,\text{sign}\, e\, \,\text{arctanh}\,\frac{2e^2\ell_*(\lambda-\lambda_0)}{-\kappa^2\ell_*^2+e^2\qty[\ell_*^2(\lambda-\lambda_0)^2+1]}\nonumber\\
    &~+\ell_*\Phi_\theta(\lambda-\lambda_0).   
\end{align}
\end{subequations}
The parametric form is
\begin{subequations}
    \begin{align}
    t(R)&=\frac{1}{\kappa}\,\text{arccosh} \frac{\abs{R+\frac{\kappa^2\ell_*}{e}}}{\sqrt{R^2-\kappa^2}},\\
    \phi(R)&=-\frac{3}{4e}\sqrt{v_{R;\kappa}(R)}+\,\text{arctanh}\,\frac{\sqrt{v_{R;\kappa}(R)}}{e+\ell_* R}+\ell_*\Phi_\theta(\lambda(R)-\lambda_0).
\end{align}
\end{subequations}
Notice that the requirements $v_{R;\kappa}\geq0$ and $R\geq\kappa$ are sufficient to guarantee the reality of the inverse hyperbolic functions involved. The trajectory reaches the horizon at Mino time $\lambda_H=\lambda_0-\,\text{sign}\, e \sqrt{\frac{2(\kappa-R_0)}{e\ell_*}}$, which is smaller than $\lambda_0$ for plunging motion and greater for bounded motion, as expected.

\paragraph{(Retrograde) Bounded$_>(e,\ell)$.} The geodesic parameters satisfy $\mathcal{C}<0$, $\ell < 0$, and $e>0$. Therefore, $R_-$ is positive, and we choose the initial condition as $R(\lambda_-)=R_-$. Defining $q\triangleq\sqrt{-\mathcal{C}}$, the explicit form reads as
\begin{align}
    R(\lambda)&= \frac{1}{\mathcal{C}}\qty[e\ell-\sqrt{(\mathcal{C}+\ell^2)(e^2+\kappa^2\mathcal{C})}\cosh{q(\lambda-\lambda_-)}].
\end{align}
The parametrized form is
\begin{subequations}
    \begin{align}
    t(R)&=\frac{1}{4\kappa}\log\frac{F_+(R)}{F_-(R)},\label{eqn:paramNNt}\\
    \phi(R)&= \frac{3\ell}{4q}\log\qty[e\ell-\mathcal{C}R+q\sqrt{v_{R;\kappa}(R)}]-\frac{1}{4}\log\frac{G_+(R)}{G_-(R)}+\ell\Phi_\theta(\lambda(R)-\lambda_-)\label{eqn:paramNNphi}
\end{align}
\end{subequations}
where we define
\begin{subequations}
    \begin{align}
    F_\pm(R)&\triangleq \qty[eR+\kappa\qty(\kappa\ell\pm\sqrt{v_{R;\kappa}(R)})]^2,\\
    G_\pm(R)&\triangleq\qty(e+\ell R\pm\sqrt{v_{R;\kappa}(R)})^2.
\end{align}
\end{subequations}
Note that using the identities
\begin{subequations}
    \begin{align}
    F_+(R)F_-(R)&=(e^2+\kappa^2\mathcal C)^2(R^2-\kappa^2)^2,\\
    G_+(R)G_-(R)&=(\mathcal C+\ell^2)^2(R^2-\kappa^2)^2,
\end{align}
\end{subequations}
one can rewrite \eqref{eqn:paramNNt} and \eqref{eqn:paramNNphi} as
\begin{subequations}
    \begin{align}
    t(R)&=-\frac{1}{2\kappa}\log\frac{F_+(R)}{(e^2+\kappa^2\mathcal{C})(R^2-\kappa^2)},\\
    \phi(R)&= \frac{3\ell}{4q}\log\qty[e\ell-\mathcal{C}R+q\sqrt{v_{R;\kappa}(R)}]\nonumber\\
    &\quad-\frac{1}{2}\log\frac{G_+(R)}{(\mathcal{C}+\ell^2)(R^2-\kappa^2)}+\ell\Phi_\theta(\lambda(R)-\lambda_-).
\end{align}
\end{subequations}
The geodesic motion starts from the past horizon, reaches $R_-$ at Mino time $\lambda_-$ and crosses the future horizon at  $\lambda_H-\lambda_-=\frac{1}{q}\text{arccosh}\frac{e\ell-\kappa\mathcal{C}}{\sqrt{(\mathcal{C}+\ell^2)(e^2+\kappa^2\mathcal{C})}}$.

\paragraph{Bounded$_<(e,\ell)$.} The parameters obey $\mathcal{C}>0$ and $e\neq0$; therefore, $R_+$ is positive and the initial condition is chosen as $R(\lambda_+)=R_+$. One defines $q\triangleq\sqrt{\mathcal{C}}$ and gets the explicit form:
\begin{align}
    R(\lambda)&= \frac{1}{\mathcal{C}}\qty[e\ell+\sqrt{(\mathcal{C}+\ell^2)(e^2+\kappa^2\mathcal{C})}\cos\qty(q(\lambda-\lambda_+))].
\end{align}
The parametrized form is given by \eqref{eqn:paramNNt} and \eqref{eqn:paramNNphi}, with the replacement rule $q\to iq$. A manifestly real form of $\phi$ is
\begin{align}
    \phi(R)&=\frac{3\ell}{4q}\arctan\frac{q\sqrt{v_{R;\kappa}}}{e\ell-\mathcal C R}-\frac{1}{2}\log\frac{G_+(R)}{(\mathcal{C}+\ell^2)(R^2-\kappa^2)}+\ell\Phi_\theta(\lambda(R)-\lambda_-).
\end{align}
The geodesic motion starts from the white hole past horizon, reaches $R_+$ at Mino time $\lambda_+$ and crosses the future horizon at  $\lambda_H-\lambda_+=\frac{1}{q}\,\text{arccos}\frac{\kappa\mathcal{C}-e\ell}{\sqrt{(\mathcal{C}+\ell^2)(e^2+\kappa^2\mathcal{C})}}$.

\paragraph{Plunging$(e,\ell)$.} The parameters satisfy $\mathcal C<0$ and $e>-\kappa q$, where $q\triangleq\sqrt{-\mathcal C}$. The roots of the radial potential are consequently either complex or negative. In the complex case ($e^2+\kappa^2\mathcal{C} < 0$), we define the real quantity 
\begin{equation}
    R_f \triangleq\frac{1}{\mathcal{C}}\qty[e\ell-\sqrt{-(\mathcal{C}+\ell^2)(e^2+\kappa^2\mathcal{C})}]
\end{equation}
and impose the final condition $R({\lambda}_f)={R}_f$, leading to
\begin{align}
    R(\lambda)&= \frac{1}{\mathcal{C}}\bigg[e\ell-\sqrt{-( \mathcal{C}+\ell^2)(e^2+\kappa^2\mathcal{C})}\times\qty(\cosh{q(\lambda-{\lambda}_f)}-\sqrt{2}\sinh{q(\lambda-{\lambda}_f)})\bigg].\nonumber
\end{align}
Note that $\lim_{x\to\pm\infty}\cosh x -\sqrt{2}\sinh x=\mp\infty$ leads to the expected behavior. For negative roots, $R=R_-$ can be used as the final condition. In both cases, the parametrized form is again as given in \eqref{eqn:paramNNt} and \eqref{eqn:paramNNphi}. The orbit starts from $R=+\infty$ at $\lambda=-\infty$ and reaches the horizon at Mino time $\lambda_H-{\lambda}_f=\frac{1}{q}\,\text{arcsinh}\, 1+\log\frac{e\ell-\kappa\mathcal{C}+q((e+\kappa\ell)}{\sqrt{-(\mathcal{C}+\ell^2)(e^2+\kappa^2\mathcal{C})}}$.

\paragraph{Def\mbox{}lecting$(e,\ell)$.} One has $\mathcal{C}<0$ and $e<0$. Choosing the initial condition as $R(\lambda_-)=R_-$ and defining again $q\triangleq\sqrt{-\mathcal{C}}$, we get
\begin{align}
    R(\lambda)&=\frac{1}{\mathcal{C}}\qty[e\ell-\sqrt{(\mathcal{C}+\ell^2)(e^2+\kappa^2\mathcal{C})}\cosh\qty(q(\lambda-\lambda_-))].
\end{align}
The parametrized form is
\begin{subequations}
    \begin{align}
    t(R)&=-\frac{1}{4\kappa}\log\frac{F_+(R)}{F_-(R)},\\
    \phi(R)&=- \frac{3\ell}{4q}\log\qty[e\ell-\mathcal{C}R+q\sqrt{v_{R;\kappa}(R)}]\nonumber\\
    &\quad+\frac{1}{4}\log\frac{G_+(R)}{G_-(R)}+\ell\Phi_\theta(\lambda(R)-\lambda_-).
\end{align}
\end{subequations}
The orbit starts from $R=+\infty$ at $\lambda=-\infty$, reaches its minimal radial value $R_-$ at Mino time $\lambda_-$, and goes back to the asymptotic region at $\lambda\to+\infty$.

\chapter{Results for the second-order quadratic conserved quantities}
\chaptermark{\textsc{Second-order conserved quantities}}

\vspace{\stretch{1}}

\section{Derivation of the quadratic quantity constraint}\label{app:quadratic_constraint}

The aim of this appendix is to provide a derivation of the reduced expression \eqref{developped_cst} for the constraint of grading $[2,3]$  for the quadratic conserved quantity \eqref{quadratic_ansatz}.

Let us denote by $\expval{f}_{(n)}$ the terms contained in $f$ that are homogeneous of order $\mathcal O(\mathcal S^n)$.
We now take for granted the validity and uniqueness of the solution \eqref{R2}. The $\mathcal O(\mathcal S^2)$ constraint equation takes the form
\begin{align}
    \expval{\dot{\mathcal{Q}}}_{(2)}=\expval{\dot{\mathcal{Q}}_R}_{(2)}+\expval{\dot{\mathcal{Q}}^{\text{quad}}}_{(2)}\stackrel{!}{=}0.
\end{align}
We will compute these two contributions separately.

\subsection{Terms coming from $\mathcal Q^{\text{quad}}$}
Using Eqs. \eqref{MPD_2}, \eqref{masses} and \eqref{p_of_v}, the variation of ${\mathcal{Q}}^{\text{quad}}$ along the trajectory is given by
\begin{align}
   \begin{split}
    \dot{\mathcal{Q}}^{\text{quad}}&=\frac{\text{D}}{\dd\tau}M_{\alpha\beta\gamma\delta}S^{\alpha\beta}S^{\gamma\delta}+2\frac{\text{D}S^{\alpha\beta}}{\dd\tau}M_{\alpha\beta\gamma\delta}S^{\gamma\delta}\\
    &=v^\lambda\qty(\nabla_\lambda M_{\alpha\beta\gamma\delta} S^{\alpha\beta}S^{\gamma\delta}+4M_{\alpha\beta\gamma\lambda}S^{\alpha\beta}p^\gamma)+\mathcal O(\mathcal S^3)\\
    &=\hat p^\lambda \qty(\nabla_\lambda M_{\alpha\beta\gamma\delta} S^{\alpha\beta}S^{\gamma\delta}+4M_{\alpha\beta\gamma\lambda}S^{\alpha\beta}p^\gamma)+\mathcal O(\mathcal S^3)\\
    &= \nabla_\lambda M_{\alpha\beta\gamma\delta}\hat p^\lambda S^{\alpha\beta}S^{\gamma\delta}+\mathcal O(\mathcal S^3).
   \end{split}
\end{align}
Recalling that $S^{\alpha\beta}=2s^{[\alpha}\hat p^{\beta]*}$, we obtain the following equation in terms of the independent variables $s_\alpha$ and $p^\mu$:
\begin{eqnarray}
   \expval{\dot{\mathcal{Q}}^{\text{quad}}}_{(2)}&=4\,\nabla_\mu \sMs\tudud{\alpha}{\nu}{\beta}{\rho}s_\alpha s_\beta \hat p^\mu \hat p^\nu \hat p^\rho
=4\,\nabla_\mu N_{\alpha\nu\beta\rho}s^\alpha s^\beta \hat p^\mu \hat p^\nu \hat p^\rho.\label{Qdot}
\end{eqnarray}

\subsection{Terms coming from $\mathcal Q_R$}
One can perform the splitting
\begin{align}
    \expval{\dot{\mathcal{Q}}_R}_{(2)}&=\expval{\dot{\mathcal{Q}}^{MD}_R}_{(2)}+\expval{\dot{\mathcal{Q}}^Q_R}_{(2)}.
\end{align}
Here, the ``monopole-dipole'' terms $\dot{\mathcal{Q}}_R^{MD}$ are those that where already present in \cite{Compere:2021kjz}:
\begin{align}
 \begin{split}
     \dot{\mathcal{Q}}_R^{MD}&\triangleq \mu^2 D\tud{\lambda}{\rho}\nabla_\lambda K_{\mu\nu}\hat p^\mu \hat p^\nu\hat p^\rho-\frac{\mu}{2}L_{\mu\nu\rho}S^{\mu\nu}R\tud{\rho}{\alpha\beta\gamma}\hat p^\alpha S^{\beta\gamma} +2\mu^2  L_{\mu\lambda\rho}D\tud{\lambda}{\nu}\hat p^\mu \hat p^\nu p^\rho  \\
    &=4 \mu \sWs \tudud{\alpha}{\mu\nu}{\beta}{\rho} s_\alpha s_\beta \hat p^\mu \hat p^\nu \hat p^\rho +O(\mathcal S^3). \label{muProblem}
 \end{split}
\end{align}
where (see Eq. (73) of \cite{Compere:2021kjz})
\begin{align}
    \sWs_{\alpha\beta\gamma\delta\epsilon}=-\frac{1}{2}\,^*L_{\alpha\beta\lambda}R\tud{*\lambda}{\gamma\delta\epsilon}.\label{Wss}
\end{align}
The \textit{relaxed spin vector} $s^\alpha$ is defined from $ S^\alpha\triangleq\Pi^\alpha_\beta s^\beta$ where the part of $\mathbf s$ aligned with $\mathbf p$ is left arbitrary, but is assumed (without loss of generality) to be of the same order of magnitude. 

The tensor $L_{\alpha\beta\gamma}$ is defined in Eq. \eqref{R2}. In order to simplify Eq. \eqref{Wss}, we have on the one hand
\begin{align}
    ^*\epsilon_{\alpha\beta\lambda\rho}\nabla^\rho\mathcal Z=-2g_{\lambda[\alpha}\nabla_{\beta]}\mathcal Z.
\end{align}
On the other hand,
\begin{align}
    \nabla_{[\alpha}K_{\beta]*\lambda}&=2Y_{\lambda[\alpha}\xi_{\beta]}+3g_{\lambda[\alpha}Y\tud{\rho}{\beta}\xi_{\rho]} =3Y_{\lambda[\alpha}\xi_{\beta]}+g_{\lambda[\alpha}\nabla_{\beta]}\mathcal Z,
\end{align}
where $\xi^\alpha=-\frac{1}{3}\nabla_\lambda Y^{*\lambda\alpha}$ is the timelike Killing vector associated with the Killing-Yano tensor.  Gathering these pieces together, we obtain the reduced expression
\begin{align}
    \sWs_{\alpha\beta\gamma\delta\epsilon}s^\alpha s^\beta \hat p^\mu \hat p^\nu \hat p^\rho=-\qty[R\tud{*\lambda}{\gamma\delta\epsilon}Y_{\lambda[\alpha}\xi_{\beta]}+\nabla_{[\alpha}\mathcal Z R^*_{\beta]\gamma\delta\epsilon}]s^\alpha s^\beta\hat p^\mu\hat p^\nu\hat p^\rho.\label{WssS}
\end{align}
The monopole-dipole piece is consequently given by
\begin{align}
    \expval{\dot{\mathcal{Q}}^{MD}_R}_{(2)}&=\left( 2\qty(Y_{\lambda\mu}\xi_\alpha-Y_{\lambda\alpha}\xi_\mu)R\tud{*\lambda}{\nu\beta\rho}-2\nabla_\mu\mathcal Z R^*_{\nu\alpha\beta\rho}\right) s^\alpha s^\beta \hat p^\mu \hat p^\nu \hat p^\rho.\label{exprO}
\end{align}

The ``quadrupolar'' terms $\dot{\mathcal{Q}}_R^{Q}$ are the ones induced by the presence of the quadru-pole, namely
\begin{align}
  \begin{split}
  \dot{\mathcal{Q}}_R^{Q}&\triangleq 2\mu K_{\mu\nu}\hat p^\mu\mathcal F^\nu-\mu\mathcal L\tud{\lambda}{\rho}\nabla_\lambda K_{\mu\nu}\hat p^\mu\hat p^\nu\hat p^\rho\\
  &\quad-2\mu L_{\mu\lambda\rho}\mathcal L\tud{\lambda}{\nu}\hat p^\mu\hat p^\nu\hat p^\rho+\mu L_{\mu\nu\rho}\mathcal L^{\mu\nu}\hat p^\rho.
  \end{split}
\end{align}

Considering only the spin-induced quadrupole \eqref{spin_induced_Q} and making use of the identities \eqref{F} to \eqref{Lv} as well as the explicit form of the tensor $L_{\mu\nu\rho}$ \eqref{R2}, we obtain the reduced expression
\begin{align}
    \begin{split}
    \expval{\dot{\mathcal{Q}}^{Q}_R}_{(2)}&=\kappa \bigg[K_{\mu\lambda}\nabla^\lambda R_{\nu\alpha\beta\rho}-\nabla_\lambda K_{\mu\nu} R\tud{\lambda}{\alpha\beta\rho} \\
    &\quad-\frac{4}{3}\nabla_{[\alpha}K_{\lambda]\nu} R\tud{\lambda}{\mu\beta\rho} -\frac{8}{3}\epsilon_{\alpha\gamma\rho\lambda}\nabla^\lambda \mathcal Z R\tud{\gamma}{\mu\beta\nu}\bigg]\Theta^{\alpha\beta}\hat p^\mu \hat p^\nu \hat p^\rho.
    \end{split}
\end{align}
Thanks to the orthogonality condition $p_\alpha S^\alpha=0$, we can substitute $\Theta^{\alpha\beta}$ in this expression with $\theta^{\alpha\beta}$ defined as 
\begin{equation}
    \theta^{\alpha\beta} = \Pi^{\alpha\beta}\mathcal S^2 -s^\alpha s^\beta .
\end{equation}
After a few algebraic manipulations and making use of Bianchi identities, we obtain 
\begin{align}
    \begin{split}
    &\expval{\dot{\mathcal{Q}}^{Q}_R}_{(2)}=\kappa\bigg[ \left( \nabla_\nu\qty(K_{\mu\lambda}R\tud{\lambda}{\alpha\beta\rho})+\nabla_\nu K_{\mu\lambda}R\tud{\lambda}{\alpha\beta\rho}-\qty(\nu\leftrightarrow\alpha)\right)\\
    &\quad+\frac{4}{3}\nabla_{[\alpha} K_{\mu ] \lambda}R\tud{\lambda}{\nu\beta\rho} {-\frac{16}{3}\nabla_\lambda\mathcal Z g_{\mu[\nu} \,^*\!R\tud{\lambda}{\alpha]\beta\rho}+}\frac{8}{3}\nabla_\mu\mathcal Z\,^*\!R_{\nu\alpha\beta\rho}\bigg]\theta^{\alpha\beta}\hat p^\mu\hat p^\nu\hat p^\rho.
    \end{split}
\end{align}
It is possible to further reduce this expression. It is useful to first derive some properties of the dualizations of Riemann tensor\textit{ in Ricci-flat spacetimes}. For any tensor $M_{abcd}$ with the symmetries of the Riemann tensor, one has
\begin{align}
    \sMs\tud{\alpha\beta}{\mu\nu}&=-6\delta^{\alpha\beta ab}_{[\mu\nu cd]}M\tdu{ab}{cd}.
\end{align}
In Ricci-flat spacetimes, this yields for the Riemann tensor,
\begin{align}
    \sRs_{\alpha\beta\gamma\delta}=-R_{\alpha\beta\gamma\delta}.\label{sRs}
\end{align}
Moreover, dualizing this equation one more time gives rise to the identity
\begin{align}
    R^*_{\alpha\beta\gamma\delta}=^*\! R_{\alpha\beta\gamma\delta}.
\end{align}
Using $R_{\alpha\beta\gamma\delta}=R_{\gamma\delta \alpha\beta}$, we also deduce 
\begin{align}
    R^*_{\alpha\beta\gamma\delta}=R^*_{\gamma\delta\alpha\beta}. \label{relR}
\end{align}
Notice that we also have the identity
\begin{align}
\nabla^\lambda R_{\lambda\alpha\beta\gamma}=0.
\end{align}
Making use of the properties of the Riemann tensor and of the identity $^*R_{\lambda\alpha\beta\rho}g^{\alpha\beta}=0$, we get the additional relations
\begin{subequations}
\begin{align}
    R_{\lambda\alpha\beta\rho}\theta^{\alpha\beta}\hat p^\rho &= -R_{\lambda\alpha\beta\rho}s^\alpha s^\beta \hat p^\rho,\\
    ^*R_{\lambda\alpha\beta\rho}\theta^{\alpha\beta}\hat p^\rho &= -^*R_{\lambda\alpha\beta\rho}s^\alpha s^\beta \hat p^\rho,\\
    R_{\lambda\nu\beta\rho}\theta^{\alpha\beta}\hat p^\nu \hat p^\rho &= R_{\lambda\nu\beta\rho}\qty(g^{\alpha\beta}\mathcal S^2-s^\alpha s^\beta) \hat p^\nu \hat p^\rho.
\end{align}
\end{subequations}
This allows to express $\expval{\dot{\mathcal{Q}}^{Q}_R}_{(2)}$ in terms of the independent variables as
\begin{align}
    \begin{split}
    \expval{\dot{\mathcal{Q}}^{Q}_R}_{(2)}&=-\kappa\bigg[\left( \nabla_\nu\qty(K_{\mu\lambda}R\tud{\lambda}{\alpha\beta\rho})+\nabla_\nu K_{\mu\lambda}R\tud{\lambda}{\alpha\beta\rho}-\qty(\nu\leftrightarrow\alpha)\right)\\
    &\quad+\frac{4}{3}\nabla_{[\alpha} K_{\mu ] \lambda}R\tud{\lambda}{\nu\beta\rho}+\qty(\frac{4}{3}\nabla_\sigma K_{\mu\lambda}R\tudud{\lambda}{\nu}{\sigma}{\rho}+\frac{2}{3}\nabla_\mu K_{\sigma\lambda}R\tudud{\lambda}{\nu}{\sigma}{\rho})g_{\alpha\beta}\\
    &\quad{+\frac{16}{3}\nabla_\lambda\mathcal Z g_{\mu[\nu} \,^*R\tud{\lambda}{\alpha]\beta\rho}}-\frac{8}{3}\nabla_\mu\mathcal Z\,^*R_{\nu\alpha\beta\rho}\bigg]s^\alpha s^\beta \hat p^\mu\hat p^\nu\hat p^\rho.\label{exprN}
    \end{split}
\end{align}
The quantity into brackets appearing in the second line of this expression is actually vanishing:
\begin{align}
    \begin{split}
    &\qty(\frac{4}{3}\nabla_\sigma K_{\mu\lambda}R\tudud{\lambda}{\nu}{\sigma}{\rho}+\frac{2}{3}\nabla_\mu K_{\sigma\lambda}R\tudud{\lambda}{\nu}{\sigma}{\rho})\hat p^\nu \hat p^\rho\\
    &\quad=\frac{2}{3}\qty(2\nabla_\sigma K_{\mu\lambda}+\nabla_\mu K_{\sigma\lambda})R\tudud{\lambda}{\nu}{\sigma}{\rho} \hat p^\nu \hat p^\rho\\
    &\quad=\frac{2}{3}\qty(\nabla_\sigma K_{\mu\lambda}+\nabla_\lambda K_{\sigma\mu}+\nabla_\mu K_{\lambda\sigma})R\tudud{\lambda}{\nu}{\sigma}{\rho} \hat p^\nu \hat p^\rho=0,
    \end{split}
\end{align}
which follows from the definition of Killing tensors.

\subsection{Reduced expression}

We can now write down the complete expression $\dot{\mathcal{Q}}^{(2)}=\dot{\mathcal{Q}}_R^{MD}+\dot{\mathcal{Q}}_R^{Q}+\dot{\mathcal{Q}}^{\text{quad}}$. Gathering all the pieces \eqref{Qdot}, \eqref{exprO}, \eqref{exprN}, we get
\begin{align}
    \begin{split}
    \dot{\mathcal{Q}}^{(2)}&=\bigg[\kappa\nabla_\alpha\qty(K_{\mu\lambda}R\tud{\lambda}{\nu\beta\rho})+\nabla_\nu\qty(4 N_{\alpha\mu\beta\rho} -\kappa K_{\mu\lambda}R\tud{\lambda}{\alpha\beta\rho})-\kappa\nabla_\mu K_{\nu\lambda}R\tud{\lambda}{\alpha\beta\rho}\\
    &\quad+\frac{2\kappa}{3}\nabla_{[\mu}K_{\lambda]\alpha}R\tud{\lambda}{\nu\beta\rho}+2\qty(Y_{\lambda\mu}\xi_\alpha-Y_{\lambda\alpha}\xi_\mu)R\tud{*\lambda}{\nu\beta\rho}-{\frac{16\kappa}{3}\nabla_\lambda\mathcal Z g_{\mu[\nu} R\tud{*\lambda}{\alpha]\beta\rho}}\\
    &\quad+2\qty(\frac{4\kappa}{3}-1)\nabla_\mu\mathcal Z R^*_{\nu\alpha\beta\rho}\bigg]s^\alpha s^\beta \hat p^\mu \hat p^\nu \hat p^\rho+\mathcal O( \mathcal S^3).\label{constraint_pre_reduced}
    \end{split}
\end{align}
In order to further simplify this expression we  first derive some additional useful identities. One has
\begin{align}
    \nabla_\mu K_{\nu\lambda}R\tud{\lambda}{\alpha\beta\rho}\hat p^\mu\hat p^\nu=Y\tud{\sigma}{\mu}\nabla_\lambda Y_{\sigma\nu} R\tud{\lambda}{\alpha\beta\rho}\hat p^\mu \hat p^\nu.
\end{align}
We obtain the relation
\begin{align}
    \nabla_\mu K_{\nu\lambda}R\tud{\lambda}{\alpha\beta\rho}\hat p^\mu\hat p^\nu=\qty[\qty(Y_{\lambda\mu}\xi_\alpha-Y_{\alpha\mu}\xi_\lambda)\,^*\!R\tud{\lambda}{\nu\beta\rho}-g_{\alpha\nu}Y_{\lambda\mu}\xi_\kappa \,^*\!R\tud{\lambda\kappa}{\beta\rho}]\hat p^\mu\hat p^\nu.
\end{align}
In a similar fashion, one can prove that
\begin{align}
    \begin{split}
    &\nabla_{[\mu}K_{\lambda]\alpha}R\tud{\lambda}{\nu\beta\rho}\hat p^\mu \hat p^\nu=\bigg[\frac{3}{2}\qty(Y_{\alpha\mu}\xi_\lambda+Y_{\lambda\alpha}\xi_\mu)\,^*\!R\tud{\lambda}{\nu\beta\rho}-\frac{3}{2}g_{\mu\nu} Y_{\lambda\alpha}\xi_\kappa\,^*\!R\tud{\lambda\kappa}{\beta\rho}\\
    &\quad+\frac{1}{2}g_{\alpha\mu}\nabla_\lambda\mathcal Z \,^*\!R\tud{\lambda}{\nu\beta\rho}-\frac{1}{2}g_{\mu\nu}\nabla_\lambda\mathcal Z\,^*\!R\tud{\lambda}{\alpha\beta\rho}-\frac{1}{2}\nabla_\mu \mathcal Z ^*\!R_{\alpha\nu\beta\rho}\bigg]\hat p^\mu\hat p^\nu.
    \end{split}
\end{align}
When the dust settles, we are left with
\begin{align}
    \begin{split}
    \dot{\mathcal{Q}}^{(2)}&=\bigg[4\nabla_\mu N_{\alpha\nu\beta\rho}+2\kappa\nabla_{[\alpha}\mathcal M^{(1)}_{\vert\mu\vert \nu]\beta\rho} +\kappa\qty(g_{\alpha\mu}Y_{\lambda\nu}-g_{\mu\nu}Y_{\lambda\alpha})\xi_\kappa\,^*\!R\tud{\lambda\kappa}{\beta\rho}\\
    &\quad +\qty(2\kappa Y_{\alpha\mu}\xi_\lambda+\qty(2-\kappa)\qty(Y_{\lambda\mu}\xi_\alpha+Y_{\alpha\lambda}\xi_\mu)+{3\kappa}g_{\alpha\mu}\nabla_\lambda\mathcal Z)\,^*\!R\tud{\lambda}{\nu\beta\rho}\\
    &\quad-3\kappa g_{\mu\nu}\nabla_\lambda\mathcal Z\,^*\!R\tud{\lambda}{\alpha\beta\rho}+(3\kappa-2)\nabla_\mu\mathcal Z R^*_{\nu\alpha\beta\rho}\bigg]s^\alpha s^\beta \hat p^\mu \hat p^\nu \hat p^\rho
    +\mathcal O( \mathcal S^3),
    \end{split}
\end{align}
where $\mathcal M^{(1)}_{\alpha\beta\gamma\delta} \triangleq K_{\alpha\lambda}R\tud{\lambda}{\beta\gamma\delta}$. This is precisely the constraint equation \eqref{developped_cst}.

\section{Reducing the constraints with the covariant building blocks: intermediate algebra}\label{app:CBB_identities}

A bit of cumbersome (but straightforward) algebra leads to the following identities written most shortly in the $\alpha$-$\omega$ formulation,
\begin{subequations}
\begin{align}
    Y_{\alpha\mu}s^\alpha\hat p^\mu&=-\alpha^{(0,-1)}_B,\\
    \nabla_\mu\mathcal Z\hat p^\mu&=-\alpha_A^{(0,-1)},\\
    R^*_{\nu\alpha\beta\rho}s^\alpha s^\beta\hat p^\nu\hat p^\rho&=3M\omega^{(0,3)}_{B^2}+M\qty(\mathcal A^2+\mathcal P^2\mathcal S^2)\omega_1^{(0,3)},\\
    \nabla_\lambda\mathcal Z R\tud{*\lambda}{\alpha\beta\rho} s^\alpha s^\beta \hat p^\rho &= \frac{3M}{2}\qty
    (E_s\omega^{(0,2)}_B-D\omega^{(1,3)}_B)\nonumber\\
    &\quad+M\qty(\mathcal S^2\alpha^{(0,-1)}_A-\mathcal A\alpha^{(0,-1)}_C)\omega_1^{(0,3)},\\
    \nabla_\lambda\mathcal Z R\tud{*\lambda}{\nu\beta\rho} s^\beta \hat p^\nu \hat p^\rho &= \frac{3M}{2}\qty(E\omega^{(0,2)}_B-F\omega^{(1,3)}_B)\nonumber\\
    &\quad+M\qty(\mathcal A\alpha^{(0,-1)}_A+\mathcal P^2\alpha^{(0,-1)}_C)\alpha^{(0,3)}_1,\\
    Y_{\alpha\lambda}R\tud{*\lambda}{\nu\beta\rho}s^\alpha s^\beta\hat p^\nu\hat p^\rho&=-M\mathcal A\omega^{(0,2)}_B+\frac{M}{2}\qty(\mathcal A\omega^{(1,3)}_{\bar B}-3G\omega^{(1,3)}_B),\\
    Y_{\mu\lambda}R\tud{*\lambda}{\nu\beta\rho} s^\beta\hat p^\mu \hat p^\nu\hat p^\rho&=M\mathcal P^2\omega^{(0,2)}_B-\frac{M}{2}\qty(\mathcal P^2\omega^{(1,3)}_{\bar B}+3H\omega^{(1,3)}_B),\\
    Y_{\mu\lambda}R\tud{*\lambda}{\alpha\beta\rho} s^\alpha s^\beta\hat p^\mu \hat p^\rho&=-M\mathcal A\omega^{(0,2)}_B+\frac{M}{2}\qty(\mathcal A\omega^{(1,3)}_{\bar B}-3G\omega^{(1,3)}_B),
    \\
    \xi_\lambda R\tud{*\lambda}{\nu\beta\rho} s^\beta\hat p^\nu\hat p^\rho &= -3M\omega^{(0,3)}_{AB}+M\qty(E_s\mathcal P^2+E\mathcal A)\omega^{(0,3)}_1,\\
    Y_{\lambda\kappa}R\tud{*\lambda\kappa}{\beta\rho} s^\beta\hat p^\rho &= 4M\omega^{(0,2)}_B,\\
    Y_{\alpha\lambda}\xi_\kappa R\tud{*\lambda\kappa}{\beta\rho} s^\alpha s^\beta\hat p^\rho&=  \frac{M}{2}\qty(E_S\omega^{(0,2)}_B-3D\omega^{(1,3)}_B)\nonumber\\
    &\quad +M\qty(\mathcal A\omega^{(0,-1)}_C-\mathcal S^2\omega^{(0,-1)}_A) \alpha^{(0,3)}_1,\\ 
    Y_{\nu\lambda}\xi_\kappa R\tud{*\lambda\kappa}{\beta\rho}  s^\beta\hat p^\nu\hat p^\rho&=\frac{M}{2}\qty(E\omega^{(0,2)}_B-3F\omega^{(1,3)}_B)\nonumber\\
    &\quad-M\qty(\mathcal A\omega^{(0,-1)}_A+\mathcal P^2\omega^{(0,-1)}_C)\alpha_1^{(0,3)},\\
    \xi_\lambda Y_{\alpha\kappa}R\tudu{*\lambda}{\rho\beta}{\kappa}s^\alpha s^\beta\hat p^\rho&=M\qty(\mathcal S^2\omega^{(0,2)}_A+E_s\omega^{(1,3)}_{\bar B}-\mathcal A \omega^{(1,3)}_{\bar C})\nonumber\\
    &\quad+\frac{M}{2}\qty(3 I\omega^{(1,3)}_A+\mathcal S^2\omega^{(1,3)}_{\bar A}),\\
    \xi_\lambda Y_{\beta\kappa}R\tudu{*\lambda}{\nu\rho}{\kappa} s^\beta\hat p^\nu \hat p^\rho&=\frac{3M}{2}\qty(\mathcal A\omega^{(0,2)}_A+G\omega^{(1,3)}_A)\nonumber\\
    &\quad+M\qty(E\alpha^{(0,-1)}_B+\mathcal P^2\alpha^{(0,-1)}_C)\omega_1^{(0,3)}.
\end{align}
\end{subequations}
Let us turn to identities involving covariant derivatives of the Riemann tensor. Making use of the identities \eqref{diff_to_alg} enforces the fundamental relation:
\begin{align}
    \begin{split}
    R_{\alpha\nu\beta\rho;\mu}&=M\nabla_\mu\Re\qty(\frac{3\qty(\mathcal R N_{\alpha\nu})\qty(\mathcal R) N_{\beta\rho}-\mathcal R^2 G_{\alpha\nu\beta\rho}}{\mathcal R^5})\\
    &=-M\Im\bigg\lbrace\mathcal R^{-4}\bigg[5N_{\mu\lambda}\xi^\lambda\qty(3 N_{\alpha\nu}N_{\beta\rho}-G_{\alpha\nu\beta\rho})\\
    &\quad-3\qty(G_{\alpha\nu\mu\lambda}\xi^\lambda N_{\beta\rho}+N_{\alpha\nu}G_{\beta\rho\mu\lambda}\xi^\lambda)+2 N_{\mu\lambda}\xi^\lambda G_{\alpha\nu\beta\rho}\bigg]\bigg\rbrace.
    \end{split}
\end{align}
It leads to the following `differential-to-algebraic' identities:
\begin{subequations}
\begin{align}
    &\nabla_\mu R_{\alpha\nu\beta\rho}s^\alpha s^\beta\hat p^\mu\hat p^\nu\hat p^\rho\nonumber\\
    &=3M\Im\qty{\mathcal R^{-4}\qty[5AB^2-2B\qty(\mathcal AE+\mathcal P^2 E_S)+A\qty(\mathcal S^2\mathcal P^2+\mathcal A^2)]}
    ,\\
    &K_{\alpha\lambda}\nabla_\mu R\tud{\lambda}{\nu\beta\rho} s^\alpha s^\beta\hat p^\mu\hat p^\nu\hat p^\rho\nonumber\\
    &=-\frac{3M}{2}\Re\qty(\mathcal R^2) \Im\qty{\mathcal R^{-4}\qty[5AB^2-2B\qty(\mathcal AE+\mathcal P^2 E_S)+A\qty(\mathcal S^2\mathcal P^2+\mathcal A^2)]}\nonumber\\
    &\quad-\frac{3M}{2}\abs{\mathcal R}^2\Im\big\lbrace\mathcal R^{-4}\big[5A\abs{B}^2+A(\mathcal P^2 I+\mathcal A G)-B\qty(GE-D\mathcal P^2)\nonumber\\
    &\quad-\bar B\qty(\mathcal A E+\mathcal P^2 E_s)\big]\big\rbrace
    ,\\
    &K_{\nu\lambda}\nabla_\mu R\tud{\lambda}{\alpha\beta\rho} s^\alpha s^\beta\hat p^\mu\hat p^\nu\hat p^\rho \nonumber\\
    &=\frac{3M}{2}\Re\qty(\mathcal R^2) \Im\qty{\mathcal R^{-4}\qty[5AB^2-2B\qty(\mathcal AE+\mathcal P^2 E_S)+A\qty(\mathcal S^2\mathcal P^2+\mathcal A^2)]}\nonumber\\
    &\quad+\frac{3M}{2}\abs{\mathcal R}^2\Im\big\lbrace\mathcal R^{-4}\big[5A\abs{B}^2+A\qty(\mathcal A G-\mathcal S^2 H)\nonumber\\
    &\quad+B\qty(E_sH+\mathcal A F+\bar AB-A\bar B)-\bar B\qty(\mathcal A E+\mathcal P^2E_s)\big]\big\rbrace
    ,\\
    &K_{\mu\lambda}\nabla_\alpha R\tud{\lambda}{\nu\beta\rho} s^\alpha s^\beta\hat p^\mu\hat p^\nu\hat p^\rho\nonumber\\
    &=\frac{3M}{2}\abs{\mathcal R}^2\Im\big\lbrace\mathcal R^{-4}\big[C\qty(\mathcal P^2 G+\mathcal A H)-B\qty(EG+\mathcal AF)+B\qty(\bar AB-A\bar B)\big]\big\rbrace.
\end{align}
\end{subequations}

\subsubsection{Ansatz terms for the quadratic invariant}\label{app:ansatz_terms}
We provide here the explicit form of the terms constituting the Ansatz for the black hole quadratic invariant in terms of the covariant building blocks. We use the notation $N^{(A)}\triangleq N_{\mu\alpha\nu\beta}^{(A)}s^\alpha s^\beta\hat p^\mu\hat p^\nu$ ($A=1,\ldots,4)$.
\begin{subequations}
\begin{align}
    N^{(1)}&=-M\abs{B}^2\alpha^{(1,2)}_1+\frac{M}{4}\bigg[3\qty(\alpha_{B^2}^{(0,1)}+\alpha_{B^2}^{(2,3)})\nonumber\\
    &\quad+\qty(\mathcal A^2+\mathcal P^2\mathcal S^2)\qty(\alpha^{(0,1)}_1+\alpha^{(2,3)}_1)],\\
    N^{(2)}&=-\frac{M}{4}\qty[\qty(\mathcal A^2+\mathcal P^2\mathcal S^2)\alpha_1^{(0,1)}+\alpha_{B^2}^{(0,1)}],\\
    N^{(3)}&=\frac{1}{4}\qty[-\qty(\mathcal A^2+\mathcal P^2\mathcal S^2)+E \qty(E \mathcal S^2-E_S\mathcal A)-E_S\qty(E_S\mathcal P^2+\mathcal A E)]\nonumber\\
    &\quad+\frac{M}{2}\qty(\mathcal A^2+\mathcal P^2\mathcal S^2)\alpha^{(0,1)}_1,\\
    N^{(4)}&=\frac{1}{2}\qty(\mathcal A^2+\mathcal P^2\mathcal S^2)\qty(2M\alpha^{(0,1)}_1-1).
\end{align}
\end{subequations}

\subsubsection{Derivatives of the scalar basis}
A direct computation gives the following identities
\begin{subequations}
\begin{align}
    \hat\nabla\mathcal S^2&=\hat\nabla\mathcal P^2=\hat\nabla\mathcal A=0,\\
    \hat\nabla A&=-i\qty(E^2+\mathcal P^2\xi^2)\mathcal R^{-1}+\frac{iM}{2}\qty(\frac{\mathcal P^2}{\mathcal R^2}+\frac{H}{\bar{\mathcal R}^2})-i A^2\mathcal R^{-1},\\
    \hat\nabla B&=i\qty(\mathcal AE+\mathcal P^2E_s)\mathcal R^{-1}-iAB\mathcal R^{-1},\\
    \hat\nabla C&=i\qty(\mathcal A\xi^2-EE_s)\mathcal R^{-1}+\frac{iM}{2}\qty(\frac{G}{\bar{\mathcal R}^2}-\frac{\mathcal A}{\mathcal R^2})-iAC\mathcal R^{-1},\\
    \hat\nabla D&=2E\omega^{(0,1)}_{\bar{C}}-2\xi^2\omega^{(0,1)}_{\bar{B}}+M\omega^{(0,2)}_{\bar{B}}+2D\omega_A^{(0,1)},\\
    \hat\nabla E&=0,\\
    \hat\nabla E_s&=-M\omega_B^{(0,2)},\\
    \hat\nabla F&=2E\omega_{\bar A}^{(0,1)}+2F\omega_A^{(0,1)},\\
    \hat\nabla G&=2G\omega_A^{(0,1)}+2E\omega_{\bar B}^{(0,1)}+2\mathcal P^2\omega_{\bar C}^{(0,1)},\\
    \hat\nabla H&=2\omega_A^{(0,1)}H+2\mathcal P^2\omega_{\bar A}^{(0,1)}
    ,\\
    \hat\nabla I&=2\mathcal S^2\omega^{(0,1)}_{\bar A}+4E_s\omega^{(0,1)}_{\bar B}-4\mathcal A\omega^{(0,1)}_{\bar C}+2I\omega_A^{(0,1)}.
\end{align}
\end{subequations}

\subsubsection{Directional derivatives of the Ansatz terms}

We aim to compute the contributions $\hat\nabla N^{(A)}$. Let us proceed step by step. First, we compute $N_{\alpha\beta\gamma\delta}^{(A)}$ for each $A=1,\dots 4$. In Ricci-flat spacetimes, using the identity \eqref{sRs}, we obtain
\begin{align}
    N^{(1)}_{\alpha\beta\gamma\delta}=-\frac{1}{2}K R_{\alpha\beta\gamma\delta}+\mathcal M^{(1)}_{[\alpha\beta]\gamma\delta}
\end{align}
where $K\triangleq K\tud{\alpha}{\alpha}$. Moreover, noticing that
\begin{align}
    ^*\!\mathcal M^{(2)}_{\alpha\beta\gamma\delta}&=Y_{\lambda[\alpha}\,^*\!R\tud{\lambda}{\beta]\sigma\delta}Y\tud{\sigma}{\gamma}=-R\tudu{*}{\sigma\delta[\alpha}{\lambda}Y_{\beta]\lambda}Y\tud{\sigma}{\gamma}
\end{align}
and using the symmetries of the Riemann tensor, we can write
\begin{align}
    N^{(2)}_{\alpha\nu\beta\rho}&=-Y_{\lambda[\alpha}\sRs\tudu{\lambda}{\nu][\beta}{\sigma}Y_{\rho]\sigma}.
\end{align}
In Ricci-flat spacetimes, this boils down to
\begin{align}
    N^{(2)}_{\alpha\nu\beta\rho}&=Y_{\lambda[\alpha} R\tudu{\lambda}{\nu][\beta}{\sigma}Y_{\rho]\sigma}.
\end{align}
The two last computations are more straightforward and give
\begin{align}
    N^{(3)}_{\alpha\beta\gamma\delta}&=\frac{1}{2}N^{(4)}_{\alpha\beta\gamma\delta}-\xi_{[\alpha}g_{\beta][\gamma}\xi_{\delta]},\qquad 
N^{(4)}_{\alpha\beta\gamma\delta}=-g_{\alpha[\gamma}g_{\delta]\beta}\xi^2.
\end{align}
We shall evaluate the following covariant derivatives:
\begin{align}
    \nabla_\mu \qty(K R_{\alpha\nu\beta\rho}) s^\alpha s^\beta \hat p^\mu \hat p^\nu \hat p^\rho &=\qty(\nabla_\mu K R_{\alpha\nu\beta\rho}+ K\nabla_\mu R_{\alpha\nu\beta\rho})s^\alpha s^\beta \hat p^\mu \hat p^\nu \hat p^\rho.\label{eq_DK}
\end{align}
Using \eqref{DK} the first term of the right-hand side of \eqref{eq_DK} can be written
\begin{align}
    \begin{split}
    &\nabla_\mu K R_{\alpha\nu\beta\rho}s^\alpha s^\beta \hat p^\mu \hat p^\nu \hat p^\rho\\
    &=\bigg\lbrace 4\qty(\xi_\lambda Y_{\alpha\mu}+\xi_\mu Y_{\lambda\alpha}-\xi_\alpha Y_{\lambda\mu})R\tud{*\lambda}{\nu\beta\rho}\\
    &\quad+2\qty[g_{\alpha\mu}\qty(2Y_{\lambda\nu}\xi_\kappa-Y_{\lambda\kappa}\xi_\nu)+g_{\mu\nu}\qty(2\xi_\lambda Y_{\kappa\alpha}+\xi_\alpha Y_{\lambda\kappa})]R\tud{*\lambda\kappa}{\beta\rho}
    \bigg\rbrace s^\alpha s^\beta\hat p^\mu \hat p^\nu \hat p^\rho\\
    &=4\bigg[ \qty(\xi_\lambda Y_{\alpha\mu}+\xi_\mu Y_{\lambda\alpha}-\xi_\alpha Y_{\lambda\mu})R\tud{*\lambda}{\nu\beta\rho}\\
    &\quad+\qty(Y_{\lambda\kappa}\xi_{[\alpha}+2\xi_\lambda Y_{\kappa[\alpha})g_{\mu]\nu}
    R\tud{*\lambda\kappa}{\beta\rho}
    \bigg] s^\alpha s^\beta\hat p^\mu \hat p^\nu \hat p^\rho.
    \end{split}
\end{align}
Gathering the two pieces above yields
\begin{align}
    \begin{split}
    &\nabla_\mu N^{(1)}_{\alpha\nu\beta\rho} s^\alpha s^\beta \hat p^\mu \hat p^\nu \hat p^\rho \\
    &=\bigg[K_{\lambda[\alpha}R\tud{\lambda}{\nu]\beta\rho;\mu}-\frac{1}{2}K\nabla_\mu R_{\alpha\nu\beta\rho}-\frac{1}{2}\nabla_\mu\mathcal Z R^*_{\alpha\nu\beta\rho}\\
    &\quad-\frac{1}{2}\qty(\nabla_\lambda\mathcal Z g_{\mu\nu}+Y_{\lambda\mu}\xi_\nu)R\tud{*\lambda}{\alpha\beta\rho}
    +\frac{1}{2}\big(\nabla_\lambda\mathcal Z g_{\mu\alpha}+{ 3Y_{\lambda\mu}\xi_\alpha+Y_{\mu\alpha}\xi_\lambda}\\
    &\quad+2Y_{\alpha\lambda}\xi_\mu\big)R\tud{*\lambda}{\nu\beta\rho}-2\qty(Y_{\lambda\kappa}\xi_{[\alpha}+\xi_\lambda Y_{\kappa[\alpha})g_{\mu]\nu}R\tud{*\lambda\kappa}{\beta\rho}
    \bigg]s^\alpha s^\beta \hat p^\mu \hat p^\nu \hat p^\rho. \label{DN1}
    \end{split}
\end{align}
On the other hand,
\begin{align}
    \begin{split}
    &\nabla_\mu N^{(2)}_{\alpha\nu\beta\rho} s^\alpha s^\beta \hat p^\mu \hat p^\nu \hat p^\rho \\
    &=\nabla_\mu \qty(Y_{\lambda[\alpha} R\tudu{\lambda}{\nu][\beta}{\sigma}Y_{\rho]\sigma}) s^\alpha s^\beta \hat p^\mu \hat p^\nu \hat p^\rho\\
    &=\nabla_\mu \qty(Y_{\lambda\alpha} R\tudu{\lambda}{\nu\beta}{\sigma}Y_{\rho\sigma}) \hat p^\mu s^{[\alpha}\hat p^{\nu]} s^{[\beta} \hat p^{\rho]}\\
    &=\qty(2\nabla_\mu Y_{\lambda\alpha}Y_{\rho\sigma}R\tudu{\lambda}{\nu\beta}{\sigma}+Y_{\lambda\alpha}Y_{\rho\sigma}\nabla_\mu R\tudu{\lambda}{\nu\beta}{\sigma})\hat p^\mu s^{[\alpha}\hat p^{\nu]} s^{[\beta} \hat p^{\rho]}\\
    &=\qty(2\epsilon_{\mu\lambda\alpha\kappa}\xi^\kappa Y_{\rho\sigma}R\tudu{\lambda}{\nu\beta}{\sigma}+Y_{\lambda\alpha}Y_{\rho\sigma}\nabla_\mu R\tudu{\lambda}{\nu\beta}{\sigma})\hat p^\mu s^{[\alpha}\hat p^{\nu]} s^{[\beta} \hat p^{\rho]}\\
    &={\frac{1}{2} }2\epsilon_{\mu\lambda\alpha\kappa}\xi^\kappa Y_{\rho\sigma}R\tudu{\lambda}{\nu\beta}{\sigma} s^{\alpha} \hat p^\mu \hat p^{\nu} s^{[\beta} \hat p^{\rho]} + Y_{\lambda[\alpha|}\nabla_\mu R\tudu{\lambda}{|\nu][\beta}{\sigma}Y_{\rho]\sigma}s^\alpha s^\beta\hat p^\mu\hat p^\nu\hat p^\rho\\
    &=\bigg[\frac{1}{2}Y_{\alpha\lambda}\xi_\mu R\tud{*\lambda}{\nu\beta\rho}-\frac{1}{2}Y_{\mu\lambda}\xi_\nu R\tud{*\lambda}{\alpha\beta\rho}-\frac{1}{2}\qty(g_{\mu\nu}Y_{\alpha\kappa}+g_{\alpha\mu}Y_{\nu\kappa})\xi_\lambda R\tudu{*\lambda}{\rho\beta}{\kappa}\\
    &\quad+\frac{1}{2}g_{\mu\nu}\xi_\lambda Y_{\rho\kappa}R\tudu{*\lambda}{\alpha\beta}{\kappa}+\frac{1}{2}g_{\alpha\mu}\xi_\lambda Y_{\beta\kappa}R\tudu{*\lambda}{\nu\rho}{\kappa}
    + Y_{\lambda[\alpha|}\nabla_\mu R\tudu{\lambda}{|\nu][\beta}{\sigma}Y_{\rho]\sigma}\bigg]\\
    &\quad\times s^\alpha s^\beta\hat p^\mu\hat p^\nu\hat p^\rho.\label{DN2bis}
    \end{split}
\end{align}

Finally, one has
\begin{align}
    \nabla_\mu\qty(\xi_{(\alpha}\xi_{\beta)})=2\nabla_\mu\xi_{(\alpha}\xi_{\beta)}.
\end{align}
In Ricci-flat spacetimes, Eq. (157) of \cite{Compere:2021kjz} boils down to the identity
\begin{align}
    \nabla_\alpha\xi_\beta=-\frac{1}{4}R^*_{\alpha\beta\gamma\delta}Y^{\gamma\delta}.
\end{align}
All in all, we obtain the relations
\begin{align}
    \nabla_\mu\qty(\xi_{(\alpha}\xi_{\beta)})&=\frac{1}{2}\xi_{(\alpha}R^*_{\beta)\mu\gamma\delta}Y^{\gamma\delta},\qquad
    \nabla_\mu\qty(\xi^2)=\frac{1}{2}\xi_\lambda R\tud{*\lambda}{\mu\gamma\delta}Y^{\gamma\delta}.
\end{align}
This yields
\begin{align}
    \begin{split}
    \nabla_\mu N^{(3)}_{\alpha\nu\beta\rho} s^\alpha s^\beta \hat p^\mu \hat p^\nu \hat p^\rho&=\qty[\frac{1}{2}\nabla_\mu N^{(4)}_{\alpha\nu\beta\rho}-\nabla_\mu\qty(\xi_{[\alpha}g_{\nu][\beta}\xi_{\rho]})]s^\alpha s^\beta\hat p^\mu\hat p^\nu\hat p^\rho\\
    &=\qty[\frac{1}{2}\nabla_\mu N^{(4)}_{\alpha\nu\beta\rho}+\frac{1}{2}\xi_{[\alpha}g_{\nu][\rho}R^*_{\beta]\mu\gamma\delta}Y^{\gamma\delta}]s^\alpha s^\beta\hat p^\mu\hat p^\nu\hat p^\rho\\
    &=\qty[\frac{1}{2}\nabla_\mu N^{(4)}_{\alpha\nu\beta\rho}+\frac{1}{4}\xi_{[\alpha}g_{\mu]\nu}R^*_{\beta\rho\gamma\delta}Y^{\gamma\delta}]s^\alpha s^\beta\hat p^\mu\hat p^\nu\hat p^\rho\label{DN3}
    \end{split}
\end{align}
as well as
\begin{align}
    \begin{split}
    \nabla_\mu N^{(4)}_{\alpha\nu\beta\rho} s^\alpha s^\beta \hat p^\mu \hat p^\nu \hat p^\rho &= -g_{\alpha[\beta}g_{\rho]\nu}\nabla_\mu\qty(\xi^2) s^\alpha s^\beta \hat p^\mu \hat p^\nu \hat p^\rho\\
    &=-\frac{1}{2}g_{\alpha[\beta}g_{\mu]\nu}\xi_\lambda R\tud{*\lambda}{\rho\gamma\delta}Y^{\gamma\delta}s^\alpha s^\beta \hat p^\mu \hat p^\nu \hat p^\rho.\label{DN4}
    \end{split}
\end{align}

Let us now use the preceding relations to write down the desired contribution in the covariant building blocks language. We will demonstrate the procedure on the $A=1$ term, which turns out to be the most involved to compute. The computations of the others contributions will not be detailed in this text. One has
\begin{align}
   \begin{split}
    &\qty(K_{\lambda[\alpha}R\tud{\lambda}{\nu]\beta\rho;\mu}-\frac{1}{2}K\nabla_\mu R_{\alpha\nu\beta\rho})s^\alpha s^\beta \hat p^\mu \hat p^\nu \hat p^\rho=\frac{15M}{4}\qty(\omega^{(0,2)}_{AB^2}+\omega^{(2,4)}_{AB^2})\\
    &\quad+\frac{3M}{4}\qty(\mathcal S^2\mathcal P^2+\mathcal A^2)\qty(\omega^{(0,2)}_A+\omega^{(2,4)}_A)\\
    &\quad-\frac{3M}{2}\qty(\mathcal AE+\mathcal P^2E_s)\qty(\omega^{(0,2)}_B+\omega^{(2,4)}_B)\\
    &\quad-\frac{3M}{4}\qty(\mathcal P^2I-\mathcal S^2 H+10\abs{B}^2+2\mathcal AG)\omega^{(1,3)}_A\\
    &\quad-\frac{3M}{4}\qty(\mathcal P^2D-EG+E_sH+\mathcal AF)\omega^{(1,3)}_B\\
    &\quad+\frac{3M}{2}\qty(\mathcal AE+\mathcal P^2E_s)\omega^{(1,3)}_{\bar B}-\frac{3M}{2}\Im\qty(\bar AB)\alpha^{(1,3)}_B.
   \end{split}
\end{align}
Using the various identities derived above in Eq. \eqref{DN1} allows to write 
\begin{align}
    \begin{split}
    \hat\nabla N^{(1)}&=\frac{15M}{4}\qty(\omega^{(0,2)}_{AB^2}+\omega^{(2,4)}_{AB^2})-\frac{3M}{2}\qty(\omega^{(0,3)}_{B^2}\alpha^{(0,-1)}_A+\omega^{(0,3)}_{AB}\alpha^{(0,-1)}_B)\\
    &\quad+\frac{M}{4}\qty(\mathcal A^2+\mathcal P^2\mathcal S^2)\qty(\omega^{(0,2)}_A+2\omega^{(1,3)}_{\bar A}+3\omega^{(2,4)}_A)\\
    &\quad-\frac{M}{2}\qty(\mathcal AE+\mathcal P^2E_s)\qty(5\omega^{(0,2)}_B-2\omega^{(1,3)}_{\bar B}+3\omega^{(2,4)}_B)\\
    &\quad-\frac{3M}{4}\qty(10\abs{B}^2+2\mathcal AG+\mathcal P^2I-\mathcal S^2H)\omega^{(1,3)}_A\\
    &\quad-3M\qty(\mathcal AF-EG+E_sH+\mathcal P^2D)\omega^{(1,3)}_B-\frac{3M}{2}\omega^{(0,0)}_{\bar AB}\alpha^{(1,3)}_B.
    \end{split}
\end{align}
Making use of the relations \eqref{lin_dep_terms} allows to express $\text{DN}^{(1)}$ in terms of linearly independent contributions. When the dust settles down, we are left with
\begin{align}
    \begin{split}
    \hat\nabla N^{(1)}&=\frac{M}{4}\qty(\mathcal A^2+\mathcal P^2\mathcal S^2)\qty(\omega^{(0,2)}_A+2\omega^{(1,3)}_{\bar A}+3\omega^{(2,4)}_A)\\
    &\quad-\frac{M}{2}\qty(\mathcal AE+\mathcal P^2E_s)\qty(5\omega^{(0,2)}_B-2\omega^{(1,3)}_{\bar B}+3\omega^{(2,4)}_B)\\
    &\quad-\frac{9M}{2}\abs{B}^2\omega^{(1,3)}_A -\frac{3M}{2}\qty(\mathcal AF-EG+E_sH+\mathcal P^2D)\omega^{(1,3)}_B\\
    &\quad+\frac{9M}{4}\omega^{(0,2)}_{AB^2}+\frac{15M}{4}\omega^{(2,4)}_{AB^2}.
    \end{split}
\end{align}

\newpage

\phantom{a}

\newpage

\setlength{\columnsep}{1cm}

\thispagestyle{empty}
\addcontentsline{toc}{chapter}{Index of the main equations}
\thispagestyle{empty}
\printindex

\addcontentsline{toc}{chapter}{Bibliography}
\bibliographystyle{unsrt}
\bibliography{author/biblio}

\pagestyle{empty}
\pagenumbering{gobble}

\newpage

\phantom{a}

\newpage

\color{white}

\newpagecolor{dukeblue}\afterpage{\restorepagecolor}

\phantom{a}

\vspace{\stretch{1}}

\begin{center}
    \includegraphics[width=0.75\textwidth]{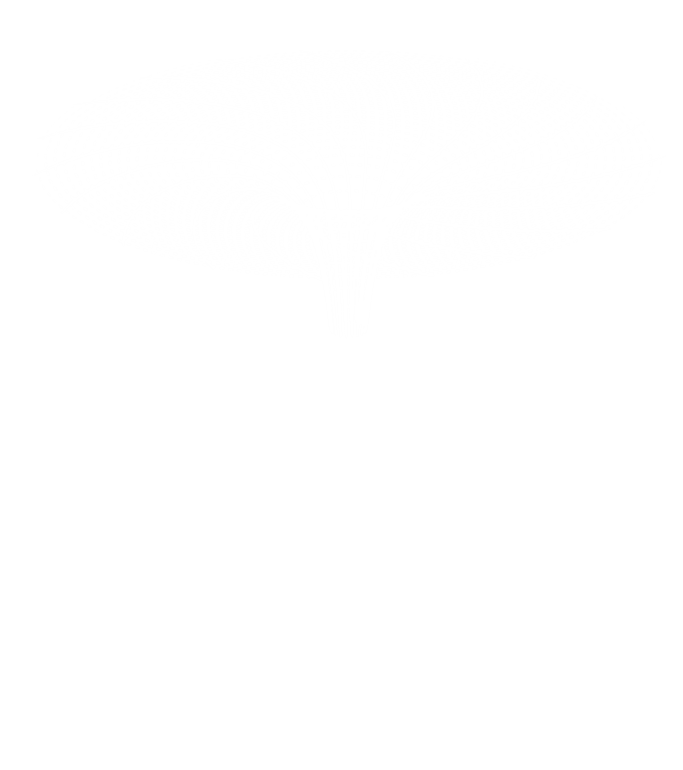}
\end{center}

\vspace{\stretch{1}}

\end{document}